%% file: bookv2.tex
\documentclass[graybox,envcountchap,sectrefs]{svmono}
\usepackage[T1]{fontenc}

\usepackage[frozencache]{minted}
\usepackage{cite}
\usepackage{bibunits}  % bibunits.sty
\usepackage{mmacells}
\usepackage[linesnumbered,ruled]{algorithm2e}

\usepackage{hyperref}
\hypersetup{
linktocpage=true,
colorlinks=true,
linkcolor=blue,          % color of internal links (change box color with linkbordercolor)
citecolor=mygreen,        % color of links to bibliography
filecolor=blue,      % magenta color of file links
urlcolor=blue           % color of external links
}  %YR
\usepackage{url}   

\usepackage{comment}

\usepackage{tikz}
\usetikzlibrary{arrows,positioning,matrix}
\usepackage{enumitem}

\usepackage{multirow}

\definecolor{spred}{rgb}{1.0, 0.0, 0.0}
\definecolor{spgreen}{rgb}{0.0, 0.27, 0.13}
\definecolor{spblue}{rgb}{0.0, 0.0, 1.0}
\definecolor{spviolet}{rgb}{0.78, 0.08, 0.52}

\newcommand{\redtt}[1]{{\color{spred}{#1}}}
\newcommand{\bluett}[1]{{\color{spblue}{#1}}}
\newcommand{\greentt}[1]{{\color{spgreen}{#1}}}
\newcommand{\violett}[1]{{\color{spviolet}{#1}}}
\usepackage{pgfplots}
\pgfplotsset{compat=newest} 

\usepackage{tikz-feynman}
\usepackage{calc}

\usepackage[htt]{hyphenat}

\usepackage{framed,amsmath}
\usepackage{type1cm}         

\usepackage{makeidx}         % allows index generation
\usepackage{graphicx}        % standard LaTeX graphics tool
\usepackage{multicol}        % used for the two-column index
\usepackage[bottom]{footmisc}% places footnotes at page bottom

\usepackage{newtxtext}       % 
\usepackage{newtxmath}       % selects Times Roman as basic font

\usepackage{subcaption}

\makeindex             % used for the subject index
\usepackage{mmacells} % allows import of Mathematica cells
\usepackage{extarrows}

\input{definitions.tex}

\begin{document}

\author{Ievgen Dubovyk, Janusz Gluza \\G\'abor Somogyi}
\title{Mellin-Barnes Integrals:\\
A Primer to Particle Physics Applications
%Mellin-Barnes integrals: \\A primer to particle physics
%applications
}
\subtitle{Lecture Notes in Physics}
\maketitle

\frontmatter%%%%%%%%%%%%%%%%%%%%%%%%%%%%%%%%%%%%%%%%%%%%%%%%%%%%%%
\bibliographystyle{elsarticle-num-ID} % long
%\bibliographystyle{bibs/utphys_zbb} % short
%\bibliographystyle{bibs/h-elsevier2}
% also long: utphys_spires_tit.bst

\include{dedication}

%\include{foreword}
\include{prefacev2}
\include{acknowledgementv2}

\tableofcontents

\include{acronym}

\mainmatter%%%%%%%%%%%%%%%%%%%%%%%%%%%%%%%%%%%%%%%%%%%%%%%%%%%%%%%\include{part}

% our stuff start
\include{chapter1v2}

\include{chapter2v2}
\include{chapter3v2}
\include{chapter4}
\include{chapter5v2}
\include{chapter6}
\include{appendixv2}
% our stuff end

%\include{appendix_SandZsums}
%\include{chapter_template}
\backmatter%%%%%%%%%%%%%%%%%%%%%%%%%%%%%%%%%%%%%%%%%%%%%%%%%%%%%%%
\printindex

%%%%%%%%%%%%%%%%%%%%%%%%%%%%%%%%%%%%%%%%%%%%%%%%%%%%%%%%%%%%%%%%%%%%%%
\end{bibunit}
\end{document}

%% file: definitions.tex
\usetikzlibrary{arrows}
\usetikzlibrary{automata}

\usetikzlibrary{shapes.misc}
\usepackage{tikz}
\usetikzlibrary{decorations.markings}
\usepackage{circuitikz}
\usepackage{calc}
\usetikzlibrary{decorations.text,fit,chains,calc,shapes.geometric,intersections}
\usetikzlibrary{patterns,angles,quotes}
\tikzset{
    cross/.pic = {
    \draw[rotate = 45] (-#1,0) -- (#1,0);
    \draw[rotate = 45] (0,-#1) -- (0, #1);
    }
}

\usepackage{hyperref}
    \hypersetup{
    linktocpage=true,
    colorlinks=true,
    linkcolor=blue,          % color of internal links (change box color with linkbordercolor)
    citecolor=mygreen,        % color of links to bibliography
    filecolor=blue,      % magenta color of file links
    urlcolor=blue           % color of external links
    }  %YR
    
\def\simasO^#1{\stackrel{#1}{\sim}}
\def\simas^#1{\hskip -3pt%
\stackrel{#1}{\lower 5pt\hbox{$\widetilde{\hphantom{#1}}$}}%
\hskip -3pt}
    
\def \mblong      {{Mellin-Barnes{}}{}}    

\def \sd      {\texttt{SD{}}{}}    
\def \mb      {\texttt{MB{}}{}}
\def \MB      {\texttt{MB{}}{}}
\def \pltest      {\texttt{PlanarityTest}{}}
\def \pltestm      {\texttt{PlanarityTest.m}{}}

\def \mbn      {\texttt{MBnumerics.m}{}}

\def \ambrem   {\texttt{AMBRE.m}{}}

\def \math {\texttt{Mathematica}{}}
\def \mbm  {\texttt{MB.m}{}}
\def \mbresolve  {\texttt{MBresolve.m}{}}
\def \mbr  {\texttt{MB}{}}

\def \ar {\texttt{AMBRE}{}}
\def \arm {\texttt{AMBRE.m}{}}
\def \la {\texttt{LA}{}}
\def \ga {\texttt{GA}{}}
\def \ha {\texttt{HA}{}}

\def\mathswitch#1{\relax\ifmmode#1\else$#1$\fi}
\def\mathswitchr#1{\relax\ifmmode{\mathrm{#1}}\else$\mathrm{#1}$\fi}

\newcommand{\bea}{\begin{eqnarray}\nonumber}
\newcommand{\eea}{\end{eqnarray}}

\newcommand{\bqa}{\begin{eqnarray}}
\newcommand{\eqa}{\nonumber\end{eqnarray}}
\newcommand{\nn}{\nonumber}

\newcommand{\eps}{{\epsilon}}
\newcommand{\ep}{{\epsilon}}
\def \litwo {{\rm{Li_2}}}
\def \litr {{\rm{Li_3}}}

\newcommand{\bq}{\begin{equation}}
\newcommand{\eq}{\end{equation}}

\newcommand{\ba}{\begin{eqnarray}}
\newcommand{\ea}{\end{eqnarray}}

\newcommand{\hint}{\underline{Hint}:}

\def\beq{\begin{equation}}
\def\eeq{\end{equation}}
\def\beqa{\begin{eqnarray}}
\def\eeqa{\end{eqnarray}}
\def\ben{\begin{enumerate}}
\def\een{\end{enumerate}}
\def\bit{\begin{itemize}}
\def\eit{\end{itemize}}

\def\mcC{\mathcal{C}}
\def\mcJ{\mathcal{J}}
\def\mcA{\mathcal{A}}
\DeclareMathOperator*{\res}{Res}

\def\zinf{z_{\infty}}

\def\sln{\mathop{\rm sln}}

\def \h {\texttt{H}{}}

\newcommand{\colbref}[1]{{\color{blue}{\href{#1}{#1}}}}

\newcommand{\wwwaux}[1]{the file \texttt{MB\_#1\_Springer.nb} in the auxiliary material in \cite{www_aux_springer}}

\newcommand{\cF}{{\cal F}}

\newcommand{\cM}{{\cal M}}

\def \litri {{\rm{Li_3}}}
\def \lin   {{\rm{Li_n}}}
\def\litwo {{\rm{Li_2}}}
\def\litr {{\rm{Li_3}}}

\def\ep   {\epsilon}
\def\eps   {\epsilon}

\def\a4{{\rm{Li_4({1/2})}}}

\def\H(#1){{{H}}_{#1}}
\def\Hh(#1,#2){{{H}}_{#1,#2}}
\def\Hhh(#1,#2,#3){{{H}}_{#1,#2,#3}}
\def\Hhhh(#1,#2,#3,#4){{{H}}_{#1,#2,#3,#4}}

%% file: dedication.tex
%%%%%%%%%%%%%%%%%%%%%%% dedic.tex %%%%%%%%%%%%%%%%%%%%%%%%%%%%%%%%%
%
% sample dedication
%
% Use this file as a template for your own input.
%
%%%%%%%%%%%%%%%%%%%%%%%% Springer %%%%%%%%%%%%%%%%%%%%%%%%%%

\begin{dedication}

%% Please have the foreword written here
We dedicate this book to the memory of \\Dr. Tord Riemann.
 
\begin{figure}
%   \centering 
%    \label{}
 %{\small}\\
%\includegraphics[height=3cm]{figs/Mellin}
\hspace{4cm}
 \includegraphics[height=4cm]{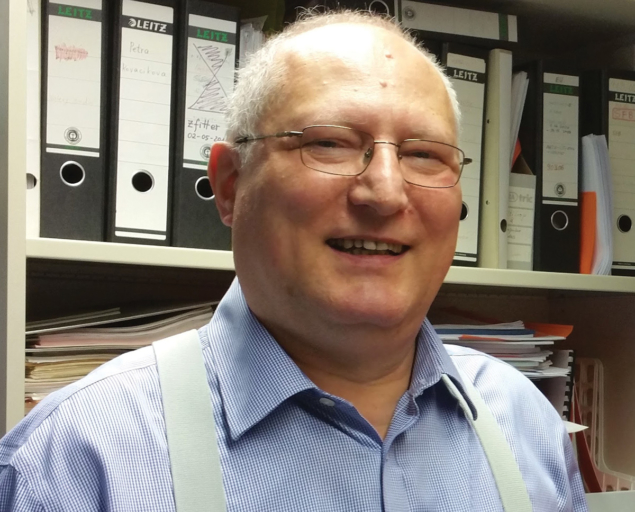}
%\hspace*{5cm}  {\small }
%\sidecaption
% \caption*{\vspace*{-.5cm} 
% Dr. Tord Riemann, 1951-2021, %\colbref{https://cerncourier.com/a/tord-riemann-1951-2021}. \newline   
% \hspace{-3cm} 
% Photo thanks to Dr. Sabine~Riemann.
% }
\end{figure}
\vspace{-.3cm}
 {\footnotesize{\rm Dr. Tord Riemann, 1951-2021, \colbref{https://cerncourier.com/a/tord-riemann-1951-2021}. %\newline   
 \hspace{4cm} 
 
 Photo thanks to Dr. Sabine~Riemann.
 }}
 \vspace{.5cm}
 \noindent
 
\hspace{-.5cm} We are thankful to Tord for his support and encouragement for writing up these notes.

%\vspace{\baselineskip}
%\begin{flushright}\noindent
%Place, month year\hfill {\it Authors}\\
%\end{flushright}

\end{dedication}

%% file: prefacev2.tex
%%%%%%%%%%%%%%%%%%%%%%preface.tex%%%%%%%%%%%%%%%%%%%%%%%%%%%%%%%%%%%%%%%%%
% sample preface
%
% Use this file as a template for your own input.
%
%%%%%%%%%%%%%%%%%%%%%%%% Springer %%%%%%%%%%%%%%%%%%%%%%%%%%

\preface

%% Please write your preface here
%%Customarily \textit{acknowledgments} are included as last part of the preface.
%}}}
%\begin{flushright}
%{\it Today in the sciences, books are usually either texts or retrospective reflections upon one aspect or another of the scientific life. The scientist who writes one is more likely to find his professional reputation impaired than enhanced. }\newline 
%Thomas~S.~Kuhn `The structure of scientific revolutions'
%\end{flushright}

\vspace{-3cm}
\begin{figure}
    \centering 
\includegraphics[height=4.5cm]{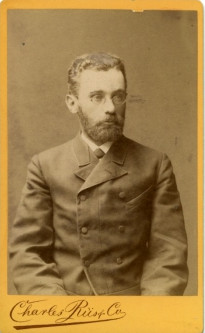}
 \includegraphics[height=4.5cm]{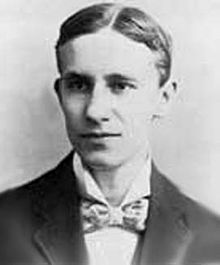}\\
 \caption*{On the left: Hjalmar Mellin (1854-1933) Finnish mathematician, Helsinki University of Technology. \copyright{}
 Atelier Charles Rüs \& Co., Hjalmar Mellin, Ident.Nr. 1992,32/003a.024, Kunstbibliothek, Staatlichen Museen zu Berlin under CC BY-NC-SA 3.0 DE,  \colbref{http://www.smb-digital.de/eMuseumPlus}.\\
 On the right: Ernest William Barnes (1874-1953)  British mathematician, liberal 
 theologian and bishop. {\copyright Foto:  \colbref{https://en.wikipedia.org/wiki/Ernest\_Barnes}}.} %\colbref{https://en.wikipedia.org/wiki}}.
\end{figure}

This book is about the mathematical basics of {\textit{Mellin-Barnes}} integrals,
however, it is motivated by the practical needs of precision calculations in particle physics. 

In modelling physical phenomena, we aim at predictability and exactness. 
In particle physics a compelling model exists, commonly called the Standard Model (SM), which describes elementary particles and their interactions.  We know since the 1960's that gauge interactions, 
mathematically based on a product of $U(1)$, $SU(2)$ and $SU(3)$ groups, describe the electromagnetic, weak and strong forces. We know the charges, spins, C and P parities, as well as the masses of 
elementary particles. We understand how fermions and bosons interact by chiral weak interactions 
and we understand the dynamics of the quarks and gluons that partake in the strong interaction. 
The SM describes innumerable physical phenomena observed in the laboratory and the cosmos, linking present-day experiments with the first $10^{-11}$ second of the Universe. 
Looking at the list of Nobel laureates in Physics, many of them are connected with particle physics. 
Their awards speak to the appreciation of the Standard Model theory builders and to the enormous efforts 
made to confirm the SM structure and predictions experimentally. 
 
The SM is surprisingly precise in its predictions and complete in the sense that all of its theorized 
elementary particles have now been observed. 
The last missing piece of the SM, the Higgs boson was discovered in 2012. This discovery was made 
possible by the construction of the LHC accelerator {and} the ATLAS and CMS detectors, whose designs 
were guided by theoretical predictions based on the SM. Thus, we now believe that we can answer the  theoretical question dating from the 1960s: "what gives mass to the electroweak bosons?".  

{\textit{However, this achievement does not stop the need for further exploration!}}
%This discovery also completed the set of particles and interactions that form the Standard Model (SM). 

First, does the ${\rm 125 GeV/c^2}$ particle that was discovered really play the role which theory says it does? Second, several unanswered mysteries call for new phenomena beyond the SM: (i) the apparent 
presence of dark matter in the Universe; (ii) the matter-antimatter asymmetry of the Universe; and (iii) 
the existence and pattern of neutrino masses to name three outstanding ones. Although the list of 
issues to be explored in the future is long\footnote{\label{briefbook} Physics Briefing Book : Input for the European Strategy for Particle Physics Update 2020,
Richard Keith Ellis et al., CERN-ESU-004,
\href{http://arxiv.org/abs/arXiv:1910.11775}{arXiv:1910.11775}.}, there is at present no clear guidance on the energy scale or strength of the new phenomena. However since the SM is complete, there is no freedom in its predictions and any observed deviation would be a discovery and a precious pointer 
into the unknown. 
An ambitious program of precision Higgs and electroweak measurements is thus in order.
However, a well-prepared and well-motivated collider project requires precise theory input to guide 
the experiment and confront the results, thus optimizing the chances for significant scientific 
progress. Ever-increasing mathematical virtuosity is however required to reach the ever-improving 
level of experimental precision. 

{\textit{And this is a place where the theory and methods connected with Mellin-Barnes integrals described in the book enter. }}

To achieve the required precision of theoretical calculations, higher-order perturbative 
corrections in Quantum Field Theory must be calculated. This involves the computation of so-called 
multi loop Feynman integrals as well as phase space integrals. 
However, currently there is no complete theory for Feynman integrals beyond one loop. This defines 
a challenge to theoreticians and calls for innovative methods and tools, both numerical and analytical. 

This book focuses on one of these, the \textit{Mellin-Barnes} (MB) method. This method has been 
applied to real applications at the frontier of research and has proven to be one of very few 
ones which could solve so far untouchable problems of precision physics.  
The book provides a self-contained overview of the necessary elements of complex analysis and 
special functions, an easy-to-follow introduction to the subject of MB representations with 
hands-on examples and a comprehensive view on the analytical and numerical methods and tools
related to MB integrals. It should be a good starting point for learning and using the MB method  in practical research.
 
\vspace{\baselineskip}
\begin{flushright}\noindent
Katowice, Cisownica, Budapest\hfill {\it Ievgen Dubovyk}\\
July 2022\hfill {\it Janusz Gluza}\\
\hfill {\it G\'abor Somogyi}
\end{flushright}

\begin{bibunit}[elsarticle-num-ID] 

\section*{Preface version 2}

The book\footnote{Mellin-Barnes Integrals: A Primer on Particle Physics Applications, Ievgen Dubovyk, Janusz Gluza, G\'abor Somogyi, Springer Nature, Lecture Notes in Physics 1008 (2022), {\color{blue}{\href{https://link.springer.com/book/10.1007/978-3-031-14272-7}{https://link.springer.com/book/10.1007/978-3-031-14272-7}}}} was published in 2022. %\cite{Dubovyk:2022obc}.

Here we provide an update covering new ideas and applications of Mellin--Barnes integrals in particle physics that have appeared in the literature since its publication. These developments are included in Chapters 1, 2, 3, 5 and the Appendix. 

We would be grateful if you could let us know about any relevant material that we may still be missing.

\vspace{\baselineskip}
\begin{flushright}\noindent
Katowice\hfill {\it Janusz Gluza} \\
Budapest\hfill {\it G\'abor Somogyi} \\
November 2026\hfill {} 
\end{flushright}

\iffalse
\putbib[%
bibs/refs,%
bibs/2loops\_LL16,%
bibs/Phd\_Dubovyk,%
bibs/LRrefa,%
bibs/2loopsreport]
\end{bibunit}
\fi

%% file: acknowledgementv2.tex
%%%%%%%%%%%%%%%%%%%%%%acknow.tex%%%%%%%%%%%%%%%%%%%%%%%%%%%%%%%%%%%%%%%%%
% sample acknowledgement chapter
%
% Use this file as a template for your own input.
%
%%%%%%%%%%%%%%%%%%%%%%%% Springer %%%%%%%%%%%%%%%%%%%%%%%%%%

\extrachap{Acknowledgements}

We would like to thank for collaborations on the papers connected with the \mb{} subject and/or many interesting discussions and insights on the subject to: \newline
Stefano Actis,  Paolo Bolzoni,  Micha\l{} Czakon, Andrei Davydychev, Claude Duhr, Wojciech Flieger, Ayres Freitas, Marek Gluza, Krzysztof Grzanka, Howard Haber, Tomasz Jeliński, Krzysztof Kajda, Mikhail Kalmykov, David Kosower, Sven-Olaf Moch, Vladimir A. Smirnov, Bas Tausk, Sam Van Thurenhout, Johann Usovitsch, Szymon Zi\k{e}ba.

This work has been supported in part by the Polish National Science Centre (NCN) under grants 2017/25/B/ST2/01987, 2020/37/B/ST2/02371 and 2023/50/A/ST2/00224.

%% file: acronym.tex
%%%%%%%%%%%%%%%%%%%%%%acronym.tex%%%%%%%%%%%%%%%%%%%%%%%%%%%%%%%%%%%%%%%%%
% sample list of acronyms
%
% Use this file as a template for your own input.
%
%%%%%%%%%%%%%%%%%%%%%%%% Springer %%%%%%%%%%%%%%%%%%%%%%%%%%

\extrachap{Acronyms}

\begin{description}%[CABR]
\item[1BL]{First Barnes Lemma}
\item[2BL]{Second Barnes Lemma}
\item[BSM]{Beyond the Standard Model}
\item[CO]{Collinear Divergence}
\item[DEs]{Differential Equations}
\item[CW]{Cheng-Wu (Variables, Theorem)}
\item[IBP]{Integration-By-Parts}
\item[LT]{Lefschetz Thimbles}
\item[IR]{Infrared Divergence}
\item[LA]{Loop-by-Loop Approach}
\item[MIs]{Master Integrals}
\item[NLO]{Next-to-Leading Order in Perturbation Theory}
%\item[eMPL]{Elliptic Functions}
\item[FI]{Feynman Integrals}
\item[GA]{Global Approach}
\item[GRMT]{Generalized Ramanujan’s Master Theorem}
\item[HPL]{Harmonic Polylogarithms}
\item[MB]{Mellin-Barnes}
\item[MPL]{Multiple Polylogarithms}
\item[SD]{Sector Decomposition}
\item[SM]{Standard Model}
\item[UV]{Ultraviolet Divergence}
\item[QCD]{Quantum Chromodynamics}
\item[QED]{Quantum Electrodynamics}
\item[QFT]{Quantum Field Theory}
\item[] 
\newpage

\item[$\mathbb{N}$]{Set of Natural Numbers $(0, 1, 2, \ldots)$} 
\item[$\mathbb{N}^+$]{Set of Natural Numbers, Zero Excluded $(1, 2, \ldots)$} 
\item[$\mathbb{Z}$]{Set of Integers $(0, \pm 1, \pm 2, \ldots)$}
\item[$\mathbb{R}$]{Set of Real Numbers} 
\item[$\mathbb{C}$]{Set of Complex Numbers} 
\item[$\Re(z)$]{Real Part of the Complex Number $z$} 
\item[$\Im(z)$]{Imaginary Part of the Complex Number $z$} 
\item[$\gamma$]{Euler's Constant}
\item[$\Gamma$]{Gamma Function}
\item[$\psi$]{Digamma Function (Polygamma Function of Order Zero)}
\item[$\psi^{(n)}$]{$n^{\rm th}$ Derivative of the Digamma Function $\psi$} 
\item[$\mathrm{Li}_{n}(z)$]{Classical Polylogarithm of Weight $n$}
\item[$S_{n,p}(z)$]{Nielsen Generalized Polylogarithm}
\item[$(a)_n$]{Pochhammer's Symbol}
\item[$\zeta(z)$]{Euler-Riemann zeta function}
\end{description}

%
% Basic math?
% \mathbb{N}^+ = \{1,2,\ldots} is the set of positive integers

%% file: chapter1v2.tex
%%%%%%%%%%%%%%%%%%%%% chapter.tex %%%%%%%%%%%%%%%%%%%%%%%%%%%%%%%%%
\begin{bibunit}[elsarticle-num-ID] % define the bib-style for the unit: elsarticle-num.bst
%  text-1; this is the corresponding section
%\putbib[2loops] % the *.bib
%\end{bibunit}
% go-on
%--- from: bibunits.sty, adapts the font size of ``References'' to section
\let\stdthebibliography\thebibliography
\renewcommand{\thebibliography}{%
\let\section\subsection
\stdthebibliography}

\chapter[Precision in Perturbative Particle Physics 
]
{
Precision in Perturbative Particle Physics%{: Basic Concepts}
}
\label{chapter:intro}  

\abstract{Precision in particle physics may lead to the discovery of anomalies which in turn may prompt the formulation of new paradigms in theory. For instance, particles thought to be elementary could turn out to be composite. 
In this  {chapter}, we discuss classes of quantum corrections in particle physics which emerge in precision studies and {discuss basic properties of Feynman integrals, notably their singularities}. As a simple example, we sketch the  standard calculation of dimensionally regulated one-loop Feynman integrals, then we introduce the basic idea of using Mellin-Barnes representations to solve {Feynman integrals} in Euclidean and Minkowskian kinematic regimes.}

\section{Loops and Real Quantum Corrections in Precision Physics}

One of the most important ways of probing the laws of nature at the smallest distance 
scales is through the study of high-energy particle collisions. Indeed, the construction 
and verification of the SM has involved many such experiments of increasing energy, 
spanning decades. Today, a steady increase of the energy and the intensity of the colliding particles is indispensable in searching for unknown feeble interactions and exploring in a deeper way the structure of the smallest quantum objects, and ultimately the vacuum itself. For instance, particles currently thought of as elementary, i.e., having no internal structure, could turn out to be composite, being built out of yet undiscovered constituents. E.g. a particle of mass 125\,GeV (the Higgs boson) corresponds to the natural wavelength of $\sim 10^{-17}$\,m, so unraveling its internal structure, if any, requires that distance scales of at least one or two magnitudes smaller are explored. This will be possible at future colliders as well as through the indirect analysis of precision quantum corrections and effective theories that will probe scales of new physics at tens of TeV level. 
This means that future colliders and studies will probe physics and particle substructures at the level of $10^{-17}$~m--$10^{-19}$~m.
Especially powerful information can be obtained by merging analyses at lepton 
and hadron colliders  
{leading to the exploration of quantum structures at the unprecedented level of attometers ($ 10^{-18}$\,m) and below.}
\begin{svgraybox}In order to fully exploit the physics potential of future facilities (establish 
discovery limits, study the viable parameter space of new particle physics models 
and so on), the experimental data will need to be confronted with precision calculations.
\end{svgraybox}
Precision calculations of quantum effects in high-energy particle collisions are performed using the methods of perturbative quantum field theory, pioneered by Dyson and Feynman among others. The application of perturbation theory is made possible by the relative smallness of all SM coupling constants at the energies relevant for these collisions, i.e. above a few tens of GeVs. (We do not consider here a problem of gauge theories like QCD which become non-perturbative at low energies, where they must be treated by different methods.) The perturbative expansion of physical quantities in the small couplings then leads to mathematical objects called Feynman Integrals (FI) which must be computed explicitly to obtain numerical predictions. Feynman diagrams provide a very useful and convenient way of representing Feynman integrals and also provide a way to visualize particle collision processes. It is not our intent here to discuss how Feynman diagrams and the corresponding Feynman integrals arise in Quantum Field Theory (QFT) and we leave these questions to the whole spectrum of excellent textbooks on the subject.

For any given particle collision process, the lowest non-vanishing perturbative order, called the Born approximation, involves diagrams with a fixed number of external lines, these represent the initial and final particles in the collision. In most cases, the diagrams of the Born approximation are tree diagrams, i.e., they do not contain internal loops. Tree diagrams are characterized by the property that cutting any internal line in the diagram makes the diagram disconnected: not every vertex can be reached from any other by traversing along internal lines. In passing we note that a given diagram has $L$ loops if one can cut $L$ different internal lines such that the diagram remains connected, but cutting $L+1$ different internal lines makes the diagram disconnected.
Higher orders in perturbation theory then lead to Feynman diagrams with more and more lines. These can either be internal to the diagram, forming closed loops, or they can exit the diagram, such that the number of external legs is increased. In the former case, the closed lines represent internal virtual particles and hence we speak of {\it virtual or loop corrections}. In the latter, the extra external lines represent the emission of extra real particles and thus such contributions are called {\it real corrections}. Obviously at higher perturbative orders mixed corrections also appear, with extra real radiation from loop diagrams. Higher order contributions in perturbation theory are also called {\it radiative corrections}, because they involve the emission, or radiation of extra (virtual or real) particles. 

Before moving on, let us make a small remark about real corrections. While the existence of loop corrections is more or less straightforward, the reason why real radiation corrections must be taken into account is perhaps less intuitive. A quick and non-rigorous explanation is as follows. Consider e.g. the process of muon-antimuon pair production in electron-positron annihilation in QED, $e^+e^- \to \mu^+ \mu^-$. The Born approximation is described by the single Feynman diagram labeled Born in Fig.~\ref{fig:epem-mupmum}. The one-loop correction to this process involves the emission and reabsorption of a virtual photon in the basic diagram. An example is shown in Fig.~\ref{fig:epem-mupmum} (Virtual). On the other hand, the real correction involves the emission of an extra real photon, as in the rightmost picture of Fig.~\ref{fig:epem-mupmum} (Real). But why should the process with an extra photon, $e^+e^- \to \mu^+ \mu^- \gamma$ be considered as a correction to the process we are studying? The answer lies in the fact that the emitted photon can have an arbitrarily small energy (soft limit) and the direction of its momentum can be arbitrarily close to that of one of the fermions (collinear limit). But in these cases, it is impossible in practice -- and in the strict limits also in principle -- to distinguish the two final states: that with just the muon-antimuon pair and that with the muon-antimuon pair and the extra photon. But the rules of quantum mechanics tell us that when computing a physical observable, say a cross section, we must sum over all final states that we do not distinguish. Thus, we must include also real corrections. This can be done in a systematic way~\cite{Yennie:1961ad}. The above argument also shows that only massless particles (in the SM photons and gluons) can enter as extra radiation in real corrections.

\begin{figure}\centering

\begin{tikzpicture}[scale=0.8]  
\begin{feynman}

\vertex at (0,2) (i1) {\(e^-\)};
\vertex at (0,0) (i2) {\(e^+\)};
\vertex at (1,1) (a);
\vertex at (2,1) (b);
\vertex at (3,2) (f1) {\(\mu^-\)};
\vertex at (3,0) (f2) {\(\mu^+\)};

\diagram*{
(i1) -- [fermion] (a) -- [fermion] (i2),
(a) -- [photon, edge label=\(\gamma/Z\)] (b),
(f2) -- [fermion] (b) -- [fermion] (f1),
};

\node at (1.5,-0.5) {Born};

\end{feynman}
\end{tikzpicture} 
\hspace{2em}
\begin{tikzpicture}[scale=0.8] 
\begin{feynman}

\vertex at (0,2) (i1) {\(e^-\)};
\vertex at (0,0) (i2) {\(e^+\)};
\vertex at (1,1) (a);
\vertex at (2,1) (b);
\vertex at (3,2) (f1) {\(\mu^-\)};
\vertex at (3,0) (f2) {\(\mu^+\)};

\vertex at (2.5,1.5) (c);
\vertex at (2.5,0.5) (d);

\diagram*{
(i1) -- [fermion] (a) -- [fermion] (i2),
(a) -- [photon, edge label=\(\gamma/Z\)] (b),
(f2) -- [fermion] (b) -- [fermion] (f1),
(c) -- [photon, edge label=\(\gamma\)] (d),
};

\node at (1.5,-0.5) {Virtual};

\end{feynman}
\end{tikzpicture} 
\hspace{2em}
\begin{tikzpicture}[scale=0.8] 
\begin{feynman}

\vertex at (0,2) (i1) {\(e^-\)};
\vertex at (0,0) (i2) {\(e^+\)};
\vertex at (1,1) (a);
\vertex at (2,1) (b);
\vertex at (3,2) (f1) {\(\mu^-\)};
\vertex at (3,0) (f2) {\(\mu^+\)};

\vertex at (2.5,1.5) (c);
\vertex at (3.5,1.5) (d) {\(\gamma\)};

\diagram*{
(i1) -- [fermion] (a) -- [fermion] (i2),
(a) -- [photon, edge label=\(\gamma/Z\)] (b),
(f2) -- [fermion] (b) -- [fermion] (f1),
(c) -- [photon] (d),
};

\node at (1.5,-0.5) {Real};

\end{feynman}
\end{tikzpicture} 

\caption{\label{fig:epem-mupmum} Electron-positron annihilation into muon-antimuon 
pair. Left: Born approximation is given by a single diagram, center: sample diagram 
contributing to the one-loop correction, right: sample diagram contributing to 
the single real emission correction.}
\end{figure}
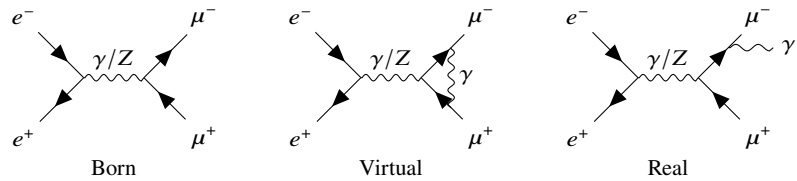

Finally we note that the integrals {corresponding to} virtual and real corrections exhibit various divergences in four spacetime dimensions. We will see in the next section that these are due either to integrating over very large momenta in loops (UV divergences) or to singularities that arise when a massless particle becomes soft (IR divergences) or their momentum becomes collinear to that of another particle (collinear divergences). The latter situation may occur both in loop and real radiation diagrams. So a naive calculation leads to infinite results! The way of treating the infinities associated with UV divergences has been known since the time of Feynman, Dyson and Tomonaga for QED, and thanks to 't Hooft and Veltman, it is also known in non-abelian gauge theories like the SM. It turns out that all such infinities are cancelled if one defines carefully the physical (i.e., measurable) couplings and masses of the theory. This procedure is called renormalization.\footnote{Besides the couplings and masses, the normalization of fields must also be considered, but we do not enter into these subtleties here.} On the other hand the infinities associated with soft and collinear particles are present even after renormalization and are in fact cancelled between virtual and real corrections. This is a somewhat more mathematical way of understanding why real radiation corrections must be included: they cancel the infrared and collinear singularities of the virtual corrections. Thus, the calculation of higher order radiative corrections to a real process involves UV renormalization and also the evaluation of both virtual and real corrections. Fig.~\ref{lepscheme} presents the a typical picture: the $2\to 2$ process of electron-positron annihilation which leads to fermion pair production, a process investigated in depth at LEP. In this figure we can find real emission and virtual loop contributions in the form of Feynman diagrams. In this book, we will consider the mathematical objects these diagrams give rise to and examine how they can be calculated using the \mblong{} method. 
%*{-1.5cm}

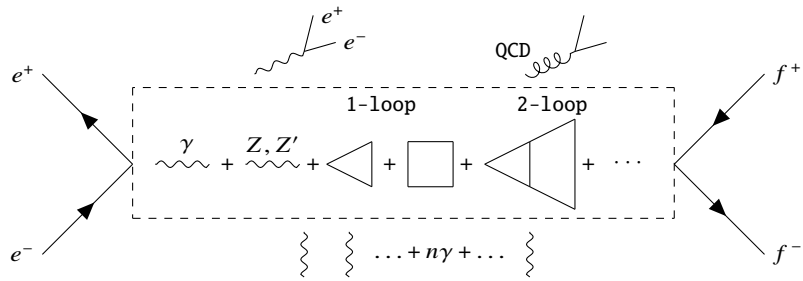
\begin{figure}[h!]
%sidecaption
\centering

\begin{tikzpicture}[scale=0.6]
\begin{feynman}
 
\vertex at (-1,4) (i1); 
\vertex at (1,3.7) (i2);
\vertex at (1,2) (i3);
\vertex at (1,0.8) (i4);
\vertex at (-1,0) (i5);

\vertex at (12,4) (i6);
\vertex at (13,3.7) (i7);
\vertex at (13,2) (i8);
\vertex at (13,0.8) (i9);
\vertex at (12,0) (i10);

\node[left] at (i1) {$e^+$};
\node[left] at (i5) {$e^-$};

\vertex at (3.7,4) (x1);
\vertex at (4.8,4.5) (x2);
\vertex at (5,5.3) (x3);
\vertex at (5.5,4.7) (x4);
\diagram*{
(x1) -- [photon] (x2) -- (x3),
(x2) -- (x4)
};
\node[right] at (x3) {\small $e^+$};
\node[right] at (x4) {\small $e^-$};

\vertex at (9.7,4) (y1);
\vertex at (10.8,4.5) (y2);
\vertex at (11,5.3) (y3);
\vertex at (11.5,4.7) (y4);
\diagram*{
(y1) -- [gluon] (y2) -- (y3),
(y2) -- (y4)
};
\node[left] at (10,4.5) {\small \texttt{QCD}};

\diagram*{
 (i5) -- [fermion] 
 (i3) --  [fermion] 
 (i1) --  [fermion], 
 (i2) --  [dashed] (i4) -- [dashed] (i9) -- [dashed] (i7)
 -- [dashed] (i2)
};

%right side
\vertex at (15,4) (f1); 
\vertex at (13,3.5) (f2);
\vertex at (13,2) (f3);
\vertex at (13,0.5) (f4);
\vertex at (15,0) (f5);

\node[right] at (f1) {$f^+$};
\node[right] at (f5) {$f^-$};

\diagram*{
 (f1) -- [fermion] 
 (f3) --  [fermion] 
 (f5)
};

\vertex at (1.5,2) (z1);
\vertex at (2.7,2) (z2);
\diagram*{
(z1) -- [photon] (z2)
};
\node at (2.2,2.4) {\small $\gamma$};
\node at (3.1,2) {$+$};

\vertex at (3.5,2) (z3);
\vertex at (4.7,2) (z4);
\diagram*{
(z3) -- [photon] (z4)
};
\node at (4.1,2.4) {\small{$Z,Z'$}};
\node at (5,2) {$+$};

\vertex at (5.3,2) (z5);
\vertex at (6.3,2.5) (z6);
\vertex at (6.3,1.5) (z7);
\diagram*{
(z5) -- (z6) -- (z7) -- (z5)
};
\node at (6.5,3.3) {\texttt{1-loop}};
\node at (6.7,2) {$+$};

%box
\vertex at (7.1,2.5) (v1);
\vertex at (8.1,2.5) (v2);
\vertex at (8.1,1.5) (v3);
\vertex at (7.1,1.5) (v4);
\diagram*{
(v1) -- (v2) -- (v3) -- (v4) -- (v1)
};
%\node at (5,2.8) {\small{1-loop}};
\node at (8.4,2) {$+$};

%2loop vertex
\vertex at (8.8,2) (w1);
\vertex at (9.8,2.5) (w2);
\vertex at (9.8,1.5) (w3);
\vertex at (10.8,3) (w4);
\vertex at (10.8,1) (w5);
\diagram*{
(w1) -- (w2) -- (w4) -- (w5) -- (w3) -- (w1), (w2) -- (w3)
};
\node at (10.3,3.3) {\texttt{2-loop}};
\node at (11.1,2) {$+$};
\node at (12,2) {$\cdots$};

%photons
\vertex at (4.8,0.5) (s1);
\vertex at (4.8,-0.5) (s2);
\vertex at (5.8,0.5) (s3);
\vertex at (5.8,-0.5) (s4);
\vertex at (7.8,0) (s5);
\vertex at (9.8,0.5) (s6);
\vertex at (9.8,-0.5) (s7);
\diagram*{
(s1) -- [photon] (s2),  (s3) -- [photon] (s4), (s6) -- [photon] (s7) };
\node at (s5) {$\ldots + n \gamma + \ldots$}; 

%\draw [yellow] (0,2.5) -- (4,2.5);
%\draw [cyan] (0, 2) -- (4,2);
\end{feynman}
\end{tikzpicture}

\caption{{A general scheme for $e^+e^- \to f \bar f$ process with typical real radiation effects (outside the dashed frame) and some chosen virtual effects (inside the dashed frame) at the lowest level (represented by $\gamma$, $Z$-boson propagators and $Z'$ extra gauge boson - which can be a part of some Beyond the Standard Model (\texttt{BSM})) and higher loop effects (represented by one-loop diagrams for vertex and box and two loops for the planar vertex).}}
\label{lepscheme}       
\end{figure}

%\section{Heart of the problems: singularities of integrals in QFT}
\section{Singularities of Amplitudes in QFT}
\label{sec:singQFT_general}

As we have described above, the basic building blocks of higher order perturbative corrections are loop and real emission diagrams. The extra emitted particles (virtual or real) carry momenta and when computing a physical observable, these momenta must be integrated over. Indeed, consider an example of the $L$-loop diagram such as the one in Fig.~\ref{fig:L-loop-diag}.
\begin{figure}\centering

\begin{tikzpicture}[scale=0.8]  
\begin{feynman}

\vertex at (0,0) (i1) {\(p_1\)};
\vertex at (0,2) (i2) {\(p_2\)};
\vertex at (1,0) (a1);
\vertex at (1,2) (a2);
\vertex at (3,0) (b1);
\vertex at (3,2) (b2);
\vertex at (5,0) (c1);
\vertex at (5,2) (c2);
\vertex at (6,0) (d1);
\vertex at (6,2) (d2);
\vertex at (7,0) (e1);
\vertex at (7,2) (e2);
\vertex at (8,0) (f1);
\vertex at (8,2) (f2);
\vertex at (10,0) (g1);
\vertex at (10,2) (g2);
\vertex at (11,0) (h1) {\(p_3\)};
\vertex at (11,2) (h2) {\(p_4\)};

\diagram*{
(i1) -- (a1) -- (b1) -- (c1) -- (d1) -- [ghost] (e1) -- (f1) -- (g1) -- (h1),
(i2) -- (a2) -- [momentum=\(k_1\)] (b2) -- [momentum=\(k_2\)] (c2) -- (d2) -- [ghost] (e2) -- (f2) -- [momentum=\(k_L\)] (g2) -- (h2),
(a2) -- (a1),
(b2) -- (b1),
(c2) -- (c1),
(f1) -- (g2),
(f2) -- (g1),
};

\end{feynman}
\end{tikzpicture}
\caption{\label{fig:L-loop-diag}  An example of the  $L$-loop diagram with four external momenta $p_1,\ldots,p_4$ and internal momenta $k_1,\ldots k_L$ which includes planar and non-planar sub-topologies.}
\end{figure}
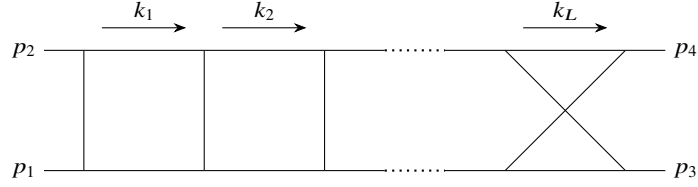
Since the diagram involves loops, momentum conservation does not fix the momenta of all internal lines. In fact, in an $L$-loop diagram, there are precisely $L$ independent internal momenta that are not fixed. The Feynman rules then instruct us to integrate over these loop momenta, thus the amplitude represented by the diagram in Fig.~\ref{fig:L-loop-diag} will have the following general structure,
\begin{equation}
G_L \propto \int d^4 k_1\, d^4 k_2\,\ldots\, d^4 k_L 
	\frac{T(k)}{(q_1^2 - m_1^2 + i \delta) (q_2^2 - m_2^2 + i \delta) \ldots 
	(q_N^2 - m_N^2 + i \delta)}\,.
\label{eq:L-loop-4-dim}
\end{equation}
Here $q_i$ and $m_i$ are the momenta and masses of the internal $N$ lines, while $T(k)$ denotes a generic numerator which itself may depend on the loop momenta $k_1\,,\ldots\, k_L$. Importantly, the range of integration is the complete four-dimensional momentum space for each $k_i$, i.e., we must integrate over each component of all loop momenta from $-\infty$ to $+\infty$.

\begin{tips}{The Role of the Feynman Prescription for Propagators}
Before moving on, let us clear up one point which may be cause for some confusion. In eq.~(\ref{eq:L-loop-4-dim}), we were careful to indicate the Feynman prescription for propagators by including the $+i\delta$ terms. The textbooks tell us that $\delta$ is a small positive quantity and that we must set $\delta \to 0^+$ at the end of the computation. However, if we take $\delta \to 0^+$, then the integrand clearly develops poles for real momenta and masses. Namely, let us consider the on-shell particle with {\it four}-momentum $p$ and a mass $m$. As 
\ba
p^2-m^2 = 0 \quad\Rightarrow\quad (p^0)^2 = \vec{p}^2+m^2
\ea
thus
\ba
p^0 &=& \pm E_p = \pm\sqrt{\vec{p}^2+m^2},
\ea  
and the solutions for the energy of the particle lie on the real axis in the complex $p^0$ plane, see Fig.~\ref{fig:tocomplex}, left. 

\begin{figure}
\centering
\begin{tikzpicture}[scale=0.6]
\begin{feynman}
 
\vertex at (2,1) (i1);
\vertex at (4,1) (i2);
\vertex at (6,1) (i3);
\vertex at (8,1) (i4);
\vertex at (10,1) (i5);
\vertex at (6,3.5) (f1);
\vertex at (6,-1.5) (f2);

\node[right] at (f1) {$\rm{Im}\; p^0$};
\node at (10,0.5) {$\rm{Re}\; p^0$};
\node[above, yshift=2] at (i2) {$-E_k$};
\node[below, yshift=-2] at (i4) {$E_k$};
\tikzset{cross/.style={cross out, draw=black, minimum size=2*(#1-\pgflinewidth), inner sep=0pt, outer sep=0pt},
%default radius will be 1pt. 
cross/.default={1pt}}
 
\path (4,1) pic[rotate = 0] {cross=4pt}; 
\path (8,1) pic[rotate = 0] {cross=4pt}; 
\draw [dashed] (4.5,1) arc[x radius=.5, y radius=.5, start angle=0, end angle=-180];

\draw [dashed] (8.5,1) arc[x radius=.5, y radius=.5, start angle=0, end angle=180];

\diagram*{
 (i1) -- [thin] 
 (i5) --  [thin] 
 (i3) --  [thin] 
 (f1) --  [thin]
 (i3) --  [thin] (f2), 
};

%%%%%
\vertex at (12,1) (j1);
\vertex at (14,1) (j2);
\vertex at (16,1) (j3);
\vertex at (18,1) (j4);
\vertex at (20,1) (j5);
\vertex at (16,3.5) (h1);
\vertex at (16,-1.5) (h2);

\node[right] at (h1) {$\rm{Im}\; p^0$};
\node at (20,0.5) {$\rm{Re}\; p^0$};
\node[above, yshift=2] at (14,1.5) {$-E_k+i \delta$};
\node[below, yshift=-2] at (18,0.5) {$E_k-i \delta$};
\tikzset{cross/.style={cross out, draw=black, minimum size=2*(#1-\pgflinewidth), inner sep=0pt, outer sep=0pt},
%default radius will be 1pt. 
cross/.default={1pt}}
 
\path (14,1.5) pic[rotate = 0] {cross=4pt}; 
\path (18,0.5) pic[rotate = 0] {cross=4pt};

\diagram*{
 (j1) -- [thin] 
 (j5) --  [thin] 
 (j3) --  [thin] 
 (h1) --  [thin]
 (j3) --  [thin] (h2), 
};

\end{feynman}
\end{tikzpicture}
\caption{On-shell poles at the complex $p^0$ plane before (left) and after (right) infinitesimal shifts.}
\label{fig:tocomplex}
\end{figure}
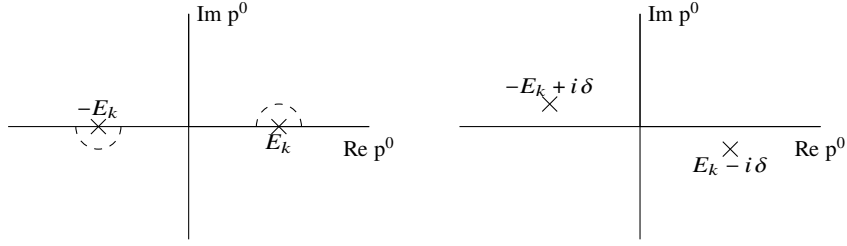

As we showed above on a simple example, one cannot simply integrate over the singularities, thus it seems like we are faced with a problem. However, the Feynman-prescription tells us precisely how we must avoid these poles. In effect, it fixes the correct contour of integration to use in the complex $p_0$ plane. But this prescription can also be implemented in the following way. As the so-called on-shell relativistic relation $p^2-m^2=0$ enters the propagator, we can avoid the singularity when integrating on the real axis by moving the kinetic energy solution $E_p$ for $p^0$ to the complex plane by adding an infinitesimal parameter $\delta$, schematically\footnote{`One of the most remarkable discoveries in elementary particle physics
 has been that of the existence of the complex plane.'
- Julian Schwinger.}    
\begin{equation}
  \frac{1}{p^2-m^2} \longrightarrow \left[ {\rm singularity\; for}\; p^2 \to m^2 \right]
  \longrightarrow \frac{1}{p^2-m^2+ i \delta},
  \label{eq:tocomplex}
\end{equation} 
see Fig.~\ref{fig:tocomplex}, right. The solutions for $p^0$ are then indeed complex, since $p^2-m^2 + i \delta = (p^0)^2-E_p^2+i\delta=0$, which implies $(p^0)^2 = E_p^2-i \delta$, and so $p^0 = \pm\sqrt{E_p-i \delta} = \pm E_p \mp i \delta'$, where $E_p = +\sqrt{\vec{p}^2+m^2}$ and $\delta' = \delta/(2E_p)$. Note that since $E_p \ge m > 0$, $\delta'$ is positive.

Thus, the singularities in the $\delta \to 0^+$ limit are only apparent and the corresponding poles are never integrated over. Indeed, taking this limit in final results for Feynman or phase space integrals never leads to singularities in $\delta$. In this respect, they are quite different form the UV, IR and CO singularities we have been discussing. We will see shortly that UV, IR and CO singularities can be regularized in a way where removing the regularization corresponds to taking the regulator $\ep$ to zero. Then these singularities will show up as poles in $\ep$.
\end{tips}

Unfortunately the integral in eq.~(\ref{eq:L-loop-4-dim}) above is often ill-defined! From mathematical point of view, the heart of the problem which we encounter in calculating the integrals that appear in radiative corrections stems from the fact that the integrands have poles inside or on the boarders of the integration region. This means that there are kinematic configurations where the denominators of the considered amplitudes and \texttt{FI} tend to zero. We will see shortly how one may make sense of these divergent integrals, however let us first show an amusing example of how naive integration can fail spectacularly when there are poles around.

\begin{tips}{Example of a Simple Singular Integral}
As the first, purely mathematical example, let us consider the definite integral 
\begin{equation}
I = \int_{-1}^1 \frac{1}{x^2}dx\,.
\label{int-simplest}
\end{equation}
Elementary integration gives $I=(-\frac{1}{x})|_{-1}^1=-2$, though, the integrand in eq.~(\ref{int-simplest}) is never negative. According to Riemann's definition of definite integrals, the result should be positive. The problem is of course that the region of integration includes $x=0$ where the integrand is infinite and the integrand exhibits {\texttt a singularity}. Such an integral is called improper. As warned in~\cite{nahin2020} {\it ``You have to be ever alert for singularities when you are doing integrals; always, stay away from singularities. Singularities are the black holes of integrals; don't `fall into' one (don't integrate across a singularity).''} 
\end{tips}

We face a similar problem as in our example above for scattering amplitudes. Namely, let us consider the simple one-loop diagrams from the left of Fig.~\ref{fig:1-loop-UV-IR}. 
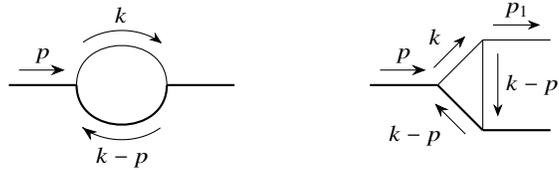
\begin{figure}\centering

\begin{tikzpicture}[scale=0.6]
\begin{feynman}
 
%tadpole 

\vertex at (0.5,1) (i1);
\vertex at (2,1) (a);
\vertex at (4,1) (b);
\vertex at (5.5,1) (f1);

\diagram*{
(i1) -- [thick,momentum=\(p\),/tikzfeynman/momentum/arrow distance=2mm] (a) -- [half left, momentum=\(k\),/tikzfeynman/momentum/arrow shorten=0.25,/tikzfeynman/momentum/arrow distance=2mm] (b) -- [thick,half left, momentum=\(k-p\),/tikzfeynman/momentum/arrow shorten=0.25,/tikzfeynman/momentum/arrow distance=2mm] (a),
(b) -- [thick] (f1),
};

\vertex at (8.5,1) (i1);
\vertex at (10,1) (a);
\vertex at (11,2) (b);
\vertex at (11,0) (c);
\vertex at (12.5,2) (f1);
\vertex at (12.5,0) (f2);

\diagram*{
(i1) -- [thick,momentum=\(p\),/tikzfeynman/momentum/arrow distance=2mm] (a) -- [momentum=\(k\),/tikzfeynman/momentum/arrow distance=2mm] (b) -- [momentum=\(p_1\),/tikzfeynman/momentum/arrow distance=2mm] (f1),
(f2) -- [thick] (c) -- [thick,momentum=\(k-p\),/tikzfeynman/momentum/arrow distance=2mm] (a),
(b) -- [momentum=\(k-p_1\),/tikzfeynman/momentum/arrow distance=2mm] (c)
};

\end{feynman}
\end{tikzpicture}

\caption{\label{fig:1-loop-UV-IR} One-loop diagrams exhibiting UV (left) and IR (right) singularities. The thick line denotes a particle of non-zero mass $m$ while the thin line represents a massless particle.}
\end{figure}
This leads to the integral ({after the earlier note on the Feynman prescription of the propagators, if not necessary,} we {will} not make the $+i\delta$ terms explicit in the propagators, notation using capital roman letters will be explained in section~\ref{sec:dimreg} where general integrals deviated from 4-dimensions will be discussed)
\begin{equation}
A_0 \propto \int \frac{d^4 k}{k^2}\,
\end{equation}
and
\begin{equation}
B_0 \propto \int \frac{d^4 k}{k^2 [(k-p)^2-m^2]}\,.
\end{equation}
Remember that the momenta of virtual particles in loops do not have to satisfy the mass-shell conditions (put differently, the loop integration over the four-momentum $k$ is fully unconstrained), and so $[(k-p)^2-m^2] = k^2 - 2k\cdot p$. (Notice that $p$ is an external momentum, so the mass-shell condition is satisfied for $p$ and thus $p^2 = m^2$.) Then for very large momenta $k \to \infty$ we find that our integrals behaves in the following way
\begin{equation}
A_0 \propto \int \frac{d^4 k}{k^2  } \xrightarrow{k \to \infty}
	\int \frac{d^4k}{k^2}
	\xrightarrow{\mathrm{Wick}}
	\int_0^\infty \frac{|K|^3 d|K| d\Omega}{|K|^2} \propto \int_0^\infty {|K| d|K|}{}\,,
\end{equation}
\begin{equation}
B_0 \propto \int \frac{d^4 k}{k^2 (k^2 - 2k\cdot p)} \xrightarrow{k \to \infty}
	\int \frac{d^4k}{(k^2)^2}
	\xrightarrow{\mathrm{Wick}}
	\int_0^\infty \frac{|K|^3 d|K| d\Omega}{|K|^4} \propto \int_0^\infty \frac{d|K|}{|K|}\,.
\end{equation}
In the second step, we have performed a Wick rotation, $k^0 \to i K^0$, $k^i \to K^i$, ($i=1,2,3$) in order to pass to {Euclidean} space in the loop momentum, so $|K|$ is the length of the Wick-rotated, {Euclidean} four-momentum $K$. Then, we passed to four-dimensional spherical coordinates and $d\Omega$ is the angular measure in four dimensions. In the last step, we have dropped the trivial angular integral that does not influence the behavior of the result at $k\to \infty$. Clearly $A_0$ and $B_0$ are UV divergent,  quadratically and    logarithmically singular at the upper limit of integration, respectively.  Similar asymptotic dependencies of integrals on the loop momentum behavior will be observed for analogous d-dimensional integrals in section~\ref{sec:dimreg}.

Turning to the one-loop diagram on the right of Fig.~\ref{fig:1-loop-UV-IR}, we find that it corresponds to the integral
\begin{equation}
C_0 = \int \frac{d^4 k}{k^2 (k-p_1)^2[(k-p)^2 - m^2]}\,.
\end{equation}
An analysis similar to the one above shows that this integral is not UV divergent. However, considering now the limit of very small momenta $k \to 0$, we find that
\begin{equation}
\begin{split}
C_0 &= \int \frac{d^4 k}{k^2 (k^2 - 2k\cdot p_1)(k^2 - 2k\cdot p)} \xrightarrow{k \to 0} 
	\int \frac{d^4 k}{k^2 (2k\cdot p_1)(2k\cdot p)}
\\&	\xrightarrow{\mathrm{Wick}}
	\int \frac{|K|^3 d|K| d\Omega}{|K^2|(2K \cdot p_1)(2K \cdot p)}
	\propto \int_0^\infty \frac{d|K|}{|K|}\,.
\end{split}
\end{equation}
As previously, we performed Wick-rotation and passed to four dimensional spherical coordinates. In the last step, we have again dropped an angular integral that does not play a role in analyzing the small $k$ behavior of the integral. Now we see that our integral is logarithmically singular at the lower limit of integration and so the integral is IR divergent. However, the integrand has yet another type of singularity: a collinear singularity when the internal momentum becomes collinear to the momentum of the neighboring massless leg. Indeed, let us parametrize the loop momentum $k$ as $k=x p_1 + k_\perp$, where $k_\perp$ points in the transverse direction with respect to $p_1$ so that $k_\perp\cdot p_1=0$. The collinear limit is reached as $k_\perp \to 0$. Then we find (recall $p_1^2=k_\perp\cdot p_1=0$)
\begin{equation}
\begin{split}
C_0 &= \int \frac{d^4 k}{k^2 (k^2 - 2k\cdot p_1)(k^2 - 2k\cdot p)}
	= \int \frac{d^4 k_\perp}
	{k_\perp^2 (k_\perp^2)[k_\perp^2 - 2(x p_1 + k_\perp)\cdot p]}
\\& \xrightarrow{k_\perp \to 0} 
	\int \frac{d^4 k_\perp}{k_\perp^2 (2k_\perp^2)(-2x p_1\cdot p)}
	\propto
	\int \frac{d^4k_\perp}{k_\perp^4} 
	\xrightarrow{\mathrm{Wick}}
	\int_0^\infty \frac{d|K_\perp|}{|K_\perp|}\,.
\end{split}
\end{equation}
In the last step we have performed the usual Wick rotation. As before, we find a logarithmic singularity for small $k_\perp$, i.e., in the collinear limit. We will see shortly how one can make sense out of the divergent integrals that we have encountered. 

We note in passing that special treatment is required also for so-called {\it{threshold}} effects connected with conspired relations among masses of virtual states~\cite{Landau:1959fi,Nakanishi:1961}. We will use thresholds for some specific constructions of \mb{} representations and their numerical evaluation in chapter~\ref{sec:5MBrepr}. 

Next, we note that real emission integrals involving massless particles can also be ill-defined due to the presence of poles in the integrand. To see this, let us consider the emission of a real massless particle in some particle scattering process, see Fig.~\ref{fig:real-emission}. We assume that the emitted particle, denoted by a gluon in the figure, is massless, i.e., $k^2=0$.
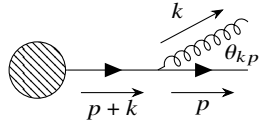
\begin{figure}\centering

\begin{tikzpicture}[scale=0.8]
\begin{feynman}

\vertex at (0,0) (m);
\vertex at (0.5,0) (ap);
\vertex at (2,0) (a);
\vertex at (3.5,0) (f1);
\vertex at (3.5,0.8) (f2);
\vertex at (3.4,0.25) (t) {\(\theta_{kp}\)};

\diagram*{
m [blob] -- (ap) -- [fermion, momentum'=\(p+k\)] (a) -- [fermion, momentum'=\(p\)] (f1),
(f2) -- [gluon, reversed momentum'=\(k\)] (a),
(t),
};

\end{feynman}
\end{tikzpicture}

\caption{\label{fig:real-emission} Schematic amplitude for single real emission.}
\end{figure}
The solid line emerging from the rest of the diagram (denoted by the blob in the figure) represents some particle's propagator with momentum $p+k$ and mass $m$, thus the complete amplitude will contain the scalar part of this propagator (i.e., its denominator without the Lorentz structure in the propagator's numerator),
\begin{equation}
\mathcal{M} \propto \frac{1}{(p+k)^2-m^2}=\frac{1}{2 p\cdot k} = \frac{1}{E_k E_p} \frac{1}{1-\beta_p \cos{\theta_{kp}}},
\label{eq:propag}
\end{equation}
where $p = (E_p, \vec{p})$ and $k = (E_k, \vec{k})$ and we set $\beta_p = \sqrt{1-\frac{m^2}{E_p^2}}$, see Problem~\ref{prob:propag}. Notice that $E_p \ge m > 0$, but since $k$ is massless, $E_k$ can go to zero. From the above we see that the amplitude has an IR divergence when $E_k \to 0$ and a collinear (CO) divergence when $\theta \to 0$ and in addition $\beta_p = 0$, i.e., the internal particle is also massless:
\begin{enumerate}[label=(\roman*)]
\item\parbox{\widthof{$\theta \to 0 \quad\mbox{and}\quad \beta_p = 0\quad$}}{$E_k \to 0$} IR,
\item\parbox{\widthof{$\theta \to 0 \quad\mbox{and}\quad \beta_p = 0\quad$}}{$\theta \to 0 \quad\mbox{and}\quad \beta_p = 0$} CO.
\end{enumerate}
Thus, {\it{infrared singularities are associated with vanishing particle energies, while collinear singularities arise when the daughter particles are emitted in the same direction.}} We note in passing that if the internal particle is not massless, the $\theta \to 0$ collinear limit is not singular, however it produces what is called a ``mass singularity''. I.e., the integration over the angle of emission is finite, but the result becomes singular if the internal mass is taken to zero.

We should warn the reader that the proper analysis of real emission singularities is more involved than the above simple calculation suggests. First of all, when we compute a physical quantity, what enters the real emission integrals is actually the squared amplitude, not the amplitude itself. Second, the integration measure also plays a role in the analysis of singularities. Indeed, the quantum mechanical rules for obtaining a cross section from a squared amplitude state that we must integrate over the momenta of final-state particles with the following phase space measure,
\begin{equation}
d\phi_n(p_1,\ldots,p_n;Q) = 
	\prod_{k=1}^{n} \frac{d^4p_k}{(2\pi)^3} \delta_{+}(p_k^2-m_k^2)(2\pi)^4
	\delta^4(p_1 + \ldots + p_n - Q)\,.
\end{equation}
Here $Q$ is the total incoming momentum, while the the $\delta_+$ function is defined as
\begin{equation}
\delta_{+}(p_k^2-m_k^2) = \delta(p_k^2-m_k^2) \Theta(p_k^0)\,.
\end{equation}
A trivial way in which the measure enters the analysis of singularities is as follows. Since overall momentum is conserved, phase space is compact for any finite incoming momentum $Q$. This implies that we cannot have $p_k \to \infty$. This is very clear physically: the energy of any outgoing particle is clearly bounded from above. Hence, {\it{real emission integrals do not have UV singularities}}. A more subtle way in which the measure enters the analysis can be seen in the soft limit. Let us suppose that momentum $k$ is massless. Then the corresponding one-particle measure becomes
\begin{equation}
\frac{d^4 k}{(2\pi)^3}\delta_{+}(k^2) = \frac{d^3 \vec{k}}{(2\pi)^3 2E_k}
	= \frac{|\vec{k}|^2 d|\vec{k}|\,d(\cos\theta)\, d\phi}{(2\pi)^3 2E_k} 
	= \frac{E_k dE_k\,d(\cos\theta)\, d\phi}{2(2\pi)^3}\,.
\end{equation}
In the first step we performed the integration over the energy component of $k$ using the $\delta_+$ function, then we passed the three-dimensional spherical coordinates for the integration over the space-like components. Finally, we used $|k| = E_k$, which is a consequence of $k^2=0$ and the fact that the energy is non-negative. Now, considering the soft limit, i.e., $E_k \to 0$, we see that the integration measure carries one factor of $E_k$, so the integrand must diverge at least like $E_k^{-2}$ in order to give a non-integrable singularity. This is precisely the case for the square of the real emission amplitude (recall the amplitude itself behaves as $E_k^{-1}$ in the soft limit) and hence the phase space integration over the soft momentum produces an ill-defined integral. E.g. if the gluon in Fig.~\ref{fig:real-emission} can be emitted from more than one leg, say those with momenta $p_1$ and $p_2$, then the matrix element will involve a sum of Feynman-diagrams and hence a sum of terms like those in eq.~(\ref{eq:propag}),
\begin{equation}
\mathcal{M} \propto \ldots + \frac{1}{2p_1\cdot k} + \frac{1}{2p_2\cdot k} + \ldots\,.
\end{equation}
Then the squared matrix element, $|\mathcal{M}|^2$ will contain terms such as
\begin{equation}
|\mathcal{M}|^2 \propto \ldots + \frac{1}{(2p_1\cdot k)(2p_2\cdot k)} + \ldots\,,
\end{equation}
which lead to integrals of the form
\begin{equation}
\begin{split}
&\int d\phi_n(p_1,\ldots,p_n;Q) \frac{1}{(2p_1\cdot k)(2p_2\cdot k)} = 
\\ & \hspace*{-.3cm} \quad =
	\int \frac{E_k dE_k\,d(\cos\theta)\, d\phi }{2(2\pi)^3} \frac{1}{E_k^2 E_{1} E_{2}} 
	\frac{1}{(1-\beta_1 \cos\theta_{k1})(1-\beta_1 \cos\theta_{k2})} 
	\times\ldots
\\ & \hspace*{-.3cm} \quad =
	\frac{1}{2(2\pi)^3 }\int_0^{Q} \frac{dE_k}{E_k} \frac{1}{E_{1} E_{2}} 
	\int d(\cos\theta)\, d\phi \frac{1}{(1-\beta_1 \cos\theta_{k1})(1-\beta_1 \cos\theta_{k2})} 
	\times\ldots\,,
\end{split}
\end{equation}
where the $\ldots$ denote the rest of the measure and integrand. The integration over $E_k$ is indeed singular on the lower limit signaling the IR divergence. We see also that the evaluation of phase space integrals can involve integrations over angular variables, a topic we will return to in section~\ref{sec:4MBrepr}.

\begin{svgraybox}We have described the basic structures and integrals in the computation of higher-order perturbative corrections in particle physics: loop integrals and phase space integrals. We have seen that both are ill-defined in four spacetime dimensions due to UV, IR, and CO divergences. While UV can be removed by passing to renormalized quantities, IR and CO divergences cancel only after summing virtual and real corrections to some given process. 
\end{svgraybox}
Note that there are technical requirements on the physical observable we want to compute in order for this cancellation to work and this is the content of the Kinoshita--Lee--Nauenberg theorem~\cite{Kinoshita:1962ur,Lee:1964is}. Quantities for which the cancellation takes place are called IR and collinear-safe observables. Observables that are not IR and collinear-safe may still be defined and measured, however they cannot be computed order-by-order in perturbation theory. For example, the definition of jets using so-called `seeded' cone algorithms is not IR and collinear-safe beyond a certain perturbative order~\cite{Salam:2007xv}. Nevertheless, they have been extensively used in past hadron collider experiments.

However, before any specific calculation can be performed, we must first give a precise meaning to the divergent integrals we encountered above. This procedure is called regularization and will be the topic of the next section.

\section{Dimensional Regularization and Evaluation of Feynman Integrals 
\label{sec:dimreg}
}

As we have just seen, in QFT we can disentangle three types of divergencies depending on the masses of the parent and emitted particles in a given interaction vertex or the whole scattering amplitude: UV, IR and CO. The presence of these singularities makes the naively defined amplitudes and cross sections ill-defined. Thus, our first task is to give a precise meaning to the various objects and integrals we have encountered so far.

There are several ways to modify our integrals such that the results become well-defined. This process is called regularization. It is important that in physical results we should be able to remove the regularization and obtain unambiguous predictions. One way of regularizing our integrals is through dimensional regularization. This method has proven to be very powerful in higher order perturbative calculations so much so that it is ubiquitous in the modern literature. The idea is to formally modify the dimensionality of spacetime from $d=4$ to $d=4-2 \ep$, where $\ep$ is an infinitesimal spacetime regulator\footnote{The factor of 2 appears because of practical reasons: many results take a simpler form with this choice, see e.g. eqs.~(\ref{eq:omegad}), (\ref{eq:tadpoleddim}) and (\ref{eq:tadpoleeps}) below.}~\cite{Smirnov:2004,Weinzierl:2006qs,Blondel:2018mad} though recently methods of direct integration in $d=4$ spacetime are developing~\cite{Gnendiger:2017pys,Weinzierl:2022eaz}. The \mb{} method has been used mostly in conjunction with the former, dimensional regularization approach and we will employ dimensional regularization throughout. In spite of that, it is not our intention here to give a comprehensive introduction to dimensional regularization and we refer the reader to the many excellent QFT textbooks that discuss this topic in detail. However, the basic idea can be demonstrated on very simple integrals, which we discuss next.

\begin{tips}{Regularization: an Elementary Example}
Consider the very simple integral
\begin{equation}
I(\alpha) = \int_0^1 \frac{1}{x^{\alpha}} dx\,.
\end{equation}
Obviously this integral is convergent for $\alpha < 1$ and can be evaluated in an elementary way, $I(\alpha) = \left.\frac{x^{1-\alpha}}{1-\alpha}\right|_0^1 = \frac{1}{1-\alpha}$. However for $\alpha=1$, the integral is logarithmically singular on the lower limit of integration and indeed the $\alpha \to 1$ limit of the result is ill-defined. In order to give meaning to the integral even for $\alpha=1$, we can add a regulator to the measure in the following way. We set $dx \to x^{-\ep} dx$ and define the regularized integral
\begin{equation}
I_{\mathrm{reg}}(\alpha;\ep) = \int_0^1 \frac{1}{x^{\alpha}} x^{-\ep} dx\,.
\end{equation}
The original integral is recovered for $\ep \to 0$. The integration is of course still elementary and we find $I_{\mathrm{reg}}(\alpha;\ep) = \left.\frac{x^{1-\alpha-\ep}}{1-\alpha-\ep}\right|_0^1 = \frac{1}{1-\alpha-\ep}$. Now, notice the following. On the one hand, when $\alpha <1$, i.e., if the original integral is convergent, we may simply take the $\ep \to 0$ limit in the regularized result and recover the original one: $I_{\mathrm{reg}}(\alpha;\ep=0)=I(\alpha)$. On the other hand, for $\alpha \to 1$, we obtain $I_{\mathrm{reg}}(\alpha=1;\ep) = \frac{1}{-\ep}$. Hence, the regularized integral is well-defined for $\alpha=1$ and moreover, whenever the original integral is convergent, we recover the correct result after removing the regularization.

Perhaps it is worth emphasizing that regularization did not magically make our divergent integral finite, since we recover the singularity in the form of a pole in $\ep$ if we try to remove the regulator in $I_{\mathrm{reg}}(\alpha=1;\ep)$. However, we now have a well-defined expression that we can associate to the divergent integral $I(\alpha=1)$. The point is that in a real computation in QFT perturbation theory, although divergent integrals show up at intermediate stages of the calculation, the final physical results are of course finite. Thus, regularization is a tool used to make all intermediate results well-defined in a consistent way. The cancellation of $\ep$ poles in physical results is in fact a strong check on the correctness of the calculations.
\end{tips}

Let us now turn to properly defining the dimensionally regularized Feynman integrals whose evaluation via the \mb{} method will be the focus of this book. The $L$-loop Feynman integral in $d=4-2\ep$ dimensions with $N$ internal lines of masses $m_i$ carrying momenta $q_i$ and $E$ external legs with momenta $p_e$ can be written in the following way:
\begin{equation}
\boxed{
G_L[T(k)] =
	\frac{e^{\ep\gamma L} (2\pi\mu)^{(4-d)L}}{(i\pi^{d/2})^L} 
	\int \frac{d^dk_1 \ldots d^dk_L\,T(k)}
     {(q_1^2-m_1^2)^{\nu_1} \ldots 
      (q_N^2-m_N^2)^{\nu_N}  }\,. 
}
\label{eq-bha}
\end{equation}
Notice in particular the $d$-dimensional measures $d^d k_i$: the integrals are now defined with all $L$ loop momenta living in $d$-dimensional momentum space.\footnote{External momenta can be defined to be either 4 or $d$-dimensional and this choice leads to different dimensional regularization schemes such as conventional dimensional regularization (CDR), 'tHooft-Veltman (HV) scheme, dimensional reduction (DR), etc. See e.g.~\cite{Weinzierl:2022eaz} for an overview.} The integration region is all of this $L\times d$-dimensional space. 
The numerator $T(k)$ is in general a tensor in the integration variables:
\begin{equation}\label{eq-T}
T(k) = 1\,, k_l^{\mu}\,, k_l^{\mu}k_n^{\nu}\,, \ldots\,.
\end{equation}
Integrals with $T(k)=1$ are referred to as basic scalar integrals. Internal four-momenta $k_i$ are connected with virtual particles propagating inside closed loops and each closed loop $i$ corresponds to one $k_i$, as shown schematically on Fig.~\ref{fig:L-loop-diag}. Each virtual particle of mass $m_i$ carries the four-momentum $q_i$, which are in general functions of both internal $k_i$ and external four-momenta $p_i$, see e.g. the right hand diagram in Fig.~\ref{fig:1-loop-UV-IR}. 
The powers $\nu_i$ in denominator are integer numbers which are not necessarily equal to one. For instance, in the so-called Integration-By-Parts (IBPs) procedure~\cite{Chetyrkin:1981qh} differentiation of the \texttt{FI} with respect to internal momenta leads to relations among integrals with different powers of scalar propagators, leading to a reduction of the \texttt{FI} to a small number of so-called master integrals (\texttt{MIs}). Hence, it is generally useful to consider arbitrary propagator powers as we have written. The overall pre-factor is chosen in a way to simplify the final expressions as we will see shortly. For example the reason for factoring $e^{\ep\gamma L}$ in eq.~(\ref{eq-bha}) becomes clear with a discussion of the expansion of gamma functions in $\ep$, see eq.~(\ref{eq:epsEuler}). Second, the parameter $\mu$ has dimensions of mass so including the factor of $(2\pi\mu)^{(4-d)L}$ restores the overall mass dimension of the integral to its original 4 dimensional value from the modified $d$-dimensional one. Finally, the factor ${(i\pi^{d/2})^L}$ is the integration measure, it could well be taken as  ${(2\pi)^{d L}}$, as discussed in~\cite{Weinzierl:2022eaz}.
%{\bf STILL TO FIX} 
The integrals in $G_L[T(k)]$ of eq.~(\ref{eq-bha}) over the internal momenta $k_i$ in $d$-dimensional space can be parametrized in  many ways and we will examine this in detail in section~\ref{sec:FIrepr}. There we will show an example how it can be evaluated with Feynman parametrization. In chapter~\ref{chapter-MBrepr} we will give examples with corresponding evaluations by \mb{} approach. 

%%%%%%%%%%=====%%%%%%%%%
However, let us first present a simple example of how such loop integrals can be computed directly. 
The most basic integrals in loop calculations arise from 1-loop diagrams. Fig.~\ref{fig:1loopmasters} shows some basic examples: a so-called tadpole\footnote{The name tadpole is due to Sydney Coleman. 
They constitute vacuum effects in one space-time point and do have physical meaning if the corresponding external field has a non-zero vacuum expectation value, as is the case e.g. for the Higgs field. Tadpole diagrams have been used for the first time in~\cite{Coleman:1963pj} to explain symmetry breaking in the strong interaction. They are also considered in the context of the free energy density or pressure of the QCD plasma, finite temperature effective potentials and phase transitions, see for instance~\cite{Boyd:1993tz,Luthe:2016spi}.}, bubble and triangle.

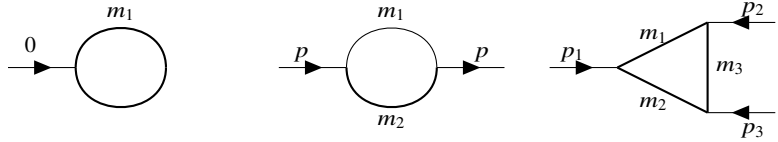
\begin{figure}[h!]
%\sidecaption
\centering

\begin{tikzpicture}[scale=0.6]
\begin{feynman}
 
%tadpole 
 
\vertex at (0.5,1) (i1);
\vertex at (2,1) (a); 
\vertex at (4,1) (b);
\vertex at (5.5,1) (f1);

\node at (1,1.5) {$0$};
\node at (3,2.2) {$m_1$};

\node at (7,1.3) {$p$};
\node at (9,2.2) {$m_1$};
\node at (9,-0.2) {$m_2$};
\node at (11,1.3) {$p$};

\diagram*{
(i1) -- [fermion] 
(a) --   [thick, half left] 
(b) -- [thick, half left] (a)
%(b) -- [thick] (f1),
};

%SE
\vertex at (6.5,1) (i1);
\vertex at (8,1) (a);
\vertex at (10,1) (b);
\vertex at (11.5,1) (f1);

\diagram*{
(i1) -- [fermion] 
(a) -- [thick, half left] 
(b) -- [thick, half left] (a),
(b) -- [fermion] (f1),
};

%vertex
\vertex at (12.5,1) (i1);
\vertex at (14,1) (a);
\vertex at (16,2) (b);
\vertex at (16,0) (c);
\vertex at (17.5,2) (f1);
\vertex at (17.5,0) (f2);

\node at (13,1.3) {$p_1$};
\node at (14.8,1.7) {$m_1$};
\node at (14.8,0.2) {$m_2$};
\node at (16.5,1) {$m_3$};
\node at (17,2.3) {$p_2$};
\node at (17,-0.4) {$p_3$};

\diagram*{
(i1) -- [fermion] (a) -- [thick] (b), (f1) -- [fermion] (b),
(f2) -- [fermion] (c) -- [thick] (a),
(b) -- [thick] (c)
};

\end{feynman}
\end{tikzpicture}

\caption{Some 1-loop Feynman diagrams with external momenta $p_i$ and internal masses $m_i$ which can be written using the basic integral eq.~(\ref{eq:ddimgeneral}).
}
\label{fig:1loopmasters}       
\end{figure}

\noindent
Omitting here the factor $e^{\gamma \ep}$ given in eq.~(\ref{eq-bha}) for a moment, let us define
\begin{equation}
I_{a}^b(M^2) \equiv 
	\frac{(2\pi\mu)^{4-d}}{i\pi^2}
	\int d^d Q \frac{(Q^2)^b}{\left(Q^2-M^2+i\delta\right)^{a}}.
\label{eq:ddimgeneral}
\end{equation}
For $a=1$ and $b=0$ we get the simplest integral, the tadpole
\begin{equation}
A_0(m^2) = \mathrm{T1l1m} \equiv I_1^0\left(m^2\right)\,. 
\label{eq:specA0}
\end{equation}
The notation by symbols $A\,, B\,, C\,,\ldots$ for tadpoles, bubbles, triangles, \ldots is due to Passarino and Veltman~\cite{Passarino:1978jh}. We also call these functions one-point (1-PF), two-point (2-PF), three-point (3-PF) functions and so on, indicating the number of legs which represent the scattered objects with external four-momenta $p_1\,, p_2\,, p_3\,,\ldots$. Returning to $I_{a}^b$, integrals with indices $a$ and $b$ greater than one appear during \texttt{IBP} reduction~\cite{Chetyrkin:1981qh} and also in derivatives of more complicated multi-point diagrams with respect to external kinematics, such as the bubble and triangle shown in Fig.~\ref{fig:1loopmasters}. They also appear as building blocks of higher-point integrals, e.g. for $B_0$ and $C_0$ we have 
(with some appropriate function $M^2$ of $x$, $y$, $p_i^2$ and $m_i^2$)
\ba
B_0(p^2;m_1^2,m_2^2) &=& \int_0^1 dx\; I_2^0\left(M^2\right), \label{eq:specB0}
\\
C_0(p_2^2,p_3^2,p_1^2;m_1^2,m_2^2,0) &=&2 \int_0^1 dx \int_0^x dy \; I_3^0\left(M^2\right).
\label{eq:specC0}
\ea 
That the integrals $B_0$ and $C_0$ can be written in terms of eq.~(\ref{eq:ddimgeneral}) is the subject of Problem~\ref{prob:Ftrick}.

Having seen that the integral in eq.~(\ref{eq:ddimgeneral}) is the basic 1-loop object, let us compute it. The calculation proceeds in two steps. First, we change the integral to the Euclidean space by performing the Wick-rotation $Q^0\to iq^0$. 
This rotation of the integration contour is allowed because the $+i\delta$ prescription for the propagators ensures that no pole is crossed. The integral then becomes
\ba
I_{a}^b(M^2) &=& 
(-1)^{b-a}\frac{(2\pi\mu)^{4-d}}{\pi^2}
\int d^d q \frac{(q^2)^b}{\left(q^2+M^2\right)^{a}}\,.
\label{eq:ddimgeneralwick}\ea
The denominator in eq.~(\ref{eq:ddimgeneralwick}) is not zero (unless $M=0$ and $q^2=0$), so the integrand has no singularity on the integration path. Second, we turn to the $d$-dimensional measure $d^d q$. The integrand in eq.~(\ref{eq:ddimgeneralwick}) is spherically symmetric (independent of angular coordinates) and therefore we can separate out the angular part of the integral by going to spherical coordinates
\bea
d^d q = q^{d-1} dq d\Omega,
\eea
where, in $d$ dimensions, we can write 
\ba
d\Omega = d\theta_1 \sin^{d-2}\theta_1 d\theta_2 \sin^{d-3}\theta_2 \ldots d\theta_{d-1}, \;\;\; d \geq 2. \label{eq:angles}
\ea
We will consider this measure in detail in section~\ref{sec:4MBrepr} where real radiation integrals are discussed. The integral over the angles in eq.~(\ref{eq:angles}) can be evaluated analytically for general $d$ and we find (see Problem~\ref{prob:dangles}),
\ba
\int d\Omega = \frac{2\pi^{\frac{d}{2}}}{{\Gamma(\frac d2)}}.
\label{eq:omegad}
\ea
Here the gamma function appears for the first time\footnote{The gamma function also appears in many basic formulae and many different areas of physics like statistical physics or in derivation of the Casimir effect. For more physical cases, see~\cite{de2019solved}.}. We will discuss its definition and properties in detail in chapter~\ref{chapter:complex}. We will also need an extension of this function to the complex plane, which will also be done in the next chapter.

Using the result for the angular integral, we can continue our evaluation of $I_{a}^b(M^2)$,
\begin{eqnarray}
I_{a}^b(M^2) &=& 
(-1)^{b-a}\frac{(2\pi\mu)^{4-d}}{\pi^2}
\frac{2 \pi^{\frac{d}{2}}}{\Gamma\big(\frac{d}{2}\big)}
\int_0^\infty d\left( \frac q M \right) \left(\frac qM \right)^{d-1} M^d \frac{M^{2b}}{M^{2a}} \frac{\left(\frac{q^2}{M^2}\right)^b}{\left(1+\frac{q^2}{M^2}\right)^{a}} \nonumber \\
&=&(-1)^{b-a}\frac{(2\pi\mu)^{4-d}}{\pi^2} 
\frac{ \pi^{\frac{d}{2}}}{\Gamma\big(\frac{d}{2}\big)}\; 2 (M^2)^{\frac{d}{2}+b-a}\int_0^\infty dt\; t^{d+2b-1}(1+t^2)^{-a},
\nonumber \\ 
\label{eq:ddimres}
\end{eqnarray}
where $t=q/M$. The last one-dimensional integral can be performed in terms of the so-called beta function. By definition
\begin{equation}
\label{eq-betadef}
B(\alpha,\beta) 
	= \frac{\Gamma(\alpha)\Gamma(\beta)}{\Gamma(\alpha+\beta)} 
	= \int_0^\infty dx\, x^{\alpha-1}(1+x)^{-\alpha-\beta} \,,\quad \Re(\alpha)\,, \Re(\beta) > 0\,,
\end{equation}
thus after setting $t^2 \to x$ we find that the last integral reduces simply to a beta function with
\begin{equation}
\alpha=\frac d2 + b\,,\quad\mbox{and}\quad \beta=a-\alpha=a-b-\frac d2\,,  
\end{equation}
and we find the result
\begin{equation}
\boxed{
I_{a}^b(M^2) = (-1)^{b-a}\frac{(2\pi\mu)^{4-d}}{\pi^2} \frac{ \pi^{\frac{d}{2}}}{\Gamma\big(\frac{d}{2}\big)} \frac{1}{(M^2)^{a-b-\frac d2}}\;\frac{\Gamma\big(\frac d2+b\big)\Gamma\big(a-b-\frac d2\big)}{\Gamma(a)}. 
}
\label{eq:tadpoleddim}
\end{equation}

Specializing to the case of $a=1$ and $b=0$ for the tadpole, from the general formula in eq.~(\ref{eq:tadpoleddim}) we get\footnote{This coincides with definitions of the one-loop functions which are commonly used in public packages, e.g. {\tt LoopTools}~\cite{Hahn:1998yk}.}
\begin{eqnarray}
A_0(m^2) &=&  
m^2~\left( \frac{4\pi\mu^2}{m^2}\right)^{\ep}
~\frac{\Gamma(\ep)}{ 1-\ep} 
\nn \\
&=& m^2~\left[
 \frac{1}{\ep} - \gamma + \ln(4\pi)
+ 1-\ln\frac{m^2}{\mu^2} \right] + {\cal O}(\ep).
\label{eq:tadpoleeps}
\end{eqnarray}
Note that in eq.~(\ref{eq:tadpoleeps}) the Euler--Mascheroni constant $\gamma$ (see chapter~\ref{chapter:complex}) and $\ln(4\pi)$ are present. Adding the factor $e^{\gamma \ep}$, as in the original definition of Feynman integrals in eq.~(\ref{eq-bha}) we find 
\ba
{\bar A}_0(m^2) 
&=& \frac{(2\pi\mu)^{2\ep}}{i\pi^2}
~ e^{\gamma\ep}~  \int \frac{d^d Q}{Q^2-m^2+i\delta}
~=~ m^2~\left(\frac{4\pi\mu^2}{m^2}\right)^{\ep} ~ e^{\gamma\ep}
~ \frac{\Gamma(\ep)}{ 1-\ep} 
\nn \\
& = & \frac{1}{\ep} + 1 + \left(1+ \frac{1}{2}\zeta_2\right)
\ep
+  \left(1 + \frac{1}{2}\zeta_2 - \frac{1}{3}\zeta_3\right)\ep^2
+\ldots
\ea
where we have set in addition
\bea
4\pi\mu^2 &=&1,
\\
m^2&=&1.
\eea
We see that the $\gamma$ can be absorbed by the appropriate pre-factor. Since
\begin{equation}
\Gamma(\ep) =
\frac{1}{\ep} -\gamma + \frac{1}{2}\left[ \gamma^2+\zeta_2\right] \ep + \frac{1}{6}\left[ -\gamma^3-3\gamma\zeta_2-2\zeta_3\right] \ep^2 +\cdots
\end{equation}
we have that
\begin{equation}
e^{\ep \gamma}\,\Gamma(\ep) =
\frac{1}{\ep} +  \frac{1}{2}\zeta_2\ep - \frac{1}{3}\zeta_3\ep^2 +\cdots
\label{eq:epsEuler}
\end{equation}
%%%%%%%%%%=====%%%%%%%%%

As already indicated in eq.~(\ref{eq:specA0}), in this book we will use also another notation for FI based on the number of massive and massless lines as done for example in~\cite{Czakon:2004wm} where in addition T stands for tadpoles, SE for self-energies, V for vertices and B for boxes, so 
\ba
{\bar A}_0(m^2=1) \equiv \texttt{T1l1m} = \frac{1}{\ep} + 1
+ \left(1+ \frac{\zeta_2}{2}\right)  \ep 
+ \left(1+ \frac{\zeta_2}{2} - \frac{\zeta_3}{3}\right)  \ep^2 
+    \ldots
\label{eq:T1l1m}
\ea  

As discussed in~\cite{'tHooft:1978xw}, one-loop Feynman integrals of higher multiplicity (i.e. more than five legs) and higher ranks (numerators) 
can be expressed in terms of a basis which includes 1-point functions -- tadpoles, 2-point functions -- self-energies, 3-point functions -- vertices and 4-point functions -- box diagrams.
 
In addition, QED integrals at the one-loop level lead to Euler dilogarithms discussed in chapter~\ref{chapter:complex} or simpler functions, for electroweak integrals we meet more general hypergeometric functions, see e.g.~\cite{Fleischer:2006ht} and section~\ref{sec:gamma_hyperg}.  
{The general one-loop integrals were tackled systematically since 1990s; see e.g.~\cite{Tarasov:1996br}.
In~\cite{Tarasov:2000sf,Fleischer:2003rm}, the class of generalized hypergeometric functions for massive one-loop Feynman integrals with unit indices ($\nu_1, \ldots \nu_N=1$ in eq.~(\ref{eq-bha})) were determined and studied with a novel approach based on dimensional difference equations:
\begin{enumerate}[label=(\roman*)]
\item
{ $_2F_1$ Gauss} hypergeometric functions are needed for self-energies;
\item
{ $F_1$ Appell} functions are needed for vertices;
\item
{ $F_S$ Lauricella-Saran} functions are needed for boxes.
\end{enumerate}
Finally, general massive one-loop one- to four-point functions with unit indices at arbitrary kinematics were determined  in~\cite{Phan:2018cnz}, where also the numerics of the generalized hypergeometric functions was worked out.
}

Below we gather how these simplest scalar integrals behave in four dimensions in the limits of the internal momentum $k \to \infty $ (\rm UV) and $k \to 0$ (\rm IR).

\begin{tabular}{lr}
\begin{minipage}{.2\textwidth}
\begin{tikzpicture}[scale=0.4]
\begin{feynman}
 
%tadpole 
 
\vertex at (0.5,1) (i1);
\vertex at (2,1) (a); 
\vertex at (4,1) (b);
\vertex at (5.5,1) (f1);

\node at (1,1.4) {\small{$0$}};
\node at (3,2.2) {\small{$m_1$}};

\diagram*{
(i1) -- [thick] 
(a) --   [thick, half left] 
(b) -- [thick, half left] (a)
%(b) -- [thick] (f1),
};
\end{feynman}
\end{tikzpicture}
\end{minipage}
&
\begin{minipage}{.5\textwidth}
\bea\nonumber
A_0 &=& \frac{1}{(i\pi^{d/2})} \int \frac{d^d k}  {D_1}
~~~~\rightarrow {\rm UV-divergent:~~}\sim \frac{d^4k}{k^2}\\
&& \hspace*{2.5cm}~~~~\rightarrow {\rm IR - finite} \nonumber
\eea
\end{minipage}
\end{tabular}
\vspace*{5mm}

\begin{tabular}{lr}
\begin{minipage}{.2\textwidth} 
\begin{tikzpicture}[scale=0.4]
\begin{feynman}
\vertex at (6.5,1) (i1);
\vertex at (8,1) (a);
1\vertex at (10,1) (b);
\vertex at (11.5,1) (f1);

\node at (7,1.3) {\small{$p_1$}};
\node at (9,2.2) {\small{$m_1$}};
\node at (9,-0.2) {\small{$m_2$}};
\node at (11,1.3) {\small{$p_1$}};

\diagram*{
(i1) -- [thick] 
(a) -- [half left] 
(b) -- [thick,half left] (a),
(b) -- [thick] (f1),
};
\end{feynman}
\end{tikzpicture}
\end{minipage}
&
\begin{minipage}{.5\textwidth}
\bea\nonumber
B_0 &=& \frac{1}{(i\pi^{d/2})} \int \frac{d^d k}  {D_1 D_2}
~~~~\rightarrow {\rm UV-divergent~~}\sim \frac{d^4k}{k^4}\nonumber\\
&& \hspace*{2.5cm}~~~~\rightarrow {\rm IR - finite} \nonumber
\eea 
\end{minipage}
\end{tabular}
\vspace*{5mm}

\begin{tabular}{lr}
\begin{minipage}{.2\textwidth} 
%\includegraphics[width=20mm]{figs/vert1l}
%vertex
\begin{tikzpicture}[scale=0.4]
\begin{feynman}
\vertex at (12.5,1) (i1);
\vertex at (14,1) (a);
\vertex at (16,2) (b);
\vertex at (16,0) (c);
\vertex at (17.5,2) (f1);
\vertex at (17.5,0) (f2);

\node at (13,1.3) {\small{$p_1$}};
\node at (14.8,1.8) {\small{$m_1$}};
\node at (14.8,0.1) {\small{$m_2$}};
\node at (16.5,1) {\small{$m_3$}};
\node at (17,2.3) {\small{$p_2$}};
\node at (17,-0.4) {\small{$p_3$}};

\diagram*{
(i1) -- [thick] (a) -- [thick] (b) -- [thick] (f1),
(f2) -- [thick] (c) -- [thick] (a),
(b) -- [thick] (c)
};
\end{feynman}
\end{tikzpicture}
\end{minipage}
&
\begin{minipage}{.5\textwidth}
\bea\nonumber
C_0 &=& \frac{1}{(i\pi^{d/2})} \int \frac{d^d k}  {D_1 D_2 D_3}
~~~~\rightarrow {\rm UV-finite~~}\sim \frac{d^4k}{k^6}\\
&& \hspace*{3cm}~~~~\rightarrow {\rm IR - divergent}\nonumber
\eea 
\end{minipage}
\end{tabular}
\vspace*{.5cm}

The propagators are:
\bea
D_1 = k^2-m_1^2,\; D_2 = (k+p_1)^2-m_2^2;\;
D_3 = (k+p_1+p_2)^2-m_3^2.
\eea

After setting \texttt{FI} and showing a traditional way how to solve \texttt{FI} in dimensional regularization procedure where singularities of the integrals are encoded in dimensional regulator $\ep$, we come to the idea of Mellin-Barnes representations and integrals.    

\section{Basic Idea of Mellin-Barnes Representations}

According to Slater~\cite{Slater:1966} the basic idea of representing a function by a contour integral with gamma functions in the integrand is due to the nineteenth century work by S.~Pincherle. It has been developed further by R.~Mellin and E.W~Barnes.
%\subsection{The main theorem for $_AF_B(1)$} \label{appgamma}
In 1895 Mellin~\cite{mellin1895} introduced theorem where
\newline
if
\begin{equation}
    f(x) = \frac{1}{2 \pi i} \int_{c-i \infty}^{c+i \infty} x^{-s}g(s) ds, \label{eq:mellin}
\end{equation}
then 
\begin{equation}
    g(s) = \int_0^\infty x^{s-1} f(x) dx.
\end{equation}
The Mellin transform function $g(s)$ is a locally integrable function
where $x$ is a positive real number and $s$ is complex in general.

Five years after the work by Mellin the paper ``The theory of the gamma function''~\cite{barnes1900}, followed by the 1907 year's series of papers by Barnes~\cite{barnes1907a,barnes1907b,barnes1907c} in which the so-called Barnes contour integrals have been explored.  
They are of the type~\cite{Slater:1966}
\begin{equation}
    I_C = \int_C \frac{dz}{2\pi i}
    \frac{\Gamma(a+z)\Gamma(b+z)\Gamma(-z)}{\Gamma(c+z)} (-s)^z.
    \label{eq:barnes-type}
\end{equation}
It can be shown that this integral (when convergent with $|s|<1$ and $|arg(-s)<\pi|$) is equivalent to the hypergeometric function $_2F_1$ which will be discussed in section~\ref{sec:gamma_hyperg}.

There are different definitions of what is called the Mellin-Barnes integral. For instance, in recent textbook~\cite{de2019solved}, the integral in eq.~(\ref{eq:mellin}) is directly called Mellin-Barnes as `any integral in the complex plane whose integrand contemplates
at least one gamma function'.  
For applications in particle physics calculations, we are interested in a slightly modified form of the above equations.  
What is nowadays commonly called the Mellin-Barnes representation\footnote{In many old textbooks they are called just Barnes integrals. Nevertheless, you may find Mellin-Barnes integrals in books by L.J.~Slater~\cite{Slater:1960:CHF} or  A.W.~Babister~\cite{Babister:1967:TFN} or more recently in a context of Feynman integrals in works by N.I.~Usyukina, A. Davydychev, B. Arbuzov, E. Boos, V. Smirnov, see e.g.~\cite{Smirnov:2004,Usyukina:1975yg,Boos:1990rg}.} in its simplest form is a representation of some power $\lambda$ of a sum of two terms $A$ and $B$, in terms of an integral on the complex plane
\begin{equation}\boxed{
%mathematics \longrightarrow  
\frac{1}{(A+B)^{\lambda}}=
  \frac{1}{\Gamma (\lambda)}
  \frac{1}{2 \pi i}
  \int_{-i \infty}^{+i \infty}dz
  \Gamma (\lambda+z)\Gamma (-z)
  \frac{B^{z}}{A^{\lambda +z}}.
  \label{mb1}
}\end{equation}
The exact conditions for validity of this formula and a proof will be given in section~\ref{sec:2MBrepr}. 
This integral follows from eq.~(\ref{eq:mellin}) when $f(x)=\frac{x^\lambda}{(1+x)^\lambda}$ and $g(s)=\frac{\Gamma[-s]\Gamma[s+\lambda]}{\Gamma[\lambda]}$ 

\begin{equation}
      \frac{x^\lambda}{(1+x)^\lambda} = \frac{1}{2 \pi i} \int_{c-i \infty}^{c+i \infty} \frac{\Gamma[-s]\Gamma[s+\lambda]}{\Gamma[\lambda]}x^{-s} ds. \label{eq:mellinMB}
\end{equation}
The proof for this equation is exactly the same as will be given for eq.~(\ref{mb1}) in  section~\ref{sec:2MBrepr}. 
Replacing x by $A/B$ eq.~(\ref{mb1}) follows.

The relation in eq.~(\ref{mb1}) has an immediate application to physics, for instance, a massive propagator can be written as 
\begin{equation}
%physics \longrightarrow   
\frac{1}{(p^2-m^2)^{a}} =
  \frac{1}{\Gamma (a)}
  \frac{1}{2 \pi i}
  \int_{-i \infty}^{+i \infty}dz
  \Gamma (a+z)\Gamma (-z)
  \frac{(-m^2)^{z}}{(p^2)^{a +z}}.
  \label{eq:tool2}
\end{equation}
 
The upshot of this change is that a mass parameter $m$ merges with a kinematic variable $p^2$ into the ratio $\left(-\frac{m^2}{p^2}\right)^z$ and the integral effectively becomes massless. 

In the context of phase space integrals the basic Mellin-Barnes formula can be employed to bring the integrand to a form where the integration over the original phase space variables can be performed trivially. Of course, the price to pay is that we must introduce \mb{} integrations, however this representation is many times much more convenient for further work than the original one. As an example, consider the integral
\begin{equation}
    \int d\phi_3(p_1,p_2,p_3,Q) 
    \frac{1}{(p_1\cdot p_2)(p_1\cdot p_2+p_1\cdot p_3)}\,,
\end{equation}
where all momenta are massless, i.e., $p_1^2=p_2^2=p_3^2=0$. It can be shown (see e.g.~\cite{Gehrmann-DeRidder:2003pne}) that in $d=4-2\ep$ dimensions, this integral is proportional to 
\begin{equation}
    \int_0^1 dx\, \int_{0}^{1-x} dy\, 
    x^{-\ep} y^{-\ep} (1-x-y)^{-\ep}
    \frac{1}{x(x+y)}\,.
\end{equation}
After performing the trivial change of variable $y\to(1-x)y$, we find
\begin{equation}
    \int_0^1 dx\, \int_{0}^{1} dy\, 
    x^{-1-\ep} (1-x)^{1-2\ep} y^{-\ep} (1-y)^{-\ep}
    \frac{1}{x+(1-x)y}\,.
\end{equation}
Were it not for the last factor, this integral would be trivial to evaluate in terms of the beta function of eq.~(\ref{eq-betadef}). However, we can use the basic \mblong{} formula to convert the sum in the denominator into factors of $x$ and $(1-x)y$ with generic powers, which allows to perform the integration over $x$ and $y$ immediately. Then we obtain the \mb{} representation
\begin{equation}
    \int_{-i\infty}^{+i\infty} dz\,
    \frac{\Gamma (1-\ep) \Gamma \left(-z_1\right) \Gamma \left(z_1+1\right) \Gamma
   \left(-\ep-z_1\right) \Gamma \left(z_1-\ep\right)}{\Gamma (1-3 \ep)}\,,
\end{equation}
which can be further manipulated (e.g. expanded in $\ep$ or evaluated analytically and numerically), as we will show in detail in the later chapters.

{Obviously, the basic \MB{} formula in eq.~(\ref{mb1}) can be applied multiple times if the problem demands and generally we encounter multi-dimensional \MB{} integrals. In fact, in this book we will refer to any multiple complex contour integral of the form
\begin{equation}
    \int_{-i\infty}^{+i\infty}\ldots \int_{-i\infty}^{+i\infty}
    \prod_{j=1}^{n}\frac{dz_j}{2\pi i}
    f(z_1,\ldots,z_n,x_1,\ldots,x_p,a_1,\ldots,a_q) 
    \frac{\prod_k \Gamma(A_k + V_k)}{\prod_l \Gamma(B_l + W_l)}\,,
\label{eq:mbmulti}
\end{equation}
as an \MB{} integral. In this expression the $x_j$ are fixed parameters (e.g. kinematic invariants and masses) while the $a_j$ are expressions of the form $a_i = n_i + b_i \ep$, with $n_i\in \mathbb{N}$ and $b_i\in \mathbb{R}$. Furthermore, $A_k$ and $B_l$ are linear combinations of the $a_i$, while $V_k$ and $W_l$ are linear combinations of the integration variables $z_i$. Finally, $f$ is an analytic function, in practice it is a product of powers of the $x_i$ with exponents that are linear combinations of $a_i$ and $z_i$.}

The use of the Mellin transform and Mellin-Barnes contour integrals in \texttt{QFT} dates back to the 1960's. Indeed, the Mellin transform was used for the analysis of scattering amplitudes and the asymptotic expansion of \texttt{FI}s already at that time, see for instance~\cite{Bjorken:1963zz} and~\cite{Bergere:1973fq, Cheng:1987a}.  Mellin-Barnes contour integrals were investigated for finite three-point functions in~\cite{Usyukina:1975yg}, followed by related works~\cite{Boos:1990rg,Davydychev:1992xr,Usyukina:1993ch}.
However, in terms of mass production of new results in the field, a real breakthrough came towards the end of the last millenium when the infra-red divergent massless planar two-loop box was solved  analytically using \mb\; methods~\cite{Smirnov:1999gc}, followed in the same year by the non-planar case~\cite{Tausk:1999vh}. These results were based on the Feynman parametrization of \texttt{FI} and the \mb{} representations were constructed by the repeated application of the basic formula of eq.~(\ref{mb1}) to the elements of the $F$ and $U$ polynomials (called Symanzik polynomials) associated to the given \texttt{FI}.\footnote{Feynman parametrization and the representation of \texttt{FI} via the Symanzik polynomials will be discussed at depth in chapter~\ref{chapter-MBrepr}.}  
This procedure was automatized initially in~\cite{Gluza:2007rt} and will be discussed in full length in chapter~\ref{chapter-MBrepr}.  
We note that the gamma functions that appear in eq.~(\ref{mb1}), and hence the \mb{} representations of \texttt{FI} and phase space integrals, play a pivotal role, changing the original singular structure of propagators in four-momentum space, see eq.~(\ref{eq:tocomplex}), into another one, connected with singularities of the gamma functions. We will explore this in chapter~\ref{chapter-singul}.

\section{Analytical and Numerical Approaches to \mblong{} Integrals}  

Regarding the evaluation of \mblong{} integrals, there is no one universal method or program covering all cases at the frontier of perturbation theory. 

Concerning analytical approaches, given that the \mb{} representation involves integrations over complex variables, the use of Cauchy's residue theorem to evaluate these integrals comes immediately to mind. We will follow this approach in chapter~\ref{chapter-MBanal}, where we will show in detail how one can obtain a series representation of \mb{} integrals by summing over residues. However, we should note that so far, due to the complex nature of the series which arise and the difficulties in establishing their convergence, mostly expansion of \mb{} integrals in some kinematic variables (see section~\ref{sec:1exp}) have been explored.
Nevertheless, some new developments in understanding the structure of these series (see~\cite{Ananthanarayan:2020fhl} and section~\ref{section:approx}) allows to hope that a full chain of actions shown in eq.~(\ref{eq:chainactions}) with a more efficient application of the approach based on convergence of series to physics studies will result.
\begin{equation}
\boxed{  {\bf Feynman\;  Integrals}  \hookrightarrow  (MB) {\bf Contour Integrals} 
\hookrightarrow (Residues) {\bf Series}}
\label{eq:chainactions}
\end{equation}

Other approaches for obtaining analytic solutions involve transforming the \mb{} integrals into real integrals of so-called Euler-type and evaluating those. We will discuss this approach in chapter~\ref{sec:1anal}, where we will show that it leads naturally to iterated integrals. Under suitable circumstances, such integrals can be evaluated analytically in terms of certain classes of functions that can be thought of as generalizations of the logarithm, so-called multiple polylogarithms.

We note that in these analytical approaches, the level of complexity, at a given loop order, is defined by the number of virtual massive particles that appear in the Feynman integrals. For one-loop integrals the most complicated mathematical objects that appear in their series expansions in $\ep$ {\emph up to finite terms} are Euler dilogarithm functions, that are special cases of multiple polylogarithms. However, in last years it has become obvious that analytical solutions to multiloop integrals, which are always the best solutions to have, involve also functions that can no longer be expressed with multiple polylogarithms, such as elliptic functions and iterated integrals of modular forms. In fact, the class of functions that arise may go even beyond the elliptic case~\cite{Adams:2017ejb}, see also section~\ref{sec:hplsmpls}.

{\it As an interesting aside, we note that elliptic curves, modularity and modular forms were some of the key ingredients to solve a long standing and very famous problem in number theory, Fermat's Last Theorem~\cite{Wiles:1995ig}. 
It is spectacular that the precision of calculations in high energy physics reached the level of sophistication in contemporary cutting-edge mathematical studies.} 

However, in our opinion, due to the many masses and large multiplicities (many external legs) that arise at higher orders in perturbation theory for precision studies in the SM, numerical integration methods may be the most promising avenues for addressing incoming challenges. In the context of \mb{} integrals, one key problem is the construction of \mb{} representations with the the lowest possible dimension. This will be explored in the main part of chapter~\ref{chapter-MBrepr}. Numerical evaluation of \MB{} integrals will be discussed thoroughly in chapter~\ref{chapter-MBnum}.

\section{Mellin and Barnes Meet Euclid and Minkowski}

As discussed in section~\ref{sec:singQFT_general}, the propagators in scattering amplitudes exhibit singularities. We can deal with them by relocating the problem into the complex plane, as schematically denoted in eq.~(\ref{eq:tocomplex}). However, there are cases where it is convenient to avoid singularities altogether by considering artificial kinematic conditions. For example, we may choose to consider some four-momenta squared to be negative, e.g. $p^2 \to -p^2$ in eq.~(\ref{eq:tocomplex}). In such cases, obviously, solutions for $E_k$ discussed in section~\ref{sec:singQFT_general} do not need special treatment in the complex plane. This kinematic configuration, which does not generate such singularities of amplitudes, is realized in Euclidean space with Euclidean kinematic conditions. The real {\it physical} kinematics is of course defined in Minkowski space and is called the Minkowskian kinematics. Obviously, the problem of singularities is related to the Minkowskian nature of kinematics in special relativity. 
By the way, we have already used Euclidean space for solving the integral in eq.~(\ref{eq:ddimgeneral}) in a standard way in section~\ref{sec:singQFT_general}: the Wick rotation effectively transforms Minkowskian momenta into Euclidean ones.

In the context of multiloop calculations and the \mb{} method, considering Euclidean kinematic has two main advantages. First, numerical cross-checks of the obtained solutions for scattering amplitudes or specific multi-loop integrals can be obtained relatively easily. The reason for this is that in Euclidean kinematics the integrands behave ``well'' in a certain sense (they are smooth functions with no rapid oscillations), hence numerical computations are straightforward. Thus, packages for manipulating \mb{} integrals that use general purpose numerical integrators, such as \mbm{}~\cite{Czakon:2005rk}, can be used directly. Second, obtaining analytical solutions can be sometimes easier using  Euclidean kinematics. One example is the fixing of boundary conditions in the differential equation method of solving Feynman integrals. We can then {\it analytically continue} the results (functions) from Euclidean to Minkowski space, checking the results numerically in the Euclidean kinematic regime. Analytic continuation of complex functions will be discussed in chapter~\ref{chapter:complex}.

\begin{tips}{When is Euclidean ({\it unphysical}) kinematics useful?}
Evaluation of integrals in Euclidean kinematics is especially useful if we want to
\begin{enumerate}
    \item check individual integrals numerically against analytical solutions;
    \item check results obtained with different approaches were analytic solutions are not always available. 
\end{enumerate}

\end{tips}

Having said that, we will show in section~\ref{sec:2num} how to use the knowledge of the threshold behaviour of Symanzik polynomials to calculate \mb{} integrals numerically, directly in Minkowskian kinematics. We will avoid the bad behaviour of integrands by a special construction of the \mb{} integrals which restores their Euclidean behaviour.

To finish this section, let us show an example of the Euclidean and Minkowskian kinematic space useful for analytical and numerical solutions of Feynman integrals.

%\section{A simple worked example as an ``invitation'' to the \mb{} topic}

\begin{tips}{Euclidean and Minkowskian Kinematics}

As kinematics depends on the number of external particles present in a given process, it is useful to define it using kinematic invariants. In the simple case of $2 \to 2$ scattering processes these are known as Mandelstam invariants. Transformations between external four-momenta and kinematic invariants are defined in many packages used in high energy physics.

Here we give an example based on the \texttt{KinematicsGen.m} package for a vertex kinematics in which invariants are defined in a standard way, see Fig.~\ref{fig:vertexkin}, as squares of external momenta $p_1$ and $p_2$ (connected with two outgoing legs) and their sum squared $s$ (connected with the incoming leg).  

\begin{figure}[h!]
\begin{center}
\begin{tikzpicture}[scale=0.5]
\begin{feynman}

\node at (0.5,0.8) {$s$=$(p_1+p_2)^2$};
\node at (7,2) [above right] {$p_1^2$=$m^2$};
\node at (7,-2) [below right] {$p_2^2$=$m^2$};
\draw [fill] (3,0) circle [radius=20pt]; 

\vertex at (-1,0) (i1);
\vertex at (3,0) (a);
\vertex at (7,2) (b);
\vertex at (7,-2) (c); 

%\node at (3,1.3) {$p_1$};

\diagram*{
(i1) -- [dashed] (a) -- [fermion] (b), (a) -- [fermion] (c)
};
\end{feynman}

\end{tikzpicture}

\end{center}
%\sidecaption
\caption{Typical assignment of kinematic variables for vertex diagrams. A black blob can include multi-loop corrections like the three-point vertex in Fig.~\ref{fig:1loopmasters}.}
\label{fig:vertexkin}
\end{figure}
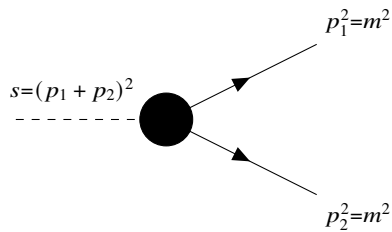
 
\begin{verbatim}
invariants = {p1^2 -> m^2, p2^2 -> m^2, p1*p2 -> -m^2 + s/2};
invEucl = {m->1, s->1};
invMink = {m->1, s->-1};
\end{verbatim} 
\begin{equation}
\label{eq:kinverb}
\end{equation}

In eq.~(\ref{eq:kinverb}) we show verbatim some input kinematic parameters related to Fig.~\ref{fig:1loopmasters} with a notation which is taken from a \math{} script file\footnote{In this book we will often use examples evaluated with the \math{} computer algebra system (CAS)~\cite{Wolfram}. \math{} is frequently used in particle physics and for the construction and analysis \mb{} integrals, as will be described in the following Chapters and the Appendix.} \newline \verb+run_script_1loop_QED_vertex_mink+, available in section `MBnumerics', ``Related and auxiliary Software'' at~\cite{ambrewww}.
\end{tips}

\section{A Simple Example as an  {Invitation}  to the \mb{} Topic}\label{sec:simpleinvitation}

In Fig.~\ref{fig:scheme_actions} a scheme of possible actions is given based on the 1-loop example. Various analytical, semi-analytical and numerical  {analysis} {for \texttt{FI}} based on the \mb{} method will be thoroughly discussed in chapters~\ref{chapter-MBanal}--\ref{chapter-MBnum}.  {`Other methods' highlighted in Fig.~\ref{fig:scheme_actions} for the \texttt{FI} evaluation include notably the powerful differential equations method (\texttt{DEs}) which is intensively developing. For recent studies and software used in analytical and numerical calculations see~\cite{Dubovyk:2022frj,Liu:2022chg,Cordero:2022gsh}. For other exploratory methods in \texttt{FI} computation, see a review~\cite{Heinrich:2020ybq}.}

\begin{figure}[h!]
\begin{center}
\begin{tikzpicture}[scale=1]

\node[rectangle, rounded corners, ultra thick, inner sep=10pt, draw=black!50, fill=black!10] (AA)
{
 
The problem
};

\node[rectangle, inner sep=10pt, draw=black!50, fill=black!10] (A) [below= of AA]
{
Feynman parameters
};

\node[rectangle, rounded corners, ultra thick, inner sep=10pt, draw=black!50, fill=black!10] (B) at (-2,-4) %[left=of A]
{
Mellin-Barnes 
representation
};

\node[rectangle, inner sep=10pt, draw=black!50, fill=black!10] (C) 
at (5,-4)%[right=of A]
{
Sector decomposition 
method
};

\node[rectangle, inner sep=10pt, draw=gray!50, fill=black!10] (D) at (4,0)
%[ right=of AA]
{
Other methods 
};

\node[rectangle, inner sep=10pt, draw=gray!50, fill=black!10] (E) at (1,-6) %[below=of B]
{
Analytical result
};
\draw [ultra thick, rounded corners] (2.4,-6.5) rectangle (6,-5);
\node at (3.7,-5.4) {\small{$\frac{1}{\ep}+2+\frac{1+y}{1-y} \ln{y}$}};
\node at (4.2,-6.1) {\small{$y=\frac{\sqrt{1-4 m^2/s}-1}{\sqrt{1-4 m^2/s}+1},\; s=p^2$}};

\draw [->] (AA) to (A);
\draw [->] (AA) to (D);
\draw [->] (A) to (B);
\draw [->] (A) to (C);
\draw [->] (B) to (E);

\draw [black, ultra thick] (E) to (A);
\draw [black, ultra thick] (B) to (C);

\node at (1.3,-3.5) {\color{black}{numerical }};
\node at (1.4,-3.8) {\color{black}{cross-checks}};

\draw [ultra thick, rounded corners] (-1.1,0.6) rectangle (1.1,2.2);
%\node[red] at (13.3,1.3) {\bf The problem};

\begin{feynman}
\vertex at (-0.9,1.4) (i1);
\vertex at (-.4,1.4) (a);
\vertex at (0.4,1.4) (b);
\vertex at (0.9,1.4) (c);
\diagram*{(i1) -- [photon, thick] (a) -- [half left, thick] (b) -- [photon, thick] (c), (a) -- [half right, thick] (b) ;
};
\end{feynman}
\node at (0,1.9) {\small $m$};
\node at (-0.6,1.6) {\small $p$};

\end{tikzpicture}
\end{center}
\caption{A scheme of actions which will be discussed in the book, the 1-loop example. 
We highlight \mb{} and sector decomposition (\sd{}) methods which are explored presently at the two- and three- loop levels~\cite{Blondel:2018mad}. The analytic result will be discussed and derived using the \mb{} method in sections~\ref{ssec:MBtoSums}~and~\ref{sec:moregensums}, see  eq.~(\ref{eq:G1SE2l2m-sol}).
\label{fig:scheme_actions}}
\end{figure}
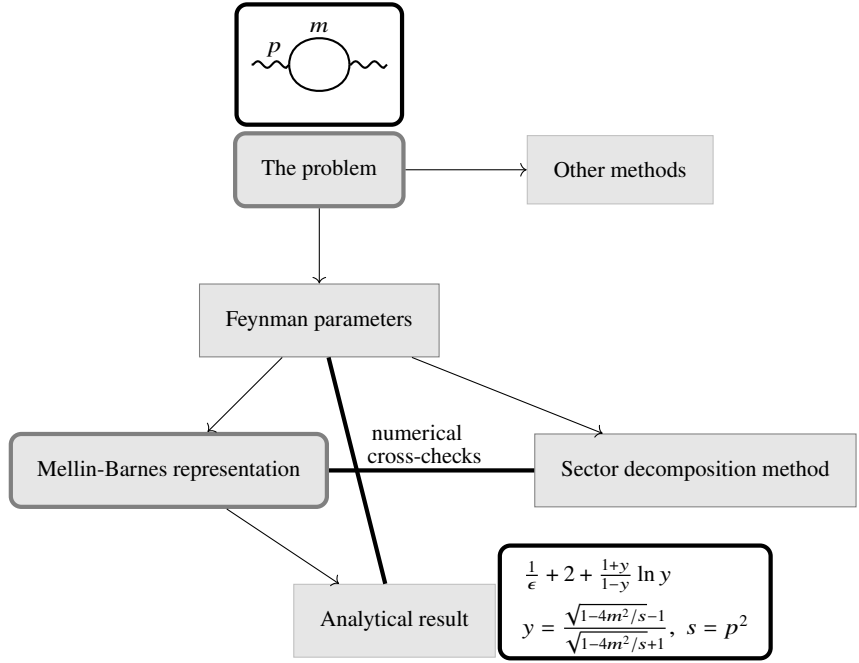

Let us consider the simple example of a one-loop correction to a three-point vertex as shown in Fig.~\ref{fig:vqed}.
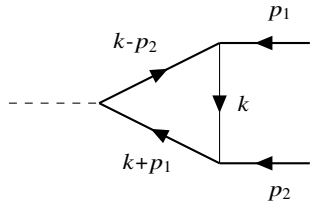
\begin{figure}
    \centering
    \begin{tikzpicture}[scale=0.8]
\begin{feynman}
\vertex at (2.5,1) (i1);
\vertex at (4,1) (a);
\vertex at (6,2) (b);
\vertex at (6,0) (c);
\vertex at (7.5,2) (f1);
\vertex at (7.5,0) (f2);

%\node at (3,1.3) {$p_1$};
\node at (4.6,2) {$k$-$p_2$};
\node at (4.8,0.) {$k$+$p_1$};
\node at (6.4,1) {$k$};
\node at (7,2.5) {$p_1$};
\node at (7,-0.5) {$p_2$};

\diagram*{
(i1) -- [dashed] (a) -- [fermion,thick] (b),
(f2) -- [fermion,thick] (c) -- [fermion,thick] (a),
(b) -- [fermion] (c), (f1) -- [fermion,thick] (b)
};
\end{feynman}
\end{tikzpicture}
    \caption{One-loop correction to the three-point vertex. Thick lines denote massive particles with mass $m$.}
    \label{fig:vqed}
\end{figure}
The corresponding Feynman-integral  is given in eq.~(\ref{vqed}), where we have anticipated that the Laurent series expansion of the result\footnote{The exact result in $\ep$ includes the hypergeometric function $_2F_1$ and can be obtained with \math{}, see \wwwaux{miscellaneous} and~\cite{Gluza:2007vdp}.}  in $\ep$, $d=4-2 \ep$ starts at ${\mathcal O}(\ep^{-1})$,
\begin{equation}
\begin{split}
V(s)  &=
\frac{e^{\ep\gamma}}{i\pi^{d/2}} \int \frac{d^d k}  {[(k+p_1)^2-m^2] [k^2] [(k-p_2)^2-m^2]} \\
&= \frac{V_{-1}(s)}{\ep} + 
V_0(s) + \cdots
\end{split}
\label{vqed}
\end{equation}
In order to illustrate the issues we meet when considering both analytical and numerical solutions of such integrals, it is enough to look at the leading divergent integral $V_{-1}(s)$, which, translated {to} the \mb{} representation, takes the following form, with $m=1,s=(p_1+p_2)^2$ (the construction of this \mb{} representation and the corresponding {\tt Mathematica} file will be discussed in chapter~\ref{chapter-MBanal}),
\begin{equation}
V_{-1}(s) = 
-
% typos corrected 2015-01-21 (Johann Usovitsch)
 ~ \frac{1}{2s} 
\int\limits_{-\frac{1}{2}-i \infty}^{-\frac{1}{2}+i \infty} 
\frac{dz}{2\pi i}~~
\underbrace{(-s)^{-z}}_{\it {Part\; I}}
\overbrace{\frac{\Gamma^3(-z)\Gamma(1+z)} {\Gamma(-2z)}}^{\it {Part\; II}}. \label{parts}
\end{equation}
{\it Parts I} and {\it II} exhibit some general features of \MB{} integrals that are important to understand when seeking either numerical or analytic solutions. In particular, if the kinematic variable $s$ is positive, the exponential function in {\it Part I} is highly oscillating along the path of integration (Problem~\ref{prob:highlyoscil}),  while the gamma functions appearing in {\it Part II} exhibit singularities (infinite values) for real integer values of the arguments. These are obvious obstacles that prevent a straightforward numerical evaluation and in general we find a fragile stability for integrations over products and ratios of gamma functions. We will learn about gamma functions starting from chapter~\ref{chapter:complex}. There are many subtleties connected also with finding analytic solutions of \mb{} integrals. For instance, for diagrams with massless propagators only, the gamma functions in \mb{} integrals have arguments that involve the integration variable $z$  in the form $\Gamma(\ldots\pm z)$. However, in massive cases some of the gamma functions will contain in their argument the integration variable $z$ multiplied by 2, as in the denominator of eq.~(\ref{parts}). The application of Cauchy's residue theorem to such integrals then leads to the appearance of so-called inverse binomial sums for massive cases. Nonetheless, for many basic integrals, like the one considered here, the analytical results are known, and can provide a good theoretical laboratory for investigating more complicated integrals. Indeed, the analytical result for $V_1(s)$ can be written as
\begin{equation}
%\begin{split}
V_{-1}(s) =  
\frac{1}{2}\sum_{n=0}^{\infty}  \frac{s^n} { \binom{2n}{n} (2n+1)} %\label{eq:anal1} 
=
% next line had a typo in acat07
 \frac{2\arcsin(\sqrt{s}/2)}{\sqrt{4-s}\sqrt{s}} 
%\label{eq:anal2} 
=   -\frac{y \ln y}{1-y^2} %\frac{(1 + y) \Log[y]}{(-1 + y)}.%{\bf JG: correct it!}
\label{eq:anal3} 
%\end{split}
\end{equation}
This result arises by applying the residue theorem to the \MB{} integral and summing up the residues of eq.~(\ref{parts}). This procedure, as well as the sums that arise will be discussed in detail in section~\ref{sec:1anal}. The conformal variable $y$ will be also discussed there.

\begin{tips}{Decomposition of \texttt{V3l2m}}
Actually, \texttt{FI} can be decomposed to MIs (\texttt{T1l1m} and \texttt{SE2l2m}) using \texttt{IBPs}, a procedure mentioned in section~\ref{sec:dimreg}.
In our case, we can use a dedicated public \texttt{IBPs} software \texttt{AIR}~\cite{Anastasiou:2004vj},  \texttt{Fire}~\cite{Smirnov:2008iw}, \texttt{Kira}~\cite{Lange:2021edb}, \texttt{LiteRed}~\cite{Lee:2013mka}  or \texttt{Reduze}~\cite{Maierhoefer:2017hyi}  to get 
\begin{equation}
    \texttt{V3l2m} = 
    \frac{{\ep}-1}{{\ep} (s-4)} \texttt{T1l1m} + \frac{1-2 \ep}{\ep (s-4)} \texttt{SE2l2m}.
     \label{eq:IBPdec1}
    % ((\ep - 1)/((s - 4) \ep))  \texttt{T1l1m}+ 
  % \texttt{SE2l2m} ((-2 \ep + 1)/((s - 4) \ep))
\end{equation}
Adding expansions for \texttt{T1l1m} in eq.~(\ref{eq:T1l1m}) and for \texttt{SE2l2m} given in Fig.~\ref{fig:scheme_actions} (see also the result for $\ep^0$ term in section~\ref{sec:moregensums}) we get 
\begin{equation}
 \texttt{V3l2m} =  \frac{y H(0,y)}{{\ep} (y^2-1)}-\frac{y \left(12 H(-1,0,y)-6 H(0,0,y)+\pi ^2\right)}{6 (y^2-1)},
 \label{eq:IBPdec2}
\end{equation}
where functions $H$ are harmonic polylogarithms which will be discussed in next chapter. By inspection (compare eq.~(\ref{eq:anal3}) with $\ep^{-1}$ term in eq.~(\ref{eq:IBPdec2}) we can see that $H(0,y) \equiv \ln y$. Decomposition in eq.~(\ref{eq:IBPdec1}) and numerical cross checks are given in \wwwaux{miscelaneous}.
\end{tips}

 \section{\mb{} applications in Particle Physics and related fields}

\MB{} integrals have numerous applications in High Energy Physics. The field is gradually moving beyond the use of Mellin--Barnes representations as a purely analytic technique toward an increasingly automated computational framework that combines geometric methods, triangulations, hypergeometric differential systems, analytic continuation, and high-precision numerical evaluation. This development is particularly important as experimental measurements demand increasingly precise perturbative calculations.

From the present perspective, applications of the Mellin--Barnes technique can be broadly divided into three classes:
\begin{enumerate}
    \item Investigation of multi-loop Feynman/conformal integrals and amplitudes \cite{Cohl:2023fwj,vanHoegaerden:2025nmh,Bezuglov:2026odp,Benfatto:2026esr,Woodward:2026hce,Banik:2026deo,Alkalaev:2025zhg,Dai:2025olm,Passarino:2025pbe,Solovtsova:2025gmj,Alkalaev:2025fgn,Bezuglov:2025xol,Bezuglov:2025msm,Smirnov:2024pbj,Bargiela:2024hgt,Bargiela:2023aiy,Alkalaev:2025fgn,Ahmed:2025yrx,Ahmed:2024pxr,Diaz:2024tdp,Manashov:2021vmx,Cappellazzi:2026nqv}.
\item Phenomenological studies of high-order QED, QCD and electroweak effects \cite{Fael:2026fxp,Kaptari:2025pmq,Ahmed:2024owh,Clavero:2024yav,Wang:2024son,Zhang:2024fcu,Agarwal:2024gws,Solovtsova:2023tsk,Solovtsova:2022cfp,Bris:2024bcq}; 
\item Investigation of cosmology and celestial physics \cite{Barvinsky:2025pjh,Aoki:2026vbc,DuasoPueyo:2025lmq,De:2025kos,Werth:2025ecc,Bockisch:2024bia,Iacconi:2023ggt,Qin:2022fbv,Choudhury:2025pzp,Guzman:2026muu,Ahmed:2026qtc}.
\end{enumerate}

Some applications and examples will be discussed in the following chapters and in the Appendices. However, we have cited the substantial body of recent literature above in hopes that it will help the interested reader in exploring the uses and development of \mb{} techniques beyond what we can present in this book.
\\

In Fig.~\ref{fig:MBtimeline} we illustrate some of the key steps in the development and application of the \mb{} method in particle physics.

%\newpage
\begin{figure}
    \centering

\begin{tikzpicture}[
    node distance=2.5cm,
    outer sep=0.0cm,
    every node/.style={
        draw, very thick,
        rectangle,
        rounded corners,
        minimum width=2cm,
        minimum height=1cm
    },
    arrow/.style={
        ->,
        >=stealth,
        very thick
    }
]

% Main timeline
\draw[
    ->,
    >=stealth,
    line width=4pt,
    gray,
    rounded corners=10pt
] (0,16) -- (11,16) -- (11,11) -- (0,11) -- (0,6) -- (11,6) -- (11,1) -- (7,1);
\draw[
    ->,
    >=stealth,
    line width=4pt,
    gray,
    rounded corners,
    dashed
] (7,1) -- (5,1);

\draw[
    ->,
    >=stealth,
    line width=4pt,
    gray,
    rounded corners
] (2,16) -- (7,16);
\draw[
    ->,
    >=stealth,
    line width=4pt,
    gray,
    rounded corners
] (9,11) -- (4,11);
\draw[
    ->,
    >=stealth,
    line width=4pt,
    gray,
    rounded corners
] (2,6) -- (7,6);

% ===== 1 =====

\node[text width=0.8cm, align=center, anchor=south, minimum width=0.8cm, minimum height=0.4cm] at (2,16.5) (B) {{\bf\footnotesize 1990}};
\node[text width=2.9cm, align=center, anchor=north] at (2,15.5) (B) {\vphantom{MB}{\bf\footnotesize MB in momentum\\ space} 
{\scriptsize  
\begin{itemize}[leftmargin=0.4cm]
\item Massive triangle Feynman diagrams
\end{itemize}
\vspace{-0.5em}
Boos, Davydychev, \\ 1990 \\[0.5em]
}};
\draw[-, very thick] (2,16) -- (2,15.5);
\fill[white] (2,16) circle (6pt);
\fill[black] (2,16) circle (4pt);

\node[text width=0.8cm, align=center, anchor=south, minimum width=0.8cm, minimum height=0.4cm] at (5.5,16.5) (B) {{\bf\footnotesize 1999}};
\node[text width=2.9cm, align=center, anchor=north] at (5.5,15.5) (B) {\vphantom{MB}{\bf\footnotesize MB in Feynman\\ parameter space} 
{\scriptsize  
\begin{itemize}[leftmargin=0.4cm]
\item Massless on-shell double box
\end{itemize}
\vspace{-0.5em}
Smirnov, 1999 \\[0.5em]
\begin{itemize}[leftmargin=0.4cm]
\item Alternate MB strategy
\end{itemize}
\vspace{-0.5em}
Tausk, 1999 \\[0.5em]
}};
\draw[-, very thick] (5.5,16) -- (5.5,15.5);
\fill[white] (5.5,16) circle (6pt);
\fill[black] (5.5,16) circle (4pt);

\node[text width=0.8cm, align=center, anchor=south, minimum width=0.8cm, minimum height=0.4cm] at (9,16.5) (B) {{\bf\footnotesize 2005}};
\node[text width=2.9cm, align=center, anchor=north] at (9,15.5) (B) {\vphantom{MB}{\bf\texttt{MB}} 
{\scriptsize  
\begin{itemize}[leftmargin=0.4cm]
\item Implements the ``Tausk method'' for resolution of $\ep$-poles
\end{itemize}
\vspace{-0.5em}
Czakon, 2005 \\[0.5em]
}};
\draw[-, very thick] (9,16) -- (9,15.5);
\fill[white] (9,16) circle (6pt);
\fill[black] (9,16) circle (4pt);

% ===== 2 =====

\node[text width=0.8cm, align=center, anchor=south, minimum width=0.8cm, minimum height=0.4cm] at (9,11.5) (B) {{\bf\footnotesize 2007}};
\node[text width=2.9cm, align=center, anchor=north] at (9,10.5) (B) {\vphantom{MB}{\bf\texttt{AMBRE}} 
{\scriptsize  
\begin{itemize}[leftmargin=0.4cm]
\itemsep=0.5em
\item MB from Feynman integrals: loop-by-loop \& global
\item MB for non-planar Feynman integrals
\end{itemize}
\vspace{-0.5em}
Gluza, Kajda, Dubovyk, Riemann, 2007--2016\\[0.5em]
}};
\draw[-, very thick] (9,11) -- (9,10.5);
\fill[white] (9,11) circle (6pt);
\fill[black] (9,11) circle (4pt);

\node[text width=0.8cm, align=center, anchor=south, minimum width=0.8cm, minimum height=0.4cm] at (5.5,11.5) (B) {{\bf\footnotesize 2009}};
\node[text width=2.9cm, align=center, anchor=north] at (5.5,10.5) (B) {\vphantom{MB}{\bf\texttt{MBresolve}} 
{\scriptsize  
\begin{itemize}[leftmargin=0.4cm]
\item Implements the ``Smirnov method'' for resolution of $\ep$-poles
\end{itemize}
\vspace{-0.5em}
Smirnov \& Smirnov, \\2009\\[0.5em]
}};
\draw[-, very thick] (5.5,11) -- (5.5,10.5);
\fill[white] (5.5,11) circle (6pt);
\fill[black] (5.5,11) circle (4pt);

\node[text width=0.8cm, align=center, anchor=south, minimum width=0.8cm, minimum height=0.4cm] at (2,11.5) (B) {{\bf\footnotesize 2011}};
\node[text width=2.9cm, align=center, anchor=north] at (2,10.5) (B) {\vphantom{MB}{\bf\footnotesize MB for phase space} 
{\scriptsize  
\begin{itemize}[leftmargin=0.4cm]
\item MB for angular integrals
\end{itemize}
\vspace{-0.5em}
Somogyi, 2011 \\[0.5em]
\begin{itemize}[leftmargin=0.4cm]
\item Multi-propagator angular integrals
\end{itemize}
\vspace{-0.5em}
Haug et al., 2024, Ahmed et al., 2024 \\[0.5em]
}};
\draw[-, very thick] (2,11) -- (2,10.5);
\fill[white] (2,11) circle (6pt);
\fill[black] (2,11) circle (4pt);

% ===== 3 =====

\node[text width=0.8cm, align=center, anchor=south, minimum width=0.8cm, minimum height=0.4cm] at (2,6.5) (B) {{\bf\footnotesize 2015}};
\node[text width=2.9cm, align=center, anchor=north] at (2,5.5) (B) {\vphantom{MB}{\bf\texttt{MBsums}} 
{\scriptsize  
\begin{itemize}[leftmargin=0.4cm]
\item Analytic evaluation of MB integrals
\end{itemize}
\vspace{-0.5em}
Ochman, Riemann, \\ 2015\\[0.5em]
}};
\draw[-, very thick] (2,6) -- (2,5.5);
\fill[white] (2,6) circle (6pt);
\fill[black] (2,6) circle (4pt);

\node[text width=0.8cm, align=center, anchor=south, minimum width=0.8cm, minimum height=0.4cm] at (5.5,6.5) (B) {{\bf\footnotesize 2018}};
\node[text width=2.9cm, align=center, anchor=north] at (5.5,5.5) (B) {\vphantom{MB}{\bf\texttt{MBnumerics}} 
{\scriptsize  
\begin{itemize}[leftmargin=0.4cm]
\item Numerical MB in physical regions
\end{itemize}
\vspace{-0.5em}
Usovitsch, Dubovyk, Riemann, 2018\\[0.5em]
}};
\draw[-, very thick] (5.5,6) -- (5.5,5.5);
\fill[white] (5.5,6) circle (6pt);
\fill[black] (5.5,6) circle (4pt);

\node[text width=0.8cm, align=center, anchor=south, minimum width=0.8cm, minimum height=0.4cm] at (9,6.5) (B) {{\bf\footnotesize 2021}};
\node[text width=2.9cm, align=center, anchor=north] at (9,5.5) (B) {\vphantom{MB}{\bf\texttt{MBConicHulls}} 
{\scriptsize  
\begin{itemize}[leftmargin=0.4cm]
\itemsep=0.5em
\item Multi-fold MB analytically via conic hulls
\item Triangulation approach
\item MB with polygammas
\end{itemize}
\vspace{-0.5em}
Ananthanarayan, Banik, Friot, Ghosh, 2021--2025\\[0.5em]
}};
\draw[-, very thick] (9,6) -- (9,5.5);
\fill[white] (9,6) circle (6pt);
\fill[black] (9,6) circle (4pt);

% ===== 4 =====

\node[text width=0.8cm, align=center, anchor=south, minimum width=0.8cm, minimum height=0.4cm] at (9,1.5) (B) {{\bf\footnotesize 2022}};
\node[text width=2.9cm, align=center, anchor=north] at (9,0.5) (B) {\vphantom{MB}{\bf\texttt{MBcreate}} 
{\scriptsize  
\begin{itemize}[leftmargin=0.4cm]
\item MB from Feynman integrals: optimal representation
\end{itemize}
\vspace{-0.5em}
Belitsky, Smirnov \& Smirnov, 2022\\[0.5em]
}};
\draw[-, very thick] (9,1) -- (9,0.5);
\fill[white] (9,1) circle (6pt);
\fill[black] (9,1) circle (4pt);

\end{tikzpicture}
    \caption{Timeline of key developments in the Mellin-Barnes method, spanning early applications in particle physics to modern computer codes. The initial two blocks represent the first applications to virtual Feynman diagrams and loop integrals. The "2011" block marks the expansion into real phase-space integration (no public code is currently available for this purpose). All other blocks denote relevant public software implementations.}
    \label{fig:MBtimeline}
\end{figure}

\newpage
  
\section*{Problems}
\addcontentsline{toc}{section}{Problems} 
\begin{problem}
%\begin{enumerate}
%\item 
\label{prob:propag} 
Derive eq.~(\ref{eq:propag}). \newline \hint{} This is a straightforward calculation by exploring the four-vector notation below eq.~(\ref{eq:propag}) and the kinematics as in Fig.~\ref{fig:real-emission}.
\end{problem}
%\item 
\begin{problem}
\label{prob:Ftrick1} 
Prove a general formula  for Feynman parameterization of products of denominators
%~\cite{Feynman:1949}:
\ba
 %####################################################################
\label{eq-appx114e}
\frac{1}{A_1^{n_1}A_2^{n_2}\ldots} &=&
\frac{\Gamma(n_1+\ldots + n_m)}{\Gamma(n_1)\ldots\Gamma(n_m)} \\
&&\int_0^1 dx_1\ldots\int_0^1dx_m
\frac{x_1^{n_1-1}\ldots x_m^{n_m-1}\delta(1-x_1\ldots -x_m)}{(x_1A_1+
  \ldots +x_mA_m)^{n_1+\ldots n_m}}  \nonumber 
%------------------------------------------------
\ea
\hint{} This relation can be proved by induction, see e.g.~\cite{Zeidler:qed}.
\end{problem}
%\item 
\begin{problem}
\label{prob:Ftrick2} Using eq.~(\ref{eq-appx114e}) show in particular  
that
\begin{eqnarray}
% corrected 21 05 2002:
%\frac{1}{AB} &=& \int_0^1 dy \frac{1}{ Ay+B(1-y)   }
 \frac{1}{AB} &=& \int_0^1 dy \frac{1}{[Ay+B(1-y)]^2}
 \label{eq:feynmantrick}
\\
\label{eq-appx114a}
\frac{1}{ABC} &=& \Gamma[3] \int_0^1dx\int_0^1dy\int_0^1dz
\frac{\delta(1-x-y-z)}{(Ax+By+Cz)^3}
\nn \\
&=& \Gamma[3] \int_0^1dx \int_0^{1-x}dy
\frac{1}{[Ax+By+C(1-x-y)]^3}
\\
&=& \Gamma(3) \int_0^1dx \int_0^{1}dy
\frac{x}{[A(1-x)+Bxy+Cx(1-y)]^3}.
\end{eqnarray}
The relation in eq.~(\ref{eq:feynmantrick}) is known in the literature as `the Feynman integration trick'~\cite{Zeidler:qed}.
\end{problem}
%\item  
\begin{problem}
\label{prob:Ftrick}
Apply relations derived in the previous problem to the definition of the $B_0$ and $C_0$ functions (see Fig.~\ref{fig:1loopmasters})

\ba
B_0(p^2;m_1^2,m_2^2) &=&
\frac{(2\pi\mu)^{4-d}}{i\pi^2}
\int
\frac{d^d q}{(q^2 - m_1^2+i\ep)[(q+p)^2 - m_2^2+i\ep]}.
\nn \\
&& \label{eq-appx115a} \\
\label{eq-114a1}
C_0(p_2^2,p_3^2,p_1^2;m_1^2,m_2^2,0)&=&
\frac{(2\pi\mu)^{4-d}}{i\pi^2} \int \frac{d^dk}{d_1 d_2d_3},
\ea
to prove that that these functions can be written in form of eqs.~(\ref{eq:specB0})~and~(\ref{eq:specC0}).
\end{problem}

\begin{problem}
%\item 
Check kinematic invariants written in the first line of eq.~(\ref{eq:kinverb}) and
find kinematic invariants needed to generate 4-PF, e.g. two-loop topology given in Fig.~\ref{fig:b5l2m} in chapter~\ref{chapter-singul}.
\\
\hint{} You can use \texttt{kinematicGen.m}~\cite{ambrewww}.
\end{problem}

\begin{problem}
%\item 
\label{prob:highlyoscil} Using known to you numerical methods and packages, discuss a problem of accuracy for the {\it Part 1} of the integral  in eq.~(\ref{parts}).\newline
\hint{} You may follow discussions in~\cite{Czakon:2005rk} and~\cite{Dubovyk:2016ocz}.
\end{problem}

\putbib[%
bibs/refs,%
bibs/2loops_LL16,%
bibs/Phd_Dubovyk,%
bibs/LRrefa,%
bibs/2loopsreport]
\end{bibunit}

%% file: chapter2v2.tex
%%%%%%%%%%%%%%%%%%%%% chapter.tex %%%%%%%%%%%%%%%%%%%%%%%%%%%%%%%%%
\begin{bibunit}[elsarticle-num-ID] % define the bib-style for the unit: elsarticle-num.bst
%  text-1; this is the corresponding section
%\putbib[2loops] % the *.bib
%\end{bibunit}
% go-on
%--- from: bibunits.sty, adapts the font size of ``References'' to section
\let\stdthebibliography\thebibliography
\renewcommand{\thebibliography}{%
\let\section\subsection
\stdthebibliography} 
 
\chapter{Complex Analysis}
\label{chapter:complex}  

\abstract{\mblong{} (\texttt{MB}) integrals  are defined in the complex plane as contour integrals over integrands involving the exponential and gamma functions in their most elementary form. Thus we consider basic features of the complex exponential, logarithm, gamma, polygamma and hypergeometric functions as well as certain generalizations of these functions that appear in particle physics applications. We also discuss general complex integrals and the notion of residues, along with Cauchy's theorem and introduce the notion of \texttt{MB} representations. We recall the basic notions and constructions of complex analysis to the level which allows to study  the structure and solutions of  \texttt{MB} integrals.}

\section{Complex Numbers and Complex Functions}
\label{sec:1complex}

Besides their relevance in most fields of mathematics, complex functions have also become indispensable tools in the natural sciences and engineering. {\it We could say that complex numbers and functions are auxiliary, technical tools there.} However, quantum mechanics is a clear exception since the wave function is complex. This is already at the heart of the Schr\"odinger equation, $i \hbar \frac{\partial \Psi}{\partial t}=\hat{H} \Psi$. If $\Psi$ were real, the Hermitian operator $\hat{H}$ acting on a real $\Psi$ would give a real number, which would contradict the left side of the Sch\"odinger equation where the time derivative acting on a real $\Psi$ would also be real, multiplied by $i$. 
In this book, complex analysis is basic for \mb{} integrals studies.

At a very basic level, a complex number is just a number of the form $a+b i$, where $a$ and $b$ are real numbers, while $i$ is simply a quantity satisfying $i^2 = -1$. Obviously $i$ cannot be a real number, as $r^2 \ge 0$ for all $r\in \mathbb{R}$, and it is called the \emph{imaginary unit}. The set of all complex numbers is denoted by $\mathbb{C}$. As a set, we have simply $\mathbb{C} = \{a+b i | a,b \in \mathbb{R}\}$. Complex numbers were initially considered as ``subtle as they are useless'' (Cardano, 1545), an ``amphibian between existence and nonexistence'' (Leibnitz, 1702), ``impossible'' and ``imaginary''. Nevertheless, using just the defining property $i^2=-1$, it is straightforward enough to perform basic arithmetic with such numbers, e.g.,
\begin{align}
    (a + b i) \pm (c + d i) &= (a\pm c) + (b\pm d)i,
    \\
    (a+b i)^2 &= (a^2-b^2) + 2 a b i
    \\
    (a + b i) \cdot (c + d i) &= (a c - b d) + (a d + b c) i,
    \\
     \frac{a+bi}{c+di} &= \frac{(a c  + b d)}{c^2+d^2} + \frac{(b c - a d)}{c^2+d^2}i,
\end{align}
and so on. In fact, the notion of complex numbers arose through the study of cubic equations, whose solution through radicals sometimes contains negative numbers under square roots even when all roots of the cubic are real. These real solutions could be computed using nothing more than the rudimentary understanding of complex numbers above.
For a more historical account on the origin of complex numbers, we refer the reader to the excellent textbook~\cite{needham1997visual}, where interesting details can be found, e.g. an account of the idea just mentioned, due to Bombelli, to solve cubic equations by constructing complex arithmetic.

A breakthrough came with Wessel, Argand and Gauss who interpreted complex numbers $z=a+b i$ geometrically, as points or vectors in the $xy$-plane having Cartesian coordinates $(a,b)$. This plane is denoted by $C$ and is called a complex plane. The coordinates $a$ and $b$ are called the \emph{real part} and \emph{imaginary part} of the complex number $z$ and are denoted as $\Re (z) \equiv a$ and $\Im (z) \equiv b$. However, in this geometric viewpoint, it is equally natural to describe a complex number by its \emph{absolute value} and \emph{argument}, which are nothing but the coordinates of the corresponding vector in two-dimensional polar coordinates, see Fig.~\ref{fig:C-plane}. Clearly the absolute value and argument are related to the real and imaginary parts by $R = |z| = \sqrt{a^2+b^2}$ and $\tan \varphi = b/a$ and conversely $\Re (z) = R\cos\varphi$ while $\Im (z) = R\cos\varphi$. Note that the argument is only defined up to multiples of $2\pi$, a fact that is obvious from the geometric interpretation. We will see a more direct way of representing a complex number with its absolute value and argument shortly.

\begin{figure}
\centering
\begin{tikzpicture}[scale=0.8]
\begin{feynman}
 
\vertex at (2,1) (i1);
\vertex at (4,1) (i2);
\vertex at (6,1) (i3);
\vertex at (8,1) (i4);
\vertex at (12,1) (i5);
\vertex at (6,4.5) (f1);
\vertex at (6,-1.5) (f2);

\node[right] at (f1) {$\rm{Im}\; z$};
\node[right] at (11.,0.5) {$\rm{Re}\; z$};

\draw [->]  (i3) -- (f1) ;
\draw [->]  (i3) -- (i5) ;
\filldraw (10,3) circle (2pt);
\draw  (i3) -- (10,3) ;
\filldraw (9,1.7) circle (2pt);
\draw [dashed]  (7.5,1.3) -- (9,1.7) ;
\draw [dashed]  (i3) -- (6.9,1.2) ;
\filldraw (3,.2) circle (2pt);
\draw [dashed]  (i3) -- (3,.2) ;

\node[right] at (2.,-.3) {$z_2=\sqrt{R} e^{(i \varphi/2+\pi)}$};
\node[right] at (9.1,1.8) {$z_1=\sqrt{R} e^{i \varphi/2}$};
\node[right] at (9.5,3.6) {$z \equiv Re^{i \varphi}$};
\node[rotate=25] at (8,2.3) (N) {$|z| \equiv R$};

%\pic [draw, ->, "$\theta$", angle eccentricity=1.5] at (4,2) {angle = mary--origo--bob};

\draw [->] (7.5,1) arc (0:35:1cm);
\node at (7.1,1.25) {$\varphi$};

\draw [->] (8.8,1) arc (0:30:1cm);
\node at (8.2,1.25) {$\varphi/2$};

\diagram*{
 (i1) -- [thin] 
 (i5) --  [thin] 
 (i3) --  [thin] 
 (f1) --  [thin]
 (i3) --  [thin] (f2), 
};

\draw [dotted]  (10,1) -- (10,3) ;
\node at (8,0.6) {$\bf{a}$};
\draw [dotted]  (6,3) -- (10,3) ;
\node at (5.7,2) {$\bf{b}$};

\end{feynman}
\end{tikzpicture}
\caption{\label{fig:C-plane}
Geometrical representation of the complex number with parameters $R$ and $\varphi$, $z=R e^{i \varphi} = |z| (\cos{\varphi} + i \sin{\varphi})$, $|z| \equiv R = \sqrt{\bf{a}^2+\bf{b}^2}, \varphi = \arctan(\bf{a/b})$. Two solutions $z_1,z_2$ of Eq.~(\ref{eq:zargum}) with $n=2$ are shown. }
\label{figcomplex} 
\end{figure}
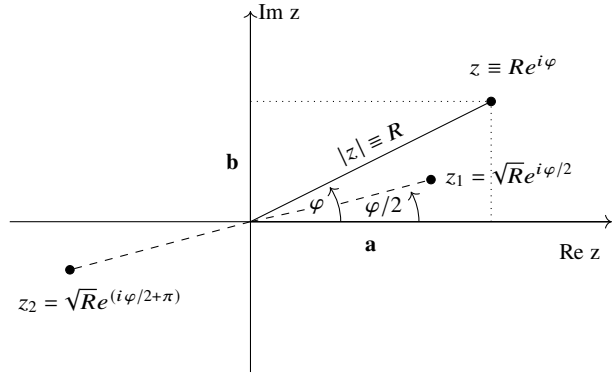

Now that we have the set of complex numbers at our disposal, we can define complex functions as maps from some subset $X \subseteq \mathbb{C}$ to some other subset $Y \subseteq \mathbb{C}$, where each element of $X$ is mapped to exactly one element of $Y$. As a very basic example, consider squaring a complex number $z$: to each $z\in\mathbb{C}$ of the form $a+b i$, we assign its square, $z^2 \equiv z\cdot z = (a^2-b^2)+2ab i$. Obviously this is a function on all of $\mathbb{C}$, as there is a unique assignment of the square for each complex number. Besides such simple arithmetic constructions, another well-known procedure for defining functions is through power series, the most obvious examples being the exponential and trigonometric functions. Similarly to the real case, let us define
\begin{equation}
    \exp(z) \stackrel{\text{def}}{=} \sum_{k=0}^{\infty}\frac{z^k}{k!}\,,
    \quad\mbox{for all } z\in \mathbb{C}\,.
\end{equation}
It can be shown that this power series is absolutely convergent for all 
$z\in \mathbb{C}$ and defines the \emph{complex exponential} of $z$. The notation $e^z = \exp(z)$ is also commonly used and we will use both. It can be proven that the usual properties of the exponential function continue to hold for complex arguments, so
\begin{equation}
    \begin{split}
        e^{z+w} &= e^z e^w\,,\quad z,w\in \mathbb{C}\,,
        \\
        \big(e^z\big)^n &= e^{nz}\,,\quad n\in\mathbb{N}\,,
        \\
        e^0 &= 1\,,
        \\
        e^z &\ne 0\,,\quad z\in\mathbb{C}\,,z\ne 0\,
        \\
        \frac{d e^z}{dz} &= e^z\,,\quad z\in\mathbb{C}\,.
    \end{split}
\end{equation}
Furthermore, for a real number $t$, the series representation of $e^{i t}$ can be written in the following way (note that the rearrangement of the series is allowed because of absolute convergence):
\begin{equation}
    e^{i t} = \left(1 - \frac{t^2}{2!} + \frac{t^4}{4!} - \ldots\right) 
    + i \left(\frac{t}{1!} - \frac{t^3}{3!} + \frac{t^5}{5!} - \ldots\right)\,, 
\end{equation}
which implies Euler's formula $e^{it} = \cos t + i \sin t$ for real $t$. Then, motivated by this observation, we define the sine and cosine of a complex variable $z$ by the corresponding power series,
\begin{equation}
\begin{split}
    \cos(z) &\stackrel{\text{def}}{=}
    \sum_{k=0}^{\infty}(-1)^k\frac{z^{2k}}{(2k)!}\,,
\\
    \sin(z) &\stackrel{\text{def}}{=}
    \sum_{k=0}^{\infty}(-1)^k\frac{z^{2k+1}}{(2k+1)!}\,,
    \quad\mbox{for all } z\in \mathbb{C}\,.
\end{split}
\end{equation}
These definitions then lead directly to Euler's formula for a general complex number, which connects the exponential and trigonometric functions of a complex variable:
\ba
e^{iz}= \cos{z} + i \sin{z}\,.
\label{eq:moivre}
\ea
This equation implies in particular
\ba
e^{i\pi} =-1
\qquad\mbox{and}\quad
e^{2 i\pi} = 1\,
\ea
and more generally $e^{2n i \pi} =1$ for all $n \in \mathbb{Z}$. This implies that the exponential function is in fact periodic in the complex plane with a period of $2\pi i$,
\ba
e^{z+2\pi i} = e^z e^{2\pi i} = e^z\cdot 1 = e^z\,.
\ea
Incidentally, Euler's relation allows us to express a complex number $z$ directly with its absolute value and argument. Recalling that $\Re (z) = R\cos\varphi$ while $\Im (z) = R\cos\varphi$, we immediately find 
\begin{equation}
    z = R e^{i\varphi}\,,
\end{equation}
a relation that will prove to be very useful in computing fractional powers, so in particular $n^{\mathrm{th}}$ roots of complex numbers.

Another familiar procedure for defining interesting mappings is to consider the inverse of particular functions. Starting with a simple example, let us consider the inverse of the function $f(z)=z^2$. That is, for a complex $z$, we now seek another complex number, denoted $\sqrt{z}$, such that $(\sqrt{z})^2 = z$. In order to find such a complex number, consider the representation of $z$ in terms of its absolute value and argument: $z = R e^{i\varphi}$. Then, it is immediately obvious that the number $\sqrt{z} = \sqrt{R} e^{i\frac{\varphi}{2}}$ has the desired property. Indeed
\begin{equation}
    \big(\sqrt{z}\big)^2 = \big(\sqrt{R} e^{i\frac{\varphi}{2}}\big)^2
    =  \big(\sqrt{R}\big)^2 \big(e^{i\frac{\varphi}{2}}\big)^2
    = R e^{2i\frac{\varphi}{2}} = R e^{i\varphi}\,.
\end{equation}
We note that by definition the absolute value of a complex number is real and non-negative, $R\ge 0$, so its real square root $\sqrt{R}$ always exists and is real. However, notice that the square root is not unique. Clearly for any number $\sqrt {z}$ all other numbers of the form $ \sqrt {z} e^ {i n \pi} $ also have the property that their square is $z$ (recall $e^{2n\pi i}=1$). Since
\begin{equation}
    e^{in\pi} = (-1)^n\,,\quad n\in\mathbb{Z}\,,
\end{equation}
(this is easily seen from the geometric interpretation of complex numbers or Euler's formula), we have precisely two distinct solutions: $+\sqrt{z}$ and $-\sqrt{z}$. This is not unexpected and follows the pattern for non-negative real numbers. We note in particular that
\begin{equation}
\sqrt{-1}=+i 
\quad\mbox{or equally well}\quad
\sqrt{-1}=-i\,.
\end{equation} 

Before moving on, let us note that the generalization of the above considerations to general $n^{\mathrm{th}}$ roots is rather straightforward. Indeed, given some complex number $z$, we are now looking for a complex number $z^{1/n}$ such that $\big(z^{1/n}\big)^n = z$. It is easy to see that the number ${z}^{1/n} = {\sqrt[n]R}e^{i\frac{\varphi}{n}}$ has this property. However, as before, this is not the only solution. In fact if $z^{1/n}$ is an $n^{\mathrm{th}}$ root of $z$, then so is any other number of the form $z^{1/n}e^{2m\pi i/n}$, for all $m\in\mathbb{Z}$. From the geometric interpretation, it is quite clear that the numbers $e^{2m\pi i/n}$ are distinct only for $m=0,1,2,\ldots,n-1$. (Indeed, the associated vectors point to the vertices of a regular $n$-sided polygon inscribed in the unit circle). Thus we find that in general
\begin{equation}
{z}^{1/n} = {\sqrt[n]R} e^{\left(i\frac{\varphi}{n}+\frac{2\pi m}{n}\right)}\,,
\quad
m=0,1,2,\ldots,n-1\,. 
\label{eq:zargum}
\end{equation}
There are then in fact $n$ separate complex $n^{\mathrm{th}}$ roots of a complex number. In Fig.~\ref{figcomplex} the complex number is presented as a vector in the complex plane. This makes it easy, among other things, to visualize fractional powers, as exemplified on the case of the square root and its two solutions $z_1,z_2$ given in Eq.~(\ref{eq:zargum}) for $n=2$. See~\cite{needham1997visual,wegert2010phase} for more useful concepts of visualizations connected with complex numbers and functions.

%\begin{tips}{Multi-valued maps, branch cuts and Riemann sheets}
As we have seen, the inverse mappings of functions can be multi-valued. Returning to the square function, its inverse is evidently a two-valued relationship, since the equation $w^2=z$ does not have a unique solution. Instead it has two solutions, $\pm\sqrt{z}$. Thus, in order to be able to think of the mapping
\begin{equation}
    w = f(z) = \pm\sqrt{z} = z^{1/2}
\end{equation}
as a function, we must make some restrictions. For example, for a non-negative real numbers $x$, we can simply declare that 
\begin{equation}
y = g(x) = +\sqrt{x} = x^{1/2}
\end{equation}
that is, we \emph{define} the square root to be the non-negative number $y$ for which $y^2 = x$. This restriction then defines a single-valued mapping, i.e., a function, on the non-negative real numbers. Nevertheless, we stress that this is a choice (however natural), and we could have defined the square root of a non-negative real number to be the non-positive number $y$ for which $y^2=x$.
A way of thinking about this restriction that will become useful very shortly is to imagine that we have erected a ``barrier'' at the single point $x=0$. This idea must be refined when working with complex numbers. To see how, let us think about how the value of $f(z)$ changes as the point $z$ moves along the unit circle in the complex plane. For definiteness, let us start at the point $z_0 = -1$ and move in the clockwise direction. We can write
\begin{equation}
    z = e^{i\varphi}
    \quad\mbox{so}\quad
    w = z^{1/2} = e^{i\frac{\varphi}{2}}\,,
    \quad \pi \ge \varphi \ge -\pi\,.
\end{equation}
We have chosen the argument $\varphi$ to run between $\pi$ and $-\pi$ and clearly $\varphi=\pm\pi$ both correspond to the same point $z_0=-1$. However, as $z$ moves all the way around the circle from $\varphi=\pi$ to $\varphi=-\pi$ and we return to our starting position, $w$ traces out only half of a circle. Thus a continuous motion in the complex plane changes the value of the square root from $e^{i\frac{\pi}{2}} = +i$ to $e^{-i\frac{\pi}{2}} = -i$.\footnote{Note that you can compute the value of $e^{i\frac{\pi}{2}}$ and $e^{-i\frac{\pi}{2}}$ unambiguously from Euler's relation.}
 
\begin{svgraybox}
This problem arises because the number $z=0$ is special: it has just one square root, while every other complex number $z\ne 0$ has two distinct square roots. This of course already happens for the reals, but in that case we can circumvent the problem by setting a barrier at the single point $x=0$. In the complex case though, we must make sure that {\it no continuous path} can completely encircle the special point $z=0$. This requires that we use a bigger ``barrier'' and this is commonly done by introducing a \emph{branch cut}. 
\end{svgraybox}
In our case, we can for example have the ``cut'' extend from the branch point\footnote{A branch point of a multi-valued map is a point such that {\it the function is discontinuous} when going around an arbitrarily small circuit around this point.} at $z=0$ along the negative real axis to the point at infinity, so that the argument of the variable $z$ in this cut plane is restricted to the range (notice the sharp inequality) $-\pi < \varphi \le \pi$.  
With this restriction, the square root mapping is single-valued, and we call this the {\it principal value} or the main branch of the argument. Notice though that the price we pay for making the function single-valued is that it is not continuous ``across the cut''. What this means is that we obtain two different values for the function if we approach a point on the cut from the two sides. For example, in our case, the cut is along the negative real axis, so we can approach some negative number, say $-R$ ($R>0$), from the upper or lower complex half-plane. So let us take $z=-R\pm i \eps$ with $\eps>0$ and consider the limit $\eps\to 0$. We find
\begin{eqnarray}
    \sqrt{z} = \sqrt{-R \pm \eps} = \sqrt{R e^{\pm i(\pi-\eps')}}
    &=& \sqrt{R} e^{\pm i(\pi/2 - \eps/2)}   \\ 
    &=& \sqrt{R}\big[\sin(\eps'/2) \pm i \cos(\eps'/2)\big]
    \xrightarrow{\eps'\to 0} \pm i \sqrt{R}\,, \nn
\end{eqnarray}
where we have used that a number of the form $-R\pm i \eps$ with a small but positive $\eps$ can be expressed as $R e^{\pm i(\pi-\eps')}$ with some $\eps'$ which is also small and positive. Notice that in this expression, the arguments are between $-\pi$ and $\pi$ in accordance with our choice above. Thus, the limit is either $+i\sqrt{R}$ or $-i\sqrt{R}$ depending on how we approach the cut and so the function is discontinuous. Incidentally, because of this discontinuity, we also have some freedom to define the value of the function along the cut itself~\cite{Vollinga:2004sn}. E.g. in our prescription, we have chosen to allow $-\pi < \varphi \le \pi$, so negative real numbers must be represented as $x = |x| e^{i\pi}$ and we find $\sqrt{x} = +i\sqrt{|x|}$. But we could also have chosen the same cut along the negative real axis but required that $-\pi \le \varphi < \pi$. This would have meant that the square root of a negative real number $x$ evaluates to $\sqrt{x} = -i \sqrt{|x|}$. All of these conventions are part of the principle value prescription, and it should be clear that choices different from those here can be made both for the position of the cut as well as the precise definition of the function on the cut. In fact, there is no general convention about the definition of the principal value, and another common choice is to take the argument in the interval $\varphi\in[0,2\pi)$, which corresponds to cutting along the positive real axis. As discussed in~\cite{wegertbook}, this ambiguity is a perpetual source of misunderstandings and errors. Indeed we have seen that even the value of the innocent-looking $\sqrt{-1}$ can be different based on the specific convention. Thus, it is very important to fix the notation and we should be aware of possible differences in software, e.g. in \math{} $\sqrt{-1}=+i$. We give an example of how to control numerical evaluation of results in Minkowskian kinematics by putting small imaginary parts to the input variable, see a pitfall note in the next section and  \wwwaux{miscellaneous}. 

Last, let us briefly turn to the notion of Riemann sheets. Staying with our example of the square root, we have seen that having introduced a cut along the negative real axis, we can still choose the value of $\sqrt{-1}$ as either $+i$ or $-i$. Thus, in order to give a complete description of the square root function, we can consider two copies of the cut $z$ plane and on one plane define the square root of $-1$ to be $e^{i\pi/2} = +i$, while on the other we define the square root of $-1$ to be $e^{-i\pi/2} = -i$. We call these two copies of the cut plane \emph{sheets}. Then, we see that the now single-valued function $w = z^{1/2}$ maps the first sheet to the right half of the $w$ plane while the second sheet is mapped to the left half of the $w$ plane. The two disconnected sheets can then be ``glued together'' along the cuts to form a single Riemann surface on which $w = f(z) = z^{1/2}$ can be defined as a continuous (in fact holomorphic, i.e., complex differentiable) function whose image is the entire $w$ plane (except $w=0$). 
In Fig.~\ref{fig:riemann-sheets} we present the Riemann-sheet plots for $w=\sqrt{z}$ and $w = \ln(z)$ generated with \math{}, see also  the source in \wwwaux{miscellaneous}. 
\begin{figure}
    \centering
    \includegraphics[width=0.45\textwidth]{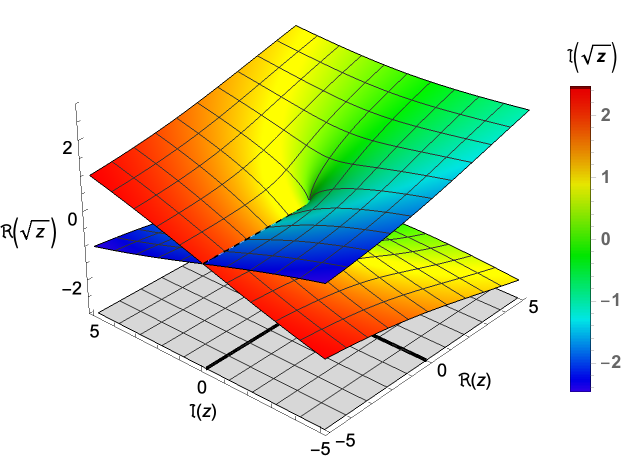}
    \includegraphics[width=0.45\textwidth]{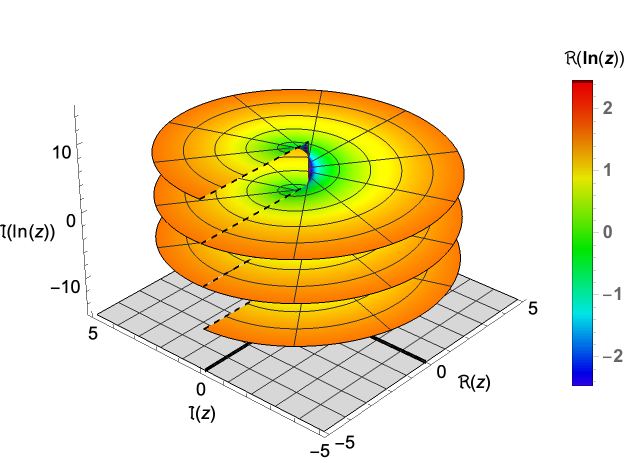}
    \caption{Riemann surfaces for the functions $w=\sqrt{z}$ (left) and $w = \ln(z)$ (right). The gray planes at the bottom represent the complex $z$ plane, while the colored sheets are the Riemann surfaces on which the functions are continuous and one-to-one. Dashed lines show the cuts along which the different branches are glued.}
    \label{fig:riemann-sheets}
\end{figure}

We will not discuss the details of this procedure any further here, but the interested reader will find more examples connected with Riemann sheets and branch cuts with nice visualizations in the textbook~\cite{Roussosimprop}.

\section{The Complex Logarithm \label{sec:3complex} }

The basic \mb{} formula in Eq.~(\ref{mb1}) involves just two types of functions: the power function $A^z$ and the gamma function $\Gamma(z)$. The power function for generic complex base and exponent is defined through the complex logarithm function, $\ln z$,
\begin{equation}
    z^w \equiv e^{w \ln z}\,,\quad z,w\in\mathbb{C}\,.
\end{equation}
But logarithms of generally complex arguments also appear explicitly, e.g. in Taylor expansions of quantities raised to $d=4-2\eps$-dependent powers
\begin{eqnarray}
a^{\epsilon} &\equiv& e^{\epsilon \ln (a)}
 = 1 +  \ln(a)~ \epsilon + \frac{1}{2} \ln^2(a)~\eps^2 +\ldots .
\label{expon}
\end{eqnarray}
Thus we will first study the complex logarithm and its generalizations in this and the next section. Then we will consider the gamma function in detail.

To begin, notice that the logarithm is the inverse of the exponential function $e^z$ that we have studied above. That is, $w=\ln z$ is by definition a complex number such that $z=e^w$. If the complex number $z$ is given in polar form, $z=R e^{i\varphi}$, then one such number is $w = \ln R + i\varphi$. However, this is not the only such number and evidently any number of the form $\ln R + i(\varphi + 2n\pi)$, $n\in\mathbb{Z}$ is also a logarithm of $z$. Thus, we are in the familiar situation that the inverse function is not single-valued. Similarly to the case of the square root function discussed previously, in order to define a single-valued function, we must make a branch cut and impose a principal value prescription. It is easy to check that again $z=0$ is a branch point: the logarithm changes by $2\pi i$ as we move on a circle of arbitrary size around this point. The branch cut will then extend in this case as well from the point $z=0$ to the point at infinity and in our principle value prescription, we will choose the cut along the negative real axis. Thus, we demand that the imaginary part of the logarithm be in the interval between $-\pi$  and $\pi$. Of course, the logarithm is then discontinuous across the cut. Indeed, for $R>0$ and $\eps>0$, we find with a calculation very similar to the one for the square root that in the limit $\eps \to 0$
\begin{align}
    \ln (-R \pm  i\eps) &\xrightarrow{\eps\to 0} \ln R \pm  i \pi\,,
    \\
    \ln (+R \pm  i\eps) &\xrightarrow{\eps\to 0} \ln R\,.
\end{align}
To finish the definition, we must choose the value of the logarithm along the negative real axis, and we employ the widely used prescription that the argument of the logarithm takes values in the interval $(-\pi,+\pi]$. In practice, this amounts to setting the logarithm of a negative real number $-R$, $R>0$ to $\ln(-R) = \ln(R) + i\pi$. We note that with this choice, the function is continuous as the cut is approached coming around the finite endpoint of the cut in the counterclockwise direction.

Summing up the above considerations, in the general case the logarithm can be defined as
\ba
\label{eq:logz1}
\ln z = \ln (a+ib) &\equiv& \ln \left( \rho e^{i\varphi} \right)
= \ln \rho + i \varphi,
\\
\rho &=& \sqrt{a^2+b^2},
\\
\varphi &=& \arctan \frac{b}{a} + \pi \theta(-a) \rm{sign}(b)\,,
\quad \Re(z)\,,\Im(z)\ne 0\,.
\label{appx147}\ea
Here we assume that $a\ne 0$ and $b\ne 0$. Then, the ratio $\frac{b}{a}$ which appears in the argument of the $\arctan$ exists and is real. Moreover, we choose the branch of the $\arctan$ function that takes values in the interval $(-\pi/2, +\pi/2)$. Furthermore $\theta$ is the step function with $\theta(a)=1$ if $a>0$ and $\theta(a)=0$ if $a<0$. Last $\rm{sign}$ is the sign of its argument, i.e.,  $\rm{sign}(b) = 1$ if $b>0$ and $\rm{sign}(b) = -1$ if $b<0$. Turning to the special cases, if $a\ne 0$ and $b=0$, we have simply
\begin{equation}
\ln z = \ln (a) = \ln|a| + i \pi \theta(-a)\,,\qquad \Im(z)=0\,,
\label{eq:logz2}
\end{equation}
while for $a=0$ and $b\ne 0$ we find
\begin{equation}
\label{eq:logz3}
\ln z = \ln (ib) = \ln|b| + i \frac{\pi}{2} \rm{sign}(b)\,,
\qquad \Re(z)=0\,.
\end{equation}
Finally, for $a=b=0$, i.e., at $z=0$, the logarithm function is not defined.
\begin{tips}{The Logarithm Function in \math{} -- a Pitfall}
The implementation of the logarithm function in \math{} follows the principal value prescription given above. In particular, the function is discontinuous across the negative real axis as shown by e.g. the numerical evaluations\footnote{In this book we use `minted' tex style to highlight the \math{} In and Out     %{\color{blue}\rm In[..]} and {\color{blue}\rm Out[..]} 
cells to \LaTeX~\cite{wwwminted}.} 
% 

%\begin{listing}[H]
\begin{minted}[frame=single,breaklines,fontsize=\small]{mathematica}
In[1]:= Log[-2 + 0.001 I]  
Out[1]:= 0.6931473055599296 + 3.14109265363146 I
In[2]:= Log[-2 - 0.001 I]  
Out[2]:= 0.6931473055599296 - 3.14109265363146 I
\end{minted}
%\end{listing}

On the cut itself, \math{} evaluates the logarithm symbolically as
 
\begin{minted}[frame=single,breaklines,fontsize=\small]{mathematica}
In[3]:= Log[-2]  
Out[3]:= I Pi + Log[2]
\end{minted}

The fact that the function is discontinuous across the negative real axis can lead to interesting pitfalls when computing with non-exact (i.e., machine precision) numbers. E.g. consider the variable
\begin{equation}
    x = \frac{\sqrt{1-\frac{4m^2}{M^2}}-1}{\sqrt{1-\frac{4m^2}{M^2}}+1}\,.
\end{equation}
For $m>M/2$, the square root is imaginary and in fact, $x$ is just a complex phase, i.e., $x=e^{i\theta}$ for some real $\theta$. E.g. choosing $m=1$ and $M=1.25$, we find 
 
\begin{minted}[frame=single,breaklines,fontsize=\small]{mathematica}
In[4]:= xval = (Sqrt[1-4 (m^2/M^2)]-1)/(Sqrt[1-4(m^2/M^2)]+ 1) /. {m -> 1, M -> 1.25}  
Out[4]:= 0.21875 + 0.9757809372497497 I
\end{minted}

Now, let us evaluate the logarithm of $(1-x)^2/x$ at this particular point in two ways. First, by substituting this expression directly into the logarithm function and second, by substituting the partial fractioned form of the expression ({\tt xval} is defined above), 

\begin{minted}[frame=single,breaklines,fontsize=\small]{mathematica}
In[5]:= {Log[(1 - x)^2/x], Log[1/x - 2 + x]} /. x -> xval
Out[5]:= {0.44628710262841936 - 3.141592653589793*I, 
  0.44628710262841953 + 3.141592653589793*I}
\end{minted}

{\it Although the two expressions are exactly equivalent mathematically, we obtain different numerical results!} The solution to this conundrum becomes obvious if we evaluate the arguments themselves
 
\begin{minted}[frame=single,breaklines,fontsize=\small]{mathematica}
In[6]:= {(1 - x)^2/x, 1/x - 2 + x} /. {x -> xval}
Out[6]:= {-1.5624999999999998 - 1.1102230246251565*^-16*I,  -1.5625 + 0.*I}
\end{minted}

Apparently $(1-x)^2/x$ is real and negative, since the imaginary parts are zero up to machine precision, but due to the finite resolution of the number representation, the formally equivalent expressions evaluate to numbers whose tiny imaginary parts have the opposite signs! Thus, in one case, due to numerical rounding errors, the logarithm is actually evaluated below the cut, giving the incorrect imaginary part. Such subtleties must be kept in mind when using numerical software.
\end{tips}

Turning to some basic properties of the complex logarithm, first it is immediately obvious that the relation $\ln(x y) = \ln(x) + \ln(y)$, well-known from the real case, cannot hold in general in the complex case. Indeed, already for two negative reals $x,y<0$ the left hand side is real, while the right hand side will have an imaginary part equal to $2\pi$. A similar situation occurs for the relation $\ln(x/y) = \ln(x) - \ln(y)$, where for non-zero real numbers $x$ and $y$ with different signs (i.e., such that $x/y <0$), the left hand side has an imaginary part of $+\pi$, while the imaginary part of the right hand side is either $+\pi$ if $x<0$ and $y>0$ or $-\pi$ if $x>0$ and $y<0$. Indeed, the correct relations read
\begin{align}
    \ln(w z) &= \ln w + \ln z + 2n\pi i\,,\qquad n\in\mathbb{Z}\,, 
\label{eq:ln_wtimesz}
\\
    \ln\left(\frac{w}{z}\right) &= \ln w - \ln z + 2n\pi i\,,\qquad n\in\mathbb{Z}\,. 
\label{eq:ln_woverz}
\end{align}
That is, the logarithm of a product (quotient) of complex numbers is the sum (difference) of the logarithms, \emph{plus a multiple of} $2\pi i$, such that the argument of the result is in the interval $(-\pi,\pi]$. In practice, there are two solutions to this situation. In numerical evaluation, assuming we have in general a product or quotient of two complex numbers, we can evaluate the product or quotient of the two complex numbers first and then use the relations in Eqs.~(\ref{eq:logz1})--(\ref{eq:logz3}) to compute the logarithm. For analytical results, we can split the logarithm of a product (quotient) of two complex numbers to the sum (difference) of two logarithms as in the positive real case, adding a proper relation for an additional phase. Surprisingly, it is not entirely straightforward to write such a relation, though and this question is taken up in Problem~\ref{prob:logphase}.

\begin{tips}{Complex Quadrature under the Logarithm} 
\label{exsquares_sec2}
A commonly occurring situation is that we encounter the logarithm of a quadratic expression, $\ln(ax^2+bx+c)$. Since the equation
\ba
ax^2+bx+c=0
\ea
has two roots, say $x_1$ and $x_2$, we can write
\begin{equation}
    ax^2+bx+c = a(x-x_1)(x-x_2)\,.
\end{equation}
However, as just discussed, when the factors of this product are not all positive real, we must be careful with using the logarithmic identities on this product. In particular, if the roots are real, we must be careful as we vary $x$ as the sign of one of the factors changes as we cross the roots (assume here that $x$ is real).

One way to proceed in this situation is to assign a small imaginary part to the quadratic expression. Then we consider the equation
\begin{equation}
ax^2+bx+c-i\epsilon=0,
\end{equation}
which has the roots (to first order in the small quantity $\eps$)
\begin{equation}
\bar{x}_1 = x_1 + i\frac{\epsilon}{a(x_1-x_2)},\qquad
\bar{x}_2 = x_2 - i\frac{\epsilon}{a(x_1-x_2)}.
\end{equation}
Let us assume for simplicity that the original roots are ordered such that $x_1 > x_2$, so that $x_1-x_2>0$. Then we can simply redefine the small quantity $\eps$ by this positive quantity and write
\begin{equation}
\bar{x}_1 = x_1 + i\frac{\epsilon}{a},\qquad
\bar{x}_2 = x_2 - i\frac{\epsilon}{a}.
\end{equation}
To derive our final expression, we need one more consideration. Namely, let $A$ and $B$ be real numbers. Then the following holds 
\ba
\ln(AB-i\epsilon) = \ln(A-i\epsilon') + \ln(B-i\frac{\epsilon}{A}),
\ea
where the signs of $\epsilon'$ and $\epsilon$ are equal. Using the above, we obtain the following useful relation:
\ba
\ln(ax^2+bx+c-i\epsilon) = 
\ln(a-i\epsilon') + \ln(x-x_1-i\frac{\epsilon}{a})
+ \ln(x-x_2 + i\frac{\epsilon}{a}) ,
\ea
where again $x_{1,2}$ are roots of $x^2+bx+c=0$ and $x_1>x_2$.
\end{tips}

\section{Generalizations of the Logarithm Function: {Classical} Polylogarithms}  
\label{sec:Li-n}

In the previous section, we have introduced the complex logarithm as the inverse of the exponential function. However, there are several other ways to think about the logarithm which generalize naturally to a set of functions that plays a very important role in modern Feynman integral calculations. Here we briefly explore these generalizations.

To start, consider computing the integral of a rational function of a single variable, say $x$,
\begin{eqnarray}
\int dx\, \frac{P(x)}{Q(x)},
\end{eqnarray}
where $P(x)$ and $Q(x)$ are polynomials in $x$. We know form algebra that after performing polynomial division and partial fraction decomposition (over the complex numbers), the integrand can be reduced to the sum of a polynomial plus functions of the form $\frac{1}{(x-a)^n}$, where $n$ is a positive integer and $a$ is a (in general complex) constant. Hence, the problem can be reduced to the computation of the following integrals
\begin{equation}
    \int dx\, x^n
    \quad\mbox{and}\quad
    \int dx\, \frac{1}{(x-a)^n},\qquad n\in\mathbb{N}.
\end{equation}
These integrals can almost all be expressed in terms of just rational functions. However the second integral for $n=1$ cannot. In fact, this integral can be used to define a new function, the logarithm of $x$,
\begin{equation}
    \int \frac{dx}{x-a} = \ln (x-a) + C.
    \label{eq:int-log-def}
\end{equation}

%\begin{tips}{The arcus tangent and the logarithm}
\subsection{The Arcus Tangent and the Logarithm}
In basic courses of calculus, it is shown that the integral of the rational function $R(x) = \frac{1}{1+x^2}$ is just the arcus tangent function, which we already encountered in  Eq.~(\ref{appx147}) above,
\begin{equation}
\int \frac{dx}{1+x^2} = \arctan (x) + C.
\label{eq:arctan-def}
\end{equation}
Usually we choose to work on the  branch of the arcus tangent function where its values fulfills the condition
\bea
-\frac{\pi}{2} < \arctan (x) < +\frac{\pi}{2},
\eea
and we formally set $\arctan (\pm \infty) = \pm \frac{\pi}{2}$.

However, following our above discussion on integrating rational functions, we expect that the integral in Eq.~(\ref{eq:arctan-def}) should be expressible in terms of at most logarithms. This is of course true and it is easy to derive the relation between the arcus tangent and logarithm functions. Indeed, by performing partial fraction decomposition over the complex numbers we have
\begin{equation}
\begin{split}
\int \frac{dx}{1+x^2}
	&= \int dx\left[\frac{1}{2 i (x-i)} - \frac{1}{2 i (x+i)}\right]
\\
	&= \frac{1}{2i}\ln(x-i) - \frac{1}{2i} \ln(x+i)
\\
	&= -\frac{i}{2}\ln \left(\frac{1 + i x}{1 - i x}\right),
\end{split}
\label{eq:arctan-log}
\end{equation}
where we have used $\ln(a) - \ln(b) = \ln(a/b)$ to simplify the result. Incidentally, we note that $\arctan(z) = -i\, \mathrm{arctanh}(iz)$, where the areatangens hyperbolicus function gives the inverse of the hyperbolic tangent of the complex number $z$, i.e., $\mathrm{arctanh}(z) = \mathrm{tanh}^{-1}(z)$. Indeed,
{\tt Mathematica} evaluates

\begin{minted}[frame=single,breaklines,fontsize=\small]{mathematica}
In[7]:= (-I)*ArcTanh[I*z]
Out[7]:= ArcTan[z]
\end{minted}

As a heuristic check of the correctness of Eq.~(\ref{eq:arctan-log}), we may compare the Taylor expansions of $\arctan(x)$ and $-\frac{i}{2}\ln \left(\frac{1 + i x}{1 - i x}\right)$ for small $x$. We find

\begin{minted}[frame=single,breaklines,fontsize=\small]{mathematica}
In[8]:= Normal[Series[ArcTan[x], {x, 0, 5}]]
Out[8]:=  x - x^3/3 + x^5/5
\end{minted}

and

\begin{minted}[frame=single,breaklines,fontsize=\small]{mathematica}
In[9]:= Normal[Series[(-I/2)*Log[(1 + I*x)/(1 - I*x)], {x, 0, 5}]]
Out[9]:=  x - x^3/3 + x^5/5
\end{minted}

The expansion can of course be carried to higher orders and we obtain the same result at any expansion order.

Of course, it turns out that the logarithm function appearing in Eq.~(\ref{eq:int-log-def}) is the same function that we have been studying previously. However, this way of thinking about the logarithm immediately offers a path towards generalization: let us now look at integrating combinations of rational functions and the logarithm of $x$. By once more performing partial fraction decomposition, it is clear that we encounter integrals of the form
\begin{equation}
    \int dx\, x^n\, \ln (x)
    \quad\mbox{and}\quad
    \int dx\, \frac{\ln(x)}{(x-a)^n},\qquad n\in\mathbb{N}.
\end{equation}
Again we find that almost all of these integrals can be expressed in terms of rational functions and the logarithm, e.g.,   
\bqa
\int dx\, x^n\, \ln (x) &=& \frac{x^{n+1}}{(n+1)^2}\left[ (n+1)\ln(x)-1\right] +C ,\qquad n\in\mathbb{Z},\; n\ne -1,
\\
\int dx\frac{\ln (x)}{x}  &=&\frac{1}{2}\ln^2(x) + C,
\\
\int dx \frac{\ln (x)}{(x-a)^2} &=&\frac{\ln(x - a)}{a} - \frac{x \ln(x)}{a (x - a)} + C,
\\
\int dx \frac{\ln (x)}{(x-a)^3} &=&\frac{1}{2a}\left[\frac{1}{a - x} + \frac{(x - 2 a) x \ln(x)}{a (x - a)^2} - \frac{\ln(x - a)}{a}\right] + C.
\eqa
However, the integral
\begin{equation}
    \int dx \frac{\ln x}{x-a},\qquad a\ne 0,
\end{equation}
cannot be expressed in terms of just rational functions and logarithms and leads to a genuinely new function. This new function is called \emph{Spence's function}, or more commonly (Euler's) \emph{dilogarithm} and is customarily defined through the integral
\begin{equation}
\mathrm{Li}_2(x) \equiv -\int_0^x \frac{dt}{t}  \ln(1-t).
\label{eq:li2}
\end{equation}
\begin{svgraybox}
The function $\mathrm{Li}_2(x)$ can not be expressed by logarithms and rational functions. 
\end{svgraybox}
For real values of the argument, the Euler dilogarithm behaves as shown in Fig.~\ref{fig-dilog}. The function has a cut along the positive real axis starting at $x=1$ and for $x>1$ develops an imaginary part, as can be seen in the Figure. Approaching $x=-\infty$, $\mathrm{Li}_2(x)$ behaves like $\ln^2(|x|)$, which explains the name dilogarithm.
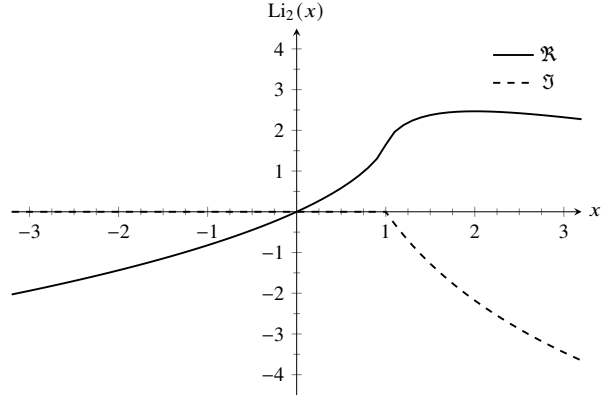
\begin{figure}[tbph] 
\begin{center}

\begin{tikzpicture}[scale=0.9]
  \begin{axis}[
      axis lines=middle,
        width=10cm, height=7cm,
        xmin=-3.2,xmax=3.2,
        ymin=-4.5,ymax=4.5,
        ytick={-4,-3,-2,-1,0,1,2,3,4},
        label style={font=\footnotesize},
        ticklabel style={font=\footnotesize},
        xlabel=$x$,
        ylabel=$\mathrm{Li}_2(x)$,
        x label style={at={(axis description cs:1,0.5)},anchor=west},
        y label style={at={(axis description cs:0.5,1)},anchor=south},
        legend style={draw=none},
        legend style={font=\footnotesize},
        minor x tick num = 3,
        minor y tick num = 1
      ]

    \addplot[
      line width=0.8pt]
    coordinates {
      (-3.2,-2.03042)
      (-3.1,-1.98524)
      (-3.,-1.93938)
      (-2.9,-1.89281)
      (-2.8,-1.84551)
      (-2.7,-1.79744)
      (-2.6,-1.74858)
      (-2.5,-1.6989)
      (-2.4,-1.64835)
      (-2.3,-1.5969)
      (-2.2,-1.54452)
      (-2.1,-1.49115)
      (-2.,-1.43675)
      (-1.9,-1.38127)
      (-1.8,-1.32465)
      (-1.7,-1.26684)
      (-1.6,-1.20778)
      (-1.5,-1.14738)
      (-1.4,-1.08558)
      (-1.3,-1.02228)
      (-1.2,-0.957405)
      (-1.1,-0.890838)
      (-1.,-0.822467)
      (-0.9,-0.752163)
      (-0.8,-0.679782)
      (-0.7,-0.605158)
      (-0.6,-0.528107)
      (-0.5,-0.448414)
      (-0.4,-0.365833)
      (-0.3,-0.280074)
      (-0.2,-0.1908)
      (-0.1,-0.0976052)
      (0.,0.)
      (0.1,0.102618)
      (0.2,0.211004)
      (0.3,0.32613)
      (0.4,0.449283)
      (0.5,0.582241)
      (0.6,0.727586)
      (0.7,0.889378)
      (0.8,1.07479)
      (0.9,1.29971)
      (1.,1.64493)
      (1.1,1.962)
      (1.2,2.12917)
      (1.3,2.24089)
      (1.4,2.31907)
      (1.5,2.3744)
      (1.6,2.41313)
      (1.7,2.43935)
      (1.8,2.45588)
      (1.9,2.46472)
      (2.,2.4674)
      (2.1,2.46506)
      (2.2,2.45859)
      (2.3,2.44869)
      (2.4,2.43594)
      (2.5,2.42079)
      (2.6,2.40362)
      (2.7,2.38473)
      (2.8,2.36439)
      (2.9,2.34281)
      (3.,2.32018)
      (3.1,2.29665)
      (3.2,2.27235)
    };
    \addplot[dashed,
      line width=0.8pt]
    coordinates {
      (-3.2,0.)
      (-3.1,0.)
      (-3.,0.)
      (-2.9,0.)
      (-2.8,0.)
      (-2.7,0.)
      (-2.6,0.)
      (-2.5,0.)
      (-2.4,0.)
      (-2.3,0.)
      (-2.2,0.)
      (-2.1,0.)
      (-2.,0.)
      (-1.9,0.)
      (-1.8,0.)
      (-1.7,0.)
      (-1.6,0.)
      (-1.5,0.)
      (-1.4,0.)
      (-1.3,0.)
      (-1.2,0.)
      (-1.1,0.)
      (-1.,0.)
      (-0.9,0.)
      (-0.8,0.)
      (-0.7,0.)
      (-0.6,0.)
      (-0.5,0.)
      (-0.4,0.)
      (-0.3,0.)
      (-0.2,0.)
      (-0.1,0.)
      (0.,0.)
      (0.1,0.)
      (0.2,0.)
      (0.3,0.)
      (0.4,0.)
      (0.5,0.)
      (0.6,0.)
      (0.7,0.)
      (0.8,0.)
      (0.9,0.)
      (1.,0.)
      (1.1,-0.299426)
      (1.2,-0.57278)
      (1.3,-0.824242)
      (1.4,-1.05706)
      (1.5,-1.27381)
      (1.6,-1.47656)
      (1.7,-1.66702)
      (1.8,-1.84659)
      (1.9,-2.01644)
      (2.,-2.17759)
      (2.1,-2.33086)
      (2.2,-2.47701)
      (2.3,-2.61666)
      (2.4,-2.75037)
      (2.5,-2.87861)
      (2.6,-3.00183)
      (2.7,-3.12039)
      (2.8,-3.23464)
      (2.9,-3.34489)
      (3.,-3.45139)
      (3.1,-3.5544)
      (3.2,-3.65415)
    };
%    \legend{$Re\mathrm{Li}_2(x)$,$Im\mathrm{Li}_2(x)$}
    \legend{$\Re$,$\Im$}
  \end{axis}
\end{tikzpicture}
\end{center}
\caption[]{
The real and imaginary part of the dilogarithm $\mathrm{Li}_2(x)$ at real values.}
\label{fig-dilog}
\end{figure}
%

%
%
%\begin{tips}{The general dilogarithmic integral}
\subsection{The General Dilogarithmic Integral}
\label{sec:gen-Li2-int}

Having introduced the dilogarithm function, we can now solve more general integrals such as
\begin{equation}
    \int \frac{dx}{m x+b}\ln(n x+a).
    \label{eq:gen-Li2}
\end{equation}
In order to bring this integral to the form were we can apply Eq.~(\ref{eq:li2}), we need to perform some transformations of the integration variables. First, let us set $m x+b=t$ and change the variable of integration from $x$ to $t$. We find
\begin{equation}
    \int \frac{dx}{m x+b}\ln(n x+a)
    = \frac{1}{m}\int \frac{dt}{t} \ln\left(\frac{a m - b n + n t}{m}\right).
\end{equation}
Next, we use $t=-\frac{a m - b n}{n} u$. This transformation is chosen such that the argument of the logarithm factorizes,
\begin{equation}
    \frac{1}{m}\int \frac{dt}{t} \ln\left(\frac{a m - b n + n t}{m}\right)
    = \frac{1}{m}\int \frac{du}{u} \ln\left(\frac{(a m - b n)(1-u)}{m}\right).
\label{eq:ln-func-rel-ex}
\end{equation}
Now we come to an interesting point. In order to proceed, we must relate the logarithm in the integrand to $\ln(1-u)$. This is formally straightforward, since the logarithm of a product is the sum of the logarithms.\footnote{Here we put aside the possible subtleties associated with imaginary parts, see Eq.~(\ref{eq:ln_wtimesz}) and the discussion there.} However, the general point to take away is that when we define new functions, we must also find ways of bringing expressions involving those new functions to standard forms, so that the definitions can actually be applied. We will come back to this point later and also in section~\ref{sec:3anal}. Returning to our integral, let us proceed formally and decompose the logarithm,
\begin{equation}
    \frac{1}{m}\int \frac{du}{u} \ln\left(\frac{(a m - b n)(1-u)}{m}\right)
    =
    \frac{1}{m}\int \frac{du}{u} \left[\ln\left(\frac{a m - b n}{m}\right)
    +\ln(1-u)\right].
\end{equation}
We see that we can now apply Eqs.~(\ref{eq:int-log-def})~and~(\ref{eq:li2}) to perform the integrations and we find the result
\begin{equation}
    \frac{1}{m}\int \frac{du}{u} \left[\ln\left(\frac{a m - b n}{m}\right)
    +\ln(1-u)\right] 
    =
    \frac{1}{m}\left[\ln\left(\frac{a m - b n}{m}\right)\ln u - \mathrm{Li}_2(u)\right] 
    + C.
\end{equation}
To be absolutely clear, here we have used the fact that Eq.~(\ref{eq:li2}) implies that the 
primitive function of $-\frac{\ln(1-u)}{u}$ is the dilogarithm, $\mathrm{Li}_2(u)$. Last, we can replace $u$ with our original variable $x$. Using $u = \frac{n (m x+b)}{-a m + b n}$, we find the final result
\begin{equation}
	\int \frac{dx}{m x+b}\ln(n x+a) 
	=
	\frac{1}{m}
	\left[\ln\left(\frac{a m - b n}{m}\right)\ln \left(\frac{n (m x+b)}{-a m + b n}\right)
	- \mathrm{Li}_2\left(\frac{n (m x+b)}{-a m + b n}\right)\right] + C.
\label{eq:gen-Li2-int}
\end{equation}

Finally, we note that in the \math{} language, the dilogarithm function $\mathrm{Li}_2(x)$ is represented by {\tt PolyLog[2,x]} and we can evaluate our integral immediately,

\begin{minted}[frame=single,breaklines,fontsize=\small]{mathematica}
In[10]:= Integrate[Log[n*x + a]/(m*x + b), x]
Out[10]:= (Log[(n*(b + m*x))/((-a)*m + b*n)]*Log[a + n*x] + PolyLog[2, (m*(a + n*x))/(a*m - b*n)])/m
\end{minted}

Apparently the result is different from what we obtained in Eq.~(\ref{eq:gen-Li2-int})! However there is no contradiction. It can be shown using some properties of the dilogarithm (namely Eq.~(\ref{eq:li2omz}) below) that the difference is a constant, independent of $x$.\footnote{At least for parameters where both expressions are real. Recall we were not careful with the imaginary parts in our derivation. See Problem~\ref{prob:Li2} for a discussion of the general case.}
%\end{tips}
%
%
\subsection{Classical Polylogarithms}
Repeating the arguments that have lead us from the logarithm to the dilogarithm, we arrive at the \emph{classical polylogarithms} $\mathrm{Li}_n(x)$. Indeed, it is easy to check that by a proper application of integration by parts, all integrals of the form
\begin{equation}
\int dx\, x^n \mathrm{Li}_2(x),\qquad\mbox{and}\qquad
\int \frac{dx}{(x-a)^n} \mathrm{Li}_2(x),\qquad n\ne -1,
\end{equation}
can be evaluated in terms of rational functions, logarithms and dilogarithms. However, the integral
\begin{equation}
\int \frac{dx}{x}  \mathrm{Li}_2(x) \equiv \mathrm{Li}_3(x) + C,
\label{eq:li3}
\end{equation}
leads to a function, the trilogarithm $\mathrm{Li}_3(x)$, that cannot be expressed with the functions already introduced. In this way, we are lead to introduce the \emph{classical polylogarithm of weight $n$}, $\mathrm{Li}_n(x)$ by the recursive definition\footnote{In the next section, we will see that Eq.~(\ref{eq:lin}) is just a special case of the iterated integral that defines a class of functions known as multiple polylogarithms. However, $\mathrm{Li}_2(x)$ and its generalizations $\mathrm{Li}_n(x)$ were historically the first functions of this class to be studied, thus in order to distinguish them from the more general cases to be studied below, we refer to them as \emph{classical polylogarithms}.}
\begin{equation}
\boxed{
\mathrm{Li}_n(x) \equiv \int_0^x \frac{dt}{t}  \mathrm{Li}_{n-1}(t).}
\label{eq:lin}
\end{equation}
The above equation holds also for $n=2$ if we define
\begin{equation}
\mathrm{Li}_1(x) \equiv -\ln(1-x).
\label{eq:li1}
\end{equation}

We have defined classical polylogarithms above via the integral in Eq.~(\ref{eq:lin}). However, it turns out that they may also be represented as simple infinite sums. This is quite relevant in connection with the Mellin-Barnes method, since as we have already mentioned, one possible way of solving \MB{} integrals is by converting them to sums using Cauchy's residue theorem. Hence let us briefly turn to this alternative way of defining classical polylogarithms. We begin by noting that
\begin{equation}
-\ln(1-x) = \sum_{k=1}^{\infty} \frac{x^k}{k}.
\end{equation}
The sum on the right hand side is convergent for $|x|<1$ and in fact can be used as a starting point to define the logarithm function in yet one more approach. Now, it can be shown that the simple generalization
\begin{equation}
\mathrm{Li}_n(x) = \sum_{k=1}^{\infty} \frac{x^k}{k^n}
\label{eq:lin-sum}
\end{equation}
leads precisely to the classical polylogarithm of weight $n$, which we have introduced in Eq.~(\ref{eq:lin}). Both representations are important and can be used to good effect when for solving \MB {} integrals analytically, as will be discussed in detail in chapter~\ref{chapter-MBanal}.

Classical polylogarithms obey a number of important functional relations. For example, it is always possible to relate $\mathrm{Li}_n(1/z)$ to $\mathrm{Li}_n(z)$,
\begin{align}
\mathrm{Li}_1\left(\frac{1}{z}\right) 
	&= -\ln\left(1-\frac{1}{z}\right) 
	= -\ln(1-z) + \ln(-z)
	= \mathrm{Li}_1(z) + \ln(-z),
\label{eq:li1z-inv}
\\
\mathrm{Li}_2\left(\frac{1}{z}\right) 
	&=-\mathrm{Li}_2(z) - \frac{1}{2} \ln^2(-z) - \zeta(2),
\label{eq:li2z-inv}
\\
\mathrm{Li}_3\left(\frac{1}{z}\right) 
	&=\mathrm{Li}_3(z) - \frac{1}{6} \ln^3(-z) - \zeta(2)\ln(-z),	
\label{eq:li3z-inv}
\end{align}
and so on. In general $\mathrm{Li}_n(z) + (-1)^n \mathrm{Li}_n(1/z)$ can be written as an $n$-th degree polynomial in $\ln(-z)$. Some other useful relations for classical polylogarithms are
\begin{align}
\litwo(1-z) 
	&=  -\litwo(z) +\zeta(2) - \ln(z)\ln(1-z),
\label{eq:li2omz}
\\
\litwo\left(1-\frac{1}{z}\right) 
	&= -\litwo(1-z) - \frac{1}{2}\ln^2(z),
\label{eq:li2om1z}
\intertext{and}
\litri\left(1-\frac{1}{z}\right) 
	&=  -\litri(z)  - \litri(1-z) + \zeta(3) + \frac{1}{6}\ln^3(z) 
	+ \zeta(2)\ln(z) 
\\	&- 
	\frac{1}{2}\ln^2(z)\ln(1-z).
\label{eq:li3om1z}
\end{align}
The square relation
\begin{equation}
\mathrm{Li}_n(z) + \mathrm{Li}_n(-z) = 2^{1-n} \mathrm{Li}_n(z^2)
\label{eq:lin-sq}
\end{equation}
is also sometimes employed in practical calculations. These relations are often useful for bringing polylogarithmic integrals to the standard form of Eq.~(\ref{eq:lin}) as well as for transforming the arguments to specific regions for numerical evaluation. For example, the convergence of the simple series expansion of the dilogarithm in Eq.~(\ref{eq:lin-sum}) is already quite bad for $|z|>1/2$. Hence an efficient evaluation of the dilogarithm (and more generally of $\mathrm{Li}_n(z)$) transforms the argument $z$ to a region where the modulus and real part are bound: $|z|\le 1$ and $\Re(z) \le 1/2$ using e.g. Eqs.~(\ref{eq:li2z-inv})~and~(\ref{eq:li2omz}).

Turning to specific values of polylogarithms, it is clear from the definition of $\lin(z)$ either as an integral, Eq.~(\ref{eq:lin}), or a sum, Eq.~(\ref{eq:lin-sum}), that $\mathrm{Li}_n(0) = 0$. Moreover, $\mathrm{Li}_n(1) = \zeta(n)$ follows immediately from the representation of $\lin(z)$ as a sum. Some other special values are also easy to determine, e.g. for the dilogarithm we have
\begin{align}
\mathrm{Li}_2\left(  0\right) &= 0,
\label{eq-r3}
\\
\mathrm{Li}_2\left(  1\right) &= \frac{\pi^2}{6} = \zeta(2) \approx 1.64493, \label{eq:zeta2}
\\
\mathrm{Li}_2\left( -1\right) &= - \frac{1}{2}\mathrm{Li}_2\left(   1\right) 
	= -\frac{\pi^2}{12} \approx -0.822467,
\\
\mathrm{Li}_2\left(  \frac{1}{2}\right) &= \frac{\pi^2}{12}-\frac{1}{2}\ln^2(2) 
	\approx 0.582241,
\\
 \mathrm{Li}_2\left(  \pm i\right) &= \pm i G -\frac{\pi^2}{48} 
 	\approx -0.205617 \pm 0.915965 i. 
\label{eq:catalan}
\end{align}
The constant $G$ in Eq.~(\ref{eq:catalan}) above is the Catalan's constant, $G=\sum\limits_{n=0}^\infty \frac{(-1)^n}{(2n+1)^2}$. It is conjectured that Catalan's number is irrational and transcendental~\cite{scott_2002}.

In addition to the functional relations discussed above, classical polylogarithms admit a number of other representations as integrals and series. For example, the integral
\begin{equation}
%\lin(z) = S_{n-1,1}(z) = 
\frac{(-1)^{n-1}}{(n-2)!}\int_{0}^{1} \frac{dt}{t}\ln^{n-2}(t)\ln(1-zt) \label{eq:nielsen}
\end{equation}
simply evaluates to the classical polylogarithm $\mathrm{Li}_n(z)$. Its straightforward generalization, 
\begin{equation}
S_{n,p}(z) 
 	= \frac{(-1)^{n+p-1}}{(n-1)!p!}\int_{0}^{1} \frac{dt}{t}\ln^{n-1}(t)\ln^p(1-zt) \label{eq:nielsenS}
\end{equation}
defines the so-called \emph{Nielsen generalized polylogarithm} $S_{n,p}(z)$. Clearly we have $S_{n-1,1}(z) = \lin(z)$. We note that in \math{} the Nielsen generalized polylogarithm function $S_{n,p}(z)$ is represented by {\tt PolyLog[n,p,z]}, while {\tt PolyLog[n,z]} gives the classical polylogarithm of weight $n$.

Turning to series expansions, an efficient  method may be found in~\cite{Vollinga:2004sn} and~\cite{'tHooft:1979xw,actis:2008br}. 
 
For $\litwo(z)$ we have~\cite{Vollinga:2004sn}
\begin{eqnarray}
\litwo(z) &=& \sum_{j=0}^{\infty} \frac{B_{j}}{(j+1)!}\left[-\ln(1 - z)\right]^{j+1}
 \\&=&
-\ln(1 - z) - \frac{1}{4}\ln^2(1 - z)
 + 4\pi \sum_{j=1}^{\infty} \zeta(2j)\frac{(-1)^j}{2j+1}\left[\frac{\ln(1 - z)}{2\pi}\right]^{2j+1}. \nonumber
\end{eqnarray}
Here the $B_j$ are Bernoulli numbers, $B_0=1, B_1=-1/2$, etc. In \math{}
$B_n$ are represented as {\tt BernoulliB[n] = BernoulliB[n, 0]} and historical facts about these numbers can be found in the chapter ``Gamma as a Decimal'' in~\cite{havilgamma}. This expansion with Bernoulli numbers ensure rapid convergence. One of many possible ways of defining the Bernoulli numbers $B_n$ is throuth a generating function,
\bea
\frac{t}{e^t - 1} = \sum_{m=0}^{\infty} B_m \frac{t^m}{m!}
\eea
%Table[4 Pi* Zeta[2 m] ((-1)^m /(2 m + 1))*(Log[1 - z]/(2 Pi))^(2 m + 1), {m, 1, 4}]
%
Other useful series expansions for $\lin(z)$ are also given in~\cite{Vollinga:2004sn} and we reproduce one here for the special case of $n=3$,
\begin{equation}
  \litri(z) = 
  \sum_{j=0}^{\infty} \frac{C_3(j)}{(j+1)!}\left[-\ln(1 - z)\right]^{j+1},
\end{equation}
with
\begin{equation}
	C_3(j)= \sum_{k=0}^j \binom{j}{k} \frac{B_{j-k} B_k}{1+k},
\end{equation}
so that $C_3(0)=1$, etc. Using these expansions for $\litwo(z)$ and $\litri(z)$ we observe typically that retaining $n$ terms in the summation gives $n\pm 1$ digit accuracy. We mention that just as for the logarithm, the evaluation of polylogarithms on their cuts (the positive real axis beginning at $z=1$) requires the adoption of a particular convention and hence must be treated with special care. We refer to~\cite{Vollinga:2004sn} for a discussion of these issues.

\begin{comment}
The expansions for $\litwo$ may be derived using an integral representation:
\bea
\mathrm{Li}_2\left( x\right) &=& - \int_0^1 \frac{dt}{t} \ln(1-xt)
~=~  - \int_0^x \frac{dt}{t} \ln(1-t).
\eea
The expansions for $\litwo$ may be derived using an integral representation:
\bea
\mathrm{Li}_2\left( x\right) &=& - \int_0^1 \frac{dt}{t} \ln(1-xt)
~=~  - \int_0^x \frac{dt}{t} \ln(1-t).
\eea
\end{comment}

\begin{tips}{One more Dilogarithmic Integral}
A special integral frequently met in one-loop Feynman integral calculations is (see~\cite{Passarino:1978jh}):
\bea
\int_0^1 \frac{dx}{x-x_0} \left[ \ln \left(x-x_A \right) - \ln \left( x_0-x_A\right) \right]
= \mathrm{Li}_2\left(  \frac{x_0}{x_0-x_A}\right)
- \mathrm{Li}_2\left(  \frac{x_0-1}{x_0-x_A}\right),
\nonumber\\
\eea
which is valid for arbitrary complex $x_A$ and real $x_0$.
The formula demonstrates the need of a subtraction at $x=x_0$ in order to make the integral well-defined for $0\leq x_0 \leq 1$.
\end{tips}

%\commgs{Keep this comment here?}
The functions and some expansions which have been discussed here are important in systematic order by order $\eps$ studies of \texttt{FI}~\cite{Devoto:1983tc}. Nowadays suitable packages exist in \math{} or \texttt{Fortran}~\cite{Vollinga:2004sn,Maitre:2005uu,Naterop:2019xaf,Gehrmann:2001pz,Buehler:2011ev}, see Appendix~\ref{introA} for details.

\section{Multiple Polylogarithms and Beyond\label{sec:hplsmpls}}
\label{sec:MPLs}

The logarithm as well as all of its generalizations discussed so far are just special cases of a set of functions known as multiple polylogarithms. Similarly to the classical polylogarithms, multiple polylogarithms\footnote{MPLs are also known as generalized polylogarithms (GPLs), see appendix~\ref{app:mplsgpls} for details and packages.}  (MPLs) can be defined recursively for $n\ge 1$, via the iterated integral~\cite{Goncharov:1998kja,Goncharov:2001iea}
\begin{equation}
G(a_1,\ldots,a_n;z) = \int_0^z \frac{dt}{t-a_1} G(a_2,\ldots,a_n;t)\,,
\label{eq:G-def}
\end{equation}
with $G(z) = G(;z) = 1$.\footnote{The above definition was already present in the works of Poincar\'e, Kummer and Lappo-Danilevskij~\cite{Lap34} as ``hyperlogarithms'' as well as implicitly in Chen's work on iterated integrals~\cite{Chen77}.} Here the $a_i\in\mathbb{C}$ are constants and $z$ is a complex variable. For the special case when all of the $a_i$'s are zero, we define
\begin{equation}
G(\vec{0}_n;z) = \frac{1}{n!} \ln^n z, 
\label{eq:G-def-0}
\end{equation}
where $\vec{0}_n$ denotes a set of $n$ zeros. The number $n$ is called the {\it weight} of the polylogarithm. It is clear from the definition that MPLs contain the ordinary logarithm as well as the classical polylogarithms as special cases. In particular we have
\begin{align}
G(\vec{a}_n;z) &= \frac{1}{n!}\ln^n \left(1 - \frac{z}{a}\right),
\\
G(\vec{0}_{n-1},1;z) &= -\mathrm{Li}_{n}(z),
\label{eq:G-lnn-Lin}
\end{align}
where $\vec{a}_n=(a,\ldots,a)$ denotes a sequence of $a$'s of length $n$. We note the definition above corresponds to what is commonly used in the physics literature (see e.g.~the usage in the {\tt PolyLogTools} pagkage, ref.~\cite{Duhr:2019tlz}), however the notation for MPLs used in the mathematical literature (e.g.~\cite{Goncharov:2005sla}) differs slightly from this:
\begin{equation}
I(a_0;a_1,\ldots,a_n;a_{n+1}) = \int_{a_0}^{a_{n+1}} \frac{dt}{t-a_n}
	I(a_0;a_1,\ldots,a_{n-1};t)
\label{eq:I-def}
\end{equation}
and $I(a_0;;a_1)=1$. The two definitions are related by (note the reversal of the arguments)
\begin{equation}
G(a_n,\ldots,a_1;a_{n+1}) = I(0,a_1,\ldots,a_n;a_{n+1}).
\end{equation}
The iterated integrals defined by Eq.~(\ref{eq:I-def}) are slightly more general than those defined in Eq.~(\ref{eq:G-def}), since they allow for a generic base point of integration $a_0$. Nevertheless, it is easy to see that every integral with a generic base point can be converted into a combination of integrals with base point zero. Even so, some properties of MPLs are more straightforward to express in the $I$ notation and its use in the mathematical literature justifies its mention here.

In addition to the integral representation of Eq.~(\ref{eq:lin}), classical polylogarithms can also be defined by their series representation as in Eq.~(\ref{eq:lin-sum}). Above we have presented the generalization of the integral definition to MPLs. However, there is also a way to generalize the series definition,
\begin{equation}
\begin{split}
\mathrm{Li}_{m_1,\ldots,m_k}(z_1,\ldots,z_k) &= \sum_{0<i_1<i_2\cdots < i_k} 
	\frac{z_1^{i_1}}{i_1^{m_1}}\frac{z_2^{i_2}}{i_2^{m_2}}\cdots\frac{z_k^{i_k}}{i_k^{m_k}}
\\ &=
	\sum_{i_k=1}^{\infty} \frac{z_k^{i_k}}{i_k^{m_k}} \sum_{i_{k-1}=1}^{i_k-1} \ldots
	\sum_{i_1=1}^{i_2-1}\frac{z_1^{i_1}}{i_1^{m_1}}.
\end{split}
\label{eq:Li-def}
\end{equation} 
This definition makes sense whenever the sums converge (e.g. for $|z_i|<1$). The number $k$ is called the \emph{depth} of the MPL. The $G$ and $\mathrm{Li}$ functions define essentially the same class of functions and are related by 
\begin{equation}
\mathrm{Li}_{m_1,\ldots,m_k}(z_1,\ldots,z_k) = 
	(-1)^k 
	G\left(\vec{0}_{m_k-1},\frac{1}{z_k},\ldots,\vec{0}_{m_1-1},\frac{1}{z_1\ldots z_k}
	;1\right)
\label{eq:Li-to-G}
\end{equation}
or alternatively ($a_i\ne 0$)
\begin{equation}
    G(\vec{0}_{m_1-1},a_1,\ldots,\vec{0}_{m_k-1},a_k;z) 
    = (-1)^k \mathrm{Li}_{m_k,\ldots,m_1}\left(
    \frac{a_{k-1}}{a_k},\ldots,\frac{a_1}{a_2},\frac{z}{a_1}\right).
\label{eq:G-to-Li}
\end{equation}
{We will make use of both the integral and sum representation of MPLs when we discuss various approaches of obtaining analytic solutions to \MB{} integrals in chapter~\ref{chapter-MBanal}.}

Regarding the basic properties of MPLs, we note first of all $G(a_1,\ldots,a_n;z)$ is divergent whenever $z=a_1$, which is easy to check using the integral representation in Eq.~(\ref{eq:G-def}). Similarly, $G(a_1,\ldots,a_n;z)$ is analytic at $z=0$ (i.e., is given by a convergent power series around $z=0$) whenever $a_n\ne 0$, which is consistent with the series representation in Eq.~(\ref{eq:Li-def}). Due to the singularities at $z=a_i$ in the integral representation, MPLs in general have a very complicated branch cut structure. 

Second, if the rightmost index $a_n$ of $G(a_1,\ldots,a_n;z)$ is non-zero, then the function is invariant under the simultaneous rescaling of all of its arguments for any non-zero complex number $k$,
\begin{equation}
G(k a_1,\ldots, k a_n;k z) = G(a_1,\ldots, a_n;z),\qquad a_n\ne 0.
\end{equation}
MPLs also satisfy the so-called \emph{H\"older convolution}. Whenever $a_1\ne 1$ and $a_n\ne 0$, we have for all non-zero complex numbers $p$ that
\begin{equation}
G(a_1,\ldots,a_n;1) = 
	\sum_{k=0}^n (-1)^k G\left(1-a_k,\ldots,1-a_1;1-\frac{1}{p}\right)
	G\left(a_{k+1},\ldots,a_n;\frac{1}{p}\right).
\end{equation}
These examples make it clear that MPLs satisfy many functional relations. These of course include the functional relations for classical polylogarithms. For example, the H\"older convolution for $a_1=0$ and $a_2=1$ reduces simply to Eq.~(\ref{eq:li2omz}) with $z=\frac{1}{p}$.

Third, the product of two MPLs of weight $n_1$ and $n_2$ with the same argument $z$ can be written as a linear combination of MPLs of weight $n_1+n_2$ with argument $z$,
\begin{equation}
G(a_1,\ldots,a_{n_1};z)G(a_{n_1+1},\ldots,a_{n_1+n_2};z) 
	= \sum_{\sigma\in\Sigma(n_1,n_2)} 
	G(a_{\sigma(1)},\ldots,a_{\sigma(n_1+n_2)};z),
\label{eq:G-shuffle}
\end{equation}
where $\Sigma(n_1,n_2)$ denotes the set of \emph{shuffles} on $n_1+n_2$ elements. A shuffle is simply a permutation of the set $(a_1,\ldots,a_{n_1+n_2})$ that leaves the orderings of $(a_1,\ldots,a_{n_1})$ and $(a_{n_1+1},\ldots,a_{n_1+n_2})$ unchanged. Formally we have
\begin{equation}
\begin{split}
\Sigma(n_1,n_2) = \{ \sigma \in S_{n_1+n_2}|& \sigma^{-1}(1) < \ldots < \sigma^{-1}(n_1)
\\
&\mbox{and}\quad \sigma^{-1}(n_1+1) < \ldots < \sigma^{-1}(n_1+n_2)\},
\end{split}
\end{equation}
where $S_{n_1+n_2}$ denotes the group of permutations on $n_1+n_2$ elements. 
\begin{tips}{The Shuffle Product of MPLs}
Let us illustrate the concept of the shuffle product by a few examples
\begin{align*}
G(a;z)G(b;z) &= G(a,b;z) + G(b,a;z),
\\ 
G(a;z)G(b,c;z) &= G(a,b,c;z) + G(b,a,c;z) + G(b,c,a;z),
\\
G(a,b;z)G(c,d;z) &= G(a,b,c,d;z) + G(a,c,b,d;z) + G(a,c,d,b;z) 
\\
	&+ G(c,a,b,d;z) + G(c,a,d,b;z) + G(c,d,a,b;z).
\end{align*}
\end{tips}
This property turns the set of MPLs into a \emph{shuffle algebra}, i.e., a vector space equipped with shuffle multiplication. The vector space structure simply means that we can take finite linear combinations of MPLs with complex coefficients, while the shuffle product is defined in Eq.~(\ref{eq:G-shuffle}) above. The proof of the shuffle algebra relations is relatively straightforward and follows from the definition of MPLs as iterated integrals. In Problem~\ref{prob:G-shuffle}, you are asked to verify the shuffle algebra relation for the product of two MPLs, both of weight one.

Besides evaluating products of MPLs with the same argument, the shuffle algebra is also useful for extracting the singularities of MPLs at $z=0$. Indeed, as noted above, $G(a_1,\ldots,a_n;z)$ is analytic (i.e., can be expanded into a power series) at $z=0$ if $a_n\ne 0$. In case $a_n=0$, we can employ the shuffle algebra to rewrite $G(a_1,\ldots,a_n;z)$ in terms of functions whose rightmost index is non-zero and MPLs of the form $G(\vec{0}_k;z)=\frac{1}{k!}\ln^k(z)$. This procedure thus makes the (logarithmic) singularities around $z$ manifest.
\begin{tips}{Extracting Singularities of MPLs}
To illustrate the extraction of singularities around $z=0$ using the shuffle algebra, consider e.g. the function $G(a,0,0;z)$. Clearly the rightmost index is zero and hence the function has logarithmic singularities at $z=0$. Using the shuffle algebra, we compute
\begin{eqnarray}
&&G(a,0,0;z) = G(0,0;z)G(a;z) - G(0,a,0;z) - G(0,0,a;z)
\nn \\
& =	&  G(0,0;z)G(a;z) - G(0,0,a;z) - [G(0,a;z)G(0;z) - 2G(0,0,a;z)]
\nn \\
&=	&  G(0,0;z)G(a;z) - G(0;z)G(0,a;z) + G(0,0,a;z).
%\end{split}
\end{eqnarray}
\end{tips}

Fourth, there is another algebra structure defined on MPLs. As noted above, the shuffle algebra structure follows from the definition of MPLs as iterated integrals. However, their definitions as sums induces another type of multiplication. We will have more to say about this in chapter~\ref{sec:1anal}, however, to illustrate the basic idea, consider the product of two MPLs of depth one,
\begin{equation}
\begin{split}
\mathrm{Li}_1(z_1)\mathrm{Li}_1(z_2) &=
	\sum_{i_1=1}^{\infty} \frac{z_1^{i_1}}{i_1}
	\sum_{i_2=1}^{\infty} \frac{z_2^{i_2}}{i_2}
	= \sum_{0<i_1,i_2} \frac{z_1^{i_1} z_2^{i_2}}{i_1 i_2}
\\
	&= \sum_{0<i_1<i_2} \frac{z_1^{i_1} z_2^{i_2}}{i_1 i_2}
	+ \sum_{0<i_2<i_1} \frac{z_1^{i_1} z_2^{i_2}}{i_1 i_2}
	+ \sum_{0<i_1=i_2} \frac{z_1^{i_1} z_2^{i_2}}{i_1 i_2}
\\
	&= \mathrm{Li}_{1,1}(z_1,z_2) + \mathrm{Li}_{1,1}(z_2,z_1) 
	+ \mathrm{Li}_{2}(z_1 z_2)
\end{split}
\label{eq:Li-stuffle}
\end{equation}
In the second line, we have simply organized the double sum over $i_1$ and $i_2$ into three sums: in the first, $i_1<i_2$, in the second $i_1>i_2$, while the last sum represents the ``diagonal'' elements of the double sum, where $i_1=i_2$. These sums are then immediately recognized as $\mathrm{Li}$ functions, see Eq.~(\ref{eq:Li-def}).\footnote{To be absolutely clear, we note that the last sum on the second line of Eq.~(\ref{eq:Li-stuffle}) is simply $
{\displaystyle \sum_{0<i_1=i_2}} \frac{z_1^{i_1} z_2^{i_2}}{i_1 i_2} = 
{\displaystyle \sum_{0<i_1}} \frac{(z_1 z_2)^{i_1}}{i_1^2} = \mathrm{Li}_{2}(z_1 z_2)
$.} 
Products of MPLs of higher depth can be handled in a similar fashion. The algebra generated this way is called the \emph{stuffle algebra} (also called the \emph{quasi-shuffle algebra}) on MPLs. We emphasize that the stuffle algebra structure is completely independent of the shuffle algebra. The former is related to the definition of MPLs as iterated integrals, while the latter is connected with the sum representation. We mention that both shuffle and stuffle products preserves the weight. This means that the shuffle or stuffle product of two MLPs of weight $w_1$ and $w_2$ is a linear combination of MPLs of weight $w_1+w_2$, as seen in the examples above. We say that the shuffle and stuffle algebras are \emph{graded} by the weight. However, notice that the stuffle product does not preserve the depth of the sums, as evident from Eq.~(\ref{eq:Li-stuffle}). Instead the depths of terms in the product are bounded by the sum of depths of the factors. In this case we say that the stuffle algebra is \emph{filtered} by the depth.  

Fifth, while the shuffle and stuffle algebras certainly generate many functional relations among MPLs, there are many more complicated relations that do not arise in this way. In particular, relations that change the arguments of the functions, such as he H\"older convolution, cannot be covered. However, such functional equations are very often crucial in evaluating Feynman integrals in terms of MPLs. We have seen a very simple example of this phenomenon already in Eq.~(\ref{eq:ln-func-rel-ex}), where we needed to use the functional relation $\ln(ab)=\ln a + \ln b$ to proceed with the calculation. Indeed, in Feynman integral calculations, the following problem often arises. We would like to evaluate a multi-dimensional integral, but after integration over some variable, say $t_1$, we find MPLs of the form $G(a_1(t_2),\ldots,a_n(t_2);z(t_2))$, where the $a_i$'s as well as $z$ are functions of $t_2$, the next integration variable. In order to proceed, we must first bring this function to ``canonical form'', i.e., one where $t_2$ only appears in the form $G(a'_1,\ldots,a'_n,t_2)$ so that we perform the integration over $t_2$ using the definition of MPLs, Eq.~(\ref{eq:G-def}). As this rewriting will generally involve the change of arguments of the functions, we cannot hope to proceed using only the shuffle and stuffle algebra structures, but rather we need a general and flexible framework for deriving such relations. This framework is given by the \emph{Hopf algebra} structure of MPLs. Even a cursory description of this framework goes beyond the scope of this chapter, however we mention that it can be used to find, and in some sense circumvent, functional equations among MPLs. We refer to~\cite{Duhr:2014woa,Weinzierl:2022eaz} for discussions of this topic.

Finally, as the name and notation imply, classical polylogarithms, Nielsen generalized polylogarithms and the harmonic polylogarithms of Remiddi and Vermaseren~\cite{Remiddi:1999ew} are all special cases of multiple polylogarithms. First of all, classical polylogarithms are clearly just depth one MPLs and from Eq.~(\ref{eq:Li-to-G}), we find
\begin{equation}
    \mathrm{Li}_n(z) = -G(\vec{0}_{n-1},1;z).
\label{eq:Lin-to-G}
\end{equation}
Nielsen generalized polylogarithms can be expressed as
\begin{equation}
    S_{n,p}(z) = (-1)^p G(\vec{0}_{n},\vec{1}_{p};z) 
    = \mathrm{Li}_{1,\ldots,1,n+1}(\underbrace{1,\ldots,1}_{p-1},z).
\label{eq:Snp-to-Li-and-G}
\end{equation}
Finally, harmonic polylogarithms~\cite{Remiddi:1997ny} (HPLs) correspond the case when all $a_i$ are $0$ or $\pm 1$. HPLs are equal to MPLs up to a sign and we have ($a_i\in\{0,\pm 1\}$)
\begin{equation}
    H_{a_1,\ldots,a_k}(z) = (-1)^p  G(a_1,\ldots,a_k;z).
\label{eq:H-to-G}
\end{equation}
Here $p$ is the number of elements of the set of indices $\{a_1,\ldots,a_n\}$ which are equal to $+1$. Thus, we have e.g. $H_{1,1,0}(z) = (-1)^2 G(1,1,0;z) = G(1,1,0;z)$ since two of the indices are $+1$, but 
$H_{1,-1,0}(z) = (-1)^1 G(1,-1,0;z) = -G(1,-1,0;z)$ as now only one index is equal to $+1$. HPLs can also be written in terms of $\mathrm{Li}$ using Eq.~(\ref{eq:G-to-Li}), however, the relation is best expressed by using a different notation for HPLs. The notation we have introduced above, where each index is either $0$ or $\pm 1$ is referred to as the `a'-notation. In contrast, the `m'-notation is given as follows. First, let us assume that the rightmost index in the `a'-notation is non-zero. Then, starting from the vector of indices $(a_1,\ldots,a_k)$, we construct a new vector $(m_1,\ldots,m_l)$
, where we drop all zeros and increase the absolute value of each non-zero element by the number of zeros immedaitely preceeding it. The sign of each index remains unchanged. So in this way e.g. $(0,1,1)$ becomes $(2,1)$ while e.g. $(0,0,1,0,-1)$ becomes $(3,-2)$ and so on. We can then extend this notation to allow for any number of rightmost zeros, so e.g. $(0,-1,0,0,1,0,0)$ goes to $(-2,3,0,0)$. The `m'-notation is introduced here becasue it makes the relation between HPLs and $\mathrm{Li}$ functions particularly simple,
\begin{equation}
    H_{m_1,\ldots,m_k}(z)
    = \mathrm{Li}_{m_k,\ldots,m_1}(1,\ldots,1,z),
\label{eq:H-to-Li}
\end{equation}
Notice that the order of indices is reversed when expressing HPLs in terms of $\mathrm{Li}$. We note that the \math{} package \texttt{HPL}~\cite{Maitre:2005uu} implements both notations for HPLs as well as appropriate functions to transform between the two.

The list of available analytical and numerical packages for evaluation of classical, Nielsen generalized, multiple polylogarithms and series summations is given in appendices~\ref{app:mplsums}~and~\ref{app:mplsgpls}. 

To finish this section, we note that although MPLs are a large class of functions, it is known that not all Feynman integrals can be evaluated in terms of them. In fact these days, Feynman integrals evaluating to HPLs or MPLs are in some sense considered to be the easiest Feynman integrals. However, the seemingly not very complicated graph of Fig.~\ref{fig:higherPL} already cannot be evaluated in terms of MPLs when all masses of the internal lines are different.
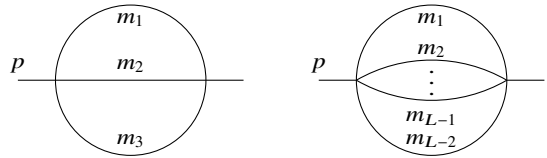
\begin{figure}
\centering
\begin{tikzpicture}[scale=1]
\begin{feynman}
 
\vertex at (0,-2) (i1);
\vertex at (0,2,1) (i2);
\vertex at (1,1) (i3);
\vertex at (1,0) (i4);
\vertex at (1,-1) (i5);
\vertex at (2,.5) (i6);
\vertex at (3,0) (f1);
\vertex at (4,0) (f2);

\draw (1,0) circle [radius=1cm];
\draw (-.5,0) -- (2.5,0); 
 
\node[above] at (-.5,0) {{{\bf $p$}}};
\node[above] at (1,0) {{{\bf $m_2$}}};
\node [below] at (1,1.) {{{\bf $m_1$}}};
\node [above] at (1,-1.) {{{\bf $m_3$}}}; 
 
\draw (5,0) circle [radius=1cm];
\draw (3.5,0) -- (4,0); 
\draw (6,0) -- (6.5,0); 
\draw (4,0) arc (120:60:2cm);
\draw (4,0) arc (240:300:2cm);

\node[above] at (3.5,0) {{{\bf $p$}}};
\node[above] at (5,0.2) {{{\bf $m_2$}}};
\node [below] at (5,1.) {{{\bf $m_1$}}};
\node [above] at (5,-.7) {{{\bf $m_{L-1}$}}}; 
\node [above] at (5,-1.) {{{\bf $m_{L-2}$}}}; 
\node [above] at (5,-.3) {{{\bf $\vdots$}}}; 

\end{feynman}
\end{tikzpicture}
    \caption{ \label{fig:higherPL}
        The two-loop banana diagram (also called the sunrise or the sunset diagram, depending on the author's mood) with generic internal masses (on left) and its generalization to the $L-1$ loop integral (on right).
      }
\end{figure}

In fact, this integral already at two-loops leads to \emph{elliptic integrals} and \emph{elliptic polylogarithms}. The study of these new classes of functions going beyond MPLs is currently an active area of research. 
%LL2022 talk Munch
In general a space of functions for Feynman integrals is extremely rich, we discussed GPLs, there are elliptic functions, modular forms and integrals over Calabi-Yau varieties~\cite{Feng:2019bdx}. Interestingly,  an upper bound on this complexity exist, 
it has been discussed in~\cite{delaCruz:2019skx}  that Feynman integrals are special cases of an A-hypergeometric function and  the scattering equations are given by the hyperdeterminant of a multidimensional
array, within the theory developed by
Gel'fand, Zelevinskii and Kapranov~\cite{GKZhyp}.  
For the state-of-the-art, see~\cite{Weinzierl:2022eaz,Bourjaily:2022bwx}.

\section{Gamma Function \label{sec:gamma}}

The gamma function is present in the basic \MB{} formula of Eq.~(\ref{mb1}) but it also appears in the definitions of a wide range of special functions, such as the (generalized) hypergeometric function and its particular cases which include Jacobi, Gegenbauer, Chebyschev or Legendre polynomials. These polynomials occur frequently both in classical and quantum physics~\cite{de2019solved}, hence the gamma function is ubiquitous in science and engineering.

Introducing the gamma function we start from factorial $n!$ which stands for the product of consecutive integers starting from 1, $n!=1 \cdot 2 \cdot \cdots n$, $n\in {\mathbb N}^+$. By convention, we set $0!=1$. Obviously we can write a recurrence relation for the factorial function, $(n+1)!=(n+1)\cdot n!$.

Interestingly, the discrete definition of the factorial can be extended to a smooth, real function with the property that
\begin{eqnarray}
\Gamma(x+1) &=& x \Gamma(x), \;\;\; x \in (0,\infty).
\label{eq-gammadiffeq}
\end{eqnarray}
Euler constructed such a function in 1729 with the normalization $\Gamma(1) = 1$ by using the convergent integral
\ba
\label{eq-gammadefreal}
\Gamma(x) &\equiv& \int_0^{\infty} e^{-t} t^{x-1} dt.
\ea
Starting from this integral representation, it is straightforward to show that the functional equation~(\ref{eq-gammadiffeq}) holds. Using partial integration we find
\begin{equation}
\begin{split}
\Gamma(z+1) &= 
    \int_0^{\infty} e^{-x}x^z dx = -x^z e^{-x}\Big|_{0}^{\infty}
    +z \int_0^{\infty} e^{-x}x^{z-1} dx 
\\    
    &=z\int_0^{\infty} e^{-x} x^{z-1} dx = z \Gamma(z).
\end{split}
\end{equation}
Furthermore, for $z \equiv n \in \mathbb{N}$, using $\Gamma(1)=1$ and Eq.~(\ref{eq-gammadiffeq}) it follows that the gamma function takes the values $\Gamma(n)=(n-1)!$. Thus the gamma function is an extension of the factorial function to numbers which are not integers. Fig.~\ref{fig:chap2_gammareim} shows the plot of the gamma function for real and purely imaginary $x$. Notice in particular the singularities at non-positive integers $x=0,-1,-2,\ldots$.
\begin{figure}[h!]
\begin{center}
    \begin{tikzpicture}[scale=1.0]
  \begin{axis}[
      axis lines=middle,
        width=6.75cm, height=5cm,
        xmin=-5.2,xmax=5.2,
        ymin=-5,ymax=5,
        label style={font=\footnotesize},
        ticklabel style={font=\footnotesize},
        xlabel=$x$,
        ylabel=$\Gamma(x)$,
        x label style={at={(axis description cs:1,0.5)},anchor=west},
        y label style={at={(axis description cs:0.5,1)},anchor=south},
        legend style={draw=none},
        legend style={font=\footnotesize},
        minor x tick num = 3,
        minor y tick num = 1
    ]
    
    \addplot [no marks,thin,black]
    coordinates {
      (-1,-5) (-1,5)
    };
        \addplot [no marks,thin,black]
    coordinates {
      (-2,-5) (-2,5)
    };
    \addplot [no marks,thin,black]
    coordinates {
      (-3,-5) (-3,5)
    };
        \addplot [no marks,thin,black]
    coordinates {
      (-4,-5) (-4,5)
    };
    \addplot [no marks,thin,black]
    coordinates {
      (-5,-5) (-5,5)
    };
    
    \addplot[
      line width=0.8pt]
    coordinates {
      (-4.99833,-5.)
      (-4.99814,-4.50528)
      (-4.99791,-4.01056)
      (-4.99762,-3.51583)
      (-4.99723,-3.02111)
      (-4.99668,-2.52639)
      (-4.99587,-2.03167)
      (-4.99453,-1.53695)
      (-4.99189,-1.04222)
      (-4.99154,-1.)
      (-4.99065,-0.905278)
      (-4.98953,-0.810556)
      (-4.98812,-0.715834)
      (-4.98626,-0.621112)
      (-4.98436,-0.547502)
      (-4.98372,-0.52639)
      (-4.98001,-0.431668)
      (-4.97412,-0.336946)
      (-4.9633,-0.242224)
      (-4.93666,-0.147502)
      (-4.85691,-0.0767638)
      (-4.77717,-0.0591614)
      (-4.69742,-0.0534566)
      (-4.61767,-0.0531792)
      (-4.53793,-0.0568571)
      (-4.45818,-0.0647498)
      (-4.37843,-0.0785568)
      (-4.29868,-0.102501)
      (-4.21894,-0.147502)
      (-4.14312,-0.242224)
      (-4.10711,-0.336946)
      (-4.08579,-0.431668)
      (-4.07162,-0.52639)
      (-4.06909,-0.547502)
      (-4.06151,-0.621112)
      (-4.05391,-0.715834)
      (-4.04799,-0.810556)
      (-4.04325,-0.905278)
      (-4.03936,-1.)
      (-4.03785,-1.04222)
      (-4.02609,-1.53695)
      (-4.01991,-2.03167)
      (-4.0161,-2.52639)
      (-4.01352,-3.02111)
      (-4.01165,-3.51583)
      (-4.01023,-4.01056)
      (-4.00912,-4.50528)
      (-4.00823,-5.)
    };
    \addplot[
      line width=0.8pt]
    coordinates {
      (-3.99156,5.)
      (-3.99066,4.52451)
      (-3.98954,4.04903)
      (-3.98813,3.57354)
      (-3.98627,3.09805)
      (-3.98371,2.62256)
      (-3.97999,2.14708)
      (-3.97405,1.67159)
      (-3.9631,1.1961)
      (-3.95529,1.)
      (-3.95133,0.924513)
      (-3.94658,0.849026)
      (-3.94079,0.773538)
      (-3.93592,0.720615)
      (-3.93358,0.698051)
      (-3.92433,0.622564)
      (-3.912,0.547077)
      (-3.89467,0.471589)
      (-3.86823,0.396102)
      (-3.82125,0.320615)
      (-3.77522,0.281388)
      (-3.72919,0.259645)
      (-3.68315,0.248597)
      (-3.63712,0.245132)
      (-3.59109,0.247827)
      (-3.54505,0.256204)
      (-3.49902,0.270457)
      (-3.45299,0.291405)
      (-3.40695,0.320615)
      (-3.32984,0.396102)
      (-3.2807,0.471589)
      (-3.24557,0.547077)
      (-3.21885,0.622564)
      (-3.19771,0.698051)
      (-3.19221,0.720615)
      (-3.18049,0.773538)
      (-3.16616,0.849026)
      (-3.15401,0.924513)
      (-3.14358,1.)
      (-3.12224,1.1961)
      (-3.09013,1.67159)
      (-3.07151,2.14708)
      (-3.0593,2.62256)
      (-3.05068,3.09805)
      (-3.04425,3.57354)
      (-3.03927,4.04903)
      (-3.0353,4.52451)
      (-3.03207,5.)
    };
    \addplot[
      line width=0.8pt]
    coordinates {
      (-2.96511,-5.)
      (-2.96181,-4.58881)
      (-2.95782,-4.17763)
      (-2.9529,-3.76644)
      (-2.94666,-3.35525)
      (-2.93849,-2.94407)
      (-2.92733,-2.53288)
      (-2.91112,-2.1217)
      (-2.88519,-1.71051)
      (-2.83576,-1.29932)
      (-2.74768,-1.)
      (-2.74163,-0.988814)
      (-2.73509,-0.977627)
      (-2.72799,-0.966441)
      (-2.72019,-0.955255)
      (-2.71152,-0.944068)
      (-2.70171,-0.932882)
      (-2.6903,-0.921695)
      (-2.67641,-0.910509)
      (-2.65777,-0.899323)
      (-2.64711,-0.894744)
      (-2.63646,-0.891403)
      (-2.6258,-0.889246)
      (-2.61515,-0.888231)
      (-2.60449,-0.888323)
      (-2.59383,-0.889494)
      (-2.58318,-0.891725)
      (-2.57252,-0.895004)
      (-2.56187,-0.899323)
      (-2.54145,-0.910509)
      (-2.52582,-0.921695)
      (-2.51269,-0.932882)
      (-2.5012,-0.944068)
      (-2.49087,-0.955255)
      (-2.48145,-0.966441)
      (-2.47274,-0.977627)
      (-2.46463,-0.988814)
      (-2.45702,-1.)
      (-2.33507,-1.29932)
      (-2.25455,-1.71051)
      (-2.20737,-2.1217)
      (-2.17565,-2.53288)
      (-2.15264,-2.94407)
      (-2.13511,-3.35525)
      (-2.12127,-3.76644)
      (-2.11005,-4.17763)
      (-2.10076,-4.58881)
      (-2.09294,-5.)
    };
    \addplot[
      line width=0.8pt]
    coordinates {
      (-1.88692,5.)
      (-1.87932,4.73024)
      (-1.87057,4.46048)
      (-1.86038,4.19072)
      (-1.84833,3.92096)
      (-1.8338,3.6512)
      (-1.81584,3.38144)
      (-1.79283,3.11169)
      (-1.78119,3.)
      (-1.77307,2.93024)
      (-1.76417,2.86048)
      (-1.76165,2.84193)
      (-1.75432,2.79072)
      (-1.74331,2.72096)
      (-1.73085,2.6512)
      (-1.71647,2.58144)
      (-1.71437,2.57217)
      (-1.6994,2.51169)
      (-1.67817,2.44193)
      (-1.64904,2.37217)
      (-1.63191,2.34329)
      (-1.61477,2.32248)
      (-1.59764,2.30918)
      (-1.58051,2.30297)
      (-1.56337,2.30358)
      (-1.54624,2.31082)
      (-1.5291,2.32463)
      (-1.51197,2.34504)
      (-1.49484,2.37217)
      (-1.46272,2.44193)
      (-1.43864,2.51169)
      (-1.42129,2.57217)
      (-1.41884,2.58144)
      (-1.40183,2.6512)
      (-1.38684,2.72096)
      (-1.37341,2.79072)
      (-1.36436,2.84193)
      (-1.36123,2.86048)
      (-1.35008,2.93024)
      (-1.3398,3.)
      (-1.32485,3.11169)
      (-1.29458,3.38144)
      (-1.27023,3.6512)
      (-1.25003,3.92096)
      (-1.2329,4.19072)
      (-1.21814,4.46048)
      (-1.20525,4.73024)
      (-1.19389,5.)
    };
    \addplot[
      line width=0.8pt]
    coordinates {
      (-0.761232,-5.)
      (-0.751442,-4.85446)
      (-0.740651,-4.70893)
      (-0.728656,-4.56339)
      (-0.715185,-4.41786)
      (-0.699848,-4.27232)
      (-0.682059,-4.12679)
      (-0.663823,-4.)
      (-0.660847,-3.98125)
      (-0.656448,-3.95446)
      (-0.648525,-3.90893)
      (-0.63995,-3.86339)
      (-0.634362,-3.83571)
      (-0.630584,-3.81786)
      (-0.620227,-3.77232)
      (-0.608573,-3.72679)
      (-0.597932,-3.69018)
      (-0.595115,-3.68125)
      (-0.578866,-3.63571)
      (-0.557299,-3.59018)
      (-0.545454,-3.57202)
      (-0.533609,-3.55853)
      (-0.521765,-3.54961)
      (-0.50992,-3.54518)
      (-0.498075,-3.54522)
      (-0.486231,-3.54969)
      (-0.474386,-3.55864)
      (-0.462541,-3.57211)
      (-0.450696,-3.59018)
      (-0.428964,-3.63571)
      (-0.412554,-3.68125)
      (-0.409706,-3.69018)
      (-0.398941,-3.72679)
      (-0.387137,-3.77232)
      (-0.376634,-3.81786)
      (-0.372799,-3.83571)
      (-0.367126,-3.86339)
      (-0.358413,-3.90893)
      (-0.350356,-3.95446)
      (-0.34588,-3.98125)
      (-0.342852,-4.)
      (-0.324274,-4.12679)
      (-0.306124,-4.27232)
      (-0.290456,-4.41786)
      (-0.27668,-4.56339)
      (-0.264406,-4.70893)
      (-0.253356,-4.85446)
      (-0.243328,-5.)
    };
    \addplot[
      line width=0.8pt]
    coordinates {
      (0.184487,5.)
      (0.200094,4.58856)
      (0.218674,4.17712)
      (0.241195,3.76568)
      (0.269112,3.35424)
      (0.304731,2.9428)
      (0.351966,2.53136)
      (0.41814,2.11992)
      (0.442877,2.)
      (0.468942,1.88856)
      (0.498663,1.77712)
      (0.519173,1.70848)
      (0.532959,1.66568)
      (0.573123,1.55424)
      (0.621055,1.4428)
      (0.679719,1.33136)
      (0.70064,1.29704)
      (0.727002,1.25726)
      (0.754159,1.21992)
      (0.854235,1.10848)
      (1.00517,0.997043)
      (1.01113,0.993699)
      (1.29525,0.898204)
      (1.57938,0.891359)
      (1.8635,0.949782)
      (1.99296,0.997043)
      (2.14763,1.07169)
      (2.21104,1.10848)
      (2.37201,1.21992)
      (2.43175,1.26852)
      (2.46458,1.29704)
      (2.50216,1.33136)
      (2.61217,1.4428)
      (2.70772,1.55424)
      (2.71588,1.56444)
      (2.79229,1.66568)
      (2.82236,1.70848)
      (2.86818,1.77712)
      (2.93701,1.88856)
      (3.,2.)
      (3.06228,2.11992)
      (3.24306,2.53136)
      (3.38785,2.9428)
      (3.50838,3.35424)
      (3.61143,3.76568)
      (3.70133,4.17712)
      (3.78095,4.58856)
      (3.85236,5.)
    };
  \end{axis}
\end{tikzpicture}
\vspace{0.5cm}
\begin{tikzpicture}[scale=1.0]
  \begin{axis}[
      axis lines=middle,
        width=6.75cm, height=5cm,
        xmin=-5.2,xmax=5.2,
        ymin=-5,ymax=5,
        label style={font=\footnotesize},
        ticklabel style={font=\footnotesize},
        xlabel=$x$,
        ylabel=$\Gamma(i x)$,
        x label style={at={(axis description cs:1,0.5)},anchor=west},
        y label style={at={(axis description cs:0.5,1)},anchor=south},
        legend style={draw=none},
        legend style={font=\footnotesize},
        minor x tick num = 3,
        minor y tick num = 1
    ]
        
    \addplot[
      line width=0.8pt]
      coordinates {
      	(-5.,-0.000271704)
		(-4.8,-0.000211807)
		(-4.6,-0.0000386985)
		(-4.4,0.000301578)
		(-4.2,0.000870178)
      	(-4.,0.00173011)
      	(-3.8,0.00293763)
      	(-3.6,0.00452973)
      	(-3.4,0.00650708)
      	(-3.2,0.00881172)
      	(-3.,0.0112987)
      	(-2.8,0.0137001)
      	(-2.6,0.0155803)
      	(-2.4,0.0162784)
      	(-2.2,0.0148358)
      	(-2.,0.00990244)
      	(-1.8,-0.000377672)
      	(-1.6,-0.0184943)
      	(-1.4,-0.0476694)
      	(-1.2,-0.091802)
      	(-1.,-0.15495)
      	(-0.8,-0.239717)
      	(-0.6,-0.343655)
      	(-0.4,-0.453643)
      	(-0.2,-0.542426)
      	(0.,-0.577216)
      	(0.2,-0.542426)
      	(0.4,-0.453643)
      	(0.6,-0.343655)
      	(0.8,-0.239717)
      	(1.,-0.15495)
      	(1.2,-0.091802)
      	(1.4,-0.0476694)
      	(1.6,-0.0184943)
      	(1.8,-0.000377672)
      	(2.,0.00990244)
      	(2.2,0.0148358)
      	(2.4,0.0162784)
      	(2.6,0.0155803)
      	(2.8,0.0137001)
      	(3.,0.0112987)
      	(3.2,0.00881172)
      	(3.4,0.00650708)
      	(3.6,0.00452973)
      	(3.8,0.00293763)
      	(4.,0.00173011)
      	(4.2,0.000870178)
      	(4.4,0.000301578)
      	(4.6,-0.0000386985)
      	(4.8,-0.000211807)
      	(5.,-0.000271704)      
    };
    \addplot[dashed,
      line width=0.8pt]
      coordinates {
      	(-5.,-0.000339933)
      	(-4.8,-0.000570007)
	    (-4.6,-0.000849562)
      	(-4.4,-0.00115169)
      	(-4.2,-0.00142339)
      	(-4.,-0.00157627)
      	(-3.8,-0.00147608)
      	(-3.6,-0.00093122)
      	(-3.4,0.000320042)
      	(-3.2,0.0026242)
      	(-3.,0.00643092)
      	(-2.8,0.0123187)
      	(-2.6,0.021035)
      	(-2.4,0.0335627)
      	(-2.2,0.0512368)
      	(-2.,0.075952)
      	(-1.8,0.110539)
      	(-1.6,0.159454)
      	(-1.4,0.230074)
      	(-1.2,0.335183)
      	(-1.,0.498016)
      	(-0.8,0.763499)
      	(-0.6,1.22859)
      	(-0.4,2.15845)
      	(-0.2,4.80974)
    };
    \addplot[dashed,
      line width=0.8pt]
      coordinates {
      	(0.2,-4.80974)
      	(0.4,-2.15845)
      	(0.6,-1.22859)
      	(0.8,-0.763499)
      	(1.,-0.498016)
      	(1.2,-0.335183)
      	(1.4,-0.230074)
      	(1.6,-0.159454)
      	(1.8,-0.110539)
      	(2.,-0.075952)
      	(2.2,-0.0512368)
      	(2.4,-0.0335627)
      	(2.6,-0.021035)
      	(2.8,-0.0123187)
      	(3.,-0.00643092)
      	(3.2,-0.0026242)
      	(3.4,-0.000320042)
      	(3.6,0.00093122)
      	(3.8,0.00147608)
      	(4.,0.00157627)
      	(4.2,0.00142339)
      	(4.4,0.00115169)
      	(4.6,0.000849562)
      	(4.8,0.000570007)
      	(5.,0.000339933)
    };
    \legend{$\Re$,$\Im$}
  \end{axis}
  
\end{tikzpicture}

\end{center}
\caption{The plot of the gamma function $\Gamma(x)$ for real (left) and purely imaginary (right) $x$. Note that for real arguments the gamma function is real and the thin vertical lines on the left panel denote asymptotes of the function.}
\label{fig:chap2_gammareim}       % Give a unique label
\end{figure}
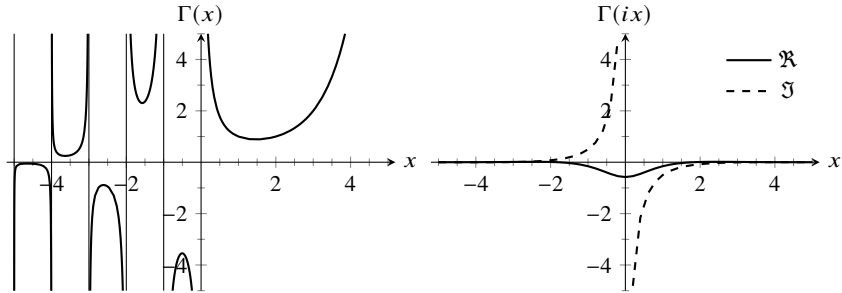

The gamma function was extended to complex variables by Gauss (1811) using the infinite product
\ba
\Gamma(z) = \lim_{n \to \infty} \frac{n! n^z}{z(z+1)\cdots (z+n)},
\ea
which defines a meromorphic function in $z$. From this representation we can see explicitly that the gamma function exhibits singularities in the complex plane, namely $z \in C \setminus \{0,-1,-2,\ldots \}$. However, there are several ways of defining the complex gamma function besides the above product representation~\cite{Slater:1966, Whittaker:1965}. For example, the integral definition of the gamma function for complex $z$ is of course a generalization of the definition of Eq.~(\ref{eq-gammadefreal}),
\ba
\label{eq-gammazdef}
\Gamma(z) &\equiv& \int_0^{\infty} e^{-x} x^{z-1} dx,
\ea
which is due to Bernoulli. It is clear that in general $\Gamma(z)$ is complex for complex $z$ and that the basic relation
\ba
\label{eq-zgammaz}
\Gamma(z+1)&=& z \Gamma(z),\quad z\in\mathbb{C}.
\ea
holds also for complex arguments.

Let us now recall some classical properties of the gamma function~\cite{Whittaker:1965,remmert:1998} which will be useful later in our study of \MB{} integrals. First, let us study the behaviour of the gamma function for large arguments. It can be shown that for $n \in \mathbb{Z}$ 
\ba 
\Gamma(n+1) = n! \simas^{n\rightarrow \infty} \sqrt{2 \pi n} n^n e^{-n},
\label{AsymptoticnGamma}
\ea
while for $z \in \mathbb{C}$ we have
\ba
\Gamma(z) \simas^{|z|\rightarrow \infty} 
\sqrt{\frac{2\pi}{z}} z^z e^{-z}\,.
\label{AsymptoticGamma}
\ea 
The second equation is the (leading term of the) Stirling formula. From the above two equations we could naively deduce, due to the different square root factors, that for large values Eqs.~(\ref{AsymptoticnGamma})~and~(\ref{AsymptoticGamma}) are in contradiction. However, this is not the case, see Problem~\ref{prob:gammaexpansion}. From Eq.~(\ref{AsymptoticGamma}) it is clear that the gamma function is strictly increasing for positive real arguments, as expected. It is also very easy to see that the rate of increase is super-exponential, i.e., $\lim_{x\to\infty} \Gamma(x)/e^{x} = \infty$. It is less obvious, but nevertheless true that for complex arguments $z=x+i y$, $x,y\in \mathbb{R}$, both the real and imaginary parts of the gamma function fall off faster than exponentially for large imaginary parts, i.e., as $|y| \to \infty$, see Problem~\ref{prob:gammaiasy}. 
In order to get a feeling for the magnitudes of gamma function values, below we give some representative numerical and analytical values obtained with \math{}.
 
\ba
\label{eq-g5}
\Gamma(1) &=&\Gamma(2) =1,
\\
\Gamma(10.1) &=& 4.54761\cdot 10^5,
\\
\Gamma(-10.1) &=& -2.21342\cdot 10^{-6},
\\
\Gamma(100.1 )  &=& 1.47845\cdot 10^{156},
\\
\Gamma(-100.1 ) &=& -6.86951\cdot 10^{-158},
\label{gammapoints} \\
\label{eq-gI}
\Gamma(-1\pm i) &=& -0.171533 + 0.326483\, i
\\
\Gamma(-1 \pm 10\; i)  &=& -4.99740\cdot 10^{-9} \pm 1.07847\cdot 10^{-8}\, i
\\
\Gamma(-1 \pm 100 \;i)  &=& 1.51438\cdot 10^{-71} \pm 1.27644\cdot 10^{-73}\, i
\\
\label{eq-g6}
\Gamma\left(\frac{1}{2}\right) &=& \sqrt{\pi},
\\
\label{eq:inversegamm}
\frac{1}{\Gamma(-1)} &=& \frac{1}{\Gamma(0)} = 0.
\ea
We note that the last equation is to be understood in the sense of a limit.

Next, consider Euler's reflection (or complement) formula, which states that for $0<z<1$
\ba
\label{eq-gsin}
\Gamma(z) \Gamma(1-z)= \frac{\pi}{\sin \pi z}.
\ea
Similarly
\ba
\label{eq-gcos}
\Gamma\left(\frac{1}{2}+z\right) \Gamma\left(\frac{1}{2}-z\right)=
\frac{\pi}{\cos \pi z}. 
\ea 
These so-called complement relations can sometimes be useful for manipulating \MB{} integrals and we will also use them when we discuss the analytic solution of one-dimensional \MB{} integrals in chapter~\ref{chapter-MBanal}. Notice moreover that Eq.~(\ref{eq-gsin}) allows us to determine the behaviour of the gamma function around $z=0$. Indeed, using $\sin(\pi z) = \pi z + \mathcal{O}(z^3)$ and $\Gamma(1)=1$, we find 
\begin{equation}
    \Gamma(z) \simas^{z\rightarrow 0} \frac{1}{z}. 
\label{eq:gamma-res-0}
\end{equation}
In fact, the above calculation can be generalized and the behaviour of the gamma function determined around any non-positive integer. To do so, let us consider $z\to -n+z$ in Eq.~(\ref{eq-gsin}), where $n\in\mathbb{N}$. Then we have
\begin{equation}
    \Gamma(-n+z) = \frac{1}{\Gamma(1+n-z)}
    \frac{\pi}{\sin \pi(-n+z)}.
\label{eq:gamma-invert-sin}
\end{equation}
However $\sin \pi(-n+z) = (-1)^n \sin(\pi z)$ for integers, while $\Gamma(1+n-z) = n!$ for $z\to 0$, hence finally
\begin{equation}
    \Gamma(-n+z) \simas^{z\rightarrow 0} 
    \frac{(-1)^n}{n!}\frac{1}{z}. 
\label{eq:gamma-res-n}
\end{equation}
Of course, the same conclusion can be reached by applying the recursion relation in Eq.~(\ref{eq-zgammaz}) to reduce $\Gamma(-n+z)$ to $\Gamma(z)$ and using Eq.~(\ref{eq:gamma-res-0}), see Problem~\ref{prob:gammares}.

Thus, we see that the gamma function has poles at all non-positive integers, a fact that will become significant shortly, when we discuss residues and Cauchy's theorem. In Fig.~\ref{fig:singul} we show the location of these poles for some gamma functions with specific arguments which will be considered in chapter~\ref{chapter-singul}, see Fig.~\ref{SEpoles}.
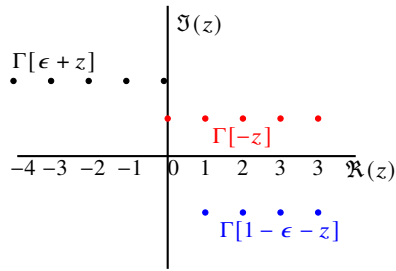
\begin{figure}
\centering
\begin{tikzpicture}[scale=0.5]
\begin{feynman}
 
\vertex at (2,1) (i1);
\vertex at (4,1) (i2);
\vertex at (6,1) (i3);
\vertex at (8,1) (i4);
\vertex at (9,1) (i5);
\vertex at (6,4.5) (f1);
\vertex at (6,-1.5) (f2);

\node[right] at (f1) {{$\Im(z)$}};
\node[right] at (10.5,0.6) {{$\Re(z)$}};

\draw [thick]  (2,1) -- (11,1) ;
\draw [thick]  (6,-2) -- (6,5) ;

\filldraw [red] (6,2) circle (2pt);
\filldraw [red] (7,2) circle (2pt);
\filldraw [red] (8,2) circle (2pt);
\filldraw [red] (9,2) circle (2pt); 
\filldraw [red] (10,2) circle (2pt);
\node [thick,red] at (8,1.5) {\bf{{$\Gamma[-z]$}}};

%\filldraw [red] (6,-1.5) circle (2pt);
\filldraw [blue] (7,-.5) circle (2pt);
\filldraw [blue] (8,-.5) circle (2pt);
\filldraw [blue] (9,-.5) circle (2pt); 
\filldraw [blue] (10,-.5) circle (2pt);
\node [thick,blue] at (9,-1.) {\bf{{$\Gamma[1-\eps-z]$}}};

\filldraw [black] (1.9,3) circle (2pt);
\filldraw [black] (2.9,3) circle (2pt);
\filldraw [black] (3.9,3) circle (2pt); 
\filldraw [black] (4.9,3) circle (2pt);
\filldraw [black] (5.9,3) circle (2pt);
\node [thick,black] at (3,3.5) {\bf{{$\Gamma[\eps+z]$}}};

%\node at (5.55,.7) {{$-\frac 12$}};
\node at (2.15,.7) {{$-4$}};
\node at (3.,.7) {{$-3$}};
\node at (4.,.7) {{$-2$}};
\node at (5.,.7) {{$-1$}};

\node at (6.15,.7) {{$0$}};
\node at (7.,.7) {{$1$}};
\node at (8.,.7) {{$2$}};
\node at (9.,.7) {{$3$}};
\node at (10.,.7) {{$3$}};  
 
\end{feynman}
\end{tikzpicture}
\caption{Positions of the singularities on the real axis (horizontal) for three gamma functions. To avoid visual overlapping of dots in the Figure, poles of single gammas are drawn at different vertical levels. The $\eps$ parameter is assumed to be real and positive, $\eps >0$. The singular behaviour of the gamma function at the poles can be seen in Fig.~\ref{fig:chap2_gammareim}.}
\label{fig:singul}
\end{figure}

Going beyond the above results, the expansion of the gamma function around integer numbers to higher orders is often needed in explicit computations. It is common that the small parameter of the expansion is the $\epsilon$ variable that appears in dimensional regularization where the space-time dimension is set to $d=4-2\epsilon$. For example the expansions of $\Gamma(z)$ around $z=0$ and $z=1$ in the small parameter $\epsilon$ read
\ba
% 
%\\
\label{eq-g7}
\Gamma(\epsilon) &=& \frac{1}{\epsilon} - \gamma + \frac{1}{2} \epsilon \left[
\zeta(2)+ \gamma^2 \right] 
+ \frac{1}{6} 
\left[-2\zeta(3)- 3  \gamma  \zeta(2)- \gamma^3 \right]
\epsilon^2 
+ {\cal O}\left(\epsilon^3\right)
\nn \\&\equiv & 
\frac{1}{\epsilon} - \gamma +g_1\epsilon+g_2\epsilon^2+ {\cal
 O}\left(\epsilon^3\right) ,
\\
\label{eq-g8}
\Gamma(1+\epsilon) &=& 1-  \gamma\epsilon
+ \frac{1}{2} \epsilon^2 \left[
\zeta(2)+ \gamma^2 \right]   + \ldots
\equiv 1- \gamma\epsilon  +g_1\epsilon^2 + \ldots.
\ea
 
We see that in the expansions a couple of irrational constants have appeared. These constants are either related to the harmonic series or are values of the zeta function at positive integers,
\ba
\label{eq-g9} 
\gamma &=& 0.577\,215\,664\,901\,532   \ldots  \quad\mbox{Euler's constant},
\\
\label{eq-g10}
\zeta_2 ~~\equiv~~ \zeta(2)&=& \mathrm{Li}_2(1) = \frac{\pi^2}{6} = 1.644\, 93 \ldots ,
\\
\label{eq-g11}
\zeta(3)&=& \mathrm{Li}_3(1) = 1.202\,056\,903\,159\,594 \ldots.
\ea
In order to explain these relations, let us define the \emph{generalized harmonic number of order $a$ of $N$} as the finite sum
\begin{equation}
    H_{N,a} \equiv \sum_{k=1}^{N} \frac{1}{k^a}.
\end{equation}
This sum is also often denoted as $S_a(N)$. With this definition, the zeta function\footnote{For a history of the zeta function and logarithm see~\cite{havilnapier}. Physical applications of spectral zeta function is discussed in Lecture Notes~\cite{elizaldezeta}.} at $a$ is simply the harmonic number of order $a$ of infinity,
\begin{equation}
    \zeta(a) = \sum_{k=1}^{\infty} \frac{1}{k^a} = H_{\infty,a}.
\end{equation}
Then Euler's constant is the limit at infinity of the difference of the harmonic number $H_{N,1}$ and the natural logarithm of $N$,
\beq
\gamma = \lim_{N\to \infty} \left[ \sum_{k=1}^{N} \frac{1}{k} - \ln(N)\right],
\eeq
while $\zeta(2)$, $\zeta(3)$ and so on are simply the values of the convergent sums
\begin{equation}
    \zeta(2) = \sum_{k=1}^{\infty}\frac{1}{k^2},\qquad
    \zeta(3) = \sum_{k=1}^{\infty}\frac{1}{k^3}.
\end{equation}
Incidentally, using the notion of harmonic numbers, it is possible to write a compact formula for the all-order expansion of the gamma function around positive integers $n$ in terms of 
$\Gamma(1+\eps)$,
\bea
\label{harmsum}
\frac{\Gamma(n+\eps)}{\Gamma(n)} &=& \Gamma(1+\eps)
\exp \left[ -\sum_{k=1}^{\infty} \frac{(-\eps)^k}{k} S_k(n-1)\right] .
\eea

Another useful relation which will be used in chapter~\ref{chapter-MBanal} is the Legendre duplication formula,
\ba\label{eq-gamdupl}
\Gamma(2z)&=&\frac{2^{2z-1}}{\sqrt{\pi}}
\Gamma(z)\Gamma\left(z+\frac{1}{2}\right), 
\ea
which is a special case of the so-called product theorem for gamma functions,
\begin{equation}
    \Gamma(nz) = (2\pi)^{\frac{1-n}{2}} n^{nz-\frac{1}{2}}
    \prod_{k=0}^{n-1} \Gamma\left(z+\frac{k}{n}\right).
\end{equation}

As we will see, one also often encounters the derivatives of the gamma function in actual calculations, see e.g. Eq.~(\ref{eq:residpolygamma}) or \wwwaux{miscellaneous} 
(polygammas can be seen in examples of $\eps$ expanded \mb{} integrals).
Hence, we recall the definition of the \emph{digamma} function $\psi(z)$ of a (in general complex) variable $z$,
\begin{equation}
    \psi(z) \equiv \frac{d}{dz} \ln \Gamma(z) = 
    \frac{1}{\Gamma(z)} \frac{d}{dz} \Gamma(z).
\label{eq:psi0-def}
\end{equation}
The polygamma function of order $m$ is then the $m$-th derivative of the digamma function
\begin{equation}
    \psi^{(m)}(z) \equiv \frac{d^m}{dz^m} \psi(z).
\label{eq:psin-def}
\end{equation}
Thus, the digamma function is simply the polygamma function of order zero, $\psi(z) = \psi^{(0)}(z)$. {\it It can be shown that polygamma functions have poles of order $m+1$ at all non-positive integers}, see Problem~\ref{prob:polygammapoles}.

For arguments whose real part is positive, $\Re(z)>0$, the digamma function can be represented also as a one-dimensional integral over a real integration variable,
\begin{equation}
    \psi(z) = \int_0^\infty [(1+t)^{-1}-(1+t)^{-z}]\frac{dt}{t} - \gamma,\qquad \Re(z)>0,
    \label{eq:polygamma-int-inf}
\end{equation}
or alternatively
\begin{equation}
    \psi(z) = \int_0^1 \frac{t^{z-1}-1}{t-1}dt - \gamma,\qquad \Re(z)>0.
    \label{eq:polygamma-int-one-ch2}
\end{equation}
{\it This relation is employed when solving \MB{} integrals via reduction to Euler-type integrals}, a method that will be explored in chapter~\ref{sec:3anal}. We note that we can easily derive similar integral representations for polygamma functions by simply taking derivatives of the above formulae with respect to $z$. (For $\Re(z)>0$ the order of differentiation and integration can be interchanged.) 

Among the many relations that hold for polygamma functions, we recall here only those that are the most relevant for manipulating \MB{} integrals. First, using Eq.~(\ref{eq-gammazdef}), it is quite straightforward to show that
\begin{equation}
    \psi(z+1) = \psi(z) + \frac{1}{z}
\end{equation}
and in the general case
\begin{equation}
    \psi^{(m)}(z+1) = \psi^{(m)}(x) + (-1)^m m! z^{-m-1}.
\end{equation}
These relations allow us to reduce the values of the polygamma functions at non-negative integers to values at one, where we have
\begin{equation}
    \psi(1) = -\gamma
    \qquad\mbox{and}\qquad
    \psi^{(m)}(1) = -(-1)^{m} m! \zeta(m + 1),\quad m\in \mathbb{N}^+.
\end{equation}
Thus, in terms of the generalized harmonic numbers introduced above, we find for $N\in \mathbb{N}^+$
\begin{align}
    \psi(N+1) &= \sum_{k=1}^{N} \frac{1}{k}-\gamma = 
    S_1(N) - \gamma,
\label{eq:psi0N}
\intertext{and}
\psi^{(m)}(N+1) &= (-1)^m m! 
    \sum_{k=1}^{m} \frac{1}{k^{m+1}}-(-1)^{m} m! \zeta(m + 1) 
\\ &=
    (-1)^m m!\left[S_{m+1}(N) - \zeta(m + 1)\right].
\label{eq:psimN}
\end{align}
Thus for example $\psi(1) = -\gamma$, $\psi(2) = 1-\gamma$, $\psi(3) = 3/2-\gamma$ and so on.

One more relation that we will make use of in chapter~\ref{chapter-MBanal} is the relation
\begin{equation}
    \psi^{(m_q)}(1-z) = (-1)^{m_q}
    \left[\psi^{(m_q)}(z) + \pi \frac{\partial^{m_q}}{\partial z^{m_q}}
    \cot(\pi z)\right],
\label{eq:psi-reflect-ch2}
\end{equation}
which can be used to convert all $\psi^{(m)}(1-z)$ functions to $\psi^{(m)}(z)$ functions.

We note finally that expanding polygamma functions around integer parameters, we again encounter the Euler $\gamma$ constant and zeta values, e.g.
\begin{align}
\psi(1+\eps) &= -\gamma + \zeta(2) \eps - \zeta(3) \eps^2 + \mathcal{O}(\eps^3)\,,
\\
\psi^{(1)}(1+\eps) &= \zeta(2) - 2\zeta(3) \eps + 3\zeta(4) \eps^2 + \mathcal{O}(\eps^3)\,,
\\
\psi^{(2)}(1+\eps) &= -2\zeta(3) + 6\zeta(4) \eps - 12\zeta(5) \eps^2 + \mathcal{O}(\eps^3)\,.
\end{align}

\section{Residues and Cauchy's Residue Theorem}
\label{sec:2complex}

{The Cauchy integral theorem (also known as the Cauchy–Goursat theorem) states that if a function is analytic in a domain (without singularities), then for any  closed contour in the domain that contour integral is zero.} 
If singularities exist, the integral over an anti-clockwise\footnote{The positive direction of the curve is defined as anti-clockwise for which the winding number is 1, see Problem~\ref{prob:curvedirection}.}  directed closed path $C_+$ is (the Cauchy's Residue Theorem):
\ba
\oint_{C_+} F(z) dz &=& 2\pi i \sum_{z=z_i} \texttt{Res}[F(z)] 
\label{eq:cauchy}
\ea
where the residues $\texttt{Res}[F(z)]|_{z=z_i}$  are coefficients $a_{-1}^i$ of the Laurent series of $F(z)$ around $z_i$:
\begin{equation}
F(z) = \sum_{n=-N}^{\infty} a_n^i (z-z_i)^n = \frac{{a_{-N}^i}}{(z-z_i)^N}+\cdots +
\frac{a_{-1}^i}{(z-z_i)}+a_0^i+\cdots
\label{eq:residue0}
\end{equation}

\begin{equation}
\boxed{
  \texttt{Res}[F(z)]|_{z=z_i} = {a_{-1}^i}
}
\label{eq:residue1}
\end{equation}

If $G(z)$ has a Taylor expansion around $z_0$ , then it is:
%\bqa 
 %\texttt{Res} [G(z)~F(z)]|_{z=z_i} = \sum_{n=1}^{N}  \frac{a_{-n}^i}{k!} \frac{d^n}{dz^n} G(z)|_{z=z_i} 
 %\eqa

\begin{equation}
\boxed{
   \texttt{Res} [G(z)~F(z)]|_{z=z_i} = \sum\limits_{n=1}^{N}  \frac{a_{-n}^i}{k!} \frac{d^n}{dz^n} G(z)|_{z=z_i} 
}
\label{eq:residue}
\end{equation}

{\it Due to this property, we need for  applications not only $\Gamma(z)$, but also its derivatives. }

\begin{svgraybox}
{The basic equation~(\ref{eq:cauchy}) says that calculation of integrals in the complex plane can be done by summing residues enclosed by the contour of integration. }
\end{svgraybox}

This is a simple trick which can help to solve many otherwise difficult integrals in a simple way. %To do that, we can use directly residue function in \math{}, which goes like this

\newpage

\begin{tips}{Residues for Integrals Over Real and Complex Variables \label{ex1_sec2}}
Using the Cauchy's theorem evaluate 
\ba
I_1 = \oint\limits_\Omega \frac{dz}{z^4+1},
\label{eq:cauchyI1}
\ea
where $\Omega$ is any closed contour\footnote{Contours of this kind are known as the Hankel contours~\cite{de2019solved}.} around $z=0$ with radius $R>1$, see Fig.~\ref{fig:cont1}, 
and the integral 
\ba
I_2 = \int\limits_{-\infty}^\infty \frac{dx}{x^4+1},
\label{eq:cauchyI2}
\ea
\\
for which the contour $\Omega_1$ in Fig.~\ref{fig:cont1} can be applied.

We have four solutions for $z^4 + 1 = 0$, $z_1 = e^{\pi i/4}, z_2=e^{3\pi i/4}, z_3=e^{5\pi i/4}, z_4=e^{7\pi i/4}$. All solutions 
$z_1, \ldots, z_4$ lie inside $\Omega$, see Fig.~\ref{fig:cont1}.

\begin{figure}
\centering
\begin{tikzpicture}[scale=0.5]
\begin{feynman}
 
\vertex at (2,1) (i1);
\vertex at (4,1) (i2);
\vertex at (6,1) (i3);
\vertex at (8,1) (i4);
\vertex at (11,1) (i5);
\vertex at (6,5) (f1);
\vertex at (6,-1.5) (f2);

\node[right] at (f1) {{$\rm{Im}\; z$}};
\node[right] at (10.,0.7) {{$\rm{Re}\; z$}};

\draw [->]  (i3) -- (f1) ;
\draw [->]  (i3) -- (i5) ;
\draw (i3) -- (f2) ;
\draw (2,1) -- (3,1);
\draw (9,1) -- (11,1);

\filldraw (6.7,1.73) circle (3pt);
\node at (6.7,2.1) {\large{$z_1$}};
\filldraw (5.3,1.73) circle (3pt);
\node at (5.3,2.1) {\large{$z_2$}};
\filldraw (6.7,0.25) circle (3pt);
\node at (6.7,-0.15) {\large{$z_4$}};
\filldraw (5.3,0.25) circle (3pt);
\node at (5.3,-0.15) {\large{$z_3$}};
\node at (7.15,.7) {\large{$1$}};
\node at (8.15,.7) {\large{$2$}};

\draw (i3) circle (1cm);
\draw [dotted] (i3) circle (2cm);
\draw [->, very thick,black] (9,1) arc (0:45:3cm);
\draw [very thick,black] (8.11,3.12) arc (45:90:3cm);
\draw [->,very thick,black] (6,4) arc (90:135:3cm);
\draw [very thick,black] (3.9,3.12) arc (135:180:3cm);
\draw [->, very thick,black] (3,1) -- (4.5,1);
\draw [->, very thick,black] (4.5,1) -- (7.5,1);
\draw [very thick,black] (7.5,1) -- (9,1);
%\diagram*{
% (i1) -- [thin] 
% (i5) --  [thin] 
% (i3) --  [thin] (f2), 
%};

\node [thick, black] at (3,3) {\bf{{$\Omega_1$}}};
\node [thick, black] at (2.8,0.6) {\bf{{$-R$}}};
\node [thick, black] at (9,0.6) {\bf{{$R$}}};

\node [thick] at (4,-0.2) {\bf{{$\Omega$}}};

\end{feynman}
\end{tikzpicture}
\caption{Contours $\Omega$ and $\Omega_1$ applied in Eqs.~(\ref{eq:cauchyI1})~and~(\ref{eq:cauchyI2}), respectively.}
\label{fig:cont1}   
\end{figure}
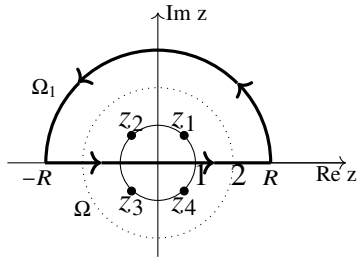

Thus we have
\begin{eqnarray}
\texttt{\footnotesize{Res}}[I_1]|_{z=z_1} &=&-\footnotesize{\frac{1}{4 \sqrt{2}}} (1+i), \;\;\; \label{eq:resz1}
\texttt{Res}[I_1]|_{z=z_2} =-\footnotesize{\frac{1}{4 \sqrt{2}}} (-1+i), \label{eq:resz2} \\
\texttt{\footnotesize{Res}}[I_1]|_{z=z_3} &=&-\footnotesize{\frac{1}{4 \sqrt{2}}} (-1-i),\;\;\; \label{eq:resz3}
\texttt{Res}[I_1]|_{z=z_4} =-\footnotesize{\frac{1}{4 \sqrt{2}}} (1-i), \label{eq:resz4} 
\end{eqnarray}
and the sum
\bq
I_1 = \sum_{i=1}^4 \texttt{\footnotesize{Res}}[I_1]|_{z=z_i} = 0. 
\eq

In \math{}, a residue calculated in Eq.~(\ref{eq:resz1}) would have the following form
\begin{minted}[frame=single,breaklines,fontsize=\small]{mathematica}
In[11]:= Residue[1/(1+z^4),{z,Exp[I*Pi/4]}]
\end{minted}
\begin{equation}
    \label{eq:exmathresid}
\end{equation} 
In contrast to the $I_1$  integral,  $I_2$ is a real integrand over real contour of integration. Nonetheless, we can apply Cauchy's theorem by extending the path of integration to the complex plane (the thick $\Omega_1$ contour in Fig.~\ref{fig:cont1}).
This is possible since the integral along the semicircle in the upper half-plane vanishes in the $R\to\infty$ limit, and so we are not changing the value of the integral by including this contribution. Then inside $\Omega_1$ there are only two singularities at $z_1$ and $z_2$, solved already in Eqs.~(\ref{eq:resz1})~and~(\ref{eq:resz2}), then the result is 
\bq
I_2 = 2 \pi i \cdot \left(-\frac{1}{4\sqrt{2}} \cdot 2 i \right) = \frac{\pi}{2\sqrt{2}}. 
\eq

These are typical examples how integrals can be solved using the Cauchy theorem. For more examples, see for instance~\cite{Roussosimprop}.

\end{tips}
  
\begin{tips}{Residues of the Gamma Function \label{ex3_sec2}}
What are residues for $\Gamma(z)$ and $\Gamma^2(z)$ function at points $z=i,1+i,-1,0,1?$
\\

Singularities of the gamma function sit on the real axis and are negative integers, see Fig.~\ref{fig:chap2_gammareim} and Eq.~(\ref{eq:inversegamm}). Series at points $z=i,i+1$
Laurent expanded as in Eq.~(\ref{eq:residue0}) have no $a_{-1}^i$ coefficients, e.g.:

\begin{minted}[frame=single,breaklines,fontsize=\small]{mathematica}
In[12]:= Series[Gamma[z],{z,I,1}] 
Out[12]:= Gamma[i]+Gamma[i]PolyGamma[0,i](z-i)+O[z-i]^2
\end{minted}

In general, 
\begin{eqnarray}
%\textcolor{myred}
%\texttt{Res}[I_1]|_{z=z_1}
\textrm{\footnotesize{Res}}[\Gamma[z]]_{z=-n} &=& \frac{(-1)^n}{n!},
\label{eq:basicresgamma}
\\
\textrm{\footnotesize{Res}}[G[z]\Gamma[z]]_{z=-n} &=& \frac{(-1)^n}{n!} G[-n], \label{eq:generalresidgamma}
\\
\textrm{\footnotesize{Res}}[G[z]\Gamma^2[z]]_{z=-n} &=& \frac{2\; %\textrm{PolyGamma}[n+1] 
\psi^{(0)}[n+1] G[-n]+G^{\prime}[-n]}{(n!)^2},\nn \\
\label{eq:residpolygamma}
\end{eqnarray}
where $G[z]$ is an analytic function of $z$.
Equation~(\ref{eq:basicresgamma}) was proved in Eq.~(\ref{eq:gamma-res-n}), while the relation in Eq.~(\ref{eq:generalresidgamma}) is a straightforward consequence of Eq.~(\ref{eq:basicresgamma}). Note the appearance of the polygamma function in Eq.~(\ref{eq:residpolygamma}), see Problem~\ref{prob:residpolygamma}. 

\end{tips}

\newpage

\begin{tips}{Summing Gamma Poles and Integration Over the Bromwich Contour \label{ex4_sec2}}

Let us calculate the sum of residues of $\Gamma[-z]$ when the contour of integration goes through $\Re[z]= \pm 1/2$ for the first 10 or 100 poles, see Fig.~\ref{fig:chap2_residuegammaz}.

In \math{}, a residue calculated for $\Gamma[-z]$ would have the following form

\begin{minted}[frame=single,breaklines,fontsize=\small]{mathematica}
In[13]:= n1=Sum[-Residue[Gamma[-z],{z,n}],{n,0,100}];
In[14]:= n2=Sum[-Residue[Gamma[-z],{z,n}],{n,1,100}];
In[15]:= n1-n2;;
\end{minted}
 
{The minus signs in the first and second In commands are due to the direction of integration, which is clockwise, see Fig.~\ref{fig:chap2_residuegammaz}. The result of the sum in the first command which starts from $n=0$ is 0.367879. In fact, taking only ten terms gives the same accuracy. The result of the sum in the second command which starts from $n=1$ is -0.632121. The difference is $n1-n2=1$, which is the value of the residue for $\Gamma[-z]$ at $z=0$.} 

In chapter~\ref{chapter-MBnum} we will discuss how to calculate \mb{} integrals numerically. 
For the above case we can make it directly in \math. The contour is over a straight line parallel to the vertical axis. Such a contour is called the Bromwich contour
~\cite{de2019solved}

\begin{minted}[frame=single,breaklines,fontsize=\small]{mathematica}
In[16]:= NIntegrate[1/(2 Pi) Gamma[-z]/.z->-1/2+I y,{y,-10,10}];
\end{minted}

\begin{equation}
\label{eq:bromwich}
\end{equation}
The integral gives the same value as the sum of residues \texttt{In[13]} above. Please note a change of variables in Eq.~(\ref{eq:bromwich}) and that to get the same accuracy as for sums, integration can be restricted to the region $(-10,10)$, as the gamma function falls down quickly with increasing $|z|$ (Fig.~\ref{fig:chap2_gammareim}).
 
\begin{figure}
\centering
\begin{tikzpicture}[scale=0.5]
\begin{feynman}
 
\vertex at (2,1) (i1);
\vertex at (4,1) (i2);
\vertex at (6,1) (i3);
\vertex at (8,1) (i4);
\vertex at (9,1) (i5);
\vertex at (6,3.5) (f1);
\vertex at (6,-1.5) (f2);

\node[right] at (f1) {{${\Im}(z)$}};
\node[right] at (9.2,0.6) {{$\Re(z)$}};

\draw [thick]  (5.5,-2) -- (5.5,1) ;
\draw [->, thick]  (5.5,1) -- (5.5,3) ;
\draw [thick]  (5.5,3) -- (5.5,4) ;
\draw [->, thick]  (6.5,-1) -- (6.5,2) ;
\draw [thick] (6.5,2) -- (6.5,3) ;
\draw [->]  (i3) -- (10,1) ;
\draw [->] (i3) -- (f1) ;
\draw (i3) -- (f2) ;
\draw (4,1) -- (6.5,1);
 
\filldraw [red] (6,1) circle (2pt);
%\filldraw (6.5,1) circle (2pt);
\filldraw [red] (7,1) circle (2pt);
\filldraw [red] (8,1) circle (2pt);
\filldraw [red] (9,1) circle (2pt); 
\node at (5.55,.7) {{$-\frac 12$}};
\node at (6.15,.7) {{$0$}};
\node at (6.6,.7) {{$\frac 12$}};
\node at (7.,.7) {{$1$}};
\node at (8.,.7) {{$2$}};
\node at (9.,.7) {{$3$}};
 
\draw [thick] (6.5,3) -- (9,3);
\draw [->, thick] (9,3) arc (90:45:2cm);
\draw [thick] (10.41,2.41) arc (45:0:2cm);
\draw [->, thick] (10.99,1) arc (0:-45:2cm);
\draw [thick] (10.41,-0.42) arc (-45:-90:2cm);
\draw [thick] (6.5,-1.01) -- (9,-1.01);

\draw [thick] (5.5,4) -- (9,4);
\draw [->, thick] (9,4) arc (90:45:3cm);
\draw [thick] (11.1,3.13) arc (45:0:3cm);
\draw [->, thick] (11.98,1.) arc (0:-45:3cm);
\draw [thick] (11.13,-1.11) arc (-45:-90:3cm);
\draw [thick] (5.5,-2) -- (9.01,-1.99);

\node [thick,red] at (8,1.5) {\bf{{$\Gamma[-z]$}}}; 

\node [thick] at (7,-2.29) {\bf{{$\Omega_1$}}}; 
\node [thick] at (8,-1.29) {\bf{{$\Omega_2$}}};

\end{feynman}
\end{tikzpicture}
\caption{Integration contours $\Omega_1$ and $\Omega_2$. Contours of integration are at fixed $\Re(z)=- \frac 12, \frac 12$ which correspond to the sums starting at $n=0,1$ at \texttt{In[13]} and \texttt{In[14]}, 
respectively.}
\label{fig:chap2_residuegammaz}
\end{figure}
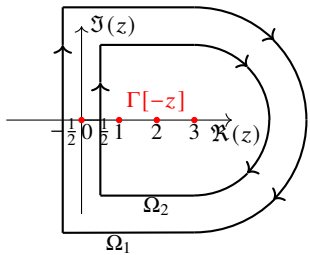
 
\end{tips}

\begin{tips}{Choosing Right Direction of the Integration Contour}
Calculate the integral of $\Gamma[z]$ along a line parallel to the imaginary axis at $\Re(z)=-1/2$.

The idea is to approximate an integral over a closed path by an integral over a
 part of the path only, as shown in Eq.~(\ref{eq:bromwich}). We can write
\beq \label{eq-oint}
\oint_{-1/3-i\infty}^{-1/3+i\infty} dz \Gamma[z]
~ \approx ~
\oint_{-1/3-9i}^{-1/3+9i} dz \Gamma[z] = (-i)~3.97173.
\eeq

This integral is assumed to be part of a closed path, closing to the left or to the right,
resulting in different sums of residues representing the integral.

Closing to the left (anti-clockwise)

\beq
\oint_{-1/2-i}^{-1/2+i} dz \Gamma[z] = 
2\pi i \sum_{n=1}^{\infty} \frac{(-1)^n}{n!} = 2\pi i \; \frac{1-e}{e} =  - 3.97173 \; i \label{eq-oint1}
\eeq
while closing contour on right
\beq 
(-1)\; 2\pi i \sum_{n=0}^{0} \frac{(-1)^n}{n!} = - 2\pi i  \neq  -  3.97173 \; i. \label{eq-oint2}
\eeq

While both the approximation Eq.~(\ref{eq-oint}) and the first sum in Eq.~(\ref{eq-oint1}) represent the integral, the second one, eq.Z~(\ref{eq-oint2}) does not. The reason is simple: from the section on the gamma function we know that the integrand in Eq.~(\ref{eq-oint}) becomes small when closing the path to the left. But when we close the integration path to the right, there is a non-negligible contribution from the positive real axis at infinity, and the sum of residues fails to represent the integral in Eq.~(\ref{eq-oint}).

\end{tips}

\begin{tips}{
QFT Example: \texttt{FI} as the \mb{} Integral and the Series \label{exQFT_sec2}}

Let us consider the diagram with two internal propagators as drawn below.

\begin{center}
\begin{tikzpicture}[scale=0.6]
\begin{feynman}
 
%tadpole 

\vertex at (0.5,1) (i1);
\vertex at (2,1) (a);
\vertex at (4,1) (b);
\vertex at (5.5,1) (f1);

\diagram*{
(i1) -- [dashed,thick,momentum=\(p\),/tikzfeynman/momentum/arrow distance=2mm] (a) -- [half left, momentum=\(k\),/tikzfeynman/momentum/arrow shorten=0.25,/tikzfeynman/momentum/arrow distance=2mm] (b) -- [thick,half left, momentum=\(k-p\),/tikzfeynman/momentum/arrow shorten=0.25,/tikzfeynman/momentum/arrow distance=2mm] (a),
(b) -- [dashed,thick] (f1),
};

\end{feynman}
\end{tikzpicture}
\end{center} 

The corresponding integral will be considered in section~\ref{sec:3MBrepr}, its \mb{} form is (see Eq.~(\ref{MB-SE1l2m-lemmas}))

 \begin{eqnarray}
G(1)_{\text{SE2l2m}} &=& \int_{-i \infty-1/8}^{+i \infty-1/8} 
dz_1 \; (-s)^{-\epsilon} {\left(-\frac{m^2}{s} \right)}^{z_1}  \frac{\Gamma[1 - \epsilon - z_1]^2 
\Gamma(-z_1) \Gamma(\epsilon + z_1)}{\Gamma(2 - 2 \epsilon - z_1)},  
\nonumber \\ 
\end{eqnarray}
    Its singular structure due to $\eps$ poles of gamma functions involved will be shown in the chapter on the resolution of singularities in Fig.~\ref{SEpoles}. 
    Expanding to order $\epsilon^{-1}$ and summing residues we get
\begin{equation} 
\sum_{n=0}^\infty \frac{s^{n}}{n!} \frac{\Gamma^3(n + 1)}{ \Gamma(2 + 2n)} 
= \frac{4 \arcsin{[\sqrt{s}/2]}}{\sqrt{4 - s}\sqrt{s}} .
\label{eq:SEsum} \end{equation}
The relation in Eq.~(\ref{eq:SEsum}) is derived in \wwwaux{SE2l2m}. 
\end{tips}

 \section{Gamma and Hypergeometric Functions, Hypergeometric Integrals}\label{sec:gamma_hyperg}

Due to their sophistication, the hypergeometric functions are more often than we think present in solutions to the exact physical problems. For instance, in classical mechanics an exact solution for the pendulum period is 
\begin{equation}
    T = 2 \pi \sqrt{\frac lg}\; \times \; _2F_1 \left[ \frac 12 , \frac 12 ; 1 ; \sin^2{\theta} \right],
\end{equation}
and in the limit of small $\theta$, the classical formula emerges $T = 2 \pi \sqrt{\frac lg} $.
For more examples, see~\cite{de2019solved}.

To prove Barnes lemmas which give  useful relations among gamma functions applied in the next chapters, relations between the   \mb{} integrals and hypergeometric functions are worth studying.   
Hypergeometric series or (ordinary) Gauss functions comes from the extension of the ordinary geometric series 
\begin{equation}
    1+x +x^2+x^3 + \ldots
\end{equation}
to the series of the form

\begin{equation}
    1+\frac{ab}{c} \frac{z}{1!}+ \frac{a{\color{black}{\bf (a+1)}}b{\color{black}{\bf(b+1)}}}{c{\color{black}{\bf(c+1)}}}\frac{z^{\color{black}{\bf 2}}}{{\color{black}{\bf 2}}!} +\ldots
    \label{eq:hyper1}
\end{equation}
\noindent
where in each next term of the sum  we increase by one the number of terms in which parameters $a,b,c$ are involved, simultanously each added parameter is increased by one comparing to the previously added factor (boldface in Eq.~(\ref{eq:hyper1})). Similarly with each next term in the sum the power of variable $z$ is increased by one, normalized by the factorial  increasing accordingly by one (denoted in boldface in Eq.~(\ref{eq:hyper1})).    

Usually the series in Eq.~(\ref{eq:hyper1}) is written in the following way 
\begin{equation}
    _2F_1[a, b; c; z].
    \label{eq:2F1}
\end{equation}
This is the Gauss function.
The two semicolons in Eq.~(\ref{eq:2F1}) separate parameters in numerators, denominators and variable (here $z$), respectively. Most of the elementary functions which appear in mathematical physics can be expressed in terms of the Gauss function, for instance

\begin{eqnarray}
e^z &=& \sum_{n=0}^{\infty} \frac{z^n}{n!} = \lim_{b \to 0} \{_2F_1[1,b;;1;z/b]\},\\
\log(1+z) &=& \sum_{n=0}^{\infty} \frac{(-z)^n}{1+n} = 
{_2F_1[1,1;2;-z]}.
\end{eqnarray}

Other functions worth mentioning which can be derived from the general function in Eq.~(\ref{eq:2F1}) are Bessel and Airy functions or the Hermite polynomials.

The definition of the Gauss function in Eq.~(\ref{eq:2F1}) can be extended to the generalized Gauss (hypergeometric) function
written in a compact way as follows

\begin{equation}
\begin{split}
&
1+\frac{a_1 a_2 \ldots a_A}{b_1 b_2 \ldots b_B}\frac{z}{1!}+
\frac{a_1 (a_1+1) a_2 (a_2+1) \ldots a_A (a_A+1)}{b_1 (b_1+1) b_2 (b_2+1) \ldots b_B (b_B+1) }\frac{z^2}{2!} +\cdots 
\\
\equiv & \sum_{n=0}^{\infty} \frac{(a_1)_n (a_2)_n \ldots (a_A)_n}{(b_1)_n(b_2)_n \ldots (b_B)_n} \frac{z^n}{n!}.
\end{split}
\end{equation}
Here $(a)_n \equiv a(a+1)\cdots (a+n-1) = \frac{\Gamma[a+n]}{\Gamma[a]}$ is  Pochhammer's symbol.
If the sum of this series with $A+B$ parameters exist, the function can be denoted as~\cite{barnes1907b,Slater:1966}
\begin{equation}
_AF_B[a_1,a_2,\ldots,a_A; b_1,b_2,\ldots,b_B;z].    
\end{equation}
Some shorter notation is also used
\begin{equation}
\sum_{n=0}^{\infty} \frac{((a)_A)_n}{((b)_B)_n} \frac{z^n}{n!} = {_AF_B[(a); (b); z]}.  
\end{equation}
Truly, it can hardly be shorter. Also, the products of several gamma functions as they appear in the \mb{} integrals can be written as

\begin{eqnarray}
\frac{\Gamma(a_1)\Gamma(a_2)\ldots \Gamma(a_A)}{\Gamma(b_1) \Gamma(b_1) \ldots \Gamma(b_1)} = \Gamma 
\left[
\begin{array}{c}
     a_1,a_2,\ldots,a_A  \\
     b_1,b_2, \ldots, b_B 
\end{array}
\right] \equiv \Gamma[(a); (b)].
\end{eqnarray}

Using the above notation, the so-called first Barnes lemma (\texttt{1BL})   can be written symbolically in the following way 

\begin{eqnarray}
\frac{1}{2 \pi i}\int_{-i\infty}^{+i\infty} \Gamma 
\left(  a+z, b+z, c-z, d-z \right) dz  &=&   \Gamma 
\left[ \begin{array}{l} a+c, a+d, b+c, b+d \\ a+b+c+d \end{array} \right] \nonumber \\
\label{eq:1stBL}
\end{eqnarray}
\noindent
with the condition $\Re (a+b+c+d)<1$. We will write this equation in a manifest form in Eq.~(\ref{barnes_lemma_1}) when considering construction of \mb{} representations for \texttt{FI}.  

\begin{tips}{\mb{} Integrals and Hypergeometric Functions}
\noindent
Using the relation in Eq.~(\ref{eq:generalresidgamma})
we can cast the \mb{} contour formula in Eq.~(\ref{eq:barnes-type}) in the following form for $|z|<1$ (by closing the contour to the right)
\begin{equation}
\begin{split}
\int_C \frac{dz}{2\pi i} 
    \frac{\Gamma(a+z)\Gamma(b+z)\Gamma(-z)}{\Gamma(c+z)} (-s)^z
        &= \sum\limits_{n=0}^\infty \frac{\Gamma(a+n) \Gamma(b+n)}{\Gamma(c+n)} \frac{s^n}{n !}    
\\
        &= \frac{\Gamma(a) \Gamma(b)}{\Gamma(c)}\;_2F_1(a,b;c;s).
\end{split}
 \label{eq:ressum2F1}  
\end{equation}
The proof for Eq.~(\ref{eq:ressum2F1}) is given in~\cite{Whittaker:1965}, see Problem~\ref{prob:mb2F1}.
Putting $b = c$, we can show that from $_2F_1(a,b;b;s)$  the \mb{} master formula in Eq.~(\ref{mb1}) follows.

The continuation of the hypergeometric series for $|s| > 1$ is performed using the intermediate formula 
\bea
F(s)&=&
\sum_{n=0}^{\infty}
\frac{\Gamma(a+n)\Gamma(1-c+a+n)\sin[(c-a-n)\pi]} {\Gamma(1+n)\Gamma(1-a+b+n)\cos(n\pi)\sin[(b-a-n)\pi]}
(-s)^{-a-n}
\\
&+& \sum_{n=0}^{\infty}
\frac{\Gamma(b+n)\Gamma(1-c+b+n)\sin[(c-b-n)\pi]} {\Gamma(1+n)\Gamma(1-a+b+n)\cos(n\pi)\sin[(a-b-n)\pi]}
(-s)^{-b-n} \nn
\eea
and yields
\begin{equation}
\begin{split}
\frac{\Gamma(a)\Gamma(b)} {\Gamma(c)} ~{}_2F_1(a,b;c;s)
&=
\frac{\Gamma(a)\Gamma(b-a)} {\Gamma(c-a)} (-s)^{-a}~{}_2F_1\left(a,1-c+a;1-b+a;\frac 1s\right)
\\
&+ \frac{\Gamma(b)\Gamma(a-b)} {\Gamma(c-b)} (-s)^{-b}~{}_2F_1\left(b,1-c+b;1-a+b;\frac 1s\right)\\
\end{split}
\label{eq:2F1inv}
\end{equation}

\end{tips}

\begin{tips}{Proof of \texttt{1BL}}

Armed with complex analysis, we can now understand how the above condition $\Re (a+b+c+d)<1$ in Eq.~(\ref{eq:1stBL}) emerges. In Fig.~\ref{fig:cond1BL} poles of gamma functions of the integrand in Eq.~(\ref{eq:1stBL}) are given.

\begin{figure}
\centering
\begin{tikzpicture}[scale=0.4]
\begin{feynman}
 
\vertex at (2,1) (i1);
\vertex at (4,1) (i2);
\vertex at (6,1) (i3);
\vertex at (8,1) (i4);
\vertex at (9,1) (i5);
\vertex at (6,3.5) (f1);
\vertex at (6,-1.5) (f2);

\node[right] at (5.5,6) {{$\Im (z)$}};
\node[right] at (13,0.7) {{$\Re(z)$}};

%a
\path (5.7,3.8) pic[rotate = 0] {cross=3pt};
\path (5.2,3.8) pic[rotate = 0] {cross=3pt};
\path (4.7,3.8) pic[rotate = 0] {cross=3pt};
\path (4.2,3.8) pic[rotate = 0] {cross=3pt};
\path (6.2,3.8) pic[rotate = 0] {cross=3pt}; 
\path (6.7,3.8) pic[rotate = 0] {cross=3pt};
\path (7.2,3.8) pic[rotate = 0] {cross=3pt};
\path (7.7,3.8) pic[rotate = 0] {cross=3pt};

%b
\path (5.7,-.7) pic[rotate = 0] {cross=3pt};
\path (5.2,-.7) pic[rotate = 0] {cross=3pt};
\path (4.7,-.7) pic[rotate = 0] {cross=3pt};
\path (4.2,-.7) pic[rotate = 0] {cross=3pt};
\path (3.7,-.7) pic[rotate = 0] {cross=3pt};
\path (3.2,-.7) pic[rotate = 0] {cross=3pt};
\path (6.2,-.7) pic[rotate = 0] {cross=3pt}; 
\path (6.7,-.7) pic[rotate = 0] {cross=3pt};
\path (7.2,-.7) pic[rotate = 0] {cross=3pt};
\path (7.7,-.7) pic[rotate = 0] {cross=3pt};

%c
\path (4.7,2.3) pic[rotate = 0] {cross=3pt};
\path (5.2,2.3) pic[rotate = 0] {cross=3pt}; 
\path (5.7,2.3) pic[rotate = 0] {cross=3pt}; 
\path (6.2,2.3) pic[rotate = 0] {cross=3pt}; 
\path (6.7,2.3) pic[rotate = 0] {cross=3pt}; 
\path (7.2,2.3) pic[rotate = 0] {cross=3pt}; 
\path (7.7,2.3) pic[rotate = 0] {cross=3pt}; 
\path (8.2,2.3) pic[rotate = 0] {cross=3pt};
\path (8.7,2.3) pic[rotate = 0] {cross=3pt};
\path (9.2,2.3) pic[rotate = 0] {cross=3pt};
\path (9.7,2.3) pic[rotate = 0] {cross=3pt};
\path (10.2,2.3) pic[rotate = 0] {cross=3pt};
\path (10.7,2.3) pic[rotate = 0] {cross=3pt};
\path (11.2,2.3) pic[rotate = 0] {cross=3pt};

%d
\path (4.7,-2.3) pic[rotate = 0] {cross=3pt};
\path (5.2,-2.3) pic[rotate = 0] {cross=3pt}; 
\path (5.7,-2.3) pic[rotate = 0] {cross=3pt}; 
\path (6.2,-2.3) pic[rotate = 0] {cross=3pt}; 
\path (6.7,-2.3) pic[rotate = 0] {cross=3pt}; 
\path (7.2,-2.3) pic[rotate = 0] {cross=3pt}; 
\path (7.7,-2.3) pic[rotate = 0] {cross=3pt}; 
\path (8.2,-2.3) pic[rotate = 0] {cross=3pt};
\path (8.7,-2.3) pic[rotate = 0] {cross=3pt};
\path (9.2,-2.3) pic[rotate = 0] {cross=3pt};
\path (9.7,-2.3) pic[rotate = 0] {cross=3pt};
\path (10.2,-2.3) pic[rotate = 0] {cross=3pt};
\path (10.7,-2.3) pic[rotate = 0] {cross=3pt}; 

%\draw [thick]  (5.5,-5) -- (5.5,1) ;
%\draw [->, thick]  (5.5,1) -- (5.5,3) ;
\draw [thick]  (5.5,3) -- (5.5,4) ;
\draw [->]  (i3) -- (13,1) ;
%\draw [->] (5.5,1) -- (5.5,6) ;
\draw (3,1) -- (6.5,1);
\draw [very thick] (5.5,-4) -- (5.5,-2.7);
\draw [very thick] (5.5,-2.7) -- (4.5,-2.7);
\draw [very thick] (5.5,-2.) -- (4.5,-2.);
\draw [very thick] (4.5,-2) arc (90:270:0.35cm);

\draw [->] (5.5,-4.3) -- (5.5,6);

%\draw [very thick] (5.5,-2) -- (5.5,-1.);
\draw [->, very thick] (5.5,-2) -- (5.5,-1.5);
\draw [very thick] (5.5,-1.5) -- (5.5,-1.1);
\draw [very thick] (5.5,-1.1) -- (8.,-1.1);
\draw [very thick] (8,-1.1) arc (-90:90:0.35cm);
\draw [very thick] (8,-0.4) -- (5.5,-0.4);
\draw [very thick] (5.5,-0.4) -- (5.5,1.9);
\draw [very thick] (5.5,1.9) -- (4.5,1.9);
\draw [very thick] (4.5,1.9) arc (270:90:0.35cm);
\draw [very thick] (4.5,2.6) -- (5.5,2.6);
\draw [->,very thick] (5.5,2.6) -- (5.5,3.1);
\draw [very thick] (5.5,3.1) -- (5.5,3.5);
\draw [very thick] (5.5,3.5) -- (8,3.5);
\draw [very thick] (8,3.5) arc (-90:90:0.35cm);
\draw [very thick] (8,4.2) -- (5.5,4.2);
\draw [very thick] (5.5,4.2) -- (5.5,5.5);

%\draw [thick] (5.5,4) -- (9,4);
\draw [->, very thick] (5.5,5.5) arc (90:0:5cm);
\draw [very thick] (10.5,1) arc (0:-90:5cm);

\node [thick,left] at (5.5,5.5) {\bf{\large{$i R$}}};  
\node [thick,left] at (5.5,-4) {\bf{\large{$-i R$}}};
\node [thick,right,below] at (5.78,0.95) {{{$0$}}};
\node [thick,right,above] at (11.2,0.0) {\bf{\large{$R$}}};
\node [thick,left] at (9.3,3.3) {\bf{\large{$a$}}};
\node [thick,left] at (4.2,1.9) {\bf{\large{$c$}}};
\node [thick,left] at (9.3,-1.23) {\bf{\large{$b$}}};
\node [thick,left] at (4.2,-2.8) {\bf{\large{$d$}}};

\end{feynman}
\end{tikzpicture}
\caption{Integration contour and singularities in the complex $z$ plane connected with Eq.~(\ref{eq:1stBL}).}
\label{fig:cond1BL}
\end{figure}
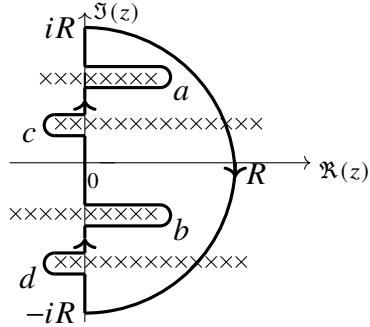

Due to the Cauchy's residue theorem discussed in section~\ref{sec:2complex} we are interested in the situation of vanishing integral round the semi-circle when approaching infinity, $R \to \infty$. 
Such integral can be parametrized with $z=R e^{it}$ as

\begin{eqnarray}
I_{\rm semi-circle}&=&\frac{-1}{2 \pi i} \int_{-\pi/2}^{\pi/2}\Gamma(a+R e^{it},b+R e^{it},c-R e^{it},d-R e^{it}) R e^{it} i dt.\nonumber \\
\end{eqnarray}
Then, taking asymptotic behaviour of the gamma function (\ref{AsymptoticGamma}), we have  $|I_{\rm semi-circle}|< {\rm const}\times  R^{-a-b+1-c+1-d} $, which vanishes with $R \to \infty$ when $\Re [a+b+c+d]<1$. 

To show that the integral over the remaining part of the closed contour in Fig.~\ref{fig:cond1BL} which is parallel to the imaginary axis from $-iR$ to $+iR$ equals to the right hand side of Eq.~(\ref{eq:1stBL}), we need to sum the residues at $z=c+n$ and $z=d+n$, $n=0,1,2,\ldots \in \mathbb{Z}$

\begin{eqnarray}
I_A &=& \sum_{n=0}^{[R]} \Gamma(a+c+n,b+c+n,d-c-n)\frac{(-1)^{n-1}}{n!}\nonumber \\
&+&\sum_{n=0}^{[R]} \Gamma(a+d+n,b+d+n,c-d-n)\frac{(-1)^{n-1}}{n!}.
\label{eq:residuesums1BL}
\end{eqnarray}

In the limit of $[R] \to \infty$ we get the two sums which involve the hypergeometric function $_2F_1$,
\begin{eqnarray}
I_A &=&  \Gamma(a+c,b+c,d-c)\; _2F_1[a+c,b+c;1+c-d;1]\nonumber \\
&+&\Gamma(a+d,b+d,c-d)\; _2F_1[a+d,b+d;1+d-c;1]. \label{eq:sum2F1}
\end{eqnarray}

If we use the Gauss's theorem~\cite{Slater:1966}
\begin{equation}
_2F_1[a,b;c;1]    = \left[ \begin{array}{l} c, c-a-b \\ c-a, c-b \end{array} \right], \label{eq:gausstheorem}
\end{equation}
we can merge the two terms into the right-hand side of Eq.~\ref{eq:1stBL}), q.e.d.

\end{tips}

The so-called second Barnes lemma reads  
\begin{eqnarray}
%\begin{equation}
&&\frac{1}{2 \pi i}\int_{-i\infty}^{+i\infty} \Gamma 
\left[ \begin{array}{l} a+z, b+z, b+c, c+z,1-d-z,-z \\ e+z \end{array} \right]  dz  \nonumber \\ 
&=&   \Gamma 
\left[ \begin{array}{l}  a,b,c,1-d+a,1-d+b,1-d+c \\ e-a,e-b,e-c \end{array} \right],  \label{eq:2ndBL}
\end{eqnarray}
with Saalschutzian condition $1+a+b+c=d+e$.

 We will write this equation in a manifest form when considering physics applications in Eq.~(\ref{barnes_lemma_2}).
 
Barnes Lemmas Eqs.~(\ref{eq:1stBL})~and~(\ref{eq:2ndBL}) can be generalized  with the help of Gauss summation theorem for the hypergeometric function $_2 F_1$ into a general theorem which includes all the Barnes-type integrals~\cite{Slater:1966} 

\begin{eqnarray}
&&\frac{1}{2 \pi i}\int_{-i\infty}^{+i\infty} \Gamma 
\left[ \begin{array}{l} (a)+z, (b)-z\\ (c)+z, (d)-z \end{array} \right] dz \label{genhyp}\\
&&\nonumber \\
&=& \sum\limits_{\mu=1}^{A} \Gamma 
\left[ \begin{array}{l} (a)-a_\mu, (b)+a_\mu\\ (c)-a_\mu, (d)+a_\mu \end{array} \right] 
 {_{B+C}F_{A+D-1}} \left[ \begin{array}{l} (b)+a_\mu, 1+a_\mu-(c);\\ 1+a_\mu-(a)', (d)+a_\mu; \end{array} (-1)^{A+C} \right] \nonumber  \\
 &&\nonumber \\
&=& {\small{\sum\limits_{\nu=1}^{B} \Gamma 
\left[ \begin{array}{l} (a)+b_\nu, (b)-b_\nu\\ (c)+b_\nu, (d)-b_\nu \end{array} \right] 
 {_{A+D}F_{B+C-1}} \left[ \begin{array}{l} (a)+b_\nu, 1+b_\nu-(d);\\ (c)+b_\nu-(a), 1+b_\nu-(b)'; \end{array} (-1)^{B+D} \right]. }}
 \nonumber
\end{eqnarray} 

 where
  \begin{eqnarray}
  \Re (c+d-a-b)&>&0,\;  A-C=B-D \geq 0.
\label{genhypcond}  \end{eqnarray}

Here prime in the notation $(\ldots)'$ indicates that denominators in Eq.~(\ref{genhyp}) give factorial elements in each term of the series.

In principle, using Eq.~(\ref{genhyp}) one can perform the integration over any free $z$ variable.
However, utility of such integration is limited due to the complexity of an emerging representation for generalized hypergeometric functions $_A F_B(;\pm 1)$. So far in practice only the first and second Barnes lemmas are in use, see next chapters.
 
More efficient way to decrease dimensionality of \mb{} integrals is to perform a transformation  of variables in such a way that  Barnes lemmas can be applied. 
Such an approach has been  implemented in the package \texttt{barnesroutines.m}~\cite{mbtools-kosower} 
for \mb-integrals with fixed contours which appear after expansion of \mb-reprepresntation in $\epsilon$. 
It performs a search for suitable  transformations with
restrictions to absolute values of matrix elements for this transformation. See section~\ref{sec:BLeff} for details. 

To finish this section, let us mention some multivariate generalizations of hypergeometric functions. This is a vast topic and here we are only scratching the surface~\cite{Abreu:2019wzk}. Let us begin with generalizations to two variables. Perhaps the most well-known of these are the Appell functions $F_1$, $F_2$, $F_3$ and $F_4$. These can be defined by the following double sums
\begin{align}
    F_1(\alpha,\beta,\beta',\gamma,x,y) &= 
    \sum_{m=0}^{\infty}\sum_{n=0}^{\infty}
    \frac{(\alpha)_{m+n} (\beta)_m (\beta')_n}
    {(\gamma)_{m+n} m! n!} x^m y^n, \qquad
    |x| < 1,\;\; |y|<1,
\label{eq:Appel-F1-def}
\\
    F_2(\alpha,\beta,\beta',\gamma,\gamma',x,y) &= 
    \sum_{m=0}^{\infty}\sum_{n=0}^{\infty}
    \frac{(\alpha)_{m+n} (\beta)_m (\beta')_n}
    {(\gamma)_{m} (\gamma')_{n} m! n!} x^m y^n, \qquad
    |x| + |y|<1,
\label{eq:Appel-F2-def}
\\
    F_3(\alpha,\alpha',\beta,\beta',\gamma,x,y) &= 
    \sum_{m=0}^{\infty}\sum_{n=0}^{\infty}
    \frac{(\alpha)_{m} (\alpha')_{n} (\beta)_m (\beta')_n}
    {(\gamma)_{m+n} m! n!} x^m y^n, \qquad
    |x| < 1,\;\; |y|<1,
\label{eq:Appel-F3-def}
\\
    F_4(\alpha,\beta,\gamma,\gamma',x,y) &= 
    \sum_{m=0}^{\infty}\sum_{n=0}^{\infty}
    \frac{(\alpha)_{m+n} (\beta)_{m+n}}
    {(\gamma)_{m} (\gamma')_{n} m! n!} x^m y^n, \qquad
    |\sqrt{x}| + |\sqrt{y}|<1,
\label{eq:Appel-F4-def}
\end{align}
where we have indicated the region of convergence for each sum. These double sums are obviously generalizations of the hypergeometric series for the Gauss function, but extended to two variables. Other generalizations can be defined along the same lines and in fact, the Appell functions are just the first four in the set of 34 so-called Horn functions. More elaborate generalizations, still of two variables, such as the Kamp\'e de F\'eriet function have also been considered. Indeed, all four Appell functions are just special cases of the more general  Kamp\'e de F\'eriet function.

Returning to the Appell functions, we mention that they also admit representations as two-dimensional real Euler-type integrals, as well as \MB{} integrals. Here we limit ourselves to presenting these for the first Appell function $F_1$ only, which we will encounter later in section~\ref{sec:expansion-special-fcns}. We have
\begin{equation}
\begin{split}
&
    F_1(\alpha,\beta,\beta',\gamma,x,y) =
    \frac{\Gamma(\gamma)}{\Gamma(\beta)\Gamma(\beta')\Gamma(\gamma-\beta-\beta')}
\\ &\qquad\times
    \int_0^1 du\, \int_0^{1-u} dv\,
    u^{\beta-1} v^{\beta'-1}
    (1-u-v)^{\gamma-\beta-\beta'-1}
    (1 - u x - v y)^{-\alpha},
\end{split}
\end{equation}
but the $F_1$ function (though not the other Appell functions) actually admits a simple one-dimensional real integral representation as well,
\begin{equation}
    F_1(\alpha,\beta,\beta',\gamma,x,y) =
    \frac{\Gamma(\gamma)}{\Gamma(\alpha)\Gamma(\gamma-\alpha)}
    \int_0^1 du\,
    u^{\alpha-1} (1-u)^{\gamma-\alpha-1}
    (1-u x)^{-\beta}
    (1-u y)^{-\beta'}.
\label{eq:F1-1d-real-int}
\end{equation}
The $F_1$ function can also be written as a two-dimensional \MB{} integral,
\begin{equation}
    \begin{split}&
    F_1(\alpha,\beta,\beta',\gamma,x,y) =
    \frac{\Gamma(\gamma)}{\Gamma(\alpha)\Gamma(\beta)\Gamma(\beta')}
    \int_{-i\infty}^{+i\infty} \frac{dz_1}{2\pi i}
    \int_{-i\infty}^{+i\infty} \frac{dz_2}{2\pi i}
\\ &\qquad\times
    \frac{\Gamma(\alpha+z_1+z_2)\Gamma(\beta+z_1)\Gamma(\beta'+z_2)}{\Gamma(\gamma+z_1+z_2)}
    \Gamma(-z_1)\Gamma(-z_2)(-x)^{z_1}(-y)^{z_2}.
    \end{split}
\end{equation}

Turning to multi-variable generalizations, we mention here the $H$-function of several variables. This function has been discussed in various forms by several authors in the literature and here we adopt the definition of~\cite{hai1995convergence}\footnote{Our definition given in Eq.~(\ref{eq:H-func-def}) is different form the $H$-function considered by~\cite{hai1995convergence} only in the replacement of $\boldsymbol{x}^{\boldsymbol{-s}}$ by $\boldsymbol{x}^{\boldsymbol{s}}$. We have made this replacement in order to unify this definition with previous expressions for special functions represented as \MB{} integrals.}. In the most general case, the $H$-function of $N$ variables is defined as follows:
\begin{equation}
    H[{\boldsymbol{x}},(\boldsymbol{\alpha},\boldsymbol{A}),(\boldsymbol{\beta},\boldsymbol{B});\boldsymbol{L}_{\boldsymbol{s}}] =
    (2\pi i)^N \int_{\boldsymbol{L}_{\boldsymbol{s}}}
    \Theta(\boldsymbol{s}) \boldsymbol{x}^{\boldsymbol{s}} 
    \boldsymbol{d}\boldsymbol{s},
\label{eq:H-func-def}
\end{equation}
where
\begin{equation}
    \Theta(\boldsymbol{s}) = 
    \frac{\prod_{j=1}^{m} 
    \Gamma\left(\alpha_j + \sum_{k=1}^{N} a_{j,k} s_k\right)}
    {\prod_{j=1}^{n} 
    \Gamma\left(\beta_j + \sum_{k=1}^{N} b_{j,k} s_k\right)}.
\end{equation}
Here $\boldsymbol{s}=(s_1,\ldots,s_N)$, $\boldsymbol{x}=(x_1,\ldots,s_N)$, $\boldsymbol{\alpha}=(\alpha_1,\ldots,\alpha_m)$ and  $\boldsymbol{\beta}=(\beta_1,\ldots,\beta_n)$ denote vectors of complex numbers, while
\begin{equation}
    \boldsymbol{A} = (a_{j,k})_{m\times N}
    \qquad\mbox{and}\qquad
    \boldsymbol{B} = (b_{j,k})_{n\times N}
\end{equation}
are matrices of real numbers. Finally
\begin{equation}
    \boldsymbol{x}^{\boldsymbol{s}} =
    \prod_{k=1}^{N} (x_k)^{s_k};\qquad
    \boldsymbol{d}\boldsymbol{s} =
    \prod_{k=1}^{N} ds_k;\qquad
    \boldsymbol{L}_{\boldsymbol{s}} = 
    L_{s_1}\times\cdots\times L_{s_N},
\end{equation}
where $L_{s_k}$ is an infinite contour in the complex $s_k$-plane running from $-i\infty$ to $+i\infty$ such that $\Theta(\boldsymbol{s})$ has no singularities for $\boldsymbol{s}\in \boldsymbol{L}_{\boldsymbol{s}}$.

The $H$-function defined in Eq.~(\ref{eq:H-func-def}) above generalizes nearly all known special functions of $N$ variables, e.g.  Lauricella functions $F^{(N)}_A$, $F^{(N)}_B$, $F^{(N)}_C$ and $F^{(N)}_D$; the $G$-function of $N$ variables; the special $H$-function of $N$ variables, etc. For the specific cases of $N = 1$ and $2$, it essentially reduces to the known Fox's $H$-function of one variable and the $H$-function of two variables defined by various authors scattered in the literature. 
Recently, the connection of Feynman integrals and Fox functions~\cite{Fox1961TheA} has been explored in~\cite{Passarino:2024ugq,Passarino:2025pbe}, where it was argued that all \texttt{FI} can be understood as $H$-functions. Formally, the statement follows from the fact that the $H$-function is essentially just a way to represent a generic \mb{} integral and we will see in the next chapter that there are algorithmic ways to dervie \mb{} representations for \texttt{FI}\footnote{Our argument here is formal, since we have not considered the question of convergence.}. The utility of this approach lies in the fact that once \texttt{FI} are represented as Fox functions, a set of tools developed for hypergeometric functions and their generalizations presents themselves. These include deriving alternate representations and analytic continuations through the Kummer transformation (which relates $H[x,\ldots]$ to $H[\frac{x}{x-1},\ldots]$), exploiting contiguity relations (these are relations between $H$ functions with parameters shifted by integers) and so on. The numerical evaluation of Fox functions directly in the physical region has also been explored in~\cite{Passarino:2025cxi}.

\section*{Problems} 
\addcontentsline{toc}{section}{Problems} 
 
\begin{problem}
\label{prob:logphase}
As we have discussed, the logarithm of the product of two complex numbers $w$ and $z$ can be written as the sum of the logarithms plus an additional phase.
\begin{equation}
    \ln(w z) = \ln w + \ln z + \eta(w,z).
    \label{eq:lnwz-prob}
\end{equation}
Show that for \emph{generic} complex numbers $w$ and $z$ (i.e., assuming that both the real and imaginary parts of $w$ and $z$ are non-zero), this phase can be written in the following form:
\begin{equation}
\eta(w,z) = 2\pi i \Bigl[ \theta(- Im\, w)\theta(- Im\, z)\theta(+ Im\, wz)
- \theta(+ Im\, w)\theta(+ Im\, z)\theta(- Im\, wz)\Bigr]\,.
%\label{appx148}
\nonumber
\end{equation}
\hint{} The argument of the left hand side of Eq.~(\ref{eq:lnwz-prob}), $\ln(w z)$, lies in the interval $(-\pi,\pi]$ by definition. Consider the circumstances under which the argument of the right hand side falls outside this interval, and how it can be mapped to this interval by adding an integer multiple of $2\pi i$.

Is the expression for $\eta(w,z)$ valid also for non-generic complex numbers? If not, why not? 
\\ \noindent \hint{} Consider the case when $w$ or $z$ (or both) is a negative real number and the case when both $w$ and $z$ are purely imaginary.

\end{problem}

\begin{problem}
\label{prob:Li2}
Consider the general dilogarithmic integral discussed in section~\ref{sec:gen-Li2-int},
\begin{equation}
    \int \frac{dx}{mx+b}\ln(nx+a).    
\end{equation}
We gave two solutions for this integral, one in Eq.~(\ref{eq:gen-Li2-int}) and one as the direct output of integration with \math{}. Show that the real parts of the two solutions always coincide up to a constant that is independent of $x$, while the imaginary parts may differ by a function of the form $C\cdot\ln(mx+b)$.
\\ \noindent \hint{} Use Eq.~(\ref{eq:li2omz}) to relate the dilogarithms in the two expressions then examine their difference. Be careful to keep in mind the proper functional identities for logarithms of complex variables such as Eqs.~(\ref{eq:ln_wtimesz})~and~(\ref{eq:ln_woverz}).
\end{problem}

\begin{problem}
\label{prob:G-shuffle}
Derive the shuffle product of $G(a;z)G(b;z)$ starting from the representations of these functions as integrals.
\\ \noindent \hint{} Consider the integral representation
\begin{equation}
G(a;z)G(b;z) = \int_0^{z} \frac{dt_1}{t_1-a}\int_0^{z} \frac{dt_2}{t_2-b}
\end{equation}
and split the two-dimensional domain of integration (the the square with corners $(0, 0)$, $(0, z)$, $(z, 0)$ and $(z, z)$) along the diagonal. See e.g.~\cite{Duhr:2014woa} for details.
\end{problem}

\begin{problem}
%\item 
\label{prob:gammaexpansion} Prove the asymptotic behavior of the gamma function for integer and complex variables, Eqs.~(\ref{AsymptoticnGamma})~and~(\ref{AsymptoticGamma}). Try to understand the origin of different behavior of the square root factors for $n \to \infty $ and  $z \to \infty$ in two formulas. 
\\ \noindent \hint{} Relation in  Eq.~(\ref{AsymptoticnGamma}) is the so-called Stirling's formula, the proof can be found among others in many statistics textbooks. More terms in expansion are given in Eq.~(\ref{gammaAsy}).
\end{problem}

\begin{problem}
\label{prob:gammaiasy} Consider the behaviour of the gamma function along the imaginary axis, i.e., $\Gamma(iy)$, with $y\in \mathbb{R}$. Show that for $|y|\to\infty$, both the real and imaginary parts of $\Gamma(iy)$ fall off as $|y|^{-1/2} e^{-i \pi |y|/2}$.
\\
\hint{} Use $z^z = e^{z\ln z}$ and the properties of the exponential function to manipulate the asymptotic formula in Eq.~(\ref{AsymptoticGamma}) to the desired form. 
\end{problem}

\begin{problem}
\label{prob:gammares}
Prove Eq.~(\ref{eq:gamma-res-n}).
\\ \noindent \hint{}  Use the recursive property of the gamma function and Eq.~(\ref{eq:gamma-res-0}).
\end{problem}

\begin{problem}
\label{prob:polygammapoles}
Argue that polygamma functions have poles of order $m+1$ at all non-positive integers.
\\ \noindent \hint{} Consider the behaviour of the gamma function around non-positive integers as given in Eq.~(\ref{eq:gamma-res-n}). Then use the definition of the polygamma functions, Eqs.~(\ref{eq:psi0-def})~and~(\ref{eq:psin-def}) to deduce the behaviour of these functions around non-positive integers.
\end{problem}

\begin{problem}
\label{prob:curvedirection}
This is a problem bridging mathematical and physics intuition. The integral in Eq.~(\ref{eq:cauchy}) in section~\ref{sec:2complex} depends on the direction of curve $C$. In the mirror world, the direction of the curve is opposite (left- and right-handedness are reversed). Changing $C_+$ to $C_-$ in Eq.~(\ref{eq:cauchy}) changes the sign of the result. Does this mean that Cauchy's integral can be determined only up to the sign and it is not well-defined if parity is conserved? How is the problem fixed in mathematics?
\\
\hint{} Consider the Jordan curve, notions of the interior and exterior, see~\cite{Zeidler:qed}.
\end{problem}

\begin{problem}
 \label{prob:residpolygamma} Prove equation~(\ref{eq:residpolygamma}). Assume that $G(z)$ is non-singular at $z=-n$, $n\in\mathbb{N}$.
 \\ \noindent \hint{} By definition, the residue is the coefficient of $(z-(-n))^{-1}$ in the Laurent-expansion of the function. One way of computing this expansion is to use Eq.~(\ref{eq:gamma-invert-sin}) to express $\Gamma(z)$ in terms of $\Gamma(1-z)$ and $\frac{\pi}{\sin\pi z}$. Then, $G(z)$ and $\Gamma(1-z)$ are regular at $z=-n$ and can be simply Taylor-expanded (note that the derivative of the $\Gamma$ function, $\psi^{(0)}$, will appear). On the other hand, $\left(\frac{\pi}{\sin\pi z}\right)^2$ is obviously singular as $\frac{1}{(z+n)^2}$ around $z=-n$ but the the Laurent-expansion is easy to construct e.g. starting form the series representation of the sine. Alternatively, you can consider the Taylor-expansion of the function $\left[\frac{(z+n)\pi}{\sin\pi z}\right]^2$. Then, it is straightforward to obtain the coefficient of $(z+n)^{-1}$.

\end{problem}
 
\begin{problem}
 \label{prob:mb2F1}  Show that from $_2F_1(a, b; c; z)$ the basic 
\mb{} formula in Eq.~(\ref{mb1}) follows: \newline  
$$\frac{1}{(A + B)^\lambda} =
        \frac{1}{\Gamma(\lambda)} \frac{1}{2 \pi i}
        \int\limits_{c-i \infty}^{c+i \infty}
        dz \Gamma(\lambda + z) \Gamma(-z)
        A^z B^{-\lambda - z}.
$$
\\ \noindent \hint{} Use a series in Eq.~(\ref{eq:ressum2F1}) with $b=c$ and apply relations given just below Eq.~(\ref{mb1}).
\end{problem}

\putbib[%
bibs/refs,%
bibs/2loops_LL16,%
bibs/Phd_Dubovyk,%
bibs/LRrefa,%
bibs/2loopsreport]
\end{bibunit}

%% file: chapter3v2.tex
%%%%%%%%%%%%%%%%%%%%% chapter.tex
%%%%%%%%%%%%%%%%%%%%%%%%%%%%%%%%%
%%%%%%%%%%%%%%%%%%%%%%%%%%%%%%%%%
\begin{bibunit}[elsarticle-num-ID] % define the bib-style for the unit: elsarticle-num.bst
%  text-1; this is the corresponding section
%\putbib[2loops] % the *.bib
%\end{bibunit}
% go-on
%--- from: bibunits.sty, adapts the font size of ``References'' to section
\let\stdthebibliography\thebibliography
\renewcommand{\thebibliography}{%
\let\section\subsection
\stdthebibliography} 

\chapter{Mellin-Barnes Representations for Feynman Integrals}
\label{chapter-MBrepr}  

\abstract{Starting from Feynman integrals defined in momentum space, Feynman and Schwinger parametrizations are introduced and their equivalence is shown. Further, the properties of Feynman graphs and the corresponding $F$ and $U$ Symanzik polynomials are discussed. The \mb{} master formula and a suite of useful programs for constructing and evaluating \mb{} integrals are introduced. Peculiarities of the construction of an \mb{} representation due to the planarity of a given diagram are examined and the notion of the loop-by-loop (\texttt{LA}) and global (\texttt{GA}) approaches are introduced. The Cheng-Wu theorem is proved and applied to the derivation of \mb{} representations for non-planar diagrams. The first and second Barnes lemmas are used in order to maximally simplify the structure and a number of complex variables in \mb{} representations. Then, specific 1-loop to 3-loop examples of \mb{} integral constructions are shown. An alternative `Method of Brackets' for constructing \mb{} representations is also introduced. Finally, the application of the \mb{} method to the computation of real phase space integrals is discussed.}

\section{Representations of Feynman Integrals \label{sec:FIrepr}}

As discussed in section~\ref{sec:singQFT_general}, our starting point is the integral in eq.~(\ref{eq-bha}). Omitting the prefactors $e^{\eps\gamma L}$ and $(2\pi\mu)^{(4-d)L}$, let us focus on scalar integrals with $T(k) =1$,

\begin{equation}\label{eq-bha_simpl}
%\boxed{
G_L[1]
=
\frac{1}{(i\pi^{d/2})^L} \int \frac{d^dk_1 \ldots d^dk_L~~}
     {(q_1^2-m_1^2)^{n_1} \ldots (q_i^2-m_i^2)^{n_j} \ldots
       (q_N^2-m_N^2)^{n_N}  }  .
%}
\end{equation}

A single Feynman propagator $D_i$ is of the form
\begin{equation}
 D_i = q_i^2 - m_i^2 + i \delta = \left[
 \sum \limits_{l=1}^{L} c_i^l k_l + \sum \limits_{e=1}^{E} d_i^e p_e
 \right]^2 - m_i^2  + i \delta,
\end{equation}
where $k_l$ and $p_e$ are internal and external momenta respectively. The $c_i^l, \, d_i^e \in [-1,1]$ are integer
coefficients and depend on a particular topology. Here we also added a small imaginary part $i \delta$ to the propagator as was discussed in section~\ref{sec:singQFT_general}.

To proceed further we introduce a generalized Feynman parameter representation
\begin{eqnarray}
 \frac{1}{D_1^{n_1} D_2^{n_2} \ldots D_N^{n_N}} &=& \frac{\Gamma(n_1+ \ldots + n_N)}{\Gamma(n_1) \ldots \Gamma(n_N)} \\
 &&
 \int \limits_{0}^{1} d x_1 \ldots \int \limits_{0}^{1} d x_N 
 \frac{x_1^{n_1-1} \ldots x_N^{n_N-1} \delta(1 - x_1 - \ldots - x_m)}
      {(x_1 D_1 + \ldots + x_N D_N)^{N_{\nu}}}  \nonumber 
\label{feynparametrization1}      
\end{eqnarray}
with $N_{\nu} = n_1+ \ldots + n_N$.

An alternative is the so-called Schwinger or $\alpha$-parameter representation, which comes from the generalized identity (Problem~\ref{prob:expIdent})
\begin{eqnarray}
 \frac{1}{D_1^{n_1} D_2^{n_2} \ldots D_N^{n_N}} &=& \frac{i^{-N_{\nu}}}{\Gamma(n_1) \ldots \Gamma(n_N)} \label{eq:expIdent}\\
&& \int \limits_{0}^{\infty} d \alpha_1 \ldots \int \limits_{0}^{\infty} d \alpha_N 
 \alpha_1^{n_1-1} \ldots \alpha_N^{n_N-1} e^{i[\alpha_1 D_1 + \ldots + \alpha_N D_N]}. \nonumber 
\end{eqnarray}

Here it is useful to show the connection between these two representations; we
will use it later in section~\ref{sec:CWtheorem} discussing the Cheng--Wu theorem. Using the identity
\begin{equation}
 1 = \int \limits_{0}^{\infty} \frac{d \lambda}{\lambda} 
 \delta \left( 1 - \frac{1}{\lambda} \sum \limits_{i=1}^{N} \alpha_i  \right) \label{eq:deltaIdent}
\end{equation}
and changing variables from $\alpha_i$ to $\alpha_i = \lambda x_i$ , one can find
\begin{eqnarray}
 \frac{1}{D_1^{n_1} D_2^{n_2} \ldots D_N^{n_N}} = \frac{i^{-N_{\nu}}}{\Gamma(n_1) \ldots \Gamma(n_N)} && 
 \int \limits_{0}^{\infty} d x_1 \ldots \int \limits_{0}^{\infty} d x_N \, x_1^{n_1-1} \ldots x_N^{n_N-1}  \label{eq:Sexp}  \\
 \times &&\int \limits_{0}^{\infty} d \lambda \lambda^{N_{\nu}-1} 
 \delta \left( 1 - \sum \limits_{i=1}^{N} x_i  \right)
 e^{i \lambda \sum \limits_{i=1}^{N} x_i D_i}. \nonumber
\end{eqnarray}
Integrating over $\lambda$ we come to the Feynman parameter representation of eq.~(\ref{feynparametrization1}).
Note that all $x_i$ are positive while the sum of $x_i$ must be unity. 
Therefore the integration region can be limited:
\begin{equation*}
 0 < x_i < 1  \,\,\, \Leftrightarrow \,\,\, 0 < x_i < \infty .
\end{equation*}

Let us now consider the momentum dependent function
\begin{equation}
 m^2(\vec x) = x_1 D_1 + \ldots + x_i D_i + \ldots + x_N D_N
    = k_i M_{ij} k_j - 2 Q_j k_j + J
\label{m2vecx}    
\end{equation}
with a $M$ -- $(L \times L)$-matrix, $Q=Q(x_i, p_e)$ -- an $L$-vector and $J=J(x_i x_j, m_i^2, p_{e_i} p_{e_j})$. 
%\vskip 15mm

Before integration over loop momenta one has to perform several preparatory steps:
\begin{itemize}

 \item Shift momenta in order to remove linear terms in $k$,
 \begin{equation}
 \label{m_shifts}
  k \rightarrow k + M^{-1} Q \Rightarrow m^2 = k M k - Q M^{-1} Q + J.
 \end{equation}
 Shifts {over internal momenta} leave the integrals unchanged.

 \item Wick rotations -- transform Minkowskian space into the Euclidean for all loop momenta:
 \begin{equation*}
  k_0 \rightarrow i k_0; \,\, k_j \rightarrow k_j (1 \leqslant j \leqslant d-1) \Rightarrow
  k^2 \rightarrow - k^2; \,\,\, d^d k \rightarrow i d^d k.
 \end{equation*}

 \item Diagonalization of the matrix $M$:
\begin{equation*}
 k^{\dag} M k = (V(x)k)^{\dag} V(x) M {V(x)}^{-1} V(x)k; \,\,\, k(x) = V(x)k;
\end{equation*}
\begin{equation*}
 V M V^{-1} = M_{\textrm{diag}}; \,\,\, (V^{\dag} = V^{-1})    
\end{equation*}{}
\begin{equation*}
  k M k \Rightarrow k(x) M_{\textrm{diag}} k(x) = \sum \limits_{i=1}^{L} \alpha_i k_i^2(x).
\end{equation*}
 The operation leaves integrals unchanged.
 %M_{diag} is useless for factorization
 
 After such manipulations the function $m^2$ has the following form:
  \begin{equation*}
    m^2 = - \sum \limits_{i=1}^{L} \alpha_i k_i^2  - Q M^{-1} Q + J.
  \end{equation*}

\item Rescale $k_i$:
 \begin{equation*}
   k_i \rightarrow \sqrt{\alpha_i} k_i \Rightarrow d^d k_i \rightarrow (\alpha_i)^{-d/2} d^d k_i \,\,\,\,
   {\textrm{and}} \,\,\,\, \prod \limits_{i=1}^{L} \alpha_i = \det M .
 \end{equation*}
 
\end{itemize}

Finally, we obtain
\begin{equation}
 G_L[1] = (-1)^{N_{\nu}}(i)^L (\det M)^{-d/2} \frac{\Gamma(N_{\nu})}{\prod \limits_{i=1}^{N} \Gamma(n_i)} 
 \int dx_1 \ldots dx_N
 \int \frac{Dk_1 \ldots Dk_L}
 {\left(\sum \limits_{i=1}^{L} k_i^2 + Q M^{-1} Q - J\right)^{N_{\nu}}}
\label{eq:Fpar0} \end{equation}
or 
\begin{equation}
 G_L[1] = \frac{(i)^{L-N_{\nu}} (\det M)^{-d/2}}{\prod \limits_{i=1}^{N} \Gamma(n_i)}
 \int d\alpha_1 \ldots d\alpha_N
 \int Dk_1 \ldots Dk_L \,\, e^{ - i \left( \sum \limits_{i=1}^{L} k_i^2 + Q M^{-1} Q - J \right) },
\label{eq:Spar0}
\end{equation}
with $Dk = \dfrac{d^dk}{i \pi^{d/2}}$.

Now the integration over loop momenta can be done in a simple way (see Problem~\ref{prob:dangles})
\begin{equation}
\label{k_integration1}
 i^L \int \frac{Dk_1 \ldots Dk_L}
 {\left(\sum \limits_{i=1}^{L} k_i^2 + \mu^2(x)\right)^{N_{\nu}}} = \frac{\Gamma\left( N_{\nu} - \frac{d}{2} L \right)}{\Gamma(N_{\nu})}
 \frac{1}{(\mu^2(x))^{N_{\nu} - \frac{dL}{2}}},
\end{equation}
\begin{equation}
 \int Dk_1 \ldots Dk_L e^{ -i \left( \sum \limits_{i=1}^{L} k_i^2 + \mu^2(\alpha) \right) } = (- i)^{-Ld/2} e^{- i \mu^2(\alpha)}, \label{k_integration2}
\end{equation}
with $\mu^2(x) = Q M^{-1} Q - J$. 
The final result (Feynman parametrization) is
\begin{equation}
\boxed{ G_L[1]  =  
 \frac{(-1)^{N_{\nu}} \Gamma\left(N_{\nu}-\frac{d}{2}L\right)}
 {\prod \limits_{i=1}^{N}\Gamma(n_i)}
 \int \prod \limits_{j=1}^N dx_j ~ x_j^{n_j-1}
 \delta(1-\sum \limits_{i=1}^N x_i)
 \frac{U(\vec x)^{N_{\nu}-d(L+1)/2}}{F(\vec x)^{N_{\nu}-dL/2}}
 }
\label{FeynSgen}
\end{equation}
where we introduced two Feynman graph polynomials $U$ and $F$
\begin{align}
%\label{uf_algdef}
 m^2 = k M k - 2 Q k + J \Leftrightarrow & \,\,\, U = \det M,  \label{uf_algdefU} \\
                                         & \,\,\, F = - \det M~J + Q M^{T} Q. \label{uf_algdefF}
\end{align}

In Schwinger or alpha representation we have 
\begin{equation}
%\boxed
{ G_L[1]  =  
 \frac{(i)^{L-N_{\nu}} (det M)^{-d/2}}{\prod \limits_{i=1}^{N} \Gamma(n_i)}
 \int \prod \limits_{j=1}^N d\alpha_j ~ \alpha_j^{n_j-1}
 e^{- i \frac{F(\alpha)}{U(\alpha)}},
 }
 \label{eq:schwinger}
\end{equation}

The relations for $U$ and $F$ in eqs.~(\ref{uf_algdefU})~and~(\ref{uf_algdefF}) are crucial for further studies. They are called Symanzik polynomials\footnote{The three independent proofs of the Feynman diagram parametrization due to Symanzik (not given in the original work) were present by Nakanishi, Shimamoto, and Kinoshita. For references, see~\cite{Nakanishi:1971}.}~\cite{Symanzik:1958}. The Schwinger representation has been used  
in the context of the Method of Brackets (see section~\ref{sec:prausa}) and  multi-fold sums~\cite{Panzer:2015ida}.
Here we framed the Feynman parametrization of eq.~(\ref{FeynSgen}) as we will explore it further in the context of \mb{} integrals.

\begin{tips}{Other Useful \texttt{FI} Representations}

Other parametrizations used in Feynman integral studies in connection with expansions by regions,  unitarity or differential approaches are~\cite{Heinrich:2020ybq}

\begin{enumerate}
    \item 
 Lee-Pomeransky representation~\cite{Lee:2013hzt}:
\begin{equation}
 I(\vec{\nu})  =\frac{(-1)^{N_\nu}\Gamma\left(D/2\right)}{\Gamma\left(\left(L+1\right)D/2-N_\nu\right)\prod_{j}\Gamma\left(\nu_{j}\right)}\int\limits _{0}^{\infty}\left(\prod_{j=1}^Ndz_{j}\,z_{j}^{\nu_{j}-1}\right)({ U}+{ F})^{-D/2}\,.\label{eq:LeePom}
\end{equation}
\item Baikov representation~\cite{Baikov:1996iu}
\begin{align}
G(\nu_1\ldots\nu_N) & =\frac{\pi^{\left(L-N\right)/2}S_{E}^{(E+1-D)/2}}{\left[\Gamma((D-E-L+1)/2)\right]_{L}}\nn\\
 & \times\int\left(\prod_{i=1}^{L}\prod_{j=i}^{L+E}ds_{ij}\right)S^{(D-E-L-1)/2}\prod_{j=1}^{N}D_{j}^{-\nu_{j}},
\end{align}
where $$\left[\Gamma(x)\right]_{L}\equiv \Gamma(x)\,\Gamma(x-1)\ldots \Gamma(x-L).$$ %and  $\nu_{(j>M)}<0$.
The quantities $S$ and $S_{E}$ come from the Jacobian of the variable transformation and have the form 
\begin{align*}
S & =\det\left\{ \left.s_{ij}\right|_{i,j=1\ldots L+E}\right\} ,\quad S_{E}=\det\left\{ \left.s_{ij}\right|_{i,j=L+1\ldots L+E}\right\} \,.
\end{align*}

Here we have loop momenta $k_i$  and external momenta $p_{1},\ldots,p_{E}$ with~\cite{Lee:2010wea}  
\begin{equation}
s_{ij}=s_{ji}=k_{i}\cdot q_{j}\,;\quad i=1,\ldots,L;\quad j=1,\ldots,K,
\end{equation}
where $q_{1},\ldots,q_{L}=k_{1},\ldots,k_{L}$, $q_{L+1},\ldots,q_{L+E}=p_{1},\ldots,p_{E}$, and
$K=L+E$.

The functions $D_{j}$ are linear functions of the variables $s_{ij}$, so that\\ $\prod_{i=1}^{L}\prod_{j=i}^{L+E}ds_{ij}\propto dD_{1}\ldots dD_{N}$.
Thus, we have 
\begin{align*}
G\left(\vec{\nu}\right) & \propto\int\left(\prod_{j=1}^{N}D_{j}^{-\nu_{j}}dD_{j}\right)\,P^{(D-E-L-1)/2},
\end{align*}
where $P\left(D_{1},\ldots D_{N}\right)$ is {\it Baikov polynomial} obtained from $S$ by expressing $s_{ij}$ via $D_{1},\ldots D_{N}$.

\end{enumerate} 
\end{tips}

\section{Topological Structure of Feynman Diagrams, Graph Polynomials}
\label{sec:1MBrepr}
 
 The functions $U$ and $F$ are called graph polynomials. 
They are polynomials in the Feynman parameters and have the following properties:
\begin{itemize}
 \item They are homogeneous in the Feynman parameters, $U$ is of degree $L$, $F$ is of degree $L + 1$.
 \item $U$ is linear in each Feynman parameter. If all internal masses are zero, then also $F$ is linear
 in each Feynman parameter.
 \item In expanded form each monomial of $U$ has a coefficient $+1$.
\end{itemize}
$U$ and $F$ are the first and the second Symanzik polynomials
of the graph, respectively.
These polynomials can be also derived from the topology of the underlying graph.

To explain the graphical method of graph polynomial construction one needs to introduce some basic definitions first:
\begin{itemize}
 \item Spanning tree $T$ for the graph $G$  
  sub-graph with the following properties:
  \begin{itemize}
   \item $T$ contains all the vertices of $G$
   \item the number of loops in $T$ is zero
   \item $T$ is connected 
  \end{itemize}
  $T$ can be obtained from $G$ by deleting $L$ edges ($L$ -- number of loops in $G$).

  \item Spanning $k$--forest ${\cal T}_k$ for the graph $G$ 
   has the same properties as $T$ but it is not required that a spanning forest is connected. \\
   $F$ can be obtained from $G$ by deleting $L+k-1$ edges.
\end{itemize}
If ${\cal T}$ is the set of all spanning forests of $G$ and ${\cal T}_k$ is set of all  spanning $k$-forests of $G$
when
\begin{equation*}
 {\cal T}  =  \bigcup \limits_{k=1}^{r} {\cal T}_k \,\,\,\, (r -{\textrm{number of vertices}}).
\end{equation*}
Each element of ${\cal T}_k$ has $k$ connected components $(T_1, \ldots, T_k)$. 
As a $P_{T_i}$ we denote a set of external momenta attached to $T_i$ for a given $k$--forest.
Depending on the ``direction'' of external momenta (whether they are incoming or outgoing) they enter the $P_{T_i}$ with a different relative sign.

The graph polynomials $U$ and $F$ can be obtained from the spanning trees and the spanning $2$--forests of a graph $G$ as follows:
\begin{equation}
 U = \sum\limits_{T\in {\mathcal T}_1} \prod\limits_{e_i\notin T} x_i,
\label{eq:spanningtree}
\end{equation}
\begin{eqnarray}
 F &=& - \sum\limits_{(T_1,T_2)\in {\mathcal T}_2}
     \left( \prod\limits_{e_i\notin (T_1,T_2)} x_i \right) 
     \left( \sum\limits_{p_i\in P_{T_1}} p_i \right) \left( \sum\limits_{p_j\in P_{T_2}} p_j \right)
 + U \sum\limits_{i=1}^n x_i m_i^2 \\
 & \equiv & F_0 + U \sum\limits_{i=1}^n x_i m_i^2. \label{eq:spanning2forest}
\end{eqnarray}
A simple one-loop example how to find Symanzik polynomials is given in Fig.~\ref{fig:FUbox}.

\begin{figure} 
\hspace*{-.5cm}

\begin{tikzpicture}[scale=0.8]
%\begin{feynman}
%\vertex at (0,0) (i1);
%\vertex at (2,0) (a);
%\vertex at (4,1) (b);
%\vertex at (4,-1) (c); 

\draw [thick] (3,-4.5) -- (6.5,-4.5);
\draw [thick] (3,-6) -- (6.5,-6);  
\draw [thick] (4,-4.5) -- (4,-6); 
\draw [thick] (5.5,-4.5) -- (5.5,-6); 
\node at (3,-4.8) {{\small{$p_1$}}};
\node at (3,-5.8) {{\small{$p_2$}}};
\node at (6.5,-4.8) {{\small{$p_3$}}};
\node at (6.5,-5.8) {{\small{$p_4$}}};
\node at (0.5,-5.2) {{\small{Basic 1-loop box diagram}}};

\end{tikzpicture}

\vspace*{.5cm}

\begin{tikzpicture}
 
\draw [thick] (0,0) -- (0,1);
\draw [thick] (0,0) -- (1,0);  
\draw [thick] (1,0) -- (1,1);  

\node at (0.5,1.2) {\scriptsize{$x_1$}};
\node at (-0.3,0.5) {\scriptsize{$x_2$}};
\node at (1.2,0.5) {\scriptsize{$x_4$}};
\node at (0.5,-0.2) {\scriptsize{$x_3$}};
\node at (4.5,-1) {{\small{Trees contributing to the $U$ polynomial for the 1-loop box diagram are drawn above.}}};

\draw [thick] (2,1) -- (3,1);
\draw [thick] (3,1) -- (3,0);  
\draw [thick] (3,0) -- (2,0); 

\draw [thick] (4,1) -- (5,1);
\draw [thick] (5,1) -- (5,0);  
\draw [thick] (4,0) -- (4,1); 

\draw [thick] (6,1) -- (7,1);
\draw [thick] (6,1) -- (6,0);  
\draw [thick] (7,0) -- (7,1); 

\node at (9,0.5) {{\bf \small{$U=x_1+x_2+x_3+x_4$}}};

\draw [thick] (0,-3) -- (0,-2);
\draw [thick] (1,-3) -- (1,-2);  

\draw [thick] (2,-2) -- (3,-2);
\draw [thick] (2,-3) -- (3,-3);  

\node at (6,-2.5) {{\bf \small{$F=t \cdot x_1 x_3+s \cdot x_2 x_4$}}};

\node at (4.5,-4) {{\small{2 - trees contributing to the $F$ polynomial for the 1-loop box diagram are drawn above.}}};

\end{tikzpicture}
 
 \caption{Graphical construction of $F$ and $U$ Symanzik polynomials. Kinematic variables $t$ and $s$ are defined as $t=(p_1-p_3)^2$ and $s=(p_1+p_2)^2$. External particles are considered massless. \label{fig:FUbox}}
% \end{center}
 \end{figure}

Cuts of internal lines (lines removed in Fig.~\ref{fig:FUbox}) are made according to eqs.~(\ref{eq:spanningtree})~and~(\ref{eq:spanning2forest}) such that:
% \vspace*{-.45cm} 
 \begin{itemize}
 \item $U$: (i) every vertex is still connected to every other vertex
 by a sequence of uncut lines; (ii) no further cuts are made without violating (i).
 \item $F$: (iii) 
 the cuts divide the graph into two disjointed parts such that within each part
 (i) and (ii) are valid and at least one external momentum line
 is connected to each part.
 \end{itemize}

\begin{svgraybox}
Regarding Fig.~\ref{fig:FUbox}, let us note that the delta function  $ \delta(1-\sum \limits_{i=1}^N x_i)$ in eq.~(\ref{FeynSgen}) in any 1-loop diagram goes over all variables $x_i$, so $U=1$. This feature will be used in section~\ref{sec:LAapproach}.
\end{svgraybox}

As a next example, for a general vertex in Fig.~\ref{fig:1lvert} which can also be a part of a multiloop diagram, we have

\begin{eqnarray}
U &=& x_1 + x_2 + x_3 \equiv 1,  \\
F_0 &=& -(q_2 + q_3)^2 x_1 x_2 - q_2^2 x_1 x_3 - q_3^2 x_2 x_3, \\
F &=& F_0 + U (x_1 m_1^2 + x_2 m_2^2 + x_3 m_3^2). \label{eq:F1lvert}
\end{eqnarray}

\begin{figure}
 \centering
\begin{tikzpicture}
%\begin{feynman}
%\vertex at (0,0) (i1);
 
\draw [thick] (0,0) -- (1,0);
\draw  (1,0) -- (3,1);  
\draw  (1,0) -- (3,-1); 
\draw  (3,1) -- (3,-1);
\draw [dashed] (3,1) -- (4,1.5);
\draw [dashed] (3,-1) -- (4,-1.5);

\node at (0.2,.2) {\scriptsize{$q_1=q_2+q_3$}};
\node at (1.6,0.7) {\scriptsize{$x_1,m_1$}};
\node at (1.6,-0.8) {\scriptsize{$x_2,m_2$}};
\node at (3.6,-1) {\scriptsize{$q_3$}};
\node at (3.6,1.) {\scriptsize{$q_2$}};
\node at (3.4,0.) {\scriptsize{$x_3,m_3$}}; 
\end{tikzpicture}
 \caption{A general one-loop vertex, which can be a part of a multiloop diagram. \label{fig:1lvert}}
% \end{center}
 \end{figure}
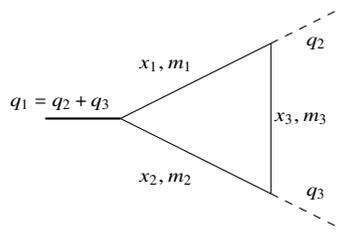

A more complicated example for a non-planar two-loop vertex is shown in Fig.~(\ref{fig:2looptrees}) where single terms for trees and forests are shown, 
see Problem~\ref{problemFU} and further examples are used in next sections.

\begin{figure}
 \centering
\begin{tikzpicture}[scale=1.1]
 
\draw [thick] (0,0) -- (1,0);
\draw  [thick] (1,0) -- (4,1);  
\draw  [thick] (1,0) -- (4,-1); 
\draw  [thick] (2,.31) -- (2.35,0.05);
\draw  [thick] (2.55,-0.08) -- (3.5,-.85);

\draw [thick] (2,-.36) -- (3.5,0.85);

\draw [dotted,red,very thick] (1.5,.4) -- (1.5,0);
\draw [dotted,red,very thick] (2.5,.7) -- (2.5,0.2);

\node at (1.2,0.4) {\scriptsize{$1$}}; 
\node at (1.2,-0.3) {\scriptsize{$6$}}; 

\node at (2.7,0.8) {\scriptsize{$2$}}; 
\node at (2.7,-0.8) {\scriptsize{$5$}}; 

\node at (3.,0.3) {\scriptsize{$3$}}; 
\node at (3.,-0.3) {\scriptsize{$4$}}; 

\node at (4.,0) {\scriptsize{{\bf{ $U \sim x_1 x_2$}}}}; 

%%%%F

\draw [thick] (5,0) -- (6,0);
\draw  [thick] (6,0) -- (9,1);  
\draw  [thick] (6,0) -- (9,-1); 
\draw  [thick] (7,.31) -- (7.35,0.05);
\draw  [thick] (7.55,-0.08) -- (8.5,-.85);

\draw [thick] (7,-.36) -- (8.5,0.85);

\draw [dotted,red,very thick] (6.5,.4) -- (6.5,-.4);
\draw [dotted,red,very thick] (7.5,.7) -- (7.5,0.2);

\node at (6.2,0.4) {\scriptsize{$1$}}; 
\node at (6.2,-0.3) {\scriptsize{$6$}}; 

\node at (7.7,0.8) {\scriptsize{$2$}}; 
\node at (7.7,-0.8) {\scriptsize{$5$}}; 

\node at (8.,0.3) {\scriptsize{$3$}}; 
\node at (8.,-0.3) {\scriptsize{$4$}}; 

\node at (9.,0) {\scriptsize{{\bf{ $F \sim -s x_1 x_2 x_6$}}}}; 

\end{tikzpicture}
  \caption{ \label{fig:2looptrees}
        Construction of the $U$ and $F$ polynomials for the non-planar massless vertex {\rm V6l0m}, single terms are shown.
See Problem~\ref{problemFU} for a complete solution.}
% \end{center}
 \end{figure}
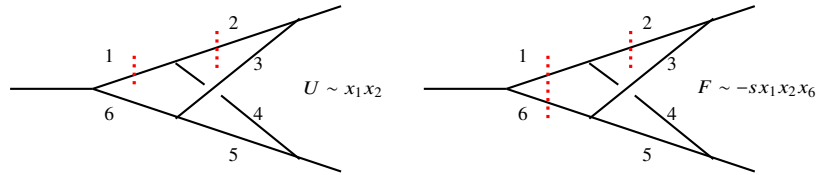
For practical purposes, $F$ and $U$ polynomials can be determined algebraically using the definitions of eqs.~(\ref{uf_algdefU})~and~(\ref{uf_algdefF}). This is encoded for instance in the package \texttt{MB.m}~\cite{Czakon:2005rk}, see Problem~\ref{prob:FUalgebraic}.  
  
\section{ Master Mellin-Barnes Formula: Prescription for the Contour, Proof  \label{sec:2MBrepr}}
 
The backbone of the procedure to build up Mellin-Barnes representations is the following relation:
\begin{equation}
        \frac{1}{(A + B)^\lambda} =
        \frac{1}{\Gamma(\lambda)} \frac{1}{2 \pi i}
        \int_{c-i \infty}^{c+i \infty}
        dz \Gamma(\lambda + z) \Gamma(-z)
        A^z B^{-\lambda - z},
 \label{Mellin-Barnes}  
\end{equation}
where 
\begin{itemize}
        \item the integration contour separates the poles of $\Gamma(-z)$ from those of $\Gamma(\lambda + z)$,
        \item $A$ and $B$ are complex numbers such that $| arg(A) - arg(B) | < \pi$.
\end{itemize}   
\begin{tips}{Proof of the Generic \mb{} Formula}
The proof based on the series and summing residues can be found in~\cite{Whittaker:1927}. Here we will briefly outline it. We start from a relation
\begin{equation} \frac{1}{(A+B)^\lambda}=
  \frac{1}{A^{\lambda}}\frac{1}{(1+B/A)^{\lambda}}
  \equiv
  \frac{1}{A^{\lambda}}\frac{1}{(1+\widetilde{B})^{\lambda}}
\end{equation}
which we expand as Taylor series
\begin{equation}
    \texttt{LHS}=
  \frac{1}{(1+\widetilde{B})^{\lambda}}=
  \sum_{n=0}^{\infty} (-1)^{n}\frac{\lambda \dots(\lambda+n-1)}{n!}\widetilde{B}^{n}.
\end{equation}
On the other hand, the right hand side of eq.~(\ref{Mellin-Barnes}) is:
\begin{equation}
    \texttt{RHS} =
  \frac{1}{(1+\widetilde{B})^{\lambda}}=
  \frac{1}{\Gamma (\nu)}\frac{1}{2 \pi i}
  \int_{-i \infty}^{+i \infty} dz \; \widetilde{B}^{z} \Gamma (\lambda+z)\Gamma (-z).\label{eq:mbRHS1}
\end{equation}
According to the Cauchy's residue theorem, eq.~(\ref{eq:cauchy}),  $\int_{C} f(z)dz=2 \pi i \sum_{i} res_{z_{i}}f$, by closing the integration contour to the right, see Fig.~\ref{fig:MBrhs}, and taking 
a series of residues (with a minus sign) at points $z=0,1,2, \dots$ (Problem~\ref{prob:mbproof}) we obtain:
\begin{equation}
  \texttt{RHS}=\frac{1}{\Gamma (\lambda)}\frac{1}{2 \pi i}
  \int_{-i \infty}^{+i \infty}dz \; 
  2 \pi i \sum_{n=0}^{\infty}
  \frac{1}{(-1)^{n} n!}\Gamma (\lambda +n)\widetilde{B}^{n}
\label{eq:mbRHS2} \end{equation}
By putting 
$\Gamma (\lambda +n)=\lambda \dots (\lambda +n-1)\Gamma (\lambda)$,
 we can see that $\texttt{LHS}=\texttt{RHS}$.

\begin{figure}[h!]
\centering
\begin{tikzpicture}[scale=0.5]
\begin{feynman}
  
\vertex at (6,4.5) (f1);
\vertex at (6,-1.5) (f2);

\node[right] at (f1) {{$\Im(z)$}};
\node[right] at (10.5,0.65) {{$\Re(z)$}};

\draw [thick]  (2,1) -- (11,1) ;
\draw [thick]  (6,-1) -- (6,5) ;
 
\draw [->,dashed]  (5.75,-.5) -- (5.75,1.5) ;
\draw [dashed]  (5.75,1.5) -- (5.75,3.5) ;
\draw [->,dashed]  (5.75,3.5) -- (7.75,3.5) ;
\draw [dashed]  (7.75,3.5) -- (10,3.5) ;
\draw [->,dashed]  (10,3.5) -- (10,1.5) ;
\draw [dashed]  (10,1.5) -- (10,-.5) ;
\draw [->,dashed]  (10,-.5) -- (7.75,-.5) ;
\draw [dashed]  (7.75,-.5) -- (5.75,-.5) ;

\filldraw [black] (6,2) circle (2pt);
\filldraw [black] (7,2) circle (2pt);
\filldraw [black] (8,2) circle (2pt);
\filldraw [black] (9,2) circle (2pt); 
\filldraw [black] (10,2) circle (2pt);
\node [thick,black] at (8,1.5) {\bf{{$\Gamma[-z]$}}};
 
\filldraw [black] (1.5,3) circle (2pt);
\filldraw [black] (2.5,3) circle (2pt);
\filldraw [black] (3.5,3) circle (2pt); 
\filldraw [black] (4.5,3) circle (2pt);
\filldraw [black] (5.5,3) circle (2pt);
\node [thick,black] at (3,3.5) {\bf{{$\Gamma[\lambda+z]$}}};

%\node at (5.55,.7) {{$-\frac 12$}};
\node at (2.15,.7) {{$-4$}};
\node at (3.,.7) {{$-3$}};
\node at (4.,.7) {{$-2$}};
\node at (5.,.7) {{$-1$}};

\node at (6.15,.7) {{$0$}};
\node at (7.,.7) {{$1$}};
\node at (8.,.7) {{$2$}};
\node at (9.,.7) {{$3$}};
\node at (10.,.7) {{$3$}};  
 
\end{feynman}
\end{tikzpicture}
\caption{Integration contour for eq.~(\ref{eq:mbRHS1}).}
\label{fig:MBrhs}
\end{figure}
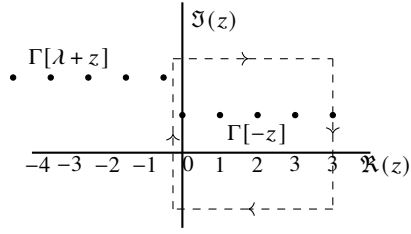

\end{tips}

The Mellin-Barnes relation in eq.~(\ref{Mellin-Barnes}) can be iterated and easily extended as a sum of several terms:
\begin{eqnarray}
  \frac{1}{(A_{1}+ \ldots +A_{n})^{\lambda}}&=&
  \frac{1}{\Gamma (\lambda)}\frac{1}{(2 \pi i)^{n-1}}
  \int_{c-i \infty}^{c+i \infty} dz_{2} \dots 
  \int_{c-i \infty}^{c+i \infty} dz_{n}
  \prod_{i=2}^{n} A_{i}^{z_{i}}
%===================
  \nonumber\\
  &\times & A_{1}^{-\lambda -z_{2}- \ldots -z_{n}}
  \Gamma (\lambda +z_{2}+ \ldots +z_{n})
  \prod_{i=2}^{n} \Gamma (-z_{i}).
        \label{MBformula}
\end{eqnarray}  
 
This formula can be applied to the Feynman parameter representation in eq.~(\ref{FeynSgen}).
Different strategies for this procedure will be discussed in the next sections, but in the end the graph polynomials
are split into pieces and the integration over Feynman parameters can be done using the relation

\begin{equation}
 \int_0^1 \prod_{i=1}^N d x_i ~ x_i^{n_i-1}
 ~ \delta(1 - x_1 - \ldots - x_N)
 =
 \frac{\Gamma(n_1) \ldots \Gamma(n_N)}
 {\Gamma\left(n_1 + \ldots + n_N \right)},
\label{XintGen} 
\end{equation}
which is a generalization of the Euler formula
\begin{equation}
 \int_0^1 d x \, x^{\alpha_1 - 1} (1 - x)^{\alpha_2 - 1} = \frac{\Gamma(\alpha_1)\Gamma(\alpha_2)}{\Gamma(\alpha_1 + \alpha_2)}.
\end{equation}

In the most general form an \mb{} representation can be written as the $I_{\mb{}}$ integral of the form

\begin{equation}
\boxed{I_{\mb{}} = \frac{1}{(2\pi i)^r} \int\limits_{-i \infty}^{+i \infty} \dots \int\limits_{-i \infty}^{+i \infty} 
 \underset{i}{\overset{r}{\Pi}} dz_i \;
 {\bf{F}}(Z,S,\epsilon)  
 \frac{\prod \limits_{j=1}^{N_n} \; \Gamma(\varLambda_j)}{\prod \limits_{k=1}^{N_d} \; \Gamma(\varLambda_k)}.
} \label{MBgenForm}
\end{equation}

\begin{equation}
\begin{array}{rl}
&\\
{\bf{F}} \;{\mbox{\rm depends on:}} 
& Z - {\mbox{\rm  set of integration variables whose length is usually smaller than $r$,}} \\
& S - {\mbox{\rm set of kinematic parameters and masses;}} \\
&\\
\varLambda_i: &  {\mbox{\rm  is a  linear combination of $z_i$ and $\epsilon$, e.g.  }} \varLambda_i = \sum \limits_l \alpha_{i l} z_l + \gamma_i + \delta_i \epsilon . \\
\end{array} \label{alphas}
\end{equation}

In practice $\bf{F}$ is a product of  $S$ elements raised to some power as shown below,
\begin{equation}
\mathbf{F} \sim \underset{k}{\Pi} \; X_k^{\sum \limits_i ( \alpha_{k i} z_i + \gamma_k + \delta_k \epsilon) },
\label{NonFreeZ}
\end{equation}
where $ \alpha_{i j},\gamma_i, \delta_i \in Z$ and $X_k$ are the ratios of kinematical invariants and masses, e.g., 
\begin{equation}
    X=\left\{ - \dfrac{s}{m_1^2}, \dfrac{m_1^2}{m_2^2}, \dfrac{s}{t}, \dots \right\}
    \label{eq:X}
\end{equation}.

 The procedure where the sum of terms in the Symanzik polynomials of eqs.~(\ref{eq:spanningtree})~and~(\ref{eq:spanning2forest}) is transformed into the \mb{} integrals is performed automatically in the \ar{} project~\cite{Gluza:2007rt,Gluza:2010rn,Dubovyk:2016ocz,ambrewww}, see also the Appendix and the webpage~\cite{mbtools} where  auxiliary packages with examples related to \mb{} calculations can be found. In the next section we show the first simple example of the construction of the \mb{} representation.

\section{Construction of \mb{} Representations for Feynman Integrals: An Example with Basic Steps} 
\label{sec:3MBrepr}  
   
Let us start from the 1-loop virtual self-energy case, Fig.~\ref{self-energy}.

\begin{figure}
\begin{center}
\begin{tikzpicture}[scale=0.6]
\begin{feynman}
 
%tadpole 

\vertex at (0.5,1) (i1);
\vertex at (2,1) (a);
\vertex at (4,1) (b);
\vertex at (5.5,1) (f1);

\diagram*{
(i1) -- [photon,thick,momentum=\(p\),/tikzfeynman/momentum/arrow distance=2mm] (a) -- [half left, momentum=\(k+p\),/tikzfeynman/momentum/arrow shorten=0.25,/tikzfeynman/momentum/arrow distance=2mm] (b) -- [thick,half left, momentum=\(k\),/tikzfeynman/momentum/arrow shorten=0.25,/tikzfeynman/momentum/arrow distance=2mm] (a),
(b) -- [photon,thick] (f1),
};
 
\end{feynman}
\end{tikzpicture}
   \caption{Two-point (self-energy) diagram. }
  \label{self-energy}
\end{center}
\end{figure}
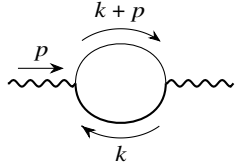
 
The scalar part of the diagram is
\begin{equation}
        G(1)_{\text{SE2l2m}} = \int \frac{d^d k}{(k^2 - m^2 + i0)^{\nu_1}((k+p)^2-m^2 + i0)^{\nu2}},
 \label{integralSE1l2m} 
\end{equation}
where all the typical constant factors were omitted. Index SE2l2m stands for self-energy (SE) one-loop (1l) with two massive (m) internal lines. We will use analogous nomenclature later on. After calculating $U$ and $F$ polynomials (see the previous section~\ref{sec:1MBrepr}) we have:
\begin{equation}
        U = x_1 + x_2, \qquad F = m^2 (x_1 + x_2)^2-s x_1 x_2-i0. \label{eq:FUse1l2m}
\end{equation}
Feynman parametrisation for this diagram reads:
\begin{eqnarray}
        G(1)_{\text{SE2l2m}} 
        &=&      
        \frac{\Gamma(\nu_1 + \nu_2 - \frac{d}{2})}{\Gamma(\nu_1) \Gamma(\nu_2)}
        \int_{0}^{1}
    \prod_{j=1}^{2}dx_j x_{j}^{\nu_{j}-1}
    \delta\left(1-\sum_{i=1}^2 x_i\right)
\nonumber\\ 
        &\times&   
        \frac{(x_1 + x_2)^{\nu_1 + \nu_2 - d}}
                 {(m^2 (x_1 + x_2)^2-s x_1 x_2-i0)^{\nu_1 + \nu_2-d/2}}. 
  \label{feynman-SE}             
\end{eqnarray}

We see that if we apply the Dirac $\delta$ function to the $U$ polynomial, which is simply a sum of Feynman parameters, we get 1. In general, every one-loop $n$-point diagram has the $U$ polynomial of the form $x_1+ \ldots +x_n$, so $U=1$ for all one-loop cases.

At this point eq.~(\ref{Mellin-Barnes}) can be used to start constructing a Mellin-Barnes representation. We use it to replace a sum of Feynman parameters in the $F$ polynomial into its product with additional integration over the complex space: 
\begin{eqnarray}
        \frac{1}{F^\lambda}
        &=&
        \frac{1}
                 {(m^2 (x_1 + x_2)^2-s x_1 x_2-i0)^{\lambda}}
\nonumber\\
    &=&                   
        \frac{1}{\Gamma(\lambda)}
        \frac{1}{2 \pi i}
        \int_{c-i \infty}^{c+i \infty} d z_1
        \Gamma(\lambda + z_1) \Gamma(-z_1)      
        (m^2-i0)^{z_1} (-s-i0)^{-\lambda - z_1}
\nonumber\\
    &\times&  
        (x_1 x_2)^{-\lambda - z_1}
        [x_1 + x_2]^{2 z_1},    
\end{eqnarray} 

where $\lambda = \nu_1 + \nu_2-d/2$. The term $[x_1 + x_2]^{2 z_1}$ in this case can be dropped since we already have the $\delta$ function. But here we show this term explicitly for generality. The term again can be changed according to eq.~(\ref{Mellin-Barnes}) resulting in:
\begin{eqnarray}
        \frac{1}{F^\lambda}
    &=&                   
        \frac{1}{\Gamma(\lambda)}
        \frac{1}{(2 \pi i)^2}
        \int_{c-i \infty}^{c+i \infty} d z_1
        \frac{1}{\Gamma(-2 z_1)}
        \int_{c-i \infty}^{c+i \infty} d z_2
        \Gamma(\lambda + z_1) \Gamma(-z_1) \Gamma(-2 z_1 + z_2) 
\nonumber\\     
        &\times&        
        \Gamma(-z_2)
        (m^2-i0)^{z_1} (-s-i0)^{-\lambda - z_1}
        x_1^{-\lambda - z_1 + z_2} x_2^{-\lambda + z_1 - z_2}.  
  \label{MBonly_SE}     
\end{eqnarray} 
Next step is to insert eq.~(\ref{MBonly_SE}) back into eq.~(\ref{feynman-SE}) and collect powers of Feynman parameters, which in our case are:
\begin{eqnarray}
        x_1^{a_1-1}&=&x_1^{(-\lambda - z_1 + z_2+\nu_1)-1},
\nonumber\\     
        x_2^{a_2-1}&=&x_2^{(-\lambda + z_1 - z_2+\nu_2)-1}.
\end{eqnarray}
The integration over Feynman parameters is performed using the following formula:
\begin{equation}
  \int_{0}^{1} \prod_{i=1}^{n}
  dx_{j}x_{j}^{a_{j}-1}
  \delta 
  \Big(
  1-\sum_{i=1}^{n}x_{i}
  \Big)=
  \frac{
    \Gamma (a_{1}) \dots \Gamma (a_{n})
    }
  {
    \Gamma (a_{1}+\ldots+a_{n})
    },
    \label{xintegration}
\end{equation}
which in practice is restricted to relevant collecting powers of Feynman parameters. The final Mellin-Barnes representation for the self-energy diagram of eq.~(\ref{integralSE1l2m}) is:
\begin{eqnarray}
        G(1)_{\text{SE2l2m}} 
        &=&      
        \frac{(-1)^{\nu_1 + \nu_2}}{\Gamma(\nu_1) \Gamma(\nu_2)}
        \int_{c-i \infty}^{c+i \infty} \frac{d z_1}{2 \pi i}  
        \int_{c-i \infty}^{c+i \infty} \frac{d z_2}{2 \pi i}
        (m^2 - i0)^{z_1} (-s - i0)^{d/2 - \nu_1 - \nu_2 -z_1}
\nonumber\\
        &\times&        
        \Gamma(-d/2 + \nu_1 + \nu_2 + z_1) \Gamma(-z_1)
        \Gamma(-2 z_1 + z_2) \Gamma(-z_2)       
\nonumber\\
        &\times&        
        \frac{\Gamma(d/2 - \nu_1 + z_1 - z_2) \Gamma(d/2 - \nu_2 - z_1 + z_2)}
        {\Gamma(-2 z_1) \Gamma(d -\nu_1 - \nu_2)}.
  \label{eq:MB-SE1l2m2dim}     
\end{eqnarray}
For more complicated cases, the $F$ polynomial often contains more than two terms. In such cases it is straightforward to use the general formula in eq.~(\ref{MBformula}). 

\begin{svgraybox}Usually in Mellin-Barnes representations infinitesimal complex part $i0$ is omitted. This doesn't mean that $i0$ is irrelevant. Even if $i0$ is omitted at the beginning, it must be recreated later to make the analytic continuation to the physical domain possible. 
\end{svgraybox} 
{\it{In general, the description of how to restore $i0$ from the propagator structure in eq.~(\ref{integralSE1l2m}) is: 
one needs to add $-i0$ to each negative invariant, e.g. $(-s)^{z_1}$ should be replaced by $(-s-i0)^{z_1}$, if $s>0$. 
This will choose the correct branch of the complex functions.}} 

The small imaginary part is not needed for positive invariants since the rising of positive numbers to a complex power doesn't lead to any ambiguities. This recipe  clearly follows from Feynman parameterization steps shown in section~\ref{sec:FIrepr}.  

Analytic continuation can be summarized nicely using the notion of $F$ and $U$ polynomials, as discussed in~\cite{Hidding:2020ytt}.

The structure of the $F$ polynomial affects the final \MB{} representation form. Obviously, this is due to the fact that eq.~(\ref{MBformula}) changes the sum of $n$ terms raised to a certain power to a $n-1$ dimensional integral over the complex space. This observation is helpful when one wants to estimate the dimensionality of the final \MB{} representation only by looking at the $F$ polynomial.

At this stage, we can further simplify the 2-dimensional MB integral in eq.~(\ref{eq:MB-SE1l2m2dim}) by applying the following Barnes lemmas  

\begin{itemize}
        \item First Barnes lemma (\texttt{1BL}) 
      \begin{equation}
      \boxed{
      \begin{aligned}
     &
            \int_{z_0-i \infty}^{z_0+i \infty} dz\,  \Gamma(a+z) \Gamma(b+z)
            \Gamma(c-z) \Gamma(d-z) 
    \\[0.5em] &=
            \frac{\Gamma(a+c)\Gamma(a+d)\Gamma(b+c)\Gamma(b+d)}{
            \Gamma(a+b+c+d)},
     \label{barnes_lemma_1}
    \end{aligned}
    }
    \end{equation}
     where 
     \begin{equation}
   a+b+c+d <1, \quad a,b,c,d \in \mathbb{R}.\label{cond1stBL}
\end{equation}
%%%%%%%%%%%%%%%%%%%%%%%%%%%%%%%
        \item Second Barnes lemma (\texttt{2BL}) 
        \begin{equation}
      \boxed{
  \begin{aligned}
    &
  \int_{z_0-i \infty}^{z_0+i \infty} dz\, \frac{\Gamma(a+z) \Gamma(b+z)
        \Gamma(c+z) \Gamma(d-z) \Gamma(e-z)}{\Gamma(f+z)}
    \\[0.5em] &=
        \frac{\Gamma(a+d)\Gamma(a+e)\Gamma(b+d)\Gamma(b+e)\Gamma(c+d)\Gamma(c+e)}{
        \Gamma(a+b+d+e)\Gamma(a+c+d+e)\Gamma(b+c+d+e)},
     \label{barnes_lemma_2}
        \end{aligned}
    }
        \end{equation}
        where 
\begin{equation}
   a+b+c+d+e = f. \label{cond2ndBL}
\end{equation}
\end{itemize}
The condition in eq.~(\ref{cond1stBL}) is fulfilled automatically as long as $\epsilon$ is not fixed.

These lemmas can be proved using a notion of hypergeometric functions as it was made for a case of \texttt{1BL} in the previous chapter, section~\ref{sec:gamma_hyperg}, for an alternative proof, see Problem~\ref{prob:1BLalternat}.

We can now apply \texttt{1BL} to eq.~(\ref{eq:MB-SE1l2m2dim}) and, simultaneously substituting powers of propagators equal one $\nu_1=\nu_2=1$ and $d=4-2 \epsilon$, we get the final, compact result:
\begin{equation}
        G_1(1)_{\text{SE2l2m}} 
        = 
        \int_{c-i \infty}^{c+i \infty} \frac{d z_1}{2 \pi i}
        (m^2)^{z_1} (-s)^{-\epsilon - z_1}
        \frac{\Gamma (1 - \epsilon - z_1)^2 \Gamma (-z_1) \Gamma (\epsilon + z_1)}
         {\Gamma (2 - 2 \epsilon - 2 z_1)}.
   \label{MB-SE1l2m-lemmas}      
\end{equation}

The above steps leading to eq.~(\ref{MB-SE1l2m-lemmas}) can be found as \wwwaux{SE2l2m}.
We will discuss this integral, solving it analytically and numerically in the following sections.
 
\section{Simplifying MB Representations}
\label{sec:5MBrepr}

In the previous section we have shown how to get the \mb{} representation for the 1-loop self-energy integral (with the same mass for the two  propagators).
If we go beyond the one loop level, the situation is more involved, especially when we consider multi-leg multi-loop scalar and tensor virtual integrals. 
Beyond one-loop:
\begin{itemize}
    \item $U(\vec x) \neq 1$,
    \item dimensionality of the \mb{} representations starts to depend on the $U(\vec x)$ structure,
    \item nontrivial simplifications of the graph polynomials are necessary.
\end{itemize}
It appears that two different strategies can be developed, which depend on the planarity of Feynman diagrams. Usually for planar diagrams a so called loop-by-loop approach (\la) is the most efficient. In this approach
the basic Mellin-Barnes relation  in eq.~(\ref{MBformula}) and Feynman parameter integrations change sums of terms in the Symanzik polynomials into integrals over complex space. These transformations are successively made for each internal momenta of the Feynman integrals.
In a planar case, there is no mismatch, and after each step an effective (planar) diagram appears. 
In addition, the \la{} has an advantage because at each step we consider one-loop subloops, so $U=1$ follows
automatically, and we have to care only about $F$ polynomial, which should be factorized using the most efficient method in order to achieve the smallest dimension of the constructed \mb{} representation.

In non-planar cases a well-defined connected diagram cannot always be drawn after each step. As a consequence, not all vertices conserve momenta~\cite{Bielas:2013rja}, see Fig.~\ref{fig:planarNP} and Problem~\ref{problem:PLNPmomenta}.
 
\begin{figure}
\begin{center}
\includegraphics[scale=0.6]{FIGS/LA}
\includegraphics[scale=0.4]{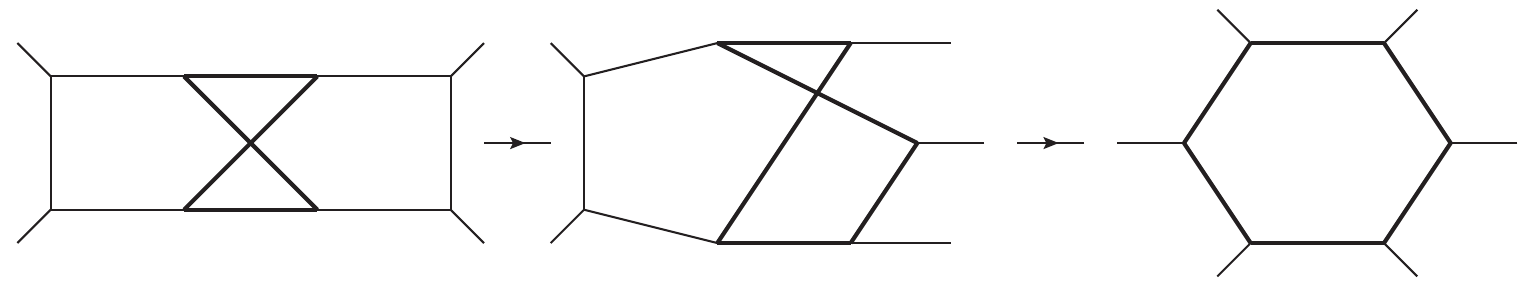}
\end{center}
\caption{Limitations of the \la{}  approach. In the planar case (the first row) momentum flow is conserved after each step in all vertices, which is not true in the non-planar case (the second row).  \label{fig:planarNP}}
\end{figure}  

That is why in this case we undertake a more natural approach where $F$ and $U$ polynomials are transformed to the \mb{} representation using eq.~(\ref{MBformula}) in one step. However, this global approach (\ga) has the price that $F$ and $U$ polynomials may result in more complicated \mb{} representations. 
Still, the 
\ga{} approach to $F$ and $U$ polynomials has an important property: {\it The polynomials are homogeneous functions of the Feynman parameters of
some given degree, depending on the underlying topology.} So, we can systematically rescale Feynman variables and in turn,
change Dirac's delta function under the integral, which defines the area of the integrand. In this case, the Dirac delta function can
be used on a restricted subset of Feynman parameters~\cite{Blumlein:2014maa}. This trick is in fact the Cheng--Wu theorem~\cite{Cheng:1987a}. We will discuss it in section~\ref{sec:CWtheorem}.

 \begin{svgraybox}
The general problem is: For which choice of Feynman variables (e.g., rescaling) the factorization of the polynomials will be the most efficient, leading to the smallest number of terms and, accordingly, to the lowest dimensionality of the \mb{} representations.
 \end{svgraybox}
%%%%%%%%%%%%%%%%%

In practice,  a choice between the \la{} and \ga{} strategies for a given integral aiming at the lowest \mb{} dimensionality of integrals seems not to be unique.
For instance, in a massless case of the non-planar two-loop diagram, a minimal four-fold \mb{} representation was derived starting from the
global Feynman parameter representation~\cite{Tausk:1999vh}. 
On the other hand, in the massive case, it appears that the \la{} is more efficient, an eightfold \mb{} representation can be obtained~\cite{Heinrich:2004iq}. In this case 
none less than a tenfold \mb{} representation is known for \ga{}~\cite{Heinrich:2004iq}. 
Moreover,  a structure of \mb{} integrals may depend also on the kinematic point and mass thresholds, which will be discussed in section~\ref{sec:2num}. Thus, it is not easy to get a general and efficient program for the construction of a broad class of \mb{} representations. 

Fig.~\ref{scheme1} shows which method is used to which type of the diagram and up to which multileg-multiloop integrals the \mb{}-suite of packages work efficiently. 

\input{FIGS/tikz.tex}

The \mb-suite  comprises several tools and the calculational procedure goes for dimensionally regulated Feynman integrals in the momentum space like this:
\begin{itemize}
\item[(i)] Start from integrals expressed by
  Feynman parameters, see eq.~(\ref{FeynSgen}) and the example in eq.~(\ref{feynman-SE}). 
 \item[(ii)]\;  Transform Feynman integrals into the \mb{} representations, see section~\ref{sec:3MBrepr} with the \texttt{AMBRE} package~\cite{Gluza:2007rt,Gluza:2010rn,Dubovyk:2016ocz,ambrewww} controlled for automation procedures by the \texttt{PlanarityTest.m} package (choice between \la{} and \ga{})~\cite{Bielas:2013v11,Bielas:2013rja}.
 \item[(iii)]\;\;   Perform an analytical continuation in $\epsilon$ with \texttt{MB.m} or \texttt{MBresolve.m}~\cite{Czakon:2005rk,Smirnov:2009up} for original integrals of dimension $d$,  $d=4-2\epsilon$,  with 
  integration paths parallel to the imaginary axis. 
 \item[(iv)]\;  Expand the Mellin-Barnes integrals as series in small $\epsilon$ with \texttt{MB.m} or \texttt{MBresolve.m}.
  \item[(v)]\;  Perform simplifications using Barnes lemmas and \texttt{barnesroutines.m}~\cite{mbtools}.
\end{itemize}

After step (v) the original representation of the Feynman integral in terms of expanded in $\epsilon$ \mb-integrals is formulated.
One may now start to calculate them, either analytically or numerically, or in a mixed approach -- See next sections and chapters.

\section{{Using Barnes Lemmas Efficiently} \label{sec:BLeff}} 

As already shown in section~\ref{sec:3MBrepr} the dimension of the representation can be decreased using the Barnes lemmas by analytical integration over $z$-variables which do not appear in powers of kinematic variables, so Barnes lemmas are an essential part of the construction of MB representations.

In the following table, we show different ways to decrease the number of MB integration for a simple polynomial which nevertheless can be a part of more complex $U$ or $F$ functions.
%sub-expression
\begin{table}
\setlength{\tabcolsep}{1em}
\renewcommand{\arraystretch}{1.5}
\centering
%\captionlistentry{factorization}
\begin{tabular}{lll}
\hline
(i) & 
$x_1 x_3 + x_1 x_4 + x_2 x_3  + x_2 x_4 $ & 3-dim representation 
\\ \hline
(ii) &
$(x_1 + x_2)(x_3 + x_4)$                  & 2-dim representation 
\\ \hline
(iii) &
$(x_1 + x_2)(x_3 + x_4)$ + BL   & 
0-dim representation 
\\ \hline
\multirow{2}{*}{(iv)} &
$(x_1 + x_2)(x_3 + x_4) \rightarrow$   
$\big[x_1 \rightarrow v_1 \xi_{11}, x_2 \rightarrow v_1 \xi_{12}$, & 
\multirow{2}{*}{0-dim representation}
\\
& $\delta(1 - \xi_{11} - \xi_{12});$ 
 $x_3 \rightarrow v_2 \xi_{21}, \ldots \big] \rightarrow v_1 v_2$ & 
\\ \hline
\end{tabular}
\caption{Simplification of graph polynomials by factorization, %, common subexpression 
Barnes lemmas and rescaling. Corresponding manipulations can be found at \wwwaux{Simpl}.}
\label{tab:Simpl1}
\end{table}
(i) Direct application of the \MB{} formula of eq.~(\ref{MBformula}) with integration using
eq.~(\ref{XintGen}) gives a three dimensional \MB{} integral. (ii) From the second row of the table, one can see that
a factorization and application of \MB{} relation to each factorized term reduces the number of \MB{} integrations to two.
(iii) Now let us consider the simplest polynomial containing only two linear terms.
As shown in the following equation, after obvious integrations steps and omitting coefficients, we
get a combination of gamma functions that exactly fits the \texttt{1BL}.
\begin{multline}
(x_1 + x_2)^p \rightarrow \int d z_1 x_1^{z_1} x_2^{p-z_1} \Gamma(-z_1) \Gamma(-p+z_1) \\
\rightarrow \int d z_1 \Gamma(-z_1) \Gamma(-p+z_1) \Gamma(z_1+1) \Gamma(p-z_1+1) / \Gamma(p+2)
\end{multline}
In this way, factorization together with \texttt{1BL} give no additional \MB{} integrations
automatically. The same situation we have had in the section~\ref{sec:3MBrepr} where we expressly
did not simplify the term $(x_1+x_2)^2$. One can recursively prove this property for a linear
combination of any length.
(iv) 
As an alternative, one can perform a rescaling of integration variables with an additional delta function, as
briefly shown in the fourth row of the Tab.~\ref{tab:Simpl1}. Such procedure leads to a very efficient
simplification that also gives no extra \MB{} integrations. This approach is used intensively for the construction
of \MB{} representations within the \texttt{GA} approach, see section~\ref{section:GA}.

Another way to achieve the same result without factorizing linear terms is to remove
as many as possible $z$-variables out of exponents of invariants and masses and then to check for \texttt{1BL}.
Such a procedure can be done by a simple linear shift of one of the $z$ variables. Usually, we have a combination
$\sum_{j=i_1}^{i_N} \alpha_j z_j$ in the exponent of one of the invariants and   
the same combination in the exponent of another one, and such combination doesn't appear in other exponents. In this situation one can make a shift $z_{i_1} \rightarrow z_{i_1} - \sum_{j=i_2}^{i_N} \alpha_j z_j$
and check Barnes' first lemma for all variables $z_{i_2} \dots z_{i_N}$. 
This algorithm is implemented in \ar~ and
works very well for the \texttt{LA} approach where the $F$ polynomial has only degree 2.
In other words, it allows to catch all linear combinations for one-loop sub-diagrams due to a
relative simplicity of the $F$ polynomial, for details see the next section.

Going beyond one-loop, graph polynomials becomes more and more complicated and factorization of linear
sub-expressions becomes not so obvious. In this case based on experience with shifts of variables described
above we can go further and try find some suitable linear transformation of integration variables which will
allow us to apply one of Barnes lemmas.
Such search have been
implemented in the package \texttt{barnesroutines.m}~\cite{mbtools-kosower} 
for \mb-integrals with fixed contours which appear after the expansion of the \mb-representation in $\epsilon$.

However, we can also apply the Barnes lemmas to the original \mb{} integral before $\epsilon$ expansion with a fixed contour of integration. This approach is interesting for two reasons. % this approach has also some disadvantages when comparing to the general \mb-representation without expansion in $\epsilon$.
First, the application of Barnes lemmas to \mb-integrals with fixed contours may lead to a large number of \mb-integrals. 
Second, starting with the $\epsilon$ expansion from a high dimensional \mb-representation we get in general a bigger cascade of integrals than starting from a representation where dimensionality was already decreased by Barnes lemmas. 

Let us discuss then a general strategy how to determine suitable $z$ transformations for applying the Barnes lemmas to the \mb{} integrals before $\eps$ expansion. 
This strategy gives the minimal number of dimensions for \mb-integrals.

We start by encoding the $z$-dependence of gamma functions of the \mb{} integral in eq.~(\ref{MBgenForm}) in a matrix form  
\begin{equation}
 M_{\Gamma} Z = \begin{bmatrix} 
   & \alpha_{ij} (\mbox{numerator}) & \\
    \hdotsfor{3} \\
   & \alpha_{ij} (\mbox{denominator}) & 
\end{bmatrix} 
\begin{pmatrix}
z_1 \\
\vdots \\
z_r
\end{pmatrix}.
\label{eq:mgamma}
\end{equation}
$M_{\Gamma}$ is a rectangular $(N_n + N_d) \times r$ matrix whose upper part contains $\alpha$-coefficients in eq.~(\ref{alphas}) 
from the numerator in eq.~(\ref{MBgenForm}) and whose bottom part contains the analogous $\alpha$-coefficients from the denominator in eq.~(\ref{MBgenForm}).
$Z$ is an $r$-vector of integration variables $z_i$. Now, any linear variable transformation can be represented as
\begin{equation}
 M_{\Gamma} Z = M_{\Gamma} U U^{-1} Z = M_{\Gamma} U Z^{\prime} , \,\,\, Z^{\prime} = U^{-1} Z,
\label{eq:Utransform} 
\end{equation}
with a non-singular $r \times r$ transformation matrix $U$. $M_{\Gamma}$ encodes a new $z$ structure of gamma functions.
Barnes lemmas can be applied if columns in $M_{\Gamma} U$ have specific structure. 
\begin{itemize}
    \item 
For the first Barnes lemma elements
in a column from $N_n + 1$ to $N_n + N_d$ must be equal to $0$ and elements from $1$ to $N_n$ must contain the set $\{1,1,-1,-1\}$
while all others must be also equal to $0$. This can be formulated in terms of overdetermined systems of linear equations
\begin{equation}
 M_{\Gamma} X = \{B_1\}.
\end{equation}
$X$ is an unknown $r$-vector representing a column in $U$ and $\{B_1\}$ is a set of all possible right hand sides.
For general $N_n$ one has $\dfrac{3 N_n !}{4! (N_n - 4)!}$ different r.h.s.

Because systems are overdetermined they may have no solutions at all. This means that there are no transformations leading to the Barnes' first lemma. If some number of solutions
$n_s$ is found, one can proceed further with the construction of the matrix 
$U$. Each solution represents a column in the matrix and can be placed on any position starting from the diagonal $U$.
The procedure of replacing columns in a diagonal matrix by our solutions continues until the matrix becomes singular.
After all these transformations, the maximal amount of added columns before the matrix becomes singular gives the number of integrations which can be done with the help 
of Barnes' first lemma.
\item 
For a discussion concerning efficient application of the second Barnes lemma to the construction of \mb{} integrals we refer the reader to~\cite{Dubovyk:2019ivv}. For \texttt{2BL}
the condition on arguments of gamma functions in eq.~(\ref{cond2ndBL}) must important.
\end{itemize}
The search for transformations for both lemmas can be performed recursively until no solution can be found.
The resulting dimensionality may also depend on the order in which the search is performed. Searching for a transformation for the second lemma can be successful after a search and application of the first lemma and vice-versa. The efficiency of this procedure will be shown in the next sections.

\section{Loop--by--loop  (\la{}) Approach
\label{sec:LAapproach}}

In the \la{} approach each subloop of a multiloop diagram is treated separately 
in an iterative way. Within this approach, the overall delta function applied to 
any 1-loop diagram results in $U=1$ and we have to care  only about the $F$ polynomial. 
Aspects of derivation \mb{} representations, such as the order of integration, 
simplifications of a result via Barnes lemmas, etc., using the \la{} are discussed in~\cite{Gluza:2007rt, Gluza:2010rn} and a lot of instructive examples can be found in~\cite{ambrewww}. 
For planar cases the automatic construction of \mb{} representations 
by the \la{} %(\ar ~v1.3.1 and v2.1.1) 
seems to be optimal, see, for example QED ladder diagrams shown in Fig.~\ref{ladder} and Tab.~\ref{planartable}.

\begin{figure}[!t]
 \begin{center}
 \includegraphics[scale = 0.5]{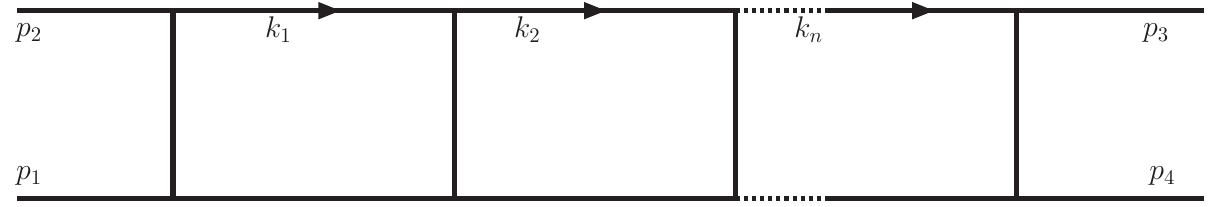}
 \end{center}
 \caption{$n$-loop ladder diagram with $k_1,...,k_n$ internal and $p_1,...,p_4$ external momenta.}
 \label{ladder}
\end{figure}

\begin{table}[!t]
\scalebox{0.96}{\parbox{\linewidth}{%
\centering
\begin{tabular}{|l|cccc|cccc|}
\hline \hline
Dimensions of planar          & \multicolumn{4}{|c|}{Massless~ cases} & \multicolumn{4}{|c|}{Massive~(QED) cases}  \\
ladder \mb{} representations  & \multicolumn{4}{|c|}{ }               & \multicolumn{4}{|c|}{ }                    \\
\hline 
Number of loops ($L$)         & 1 & 2 & 3 & 4& {1}          &  {2}          &  {3}           &  {4}                \\
\hline 
No Barnes first lemma (BL1) & 1 & 4 & 7 & 10 &  3           &   8           &   13           &   18                \\
With BL1                    &{1}&{4}&{7}&{10}&  2 ({1}+{1}) &   6 ({4}+{2}) &   10 ({7}+{3}) &   14 ({10}+{4})     \\
\hline \hline
\end{tabular}
}}
\caption[]{ Optimal results for ladder diagrams defined in Fig.~\ref{ladder}. It has been found in~\cite{Blumlein:2014maa} that $\rm{Dim(massive~ case)} = \rm{Dim(massless~ 
case)}+\#loops$.}
\label{planartable}
\end{table}
 
\subsection{\texttt{LA} Approach, Planar Example}

Let us see how it works for the 2-loop vertex in Fig.~\ref{fig:2lvertPLB}. 
The steps given below can be found at \wwwaux{V63MZ1MH}.
Such a vertex appears in the two-loop $Z$-boson decay calculations~\cite{Dubovyk:2016ocz,Dubovyk:2018rlg}. 
We consider the integral with three massive $Z$-boson $M_Z$ and one Higgs boson $M_H$ propagators,  
$m_1=m_3=m_6=M_Z$, $m_2=M_H$ and $m_4=m_5=0$. The integral representation of this diagram is the following:
\begin{align}
 I_{\rm V6l3m1M} = \int d^d k_1 d^d k_2 & \frac{1}{(k_1^2 - M_H^2)((k_1-k_2)^2 - M_Z^2)(k_2^2 - M_Z^2)} \nonumber \\ 
              & \times \frac{1}{(k_1+p_1)^2(k_2+p_1)^2((k_1+p_1+p_2)^2 - M_Z^2)}
\label{xh0w_3}
\end{align}
with $p_1^2 = p_2^2 = 0$ and $(p_1 + p_2)^2 = s$.

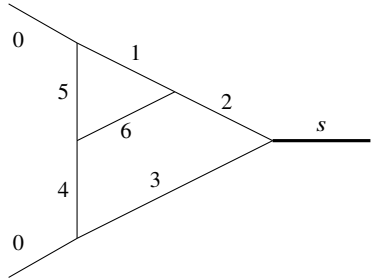
\begin{figure}
\centering
\begin{tikzpicture}[scale=1.3]
\begin{feynman}
 
\vertex at (0,-2) (i1);
\vertex at (0,2,1) (i2);
\vertex at (1,1) (i3);
\vertex at (1,0) (i4);
\vertex at (1,-1) (i5);
\vertex at (2,.5) (i6);
\vertex at (3,0) (f1);
\vertex at (4,0) (f2);
 
\draw    (0.3,1.4) -- (1,1) ; 
\draw    (1,1) -- (1,0) ;
\draw    (1,0) -- (1,-1) ;
\draw    (0.3,-1.4) -- (1,-1) ;
\draw    (1,0) -- (2,.5) ;
\draw    (2,.5) -- (1,1) ;
\draw   (2,.5) -- (3,0) ;
\draw    (1,-1) -- (3,0) ;
\draw [very thick]  (3,0) -- (4,0) ;
 
\node[below] at (0.4,1.2) {{{\bf $0$}}};
\node[above] at (0.4,-1.2) {{{\bf $0$}}};
\node at (3.5,0.15) {{{\bf $s$}}};
\node[left] at (1.,.5) {{$5$}};
\node[left] at (1.,-.5) {{$4$}};
\node at (1.5,0.1) {{$6$}};  
\node at (1.6,0.9) {{$1$}};
\node at (2.53,.4) {{$2$}};
\node at (1.8,-.4) {{$3$}};

\end{feynman}
\end{tikzpicture}
  \caption{\label{fig:2lvertPLB}
Two-loop vertex diagram for which \la{} is applied. Different mass configurations 
for propagators 1-6 lead to different dimensions of \mb{} integrals, 
see example $I^{\rm MB}_{\rm V6l3m1M}$ and eq.~(\ref{Ex1MBrepr}).}
\end{figure}
 
The minimal \mb{} dimensionality for this integral can be obtained if one first integrates over 
the subloop triangle marked by dashed lines on Fig.~\ref{LAex1fig}, and then over remaining propagators,
so the order of integration is $k_2 \rightarrow k_1$. 

For the first iteration we have
\begin{equation}
 I^{(1)}_{\rm V6l3m1M} = \int d^d k_2 \frac{1}{((k_1-k_2)^2 - M_Z^2)(k_2^2 - M_Z^2)(k_2+p_1)^2}.
\end{equation}
The $F$ polynomial corresponding to this integral is
\begin{equation}
 F^{(1)}(\vec x) = M_Z^2 x_1 + M_Z^2 x_2 - k_1^2 x_1 x_2 - (k_1 + p_1)^2 x_1 x_3.
\end{equation}
The linear dependence on Feynman parameters for $M_Z^2$ terms takes place due to the fact that $U = 1$.
The \mb{} representation for this subloop can be written as
\begin{eqnarray}
I^{\rm MB (1)}_{\rm V6l3m1M} &= & ... (-1)^{2 - \epsilon - z_1 - z_2} (M_Z^2)^{z_1 + z_2} 
(k_1^2)^{-z_3} ((k_1 + p_1)^2)^{-1 - \epsilon - z_1 - z_2 - z_3} \nonumber \\
&\times&    \Gamma(-z_1) \Gamma(-\epsilon - z_2) \Gamma(-z_2)  \Gamma(-\epsilon - z_1 - z_2 - z_3) \Gamma(-z_3) \nonumber \\
& \times & \Gamma(1 + z_2 + z_3) \Gamma(1 + \epsilon + z_1 + z_2 + z_3)/ \Gamma(1 - 2 \epsilon - z_1 - z_2).
\end{eqnarray}
Integrations over $z$-variables in a complex plane are assumed.

In the next stage we integrate over the remaining loop momentum $k_1$ 
\begin{equation}
 I_{\rm V6l3m1M}^{(2)} = \int \frac{ d^d k_1}{(k_1^2)^{z_3}(k_1^2 - M_H^2)((k_1+p_1)^2)^{2 + \epsilon + z_1 + z_2 + z_3}((k_1+p_1+p_2)^2 - M_Z^2)}.
\label{ex1vertex} 
\end{equation}
This integral has four propagators. Note that now the powers of $k_1^2$ and $(k_1+p_1)^2$ 
are complex and depends on $z_i$ variables. One notices
that during the first step it is possible to modify $F$ by the term $\pm M_H^2 x_1 x_2$ and get in the next iteration
only one propagator $k_1^2 - M_H^2$ with the complex power. However, in this case the \mb{} representation
for that modified $F$ function will have a term $(-M_H^2)^z$ which forces the integration to be always of the `Minkowskian' type, details will be explained in chapter~\ref{chapter-MBnum}. 
%In section ... we will discuss threshold effects and how to avoid such situation.
Nonetheless, such a construction is according to the \la{} strategy  where the propagators 
in eq.~(\ref{ex1vertex}) can be represented as a vertex diagram where propagators $k_1^2$ and $k_1^2-M_H^2$ are
combined into one line as shown in Fig.~\ref{LAex1fig} and this doesn't bring any complication for further construction.

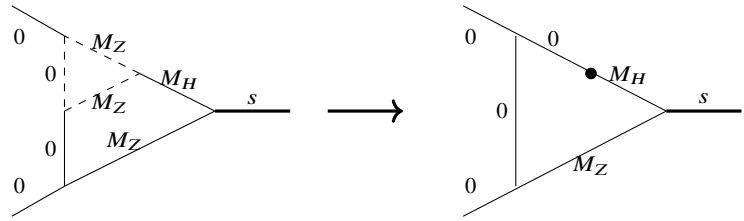
\begin{figure}
\centering
\begin{tikzpicture}[scale=1.]
\begin{feynman}
 
\vertex at (0,-2) (i1);
\vertex at (0,2,1) (i2);
\vertex at (1,1) (i3);
\vertex at (1,0) (i4);
\vertex at (1,-1) (i5);
\vertex at (2,.5) (i6);
\vertex at (3,0) (f1);
\vertex at (4,0) (f2);
 
\draw    (0.3,1.4) -- (1,1) ; 
\draw [dashed]    (1,1) -- (1,0) ;
\draw    (1,0) -- (1,-1) ;
\draw    (0.3,-1.4) -- (1,-1) ;
\draw  [dashed]  (1,0) -- (2,.5) ;
\draw  [dashed]  (2,.5) -- (1,1) ;
\draw   (2,.5) -- (3,0) ;
\draw    (1,-1) -- (3,0) ;
\draw [very thick]  (3,0) -- (4,0) ;
 
\node[below] at (0.4,1.2) {{{\bf $0$}}};
\node[above] at (0.4,-1.2) {{{\bf $0$}}};
\node at (3.5,0.15) {{{\bf $s$}}};
\node[left] at (1.,.5) {{$0$}};
\node[left] at (1.,-.5) {{$0$}};
\node at (1.6,0.1) {{$M_Z$}};  
\node at (1.6,0.9) {{$M_Z$}};
\node at (2.53,.42) {{$M_H$}};
\node at (1.8,-.4) {{$M_Z$}};

\draw [->,very thick] (4.5,0) -- (5.5,0);

\draw    (6.3,1.4) -- (9,0); ; 
\draw    (7,1) -- (7,-1) ;
\draw    (6.3,-1.4) -- (9,0) ;
\draw [very thick]  (9,0) -- (10,0) ;
\draw [fill] (8,0.5) circle [radius=2pt];
\node[below] at (6.4,1.2) {{{\bf $0$}}};
\node[above] at (6.4,-1.2) {{{\bf $0$}}};
\node[left] at (7,0) {{{\bf $0$}}};
\node[above] at (9.5,0) {{{\bf $s$}}};
\node[above] at (7.5,0.75) {{{\bf $0$}}};
\node[above] at (8.5,0.25) {{{\bf $M_H$}}}; 
\node[below] at (8,-0.5) {{{\bf $M_Z$}}}; 
  
\end{feynman}
\end{tikzpicture}
 \caption{Example of the \la{}, first iteration -- triangle subloop (dashed lines), second -- four propagators effectively combined into a vertex.}
 \label{LAex1fig}
\end{figure}

The $F$ polynomial for the final iteration is
\begin{equation}
 F^{(2)}(\vec x) = M_H^2 x_2 + M_Z^2 x_4 - s x_1 x_4 - s x_2 x_4,
\label{F2_1} 
\end{equation}
and we end up with a six dimensional \mb~ representation:
\begin{eqnarray}
 I^{\rm MB}_{\rm V6l3m1M}&=& (...) (M_H^2)^{z_4} (M_Z^2)^{z_{125}} (-s)^{-2 - 2 \epsilon - z_{1245}}\\
& \times&    \Gamma(-z_1) \Gamma(-\epsilon - z_2) \Gamma(-z_2) \Gamma(-\epsilon - z_{123}) 
    \Gamma(1 + z_{23}) \nonumber \\ 
    & \times& \Gamma(1 + \epsilon + z_{123}) \Gamma(-1 - 2 \epsilon - z_{124}) 
    \Gamma(-z_4) \Gamma(-z_5) \Gamma(-1 - 2 \epsilon - z_{1256}) \nonumber \\ 
    & \times& \Gamma(-z_6) \Gamma(-z_3 + z_6) 
    \Gamma(2 + 2 \epsilon + z_{12456}) / \Gamma(1 - 2 \epsilon - z_{12}) \Gamma(-3 \epsilon - z_{1245}). \nonumber 
\label{Ex1MBrepr}    
\end{eqnarray}
where we abbreviated $z_{123} = z_1 + z_2 + z_3$, etc.

We can further simplify this representation. Already at the stage 
of $F$ polynomials one can see that, for example, $F^{(1)}(\vec x)$ can be rewritten in the form
\begin{equation}
 F^{(1)}(\vec x) = M_Z^2 (x_1 +  x_2) - k_1^2 x_1 x_2 - (k_1 + p_1)^2 x_1 x_3 
\label{eq:F1V6l3mZ1mH} \end{equation}
and $x_1 + x_2$ can be considered as one term and the \mb{} formula can be applied recursively,
first to the $F$ polynomial as a whole, and then to the term $x_1 + x_2$. This always allows
to apply \texttt{BL1} to a $z$-variable used in the \mb{} transformation of the sum 
$x_1 + x_2$ as explained in section~\ref{sec:BLeff}. The similar trick can be also done 
for the term $-s x_4 (x_1 + x_2)$ in $F^{(2)}(\vec x)$.

However, for automation in \ar~ package another strategy was chosen (see previous section). After the shift
$z_1 \rightarrow z_1-z_2-z_5$, Barnes' first lemma is applied for variables $z_2$ and $z_6$ 
giving a 4-dimensional final representation, see \wwwaux{V6l3m1M}.

\subsection{\texttt{LA} Approach, Non-Planar Example}

Next example is a non-planar two-loop vertex shown in Fig.~\ref{fig:2lvertPLBc}.

\begin{figure}
\centering
\begin{tikzpicture}[scale=1.3]
\begin{feynman}
 
\vertex at (0,-2) (i1);
\vertex at (0,2,1) (i2);
\vertex at (1,1) (i3);
\vertex at (1,0) (i4);
\vertex at (1,-1) (i5);
\vertex at (2,.5) (i6);
\vertex at (3,0) (f1);
\vertex at (4,0) (f2);
 
\draw    (0.3,1.4) -- (1,1) ; 
\draw    (1,1) -- (2,-.5) ;
%\draw    (1,0) -- (1.8,-.5) ;
\draw    (0.3,-1.4) -- (1,-1) ;
\draw    (1,-1) -- (2,.5) ;
\draw    (2,.5) -- (1,1) ;
\draw   (2,.5) -- (3,0) ;
\draw    (1,-1) -- (3,0) ;
\draw [very thick]  (3,0) -- (4,0) ;
 
\node[below] at (0.4,1.2) {{{\bf $0$}}};
\node[above] at (0.4,-1.2) {{{\bf $0$}}};
\node at (3.5,0.15) {{{\bf $s$}}};
\node[left] at (1.3,.5) {{$5$}};
\node[left] at (1.3,-.5) {{$6$}};
%\node at (1.5,0.1) {{$6$}};  
\node at (1.6,0.9) {{$1$}};
\node at (2.53,.4) {{$2$}};
\node at (2.53,-.4) {{$3$}};
\node at (1.6,-0.83) {{$4$}};
\end{feynman}
\end{tikzpicture}
        \caption{\label{fig:2lvertPLBc}
The two-loop non-planar vertex diagram for which \la{} is also applied.}
\end{figure}
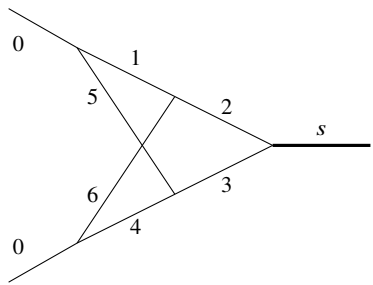
 
There is only one massive propagator, $m_1=m_2=m_3=m_4=m_5=0$ and $m_6=M_Z$. This integral is especially demanding for the \texttt{SD} method, for a discussion see~\cite{Dubovyk:2016aqv}~and~\cite{Jahn:2018zsh}.
The integral representation in this case is (see \wwwaux{V6l1MZ}) 
\begin{align}
 I^{(1)}_{\rm V6l1m} = \int \int  & \frac{d^d k_1 d^d k_2}{k_1^2 (k_1-k_2)^2 k_2^2
  ((k_1 - k_2 + p_1)^2-M_Z^2) (k_2+p_2)^2 (k_1+p_1+p_2)^2}.
\label{0h0w_15}
\end{align}
The main purpose of the example is to show difficulties related to the application of the \la{} method to non-planar diagrams.

We start with a one-loop subdiagram and the order of integrations $k_2, k_1$.
There are four propagators which include $k_2$ (box integral)
\begin{equation}
 I^{(1)}_{\rm V6l1m}= \int d^d k_2 \, \frac{1}{(k_1 - k_2)^2 (k_2)^2 ((k_1 - k_2 + p_1)^2 - M_Z^2) (k_2 + p_2)^2},
\end{equation}
which has a more complicated $F^{(1)}(\vec x)$ function:
\begin{align}
 F^{(1)}(\vec x) = & - k_1^2 x_1 x_2 - ((k_1 + p_1)^2 - M_Z^2) x_2 x_3 -(k_1 + p_2)^2 x_1 x_4 \nonumber \\ & - (k_1 + p_1 + p_2)^2 x_3 x_4 
+ M_Z^2 x_3 (x_1 + x_3 + x_4).
\end{align}

More problems appear when we go to the integration over loop momenta $k_1$
\begin{align}
 I^{(2)}_{\rm V6l1m} = \int d^d k_1 \, & \frac{1}{((k_1)^2)^{1 - z_1} ((k_1 + p_1)^2 - M_Z^2)^{-z_3} ((k_1 +  p_2)^2)^{-z_5}} \nonumber \\ 
 & \times \frac{1}{((k_1 + p_1 + p_2)^2)^{3 + \epsilon + z_{123456}}}.
\label{0h0w_15_2}  
\end{align}
In eq.~(\ref{0h0w_15_2}) the propagators do not form a diagram with conserved momentum flow. Propagators
$(k_1)^2 ((k_1 + p_1)^2 - M_Z^2) (k_1 + p_1 + p_2)^2$ correspond to a 1-loop vertex diagram but
$(k_1 +  p_2)^2$ must be considered as an artificial numerator,  in non-planar case the exponent of such numerator is not an integer but a complex number.  

After the $k_1$ integration $F^{(2)}(\vec x)$ takes the form
\begin{equation}
 F^{(2)}(\vec x) = M_Z^2 x_2 + 2 s x_2 x_3 - 2 s x_1 x_4, 
\end{equation}
and the final 6-dimensional representation can be written in the form\footnote{After $\epsilon$ expansion 
%via \mbm~ 
only maximally 5-dimensional \mb{} integrals remain, up to the constant in $\eps$.}
\begin{align}
  I^{\rm MB}_{\rm V6l1m} = & \frac{1}{(-2s)^{2 + 2 \epsilon}}  \int\limits_{-i\infty}^{i\infty} \frac{dz_1}{2 \pi i} \ldots \frac{dz_8}{2 \pi i} (-1)^{z_8} 
  \left( \frac{M_Z^2}{-2s} \right)^{z_2} \Gamma(2 + \epsilon + z_{1235} - z_7) \Gamma(-z_7)
  \nonumber \\ & \times \Gamma(-z_1) \Gamma(-1 - 2 \epsilon - z_{13}) \Gamma(1 + z_{13})  
 \Gamma(-1 - \epsilon - z_{15})
  \Gamma(1 + z_{15}) 
 \nonumber \\
 & \times \Gamma(-1 - 2 \epsilon - z_{128}) \Gamma(1 - \epsilon + z_{135} - z_{78}) \Gamma(-z_8) \Gamma(2 + 2 \epsilon + z_{28}) \Gamma(-z_5 + z_8)
 \nonumber \\
 & \times \Gamma(-z_3 + z_{78})  \Gamma(-z_2 + z_7) \Gamma(-1 - \epsilon - z_{123} + z_7) / \Gamma(-2 \epsilon) \Gamma(1 - z_1)
 \nonumber \\
 & \times \Gamma(-1 - 2 \epsilon - z_{123} + z_7)  \Gamma(-3 \epsilon - z_2) \Gamma(3 + \epsilon + z_{1235} - z_7)  \label{eq:MBnplanar_bad}.
\end{align}
The main disadvantage of this representation is the factor $(-1)^{z_8}$ with a complex variable $z_8$ which forces 
the integral to be always of the Minkowskian type.

By construction this representation is correct and in principle can be used for  the evaluation
of the integral but due to its Minkowskian form it cannot be evaluate numerically by \mbm{} which works for Euclidean kinematic. In the next sections, we will describe another approach to construct \mb{} representations for all types of non-planar diagrams and which form is similar to planar diagrams. To achieve it, we consider first the powerful Cheng-Wu theorem.

\section{Cheng-Wu Theorem \label{sec:CWtheorem}}

The theorem (\texttt{CW}) has been considered by Cheng and Wu in~\cite{Cheng:1987a}.
\begin{theorem} \label{CWtheorem0}
For the Feynman parameter representation in eq.~(\ref{FeynSgen})
the Cheng--Wu (CW) theorem states that the same formula
holds with the delta function $ \delta(1-\sum \limits_{i=1}^N x_i)$ replaced by
\begin{equation}
 \delta\left( \sum \limits_{i \in \Omega} x_i -1 \right),
 \label{CWtheorem}
\end{equation}
\noindent
where $\Omega$ is an arbitrary subset of the lines $1, \ldots, L$, when the integration over
the rest of the variables, i.e. for $i \notin \Omega$, is extended to the integration from
zero to infinity. 
\end{theorem}

\begin{proof}
% appendix C.10, Cheng:1987a
For the proof we use the identity in eq.~(\ref{eq:deltaIdent}) and restrict it to the subset $\Omega$, namely  
\begin{equation}
 1 = \int \limits_{0}^{\infty} \frac{d \lambda}{\lambda} 
 \delta \left( 1 - \frac{1}{\lambda} \sum \limits_{i=1}^{N} \alpha_i  \right) \Leftrightarrow
1 = \int \limits_{0}^{\infty} \frac{d \lambda}{\lambda} 
 \delta \left( 1 - \frac{1}{\lambda} \sum \limits_{i \in \Omega} \alpha_i  \right). \label{eq:xomega}
\end{equation}
Then we change variables from $\alpha_i$ to $\alpha_i = \lambda x_i$,  as applied for obtaining the representation in eq.~(\ref{eq:Sexp}).
The key point in the calculation is that, as noticed in~\cite{Cheng:1987a}, the terms $\prod \limits_{j=1}^N dx_j ~ x_j^{n_j-1}$, $U(x)^{N_\nu-d(L+1)/2}$ and $F(x)^{N_\nu-d L/2}$ in eq.~(\ref{FeynSgen}) are homogeneous in $x$ of order 0.
Indeed, $\prod \limits_{j=1}^N dx_j ~ x_j^{n_j-1} \sim N_\nu \equiv A$, $U(x)^{N_\nu-d(L+1)/2} \sim (N_\nu-d(L+1)/2) L \equiv B$, $F(x)^{N_\nu-d L/2} \sim (N_\nu-d L/2) (L+1) \equiv C$, as $U$ and $F$ are of order $L$ and $L+1$ respectively, so $A+B+C=0.$
%discussion after eq C66 in cheng-wu 
%Nnu+(Nnu-d(L+1)/2)*L-(Nnu-d L/2)*(L+1)//Simplify
So, we can freely rescale any subset of $x$ parameters in (\ref{eq:xomega}).
\end{proof}

In~\cite{Heinrich:2021dbf} an alternative proof of the \texttt{CW} theorem is given using a notion of sector decomposition~\cite{Binoth:2000ps} where independently of the choice of the coefficients $a_i$ in the general integral $\int_0^\infty d{\bf x}\,f({\bf x})\delta\left(1-  \sum_{i} a_ix_i\right)$,  the considered integrand does not change, $f({\bf x})$ must be the homogenous function in variables $\bf x$. We propose this proof as the Problem~\ref{problem_sd_cw}. Another proof of \texttt{CW} is given in~\cite{Jantzen:2012mw}.

\section{Global (\ga) Approach \label{section:GA}}

The second possibility to construct an \mb{} representation for a given Feynman integral is to integrate 
simultaneously over all loop momenta. In this case $U(\vec x)$ is not equal to $1$ anymore but we can avoid highly oscillating additional factors of the from $(-1)^z$ as given in the previous example, see eq.~(\ref{eq:MBnplanar_bad}). A naive way to construct a representation
is to apply the \mb{} master formula of eq.~(\ref{MBformula}) to both $U(\vec x)$ and $F(\vec x)$ polynomials and then try to simplify the result using
Barnes lemmas. However, in this way, one faces several problems. First, after integration over Feynman parameters,
one always gets $\Gamma(0)$ in the denominator. This happens precisely because of the homogeneity of the original Feynman parameters
representation. We can see it in the following way.

\begin{tips}{\ga{} and the $\Gamma(0)$ Problem}
Let us consider the sunset diagram in Fig.~\ref{fig:sunset}.
\begin{figure}
\centering
\begin{tikzpicture}[scale=1]
\begin{feynman}
 
\vertex at (0,-2) (i1);
\vertex at (0,2,1) (i2);
\vertex at (1,1) (i3);
\vertex at (1,0) (i4);
\vertex at (1,-1) (i5);
\vertex at (2,.5) (i6);
\vertex at (3,0) (f1);
\vertex at (4,0) (f2);

\draw (1,0) circle [radius=1cm];
\draw (-.5,0) -- (2.5,0); 
 
\node[above] at (-.5,0) {{{\bf $p$}}};
\node[above] at (1,0) {{{\bf $x_2, 0$}}};
\node [below] at (1,1.) {{{\bf $x_1,m$}}};
\node [above] at (1,-1.) {{{\bf $x_3, 0$}}}; 
 
\end{feynman}
\end{tikzpicture}
    \caption{ \label{fig:sunset}
        The 2-loop sunset diagram with one non-zero mass and Feynman parameters $x_1$, $x_2$, $x_3$. 
      }
\end{figure}
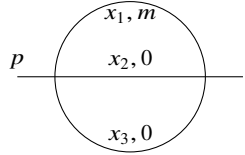

The Symanzik polynomials for the corresponding \texttt{FI} are 
\begin{eqnarray} 
U&=&x_1 x_2+x_2 x_3 +x_1 x_3, \\
F&=& -p^2 x_1 x_2 x_3 + m^2 x_1 U, \label{eq:UFsunset}
\end{eqnarray}
\noindent
for which, see eq.~(\ref{FeynSgen}), the relevant part of the integration with Feynman parameters is

\begin{equation}
I \propto \int dx_1 dx_2 dx_3 \delta(...) \frac{U^{N_\nu-d(L+1)/2}}{F^{N_\nu-dL/2}} \label{eq:FIsunset}.
\end{equation}
  Substituting eq.~(\ref{eq:UFsunset}) into eq.~(\ref{eq:FIsunset}) and using the \MB{} relation of eq.~(\ref{Mellin-Barnes}) iteratively, we have
\begin{eqnarray} I &\propto & (-p^2 x_1 x_2 x_3)^{z_1} ( m^2 x_1 U)^{-(N_\nu-dL/2)-z_1} U^{N_\nu-d(L+1)/2} \\
 \label{eq:gamma0_part1}
&\propto & (x_1 x_2 x_3)^{z_1} x_1^{-N_\nu+dL/2-z_1} U^{-d/2-z_1}\\
&\propto & (x_1 x_2 x_3)^{z_1} x_1^{-N_\nu+dL/2-z_1} (x_1 x_2)^{z_2} (x_2 x_3)^{z_3} (x_1 x_3)^{-d/2-z_2-z_3-z_1}.
\end{eqnarray}
Now we make x-integration, the important is the sum of exponents of Feynman parameters, see denominator of eq.~(\ref{xintegration}). Gathering Feynman terms together we get
\begin{equation}
x_1^{-N_\nu+dL/2-d/2-z_3-z_1} x_2^{z_1+z_2+z_3} 
x_3^{-d/2-z_2}.
\label{eq:gamma0}
\end{equation}
Hence, keeping in mind the overall prefactor $\prod_i x_i^{n_i-1}$, the exponents give $d(L-2)/2$
which for $L=2$ is $0$.
According to eq.~(\ref{xintegration}) the result translates to the zero argument of the gamma function in the denominator.
%and vanishing expression, which is bad.
From another side, we know that \MB{} representation for this integral does not vanish, so the gamma function in the denominator should be canceled after suitable integration with the help of \texttt{BL}. We propose this as the Problem~\ref{gammazero}.

To overcome this problem we apply \texttt{CW} theorem by starting from the expression in eq.~(\ref{eq:gamma0_part1}) 
\begin{equation} \label{eq:gamma0_cw1}
(x_1 x_2 x_3)^z_1 x_1^{-N_\nu+dL/2-z_1} (x_1 x_2 + x_3 (x_1+x_2))^{-d/2-z_1}.
\end{equation}
Now we use \texttt{CW} to the $\Omega$ subset in eq.~(\ref{eq:gamma0_cw1})
\begin{equation} \label{eq:gamma0_cw2}
    \int_0^1 dx_1 dx_2 dx_3 \delta(1-x_1-x_2-x_3) \to \int_0^\infty dx_3 \int_0^1 dx_1 dx_2 \delta(1-x_1-x_2)
\end{equation}
and we get 
\begin{equation} \label{eq:gamma0_cw3}
x_1^{-N_\nu+dL/2} x_2^{z_1} x_3^{z_1} (x_3+ x_1 x_2)^{-d/2-z_1}.
\end{equation}
We can evaluate the integral over $x_3$
\begin{eqnarray}
    \int\limits_0^\infty dx_2 x_3^{z_1} (x_3+ x_1 x_2)^{-d/2-z_1} = 
    x_1^{1- \frac d2} x_2^{1-\frac d2} \frac{\Gamma[1+z_1]\Gamma[-1+d/2]}{\Gamma[d/2+z_1]},
\end{eqnarray}
and adding integrations over $x_2$ and $x_3$ we get 
\begin{eqnarray}
&&    \int\limits_0^1 dx_1 \int\limits_0^1 dx_2 \delta(1-x_1-x_2) 
(x_1 x_2)^{-N_\nu+dL/2 +1- \frac d2} x_2^{z_1+1-\frac d2} \\
&=& \frac{\Gamma[-N_\nu +1 +d(L-1)/2]\Gamma[z_1+1 -d/2]}{\Gamma[-N_\nu+2+z_1+d(L-2)/2]}.
\end{eqnarray}

This time, coming back to eq.~(\ref{xintegration}) the argument of gamma in the denominator is zero and we can construct \mb{} representation. It can be shown that in general the result as given in eq.~(\ref{eq:gamma0}) gives `zero' exponents, independently of the number of loops considered.

\end{tips}

An option to regulate this singularity is to shift the exponent of one of the Feynman parameters by arbitrary $\delta$
and then take the limit $\delta \rightarrow 0$, before doing the analytic continuation in the dimensional parameter $\epsilon \rightarrow 0$.
One should stress that shifting the original exponents of propagators $n_i \rightarrow n_i + \eta$,
similar to what we do in example~\ref{tips_trick_eta}, does not help because $\Gamma(0)$ appears automatically and it does not depend on powers of propagators $n_i$. Second, even at the two-loop level, the amount of terms in $U(\vec x)$ and $F(\vec x)$ is quite large;
the number of \mb{} integrations is equal to the number of terms in a polynomial, minus one, see eq.~(\ref{MBformula}). It leads to a very high-dimensional \mb{} representation
for which it is difficult to catch all possible simplifications by Barnes lemmas.

\subsection{\texttt{GA} Approach and \texttt{CW} Theorem, the Non-Planar Double Box}

To show how to construct the \mb{} representations properly we start  with a detailed explanation of one of
the first successful implementations of the global approach to a non-trivial Feynman integral,    
namely, the non-planar massless double box diagram~\cite{Tausk:1999vh} given in Fig.~\ref{TauskNPdbox}, with on-shell external legs ($p_i^2 = 0$).

\begin{figure}[h!]
\begin{center}
\includegraphics[scale = 0.6]{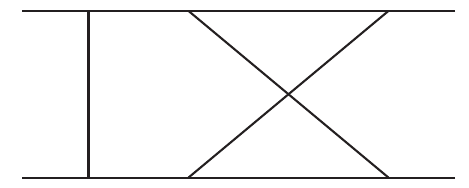}
\end{center}
\caption{The non-planar double box topology.}
\label{TauskNPdbox}
\end{figure}

We define an integral for this diagram as
(all external momenta $p_i$ are incoming in cyclic notation):
\begin{multline}
 B_7^{NP} = \iint d^d k_1 d^d k_2  
 \frac{1}{
 [(k_1+k_2+p_1+p_2)^2]^{n_1} 
 [(k_1+k_2+p_2)^2]^{n_2} 
 [(k_1+k_2)^2]^{n_3} } \\
 \frac{1}{[(k_1-p_3)^2]^{n_4} 
 [(k_1)^2]^{n_5} 
 [(k_2-p_4)^2]^{n_6} 
 [(k_2)^2]^{n_7}}.
 \label{B7NP}
\end{multline}
An explicit form of the Symanzik polynomials for this integral is, see Problem~\ref{prob:FUalgebraic}
\begin{eqnarray}
 U(x) &=& x_1 x_2 + x_1 x_4 + x_2 x_4 + x_1 x_5 + x_2 x_5 + 
 x_2 x_6 + x_4 x_6 + x_5 x_6  \label{umassless} \\ 
 &+& x_1 x_7 + x_4 x_7 + 
 x_5 x_7 + x_6 x_7,\nonumber \\
 F(x) &=& -s~ x_1 x_2 x_5 - s ~x_1 x_3 x_5 - s~ x_2 x_3 x_5 - 
 u ~x_2 x_4 x_6 \nonumber \\ 
 &-& s~ x_3 x_5 x_6 - t~ x_1 x_4 x_7 - 
 s~ x_3 x_5 x_7 - s~ x_3 x_6 x_7.
 \label{fmassless}
\end{eqnarray}
Changing sums of monomials in $x$ into products using the \texttt{MB} master formula, eq.~(\ref{MBformula}),
leads to an 18-fold \mb-integral (11 and 7 complex variables come from $U$ and $F$, respectively). 
Certainly, it can be factorized in a better way. %, if $U$ and $F$ are factorized properly.
 The factorization proposed in~\cite{Tausk:1999vh} looks as follows
\begin{eqnarray}
 U(x) &=& (x_1+x_6)(x_2+x_7) + (x_3+x_4+x_5)(x_1+x_2+x_6+x_7),
 \label{umassless1}\\
 F(x) &=&  - t~ x_1 x_4 x_7 - u ~x_2 x_4 x_6 - s~ x_1 x_2 x_5 
  - s~ x_3 x_6 x_7  - s ~ x_3 x_5 (x_1+x_2+x_6+x_7),\nonumber \\&&
  \label{FTausk2}
\end{eqnarray}
where the longest factorized term has four Feynman parameters $x_1+x_2+x_6+x_7$.
The factorizations in $U$ and $F$ are connected with a proper collection of Feynman parameters in 
front of the variables $k_1,k_2$, as shown schematically in Fig.~\ref{schemeF}.
\begin{figure}[h!]
 \begin{center} 
\scriptsize
\begin{tikzpicture}%[transform canvas={scale=0.9}]
\matrix (m1) [matrix of math nodes] at (0,2)
{
\greentt{k_1^2} x_1 & \violett{k_2^2} x_2 & (\bluett{k_1+k_2})^2 x_3 &
(\bluett{k_1+k_2}+p_2)^2 x_4 & (\bluett{k_1+k_2}+p_1+p_2)^2 x_5 &
(\greentt{k_1}-p_3)^2 x_6 & (\violett{k_2}-p_4)^2 x_7 \\
};

\matrix (m2) [matrix of math nodes] at (0,0)
{
(\greentt{x_1} & \greentt{+} & \greentt{x_6}) & 
(\violett{x_2} & \violett{+} & \violett{x_7}) & + & 
(\bluett{x_3} & \bluett{+} & \bluett{x_4} & \bluett{+} & \bluett{x_5}) & 
(\redtt{x_1} & \redtt{+} & \redtt{x_2} & \redtt{+} & \redtt{x_6} & \redtt{+} & \redtt{x_7}) \\
};

\draw[->,spgreen] (m1-1-1.south) -- (m2-1-1.north);
\draw[->,spred] (m1-1-1.south) -- (m2-1-13.north);

\draw[->,spviolet] (m1-1-2.south) -- (m2-1-4.north);
\draw[->,spred] (m1-1-2.south) -- (m2-1-15.north);

\draw[->,spblue] (m1-1-3.south) -- (m2-1-8.north);
\draw[->,spblue] (m1-1-4.south) -- (m2-1-10.north);
\draw[->,spblue] (m1-1-5.south) -- (m2-1-12.north);

\draw[->,spgreen] (m1-1-6.south) -- (m2-1-3.north);
\draw[->,spred] (m1-1-6.south) -- (m2-1-17.north);

\draw[->,spviolet] (m1-1-7.south) -- (m2-1-6.north);
\draw[->,spred] (m1-1-7.south) -- (m2-1-19.north);

\end{tikzpicture}

\end{center}
 \caption{Efficient factorization scheme for the $U$ polynomial in eq.~(\ref{umassless1}). }
 \label{schemeF} 
\end{figure}
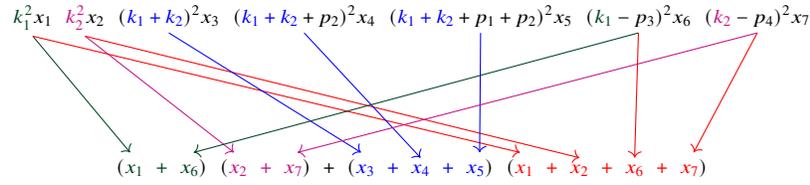

Now we can apply the CW theorem~\ref{CWtheorem0} with eq.~(\ref{CWtheorem}), keeping in the $\delta$-function
the longest factorized subset of Feynman parameters, so it can 
be dropped out from the Symanzik polynomials, and the integral becomes
\begin{multline}
 B_7^{NP} = 
 \frac{(-1)^{N_{\nu}} \Gamma\left(N_{\nu}-d\right)}{\Gamma(n_1) \ldots \Gamma(n_N)}
 \int \limits_0^{\infty} dx_3 dx_4 dx_5 \int \limits_0^{1} dx_3 dx_4 dx_5 \delta(1-(x_1+x_2+x_6+x_7)) \\
 \frac{((x_1 + x_6)(x_2 + x_7) + x_3 + x_4 + x_5)^{N_{\nu} - \frac{3d}{2}}}
 {(- t~x_1 x_4 x_7 - u~x_2 x_4 x_6 - s~x_1 x_2 x_5 - s~x_3 x_6 x_7  - s~x_3 x_5)^{N_{\nu} - d}}.
 \label{eq:B7NPa}
\end{multline}
In the next step we apply the \mb{} relation to $U$ and $F$ not expanding the term $(x_1+x_6)(x_2+x_7)$:
\begin{multline}
B_7^{NP} =  
\frac{(-1)^{N_{\nu}}}{\Gamma(n_1) \ldots \Gamma(n_N)} \int \limits_{-i \infty}^{i \infty} \frac{dz_1}{2 \pi i} \ldots \frac{dz_4}{2 \pi i} 
 \int dx_1 \ldots dx_7 ~
 (-s)^{ - N_{\nu} + d - z_2 - z_3} (-t)^{z_2} (-u)^{z_3} \\
 \shoveright{\times \Gamma (-z_1) \Gamma (-z_2) \Gamma (-z_3) \Gamma (-z_4) 
 \Gamma(N_{\nu} - d + z_1 + z_2 + z_3 + z_4)} \\
 \shoveleft{\times x_1^{ - N_{\nu} + d - z_1 - z_2 - z_3}
 x_2^{z_2 + z_3}
 \textcolor{black}{x_3}^{ - N_{\nu} + d - z_2 - z_3 - z_4}
 \textcolor{black}{x_4}^{z_1 + z_3}
 \textcolor{black}{x_5}^{z_2 + z_4}   
 x_6^{z_1 + z_2}    
 x_7^{z_3 + z_4} }\\
 \times (\textcolor{black}{ x_3 + x_4 + x_5} + (x_1 + x_6)(x_2 + x_7))^{N_{\nu} - \frac{3d}{2}}.\label{eq:B7NPb}
\end{multline}
To perform the integration over \texttt{CW} variables we can use iteratively the following integration formula
\begin{equation}
 \int \limits_0^{\infty} dx ~ x^{z_1}(x + y)^{z_2} = \frac{ y^{1 + z_1 + z_2} \Gamma(1 + z_1) \Gamma(-1 - z_1 - z_2)}{\Gamma(-z_2)}.
 \label{CWintRule}
\end{equation}
In the last step we apply the \texttt{MB} master relation of eq.~(\ref{MBformula}) to the terms $(x_1+x_6)$ and $(x_2+x_7)$. As we know
from the previous section, the corresponding $z$-variables can be removed using first Barnes lemma.
Finally, we end up with a 4-dimensional \mb{} representation,
\begin{multline}
 B_7^{NP} = \frac{(-1)^{N_{\nu}}}
 {\Gamma(n_1) \ldots \Gamma(n_7)} \int \limits_{-i \infty }^{i \infty} \frac{dz_1}{2 \pi i} \ldots \frac{dz_4}{2 \pi i} 
 (-s)^{4 - 2 \epsilon - N_{\nu} -z_{23}} (-t)^{z_3} (-u)^{z_2} \\
 \frac{\Gamma(- z_1) \Gamma(- z_2) \Gamma(- z_3) \Gamma(- z_4) \Gamma(2 - \epsilon - n_{45}) 
 \Gamma(2 - \epsilon - n_{67})}
 {\Gamma(4 - 2 \epsilon -  n_{4567}) \Gamma(n_{45} +  z_{1234})\Gamma(n_{67} +  z_{1234}) \Gamma(6 -3 \epsilon -  N_{\nu})} \\
 \Gamma(n_2 + z_{23}) \Gamma(n_4 + z_{24}) 
 \Gamma(n_5 + z_{13} ) \Gamma(n_6 + z_{34}) \Gamma(n_7 + z_{12})
 \Gamma^3(- 2 + \epsilon + n_{4567} + z_{1234})\\ 
 \Gamma(4 - 2 \epsilon - n_{124567} - z_{123}) 
 \Gamma(4 -2 \epsilon -  n_{234567} -  z_{234}) 
 \Gamma(- 4 + 2 \epsilon +  N_{\nu} +  z_{1234}),
\end{multline}
with notations $z_{i\ldots j \ldots k} = z_i + \ldots + z_j + \ldots + z_k$ 
and
$n_{i\ldots j \ldots k} = n_i + \ldots + n_j + \ldots + n_k$.

\subsection{General 2-loop Skeleton Diagrams \label{sec:2loopMBgeneral}}

As discussed in~\cite{Cvitanovic:1974uf}, 
{\it to get any two-loop diagram it is sufficient to attach external lines/legs
to the lines and/or vertices of the skeleton diagram given in} Fig.~\ref{fig:2lsceleton}. 
Lines $\alpha$, $\beta$ and $\gamma$ in the skeleton diagram  are also called the chains.

\begin{tips}{Non-planar Vertex from the Skeleton Diagram}
To get a non-planar two-loop 
vertex diagram in Fig.~\ref{fig:2lsceleton} (b) one has to attach one external leg to each line on the skeleton diagram (a) while
for the two-loop planar box we would have to attach external legs to the  line $\alpha$ and two external legs to line $\gamma$. For the non-planar double-box we would have to attach two legs to line $\alpha$ and one leg to each of lines $\beta$ and $\gamma$.
\end{tips}

\begin{figure}
\centering
\begin{tikzpicture}[scale=1]
\begin{feynman}
 
\vertex at (0,-2) (i1);
\vertex at (0,2,1) (i2);
\vertex at (1,1) (i3);
\vertex at (1,0) (i4);
\vertex at (1,-1) (i5);
\vertex at (2,.5) (i6);
\vertex at (3,0) (f1);
\vertex at (4,0) (f2);

\draw (1,0) circle [radius=1cm];
\draw (1,1) -- (1,-1); 
 
\node[left] at (0.,0) {{{\bf $\alpha$}}};
\node[left] at (1,0) {{{\bf $\beta$}}};
\node [left] at (2,0.) {{{\bf $\gamma$}}};

\draw [->,very thick] (2.5,0) -- (3.5,0);

\draw    (4.,1.) -- (6,0) ; 
\draw    (4.,-1.) -- (6,0) ;
\draw   (6,0) -- (6.5,0) ;
\draw    (4.5,.75) -- (5.,-.5) ;
\draw    (4.5,-.75) -- (5.,.5) ;
%\draw    (1,0) -- (1.8,-.5) ;
 
\draw (9,0) circle [radius=1cm];
\draw (9,1) -- (9,-1);
\draw (7.5,0) -- (8,0);
\draw (9,0) -- (9.5,0);
\draw (10,0) -- (10.5,0);
 
 \node[below] at (1.,-1.2) {{{ (a)}}}; 
  \node[below] at (7.,-1.2) {{{ (b)}}}; 
  
\end{feynman}
\end{tikzpicture}
    \caption{ \label{fig:2lsceleton}
             The  two-loop skeleton diagram (a) and creation of the two-loop non-planar vertex diagram (b). In (b) both diagrams are topologically equivalent.}
\end{figure}
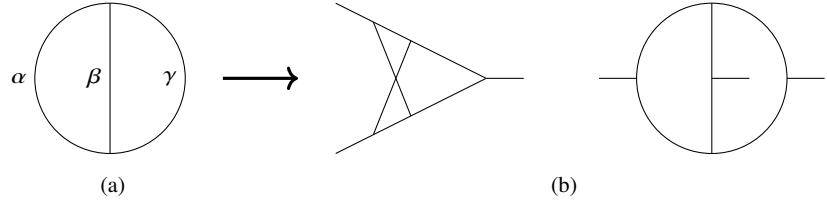

Let us consider the two-loop skeleton diagram in Fig.~\ref{fig:2lsceleton}.
A possible choice of  the loop momenta $k_1$ and $k_2$ in the chains 
contain the following combinations of loop momenta: $k_1$, $k_2$ and $k_1 + k_2$ or $k_1 - k_2$ (depending on the orientation of momentum flow in the chains), plus corresponding 
external momenta, depending on the topology.
Let now introduce the following transformation of the Feynman parameters belonging to each chain:
\begin{equation}
 \{ \vec x \}_i: \,\,\, x_k \rightarrow v_i \xi_{ik} \,\,\, \times \delta \left( 1 - \sum^{\eta_i}_{k=1} \xi_{ik} \right),
 \label{TransformRule}
\end{equation}
where $i$ denotes the chain index and $k \in [1, \eta_i]$, with $\eta_i$ -  the number of propagators in the chain.
The $\delta$-function keeps the number of variables unchanged.

Now the integration over Feynman parameters looks as follow:
\begin{multline}
 \int^1_0 \,\, \prod^{N}_{i=1} dx_i \,\, \delta \left( 1 - \sum^{N}_{i=1} x_i \right) \rightarrow \\
 \int^1_0 \,\, \prod^{3}_{i=1} dv_i v_i^{\eta_i - 1}  \,\, \delta \left( 1 - \sum^{3}_{i=1} v_i \right)
 \int^1_0 \,\, \prod^{3}_{i=1} \prod^{\eta_i}_{k=1} d\xi_{ik}  \,\, \prod^{3}_{i=1} \delta \left( 1 - \sum^{\eta_i}_{k=1} \xi_{ik} \right).
 \label{eq:vixi}
\end{multline}

This transformation dramatically simplifies graph polynomials.

\begin{tips}{Change of Variables for \texttt{CW} Theorem  \label{ex_masslessvertNP}}

Let us consider a two-loop non-planar completely massless vertex, whose $U$ and $F$ are represented in Fig.~\ref{fig:2looptrees}, see \wwwaux{V6l0m}  
for detailed derivations. We can see that applying the \MB{} formula to the constructed $U$ and $F$ terms in Fig.~\ref{fig:2looptrees} would result in multi-dimensional \mb{} representations, 11- and 5-dimensional, respectively. Let us see how the transformation of eq.~(\ref{TransformRule}) will change the situation.
The starting Feynman integral is
\begin{equation}
  \int  \frac{d^d k_1 d^d k_2}{
  [k_1^2]^{n_1}
  [(p_1 - k_1)^2]^{n_2}
  [(p_1 - k_1 - k_2)^2]^{n_3}
  [(p_2 + k_1 + k_2)^2]^{n_4}
  [(p_2 + k_2)^2]^{n_5}
  [k_2^2]^{n_6}}.
\label{npvertex}
\end{equation}
The momentum dependent function $m^2(\vec x)$, see eq.~(\ref{m2vecx}), and the parameters
transformations for the integral are the following
\begin{equation}
\begin{array}{cclcc}
 m^2(\vec x)        & = & x_1 (p_1 - k_1 - k_2)^2     & ~~~~~~~~~ & x_1 \rightarrow v_1 \xi_{11}  \\
                        & + & x_2 (p_2 + k_1 + k_2)^2 &           & x_2 \rightarrow v_1 \xi_{12}  \\
                        & + & x_3 (k_1)^2             &           & x_3 \rightarrow v_2 \xi_{21}  \\
                        & + & x_4 (p_1 - k_1)^2       &           & x_4 \rightarrow v_2 \xi_{22}  \\
                        & + & x_5 (p_2 + k_2)^2       &           & x_5 \rightarrow v_3 \xi_{31}  \\
                        & + & x_6 (k_2)^2             &           & x_6 \rightarrow v_3 \xi_{32}.  
\end{array} \label{eq:transfexnplvertex}
\end{equation}
\\
Utilizing the $\delta$-functions $\prod^{3}_{i=1} \delta \left( 1 - \sum^{\eta_i}_{k=1} \xi_{ik} \right)$ in eq.~(\ref{eq:vixi}),
the first Symanzik polynomial $U$ for any two-loop diagram becomes
\begin{equation}
 U(\vec x) \rightarrow U(\vec v) = v_1 v_2 + v_2 v_3 + v_1 v_3. 
\label{UCW2loop} 
\end{equation}
Due to dependence on kinematic variables, the second Symanzik polynomial $F$ has more complicated structure and depends on a definite topology. 
%, as we will discuss now.
In our case, after the transformation in eq.~(\ref{eq:transfexnplvertex}) we get
\begin{equation}
 F = - s \xi_{11} \xi_{22} \xi_{31} v_1 v_2 v_3 - s \xi_{12} \xi_{21} \xi_{32} v_1 v_2 v_3
     - s \xi_{31} \xi_{32} v_1 v_3^2 - s \xi_{31} \xi_{32} v_2 v_3^2.
\label{FNp2lV1}     
\end{equation}

{In summary, using eq.~(\ref{MBformula}) after rescaling in eq.~(\ref{UCW2loop}), 
for {\it any} two-loop diagram the $U$ polynomial can generate maximally a two-dimensional \MB{} integral.} Let us see now how applying the \texttt{CW} theorem~\ref{CWtheorem0} we can simplify further $U$ and $F$ structures, so dimensionality of \mb{} representations. 
\end{tips}

\begin{tips}{Applying the \texttt{CW} Theorem \label{tip:banana}}

Looking on the homogeneity of variables $v$ and $\xi$ in eq.~(\ref{FNp2lV1}) we can see that the CW Theorem~\ref{CWtheorem0} can be applied only to the $v$ variables. Thus, according to~\cite{Cvitanovic:1974uf},  we can generalize this fact  of the homogeneity of  the $v$ variables: {\it For any 2-loop diagram the graph polynomials
$U$ and $F$ can be represented formally in the form corresponding 
to the so-called sunrise diagram} in Fig.~\ref{fig:sunrise}, where $p^2$ and $m^2_i$
depend on $\xi_{ik}$ and the kinematic invariants (S) of the initial diagram:
\begin{multline}
 G(X) \sim \int \prod d \xi_{ik} \delta \left( 1 - \sum_k \xi_{ik} \right) \\
 \int 
 \frac{d^d k_1 d^d k_2 }{[k_1^2 - m_1^2(S,\xi_{ik})]^{n_1} [k_2^2 - m_2^2(S,\xi_{ik})]^{n_2} [(p + k_1 + k_2)^2 - m_3^2(S,\xi_{ik})]^{n_3}}. 
\label{CWsunrise} 
\end{multline}

\begin{figure}
\centering
\begin{tikzpicture}[scale=1.2]
\begin{feynman}
 
\vertex at (0,-2) (i1);
\vertex at (0,2,1) (i2);
\vertex at (1,1) (i3);
\vertex at (1,0) (i4);
\vertex at (1,-1) (i5);
\vertex at (2,.5) (i6);
\vertex at (3,0) (f1);
\vertex at (4,0) (f2);

\draw (1,0) circle [radius=1cm];
\draw (-.5,0) -- (2.5,0); 
 
\node[above] at (-.5,0) {{{\bf $p$}}};
\node[above] at (1,0) {{{\bf $v_2,m_2$}}};
\node [below] at (1,1.) {{{\bf $v_1,m_1$}}};
\node [above] at (1,-1.) {{{\bf $v_3,m_3$}}}; 
 
\end{feynman}
\end{tikzpicture}
    \caption{ \label{fig:sunrise}
        The general sunrise diagram.  
      }
\end{figure}
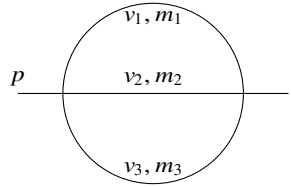
 
In case of the considered massless non-planar vertex in eq.~(\ref{npvertex}) we have specifically
\begin{eqnarray}
 p^2  &=& - s ( \xi_{12} \xi_{22} \xi_{31} + \xi_{12} \xi_{21} \xi_{32} - \xi_{31} \xi_{32}  ), \\
 m_1^2 &=& - s \xi_{31} \xi_{32}, \,\,\,\, m_2^2 = m_3^2 = 0.
\end{eqnarray}

In order to use the CW Theorem~\ref{CWtheorem0} there are only three possibilities to choose a subset in $\delta \left( 1 - \sum^{3}_{i=1} v_i \right)$ during the transformation of eq.~(\ref{TransformRule}).
This choice depends on the structure of $F$, from 
$v_1 + v_2$, $v_2 + v_3$ or $v_1 + v_3$ one should take the combination which makes F as simple as possible.
In case of eq.~(\ref{FNp2lV1}) it is $v_1 + v_2$, so we have
\begin{eqnarray}
 U &=& v_3 + v_1 v_2, \label{eq:afterCW2loop} \\ F &=& - s \xi_{11} \xi_{22} \xi_{31} v_1 v_2 v_3 - s \xi_{12} \xi_{21} \xi_{32} v_1 v_2 v_3     -  s \xi_{31} \xi_{32} v_1 v_3^2. \label{eq:bananaF}
\end{eqnarray}
Applying the \texttt{MB}-master formula to $F$ and integrating over $v_3$ we get a two-fold \mb-integral 
 \begin{multline}
  I^{\rm MB}_{\rm V6l0m} = (-1)^{N_{\nu}} (-s)^{4 - 2 \epsilon - N_{\nu}} \int\limits_{-i\infty}^{i\infty} \int\limits_{-i\infty}^{i\infty}  \frac{dz_1}{2 \pi i} \frac{dz_2}{2 \pi i}   \\
  \times \frac{\Gamma(2 - \epsilon - n_{12}) \Gamma(2 - \epsilon - n_{56}) 
        \Gamma(4 - 2 \epsilon - n_{12356} - z_1) \Gamma(-z_1) \Gamma(n_1 + z_1)}
   {\Gamma(n_1) \Gamma(n_2) \Gamma(n_3) \Gamma(n_4) \Gamma(n_5) \Gamma(4 - 2 \epsilon - n_{1256}) \Gamma(n_6)
     \Gamma(n_5 + z_1) } \\
  \times  \frac{\Gamma(-z_2) \Gamma(n_2 + z_2) \Gamma(n_6 + z_2) \Gamma(-4 + 2 \epsilon + N_{\nu} + z_{12})}
        {\Gamma(8 - 4 \epsilon - 2 n_{1256} - n_{34} - z_{12}) \Gamma(n_{12} + z_{12}) \Gamma(n_{56} + z_{12})
         \Gamma(4 - 2 \epsilon - n_{12456} - z_2)}. \label{eq:V6l0m2dim}
 \end{multline}
 
\end{tips}

\begin{svgraybox}
In case of massive diagrams, namely when some of the propagators have masses $m^2_i$, the best
option is to not expand in the first place the term $U\sum x_i m^2_i$, instead to apply eq.~(\ref{MBformula}) to $F$ {\it considering $U$
as an independent variable}. 
\end{svgraybox}

Adding a mass parameter to one of the propagators in the non-planar vertex results only in one additional \mb{} integration.
However, this is not a general way, for example in case of the non-planar QED double box in Fig.~\ref{TauskNPdbox} with $p_i^2 = m^2$ and masses $m^2$ in propagators with powers 
$n_1, n_3, n_4$ and $n_7$ in eq.~(\ref{B7NP}), the above procedure leads to the appearance of a `pseudo' Minkowskian factor $(-m^2)^{z_i}$ in the \mb-representation, 
similarly to the case discussed already in section~\ref{sec:LAapproach}, see eq.~(\ref{eq:MBnplanar_bad}), which 
restricts this approach.

On the other hand, using \la{}, one gets an 8-dimensional representation~\cite{Heinrich:2004iq} 
which points to the applicability of \la{} also to some non-planar diagrams. 

\subsection{Generalization to 3-loop Integrals}

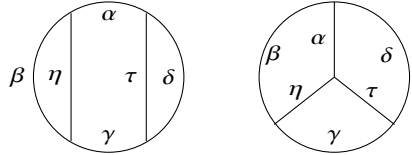
\begin{figure}
\centering
\begin{tikzpicture}[scale=1]
\begin{feynman}

\draw (1,0) circle [radius=1cm];
\draw (.5,.85) -- (.5,-.85); 
\draw (1.5,.85) -- (1.5,-.85); 

\node[below] at (1.,1) {{{\bf $\alpha$}}};
\node[left] at (0,0) {{{\bf $\beta$}}};
\node[left] at (0.5,0) {{{\bf $\eta$}}};
\node[left] at (1.5,0) {{{\bf $\tau$}}};
\node[left] at (2,0) {{{\bf $\delta$}}};
\node [above] at (1,-1.) {{{\bf $\gamma$}}};
  
\draw (4,0) circle [radius=1cm];
\draw (4,0) -- (4.8,-.63);
\draw (4,0) -- (3.2,-.63);
%\draw (4,0) -- (5.,1);
\draw (4,0) -- (4,1);
  
\node[left] at (4.,.5) {{{\bf $\alpha$}}};
\node[left] at (3.4,0.3) {{{\bf $\beta$}}};
\node[left] at (3.7,-.2) {{{\bf $\eta$}}};
\node[left] at (4.7,-.2) {{{\bf $\tau$}}};
\node[left] at (4.9,0.3) {{{\bf $\delta$}}};
\node [above] at (4,-1.) {{{\bf $\gamma$}}};

\end{feynman}
\end{tikzpicture}
   \caption{Basic skeleton generating diagrams with propagators $\alpha , \ldots,  \tau$ for the 3-loop topologies discussed in~\cite{Cvitanovic:1974uf}, so-called {\texttt{Ladder}} (on left) and {\texttt{Mercedes}} (on right) topologies. \label{fig:skel}}
\end{figure} 

At the three-loop level there are two generating topologies indicated in  Fig.~\ref{fig:skel} with different properties~\cite{Cvitanovic:1974uf}. The diagram on the left side of Fig.~\ref{fig:skel}
and all its derivatives can be cut to two disconnected one-loop pieces. Propagators corresponding to the 3-loop diagrams can have two combinations of two loop momenta (if three different loop momenta in one propagator are present, one loop momenta can be always eliminated). %Because chains (a) and (b) carry the same internal momenta they can be considered as one chain,
In the 3-loop case the number of $v$-variables in the general transformation rule in eq.~(\ref{TransformRule}) is five. 
Typical examples of this type of diagrams for the box topologies are depicted in Fig.~\ref{fig:skel1ex}. 
\begin{figure}[!h]
\begin{center}
\includegraphics[scale=0.5]{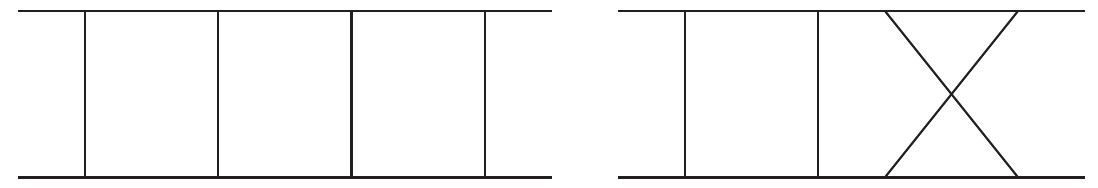}
\end{center}
\caption{The 3-loop box topologies obtained from the {\texttt{Ladder}} skeleton diagram in Fig.~\ref{fig:skel}. Diagrams originated from this skeleton  have the same $U$ polynomial structure. \label{fig:skel1ex} }
\end{figure} 

For the left skeleton diagram in Fig.~\ref{fig:skel}, after a transformation of variables,
all diagrams up to replacing of indices have a $U$ polynomial in the following form (Problem~\ref{prob:3loopskeletonA}).
\begin{equation}
 U = v_1 v_2 v_3+v_1 v_2 v_4+v_2 v_3 v_4+v_1 v_2 v_5+v_1 v_3 v_5+v_2 v_3 v_5+v_1 v_4 v_5+v_3 v_4 v_5. 
 \label{skel1exU3loop}
\end{equation}
Many diagrams of this type have planar sub-loops and can be treated in a different way, see the next section.

The diagram on the right side in Fig.~\ref{fig:skel} can not be divided into two one-loop pieces. Its propagators have all three different combinations of two loop momenta and number of $v$-variables is six. This skeleton generates the most complicated non-planar topologies, see for example Fig.~\ref{fig:skel2ex}
and Problem~\ref{prob:vertexSkeleton}. 

\begin{figure}[!h]
\begin{center}
\includegraphics[scale=0.4]{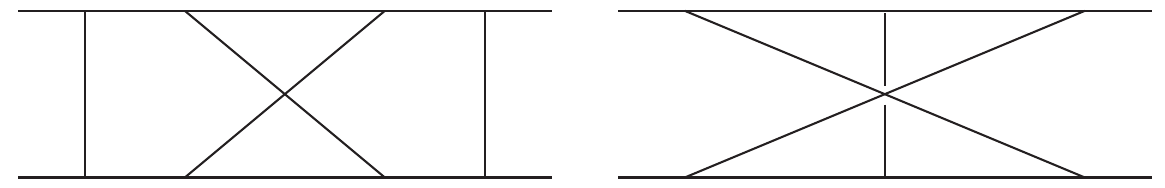}
\end{center}
\caption{The 3-loop box topologies obtained from the right skeleton diagram in Fig.~\ref{fig:skel} and have another $U$ polynomial structure. \label{fig:skel2ex}}
\end{figure}

All child diagrams have now again the same $U$ function in the form
\begin{eqnarray}
U &=& v_1 v_2 v_3+v_1 v_2 v_4+v_1 v_3 v_4+v_1 v_2 v_5+v_1 v_3 v_5+v_2 v_3 v_5+v_2 v_4 v_5+v_3 v_4 v_5 \nonumber \\
    &+& v_1 v_2 v_6 +v_2 v_3 v_6+v_1 v_4 v_6+v_2 v_4 v_6+v_3 v_4 v_6+v_1 v_5 v_6+v_3 v_5 v_6+v_4 v_5 v_6.\nonumber \\ 
\label{skel2exU}
\end{eqnarray}

As in the two-loop case, further simplification can be done with the help of the CW Theorem~\ref{CWtheorem0}.
The shortest form of $U$ can be obtained if we choose three CW variables. An analysis of the form of $U$ shows that that there are four different
possibilities to choose CW variables. For eq.~(\ref{skel2exU}) one of them is
\begin{equation}
 \int^{\infty}_0 d v_2 d v_3 d v_4 \int_0^1 d v_1 d v_5 d v_6 \delta (1 - v_1 - v_5 - v_6 ),
\end{equation}
which gives
\begin{equation}
 U_{CW} = v_2 v_3 + v_2 v_4 + v_3 v_4 + v_1 v_2 v_5 + v_1 v_3 v_5 + 
           v_1 v_2 v_6 + v_1 v_4 v_6 + v_1 v_5 v_6 + v_3 v_5 v_6 + v_4 v_5 v_6
           \label{eq:ucw_3l}
\end{equation} 
and reduces the number of terms in $U$ from 16 to 10.

One of possible ways to make the integration over suitable CW variables is to use the following
factorization trick:
\begin{equation}
 U_{CW} = v_2 (v_3 + v_4 + v_1 v_5) + v_3 (v_4 + v_1 v_5) + 
           v_1 v_6 (v_2 + v_5) + v_4 v_6 (v_1 + v_5) + v_3 v_5 v_6.
\label{skel2exUfactor}           
\end{equation}
There are six different possibilities to get four terms in $U$ and 
altogether we have 24 variants to choose \texttt{CW} variables and factorize $U$.
A final choice is based, first, on a minimization of the amount of terms in $F$ and, second,
on the presence in $F$ of the same factorization patterns as in eq.~(\ref{skel2exUfactor}) for $U$.

Now, to construct the \mb{} representation as in the two-loop case we do not have to necessarily expand mass terms in $F$
(see the previous section). Similarly, it is not necessary to expand any factorized combination
which corresponds to the pattern in eq.~(\ref{skel2exUfactor}). After applying  eq.~(\ref{MBformula}) to $U$ and $F$
we can integrate recursively the polynomials over $v_3$ and then over $v_4$ using eq.~(\ref{CWintRule}). The integration over $v_2$
can be also done using eq.~(\ref{CWintRule}). Finally, we apply again eq.~(\ref{MBformula}) to the term $v_1+v_5$. In the end this \mb{} integration
can be removed using \texttt{1BL}.
As one can see within this algorithm the $U$ polynomial
gives already 4 \mb{} integrations.
$F$ is usually more complex and the final \mb{} representation can be very high dimensional.
This is a natural limitation of the method.

{\it As a rule of thumb, the \ga{} usually gives optimal representations if from the beginning $\mbox{Length[U]} \leq \mbox{Length[F]}$.}

\begin{tips}{\ga{} Approach for a Planar Diagram: An Example}

This example shows an application of \ga{} and \texttt{1BL}  to
a planar 3-loop self-energy massless diagram, resulting in the same \mb{} integral dimension as by \la{}.
The integral is 
\begin{equation}
\label{eq:SE6l0m}
    I_{\rm SE6l0m} = \int \frac{ dk_1 dk_2 dk_3}{[k_1^2]^{n_1} [k_2^2]^{n_2}  [k_3^2]^{n_3} [(k_1+k_3)^2]^{n_4} [(k_2+k_3)^2]^{n_5}   [(k_2-k_1+q)^2]^{n_6}.} 
%PR[k1 + k3, 0, n4] PR[k2 + k3, 0, n5] PR[k2 - k1 + q, 0, n6]}}
\end{equation}
The diagram is given in  Fig.~(\ref{fig:SE6l0m}) where

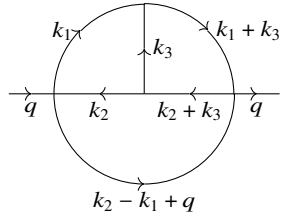
\begin{figure}
\centering
\begin{tikzpicture}[scale=1.2]
\begin{feynman}

%\draw (1,0) circle [radius=1cm];
\draw [->] (0,0) arc (180:135:1cm);
\draw (.3,.7) arc (135:90:1cm);
\draw [->] (1,1) arc (90:45:1cm);
\draw (1.7,.7) arc (45:0:1cm);
\draw [->] (0,0) arc (180:270:1cm);
\draw (1,-1) arc (270:360:1cm);

\draw [->] (-.5,0) -- (-.25,0);
\draw  (-.25,0) -- (0,0);
\draw [->] (2,0) -- (2.25,0);
\draw  (2.25,0) -- (2.5,0);

\draw [->] (1,0) -- (1,0.5); 
\draw (1,0.5) -- (1,1); 
\draw [->] (2,0) -- (1.5,0); 
\draw (1.5,0.) -- (1,0); 
\draw [->] (1,0) -- (0.5,0); 
\draw (0.5,0) -- (0,0); 

\node[below] at (1.5,0) {{{\bf $k_2+k_3$}}};
\node[below] at (0.5,0) {{{\bf $k_2$}}};
\node [below] at (-.25,0.) {{{\bf $q$}}};
\node [below] at (2.25,0.) {{{\bf $q$}}};
\node [below] at (1,-1.) {{{\bf $k_2-k_1+q$}}}; 
\node [above,left] at (.3,0.7) {{{\bf $k_1$}}};
\node [above,right] at (1.7,.7) {{{\bf $k_1+k_3$}}}; 
\node [right] at (1,.5) {{{\bf $k_3$}}}; 
\end{feynman}
\end{tikzpicture}
  \caption{The 3-loop planar self-energy diagram  $I_{\rm SE6l0m}$ defined in eq.~(\ref{eq:SE6l0m}).     }
\label{fig:SE6l0m}
\end{figure}

The \la{} method gives here 4-dimensional \mb{} representation, see section E4 in~\cite{Blondel:2018mad} and a Problem~\ref{prob:3loopmerc}. 
This diagram corresponds to the right skeleton diagram in Fig.~\ref{fig:skel} with two external legs attached to vertices on the circle, so variable transformation of the type in eq.~(\ref{TransformRule}) is not needed. 
The $U$ polynomial has the same structure as in eq.~(\ref{skel2exUfactor}), corresponding $F$ polynomial can be simplified to the following
form:
\begin{equation}
 F = - s v_3 v_4 v_6 (v_1 + v_5) - s v_2 v_3 (v_4 + v_1 v_5) - s v_1 v_3 v_6 (v_2 + v_5).
\end{equation}
This means three additional \mb{} integrations plus 4 integrations which come from the $U$ polynomial.
The resulting 7-dimensional output from \texttt{AMBREv3.2} is given in \wwwaux{SE6l0m}. 

The integration over \texttt{z7} corresponds to the term $v_1+v_5$ and can be immediately removed by applying \texttt{1BL} and making the shift \texttt{z6 -> z6 - z3 - z2}, one can apply \texttt{1BL} to the variable \texttt{z2} and \texttt{2BL} to the variable \texttt{z3}, 
obtaining a 4-dimensional representation as in case of the \la{}. The  shift in momenta was found using the algorithm described in section~\ref{sec:BLeff}.

To complete this section and the discussion of the \ga{} method for 3-loop cases we present a strategy of choosing CW variables and a factorization
for the first type of 3-loop diagrams in Fig.~(\ref{fig:skel1ex}). Because here we have only five $v$-variables, the optimal choice is to consider two CW variables.
For eq.~(\ref{skel1exU}) this leads to
\begin{equation}
 \int^{\infty}_0 d v_1 d v_2 \int_0^1 d v_3 d v_4 d v_5 \delta (1 - v_3 - v_4 - v_5 )
\end{equation}
and
\begin{equation}
 U_{CW} = v_1 v_2 + v_2 v_3 v_4 + v_1 v_3 v_5 + v_2 v_3 v_5 + v_1 v_4 v_5 + 
           v_3 v_4 v_5. 
\end{equation}
his operation reduces the number of terms in U from 8 to 6. As in the previous case, there are also 4 possibilities to choose CW variables.
A factorization scheme to integrate over CW variables can be the following:
\begin{equation}
 U_{CW} = v_2(v_1 + v_3 v_4) + v_3 v_5 (v_2 + v_4) + v_1 v_5 (v_3 + v_4).
 \label{eq:Ucw_2dim}
\end{equation}
Due to the smaller number of variables and terms, there are now 
only 4 possibilities to factorize $U$.
All other steps go the same way, as for eq.~(\ref{skel2exUfactor}).
After applying eq.~(\ref{MBformula}) to $U$ and $F$
we integrate over $v_1$ and $v_2$ using eq.~(\ref{CWintRule}). The combination $v_3+v_4$ again can be removed using \texttt{1BL}.
In this case, $U$ polynomial results in two additional \mb{} integrations.

\end{tips}

\section{\label{sec:hybrid} Hybrid (\ha{}) Approach }
 
  At the three-loop we can  approach propagators and sub-loops combining the \la{} and \ga{} approaches. This depends on whether 
or not a given topology includes a planar sub-topology which can be disconnected or not.  In the first case we start with the \la{} and integrate over one of the loop momenta. After that the obtained effective two-loop diagram is treated by the \ga{}
In Fig.~\ref{fig:hyb1to2}, propagators connected with a planar sub-loop in the form 
of the one-loop box are  transformed into an \mb{} form in a first step, defining a new effective propagator. 
In this way an effective two-loop diagram is created. 
%\vskip 20mm
\begin{figure}[!h]
\begin{center}
\includegraphics[scale=0.5]{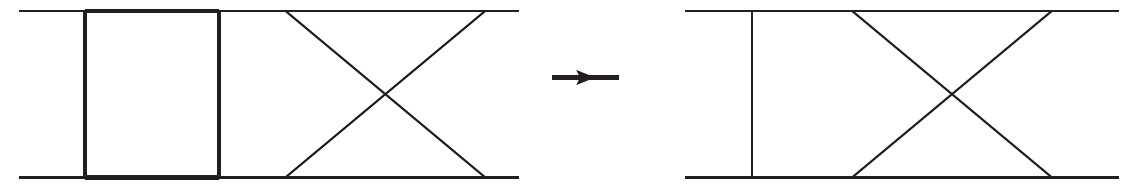}
\end{center}
\caption{Hybrid approach. \la{} is applied first to the 1-loop subloop, \h{} ($1 \rightarrow 2$).}
\label{fig:hyb1to2}
\end{figure}
 
Another example of this kind is a non-planar vertex shown in Fig.~\ref{fig:3loopNP1t02_1}. An exemplary  file is in \wwwaux{V8l0m}.
As was mentioned before the combination of three different loop momenta in one propagator usually does not take place, they appear here due to
the procedure of generation of diagrams using some automatic codes like \texttt{FeynArts} or \texttt{qgraph}. A simpler form can be obtained by a simple shift, $k_3 \rightarrow k_3 + k_2$.
In this particular case that plays no role because in the first iteration the momentum $k_1$ circulates only in the one-loop sub-diagram shown in bold, and
for  \ga{}, in the second iteration the momentum flow is not important. 
Nonetheless, 
%At this point one should stress that 
the momentum flow is important in general  
%plays a significant role for the hybrid approach and  \la{} 
-- each loop momentum should go through a minimal possible topological construction. For example, due to some specific construction procedure, $k_1$ could appear in more than three propagators. In this case a preliminary shift of loop momenta is needed for a successful/efficient application of \ha{}.

%\begin{figure}[!t]
%\begin{center}
%\includegraphics[scale=.4]{figs/E19.pdf}
%\end{center}
%\caption{Hybrid method, an example.
%The bold lines show the choice of the one-loop sub-diagram in the %first step of \h{} ($1 \rightarrow 2$).
%\label{fig:3loopNP1t02_1}}
%\end{figure}  

\begin{figure}
\centering
\begin{tikzpicture}[scale=1.2]
\begin{feynman}

\draw (-.5,0.5) -- (0,0);
\draw  (0,0) -- (2,0);
\draw (2,0) -- (2.5,0.5);

\draw  (0,0) -- (1,-2);
\draw  (1,-2) -- (2.,0);
\draw  (1,-2) -- (1,-2.5);
\draw [very thick] (.41,-.8) -- (1.63,-.8);
\draw [very thick] (.41,-.8) -- (1.,-2);
\draw [very thick] (1.63,-.8) -- (1.,-2);

\draw (1,0.) -- (1,-2);  

\node [below] at (-.5,.5) {{{\bf $q_1$}}};
\node [below] at (2.8,.5) {{{\bf $q_1+q_2$}}};
\node [right] at (1,-2.25) {{{\bf $q_2$}}};
\node [above,small] at (.5,0.)  {$k_2+q_1$};
\node [above] at (1.5,0.) {{{\bf $k_3+q_2$}}};
\node[below] at (0,-.3) {{{\bf $k_2$}}};
\node[below] at (0.5,-1.2) {{{\bf $k_1$}}};
\node [right] at (1.75,-.4) {{{\bf $k_3-q_1$}}};
\node at (1.02,-.2) {{{\bf $k_3-k_2-q_1+q_2$}}};
\node at (1.,-.9) {{{\bf $k_2-k_1$}}};
\node[below] at (2.3,-1.2) {{{\bf $k_1-k_2+k_3-q_1$}}};

\end{feynman}
\end{tikzpicture}
\caption{Hybrid method, an example.
The bold lines show the choice of the one-loop sub-diagram in the first step of \h{} ($1 \rightarrow 2$).
\label{fig:3loopNP1t02_1}}
\end{figure}
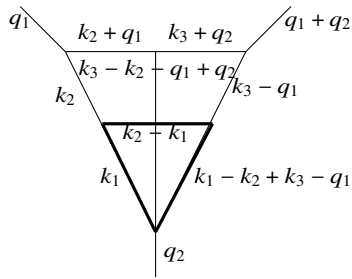

%\end{tips}

In Fig.~\ref{fig:hyb2to1}, a non-planar 
disconnected sub-graph can be identified and  we go in the opposite direction, first, we apply the \ga{} to a two-loop sub-diagram, basically a non-planar one,
and the remaining one-loop integral is processed in a simple way by the \la.

\begin{figure}[!h]
\begin{center}
\includegraphics[scale=0.5]{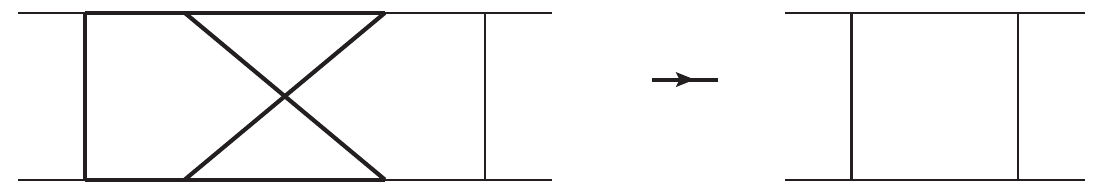}
\end{center}
\caption{Hybrid approach. The \ga{} is applied first to the 2-loop sub-diagram,  \h{} ($2 \rightarrow 1$).}
\label{fig:hyb2to1}
\end{figure}

Yet another example for \ha{} ($2 \to 1$)  is shown in Fig.~\ref{fig:3loopNP2to1_1}. In the first iteration the two-loop non-planar vertex effectively has two off-shell external legs, so after transformation of eq.~(\ref{TransformRule}), in contrast to the non-planar vertex from example in section~\ref{sec:2loopMBgeneral}, one gets a 6-dimensional representation. The remaining three propagators form one-loop vertex and give no additional
\mb{} integration.

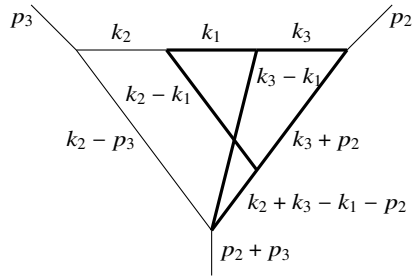
\begin{figure}[h!]
\centering
\begin{tikzpicture}[scale=1.2]
\begin{feynman}

\draw (-1,0.5) -- (-0.5,0); %ex1 
\draw  (-0.5,0) -- (0.5,0);
\draw [very thick] (0.5,0) -- (1.5,0);
\draw [very thick] (1.5,0) -- (2.5,0);
\draw (2.5,0) -- (3,0.5); %ex2
\draw  (1,-2) -- (1,-2.5); %ex3
\draw (-0.5,0) -- (1.,-2);
\draw [very thick] (2.5,0) -- (1.,-2);
\draw [very thick] (1.5,0) -- (1.,-2);
\draw [very thick] (1.5,-1.33) -- (0.5,0);

\node [below] at (-1.1,.5) {{{\bf $p_3$}}};
\node [below] at (3.1,.5) {{{\bf $p_2$}}};
\node [right] at (1,-2.25) {{{\bf $p_2+p_3$}}};
\node [above] at (0,0) {{{\bf $k_2$}}};
\node [above] at (1,0) {{{\bf $k_1$}}};
\node [above] at (2,0) {{{\bf $k_3$}}};
\node [left] at (0.25,-1) {{{\bf $k_2-p_3$}}};
\node [right] at (1.3,-1.7) {{{\bf $k_2+k_3-k_1-p_2$}}};
\node [right] at (1.8,-1) {{{\bf $k_3+p_2$}}};
\node [right] at (1.4,-0.3) {{{\bf $k_3-k_1$}}};
\node [right] at (-0.05,-0.5) {{{\bf $k_2-k_1$}}};

\end{feynman}
\end{tikzpicture}

\caption{ \label{fig:3loopNP2to1_1}
Hybrid approach.
The bold lines show the choice of the two-loop sub-diagram in the first step of the \h{} ($2 \rightarrow 1$) procedure.}
\end{figure}   
 
\section{Minimal Dimensions of Symanzik Polynomials: Summary}
 
For the 2-loop skeleton diagram in Fig.~\ref{fig:2lsceleton}, rescaling of $U$ gives a 2-dim \mb{} integral, see eq.~(\ref{UCW2loop})
\begin{equation*}
U_{\mbox{\tiny 2-loop}} = v_1 v_2 + v_2 v_3 + v_1 v_3. 
\end{equation*}
For the 3-loop \texttt{Ladder} skeleton diagram in Fig.~\ref{fig:skel}, rescaling of $U$ gives a 7-dim \mb{} integral, see eq.~(\ref{skel1exU3loop})
\begin{equation*}
U_{\mbox{\tiny 3-loop(I)}} = v_1 v_2 v_3+v_1 v_2 v_4+v_2 v_3 v_4+v_1 v_2 v_5+v_1 v_3 v_5+v_2 v_3 v_5+v_1 v_4 v_5+v_3 v_4 v_5. 
 \label{skel1exU}
\end{equation*}
For the 3-loop \texttt{Mercedes} skeleton diagram in Fig.~\ref{fig:skel}, rescaling of $U$ gives a 15-dim \mb{} integral, see eq.~(\ref{skel2exU})
\begin{multline*}
U_{\mbox{\tiny 3-loop(II)}} = v_1 v_2 v_3+v_1 v_2 v_4+v_1 v_3 v_4+v_1 v_2 v_5+v_1 v_3 v_5+v_2 v_3 v_5+v_2 v_4 v_5+v_3 v_4 v_5 \\
    +  v_1 v_2 v_6 +v_2 v_3 v_6+v_1 v_4 v_6+v_2 v_4 v_6+v_3 v_4 v_6+v_1 v_5 v_6+v_3 v_5 v_6+v_4 v_5 v_6. 
\label{skel2exU}
\end{multline*}
Finally, using CW Theorem~\ref{CWtheorem0}, we get further simplifications:

\begin{svgraybox}
\begin{itemize}
    \item For any 2-loop diagram: $\delta(1 - v_1 - v_2)$, $U(\vec v) = v_3 + v_1 v_2$. \\
    {\color{black}{There are no additional \mb{} integrations}} generated $U$, as in the 1-loop case, see a discussion around eq.~(\ref{eq:afterCW2loop}).
    \item For any 3-loop diagram: $\delta(1 - v_1 - v_2 - v_3)$
    \begin{itemize}
        \item \texttt{Ladder} - $U$ generates {\color{black}{2 additional \mb{} integrations}}, see a discussion around eq.~(\ref{eq:Ucw_2dim})
        \item \texttt{Mercedes} - $U$ generates {\color{black}{4 additional \MB{} integrations}}, see a discussion around eq.~(\ref{skel2exUfactor}).
    \end{itemize}
\end{itemize}
\end{svgraybox}

For the $F$ polynomial situation is more complicated and additional dimensions of  \mb{} integrals depends on kinematics and internal masses. As discussed, the expression $F = F_0 + U \sum\limits x_i m_i^2$ must be treated individually for each \texttt{FI}, without or with suitable expansions of $F_0$ and $U$ terms.
%%%%
% MBcreate
%%%%
In fact, given a generic Feynman integral, we are not aware of general results on the dimensionality of its optimal \mb{} representation (i.e., the \mb{} representation with a minimal number $n_0$ of complex integrations). One can of course proceed by trial and error for specific cases, however this can beomce very tedious. One way to address this issue is to build an alogrithm which attempts to automate the minimalization of the number of \mb{} integrations. This is precisely the approach of the \texttt{MBcreate.m} package of~\cite{Belitsky:2022gba}, which implements strategies to automatically deduce a concise \mb{} representatons with as few integrations as possible.

\section{Beyond 3-loops}

At 4-loops planar diagrams can still be treated with \la{}, the complexity of $U$ and $F$ leaves not much space for constructing good \mb{} representations using the \ga{}. Interestingly, the hybrid method opens a possibility to build \texttt{MB} representations beyond three loops. For one of the possible implementations, see~\cite{Boels:2015yna}. However, these issues go beyond the content of this textbook.

\section{\mb{} and the Method of Brackets \label{sec:prausa} }

In~\cite{Prausa:2017frh}  a very interesting approach to the construction of \mbr{} representations  was introduced.  The approach is based on the `Method of Brackets'  defined in~\cite{Gonzalez:2007ry,Gonzalez:2008xm,Gonzalez:2010uz,Gonzalez:2021vqh}, which is a multi-fold construction technique sums starting directly from a Schwinger-parameterized Feynman integral presented in chapter~\ref{chapter-MBrepr}.
The technique of brackets transforms the parameter integral into a series-like
structure called: `brackets expansion' with four basic rules. 
\begin{enumerate}
    \item {\underline{Exponential function expansion}}
     
\begin{equation}
\exp \left( -xA\right) =\sum\limits_{n}\frac{\left( -1\right) ^{n}}{\Gamma
\left( n+1\right) }x^{n}A^{n}.
\end{equation}%
If the argument of exponential function is $\exp (xA)$

\begin{equation}
\exp \left( xA\right) =\sum\limits_{n}\frac{\left( -1\right) ^{n}}{\Gamma
\left( n+1\right) }x^{n}\left( -A\right) ^{n}.
\end{equation}%
The reason for this is to associate to each expansion the factor $\phi _{n}=%
\frac{\left( -1\right) ^{n}}{\Gamma \left( n+1\right) }$ as a simple
convention.
\item {\underline Integration symbol and its equivalent bracket}

This rule corresponds to the definition of the bracket symbol. The structure
$\int x^{a_{1}+a_{2}+...+a_{n}-1}$ $dx$ is replaced by its respective
bracket representation

\begin{equation}
\int x^{a_{1}+a_{2}+...+a_{n}-1}dx=\left\langle
a_{1}+a_{2}+...+a_{n}\right\rangle .
\end{equation}

\item \underline{Polynomials expansion}

For polynomials the following representation in terms of series of
brackets is used

\begin{equation}
\begin{array}{l}
\left( A_{1}+...+A_{r}\right) ^{\pm \mu }= \\
\\
\sum\limits_{n_{1}}...\sum\limits_{n_{r}}\phi _{n_{1}}...\phi
_{n_{r}}\;\left( A_{1}\right) ^{n_{1}}...\left( A_{r}\right) ^{n_{r}} 
\times \frac{\left\langle \mp \mu +n_{1}+...+n_{r}\right\rangle }{\Gamma
\left( \mp \mu \right) }.%
\end{array}%
\end{equation}%
This rule is derived using rule 1. and 2. 
after applying the Schwinger's parametrization to this polynomial.  

\item \underline{Finding the solution}

For the case of a generic series of brackets $J$

\begin{equation}
\begin{array}{ll}
J= & \sum\limits_{n_{1}}...\sum\limits_{n_{r}}\phi _{n_{1}}...\phi
_{n_{r}}\;\digamma (n_{1},...,n_{r}) \\
&  \\
& \times \;\left\langle a_{11}n_{1}+...+a_{1r}n_{r}+c_{1}\right\rangle \ldots \left\langle a_{r1}n_{1}+...+a_{rr}n_{r}+c_{r}\right\rangle ,%
\end{array}%
\end{equation}%
the solution is obtained using the general formula

\begin{equation}
%\begin{array}{c}
J=\frac{1}{\left\vert \det \left( \mathbf{A}\right) \right\vert }\Gamma
\left( -n_{1}^{\ast }\right) ...\Gamma \left( -n_{r}^{\ast }\right) \digamma
(n_{1}^{\ast },...,n_{r}^{\ast })%
%\end{array}
\label{Brackets}
\end{equation}%
where $\det \left( \mathbf{A}\right) $ is evaluated by the following
expression

\begin{equation}
\det \left( \mathbf{A}\right) =\left\vert
\begin{array}{ccc}
a_{11} & \ldots & a_{1r} \\
\vdots & \ddots & \vdots \\
a_{r1} & \cdots & a_{rr}%
\end{array}%
\right\vert ,
\end{equation}%
and $\left\{ n_{i}^{\ast }\right\} $ $\;\left( i=1,...,r\right) $ is the
solution of the linear system obtained by the vanishing of the brackets

\begin{equation}
\left\{
\begin{array}{cc}
a_{11}n_{1}+...+a_{1r}n_{r}= & -c_{1} \\
\vdots & \vdots \\
a_{r1}n_{1}+...+a_{r}n_{r}= & -c_{r}.%
\end{array}%
\right.
\end{equation}%
The value of $J$ is not defined if the matrix $\mathbf{A}$ is not invertible.
\end{enumerate}
 
The relation in eq.~(\ref{Brackets}) is the Generalized Ramanujan's Master Theorem (\texttt{RMTG}) derived in~\cite{Gonzalez:2008xm}. For simple application of the above formalism to a construction of \mb{} integrals, based on \texttt{RMTG}, see Problem~\ref{problem_prausa1}
and original discussion in~\cite{Gonzalez:2008xm}, Problem~\ref{problem_prausa2}.

Here we compare some basic results from this method of brackets  with results by \ar{} packages, in which the above discussed 3-loop non-planar cases were optimized.

  \input{mprausasubfigs.tex}
\begin{table}[h!]
      \centering
      \begin{tabular}{lcccc} \hline
        diagram & Method of Brackets & \texttt{AMBRE} & planarity & \texttt{AMBRE} $^*$ \\ \hline
        (a)  & \textbf{7} & 13  & NP & \textbf{4}, $\h{} (2 \to 1)$\\
        (b)  & \textbf{1} & 2 & P & \textbf{1} \\
        (c)  & \textbf{7} & 9  & NP & \textbf{5}, $\h{} (1 \to 2)$\\
        (d)  & \textbf{7} & 8  & NP & 8, $\h{} (1 \to 2)$ \\
        (e)  & 5 & \textbf{3}  & P & \textbf{3} \\
        \hline
      \end{tabular}        
\caption{The number of MB integrations of the representation fconstructed by the Method of Brackets and \texttt{AMBRE} for cases (a)-(e) given in Fig.~\ref{fig:prausa}. P (NP) stands for planarity (non-planarity) of the diagram. The last column gives improvements discussed in the text when using the \ha{} method.\label{tab:comp}}
\end{table}

By \texttt{AMBRE}$^*$ in Tab.~\ref{tab:comp} we denote \texttt{AMBREv1} for planar cases and manual combination of \texttt{AMBREv1} and \texttt{AMBREv3} for the hybrid method.  
The \texttt{AMBRE}$^*$ results from Tab.~\ref{tab:comp} are given in \wwwaux{AMBREnew}.

\begin{comment}
\begin{table}[h!]
      \centering
      \begin{tabular}{lcccc} \hline
        diagram & Method of Brackets & \texttt{AMBRE} & planarity & AMBRE 4$^*$/method \\ \hline
        f(a)  & \textbf{7} & 13  & NP & \textbf{4} (2 \to 1)\\
        fig.5.1(b)  & \textbf{1} & 2 & P & \textbf{1} \\
        fig.5.1(c)  & \textbf{7} & 9  & NP & \textbf{5} \\
        fig.5.1(d)  & \textbf{7} & 8  & NP & 8 \\
        fig.5.1(e)  & 5 & \textbf{3}  & P & \textbf{3} \\
        fig.5.1(f)  & 9 & \textbf{4}  & P & \textbf{4} \\
        fig.5.1(g)  & 7 & \textbf{4} & P & \textbf{4} \\
        fig.5.1(h)  & 5 & \textbf{4} & P & \textbf{4}\\
        fig.5.1(i)  & \textbf{2} & \textbf{2}  & P & \textbf{2} \\
        fig.5.1(j) & \textbf{2} & \textbf{2}  & P & \textbf{2} \\
        \hline
      \end{tabular}        
\caption{The number of MB integrations of the representation constructed by the Method of Brackets 
compared to the previous \ar{} versions~\cite{Gluza:2007rt,Gluza:2010rn,Blumlein:2014maa,Dubovyk:2016ocz} 
and the last version AMBRE 4.  P (NP) stands for planarity (non-planarity) of the diagram. \label{tab:comp}}
\end{table}
\end{comment}

\section{{Phase Space \MB{} Integrals} 
}
\label{sec:4MBrepr}

When computing predictions for physical observables  
such as cross sections or decay rates in perturbative 
quantum field theory, in addition to the Feynman integrals 
discussed above, we also we also encounter integrations over 
the momenta of outgoing external particles in a scattering 
or decay process. We will refer to such integrals as phase 
space integrals, since the associated integration measure 
is commonly referred to as the phase space measure. In 
particular, the measure for $n$ outgoing particles of total 
momentum $Q$ in $d$ space-time dimensions reads
\begin{equation}
    d\phi_n(p_1,\ldots,p_n;Q) = 
    \prod_{k=1}^{n}
    \frac{d^d p_k}{(2\pi)^{d-1}}\delta_+(p_k^2-m_k^2)
    (2\pi)^d \delta^{(d)}(p_1+\ldots+p_n-Q),
\end{equation}
where $\delta^{(d)}$ is the $d$-dimensional Dirac delta 
function while
\begin{equation}
\delta_+(p_k^2-m_k^2) = \delta(p_k^2-m_k^2)\Theta(p_k^0)
\end{equation}
and the various factors of $2\pi$ are conventional. We 
see that in contrast to loop integrals, here the integrations 
over the momenta are constrained by both the mass-shell 
conditions (i.e., $p_k^2=m_k^2$) as well as overall momentum 
conservation. Thus, it is less straightforward to develop 
general representations of phase space integrals than for 
loop integrals. Thus, in practical applications one often 
simply chooses some particular explicit parametrization of 
the components of the outgoing momenta, resolves the 
constraints, and attempts to evaluate the resulting parametric 
integral. We note in passing that the resolution of the 
constraints can lead to very elaborate expressions and so 
it can be crucial to correctly tailor the chosen 
parametrization to the problem at hand.

\begin{tips}{A Three-Particle Phase Space Integral}
As an example, consider the three-particle phase space where all three outgoing momenta are massless, $p_1^2=p_2^2=p_3^2$. A typical situation is that one must compute the integral of some function of the invariants $s_{ij} = 2p_i\cdot p_j$, $i,j=1,2,3$ over the full phase space. E.g. consider the simple example
\begin{equation}
\begin{split}
    I &= \int d\phi_3(p_1,p_2,p_3;Q) \frac{1}{s_{13} s_{23}}
\\ &=
    \int \frac{d^d p_1}{(2\pi)^{d-1}} \delta_{+}(p_1^2)
    \frac{d^d p_2}{(2\pi)^{d-1}} \delta_{+}(p_2^2)
    \frac{d^d p_3}{(2\pi)^{d-1}} \delta_{+}(p_3^2)
\\ &\times
    (2\pi)^d \delta^{(d)}(Q-p_1-p_2-p_3)
    \frac{1}{s_{13} s_{23}}.
\end{split}
\end{equation}
In order to compute this integral, let us first use the Dirac delta function expressing overall momentum conservation to integrate out e.g. the momentum $p_3$,
\begin{equation}
\begin{split}
    I &= (2\pi)^{3-2d} \int d^d p_1\, \delta_{+}(p_1^2)
    d^d p_2\, \delta_{+}(p_2) \delta_{+}(p_3^2)
    \delta_{+}[(Q-p_1-p_2)^2]
    \frac{1}{s_{13} s_{23}}.
\end{split}
\label{eq:PS3-ex-I}
\end{equation}
Now, since both the integration measure and the integrand are rotationally invariant, we can choose to work in the following convenient Lorentz-frame:
\begin{equation}
    \begin{split}
        Q^\mu &= \sqrt{s}(1,\vec{0}_{d-1}),
        \\
        p_1^\mu &= E_1(1,\vec{0}_{d-2}, 1),
        \\
        p_2^\mu &= E_2(1,\vec{0}_{d-3},
        \sin\theta,
        \cos\theta),
        \\
        p_3^\mu &= Q-p_1-p_2.  
    \end{split}
    \label{eq:L-frame-3}
\end{equation}
Above $\vec{0}_j$ denotes a vector of $j$ zeros. It is not difficult to show that in this frame we may write the various pieces of the integration measure as follows,
\begin{equation}
\begin{split}
    d^d p_1 \delta_{+}(p_1^2) &= 
    \frac{E_1^{d-3}}{2} dE_1\,d\Omega_{d-1}\Theta(E_1),
\\
    d^d p_2 \delta_{+}(p_2^2) &= 
    \frac{E_2^{d-3}}{2} dE_2\,d(\cos\theta)\,(\sin\theta)^{d-4} 
    d\Omega_{d-2}\Theta(E_2),
\end{split}
\end{equation}
where $d\Omega_j$ is the measure for the integration over the $j$ angular variables on which the integrand does not depend (i.e., the angles of rotations that are necessary to bring a general configuration to the form of eq.~(\ref{eq:L-frame-3})). On the other hand, the invariants in this frame take the form
\begin{equation}
    \begin{split}
        s_{12} &= 2E_1 E_2(1-\cos\theta),
\\
        s_{13} &= 2E_1\sqrt{s} -  2E_1 E_2(1-\cos\theta),
\\
        s_{23} &= 2E_2\sqrt{s} - 2E_1 E_2(1-\cos\theta).
    \end{split}
\label{eq:PS3-ex-inv}
\end{equation}
Then, our integral in eq.~(\ref{eq:PS3-ex-I}) takes the form
\begin{equation}
    \begin{split}
        I &= (2\pi)^{3-2d} 2^{-2} \Omega_{d-1} \Omega_{d-2}
        \int d E_1\, d E_2\,(E_1 E_2)^{d-3} \int_{-1}^{1} d(\cos\theta)
        (\sin\theta)^{d-4}
\\ &\times   
        \delta[s - 2 E_1 \sqrt{s} - 2 E_2 \sqrt{s} 
        + 2E_1 E_2 (1-\cos\theta)]\Theta(\sqrt{s}-E_1-E_2)
        \Theta(E_1)\Theta(E_2)
\\ &\times   
    \frac{1}{[2E_1\sqrt{s} -  2E_1 E_2(1-\cos\theta)]
    [2E_2\sqrt{s} -  2E_1 E_2(1-\cos\theta)]}.
    \end{split}
\label{eq:PSint-ex-int}
\end{equation}
Here we have used that $\delta_+(p^2) = \delta(p^2)\Theta(p^0)$ and we have been careful to indicate all constraints on the energy integrals. It is possible to evaluate this integral by solving the delta function for one of the energies, which leads to an integration over the angle $\theta$ and the other energy, $E_1$. Being careful to determine the correct limits of integration on the remaining variables, we find that the result can be expressed in terms of just gamma functions, see problem~\ref{problem_PSint}. 
Below we will have much more to say about how the \MB{} method can be used to compute the angular parts of such phase space integrals.

However, in order to highlight the simplifications that can result from a well-chosen parametrization, let us not follow this path, but instead choose a new set of integration variables, namely the invariants $s_{12}$, $s_{13}$ and $s_{23}$. Obviously we can use eq.~(\ref{eq:PS3-ex-inv}) to solve for $E_1$, $E_2$ and $\cos\theta$ in terms of these invariants. Performing the change of variables, we find simply that
\begin{equation}
    \begin{split}
        I &= (2\pi)^{3-2d} 2^{-1-d} s^{\frac{2-d}{2}}
        \Omega_{d-1} \Omega_{d-2} \int ds_{12}\,ds_{13}\,ds_{23}\,
        (s_{12} s_{13} s_{23})^{\frac{d-4}{2}} 
\\ &\times   
        \delta(s-s_{12}-s_{13}-s_{23})
        \frac{1}{s_{13} s_{23}},
    \end{split}
\end{equation}
where all integrations run between 0 and $s$. This last integral is very easy to evaluate by resolving the delta function for $s_{12}$ and we find simply
\begin{equation}
        I = 
        \frac{2^{-8+6\epsilon} \pi^{-\frac{5}{2}+2\epsilon} 
        \Gamma^2(-\epsilon)}{\Gamma(1-3\epsilon)\Gamma\left(\frac{3}{2}-\epsilon\right)}s^{-1-2\epsilon}.
\label{eq:PSint-ex-res}
\end{equation}
Here we have set $d=4-2\epsilon$ and used that the the total angular volume in $p$ dimensions is $\Omega_p = \frac{2\pi^{p/2}}{\Gamma(p/2)}$ (see eq.~(\ref{eq:Omp}) below).

As a final word of caution, we mention that beyond three-particle integrals, the parametrization in terms of invariants is often times not as useful as our example might lead one to believe. The reason behind this is that although the measure of integration can be written in a very nice and compact form also for phase spaces of higher multiplicity (see e.g.~\cite{Gehrmann-DeRidder:2003pne,Heinrich:2006sw}), the limits of integration are actually quite non-trivial to express in a way that is convenient for actual calculations.
\end{tips}

In practical applications the parametrizations many times 
involve angles between three-momenta, defined in some suitable 
Lorentz-frame, and so the evaluation of phase space integrals 
will involve the computation of integrals over these angles~\cite{Anastasiou:2013srw}. 
Such angular integrals are particularly well suited to 
treatment by the Mellin-Barnes method, as we will describe 
below.

To start, let us define what we will call an angular integral 
with $n$ denominators. Following the conventions of~\cite{Somogyi:2011ir}, we set:
\begin{equation}
    \Omega_{j_1,\ldots,j_n} \equiv
    \int d\Omega_{d-1}(q)
    \frac{1}{(p_1\cdot q)^{j_1}\ldots (p_n\cdot q)^{j_n}}.
    \label{eq:Oj1jn-def}
\end{equation}
Here $p_1^\mu,\ldots,p_n^\mu$ are fixed $d$-dimensional 
vectors, and we are integrating over angular variables 
of the massless $d$-dimensional vector $q^\mu$. Thus 
$d\Omega_{d-1}(q)$ is the rotationally invariant angular 
measure in $d$ dimensions, whose explicit expression will 
be given below in eq.~(\ref{eq:Omp}). It turns out that the 
general integral in eq.~(\ref{eq:Oj1jn-def}) admits a 
very nice representation in terms of Mellin-Barnes 
integrals which can be a suitable starting point for 
both analytic and numeric evaluation.

To derive this representation, we begin by noting that 
the overall normalization of the $p_i^\mu$ and $q^\mu$ 
clearly does not play an essential role and so without 
loss of generality, we can simply choose to normalize 
these vectors in whatever way is most convenient. Thus, 
we choose a Lorentz-frame where
\begin{equation}
    \begin{split}
        p_1^\mu &= (1,\vec{0}_{d-2},\beta_1),
        \\
        p_2^\mu &= (1,\vec{0}_{d-3},
        \beta_2\sin\chi_2^{(1)},
        \beta_2\cos\chi_2^{(1)}),
        \\
        p_3^\mu &= (1,\vec{0}_{d-4},
        \beta_3\sin\chi_3^{(2)}\sin\chi_3^{(1)},
        \beta_3\cos\chi_3^{(2)}\sin\chi_3^{(1)},
        \beta_3\cos\chi_3^{(1)}),
        \\ &\vdots \\
        p_n^\mu &= (1,\vec{0}_{d-1-n},
        \beta_n \prod_{k=1}^{n-1}\sin\chi_n^{(k)}, \\&
\hspace{.5cm}        \beta_n \cos\chi_n^{(n-1)} \prod_{k=1}^{n-2}
        \sin\chi_n^{(k)},\ldots,
        \beta_n \cos\chi_n^{(2)}\sin\chi_n^{(1)},
        \beta_n \cos\chi_n^{(1)}).
    \end{split}
    \label{eq:pmu-def}  
\end{equation}
%\commgs{TODO FIGURE}
Again, $\vec{0}_j$ denotes a vector of $j$ zeros, and thus 
in words, we have chosen a frame where the direction of 
$p_1^\mu$ fixes the $d$-th axis, then $p_2^\mu$ fixes the 
plane of the $d$-th and $(d-1)$-st axis and so on. Notice 
that we have written all vectors in $d$-dimensional polar 
coordinates and used the freedom to fix the normalizations 
to fix each zeroth component to be one, $p_j^0=1$. In this 
frame, $q^\mu$ can be written in the following form, again 
using $d$-dimensional polar coordinates,
\begin{equation}
    \begin{split}
        q^\mu = (1,..\mbox{``angles''}..,
        \cos\vartheta_n\prod_{k=1}^{n-1}\sin\vartheta_k,
        \cos\vartheta_{n-1}\prod_{k=1}^{n-2}
        \sin\vartheta_k,\ldots,
        \cos\vartheta_2\sin\vartheta_1,
        \cos\vartheta_1).
    \end{split}
    \label{eq:qmu-def}
\end{equation}
Here $..\mbox{``angles''}..$ denotes those angular variables 
on which the integral does not depend and that can thus be 
integrated trivially. Also, as before, we have used the 
freedom to fix the normalization to set $q^0=1$. It can 
be shown that the angular measure appearing in eq.~(\ref{eq:Oj1jn-def}) above can be written as follows
\begin{equation}
    d\Omega_{d-1}(q) = \prod_{k=1}^{n}d(\cos\vartheta_k)
    (\sin\vartheta_k)^{-k+1-2\epsilon}
    d\Omega_{d-1-n}(q)
    \label{eq:dOmq-def}
\end{equation}
where we have now set $d=4-2\epsilon$. Notice that 
$d\Omega_{d-1-n}(q)$ corresponds to the angular measure 
of the variables denoted simply as $..\mbox{``angles''}..$ 
in eq.~(\ref{eq:qmu-def}) above. 
For later use, we note that the total angular volume in 
$p$ dimensions is simply (Problem~\ref{prob:dangles})
\begin{equation}
    \Omega_p = \int d\Omega_p(q) = 
    \frac{2\pi^{\frac{p}{2}}}{\Gamma\left(\frac{p}{2}\right)}.
    \label{eq:Omp}
\end{equation}
For a sanity check, notice that the angular volume 
$\Omega_p$ is just the integral of the angular part 
of the total volume measure in polar coordinates in 
$p$ dimensions, i.e., the surface of the $p$-dimensional 
unit sphere. Hence in $p=2$ we expect $\Omega_2 = 2\pi$, 
which is clearly in agreement with eq.~(\ref{eq:Omp}). 
For $p=3$ eq.~(\ref{eq:Omp}) gives 
$\Omega_3 = 2\pi^{\frac{3}{2}}/\Gamma(\frac{3}{2})$, but 
$\Gamma(\frac32) = \sqrt{\pi}/2$, so in fact we obtain 
the expected result of $\Omega_3 = 4\pi$.

Then, inserting 
eq.~(\ref{eq:dOmq-def}) into eq.~(\ref{eq:Oj1jn-def}), we 
find the explicit integral representation
\begin{equation}
    \begin{split}
    &
        \Omega_{j_1,\ldots,j_n} = 
        \int d\Omega_{d-1-n}(q)
        \int_{-1}^{1} \prod_{k=1}^{n}
        \left[d(\cos\vartheta_k)
        (\sin\vartheta_k)^{-k+1-2\epsilon}\right]
    \\&\quad\times
        \prod_{k=1}^{n}\left\{
        1 - \beta_k\sum_{l=1}^k\left[
        \left(\delta_{lk}+(1-\delta_{lk})
        \cos\chi_k^{(l)}\right)\cos\vartheta_l
        \prod_{m=1}^{l-1}\left(
        \sin\chi_k^{(m)}\sin\vartheta_m
        \right)
        \right]
        \right\}^{-j_k}.
    \end{split}
    \label{eq:Oj1jn-1}
\end{equation}
As usual, $\delta_{kl}$ is equal to one if $k$ and $l$ 
coincide, otherwise it is zero.
{We will spell out this formula for small values of $n$ (in particular $n=1$ and $n=2$) explicitly when we discuss examples later, see eqs.~(\ref{eq:Ojm0}), (\ref{eq:Ojm}), (\ref{eq:Ojkm0}) and (\ref{eq:Ojkm}).} 

In general, the parametric integral in eq.~(\ref{eq:Oj1jn-1}) 
is quite elaborate. E.g. for $j_k=1$ and $\beta_k=1$, 
already for $n=2$ the integrand has a line singularity inside 
the (two real dimensional) integration domain. Thus even 
in this simple case, already the resolution of singularities 
as $\epsilon\to 0$ is non-trivial. However, the Mellin-Barnes 
method provides a particularly nice way to approach the 
evaluation of eq.~(\ref{eq:Oj1jn-1}). We proceed as follows. 
First, let us go back to the original definition, 
eq.~(\ref{eq:Oj1jn-def}) and use Feynman parametrization 
to combine all denominators:
\begin{equation}
    \begin{split}
        \Omega_{j_1,\ldots,j_n} &= 
        \int d\Omega_{d-1}(q) 
        \frac{\Gamma(j)}{\prod_{k=1}^{n}\Gamma(j_k)}
    \\&\times
        \int_0^1 \left[\prod_{k=1}^{n} 
        d x_k (x_k)^{j_k-1}\right]
        \delta\left(\sum_{k=1}^{n}x_k - 1\right)
        \left[
        \left(\sum_{k=1}^{n} x_k p_k\right)\cdot q
        \right]^{-j},
        \end{split}
    \label{eq:Oj1jn-2}
\end{equation}
where $j=j_1+\ldots+j_n$ is the sum of the exponents. 
Now we make an important observation: we can exploit the 
rotational invariance of the original expression and 
evaluate the integral in a frame where the direction of the 
weighted sum of momenta $x_1 p_1^\mu + \ldots + x_n p_n^\mu$ 
points along the $d$-th direction. In this frame, we have 
simply
\begin{equation}
    \sum_{k=1}^{n} x_k p_k^\mu = (1,\vec{0}_{d-2},\beta)
    \qquad\mbox{and}\qquad
    q^\mu = (1,..\mbox{``angles''}..,
        \sin\vartheta,
        \cos\vartheta).
\end{equation}
A quick computation shows that the variable $\beta$ 
introduced here can be expressed as the solution of the 
following equation
\begin{equation}
    1 - \beta^2 = \left(\sum_{k=1}^{n} x_k p_k^\mu\right)^2 
    = \sum_{k=1}^{n} \sum_{l=k+1}^{n} 2x_k x_l (p_k\cdot p_l)
    + \sum_{k=1}^{n} x_k^2 p_k^2
    \label{eq:beta}
\end{equation}

Thus, eq.~(\ref{eq:Oj1jn-2}) takes the form
\begin{equation}
    \begin{split}
        \Omega_{j_1,\ldots,j_n} &= 
        \frac{\Gamma(j)}{\prod_{k=1}^{n}\Gamma(j_k)}
        \int_0^1 \left[\prod_{k=1}^{n} 
        d x_k (x_k)^{j_k-1}\right]
        \delta\left(\sum_{k=1}^{n}x_k - 1\right)
    \\&\times
        \int d\Omega_{d-2}(q)
        \int_{-1}^{1}d(\cos\vartheta)
        (\sin\vartheta)^{-2\epsilon}
        \left[
        1 - \beta\cos\vartheta
        \right]^{-j}.
    \end{split}
    \label{eq:Oj1jn-3}
\end{equation}
The only non-trivial angular integration above can be 
performed in terms of a Gauss ${}_2F_1$ hypergeometric 
function via the substitution $\cos\vartheta \to 2s-1$, 
while $\Omega_{d-2}$ is given in eq.~(\ref{eq:Omp}):
\begin{equation}
    \begin{split}
        &\int d\Omega_{d-2}(q)
        \int_{-1}^{1}d(\cos\vartheta)
        (\sin\vartheta)^{-2\epsilon}
        \left[
        1 - \beta\cos\vartheta
        \right]^{-j}
    \\&=
    2^{2-2\epsilon}\pi^{1-\epsilon}(1+\beta)^{-j} 
    \frac{\Gamma(1-\epsilon)}{\Gamma(2-2\epsilon)}
    {}_2F_1\left(j,1-\epsilon,2-2\epsilon,
    \frac{2\beta}{1+\beta}\right).
    \end{split}
\end{equation}
This result can be put in a more convenient form by 
using the quadratic hypergeometric identity
\begin{equation}
    {}_2F_1(a,b,2b,z) = \left(1-\frac{z}{2}\right)^{-a}
    {}_2F_1\left[\frac{a}{2},\frac{a+1}{2},b+\frac{1}{2},
    \left(\frac{z}{2-z}\right)^2\right]
\end{equation}
which leads to
\begin{equation}
    \begin{split}
        &\int d\Omega_{d-2}(q)
        \int_{-1}^{1}d(\cos\vartheta)
        (\sin\vartheta)^{-2\epsilon}
        \left[
        1 - \beta\cos\vartheta
        \right]^{-j}
    \\&=
    2^{2-2\epsilon}\pi^{1-\epsilon}
    \frac{\Gamma(1-\epsilon)}{\Gamma(2-2\epsilon)}
    {}_2F_1\left(\frac{j}{2},\frac{j+1}{2},
    \frac{3}{2}-\epsilon,
    \beta^2\right).
    \end{split}
    \label{eq:Om-j-m}
\end{equation}
The utility of this later form lies in the fact that it 
is $\beta^2$ and not $\beta$ itself that has a simple 
expression in terms of the dot-products of the $p_i^\mu$ 
as evidenced by eq.~(\ref{eq:beta}). In fact, it turns 
out that it is convenient to write eq.~(\ref{eq:beta}) in 
a more compact form by introducing the variables $v_{kl}$ 
such that
\begin{equation}
    v_{kl} \equiv 
    \begin{cases}
    \frac{p_k\cdot p_l}{2}, & k\ne l
    \\
    \frac{p_k^2}{4}, & k=l
    \end{cases}.
    \label{eq:vkl-def}
\end{equation}
Then eq.~(\ref{eq:beta}) takes the simple form
\begin{equation}
    1-\beta^2 = 4\sum_{k=1}^{n}\sum_{l=k}^{n} x_k x_l v_{kl}.
\label{eq:1-beta2}
\end{equation}
The reason for this particular choice of normalization 
will become clear shortly.

To continue, we substitute eq.~(\ref{eq:Om-j-m}) into 
eq.~(\ref{eq:Oj1jn-3}) and obtain
\begin{equation}
    \begin{split}
        \Omega_{j_1,\ldots,j_n} &= 
        \frac{\Gamma(j)}{\prod_{k=1}^{n}\Gamma(j_k)}
        \int_0^1 \left[\prod_{k=1}^{n} 
        d x_k (x_k)^{j_k-1}\right]
        \delta\left(\sum_{k=1}^{n}x_k - 1\right)
    \\&\times
        2^{2-2\epsilon}\pi^{1-\epsilon}
        \frac{\Gamma(1-\epsilon)}{\Gamma(2-2\epsilon)}
        {}_2F_1\left(\frac{j}{2},\frac{j+1}{2},
        \frac{3}{2}-\epsilon,
        1-4\sum_{k=1}^{n}\sum_{l=k}^{n} x_k x_l v_{kl}\right).
    \end{split}
    \label{eq:Oj1jn-4}
\end{equation}
At this point, we have traded all angular integrations for 
integrals over Feynman parameters, $x_j$. However, the 
dependence of the integrand on the $x_j$ is quite complicated. 
This is where the Mellin-Barnes representation comes to 
our help. To further manipulate eq.~(\ref{eq:Oj1jn-4}), we 
proceed as follows. First, we will represent the 
hypergeometric function as a one-dimensional Mellin-Barnes 
integral. Then, we will transform the double sum in the 
argument of the hypergeometric function to a product of 
factors by using the basic Mellin-Barnes formula. This 
will then allow us to perform all integrations over the 
Feyman parameters, leading to our final result: a pure 
Mellin-Barnes representation of the general $n$ denominator 
angular integral.

To implement the steps outlined above, first use the 
following well-known Mellin-Barnes representations of 
the ${}_2F_1$ hypergeometric function,
\begin{equation}
\begin{split}
    {}_2F_1(a,b,c,x) &= \frac{\Gamma(c)}{\Gamma(a)\Gamma(b)
    \Gamma(c-a)\Gamma(c-b)}
    \\ &\times
    \int_{-i\infty}^{+i\infty} \frac{dz_0}{2\pi i}
    \Gamma(a+z_0)\Gamma(b+z_0)\Gamma(c-a-b-z_0)\Gamma(-z_0)
    (1-x)^{z_0},
\end{split}
\label{eq:2F1-MB}
\end{equation}
to write
\begin{equation}
\begin{split}
    &
    {}_2F_1\left(\frac{j}{2},\frac{j+1}{2},
        \frac{3}{2}-\epsilon,
        1-4\sum_{k=1}^{n}\sum_{l=k}^{n} x_k x_l v_{kl}\right)
\\&\quad=
    \frac{\Gamma\left(\frac{3}{2}-\epsilon\right)}
    {\Gamma\left(\frac{j}{2}\right)
    \Gamma\left(\frac{j+1}{2}\right)
    \Gamma\left(\frac{3-j}{2}-\epsilon\right)
    \Gamma\left(\frac{2-j}{2}-\epsilon\right)}
    \int_{-i\infty}^{+i\infty} \frac{dz_0}{2\pi i}
\\&\qquad\times
    \Gamma\left(\frac{j}{2}+z_0\right)
    \Gamma\left(\frac{j+1}{2}+z_0\right)
    \Gamma(1-j-\epsilon-z_0)\Gamma(-z_0)
    \left(4\sum_{k=1}^{n}\sum_{l=k}^{n} 
    x_k x_l v_{kl}\right)^{z_0}
\\ &\quad=
    2^{-j} 
    \frac{\Gamma(2 - 2 \epsilon)}
    {\Gamma(1 - \epsilon) \Gamma(2 - j - 2 \epsilon) 
    \Gamma(j)}
    \int_{-i\infty}^{+i\infty} \frac{dz_0}{2\pi i}
\\&\qquad \times 
    \Gamma(j + 2 z_0)
    \Gamma(1-j-\epsilon-z_0)\Gamma(-z_0) 
    \left(\sum_{k=1}^{n}\sum_{l=k}^{n} 
    x_k x_l v_{kl}\right)^{z_0}.
\end{split}
\label{eq:ang-2F1-MB}
\end{equation}
To write the last equality, we used the doubling relation for the 
gamma function,
\begin{equation}
    \Gamma(2x) = \frac{2^{2x-1}}{\sqrt{\pi}}
    \Gamma(x)\Gamma\left(x+\frac{1}{2}\right),
\label{eq:doubling-rel}
\end{equation}
to simplify the arguments of the gamma functions. Notice 
in particular, that the factor of $4^{z_0}$ in the second 
line is cancelled, which explains our normalizaion of the 
variables in eq.~(\ref{eq:vkl-def}). Then, substituting 
eq.~(\ref{eq:ang-2F1-MB}) into eq.~(\ref{eq:Oj1jn-4}), 
we find
\begin{equation}
    \begin{split}
        \Omega_{j_1,\ldots,j_n} &= 
        \frac{2^{2-j-2\epsilon}\pi^{1-\epsilon}}
        {\prod_{k=1}^{n}\Gamma(j_k)
        \Gamma(2 - j - 2 \epsilon)}
        \int_0^1 \left[\prod_{k=1}^{n} 
        d x_k (x_k)^{j_k-1}\right]
        \delta\left(\sum_{k=1}^{n}x_k - 1\right)
    \\&\times
        \int_{-i\infty}^{+i\infty} \frac{dz_0}{2\pi i}
        \Gamma(j + 2 z_0)
        \Gamma(1-j-\epsilon-z_0)\Gamma(-z_0)
        \left(\sum_{k=1}^{n}\sum_{l=k}^{n} 
        x_k x_l v_{kl}\right)^{z_0}.
    \end{split}
    \label{eq:Oj1jn-5}
\end{equation}
Next, we write the double sum to the $z_0$-th power in a 
factorized form. Using the basic Mellin-Barnes formula, 
it is not difficult to show that (see Problem~\ref{problem_OMB})
\begin{equation}
    \begin{split}
        \left(\sum_{k=1}^{n}\sum_{l=k}^{n} 
        x_k x_l v_{kl}\right)^{z_0} &=
        \frac{1}{\Gamma(-z_0)}\int_{-i\infty}^{+\infty}
        \left[\prod_{k=1}^{n-1}\prod_{l=k}^{n}
        \frac{dz_{kl}}{2\pi i} \Gamma(-z_{kl})
        (x_k x_l v_{kl})^{z_{kl}}\right]
        \\ &\times
        \Gamma\left(-z_0 + \sum_{k=1}^{n-1}\sum_{l=k}^{n}
        z_{kl}\right) (x_n^2 v_{nn})^{z_0 - 
        \sum_{k=1}^{n-1}\sum_{l=k}^{n}z_{kl}}
    \end{split}
    \label{eq:vklsum-MB}
\end{equation}
%\deleted{\commgs{EXERCISE: prove this equation!}}
Substituting eq.~(\ref{eq:vklsum-MB}) into 
eq.~(\ref{eq:Oj1jn-5}), we obtain
\begin{equation}
    \begin{split}
        \Omega_{j_1,\ldots,j_n} &= 
        \frac{2^{2-j-2\epsilon}\pi^{1-\epsilon}}
        {\prod_{k=1}^{n}\Gamma(j_k)
        \Gamma(2 - j - 2 \epsilon)}
        \int_0^1 \left[\prod_{k=1}^{n} 
        d x_k (x_k)^{j_k-1}\right]
        \delta\left(\sum_{k=1}^{n}x_k - 1\right)
    \\&\times
        \int_{-i\infty}^{+i\infty} \frac{dz_0}{2\pi i}
        \Gamma(j + 2 z_0)
        \Gamma(1-j-\epsilon-z_0)
    \\&\times
        \int_{-i\infty}^{+\infty}
        \left[\prod_{k=1}^{n-1}\prod_{l=k}^{n}
        \frac{dz_{kl}}{2\pi i} \Gamma(-z_{kl})
        (x_k x_l v_{kl})^{z_{kl}}\right]
        \\ &\times
        \Gamma\left(-z_0 + \sum_{k=1}^{n-1}\sum_{l=k}^{n}
        z_{kl}\right) (x_n^2 v_{nn})^{z_0 - 
        \sum_{k=1}^{n-1}\sum_{l=k}^{n}z_{kl}}.
    \end{split}
    \label{eq:Oj1jn-6}
\end{equation}
Now let us change the variable of integration from $z_0$ 
to $z_{nn} \equiv z_0 - \sum_{k=1}^{n-1}\sum_{l=k}^{n}z_{kl}$. 
Then $z_0 = \sum_{k=1}^{n-1}\sum_{l=k}^{n}z_{kl} + z_{nn} = 
\sum_{k=1}^{n}\sum_{l=k}^{n}z_{kl}$, i.e., $z_0$ is simply 
the sum of all $\frac{n(n+1)}{2}$ integration variables. To 
avoid any confusion, let us denote this sum simply as $z$ 
in the following,
\begin{equation}
    z=\sum_{k=1}^{n}\sum_{l=k}^{n}z_{kl}\,. 
\label{eq:z-def}
\end{equation}
Then we find
\begin{equation}
    \begin{split}
        \Omega_{j_1,\ldots,j_n} &= 
        \frac{2^{2-j-2\epsilon}\pi^{1-\epsilon}}
        {\prod_{k=1}^{n}\Gamma(j_k)
        \Gamma(2 - j - 2 \epsilon)}
        \int_0^1 \left[\prod_{k=1}^{n} 
        d x_k (x_k)^{j_k-1}\right]
        \delta\left(\sum_{k=1}^{n}x_k - 1\right)
    \\&\times
        \left[\prod_{k=1}^{n}\prod_{l=k}^{n}
        \frac{dz_{kl}}{2\pi i} \Gamma(-z_{kl})
        (x_k x_l v_{kl})^{z_{kl}}\right]
        \Gamma(j + 2 z)
        \Gamma(1-j-\epsilon-z).
    \end{split}
    \label{eq:Oj1jn-7}
\end{equation}
Collecting all factors of the $x$'s we find
\begin{equation}
    \begin{split}
        \Omega_{j_1,\ldots,j_n} &= 
        \frac{2^{2-j-2\epsilon}\pi^{1-\epsilon}}
        {\prod_{k=1}^{n}\Gamma(j_k)
        \Gamma(2 - j - 2 \epsilon)}
        \int_0^1 \left[\prod_{k=1}^{n} 
        d x_k (x_k)^{j_k-1+z_k}\right]
        \delta\left(\sum_{k=1}^{n}x_k - 1\right)
    \\&\times
        \left[\prod_{k=1}^{n}\prod_{l=k}^{n}
        \frac{dz_{kl}}{2\pi i} \Gamma(-z_{kl})
        (v_{kl})^{z_{kl}}\right]
        \Gamma(j + 2 z)
        \Gamma(1-j-\epsilon-z),
    \end{split}
    \label{eq:Oj1jn-8}
\end{equation}
where we have introduced
\begin{equation}
    z_k = \sum_{l_1}^k z_{lk} + \sum_{l=k}^n z_{kl}.
\label{eq:zk-def}
\end{equation}
In words, $z_k$ is the sum of all variables that involve $k$ 
as one of their indices, such that $z_{kk}$ itself is counted twice, i.e., $z_k = z_{1k} + \ldots + z_{k-1k} + 2z_{kk} + 
z_{kk+1} + \ldots + z_{kn}$. To finish the computation, realize that the integration over the Feynman parameters 
can now be performed using the formula
\begin{equation}
    \int_0^1 \left[\prod_{k=1}^{N} 
        d x_k (x_k)^{p_k-1}\right]
        \delta\left(\sum_{k=1}^{N}x_k - 1\right)
        = \frac{\prod_{k=1}^{N} \Gamma(p_k)}
        {\Gamma\left(\sum_{k=1}^{N} p_k\right)}.
\end{equation}
Applying the above formula with $p_k = j_k+z_k$ and noting 
that $\sum_{k=1}^n (j_k + z_k) = j + 2z$ (recall $j$ is the 
sum of all exponents $j_k$), we arrive at our final result:
\begin{equation}
    \begin{split}
        \Omega_{j_1,\ldots,j_n} &= 
        \frac{2^{2-j-2\epsilon}\pi^{1-\epsilon}}
        {\prod_{k=1}^{n}\Gamma(j_k)
        \Gamma(2 - j - 2 \epsilon)}
        \left[\prod_{k=1}^{n}\prod_{l=k}^{n}
        \int_{-i\infty}^{+i\infty}
        \frac{dz_{kl}}{2\pi i} \Gamma(-z_{kl})
        (v_{kl})^{z_{kl}}\right]
        \\&\times
        \left[\prod_{k=1}^{n}\Gamma(j_k+z_k)\right]
        \Gamma(1-j-\epsilon-z).
    \end{split}
    \label{eq:Oj1jn-fin}
\end{equation}
We have thus derived an $\frac{n(n+1)}{2}$-fold Mellin-Barnes 
integral representation for the general angular integral with 
$n$ denominators.

Before presenting some examples, let us make some comments. First, 
notice that the derivation of eq.~(\ref{eq:Oj1jn-fin}) implicitly 
assumes that the exponents $j_k$ are not zero or negative integers, 
and indeed, eq.~(\ref{eq:Oj1jn-fin}) is clearly ill-defined if any 
exponent is a non-positive integer. The case of an exponent being 
zero is obviously uninteresting from a practical point of view: it 
simply signals that the given denominator is not actually present 
in the integrand. On the other hand, negative integer powers, i.e., 
dot-products $p_j\cdot q$ in the \emph{numerator} do sometimes appear 
in practical applications. In such cases, when say 
$-j_k \in \mathbb{N}^+$, we can attempt to analytically continue 
eq.~(\ref{eq:Oj1jn-fin}) to the required value of $j_k$, e.g. by 
setting $j_k \to j_k + \delta$ and performing the analytic 
continuation $\delta \to 0$. In practice, this analytic continuation 
can be performed using the same methods and tools that we employ to 
analytically continue Mellin-Barnes integrals in the parameter of 
dimensional regularization $\epsilon$ to zero and that will be 
explained in depth in section~\ref{sec:2poles}.

Second, the derivation of eq.~(\ref{eq:Oj1jn-fin}) also assumes that 
all variables $v_{kl}$ are non-zero (in fact, positive). However, 
it may well happen that some $v_{kl}$ is zero, say when some momentum 
$p_i$ in eq.~(\ref{eq:pmu-def}) is massless, which implies $v_{ii} = 0$. 
In such cases eq.~(\ref{eq:Oj1jn-fin}) clearly cannot be used as it 
stands. However, it is straightforward to adapt the derivation to 
such cases, since if some $v_{ij}$ is identically zero, the only 
modification is that the corresponding term is missing from the sum 
in eq.~(\ref{eq:1-beta2}). Then, the integration over $z_{ij}$ 
in eq.~(\ref{eq:vklsum-MB}) is absent, but the rest of the derivation 
goes through as before. The final result is that we must drop all 
integrations that correspond to variables which are identically zero. 
Practically this amounts to the following simple changes in 
eq.~(\ref{eq:Oj1jn-fin}): first, the double product over $k$ and $l$ 
in the first bracket in eq.~(\ref{eq:vklsum-MB}) is restricted to those 
values of $k$ and $l$ for which $v_{kl} \ne 0$ and second, the sums 
defining $z$ and $z_k$ in eqs.~(\ref{eq:z-def})~and~(\ref{eq:zk-def}) 
are similarly restricted to those $z_{kl}$ for which $v_{kl} \ne 0$.

Finally, we mention that the general angular integral $\Omega_{j_1,\ldots,j_n}$ can be expressed in a compact way with the $H$-function of several variables introduced in eq.~(\ref{eq:H-func-def}). The details of this representation are given in~\cite{Somogyi:2011ir}.
%\deleted{\commgs{QUESTION: mention $H$ function of several variables?}}

Let us now turn to some examples. First, we consider the angular 
integral with a single massless denominator,
\begin{equation}
    \Omega_j(0,\epsilon) = 
    \int d\Omega_{d-2} \int_{-1}^{1} d\cos(\theta_1)
    (\sin\theta_1)^{-2\epsilon}(1 - \cos\theta_1)^{-j}\,.
\label{eq:Ojm0}
\end{equation}
The single momentum $p_1$ is massless and so $v_{11} = 0$, hence 
the discussion above regarding variables that are identically zero 
applies. Noting that $z_1=z=0$, we obtain the zero dimensional Mellin-Barnes representation,
\begin{equation}
    \Omega_j(0,\epsilon) = 2^{2-j-2\epsilon}\pi^{1-\epsilon}
    \frac{\Gamma(1-j-\epsilon)}{\Gamma(2-j-2\epsilon)}\,.
\label{eq:Om-j-0}
\end{equation}
This result is simple to verify, since the integral can be 
performed in terms of gamma functions after setting 
$\cos(\theta_1) \to 2s-1$. Then using eq.~(\ref{eq:Omp}) for the 
angular volume $\Omega_{d-2}$, we obtain eq.~(\ref{eq:Om-j-0}) 
immediately.

Although we have already derived an expression for the angular 
integral with one massive denominator, 
\begin{equation}
    \Omega_j(v_{11},\epsilon) = 
    \int d\Omega_{d-2} \int_{-1}^{1} d\cos(\theta_1)
    (\sin\theta_1)^{-2\epsilon}(1 - \beta_1\cos\theta_1)^{-j}\,,
\label{eq:Ojm}
\end{equation}
in eq.~(\ref{eq:Om-j-m}), for the sake of completeness, let 
us nevertheless discuss how this result may be derived from our 
master formula eq.~(\ref{eq:Oj1jn-fin}). Since now $v_{11} \ne 0$ 
and furthermore $z_1 = 2 z_{11}$, $z=z_{11}$ (see eqs.~(\ref{eq:zk-def})~and~(\ref{eq:z-def})), we find the one-dimensional Mellin-Barnes integral representation
\begin{equation}
\begin{split}
    \Omega_j(v_{11},\epsilon) &= 
    \frac{2^{2-j-2\epsilon} \pi^{1-\epsilon}}
    {\Gamma(j)\Gamma(2-j-2\epsilon)}
    \int_{-i\infty}^{+i\infty} \frac{dz_{11}}{2\pi i}
\\&\times
    \Gamma(-z_{11}) \Gamma(j+2z_{11}) \Gamma(1-j-\epsilon-z_{11})
    (v_{11})^{z_{11}}\,.
\end{split}
\label{eq:Om-1-m-MB}
\end{equation}
This Mellin-Barnes integral can be evaluated by repeating the 
steps that lead from eq.~(\ref{eq:Oj1jn-4}) to eq.~(\ref{eq:Oj1jn-5}): 
after writing $\Gamma(j+2z_{11})$ as a product of the gamma 
functions $\Gamma(\frac{j}{2}+z_{11})$ and $\Gamma(\frac{j+1}{2}+z_{11})$ 
using the relation in eq.~(\ref{eq:doubling-rel}), the resulting 
integral is of the form given in eq.~(\ref{eq:2F1-MB}) and can 
be evaluated in terms of a ${}_2F_1$ hypergeometric function. In 
fact, we may read off the final result simply by comparing our 
integral in eq.~(\ref{eq:Om-1-m-MB}) to the right hand side of eq.~(\ref{eq:ang-2F1-MB}),
\begin{equation}
    \Omega_j(v_{11},\epsilon) = 
    2^{2-2\epsilon}\pi^{1-\epsilon}
    \frac{\Gamma(1-\epsilon)}{\Gamma(2-2\epsilon)}
    {}_2F_1\left(\frac{j}{2},\frac{j+1}{2},
    \frac{3}{2}-\epsilon,
    1-4v_{11}\right)\,.
\end{equation}
This result clearly agrees with our earlier one in 
eq.~(\ref{eq:Om-j-m}) since $v_{11} = (1-\beta_1^2)/4$ 
(see eq.~(\ref{eq:vkl-def})).

Turning to a less trivial example, consider now the angular 
integral with two massless denominators,
\begin{equation}
    \begin{split}
    \Omega_{j,k}(v_{12},\epsilon) &= 
        \int d\Omega_{d-3} 
        \int_{-1}^{1} d(\cos\theta_1)(\sin\theta_1)^{-2\eps}
        \int_{-1}^{1} d(\cos\theta_2)(\sin\theta_2)^{-1-2\eps}
\\& \times
        (1-\cos\theta_1)^{-j}
        (1-\cos\chi_2^{(1)}\cos\theta_1
        -\sin\chi_2^{(1)}\sin\theta_1\cos\theta_2)^{-k}\,.
    \end{split}
\label{eq:Ojkm0}
\end{equation}
Since $p_1^2=p_2^2=0$, only $v_{12}$ is non-zero and we must drop 
from eq.~(\ref{eq:Oj1jn-fin}) the integrations corresponding to 
$v_{11}$ and $v_{22}$ as discussed above. Hence, we obtain a 
one-dimensional Mellin-Barnes representation. Using 
eqs.~(\ref{eq:zk-def})~and~(\ref{eq:z-def}), we see that
$z_1 = z_2 = z = z_{12}$, so
\begin{equation}
    \begin{split}
    \Omega_{j,k}(v_{12},\epsilon) &= 
        \frac{2^{2-j-k-2\epsilon} \pi^{1-\epsilon}}
        {\Gamma(j)\Gamma(k)\Gamma(2-j-k-2\epsilon)}
        \int_{-i\infty}^{+i\infty} \frac{dz_{12}}{2\pi i}
\\& \times
        \Gamma(-z_{12})\Gamma(j+z_{12})\Gamma(k+z_{12})
        \Gamma(1-j-k-\epsilon-z_{12})(v_{12})^{z_{12}}\,.
    \end{split}
\label{eq:Om-2-0m-MB}
\end{equation}
This Mellin-Barnes integral can be evaluated immediately in 
terms of a ${}_2F_1$ hypergeometric function using 
eq.~(\ref{eq:2F1-MB}) and we find
\begin{equation}
    \Omega_{j,k}(v_{12},\epsilon) = 
    2^{2-j-k-2\epsilon} \pi^{1-\epsilon}
    \frac{\Gamma(1-j-\epsilon) \Gamma(1-k-\epsilon)}
    {\Gamma(1-\epsilon) \Gamma(2-j-k-2\epsilon)}
    {}_2F_1(j,k,1-\epsilon,1-v_{12})\,.
\label{eq:Om-2-0m-fin}
\end{equation}
Before moving on to the next example, let us make some comments.
First, the argument of the hypergeometric function is simply 
related to the angle enclosed by the spatial parts of the 
vectors $p_1$ and $p_2$. Using eq.~(\ref{eq:vkl-def}), we see 
immediately that $v_{12} = (1+\cos\chi_2^{(1)})/2$. Second, 
although eq.~(\ref{eq:Om-2-0m-MB}) was derived under the assumption 
that both $j$ and $k$ are not zero or negative integers. Nevertheless, 
the final result in eq.~(\ref{eq:Om-2-0m-fin}) applies in such cases 
as well. Finally, we note that in practical applications we often 
require the expansion of this result in $\epsilon$ around zero for 
specific integers $j$ and $k$. We may arrive at such an expansion by 
resolving the poles of the Mellin-Barnes integral in 
eq.~(\ref{eq:Om-2-0m-MB}) using the methods of section~\ref{sec:1poles} 
and expanding the integrand. Then the expansion coefficients will 
generically be given by (sums of) one-dimensional finite Mellin-Barnes 
integrals that no longer depend on $\epsilon$. In chapter~\ref{chapter-MBanal} 
we will present several tools that can be used to analytically 
compute these integrals.

As our last example, let us consider the generalization of the 
previous integral to the case when one of the momenta, say $p_1$, 
is massive,
\begin{equation}
    \begin{split}
    \Omega_{j,k}(v_{11},v_{12},\epsilon) &= 
        \int d\Omega_{d-3} 
        \int_{-1}^{1} d(\cos\theta_1)(\sin\theta_1)^{-2\eps}
        \int_{-1}^{1} d(\cos\theta_2)(\sin\theta_2)^{-1-2\eps}
\\& \times
        (1-\beta_1 \cos\theta_1)^{-j}
        (1-\cos\chi_2^{(1)}\cos\theta_1
        -\sin\chi_2^{(1)}\sin\theta_1\cos\theta_2)^{-k}\,.
    \end{split}
\label{eq:Ojkm}
\end{equation}
Since now only $v_{22}$ is identically zero, we obtain a two-dimensional 
Mellin-Barnes representation. Using $z_1 = 2z_{11}+z_{12}$, 
$z_2 = z_{12}$ and $z = z_{11}+z_{12}$ (see eqs.~(\ref{eq:zk-def})~and~(\ref{eq:z-def})), we find
\begin{equation}
    \begin{split}
    \Omega_{j,k}(v_{11},v_{12},\epsilon) &= 
        \frac{2^{2-j-k-2\epsilon}\pi^{1-\epsilon}}
        {\Gamma(j)\Gamma(k)\Gamma(2-j-k-2\epsilon)}
        \int_{-i\infty}^{+i\infty} 
        \frac{dz_{11}}{2\pi i} \frac{dz_{12}}{2\pi i}
\\& \times
        \Gamma(-z_{11})\Gamma(-z_{12})
        \Gamma(j+2z_{11}+z_{12})\Gamma(k+z_{12})
\\& \times
        \Gamma(1-j-k-\epsilon-z_{11}-z_{12})
        (v_{11})^{z_{11}} (v_{12})^{z_{12}}\,.
    \end{split}
\label{eq:Om-2-1m-MB}
\end{equation}
We will show in section~\ref{sec:MBtoEulertoInt} that this integral can be evaluated 
in a closed form in terms of the Appell function of the first kind,
\begin{equation}
    \begin{split}
    &
    \Omega_{j,k}(v_{11},v_{12},\epsilon) =
        2^{2-j-k-2\epsilon}\pi^{1-\epsilon} 
        \frac{\Gamma(1-k-\epsilon)}{\Gamma(2-k-2\epsilon)}
        v_{12}^{-j}
   F_1\bigg(
    j, 1-k-\epsilon, 
    \\&\quad
    1-k-\epsilon, 
    2-k-2\epsilon,
    \frac{2v_{12}-1-\sqrt{1-4v_{11}}}{2v_{12}},
    \frac{2v_{12}-1+\sqrt{1-4v_{11}}}{2v_{12}}
    \bigg)\,.
    \end{split}
\label{eq:Om-2-1m-fin}
\end{equation}
As in the case of the massless two-denominator angular integral, 
it is interesting to note that although eq.~(\ref{eq:Om-2-1m-MB}) 
was derived under the assumptions that $j$ and $k$ are not zero or 
negative integers, nevertheless, the solution in 
eq.~(\ref{eq:Om-2-1m-fin}) is valid even for these cases. For practical applications we are often interested in the $\epsilon$ expansion of the final result in eq.~(\ref{eq:Om-2-1m-fin}). In section~\ref{sec:expansion-special-fcns}, we will present methods that allow to perform this expansion for integer values of the parameters $j$ and $k$, based on the double sum representation of the Appell $F_1$ function given in eq.~(\ref{eq:Appel-F1-def}). 

As pointed out already in~\cite{Somogyi:2011ir}, angular integrals with three or more denominators involve at least 3-dimensional \mb{} representations (depending on how many momenta are massive), whose analytic evaluation is not straightforward. Nevertheless, in recent years, progress has been made in this direction and results expanded in the parameter of dimensional regularization up to finite terms are now available for the three- and four-denominator cases~\cite{Ahmed:2024pxr,Ahmed:2025yrx}. They involve MPLs with algebraic arguments (i.e., square root expressions) in terms of the $v_{ij}$ variables. Moreover, in~\cite{Smirnov:2024pbj} a conjecture for the small-mass asymptotics of the general massive $n$-denominator angular integral was presented. The conjecture is underpinned by explicit results for the three- and four-denominator cases.

\section*{Problems}
\addcontentsline{toc}{section}{Problems}
\begin{problem}
Prove the identity in eq.~(\ref{eq:expIdent}). \label{prob:expIdent}
\\ \noindent \hint{} For proofs of relations including also Symanzik polynomials, see~\cite{Nakanishi:1971,Nakanishi:1961}.
\end{problem}
\begin{problem}
Prove the equivalence of the Lee-Pomeransky representation of eq.~(\ref{eq:LeePom}) with the Feynman parametrization of eq.~(\ref{FeynSgen}).
\\
\hint{} Insert $1=\int_0^\infty d\eta\,\delta(\eta-\sum_{j=1}^{N}z_j)$, substitute
$z_j=\eta\,x_j$ for $ j=1,\ldots, N$ and identify the corresponding gamma functions. %heinrich - collider
\end{problem}
\begin{problem}  It is hard by a naked eye to see how cutting lines in Fig.~\ref{fig:2looptrees}
we get appropriate spanning tree $T$ and $k-$forest ${\cal{T}}_k$ for a given diagram (here $k=2$).  Rearrange cut diagrams in Fig.~\ref{fig:2looptrees} to see it explicitly, in a similar way as  shown in Fig.~\ref{fig:FUbox}. \label{problemFU}
\\ \noindent \hint{} For a complete solution, see~\cite{Dubovyk:2019ivv}, section 4.2.2.
\end{problem}

\begin{problem}
Generate $F$ and $U$ polynomials for integral in Fig.~\ref{fig:2looptrees}.
Do the same to find Symanzik polynomials as in eqs.~(\ref{umassless})~and~(\ref{fmassless}). 
\\ \hint{}
 Use \mbm{} or the \texttt{Mathematica} module which can be found in~\cite{wwwFU}. 
 \label{prob:FUalgebraic}
\end{problem}

\begin{problem} 
\label{prob:mbproof} Show that eq.~(\ref{eq:mbRHS1}) gives eq.~(\ref{eq:mbRHS2}).
\\ \hint{} Use the  Cauchy's residues theorem.
\end{problem} 
\begin{problem} 
 Following  section~\ref{sec:3MBrepr}, show that for (\ref{eq:F1lvert}) 
the \mb{} master formula in eq.~(\ref{MBformula}) gives
\begin{eqnarray}
    I_{\rm MB} &\sim &\int dz_1 dz_2 dz_3 \, (- s x_1 x_2)^{z_1} (- q_2^2 x_1 x_3)^{z_2}
    (- q_3^2 x_2 x_3)^{z_3} \nonumber \\
    &&\times \left(x_1 m_1^2 + x_2 m_2^2 + x_3 m_3^2\right)^{-z_1-z_2-z_3 - N_{\nu} + d/2}.
\end{eqnarray}
\end{problem}

\begin{problem} 
Assign external and internal momenta to the diagrams in Fig.~\ref{fig:planarNP}. Show that after integrating one internal momenta (removing corresponding lines where the integrated internal momenta is the only internal momenta), not all vertices in the non-planar case conserve momenta. \label{problem:PLNPmomenta}
\end{problem}

\begin{problem}
Built up the matrix $M_{\Gamma}$ of eq.~(\ref{eq:mgamma}) for the transformation of $z$ variables in the \mb{} integrals given in eqs.~(\ref{eq:MB-SE1l2m2dim})~and~(\ref{eq:B7NPb}).  
\\ \noindent \hint{} See the solution in \wwwaux{Zmatrix}.
\end{problem}

\begin{problem} 
Analyze a proof of the \texttt{CW} theorem  using  a notion of sector decomposition~\cite{Binoth:2000ps} as given in the appendix of~\cite{Heinrich:2021dbf}.  \label{problem_sd_cw}
\end{problem}

\begin{problem}
For the diagrams in Fig.~\ref{fig:skel1ex} find a suitable set of transformations of variables analogously as in eq.~(\ref{eq:transfexnplvertex}) and derive eq.~(\ref{skel1exU3loop}). 
\label{prob:3loopskeletonA}
\\ \noindent \hint{} See the solution in \wwwaux{3L}.
\end{problem}

\begin{problem}
 \label{gammazero} Consider the 2-loop sunset diagram in Fig.~\ref{fig:sunset}. Derive corresponding \mb{} integral and find relations among gamma functions in the numerator which cancel against $\Gamma[0]$ in the denominator. \\ \noindent
 \hint{} Find suitable transformation of integration variables in the \mb{} integral and use \texttt{1BL}. See the solution in \wwwaux{Gamma0}.
\end{problem}

\begin{problem} \label{prob:beyondMPL}
Consider the banana diagram given in chapter~\ref{chapter:complex},  Fig.~\ref{fig:higherPL}. The partial fractioning for the banana's $F$ polynomial is discussed in the note~\ref{tip:banana} `Applying the CW theorem' and eq.~(\ref{eq:bananaF}). It eventually leads to square roots, which cannot be integrated in an iterative way using HPL and MPL formalism (so solutions go beyond MPLs).
For pedagogical example in $d=2$, see~\cite{Bourjaily:2022bwx}.
\end{problem}

\begin{problem}
 \label{prob:vertexSkeleton}.
Starting from the skeleton diagram on right of Fig.~\ref{fig:skel} generate diagrams in Fig.~4.22 in~\cite{Dubovyk:2019ivv}. At this level it is already more complicated to recognize different topologies by eye,  the two diagrams in Fig.~\ref{fig:int11} %have been treated as an independent diagram though they 
are actually the same, for topologies recognition, see e.g.~\cite{Bielas:2013rja,Gerlach:2022qnc}.
\begin{figure}[h!]
\centering 
\includegraphics[scale=.3]{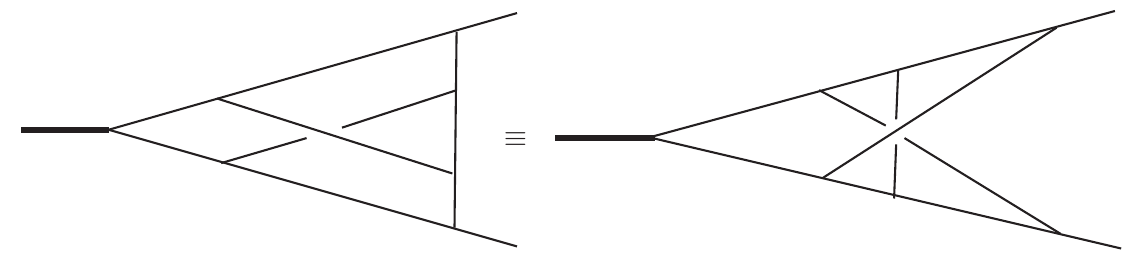}
\caption{Apparently two the same topologies.}
\label{fig:int11}
%\end{subfigure}%
\end{figure}

\end{problem}

\begin{problem} Apply the \la{} method to get a 4-dimensional \mb{} representation for the 3-loop integral $I_{\rm SE6l0m}$ of eq.~(\ref{eq:SE6l0m}).  \label{prob:3loopmerc}
\\ \noindent \hint{} For a solution, see \wwwaux{SE6l0m}. 
\end{problem}

\begin{problem}
\label{problem_prausa1} The Method of Brackets consists of a small set of simple rules on how to rewrite a Schwinger-parameterized integral into a form to which this master formula can be applied, as discussed in section~\ref{sec:prausa}.
To see all the steps in practice, follow a non-planar two-loop didactic example discussed in section E.5 in~\cite{Blondel:2018mad}. 
\end{problem}

\begin{problem}
\label{problem_prausa2} For applying the method of brackets and \texttt{GRMT} to solve some one-loop cases, follow examples discussed in~\cite{Gonzalez:2008xm}. 
\end{problem}

\begin{problem}
\label{problem_PSint}
Evaluate the integrals in eq.~(\ref{eq:PSint-ex-int}) and verify the result in eq.~(\ref{eq:PSint-ex-res}). You will have to pay careful attention to the limits of integration imposed by the positivity of energies as well as the range of $\cos\theta$.
\\ \noindent \hint{} You may find it useful to partial fraction the integrand. See also \wwwaux{PSint}

\end{problem}

\begin{problem}
 \label{prob:dangles} Using the basic $d$-dimensional Gaussian integral $\int d^d k e^{-\frac{k^2}{2}}=(2\pi)^{\frac d2}$  prove eq.~(\ref{eq:dOmq-def}) for integer $p$. 
 \\ \noindent \hint{} This is rather an elementary derivation, see for instance~\cite{Zeidler:qed}, however, often applied, see e.g. eq.~(\ref{k_integration2}).  
\end{problem}

\begin{problem}
 \label{problem_OMB}
 Prove eq.~(\ref{eq:vklsum-MB}).
  \\ \noindent \hint{} The result is actually just an application of the general formula in eq.~(\ref{MBformula}), the only non-trivial part is keeping track of the indexing. In this regard, it may help to write out the double sum explicitly, 
  $\sum_{k=1}^n\sum_{l=k}^n x_k x_l v_{kl} = x_1 x_1 v_{11} + x_1 x_2 v_{12} + \ldots + x_{n-1} x_n v_{n-1n} + x_n x_n v_{nn}$.
  \end{problem}

\putbib[%
bibs/refs,%
bibs/2loops_LL16,%
bibs/Phd_Dubovyk,%
bibs/LRrefa,%
bibs/2loopsreport]
\end{bibunit}

%% file: FIGS/tikz.tex
\begin{figure}[!h]

\centering

\begin{tikzpicture}
 
\node[rectangle, inner sep=10pt, draw=black!50, fill=black!10] (A)
{
Input integral
};

\node[rectangle, inner sep=5pt, align=left, draw=black!50, fill=black!10] (B) [right=of A]
{
Diagram's analysis:\\ 
PlanarityTestv1.3.4 \cite{Bielas:2013v12}
};

\node[rectangle, inner sep=5pt, align=left, draw=black!50, fill=black!10] (C) [right=of B]
{
\MB{} construction: \\
AMBREv1.3.1 \\ 
AMBREv2.1.1 \cite{Gluza:2010v22} \\ 
AMBREv3.1.1 \cite{Dubovyk:201509v30x}
};

\node[rectangle, inner sep=5pt, align=left, draw=black!50, fill=black!10] (D) [below=of C]
{
$\epsilon$ continuation: \\
MB.m \cite{Czakon:2005rk} \\
MBresolve.m \cite{mbtools-smirnov}
};

\node[rectangle, inner sep=5pt, align=left, draw=black!50, fill=black!10] (E) [left=of D]
{
Optimization of output: \\
barnesroutines.m \cite{mbtools-kosower}
};

\node[rectangle, inner sep=5pt, align=left, draw=black!50, fill=black!10] (F) [left=of E]
{
Numerical integration: \\
MB.m \\ 
MBnumerics.m \cite{Usovitsch:201606} \\
or analitycal approaches
%summation: \\
%MBsums.m \cite{Ochman:2015fho}
};

\draw [->] (A) to (B);
\draw [->] (B) to (C);
\draw [->] (C) to (D);
\draw [->] (D) to (E);
\draw [->] (E) to (F);

\end{tikzpicture}

\caption{ \label{scheme1}
        The operational sequence of the \MB{}-suite. This flowchart shows the main steps and the correponding software to produce a Mellin-Barnes representation of a Feynman integral and perform its numerical or analytical solution.}
\end{figure}

%% file: mprausasubfigs.tex
\begin{figure}
      \centering
      \begin{subfigure}[b]{.19\textwidth}
        \centering
        \includegraphics[scale=.8]{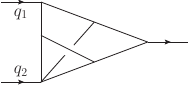}
        \caption{}
        \label{fig:int1}
      \end{subfigure}%
      \begin{subfigure}[b]{.19\textwidth}
        \centering
        \includegraphics[scale=.8]{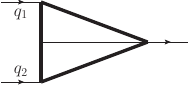}
        \caption{}
        \label{fig:int2}
      \end{subfigure}%
      \begin{subfigure}[b]{.19\textwidth}
        \centering
        \includegraphics[scale=.8]{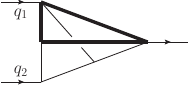}
        \caption{}
        \label{fig:int3}
      \end{subfigure}%
      \begin{subfigure}[b]{.19\textwidth}
        \centering
        \includegraphics[scale=.8]{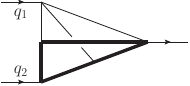}
    \caption{}
        \label{fig:int4}
      \end{subfigure}%
      \begin{subfigure}[b]{.19\textwidth}
        \centering
        \includegraphics[scale=.8]{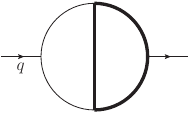}
        \caption{}
        \label{fig:int5}
      \end{subfigure}%
      \caption{ \label{fig:prausa} Examples of two- and three-loop diagrams. 
      Bold (thin) lines represent
massive (massless) propagators. Taken from \cite{Prausa:2017frh}, results discussed in Tab.~\ref{tab:comp}.}
\end{figure}

%% file: chapter4.tex
%%%%%%%%%%%%%%%%%%%%% chapter.tex %%%%%%%%%%%%%%%%%%%%%%%%%%%%%%%%%
\begin{bibunit}[elsarticle-num-ID] % define the bib-style for the unit: elsarticle-num.bst
%  text-1; this is the corresponding section
%\putbib[2loops] % the *.bib
%\end{bibunit}
% go-on
%--- from: bibunits.sty, adapts the font size of ``References'' to section
\let\stdthebibliography\thebibliography
\renewcommand{\thebibliography}{%
\let\section\subsection
\stdthebibliography} 
%
% sample chapter
%
% Use this file as a template for your own input.
%
%%%%%%%%%%%%%%%%%%%%%%%% Springer-Verlag %%%%%%%%%%%%%%%%%%%%%%%%%%
%\motto{Use the template \emph{chapter.tex} to style the various elements of your chapter content.}
\chapter{Resolution of Singularities}
\label{chapter-singul}  

\abstract{We discuss strategies for extracting singularities of Feynman integrals connected with dimensional regularization ($\eps$ poles) using Mellin-Barnes representations. Based on a concrete example, we outline an algorithmic procedure for the Laurent expansion of \mb{} integrals in the parameter $\eps$. Some peculiarities in the analytic continuation for $\eps \to 0$ are discussed.}

\section{Where do the Poles Come From?}
\label{sec:1poles}
 
In chapter~\ref{chapter:intro} we discussed nature of singularities which appear in \texttt{FI}.  
In the context of \mb{} integrals we are especially interested in singularities connected with the dimensional parameter $d=n-2 \epsilon$  ($n$ must not necessarily be four~\cite{Tarasov:1996br}, see Problem~\ref{prob:dimnot4}).  
The $\epsilon$ regulator is transformed from Feynman parametrization of integrals to the arguments of gamma functions, as was discussed on the occasion of construction of \mb{} integrals in chapter~\ref{chapter-MBrepr}.
We will discuss now how these singularities can be resolved.
\mb{} representations allow in a systematic way to treat this class of IR and UV singularities. 

\section{Resolving Poles: Straight Line and Deformed Contours}
\label{sec:2poles} 

After defining the Mellin-Barnes representation for the Feynman integral, we are interested in further processing, aiming at analytical and numerical solutions.  Let us then evaluate further Eq.~(\ref{MB-SE1l2m-lemmas}). This integrals depends on the dimensional parameter $\epsilon$, which in the four-dimensional physical limit tends to zero. Naturally, we need the Laurent expansion of the constructed MB integral representations.  Then we can use the Cauchy theorem to get a representation in terms of a sum of contour integrals, valid at $\epsilon=0$. As discussed in section~\ref{sec:2MBrepr}, it is important that the Mellin-Barnes representation is only well defined, if the integration contour separates the left poles $\Gamma (\ldots + z)$ from the right poles $\Gamma (\ldots - z)$ and in general for the combination of gamma functions with $\epsilon =0$ that will not be the case (if the contour is chosen to be a straight line parallel to the imaginary axis). In practice two solutions to this problem have been proposed:
\begin{itemize}
        \item The ''Tausk method'' - which fixes the contours parallel to the imaginary axis and accounts for the poles crossing in the analytic continuation~\cite{Tausk:1999vh}. The idea of it is presented in Fig.~\ref{SEpoles}. This method was implemented in the \texttt{MB} program~\cite{Czakon:2005rk} written in \math{} and implemented for numerical calculations in~\cite{Anastasiou:2005cb} (unpublic), see appendix~\ref{app:MB.m}.
        \item The ''Smirnov method'' - which identifies gamma functions responsible for the generation of poles in $\epsilon$. The general observation is that $\Gamma[a+z_1]\Gamma[b-z_1]$ with $a$ and $b$ which can depend on other variables $z_i$, generates the pole $\Gamma[a+b]$. So, for instance, $\Gamma[1+z_1]\Gamma[-1-z_1-\epsilon]$ generates directly a pole $\Gamma[-\eps]$.  The algorithm for identifying all poles in $\eps$ for a given set of gamma functions is implemented in the \texttt{MBresolve} program~\cite{Smirnov:2009up},
        see appendices~\ref{app:MB.resolve}~and~\ref{appB} (Smirnov's lecture).
\end{itemize}

\begin{figure}[h!] 
\begin{center}
        \includegraphics[scale=0.6]{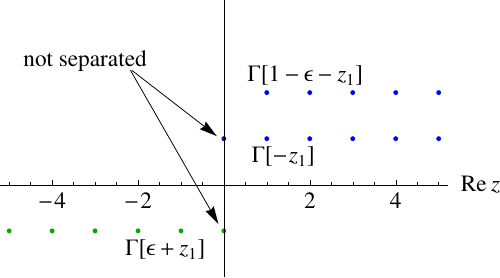}
        \includegraphics[scale=0.6]{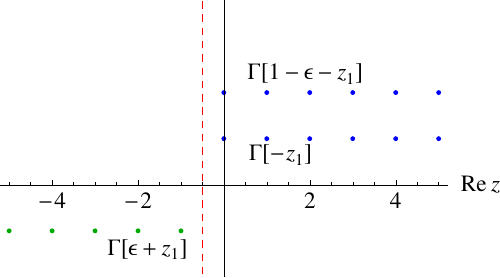}    
        \caption{The poles of the \mb{} representation in Eq.~(\ref{MB-SE1l2m-lemmas}) for negative (positive) values of ${\rm \Re}[z]$ denoted in red (blue). Left plot shows situation when $\epsilon \rightarrow 0$ which leads to non-separated left and right poles. The right plot presents a proper shift which separates the poles ($\epsilon \rightarrow 1$, $\Re(z) \rightarrow -1/2$).}
  \label{SEpoles}
\end{center}
\end{figure}

Graphically we can see what happens if we put $\epsilon \rightarrow 0$ in the derived \mb{} representation, the left plot of 
Fig.~\ref{SEpoles}. We note that left and right poles are not separated from each other. To avoid that, we take $\epsilon \neq 0$, $\epsilon \rightarrow 1$ is a good option, see  the right plot of Fig.~\ref{SEpoles}. Using this scheme according to the Tausk method, the final result ($\epsilon = 0$) will be obtained from the case where the poles are separated ($\epsilon \neq 0$). Depending whether the poles crossed the contours from left or right (when $\epsilon \rightarrow 0$), one should add or subtract the residue of the integrand on that pole.

\subsection{Bromwich Contours and Se\-pa\-ra\-tion of Poles \label{ex_contour_paral}}

To discuss a pole structure of Eq.~(\ref{MB-SE1l2m-lemmas}), let us write this integral in a form which helps to identify individual gamma functions 
(what matters are gamma functions in numerators, which we assign as $G_1,G_2,G_3$; possible poles of gamma functions in the denominator makes the integral to vanish)
 
\begin{eqnarray} 
&&aux= 
        \int_{c-i \infty}^{c+i \infty} \frac{d z_1}{2 \pi i}
        (ms)^{z_1} (-s)^{-\epsilon}
        \frac{G_1 (1 - \epsilon - z_1)^2 G_2(-z_1) G_3 (\epsilon + z_1)}
         {\Gamma (2 - 2 \epsilon - 2 z_1)}, \; ms = -\frac{m^2}{s}.\nn
         \\
         &&
         \label{eq:aux}
\end{eqnarray}

According to Fig.~\ref{SEpoles}, the following gamma functions contribute to the left and right poles in Eq.~(\ref{MB-SE1l2m-lemmas}):
\begin{itemize}
        \item Left poles: $G_3(\epsilon+z_1)$,
        \item Right poles: $G_1(1-\epsilon-z_1)$,$G_2(-z_1)$. 
%$\Gamma(2-2\epsilon-2z_1)$,
\end{itemize}

As discussed in chapter 3, positive arguments of gamma function makes the function regular, see Fig.~\ref{fig:chap2_gammareim}. However, it appears, that for $\epsilon =0$, the set of inequalities build from arguments of $G_1,G_2,G_3$

\begin{eqnarray}
1-\epsilon-z_1 & > & 0, \label{ineq1}\\
-z_1 & > & 0, \label{ineq2}\\
\epsilon+z_1 &> & 0, \label{ineq3}
\end{eqnarray}
\noindent 
results in no solution in $z_1$. This can be immediately seen comparing Eqs.~(\ref{ineq2})~and~(\ref{ineq3}).

At this stage, to define contours of integration for \mb{} integrals, we follow the Tausk approach. We fix a contour of integration for \mb{} integrals parallel to the imaginary axis (Bromwich contours, see chapter~\ref{chapter:complex}). 
\begin{svgraybox}
To separate poles of gamma functions % at $\pm i \infty$, 
we start with $\eps \neq 0$ and with real parts of integration variables $z_i$ for which \mb{} integrands are well defined (positive arguments of gamma functions in a numerator of the integrand). Then we consider poles which appear when a limit  $\eps \to 0$ is taken.  
 \end{svgraybox}

Taking non-zero $\epsilon$, e.g. $\eps=1/4$, we can already find $z_1$ which satisfies Eqs.~(\ref{ineq1})--(\ref{ineq3}), e.g. $z_1=-1/4$.

So, in Step 1, taking $\epsilon = 1/4, \Re(z_1)=-1/4$, we get 
\begin{equation}
\rm{subst0} = G_1(3/4), G_2(1/4), G_3(1/4). \label{eq:step1}    
\end{equation}

We can see that with such a choice of $\eps, z_1$, the arguments of all gamma functions in the numerator of the integrand are positive (the integral is well defined). The very next step in the calculation of Eq.~(\ref{MB-SE1l2m-lemmas}) is the analytic continuation i.e. we have to finally go down to the case where $\epsilon \rightarrow 0$.

So in Step 2, taking $\epsilon = 0, \Re(z_1)=-1/4$, we get 
\begin{equation}
\rm{subst1} = G_1(5/4), G_2(1/4), G_3(-1/4).\label{eq:step2}    
\end{equation}
   
As we can see, the analytic continuation in $\epsilon \to 0$ hits the pole at $\Re(z_1)=-1/4$  for gamma $G_3[\eps+z_1]$.
So, the original integral in Eq.~(\ref{eq:aux}) can be decomposed in the following way

\begin{eqnarray}
    aux(\epsilon =-1/4, \Re(z_1)=-1/4) &=& aux(\epsilon = 0, \Re(z_1)=-1/4) \nonumber \\&+& \mathrm{Res}_{z\to -\eps} aux(\Re(z_1)=-1/4).
    \label{eq:auxbasic}
\end{eqnarray}

As $aux$ is a 1-dimensional function in $z_1$, the residue of the $aux$ function at $z =\eps$ returns a function 
   \begin{equation}
   \mathrm{Res}_{z\to -\eps} aux(\Re(z_1)=-1/4) = ms^{-\eps} \Gamma[\eps].
   \label{eq-se2l2m-aux1}
   \end{equation}

We will solve Eq.~(\ref{eq:auxbasic}) analytically using the Cauchy's residue theorem in chapter~\ref{chapter-MBanal} and purely numerically in chapter~\ref{chapter-MBnum}, taking the integral $aux(\epsilon = 0, \Re(z_1)=-1/4)$ in the complex plane.

\subsection{Analytic Continuation in $\eps$ \label{ex_anal_eps}}
 
Let us discuss more complicated, two-loop example, with some massive propagators, see Fig.~\ref{fig:b5l2m}, which has been generated in \math{} using two following commands. To get it, in addition kinematic variables must be defined for the 4PF, which is a subject of the Problem~2 in chapter~\ref{chapter-MBrepr}.

\begin{minted}[frame=single,breaklines,fontsize=\small]{mathematica}
In[1]:= int = PR[k1, 0, n1]*PR[k2 - p2, 0, n2]*PR[k1 + k2, m, n3]*PR[k2 + p3, 0, n4]*PR[k2 - p1 - p2, m, n5];
In[2]:= PlanarityTest[{int}, {k1, k2}, DrawGraph -> True]; 
\end{minted}

\begin{equation}
    \label{eq:propags}
\end{equation}

\begin{figure}[h!] 
\begin{center}
        \includegraphics[scale=0.4]{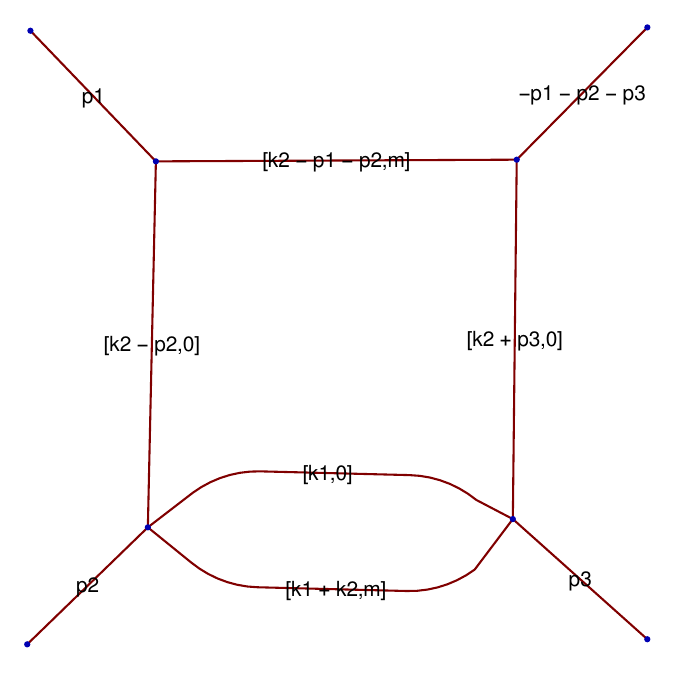}    
        \caption{Two-loop topology for Feynman integral with five internal lines (propagators, two massive internal lines). The figure generated by \pltest{}~\cite{Bielas:2013rja,ambrewww}, see appendix~\ref{app:pltest}. We will call this topology \texttt{B5l2m}, as this is a 4PF (Box) with 5 internal lines, of which 2 are massive.}
  \label{fig:b5l2m}
\end{center}
\end{figure}

\mb{} representation for this case is (see chapter~\ref{chapter-MBrepr} for details of its construction)

%B5l2m2_Springer.nb
\begin{eqnarray}
&&G(1)_{\text{B5l2m}} = -(-m^2)^{-\eps - z_1 + z_2} (-s)^{z_3}
(-t)^{-1 - \eps + z_1 - z_2 - z_3}\ \label{eq:mb_b5l2m} \\
&&\nn \\
&&\int_{c_1-i \infty}^{c_1+i \infty} \frac{d z_1}{2 \pi i} \int_{c_2-i \infty}^{c_2+i \infty} \frac{d z_2}{2 \pi i} \int_{c_3-i \infty}^{c_3+i \infty} \frac{d z_3}{2 \pi i} \nn \\
&&\nn \\
&  & \bigg( 
 G_1[1 - 2 \eps - z_1] G_2[1 + z_1] G_3[
  \eps + z_1] G_4[-z_2] G_5[-\eps + z_1 - z_2 - z_3] G_6[-z_3]
  \nn  \\
      &&\nn \\
&\times &
    \frac{  G_7[
  1 + z_3] G_8[-z_1 + z_3] G_9[1 + \eps - z_1 + z_2 + z_3] G_{10}[
  1 - z_1 + 2 z_2 + 2 z_3]}{\Gamma[2 - 2 \eps] \Gamma[1 - 2 \eps + z_1] \Gamma[1 - z_1 + 2 z_3]} \bigg). \nn
\end{eqnarray}     

Notation $G(1)_{\text{B5l2m}}$ means that \mb{} representation is for a two-loop scalar integral, see Eq.~(\ref{eq-bha}). As in Eq.~(\ref{eq:aux}), to track a discussion of identified poles in $\eps$, we numbered gamma functions in the numerator. 
This integral is three dimensional. Similarly to the previous case, we are looking for singular resolution in $\epsilon$, i.e. we want to regularize the integral in the limit $\epsilon \to 0$. Taking $\eps =0$ in arguments of all gamma functions in the numerator of Eq.~(\ref{eq:mb_b5l2m}), we can seek a solution for a set of inequalities analogous to those in Eqs.~(\ref{ineq1})--(\ref{ineq3}).
There are cases where we get no solution, see Problem~\ref{prob:chap4_nosol}. Incidentally, to find that some system of inequalities has no solutions can be quite tricky, fortunately systems like \math{} are able to solve such problems with eye blick speed using command \texttt{FindInstance}\footnote{The shortest \math{} description of the function: \\{\rm FindInstance[expr,vars] finds an instance of vars that makes the statement expr be True.}}. 

It is useful to see how it works in details for $G(1)_{\text{B5l2m}}$ as in a case of $aux$ function of Eq.~(\ref{eq:aux}) we had one-dimensional integral which terminated in a number Eq.~(\ref{eq-se2l2m-aux1}) after taking a single residue. Learning by doing, we will be able to understand the way how a general algorithm for analytic continuation in $\eps$ works. 

Taking $\epsilon = 1/10$, we find that the arguments of the $G_1$--$G_{10}$ functions in Eq.~(\ref{eq:mb_b5l2m}) are positive

\begin{eqnarray}
  \rm{rule0} =  
   \{z_1 \to  - \frac{1}{20}, z_2 \to -\frac{5}{16}, z_3 \to -\frac{1}{40} \}    \label{eq:b5l2m_step1} 
\end{eqnarray}
 
For $z_1,z_2,z_3$ values as in Eq.~(\ref{eq:b5l2m_step1}), and $\eps =0$ (which we aim at), we get following values of arguments for $G_1$--$G_{10}$

\begin{eqnarray}
&& G_1[7/6]  G_2[5/6] G_3[-(1/6)] G_4[11/24] G_5[3/8]^2 G_6[1/12] G_7[
  11/12] \nn \\
  &&G_8[1/12] G_9[5/8] G_{10}[1/12].
\end{eqnarray}

Coming from $\eps=1/10$ to $\eps =0$, function $G_3[\eps+z_1]$ changes sign, it goes through the pole. We pick up residue at $z_1=-\eps$

\begin{eqnarray}
&& Res[G(1)_{\text{B5l2m}}]|_{z_1=-\eps} =
-(m^2)^{z_2} (-s)^{z_3}
(-t)^{-2\eps - z_2 - z_3}  \int_{c_2-i \infty}^{c_2+i \infty} d z_2 \int_{c_3-i \infty}^{c_3+i \infty} d z_3 \nn \\
&&\nn \\
&  &  G_1[1 - \eps]^2
  G_2[-z_2] G_3[-2 \eps - z_2 - z_3]^2 G_4[-z_3] G_5[1 + z_3] G_6[\eps +z_3]
  \nn \\
  &&\nn \\
&  &    \frac{   G_7[1 + 2 \eps + z_2 + z_3] G_8[1 + \eps + 2 z_2 + 2 z_3]}{G[1 - 3 \eps] \Gamma[2 - 2 \eps] \Gamma[1 + \eps + 2 z_3]}.  \label{eq:mb_b5l2m_res1}
\end{eqnarray}

Now, we are left with two variables $z_2$, $z_3$ and a new value of $\epsilon$, which is, according to Eqs.~(\ref{eq:b5l2m_step1})~and~(\ref{eq:mb_b5l2m_res1}) equal to $\eps=-z_1=1/20$.
With this new value of $\eps$, we continue analytic continuation with $\eps \to 0$.
 
Taking $z_2$, $z_3$ as in Eq.~(\ref{eq:b5l2m_step1}), and $\eps=0$,
we get for arguments of $G_1$--$G_8$ functions in Eq.~(\ref{eq:mb_b5l2m_res1})

\begin{eqnarray}
&&G_1[1] G_2[5/16] G_3[27/80] G_4[1/40] G_5[39/40] G_6[-(1/40)] \nn \\
&&G_7[53/
  80] G_8[13/40].    
\end{eqnarray}

This time, coming from $\eps=1/20$ to $\eps =0$, function $G_6[\eps+z_3]$ changes sign, it goes through the pole. We pick up residue at $z_3=-\eps$

\begin{eqnarray}
&& Res[G(1)_{\text{B5l2m}}]|_{z_1=-\eps, z_3=-\eps} =
-(m^2)^{z_2} (-s)^{-\eps}
(-t)^{-\eps - z_2}  \int_{c_2-i \infty}^{c_2+i \infty} d z_2 \nn \\
&&\nn \\
&  &  \frac{G_1[1 - \eps]^2 G_2[\eps] G_3[-\eps - z_2]^2 G_4[-z_2] G_5[1 + \eps + z_2] G_6[
  1 - \eps + 2 z_2]}{\Gamma[1 - 3 \eps] \Gamma[2 - 2 \eps]}. \nn\\&& \label{eq:mb_b5l2m_res2}
\end{eqnarray}

Please note that if we took instead of $1/10$ in Eq.~(\ref{eq:b5l2m_step1}) a larger value, $1/3$ say, in the last step also $G_8$ would result in negative value of the argument, and we would have to take additional residue at $1 + \eps + 2 z_2 + 2 z_3 =0$, adding one more integral of lower dimension.  

Now, we are left with one variable $z_2$, and a new value of $\epsilon$, which is, according to Eq.~(\ref{eq:b5l2m_step1}) $\eps=-z_3=1/12$.
With this new value of $\eps$, we continue analytic continuation with $\eps \to 0$.

Taking $z_2$ as in Eq.~(\ref{eq:b5l2m_step1}), and $\eps=0$,
we get for arguments of $G_1$--$G_6$ functions in Eq.~(\ref{eq:mb_b5l2m_res2})

\begin{equation}
G_1[1] G_2[0] G_3[5/16]^2 G_4[5/16] G_5[11/16] G_6[(3/8)].
\end{equation}

We can see that there is no negative arguments anymore, and we ended the procedure\footnote{Function $G_2[0]$ is actually $\Gamma[\eps]$, which can be Taylor expanded, adding to the singularity at $\eps$ expansion of final functions in Eq.~(\ref{eq:B5l2m2fin}).}. The final result is
\begin{equation}
\begin{split}
&
G(1)_{\text{B5l2m}}{}_{\{\eps \to \frac{1}{10},\, z_1 \to  - \frac{1}{20},\, z_2 \to -\frac{5}{16},\, z_3 \to -\frac{1}{12} \} }
\\ &\quad =
G_2(1)_{\text{B5l2m}}{}_{\{\eps \to 0,\, z_1 \to  - \frac{1}{20},\, z_2 \to -\frac{5}{16},\, z_3 \to -\frac{1}{12} \} } 
\\ &\qquad + 
Res[G(1)_{\text{B5l2m}}{}_{\{\eps \to 0,\, z_1 \to  - \frac{1}{20},\, z_2 \to -\frac{5}{16},\, z_3 \to -\frac{1}{12} \} } ]\big|_{z_1=-\eps}
\\ &\qquad +
Res[G(1)_{\text{B5l2m}}{}_{\{\eps \to 0,\, z_1 \to  - \frac{1}{20},\, z_2 \to -\frac{5}{16},\, z_3 \to -\frac{1}{12} \} } ]\big|_{z_1=-\eps,\, z_3=-\eps} \label{eq:B5l2m2fin}
\end{split}
\end{equation}

The whole procedure of analytic continuation in $\eps$ is automatized in the \mbm{} package~\cite{Czakon:2005rk}, see \texttt{MBoptimezedRules} command there. 
Using this command, we can easily find a solution for a set of inequalities (Problem~\ref{prob:dimnot4}), which written in terms of substitutions can take a form

\begin{minted}[frame=single,breaklines,fontsize=\small]{mathematica}
In[3]:= rules = MBoptimizedRules[fin, eps -> 0, {}, {eps}]}
Out[3]:= {{eps -> 11/64, {z1 -> -3/64, z2 -> -7/32, z3 -> -1/32}}
\end{minted}

The notebook file with a complete solution can be found in  \wwwaux{B5l2m2}. Algorithm 1 summarizes the main steps we discussed above for analytic continuation in $\eps$. %Generalization of Example~1. 
More elaborated algorithms can be found in~\cite{Czakon:2005rk}~and~\cite{Anastasiou:2005cb}.

\begin{algorithm}
\label{algor1}
	 \SetAlgoNoLine
	\LinesNumbered
	\SetKw{KwGoTo}{go to}
	\caption{Pseudocode for analytic continuation in $\eps$, basic steps.}
	\label{alg:alg1}
	%\SetAlgoLined
	Let $F(\{z_i\},\eps)$ be a \mb{} integral over $z_i$ variables and $\eps$ is the dimensional regulator to be analytically continued to 0  \;
	\For{\label{marker} a set of  $\{\Gamma_j \}$ in the numerator of $F$} 
	{Solve[$Arg(\Gamma_j)>0,\{\{z_i\},\eps\} \subseteq X$] \;} 
	\For{$\{z_i\}$ Looking for change of sign in arguments of gammas}{Find  $\{\Gamma_k \} \in \Omega_0$; $k<j$: Solve[$Arg(\Gamma_j)<0,\{\{z_i\},\eps \} \subseteq X, \eps=0$] \;}
	{\If{$\Omega_0 = \emptyset$}{$F(\{z_i\},\eps)$ is a solution \;
		\Else{	Take residue of $F$ for each $\Gamma_i \in \Omega_0$, Res[F]$|_{Arg[\Gamma_i]=0} \equiv R_i$ \;
For each $\{\Gamma_k \} \in \Omega_0$ and $\{z_i\} \in X$, 
Solve[$Arg(\Gamma_k)==0,\{\eps_k\},  \{X, \eps_k \} \equiv  X_k$] \; 
For $\eps =\eps_k $  and $\{z_i, \eps \} \subseteq X$ and $F(\{z_i\},\eps) = R_i$ \;
\KwGoTo~\ref{marker}}}
}

\end{algorithm}
 
\subsection{ Analytic Continuation in Auxiliary Parameters \label{tips_trick_eta}}

In order to find proper integration paths for the MB-integrations, i.e.  the condition that all arguments of gamma functions are positive,  it may happen that starting with $\eps \neq 0$ no solutions for the inequalities like those in Eqs.~(\ref{ineq1})--(\ref{ineq3}) can be found.
Fortunately,  we have a freedom to increase the number of manipulated variables and add to the set of $\eps$ and $z_i$ parameters the powers of propagators $n_k$, see e.g. Eq.~(\ref{eq:propags}). We do not fix power of propagators, instead we take $n_k \to n_k +\eta_k$ ($\eta \in \mathbb{R}$). Then, of course, after finding
a solution for $\eps, z_i, \eta_k$, 
analytic continuation must be done $\eps \to 0$ {\it and} $\eta \to 0$.  
This can be automatized by a modification of Algorithm~\ref{alg:alg1}. Namely, instead starting with $F(\{z_i\}, \eps )$, we take $F(\{z_i\}, \eps , \eta_1 )$ where $\eta_1$ is  the power of the first propagator. It is better to start with one by one propagator as analytical continuation of each $\eta_i$ may proliferate additional integrals. 
If we cannot find a valid rule, $\nu_2$ is treated analogously, and the procedure can be continued until we find a set  of $\{z_i, \eps , \eta_j \}$ which is successful. {If it happens that an analytical continuation in only one additional parameter $\eta$ is not sufficient, the program will stop with a proper remark.}
We introduced the auxiliary file \texttt{MBnum.m} \cite{ambrewww} to \mbm{} which realizes this procedure
and the subsequent automatized analytic continuation, $\epsilon$-expansion,
and numerics. The complete solution is in \wwwaux{B5nf}. 

The integral under consideration is depicted in Fig.~\ref{fig:b5nf}. 

\begin{figure}[h!] 
\begin{center}
        \includegraphics[scale=0.4]{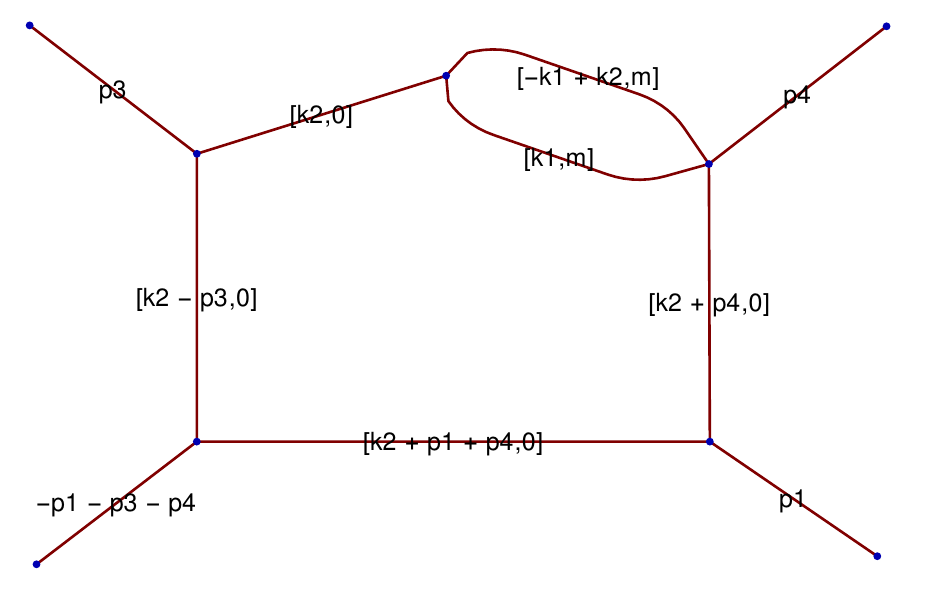}    
        \caption{Two-loop topology for Feynman integral with five internal lines (propagators, two massive internal lines). The figure generated by \pltest{}~\cite{Bielas:2013rja,ambrewww}, see appendix~\ref{app:pltest}. We will call this topology \texttt{B5nf},  with 4 internal lines and one SE insertion with two massive propagators.}
  \label{fig:b5nf}
\end{center}
\end{figure}

This is an interesting integral as \texttt{MBRules} or \texttt{MBoptimezedRules} gives solutions only for $n_5$ (and $n_6$).
There are the following additional steps (see the frame below) connected with $n_5$ in analytic continution ($n_5 \to n_5 + \eta$). Please note that in the last line, \texttt{MBexpand} is used with factor one, for $\eps$ expansion at two loops the factor is \texttt{Exp[2*eps*EulerGamma]}. 

\begin{minted}[frame=single,breaklines,fontsize=\small]{mathematica}
In[4]:= rules = MBoptimizedRules[fin,eta->0,{},{eps,eta}]}
Out[4]:= {{eps->-9/8, eta->25/16}, {z1->-3/8, z2->-1/2, z3->-31/32}}
In[5]:= Step1cont = MBcontinue[fin,eta->0,rules]};
Out[5]:= 10 integrals found ...
In[6]:= after = MBexpand[Step1cont,1,{eta,0,0}]} 
\end{minted}

\begin{equation}
    \label{eq:etas}
\end{equation}

%\end{tips}

In the examples for resolution of singularities, we explored the `Tausk method', used in \mbm{}. It is an intuitively easy approach to follow and apply. However, we should acknowledge that the `Smirnov method' is a good alternative. In fact, it has been used in~\cite{Gluza:2010ws}~or~\cite{Dubovyk:2019szj} for a simple reason that it does not need additional regulators as in subsection~\ref{tips_trick_eta} (so no proliferation of additional \mb{} integrals during analytic continuation).
 The file \texttt{MB\_B5nf\_Springer.nb} in \cite{www_aux_springer}   
 shows a difference between applying \mbm{} and \texttt{MBresolve.m}.

\section*{Problems}
\addcontentsline{toc}{section}{Problems}
\begin{problem}
 \label{prob:dimnot4} \mb{} integrals must not necessarily be considered for $n=4$ where in general the dimensional parameter is $d=n-2 \eps$. See for instance~\cite{Tarasov:1996br}~or~\cite{Dubovyk:2022frj} for application in recent studies.  
Make analytic continuation for some \mb{} representations for general $n$. 
\\ \noindent \hint{} For examples, see~\cite{ambrewww}. 
\end{problem}

\begin{problem} \label{prob:chap4_nosol}
 Show that for $\eps \to 0$, the set of solutions for inequalities constructed from arguments of gamma functions being positive in the numerator of the \mb{} representation  represented by Fig.~\ref{fig:b5nf} is empty. Find a solution, this time with $\eps \neq 0$.  \\
\hint{} The solution can be found using semidefinite programming~\cite{matouek07understanding}, e.g. in Python. Alternatively, it can be coded in \math{}, using the aforementioned function \texttt{FindInstance}. See \wwwaux{B5nf}. 
\end{problem}
\begin{problem}
Analyze analytic continuation in $\eps$ and $\eta$ parameters and identify poles of gamma functions, finding corresponding residues in \newline  \wwwaux{B5l2m2}.  
\end{problem}

\putbib[%
bibs/refs,%
bibs/2loops_LL16,%
bibs/Phd_Dubovyk,%
bibs/LRrefa,%
bibs/2loopsreport]
\end{bibunit}

%% file: chapter5v2.tex
%%%%%%%%%%%%%%%%%%%%% chapter.tex %%%%%%%%%%%%%%%%%%%%%%%%%%%%%%%%%
\begin{bibunit}[elsarticle-num-ID] % define the bib-style for the unit: elsarticle-num.bst
%  text-1; this is the corresponding section
%\putbib[2loops] % the *.bib
%\end{bibunit}
% go-on
%--- from: bibunits.sty, adapts the font size of ``References'' to section
\let\stdthebibliography\thebibliography
\renewcommand{\thebibliography}{%
\let\section\subsection
\stdthebibliography}  

\chapter{Analytic Solutions}
\label{chapter-MBanal}  

\abstract{We discuss analytic solutions of \mb{} representations for virtual loop and real phase space integrals. We present examples of how \mb{} loop integrals can be transformed to nested sums over residues and derive analytic solutions for some one- and two-dimensional sums. The evaluation of real phase space angular integrals using \mb{} representations and corresponding sums is also examined. We show how higher dimensional \mb{} integrals can be decoupled by changing \mb{} variables and how \mb{} integrals can be expressed via real Euler-type integrals. We also present a brief discussion on the symbolic evaluation of Euler-type integrals as a way of obtaining analytic solutions. Finally, we consider approximations of \mb{} integrals in ratios of both kinematic parameters and masses present in propagators of \texttt{FI}, as well as by the method of expansion by regions.}

\section{Residues and Symbolic Summations}
\label{sec:1anal}

 \begin{svgraybox}
  Summing over residues and performing the limit $\eps \to 0$ may be exchanged. In any case: When crossing a pole by letting ${\eps \to 0}$, take the residue and add it to the
expression. 
Generates lower dimensional sums. 
Then look at the resulting Mellin-Barnes integrals and for the series of remaining poles.
Take the sum over residues when closing the contour to the right or to the left.
\end{svgraybox}

This is the only recipe, doesn't it look simple? Well, at first sight, yes. Indeed, some sums can be done immediately, see e.g. section~\ref{sec:simpleinvitation} and evaluation of a binomial sum in \wwwaux{miscellaneous}.

Mathematically, Mellin-Barnes integrals are simply complex contour 
integrals whose integrands have a special structure: they involve 
products of Euler gamma functions whose arguments depend on polynomials 
of the integration variables, as well as fixed quantities raised to 
powers that again can depend on the variables of integration. Thus, 
it is quite natural to try to apply Cauchy's residue theorem to perform the 
integraions over the Mellin-Barnes variables. In order to do so, 
we must first find a suitable closed contour to which we can apply 
the theorem. As the Mellin-Barnes contours themselves run from 
$-i\infty$ to $+i\infty$, the obvious and it turns out correct guess 
is to add to this contour a semi-circle ``at infinity'' and obtain a 
closed contour in this way. We are allowed to close the contour as long as the integrand vanishes at complex infinity. 

\subsection{Choosing {the} Contour \label{sec:choosingcontour}}
In practical applications, we can usually fulfill this requirement if 
we close the contour in the correct direction. In order to get a 
feeling for why this is true and what complications might arise, let 
us consider the following simple one-dimensional Mellin-Barnes integral
\begin{equation}
    \int_{-i\infty}^{+i\infty} \frac{dz}{2\pi i}
    \Gamma(a+z)\Gamma(b-z) x^z\,.
\end{equation}
Let us consider how the integrand behaves at complex infinity. First, 
using the asymptotic formula for $\Gamma(z)$ as $z\to\infty$, eq.~(\ref{AsymptoticGamma}), 
we have
\begin{equation}
    \Gamma(a+z) \simas^{z\rightarrow \infty} 
\sqrt{\frac{2\pi}{z}} z^a e^{z(\ln z-1)}\,.
\end{equation}
Furthermore, using
\begin{equation}
    \Gamma(1-b+z)\Gamma(b-z) = \frac{\pi}{\sin[\pi(b-z)]}
\end{equation}
we obtain
\begin{equation}
    \Gamma(b-z) \simas^{z\rightarrow \infty} 
\frac{1}{\sqrt{2\pi z}} z^b e^{-z(\ln z-1)} \frac{\pi}{\sin[\pi(b-z)]}\,.
\end{equation}
Then, the full integrand behaves as
\begin{equation}
\begin{split}
    \Gamma(a+z)\Gamma(b-z) x^z  \simas^{z\rightarrow \infty} &
    \sqrt{\frac{2\pi}{z}} z^a e^{z(\ln z-1)}
    \frac{1}{\sqrt{2\pi z}} z^b e^{-z(\ln z-1)} \frac{\pi}{\sin[\pi(b-z)]}
    e^{z \ln x}
\\ &=
    z^{a+b-1} \frac{\pi}{\sin[\pi(b-z)]} e^{z \ln x}\,.
\end{split}
\label{eq:exinf1}
\end{equation}
This clearly goes to zero for $z\to \infty$, provided we choose 
$\Re(z) > 0$ (i.e., we close the contour to the right) if $0<x<1$ 
and $\Re(z) < 0$ (i.e., we close the contour to the left) if $1<x$, see Problem~\ref{prob:inf1}. 
(Obviously we must avoid the poles at $z=b+n$, $n\in \mathbb{N}$.) 
Hence, as promised, we can arrange to have the integrand go to zero 
along the semi-circle at infinity if we choose the correct contour.

\begin{tips}{Specific Contours}
This example is not entirely generic, e.g. consider
\begin{equation}
    \int_{-i\infty}^{+i\infty} \frac{dz}{2\pi i}
    \Gamma(z) A^z
\end{equation}
where for $z\to -\infty$ the $\Gamma(z)$ goes to zero faster 
than $A^z$ goes to infinity for any $A$, so the contour must always 
be closed to the left. {Notice, however that this particular integral does not come from applying the basic \MB{} identity eq.~(\ref{mb1}) to some expression, since the gamma function of the form $\Gamma(\ldots + z)$ is not accompanied by another gamma function with $z$-dependence $\Gamma(\ldots - z)$.}
\end{tips}
 
{For multi-dimensional integrals, a similar analysis can be performed for each integration variable separately, taking into account all gamma functions whose argument involves the particular integration variable. In order to make this precise,} 
let us consider the generic multi-dimensional Mellin-Barnes integral {in eq.~(\ref{eq:mbmulti}),} 
\begin{equation}
    \int_{-i\infty}^{+i\infty}\ldots \int_{-i\infty}^{+i\infty}
    \prod_{j=1}^{n}\frac{dz_j}{2\pi i}
    f(z_1,\ldots,z_n,x_1,\ldots,x_p,a_1,\ldots,a_q) 
    \frac{\prod_k \Gamma(A_k + V_k)}{\prod_l \Gamma(B_l + W_l)}.
\end{equation}
Recall that in this expression the $x_j$ are fixed parameters while $A_k$ and $B_l$ are linear combinations of the parameters $a_i$. The latter take the form $a_i = n_i + b_i \epsilon$, with $n_i\in \mathbb{N}$ and $b_i\in \mathbb{R}$. Last, $V_k$ and $W_l$ are linear combinations of the integration variables $z_i$ and $f$ in practice is a product of powers of the $x_i$ with exponents that are linear combinations of $a_i$ and $z_i$ (more generally it can be an analytic function). Let us analyze the behavior of the integrand at complex infinity in the variable $z_n$. Clearly this is no loss of generality, since we may simply relabel the integration variables such that the variable of interest becomes $z_n$. Then our integral takes the form
\begin{equation}
\begin{split}
    \int_{-i\infty}^{+i\infty}\ldots &\int_{-i\infty}^{+i\infty}
    \prod_{j=1}^{n-1}\frac{dz_j}{2\pi i}
    g(z_2,\ldots,z_n,x_1,\ldots,x_p,a_1,\ldots,a_q)
    \frac{\prod'_k \Gamma(A_k + V_k)}{\prod'_l \Gamma(B_l + W_l)}
\\ \times
    &\int_{-i\infty}^{+i\infty}\frac{dz_n}{2\pi i}
    \prod_p \Gamma(C_p + z_n)^{\alpha_p}
    \prod_q \Gamma(D_q - z_n)^{\beta_q} X^{z_n}\,.
\end{split}
\label{eq:MB-gen}
\end{equation}
Above, the primes on the products $\prod'_k$ and $\prod'_l$ mean that 
only those factors for which $V_k$ and $W_l$ do not involve $z_n$ 
are included, while gamma functions whose arguments depend on $z_n$ 
are made explicit in the second line. Notice that we have separated 
the gamma functions that depend on $z_n$ as $\Gamma(\ldots+z_n)$ and 
those whose argument involves $-z_n$ as $\Gamma(\ldots-z_n)$ and $C_p$, 
$D_q$ are the corresponding linear combinations of parameters {\em and} 
other integration variables. We have moreover made explicit the 
(integer) powers $\alpha_p, \beta_q\in \mathbb{Z}$ for the gamma 
functions involving $z_n$ as well as the 
piece of the $f$ function which involves $z_n$. Thus $X$ is a 
product of the parameters $x_j$, perhaps raised to powers whose exponents 
do not depend on any other integration variable $z_j$. We note that 
technically this is not the most general expression possible, as linear 
combinations of the integration variables in the arguments of gamma 
functions can also appear where the integration variables appear with 
coefficients different from $\pm 1$ (e.g. $\Gamma(1+2z_n)$, etc.). 
However, the analysis of the behavior of the integrand at complex infinity 
that follows is easy to repeat for such a more general integrand which is considered for instance in~\cite{Gluza:2016fwh}. 

Turning finally to examining the behavior of the integrand in 
eq.~(\ref{eq:MB-gen}) at complex infinity in $z_n$, we can use the 
following asymptotic fomulae for gamma functions, discussed in 
section~\ref{sec:gamma},
\begin{equation}
    \Gamma(a+z) \simas^{z\rightarrow \infty} 
\sqrt{\frac{2\pi}{z}} z^a e^{z(\ln z-1)}\,,
\end{equation}
and 
\begin{equation}
    \Gamma(b-z) \simas^{z\rightarrow \infty} 
\frac{1}{\sqrt{2\pi z}} z^b e^{-z(\ln z-1)} \frac{\pi}{\sin[\pi(b-z)]}\,.
\end{equation}
to obtain
\begin{equation}
\begin{split}
&
    \prod_p \Gamma(C_p + z_n)^{\alpha_p}
    \prod_q \Gamma(D_q - z_n)^{\beta_q} X^{z_n} 
    \simas^{z_n\rightarrow \infty}
    (2\pi)^{\frac{1}{2}\sum_p \alpha_p - \frac{1}{2}\sum_q \beta_q}
\\ &\qquad\times
    z_n^{z_n(\sum_p \alpha_p - \sum_q \beta_q) 
    + \sum_p(C_p-\frac{1}{2})\alpha_p + \sum_q(D_q-\frac{1}{2})\beta_q}
    e^{-(\sum_p\alpha_p - \sum_q\beta_q)z_n}
\\ &\qquad\times
    \prod_q \left(\frac{\pi}{\sin[\pi(D_q-z_n)]}\right)^{\beta_q} 
    X^{z_n}\,.
\end{split}
\end{equation}
The behavior of this expression at complex infinity is controlled by 
the expression on the second line. In particular, we see that if 
the difference $\sum_p \alpha_p - \sum_q \beta_q$ is negative (recall 
$\alpha_p$ and $\beta_q$ are integers), then for large and positive 
$\Re(z)$ the expression goes to zero because of the presence of the 
factor $z_n^{-c \cdot z_n}$ with $c>0$. So in this case, we must close 
the contour ``to the right'', i.e., such that $Re(z) \to +\infty$. 
On the other hand, if $\sum_p \alpha_p - \sum_q \beta_q$ is positive 
the contour must be closed ``to the left'' such that $\Re(z) \to -\infty$. 
However, it may happen that $\sum_p \alpha_p - \sum_q \beta_q = 0$ and 
in fact, this is the typical situation in practical applications. 
\begin{tips}{Balanced Mellin-Barnes Integrals}
The Mellin-Barnes integral
\begin{equation}
    \int_{-i\infty}^{+i\infty} \frac{dz_j}{2\pi i} 
    \prod_{k=1}^{n_+} \Gamma(a_{k} + z_j)^{\alpha_{k}}
    \prod_{l=1}^{n_-} \Gamma(a_{l} - z_j)^{\beta_{l}}\,,
    \qquad
    \alpha_k\,,\beta_l \in {\mathbb Z}
\end{equation}
is \emph{balanced} in the variable $z_j$ if 
$\sum_{k=1}^{n_+} \alpha_{k} = \sum_{l=1}^{n_-} \beta_{l}$. 
\begin{svgraybox}
Informally, an integral is balanced if for each function of the 
form $\Gamma(\ldots+z_j)$ the integrand contains also a function 
of the form $\Gamma(\ldots-z_j)$, with gamma functions in the 
denominator being counted as appearing a negative number of times. 
\end{svgraybox}
A multi-dimensional Mellin-Barnes integral is balanced if it is balanced 
in each integration variable. The Mellin-Barnes representations we encounter in actual applications to Feynman integrals are typically balanced by construction. This is because the basic identity expressing the sum $(A+B)^{-\lambda}$ as a Mellin-Barnes integral,
\begin{equation}
    \frac{1}{(A+B)^\lambda} = \frac{1}{\Gamma(\lambda)}\int_{-i\infty}^{+i\infty} 
    \frac{dz}{2\pi i}
    \Gamma(\lambda+z)\Gamma(-z) A^{-\lambda-z} B^z
    \label{eq:MBformula}
\end{equation}
is obviously balanced, and it is not hard to see that repeated 
applications of this formula will also lead to balanced integrands, 
e.g.
\begin{equation}
\begin{split}
    \frac{1}{(A+B+C)^\lambda} &= \frac{1}{\Gamma(\lambda)}\int_{-i\infty}^{+i\infty} 
    \frac{dz_1}{2\pi i}
    \Gamma(\lambda+z_1)\Gamma(-z_1) A^{-\lambda-z_1} (B+C)^{z_1}
\\ &=
    \frac{1}{\Gamma(\lambda)}
    \int_{-i\infty}^{+i\infty} \frac{dz_1}{2\pi i} \frac{dz_2}{2\pi i}
    \Gamma(\lambda+z_1)\Gamma(-z_1+z_2)\Gamma(-z_2)
    A^{-\lambda-z_1} B^{z_1-z_2} C^{z_2}\,.
\end{split}
\end{equation}
The integrand above is seen to be balanced both in $z_1$ and $z_2$ 
and we call attention to the fact that one specific gamma function 
may play a role for determining if the integrand is balanced for 
more than one variable, as is the case with $\Gamma(-z_1+z_2)$ 
above. This particular gamma function must be included when 
counting both $\Gamma(\ldots-z_1)$ functions as well as
$\Gamma(\ldots+z_2)$ functions. Hence, the Mellin-Barnes integrals 
we encounter when evaluating Feynman integrals are essentially 
balanced by construction, and so typically 
$\sum_p \alpha_p - \sum_q \beta_q = 0$ in eq.~(\ref{eq:MB-gen}) in 
practical applications. Finally, we note that polygamma functions
\begin{equation}
    \psi(z) = \frac{\Gamma'(z)}{\Gamma(z)}\,,
    \qquad
    \psi^{(n)}(z) = \frac{d^n \psi(z)}{dz^n}\,,
\end{equation}
do not enter the balance counting, so an integral like
\begin{equation}
    \int_{-i\infty}^{+i\infty} \frac{dz}{2\pi i} 
    \Gamma(1+z)\Gamma(-z) \psi(2+z)
\end{equation}
is considered balanced.  
\end{tips}
For balanced integrals, i.e., when $\sum_p \alpha_p - \sum_q \beta_q = 0$ in eq.~(\ref{eq:MB-gen}), the asymptotic expression simplifies and we find
\begin{equation}
\begin{split}
&
    \prod_p \Gamma(C_p + z_n)^{\alpha_p}
    \prod_q \Gamma(D_q - z_n)^{\beta_q} X^{z_n} 
    \simas^{z_n\rightarrow \infty}
    z_n^{\sum_p(C_p-\frac{1}{2})\alpha_p + \sum_q(D_q-\frac{1}{2})\beta_q}
\\ &\qquad\times
    \prod_q \left(\frac{\pi}{\sin[\pi(D_q-z_n)]}\right)^{\beta_q} 
    X^{z_n}\,.
\end{split}
\label{eq:MB-balanced-asy}
\end{equation}
Away from the poles at $z=b+n$, $n\in \mathbb{N}$, this expression 
essentially has the form $\sim z_n^c X^{z_n}$ with some constant $c$. 
Then we find that it is the magnitude of $X$ which determines if the 
integrand vanishes at complex infinity in a particular direction. 
Assuming that $X$ is real and positive, the integrand clearly goes 
zero for $z\to \infty$, provided we choose $\Re(z) > 0$ and close 
the contour to the right if $0<X<1$, while for $1<X$ we must demand 
$\Re(z) < 0$ and close the contour to the left. (Obviously the poles 
at $z=b+n$, $n\in \mathbb{N}$ need to be avoided.) Hence, as promised, 
we can arrange to have the integrand go to zero along the semi-circle 
at infinity if we choose the correct contour. 

Before moving on, let 
us note one further complication which can arise. If both 
$\sum_p \alpha_p - \sum_q \beta_q = 0$ {\em and} $X=1$, then the integrand 
simply behaves as a power of $z_n$ at infinity. Evidently the exponent 
is given by $\sum_p(C_p-\frac{1}{2})\alpha_p + \sum_q(D_q-\frac{1}{2})\beta_q$ which may be negative, positive or 
zero (e.g. at $\epsilon=0$) and the integrand is not guaranteed to 
vanish at complex infinity. In this case one option is to introduce 
an auxiliary variable $0<X<1$, obtain the solution with this $X\ne 1$ 
and finally, take the $X\to 1$ limit. We will come back briefly to this point 
in the next section. 

\subsection{From \MB{} Integrals to Sums}
\label{ssec:MBtoSums}
Let us then assume that we have identified the proper closed contour 
for some Mellin-Barnes integration variable $z$. It is now in principle 
a simple matter to apply Cauchy's residue theorem to perform the integration 
in $z$ and obtain the result in the form of a sum over residues. 
As an example, consider the simple one-dimensional Mellin-Barnes 
integral 
\begin{equation}
    I(x) = \int_{\bar{z}-i\infty}^{\bar{z}+i\infty} \frac{dz}{2\pi i}
    \Gamma(1+z)\Gamma(-z) x^z\,.
\label{eq:Ix-example}
\end{equation}
For definiteness, let us assume that the contour runs parallel to the 
imaginary axis with constant real part that is between zero and one, 
$0< \bar{z} <1$ and that $x$ is real and positive. By our previous 
analysis, we must close the contour to the right or the left depending 
on the magnitude of $x$. For $x<1$, we must have $\Re(z) \to +\infty$ 
and so our closed contour encircles the poles of $\Gamma(-z)$ at 
$z=1,2,\ldots$. Then, using eq.~(\ref{eq:basicresgamma})
\begin{equation}
    \mathrm{Res}_{z\to n} \Gamma(-z) = - \frac{(-1)^n}{n!}\,,
    \qquad n\in \mathbb{N}
\end{equation}
we immediately find (note that our contour runs clockwise)
\begin{equation}
    I(x) = -2\pi i \sum_{n=1}^{\infty}\frac{1}{2\pi i} 
    \left[- \Gamma(1+n)\frac{(-1)^n}{n!} x^n\right]
    =
    \sum_{n=1}^{\infty} (-x)^n\,.
\label{eq:Ix-example-right}
\end{equation}
The sum is convergent for $0<x<1$ and it is of course elementary. 
We obtain
\begin{equation}
    I(x) = -\frac{x}{1+x}\,,\qquad 0<x<1\,.
\end{equation}
\begin{tips}{Elementary Sums}
We recall some elementary sums that we will make use of in the following. We begin with the simple geometric series and discuss some simple generalizations. To start recall that
\begin{equation}
    \sum_{n=0}^{\infty} x^n = \frac{1}{1-x},
\label{eq:geometric-sum}
\end{equation} 
where the sum is convergent for $|x|<1$. One simple derivation of this result proceeds to show that the $N$-th partial sum $s_N$ is simply $S_N = \frac{1+x^{N+1}}{1-x}$. Then for $|x|<1$ we find eq.~(\ref{eq:geometric-sum}) by taking the $N\to \infty$ limit. However, eq.~(\ref{eq:geometric-sum}) can be established also by considering the Taylor expansion of $f(x)=(1-x)^{-1}$ around $x=0$. Clearly $f(0) = 1$, $f'(0)=1\cdot(1-x)^{-2}|_{x=0}=1$, $f''(0)=1\cdot2\cdot(1-x)^{-2}|_{x=0}=2!$ and so on, hence $f^{(n)}(0) = 1\cdot 2 \cdot \ldots \cdot n(1-x)^{-n-1}|_{x=0} = n!$. Thus
\begin{equation}
    (1-x)^{-1} = 1 + x + \frac{2!}{2!}x^2 + \ldots = \sum_{n=0}^{\infty} x^n.
\end{equation}

The advantage of this second method is that it allows to consider some immediate generalizations. In particular, consider first $f(x)=(1-x)^{-m}$ for some positive integer $m$. In this case $f(0)=1$, $f'(0) = m(1-x)^{-m-1}|_{x=0} = m$, $f''(0) = m(m+1)(1-x)^{-m-2}|_{x=0} = m(m+1)$ and in general $f^{(n)}(0) = m(m+1)\cdots(m+n-1)(1-x)^{-m-n}|_{x=0} = m(m+1)\cdots(m+n-1)$. Using the ratio test, one can show that the Taylor series converges for $|x|<1$ and thus
\begin{equation}
    (1-x)^{-m} = 1 + m x + \frac{m(m+1)}{2!} x^2 + \ldots 
    = \sum_{n=0}^{\infty} \frac{m(m+1)\cdots(m+n-1)}{n!} x^n
\end{equation}
so we find the summation formula
\begin{equation}
    \sum_{n=0}^{\infty} \frac{(m+n-1)!}{n!(m-1)!} x^n = \frac{1}{(1-x)^m}.
\label{eq:powersum1}
\end{equation}

We may also apply the same arguments for fractional powers. So let us set $\alpha = \frac{p}{q}$ where $p$ and $q$ are positive integers and examine the Taylor expansion of the function $f(x)=(1+x)^\alpha$ around $x=0$. We see that $f(0)=1$, while $f'(0) = \alpha (1+x)^{\alpha-1}|_{x=0} = \alpha$, $f''(0) = \alpha(\alpha-1) (1+x)^{\alpha-2}|_{x=0} = \alpha(\alpha-1)$ and in general $f^{(n)}(0) = \alpha(\alpha-1)\cdots(\alpha-n+1) (1+x)^{\alpha-n}|_{x=0} = \alpha(\alpha-1)\cdots(\alpha-n+1)$. The convergence of the Taylor series for $|x|<1$ can be established with the ratio test. Thus
\begin{equation}
    (1+x)^\alpha = 1 + \alpha x + \frac{\alpha(\alpha-1)}{2!} x^2 + \dots 
    = \sum_{n=0}^{\infty} \frac{\alpha\cdot(\alpha-1)\cdot\ldots\cdot(\alpha-n+1)}{n!} x^n.
\label{eq:frac-binom-thrm}
\end{equation}
Introducing the \emph{generalized binomial coefficient} $\binom{\alpha}{n}$,
\begin{equation}
    \binom{\alpha}{n} = \frac{\alpha\cdot(\alpha-1)\cdot\ldots\cdot(\alpha-n+1)}{n!},
\end{equation}
we can write
\begin{equation}
    \sum_{n=0}^{\infty}  \binom{\alpha}{n} x^n = (1+x)^\alpha.
\label{eq:powersum2}
\end{equation}
This result is known as the \emph{binomial theorem for fractional exponents}.

We note finally that all of the sums above are special cases of the Gauss hypergeometric sum introduced in section~\ref{eq:hyper1}. In particular, they are all just realizations of the general formula
\begin{equation}
    {}_2F_1(-\alpha,\beta,\beta,-z) = (1+z)^\alpha\qquad
    \beta\mbox{ arbitrary}.
\end{equation}

Finally, let us recall one more well-known but useful trick for evaluating sums where some positive integer power of the index of summation appears in the summand. Consider e.g.,
\begin{equation}
    \sum_{n=0}^{\infty} n^k x^n,\qquad k\in \mathbb{N}^+.
\label{eq:nk-geometric-sum}
\end{equation}
One approach to dealing with this sum is to notice that formally
\begin{equation}
    x \frac{d}{dx}\frac{1-x} = x \frac{d}{dx} \sum_{n=0}^{\infty} x^n
    =
    x \sum_{n=0}^{\infty} n x^{n-1},
\end{equation}
and so 
\begin{equation}
    \sum_{n=0}^{\infty} n x^n = \frac{x}{(1-x)^2}.
\label{eq:n1-geometric-sum}
\end{equation}
Obviously this operation can be iterated $k$ times once the specific value of $k$ is known to compute the sum in eq.~(\ref{eq:nk-geometric-sum}). The same trick can be useful also for the generalized sums discussed above.
\end{tips}
Continuing with our example, for $1<z$ we must close the contour to the left. 
Then the contour encircles the poles of $\Gamma(1+z)$ at $z=-1,-2,\ldots$ 
{\em and} the pole of $\Gamma(-z)$ at $z=0$
Using
\begin{equation}
    \mathrm{Res}_{z\to -n} \Gamma(1+z) = \frac{(-1)^{n-1}}{(n-1)!}\,,
    \qquad n\in \mathbb{N}^+\,,
\end{equation}
we obtain (notice that the contour now runs counter-clockwise)
\begin{equation}
    I(x) = 2\pi i \left\{\sum_{n=1}^{\infty}\frac{1}{2\pi i}
    \left[\frac{(-1)^{n-1}}{(n-1)!} \Gamma(n) x^{-n}\right]
    -\Gamma(1+0)x^0\right\}
    =
    -\sum_{n=1}^{\infty} (-x)^{-n} - 1\,.
\end{equation}
The sum is convergent now for $1<x$ and is again elementary. We find
\begin{equation}
    I(x) = -\frac{x}{1+x}\,,\qquad 1<x\,.
\end{equation}
So in this case, we indeed obtain the same functional form for 
both $0<x<1$ and $1<x$. On the one hand, this should not come as 
a surprise. Indeed, our starting Mellin-Barnes integral only differed 
form the basic Mellin-Barnes formula of eq.~(\ref{mb1}) by the position 
of the contour: it did not separate all poles of $\Gamma(1+z)$ and 
$\Gamma(-z)$ as the pole of $\Gamma(-z)$ at $z=0$ was on the wrong 
side of the contour. So we could have evaluated the integral simply 
by shifting the contour to some position between $-1$ and $0$ and 
accounting for the residue of the pole at $z=0$. So letting 
$\bar{z}' \in (-1,0)$, we compute for all $0<x$
\begin{equation}
\begin{split}
    I(x) &= \int_{\bar{z}-i\infty}^{\bar{z}+i\infty} \frac{dz}{2\pi i}
    \Gamma(1+z)\Gamma(-z) x^z
\\ &=
    \int_{\bar{z}'-i\infty}^{\bar{z}'+i\infty} \frac{dz}{2\pi i}
    \Gamma(1+z)\Gamma(-z) x^z - 1
\\ &=
    \Gamma(1)(1+x)^{-1} - 1 = -\frac{x}{1+x}\,,\qquad 0<x\,.
\end{split}\label{int1dim1}
\end{equation}
On the other hand, this phenomena of obtaining the same functional 
form of the solution for both $x<1$ and $x>1$ is not generic as the next 
example demonstrates.

%%%% PITFALS example!
\begin{tips}{Regions of Validity of the Solution}
To 
illustrate this, we consider the Mellin-Barnes integral
\begin{equation}
    I(x) = \int_{\bar{z}-i\infty}^{\bar{z}+i\infty} \frac{dz}{2\pi i}
    \frac{\Gamma(-z)}{\Gamma(3-z)} x^z\,,
\end{equation}
where again let us assume that the contour is parallel to the 
imaginary axis and $\bar{z} $ is between zero and one. Applying our 
analysis regarding the choice of contours, we see that once more we 
must close the contour to the right if $0<x<1$ and to the left if 
$1<x$. In the former case, the closed contour, running clockwise, 
encircles the poles of the integrand at $z=1$ and $z=2$. Notice 
that in this case there are in fact no other poles inside the contour, 
since 
\begin{equation}
    \frac{\Gamma(-z)}{\Gamma(3-z)} = \frac{1}{(2-z)(1-z)(-z)} 
    = -\frac{1}{2(z-2)} + \frac{1}{z-1} - \frac{1}{2z}\,.
\end{equation}
It is also trivial to read off the residues at the poles,
\begin{equation}
    \mathrm{Res}_{z\to 1}\frac{\Gamma(-z)}{\Gamma(2-z)} = 1\,,\qquad
    \mathrm{Res}_{z\to 2}\frac{\Gamma(-z)}{\Gamma(2-z)} = -\frac{1}{2}
\end{equation}
hence, the sum over poles is finite and we simply have
\begin{equation}
    I(x) = - \sum_{n=1}^{2} 1 \cdot x^n = -x + \frac{x^2}{2}\,,\qquad 0<x<1\,.
\end{equation}    
On the other hand, when $1<x$, the contour runs counter-clockwise 
and encircles the single pole at $z=0$ with residue
\begin{equation}
    \mathrm{Res}_{z\to 0}\frac{\Gamma(-z)}{\Gamma(2-z)} = -\frac{1}{2}\,.
\end{equation}
Then the sum over poles involves just a single term and we find
\begin{equation}
    I(x) = \sum_{n=0}^{0} \left(-\frac{1}{2}\right) \cdot x^n 
    = -\frac{1}{2}\,,\qquad 1<x\,.
\end{equation}
So finally we have
\begin{equation}
    I(x) = \begin{cases}
        -x + \frac{x^2}{2} & \text{if $0<x<1$} \\ &\\
        -\frac{1}{2} & \text{if $1 \le x$}
    \end{cases}\,.
\end{equation}
The lesson to take away form these simple examples is that one must be 
mindful of both the direction in which the contour is closed as well as 
the range over which the summation must be taken. 

In the case of multi-dimensional Mellin-Barnes integrals, closing the countours for each integration variable in different directions can lead to different multiple sum representations of the result, each convergent in different ranges of the variables, see ref.~\cite{Ananthanarayan:2020fhl} for further details.
\end{tips}
Finally, let us briefly comment on the case of $x=1$. First, notice that if $x=1$ in 
eq.~(\ref{eq:Ix-example}), then we are precisely in the situation discussed below 
eq.~(\ref{eq:MB-balanced-asy}): $\sum_p(C_p-\frac{1}{2})\alpha_p + \sum_q(D_q-\frac{1}{2})\beta_q = (1-\frac{1}{2})^1 + (0-\frac{1}{2})^1 = 0$ and $X=1$. So it is not evident 
how the contour should be closed. In fact, we can also see that this case requires 
special care since if we proceed naively and close the contour, say, to the right, we 
obtain (just set $x=1$ in eq.~(\ref{eq:Ix-example-right}))
\begin{equation}
I(1) = \sum_{n=1}^{\infty} (-1)^n.
\end{equation}
Clearly this sum does not converge! Had we closed the contour to the left, the same 
divergent sum would have appeared. In such situations, one option is in fact to compute 
$I(x)$ instead of $I(1)$ directly and take the $x\to 1$ limit. In our case, we find 
simply
\begin{equation}
I(1) = \lim_{x\to 1} I(x) = \lim_{x\to 1} -\frac{x}{1+x} = -\frac{1}{2}.
\end{equation}
%\begin{svgraybox}
The correctness of this result can be checked simply by evaluating the integral 
numerically. This is easily done in \math{}
\begin{minted}[frame=single,breaklines,fontsize=\small]{mathematica}
In[1]:= NIntegrate[I*1/(2 Pi I)*Gamma[1 + z] Gamma[-z] /. {z -> 1/2 + I t}, {t, -Infinity, Infinity}]
Out[1]:= -0.5 + 0. I
\end{minted}
The extra factor of $i$ is the Jacobian associated with the change of integration 
variable $z \to 1/2 + i t$ which we have used to integrate along the correct contour.
%\end{svgraybox}

In order to highlight how these considerations apply in the 
multi-dimensional case, let us look at the following example. Consider 
\begin{equation}
    I(x,y) = 
    \int_{\bar{z}_1-i\infty}^{\bar{z}_1+i\infty}
    \int_{\bar{z}_2-i\infty}^{\bar{z}_2+i\infty}
    \frac{dz_1}{2\pi i}\frac{dz_2}{2\pi i}
    \Gamma(1+z_1)\Gamma(1+z_2)\Gamma(-z_1-z_2) x^{z_1} y^{z_2}
\end{equation}
Let us assume that both $\bar{z}_1$ and $\bar{z}_2$ are between 
zero and one and that $\bar{z}_1 + \bar{z}_2 < 1$. Consider the 
situation when both $x$ and $y$ are positive but smaller than one, 
$0<x,y<1$. Then, starting with the $z_1$ integration, we have that the 
contour in this variable must be closed to the right and will thus 
encircle the poles of $\Gamma(-z_1-z_2)$, which lie at $z_1 = n-z_2$, 
$n\in\mathbb{N}$. However, notice that not all of these poles are inside 
the contour! In fact, the pole at $z_1 = -z_2$, corresponding to $n=0$ 
has a negative real part (since $\bar{z}_2 > 0$) and hence is not inside 
the closed contour. The next pole at $z_1 = 1-z_2$ (corresponding to $n=1$) 
is inside the contour only if $\Re(z_1) < 1-\Re(z_2)$ and we see now 
the significance of the condition $\bar{z}_1 + \bar{z}_2 < 1$ above. 
The rest of the poles at $z_1 = n - z_2$ for $n\ge 2$ are all inside the 
contour since $\Re(z_1) < n-\Re(z_2)$ is clearly satisfied for 
$0<\bar{z}_1,\bar{z}_2<1$ and $n\ge 2$. Then, using
\begin{equation}
\begin{split}
    \mathrm{Res}_{z_1 \to n - z_2}&
    \left[\Gamma(1+z_1)\Gamma(1+z_2)\Gamma(-z_1-z_2) x^{z_1} y^{z_2}\right]
    = 
\\&=
\Gamma(1+n-z_2)\Gamma(1+z_2)\frac{-(-1)^n}{n!}x^{n-z_2} y^{z_2}
\end{split}
\end{equation}
we can perform the integration for $z_1$ and obtain
\begin{equation}
\begin{split}
I(x,y) &= \int_{\bar{z}_2-i\infty}^{\bar{z}_2+i\infty}
    \frac{dz_2}{2\pi i} (-2\pi i) \frac{1}{2\pi i}
    \sum_{n_1=1}^{\infty}\Gamma(1+n_1-z_2)\Gamma(1+z_2)
    \frac{-(-1)^{n_1}}{n_1!}x^{n_1-z_2} y^{z_2}
\\ &=
    \int_{\bar{z}_2-i\infty}^{\bar{z}_2+i\infty}
    \frac{dz_2}{2\pi i}
    \sum_{n_1=1}^{\infty}\Gamma(1+n_1-z_2)\Gamma(1+z_2)
    \frac{(-x)^{n_1}}{n_1!} \left(\frac{y}{x}\right)^{z_2}\,.
\end{split}
\end{equation}
Notice how the summation over $n_1$ runs from one to infinity, as 
explained above. Next, we would like to perform the integration over 
$z_2$ by using the residue theorem. Examining the behaviour of the 
integrand at infinity, we see that we must close the integration 
contour to the right if $0<y/x<1$ or to the left if $1< y/x$. Let 
us assume the former for the sake of example, so that the closed 
contour encircles the residues of $\Gamma(1+n_1-z_2)$ at 
$z_2 = 1+n_1+n_2$ with $n_2 \in \mathbb{N}$. In order to determine 
the range of summation in $n_2$, we note that $1+n_1\ge 2$, since 
$n_1\ge 1$ and so each pole of $\Gamma(1+n_1-z_2)$ lies inside the 
closed contour since $Re(z_2)<1$. The residues at $z_2 = 1+n_1+n_2$ 
are easy to compute and we find
\begin{equation}
\begin{split}
    \mathrm{Res}_{z_2 \to 1+n_1+n_2}&
    \left[
    \Gamma(1+n_1-z_2)\Gamma(1+z_2)
    \frac{(-x)^{n_1}}{n_1!} \left(\frac{y}{x}\right)^{z_2}
    \right]
\\ &=
    \frac{-(-1)^{n_2}}{n_2!} \Gamma(2+n_1+n_2) \frac{(-x)^{n_1}}{n_1!}
    \left(\frac{y}{x}\right)^{1+n_1+n_2}\,.
\end{split}
\end{equation}
Then the integration over $z_2$ gives the result
\begin{equation}
\begin{split}
    I(x,y) &= -(2\pi i) \frac{1}{2\pi i} \sum_{n_2=0}^{\infty}
    \sum_{n_1=1}^{\infty}  
    \frac{-(-1)^{n_2}}{n_2!} \Gamma(2+n_1+n_2) \frac{(-x)^{n_1}}{n_1!}
    \left(\frac{y}{x}\right)^{1+n_1+n_2}
\\ &=
    \frac{y}{x}\sum_{n_1=1}^{\infty} \sum_{n_2=0}^{\infty}
    \frac{\Gamma(n_1+n_2+2)}{\Gamma(n_1+1)\Gamma(n_2+1)}(-y)^{n_1}
    \left(-\frac{y}{x}\right)^{n_2}\,,
    \qquad 0<y<x<1\,.
\end{split}
\label{eq:simple-double-sum}
\end{equation}
%\commjg{can this be example with summing up or in miscellanous file?}
In the last step, we have exchanged the order of the summations 
which we are allowed to do if the sums are absolutely convergent. 
We will not pursue this example any further, as hopefully it is clear 
how one would derive similar multiple sum representations of for the 
cases we have not considered explicitly, e.g. for $0<x<y<1$ and so on.

\begin{tips}{Evaluating Sums with \math{}}
The double sum in eq.~(\ref{eq:simple-double-sum}) is actually elementary and can be computed in \math{} immediately, see \wwwaux{basicsums}
\begin{minted}[frame=single,breaklines,fontsize=\small]{mathematica}
In[2]:= summand = y/x Gamma[n1+n2+2]/Gamma[n1+1]/Gamma[n2+1] (-y)^n1 (-y/x)^n2;
In[3]:= sum = Sum[summand, {n2, 0, Infinity}]
Out[3]:= (x (-y)^n1 y ((x+y)/x)^-n1 Gamma[2+n1])/((x+y)^2 Gamma[1+n1])
In[4]:= sum = Sum[sum, {n1, 1, Infinity}]
Out[4]:= -((x^2 y^2 (2 x+2 y+x y))/((x+y)^2 (x+y+x y)^2))
\end{minted}
It is simple to understand both steps of the computation. First, consider the inner summation over $n_2$ for some fixed $n_1 \in \mathbb{N}^+$. This sum can easily be brought to the form of the elementary sum in eq.~(\ref{eq:powersum1}),
\begin{equation}
\begin{split}
    \sum_{n_2=0}^{\infty}
    \frac{\Gamma(n_1+n_2+2)}{\Gamma(n_1+1)\Gamma(n_2+1)} \left(-\frac{y}{x}\right)^{n_2} 
    &=
    \sum_{n_2=0}^{\infty}
    \frac{(n_1+n_2+1)!}{n_1! n_2!} \left(-\frac{y}{x}\right)^{n_2}    
\\&=
    (n_1+1)\sum_{n_2=0}^{\infty}
    \frac{(n_1+2+n_2-1)!}{(n_1+2-1)! n_2!} \left(-\frac{y}{x}\right)^{n_2} 
\\&=
    (n_1+1)\left(1+\frac{y}{x}\right)^{-n_1-2},
\end{split}
\end{equation}
where in the last line we used eq.~(\ref{eq:powersum1}) with $m=n_1+2$ and $n=n_2$. The final summation in $n_1$ is then essentially a geometric sum, with the generalization that $n_1$ also appears in the numerator of the summand, as in eq.~(\ref{eq:nk-geometric-sum}). We have seen that such sums are easy to evaluate and we can use eqs.~(\ref{eq:geometric-sum})~and~(\ref{eq:n1-geometric-sum}) to find the final result (notice though, that the $n_1$ summation runs from $n_1=1$, not $n_1=0$).
\end{tips}

The previous examples make it clear that computing the Mellin-Barnes 
integrals by use of the residue theorem generally leads to multiple sums 
involving gamma functions as well as powers of parameters. Thus, our next 
task is to develop tools for the evaluation of such sums. However, before 
we turn to this, let us examine one further example which will serve 
also to motivate the class of sums we will examine. Thus, consider the 
following Mellin-Barnes integral,
\begin{equation}
    I(x) = \int_{\bar{z}-i\infty}^{\bar{z}+i\infty}\frac{dz}{2\pi i}
    \Gamma(z)^2 \Gamma(-z)^2 x^z\,,
\end{equation}
with $0<\bar{z}<1$ and $0<x<1$. We can perform the integration using 
Cauchy's theorem by closing the contour to the right, which thus 
encircles the poles of the integrand at $z = n$ for $n=1,2,\ldots$. 
Since $0<\bar{z}<1$, the pole at $z=0$ is not inside the contour. 
We can compute the residue and we find
\begin{equation}
    \mathrm{Res}_{z \to n}
    \left[\Gamma(z)^2 \Gamma(-z)^2 x^z\right]
    =
    \frac{x^n}{n^2}\left[
    \ln x + 2 \psi^{(0)}(n) - 2 \psi^{(0)}(n+1)
    \right]\,.
\end{equation}f
We will explain one approach to computing such symbolic residues in detail in section~\ref{sect:1dmb-sol-with-sums} below, however in many situations they are simple to obtain also with \math{},
\begin{minted}[frame=single,breaklines,fontsize=\small]{mathematica}
In[5]:= Assuming[n > 0 && Element[n, Integers], 
 Residue[Gamma[z]^2 Gamma[-z]^2 x^z, {z, n}]]
Out[5]:= ((-1)^(-2 n) x^n Gamma[n]^2 (Log[x]+2 PolyGamma[0,n] - 2 PolyGamma[0,1+n]))/(n!)^2
\end{minted}
Since the contour runs clockwise we find
\begin{equation}
    I(x) = -(2\pi i) \frac{1}{2\pi i}\sum_{n=1}^{\infty}
    \frac{x^n}{n^2}\left[
    \ln x + 2 \psi^{(0)}(n) - 2 \psi^{(0)}(n+1)
    \right]\,.
\end{equation}
From eq.~(\ref{eq:psi0N}) we see that the polygamma function can be written as
\begin{equation}
    \psi^{(0)}(n) = -\gamma + \sum_{k=1}^{n-1}\frac{1}{k}\,,
\label{eq:psi0_sumrep}
\end{equation}
and so we obtain the following double sum representation of the integral
\begin{equation}
    I(x) = -(2\pi i) \frac{1}{2\pi i}\sum_{n=1}^{\infty}
    \frac{x^n}{n^2}\left[
    \ln x + 2\sum_{k=1}^{n}\frac{1}{k} 
    - 2\sum_{k=1}^{n-1}\frac{1}{k}
    \right]\,.
\end{equation}
Notice that this representation involves {\em nested sums}: the upper 
limit of the inner summation over $k$ depends on the index of the outer 
summation over $n$. For completeness, we note that in this particular 
case, the nesting of the sums is illusory, since using 
$\psi^{(0)}(n+1) = \psi^{(0)}(n) + \frac{1}{n}$, which follows immediately 
from eq.~(\ref{eq:psi0_sumrep}), and we obtain simply
\begin{equation}
    I(x) = -\sum_{n=1}^{\infty}
    \frac{x^n}{n^2}\left(
    \ln x - \frac{2}{n}
    \right)\,.
\label{eq:Iex-Zsums}
\end{equation}
The sums that appear here are no longer elementary in the sense that 
the results can no longer be expressed with elementary functions such 
as rational functions and logarithms. Hence, we must develop tools and 
techniques to evaluate such sums. Furthermore, we emphasize that 
in more elaborate applications, evaluating Mellin-Barnes integrals via 
the residue theorem produces genuine multiple nested sums. Thus, next 
we turn our attention to a class of such sums.

\subsection{Z- and S- Nested Sums}
\label{ssec:nestedsums}

Motivated by the previous examples, we define a class of functions 
called {\em Z-sums} by the following recursive formula~\cite{Moch:2001zr}
\begin{equation}
\begin{split}
Z(n) &= \begin{cases}
        1 & \text{if $n\ge 0$}\,, \\
        0 & \text{if $n < 0$}\,,
    \end{cases}      
\\
Z(n;m_1,\ldots,m_k;x_1,\ldots,x_k) &= 
    \sum_{i=1}^{n} \frac{x_1^i}{i^{m_1}} 
    Z(i-1;m_2,\ldots,m_k;x_2,\ldots,x_k)\,.
\end{split}
\label{eq:Z-sum-def}
\end{equation}
Here $k$ is called the {\em depth} of the sum, while 
$w = m_1 + \ldots + m_k$ is the {\em weight} of the sum. Unfolding 
the recursion, one obtains the equivalent definition
\begin{equation}
    Z(n;m_1,\ldots,m_k;x_1,\ldots,x_k) =
        \sum_{n\ge i_1 > i_2 > \ldots > i_k > 0}
        \frac{x_1^{i_1}}{i_1^{m_1}} \ldots \frac{x_k^{i_k}}{i_k^{m_k}}\,.
\label{eq:Z-sum-def2}
\end{equation}
We note that the upper limit of summation, $n$, can be infinity.
In much the same way, we also define {\em S-sums} by
\begin{equation}
\begin{split}
S(n) &= \begin{cases}
        1 & \text{if $n > 0$}\,, \\
        0 & \text{if $n \le 0$}\,,
    \end{cases}      
\\
S(n;m_1,\ldots,m_k;x_1,\ldots,x_k) &= 
    \sum_{i=1}^{n} \frac{x_1^i}{i^{m_1}} 
    S(i;m_2,\ldots,m_k;x_2,\ldots,x_k)\,.
\end{split}
\label{eq:S-sum-def}
\end{equation}
Once again, we can unfold the recursion to arrive at the equivalent 
definition
\begin{equation}
    S(n;m_1,\ldots,m_k;x_1,\ldots,x_k) =
        \sum_{n\ge i_1 \ge i_2 \ge \ldots \ge i_k \ge 0}
        \frac{x_1^{i_1}}{i_1^{m_1}} \ldots \frac{x_k^{i_k}}{i_k^{m_k}}\,.
\label{eq:S-sum-def2}
\end{equation}
Again, $n=\infty$ is allowed. The difference between eqs.~(\ref{eq:Z-sum-def2})~and~(\ref{eq:S-sum-def2}) is only in the upper summation boundary, see the discussion of the property {\bf S-Z1} further in this section. 

$Z$- and $S$-sums define a rather general class of functions which 
includes as special cases classical polylogarithms $\mathrm{Li}_m(x)$, 
Nielsen's polylogarithms $S_{n,p}(x)$, the harmonic polylogarithms of 
Remiddi and Vermaseren $H_{m_1,\ldots,m_k}(x)$, as well as Goncharov's 
multiple polylogarithms $G(a_1,\ldots,a_w;z)$ discussed in chapter~\ref{chapter:complex} (see in particular sections~\ref{sec:Li-n}~and~\ref{sec:MPLs}). Objects such as multiple 
$\zeta$-values $\zeta(m_1,\ldots,m_k)$, Euler-Zagier sums 
$Z_{m_1,\ldots,m_k}(n)$ and harmonic sums $S_{m_1,\ldots,m_k}(n)$, 
which are perhaps less than well-known are also special cases of $Z$- and 
$S$-sums~\cite{Moch:2001zr}.  

For $n=\infty$ $Z$-sums are by definition the multiple polylogarithms 
of Goncharov, introduced in eq.~(\ref{eq:Li-def})
\begin{eqnarray}
    Z(\infty;m_1,\ldots,m_k;x_1,\ldots,x_k) = 
    \mathrm{Li}_{m_k,\ldots,m_1}(x_k,\ldots,x_1)\,.
    \label{eq:Z-sum-to-Li}
\end{eqnarray}
Notice that the indices and arguments on the right hand side are 
reversed with respect to the left hand side. 
In particular the classical polylogarithms of order $n$ are simply the 
depth-one $Z$-sums to infinity,
\begin{equation}
     Z(\infty;m,x) = \sum_{i=1}^{\infty}\frac{x^i}{i^m} =
     \mathrm{Li}_m(x).
\label{eq:Z1-to-LIm}
\end{equation}

We next turn to developing the basic tools that are useful for performing 
the summation of residues when evaluating Mellin-Barnes integrals, but 
first let us remark that already at this stage we can write a symbolic 
solution for the example in eq.~(\ref{eq:Iex-Zsums}). Indeed, we have simply
\begin{equation}
    I(x) = -\sum_{n=1}^{\infty}
    \frac{x^n}{n^2}\left(
    \ln x - \frac{2}{n}
    \right)
    =
    -\ln x\, Z(\infty;2;x) + 2 Z(\infty;3,x)\,.
\end{equation}
Since depth-ones sums to infinity simply correspond to the 
classical polylogarithms, we have
\begin{equation}
    I(x) = -\sum_{n=1}^{\infty}
    \frac{x^n}{n^2}\left(
    \ln x - \frac{2}{n}
    \right)
    =
    -\ln x\, \mathrm{Li}_2(x) + 2 \mathrm{Li}_3(x)\,.
\end{equation}

Of course, in most real applications, we will not simply obtain sums 
that can be instantly recognised to be a $Z$-sum or an $S$-sum. Hence, 
we must study the properties of these functions and work out algorithms 
for expressing sums over residues in terms of them, when possible.
To do so, we begin by describing some basic operations on $Z$- and 
$S$-sums. 
\begin{svgraybox}
The four essential properties we will need are the 
following:
\begin{enumerate}
    \item[{\bf S-Z1}] $Z$- and $S$-sums define the same class of functions 
    and can be converted into each other. Therefore, we will 
    mostly work with $Z$-sums in the following.
    \item[{\bf S-Z2}] $Z$- and $S$-sums form an algebra. This means that 
    (finite) products of $Z$- and $S$-sums can be written as 
    linear combinations of $Z$- and $S$-sums of increased 
    depth.
    \item[{\bf S-Z3}] Sums with (positive) offsets can be algebraically reduced 
    to regular $Z$- and $S$-sums. Offsets here refer either to 
    sums with an upper limit of summation of the form $n+c$, where $c>0$ 
    is a fixed integer, or to denominators of the type $(i+c)^m$ in 
    sums.
    \item[{\bf S-Z4}] $Z$- and $S$-sums appear as expansion coefficients when 
    expanding gamma-functions around integer values.
\end{enumerate}
\end{svgraybox}

First, the fact that $Z$- and $S$-sums can be converted into one 
another is quite straightforward, since the two definitions differ 
only in the upper summation boundary of the nested sums: $(i-1)$ 
for $Z$-sums and $i$ for $S$-sums. Hence sums of depth one ($k=1$) 
simply coincide,
\begin{equation}
    S(n;m_1;x_1) = Z(n;m_1;x_1)\,,
\label{eq:SZ-conv-depth1}
\end{equation}
both being equal to $\sum_{i_1} \frac{x_1^{i_1}}{i_1^{m_1}}$. For depths 
greater than one, the conversion is in principle a straightforward, if 
somewhat tedious computation, e.g.,
\begin{equation}
\begin{split}
    S(n;m_1,m_2;x_1,x_2) &= 
    \sum_{i_1=1}^{n} \frac{x_1^{i_1}}{i_1^{m_1}} S(i_1;m_2;x_2) =
    \sum_{i_1=1}^{n} \frac{x_1^{i_1}}{i_1^{m_1}}
    \sum_{i_2=1}^{i_1} \frac{x_2^{i_2}}{i_2^{m_2}}
\\ &=
    \sum_{i_1=1}^{n} \frac{x_1^{i_1}}{i_1^{m_1}}
    \left[
    \sum_{i_2=1}^{i_1-1} \frac{x_2^{i_2}}{i_2^{m_2}}
    + \frac{x_2^{i_1}}{i_1^{m_2}}
    \right]
\\ &=
    \sum_{i_1=1}^{n} \frac{x_1^{i_1}}{i_1^{m_1}} Z(i_1-1;m_2;x_2) + 
    \sum_{i_1=1}^{n} \frac{(x_1 x_2)^{i_1}}{i_1^{m_1+m_2}}
\\ &=
    Z(n;m_1,m_2;x_1,x_2) + Z(n;m_1+m_2;x_1 x_2)\,.
\end{split}
\label{eq:SZ-conv-depth2}
\end{equation}
In general, the conversion between $Z$- and $S$-sums can be performed 
in a recursive fashion using an easy generalization of the computation 
above. The details of this algorithm are discussed in~\cite{Moch:2001zr}. 
Thus one can easily convert $Z$-sums to $S$-sums and vice versa. {\it However, 
some properties are more naturally expressed in terms of $Z$-sums while 
others in terms of $S$-sums. Thus, we prefer to keep both definitions 
and convert between them when and if convenient.}

Second, $Z$- and $S$-sums form an algebra. This simply means that the 
product of two $Z$-sums (or $S$-sums) with the same upper limit of 
summation $n$ can be expressed as a linear combination of single $Z$-sums 
(or $S$-sums) of the same upper limit of summation. In order to get a 
feeling for why this is true, let us try and express the product of 
$Z(n;m_1;x_1)$ and $Z(n;m_2;x_2)$ as a single $Z$-sum. We begin by 
using the definition of $Z$-sums to write
\begin{equation}
Z(n;m_1;x_1) Z(n;m_2;x_2) =
    \sum_{i=1}^{n} \frac{x_1^{i}}{i^{m_1}} Z(n)
    \sum_{j=1}^{n} \frac{x_2^{j}}{j^{m_2}} Z(n)\,.
\end{equation}
Now, let us rewrite the double summation over $i$ and $j$ as three 
terms using the relation
\begin{equation}
    \sum_{i=1}^{n}\sum_{j=1}^{n} a_{ij} = 
    \sum_{i=1}^{n}\sum_{j=1}^{i-1} a_{ij}
    + \sum_{j=1}^{n}\sum_{i=1}^{j-1} a_{ij}
    + \sum_{i=1}^{n} a_{ii}\,,
\end{equation}
which leads to
\begin{equation}
\begin{split}
Z(n;m_1;x_1) Z(n;m_2;x_2) =&
    \sum_{i=1}^{n}\sum_{j=1}^{i-1} 
    \frac{x_1^{i}}{i^{m_1}} \frac{x_2^{j}}{j^{m_2}} [Z(n)]^2
    + \sum_{j=1}^{n}\sum_{i=1}^{j-1} 
    \frac{x_1^{i}}{i^{m_1}} \frac{x_2^{j}}{j^{m_2}} [Z(n)]^2
\\ &
    + \sum_{i=1}^{n}
    \frac{(x_1 x_2)^{i}}{i^{m_1+m_2}} [Z(n)]^2\,.
\end{split}
\end{equation}
Next, we use $[Z(n)]^2=Z(n)$, since $Z(n)=1$ for $n\ge 0$, 
and the definition of $Z$-sums in eq.~(\ref{eq:Z-sum-def}) 
to obtain
\begin{equation}
\begin{split}
Z(n;m_1;x_1) Z(n;m_2;x_2) =&
    \sum_{i=1}^{n} \frac{x_1^{i}}{i^{m_1}} 
    \sum_{j=1}^{i-1} \frac{x_2^{j}}{j^{m_2}} Z(n)
    + \sum_{j=1}^{n} \frac{x_2^{j}}{j^{m_2}} 
    \sum_{i=1}^{j-1} \frac{x_1^{i}}{i^{m_1}} Z(n)
\\ &
    + \sum_{i=1}^{n}
    \frac{(x_1 x_2)^{i}}{i^{m_1+m_2}} Z(n)
\\ =&
    Z(n;m_1,m_2;x_1,x_2) 
    + Z(n;m_2,m_1;x_2,x_1)
\\ &    
    + Z(n;m_1+m_2;x_1 x_2)\,.
\end{split}
\end{equation}
As promised, we have expressed the product of two $Z$-sums as a 
linear combination of single $Z$-sums. The decomposition of a product 
of two general $Z$-sums into single $Z$-sums can be performed using a 
recursion on the depth of the sums. This is discussed further in~\cite{Moch:2001zr}.
 
Third, for any fixed positive integer $c$, the $Z$-sum $Z(n+c-1,\ldots)$ 
and the $S$-sum $S(n+c,\ldots)$ can be related to $Z(n-1,\ldots)$ and 
$S(n,\ldots)$. To see this, let us consider e.g. $Z(n+2-1;m_1,m_2;x_1,x_2)$. 
Start by resolving the outermost summation and splitting it as follows:
\begin{equation}
\begin{split}
    Z(n+2-1;m_1,m_2;x_1,x_2) &= \sum_{i=1}^{n-1} \frac{x_1^i}{i^{m_1}} 
    Z(i-1;m_2;x_2) 
    + \frac{x_1^n}{n^{m_1}} Z(n-1;m_2;x_2)
\\ &
    + \frac{x_1^{n+1}}{(n+1)^{m_1}} Z(n+1-1;m_2;x_2)\,.
\end{split}
\label{eq:Z-shift}
\end{equation}
The first term is simply $Z(n-1;m_1,m_2;x_1,x_2)$ by definition, while 
the second term is already expressed in term of a $Z$-sum to $n-1$. 
Moreover, in the third term the depth of the remaining $Z$-sum is reduced. 
So repeating the above idea, we find
\begin{equation}
    Z(n+1-1;m_2;x_2) = \sum_{j=1}^{n-1}\frac{x_2^j}{j^{m_2}} 
    + \frac{x_2^n}{n^{m_2}} 
    = Z(n-1;m_2;x_2) + \frac{x_2^n}{n^{m_2}}\,. 
\label{eq:Z-shift-1}
\end{equation}
Substituting eq.~(\ref{eq:Z-shift-1}) into eq.~(\ref{eq:Z-shift}), we obtain 
\begin{equation}
\begin{split}
    Z(n+2-1;m_1,m_2;x_1,x_2) &= Z(n-1;m_1,m_2;x_1,x_2) 
    + \left(\frac{x_1^n}{n^{m_1}} + \frac{x_1^{n+1}}{(n+1)^{m_1}}\right)
\\ & \times
    Z(n-1;m_2;x_2)
    + \frac{x_2^n}{n^{m_2}}\frac{x_1^{n+1}}{(n+1)^{m_1}}\,.
\end{split}
\end{equation}
{\it Hence, we have rewritten the original expression in terms of only $Z$-sums 
to $n-1$.} Generally, offsets in $Z$-sums higher depth, as well as $S$-sums, 
can be reduced in a recursive way, see~\cite{Moch:2001zr} for a discussion. This procedure is called {\it the synchronization of sums}.

In practical calculations, we are also often faced with sums of the form
\begin{equation}
    \sum_{i=1}^{n}\frac{x^i}{(i+c)^m} Z(i-1,\ldots)\,,
\label{eq:c-shifted-sum}
\end{equation}
with $c$ a positive integer. We would like to relate this sum to sums 
where the offset $c$ is zero. The basic idea is simple: if the depth of 
the subsum $Z(i-1,\ldots)$ is zero, we can shift $c$ down by one using 
the identity
\begin{equation}
    \sum_{i=1}^{n}\frac{x^i}{(i+c)^m} = 
    \frac{1}{x} \sum_{i=1}^{n}\frac{x^i}{(i+c-1)^m} - \frac{1}{c^m} 
    + \frac{x^n}{(n+c)^m}\,.
\label{eq:c-shift}
\end{equation}
This relation is easily derived by replacing the index of summation as 
$i \to i-1$ and accounting for the changed limits of summation. The 
derivation also makes it obvious that the last term on the right 
hand side only contributes if $n$ is not equal to infinity. Repeated 
applications of eq.~(\ref{eq:c-shift}) then reduce the offset to zero.
In order to see what happens if the depth of the subsum is not zero, 
let us consider e.g.,
\begin{equation}
    \sum_{i=1}^n \frac{x_1^i}{(i+1)^{m_1}} Z(i-1;m_2;x_2)\,.
\end{equation}
We start by resolving the first (and in this case) only summation in the 
definition of the subsum, and shifting the index of summation $i \to i-1$, 
as before
\begin{equation}
\begin{split}
    \sum_{i=1}^n \frac{x_1^i}{(i+1)^{m_1}} Z(i-1;m_2;x_2)
    &= \sum_{i=1}^n \frac{x_1^i}{(i+1)^{m_1}} 
    \sum_{j=1}^{i-1} \frac{x_2^j}{j^{m_2}}Z(j-1)
\\
    &= \sum_{i=2}^{n+1} \frac{x_1^{i-1}}{(i-1+1)^{m_1}} 
    \sum_{j=1}^{i-1-1} \frac{x_2^j}{j^{m_2}}Z(j-1)\,.
    \end{split}
\end{equation}
Next, separate the term with $i=n+1$ to obtain a sum in $i$ whose 
upper limit is $n$
\begin{equation}
\begin{split}
    \sum_{i=1}^n \frac{x_1^i}{(i+1)^{m_1}} Z(i-1;m_2;x_2)
    &=
    \sum_{i=2}^{n} \frac{x_1^{i-1}}{(i-1+1)^{m_1}} 
    \sum_{j=1}^{i-1-1} \frac{x_2^j}{j^{m_2}}Z(j-1)
\\ &+ 
    \frac{x_1^{n}}{(n+1)^{m_1}} 
    \sum_{j=1}^{n-1} \frac{x_2^j}{j^{m_2}}Z(j-1)\,.
    \end{split}
\end{equation}
Obviously, this step is only necessary if $n$ is finite, thus the last 
term is not present if the limit of summation if infinity. Now, rewrite 
the sum over $j$ on the first line such that its upper limit of summation 
is $i-1$
\begin{equation}
\begin{split}
&
\sum_{i=1}^n \frac{x_1^i}{(i+1)^{m_1}} Z(i-1;m_2;x_2)
    =
\\ &\qquad=
    \sum_{i=2}^{n} \frac{x_1^{i-1}}{(i-1+1)^{m_1}} 
    \bigg[\sum_{j=1}^{i-1} \frac{x_2^j}{j^{m_2}}Z(j-1)-
    \frac{x_2^{i-1}}{(i-1)^{m_2}}Z(i-1-1)\bigg]
\\ &\qquad+ 
    \frac{x_1^{n}}{(n+1)^{m_1}} 
    \sum_{j=1}^{n-1} \frac{x_2^j}{j^{m_2}}Z(j-1)\,.
    \end{split}
\end{equation}
The two sums over $j$ on the second and third lines can be evaluated as 
$Z$-sums to $i-1$ and $n-1$, while the contribution coming from the second 
term in the square brackets does not depend on $j$. In this contribution 
change the index of summation as $i \to i+1$
\begin{equation}
\begin{split}
&
\sum_{i=1}^n \frac{x_1^i}{(i+1)^{m_1}} Z(i-1;m_2;x_2)
    =
\\ &\qquad=
    \frac{1}{x_1}\sum_{i=2}^{n} \frac{x_1^{i}}{(i-1+1)^{m_1}} Z(i-1;m_2;x_2) 
    - \sum_{i=1}^{n-1} \frac{x_1^{i}}{(i+1)^{m_1}}
    \frac{x_2^{i}}{(i)^{m_2}}Z(i-1)
\\ &\qquad+ 
    \frac{x_1^{n}}{(n+1)^{m_1}} Z(n-1;m_2;x_2)\,.
    \end{split}
\label{eq:c-shift-2}
\end{equation}
To finish, we note the following: in the first term, we have lowered the 
offset by one, which was our goal and in this particular case, the offset 
is already zero at this stage. Moreover, the lower limit of summation in 
this term can formally be put to $i=1$, since the term with $i=1$ would be 
zero because $Z(n,\ldots) = 0$ if $n<0$ (the outermost sum is empty). 
In the second term, we have reduced the depth of the inner sum by one, 
and in this particular case, the inner sum is now depth zero. Finally, 
the last term is only present if $n$ is not infinity. The only contribution 
that is not yet in the form of a $Z$-sum is the second term on the right 
hand side, as products of $(i+1)$ and $i$ appear in the denominator. 
However, for specific integers $m_1$ and $m_2$, partial fractioning can 
be used to reduce the product in the second term to a sum of terms that 
contain either $(i+1)$ or $i$ in their denominator, but not both. In 
general, iterated applications of the formula
\begin{equation}
    \frac{x_1^i}{(i+c_1)^{m_1}}\frac{x_2^i}{(i+c_2)^{m_2}}
    = \frac{1}{c_1 - c_2}
    \left[\frac{x_1^i}{(i+c_1)^{m_1}}\frac{x_2^i}{(i+c_2)^{m_2-1}}
     - \frac{x_1^i}{(i+c_1)^{m_1-1}}\frac{x_2^i}{(i+c_2)^{m_2}}\right]
\end{equation}
will achieve the desired form. At this stage, all remaining sums will be 
of the form
\begin{equation}
    \sum_{i=1}^{n}\frac{x^i}{i^m}Z(i-1,\ldots)\,,
\end{equation}
which are simply $Z$-sums. Finally, for offsets greater than one and subsums 
of general depth, we must apply the procedure outlined above iteratively, 
until the offset has been completely reduced to zero. The general case is discussed in~\cite{Moch:2001zr}.

%\begin{equation}
%    \sum_{i=1}^{n}\frac{x^i}{(i+c)^m} = 
%    \sum_{i=2}^{n+1}\frac{x^{i-1}}{(i-1+c)^m} =
%    \frac{1}{x}\sum_{i=1}^{n}\frac{x^i}{(i-1+c)^m} - \frac{1}{c^m}
%    + \frac{x^n}{(n+c)^m}\,.
%\end{equation}

Finally, the fourth property {\bf S-Z4} that we will discuss is that 
$Z$-sums and $S$-sums appear in expansions of the gamma function around integer values. 
This property makes them particularly useful for computing 
$\epsilon$-expansions of sums obtained from Mellin-Barnes integrals using 
the residue theorem. Indeed, one can show that for positive integers $n$ 
we have
\begin{equation}
\begin{split}
    &\Gamma(n+\epsilon) = \Gamma(1+\epsilon) \Gamma(n)
\\ &\qquad\times
    \left[
    1 + \epsilon Z_1(n-1) + \epsilon^2 Z_{11}(n-1) + \epsilon^3 Z_{111}(n-1)
    + \ldots + \epsilon^{n-1} Z_{11\ldots 1}(n-1)
    \right]\,,
\end{split}
\label{eq:Gamma-exp-pos}
\end{equation}
where
\begin{equation}
    Z_{\underbrace{11\ldots 1}_{k}}(n-1) = Z(n-1;\underbrace{1,1,\ldots,1}_k;
    \underbrace{1,1,\ldots,1}_k)\,.
\end{equation}
Notice that in eq.~(\ref{eq:Gamma-exp-pos}), since we are keeping the 
factor of $\Gamma(1+\epsilon)$ unexpanded on the right hand side, the sum 
in the square brackets is finite. The expansion around negative integers 
takes the form (again $n>0$)
\begin{equation}
\begin{split}
    &\Gamma(-n+1+\epsilon) = \frac{\Gamma(1+\epsilon)}{\epsilon} 
    \frac{(-1)^{n-1}}{\Gamma(n)}
\\ &\qquad\times
    \left[
    1 + \epsilon S_1(n-1) + \epsilon^2 S_{11}(n-1) + \epsilon^3 S_{111}(n-1)
    + \ldots \right]\,,
\end{split}
\label{eq:Gamma-exp-neg}
\end{equation}
with
\begin{equation}
    S_{\underbrace{11\ldots 1}_{k}}(n-1) = S(n-1;\underbrace{1,1,\ldots,1}_k;
    \underbrace{1,1,\ldots,1}_k)\,.
\end{equation}
In contrast to the expansion around positive integers, in this case the 
sum in the square brackets is infinite. Last, we note the inverse formula
\begin{equation}
\begin{split}
&    \left[
    1 + \epsilon Z_1(n-1) + \epsilon^2 Z_{11}(n-1) + \epsilon^3 Z_{111}(n-1)
    + \ldots + \epsilon^{n-1} Z_{11\ldots 1}(n-1)
    \right]^{-1} =
\\ &\qquad = 
    1 - \epsilon S_1(n-1) + \epsilon^2 S_{11}(n-1) - \epsilon^3 S_{111}(n-1)
    + \ldots\,,
\end{split}
\label{eq:Zsum-invert}
\end{equation}
which can be useful in practical applications.

\subsection{Solving {1-dim}ensional Mellin-Barnes Integrals with Nested Sums}
\label{sect:1dmb-sol-with-sums}

The tools developed in the previous section allow us to derive analytic 
solutions for a large and important class of one-dimensional Mellin-Barnes 
integrals. Thus, let us consider the integral
\begin{equation}
    \int_{z_0-i\infty}^{z_0+i\infty} \frac{dz}{2\pi i} 
    \frac{\prod_{i=1}^{m} \Gamma(a_{i} + z)
    \prod_{j=1}^{n} \Gamma(b_{j} - z)}
    {\prod_{k=1}^{\bar{m}} \Gamma(c_{k} + z)
    \prod_{l=1}^{\bar{n}} \Gamma(d_{l} - z)}
    \prod_{p=1}^{r} \psi^{(h_p)}(e_{p}+z)
    \prod_{q=1}^{s} \psi^{(g_q)}(f_{q}-z)
    A^z\,,
\label{eq:1dMB-gen}
\end{equation}
where $a_i$, $b_j$, $c_k$ , $d_l$, $e_p$ and $f_q$ are assumed to be 
integers, not necessarily distinct. (Hence, e.g. $a_{i_1} = a_{i_2}$, etc. 
is allowed.) {\it Such expressions arise from dimensionally-regularized 
Mellin-Barnes integrals after one has resolved the poles and expanded in 
$\epsilon$, see chapter~\ref{chapter-singul}.} It may also be possible to construct first the sum representation of the Mellin-Barnes integral and perform the $\epsilon$-expansion on this 
sum and we will come back briefly to this approach below. However, let us 
first concentrate on eq.~(\ref{eq:1dMB-gen}). To start, let us observe that 
because $a_i,\ldots,f_q \in \mathbb{Z}$, the integrand only has poles at 
integer values of $z$. %Then 
\begin{svgraybox}
Our aim will essentially be to show that one 
can write the residue at a generic integer $n$ explicitly and that the 
summation over the residues can be performed in terms of $S$- or $Z$-sums. 
\end{svgraybox}
In order to be specific, let us assume that $0<Re A<1$, so that we must close 
the contour of integration to the right, but the procedure can be adapted 
also to the case of $1<Re A$ as well.

To show the above, we first perform a change of integration variables, 
whose utility will become clear in a moment. Thus, set
\begin{equation}
\bar{a} \equiv \min\{0,a_i\}\,,\quad
\bar{d} \equiv \max\{0,d_l-1\}\,,\quad
\bar{e} \equiv \min\{0,e_p\}\,,\quad
\bar{f} \equiv \max\{0,f_q-1\}\,.
\end{equation}
Then define
\begin{equation}
r \equiv \max\{|\bar{a}|,\bar{d},|\bar{e}|,\bar{f}\}
\label{eq:r-def}
\end{equation}
and let $z\to z+r$ in eq.~(\ref{eq:1dMB-gen}). Introducing the notation
\begin{equation}
    \bar{a}_i \equiv a_i+r\,,\quad
    \bar{b}_j \equiv b_j-r\,,\quad
    \bar{c}_k \equiv c_k+r\,,\quad
    \bar{d}_l \equiv d_l-r\,,\quad
    \bar{e}_p \equiv e_p+r\,,\quad
    \bar{f}_q \equiv f_q-r\,,\quad
\end{equation}
we find
\begin{equation}
    A^r \int_{z_0-r-i\infty}^{z_0-r+i\infty} \frac{dz}{2\pi i} 
    \frac{\prod_{i=1}^{m} \Gamma(\bar{a}_{i} + z)
    \prod_{j=1}^{n} \Gamma(\bar{b}_{j} - z)}
    {\prod_{k=1}^{p} \Gamma(\bar{c}_{k} + z)
    \prod_{l=1}^{q} \Gamma(\bar{d}_{l} - z)}
    \prod_{p=1}^{r} \psi^{(h_p)}(\bar{e}_{p}+z)
    \prod_{q=1}^{s} \psi^{(g_q)}(\bar{f}_{q}-z)
    A^{z}\,.
\label{eq:1dMB-gen-1}
\end{equation}
The point of this transformation is that by construction $\bar{a}_i\ge 0$, 
$\bar{d}_l\le 1$, $\bar{e}_p\ge 0$ and $\bar{f}_q\le 1$, which will be 
convenient in what follows. Second, we shift the contour of integration to 
some standard contour. Let us choose e.g. the straight line contour 
parallel to the imaginary axis running form $\bar{z}_0-i\infty$ to 
$\bar{z}_0+i\infty$ with $\bar{z}_0 \in (0,1)$. As $\bar{a}_i,\ldots,\bar{f}_q$ 
are all integers, such a contour avoids all poles of the integrand. Clearly, 
when shifting the contour from $z_0-r$ to $\bar{z}_0$, we must account for 
the residues of the poles that are crossed, so our integral becomes
\begin{eqnarray}
%\begin{split}
    A^r && \int_{\bar{z}_0-i\infty}^{\bar{z}_0+i\infty} \frac{dz}{2\pi i} 
    \frac{\prod_{i=1}^{m} \Gamma(\bar{a}_{i} + z)
    \prod_{j=1}^{n} \Gamma(\bar{b}_{j} - z)}
    {\prod_{k=1}^{p} \Gamma(\bar{c}_{k} + z)
    \prod_{l=1}^{q} \Gamma(\bar{d}_{l} - z)} \nn \\
&\times& \prod_{p=1}^{r} \psi^{(h_p)}(\bar{e}_{p}+z)
    \prod_{q=1}^{s} \psi^{(g_q)}(\bar{f}_{q}-z)
    A^{z} 
    + \sum \mbox{Residues}\,.
%\end{split}
\label{eq:1dMB-gen-2}
\end{eqnarray}
Notice that as we shift the contour, we cross only finitely many poles 
at specific known integers, thus the sum of residues can be computed 
explicitly once the values of $a_i,\ldots,f_p$ as well as $h_p,\,g_q$ 
and $z_0,\, \bar{z}_0$ are fixed.

Thus, we concentrate on the left-over integral. To proceed, we perform 
the synchronization of the gamma and $\psi$ functions. This simply 
means that we apply the identities $\Gamma(1+x) = x \Gamma(x)$ and 
$\psi^{(n)}(1+x) = \psi^{(n)}(x) + (-1)^n n! z^{-n-1}$ repeatedly, 
until all gamma functions have been reduced to just $\Gamma(z)$ or 
$\Gamma(1-z)$ and all $\psi$ functions to $\psi^{(n)}(z)$ or $\psi^{(m)}(1-z)$. 
Since $\bar{a}_i,\ldots,\bar{f}_q$ are all integers, we can achieve this in 
a finite number of steps. In fact, clearly
\begin{align}
    \Gamma(\bar{a}_i+z) = (\bar{a}_i-1+z)(\bar{a}_i-2+z)\ldots z\Gamma(z)\,,
    \qquad \bar{a}_i > 0\,,
\label{eq:Gamma-sync-1}
\\
\Gamma(\bar{b}_j - z) = 
    \begin{cases}
        (\bar{b}_j-1-z)(\bar{b}_j-2-z)\ldots (1-z) \Gamma(1-z) 
        & \text{if $\bar{b}_j > 1$}\,, \\
        \frac{\Gamma(1-z)}{(-|\bar{b}_j|-z)(-|\bar{b}_j|+1-z)\ldots (-z)} 
        & \text{if $\bar{b}_j < 1$}\,,
    \end{cases} 
\label{eq:Gamma-sync-2}
\\
\Gamma(\bar{c}_k + z) = 
    \begin{cases}
        (\bar{c}_k-1+z)(\bar{c}_k-2+z)\ldots z\Gamma(z) 
        & \text{if $\bar{c}_k > 0$}\,, \\
        \frac{\Gamma(z)}{(-|\bar{c}_k|+z)(-|\bar{c}_k|+1+z)\ldots (-1+z)} 
        & \text{if $\bar{c}_k < 0$}\,,
    \end{cases} 
\label{eq:Gamma-sync-3}    
\\
    \Gamma(\bar{d}_l-z) = \frac{\Gamma(1-z)}{(-|\bar{d}_l|-z)(-|\bar{d}_l|+1-z)\ldots (-z)}\,,
    \qquad \bar{d}_l < 1\,,
\label{eq:Gamma-sync-4}
\end{align}
and furthermore
\begin{align}
\psi^{(n_p)}(\bar{e}_p+z) &= \psi^{(n_p)}(z)
    + (-1)^n n! \left(\frac{1}{(\bar{e}_p-1+z)^{n_p+1}} + \ldots + 
    \frac{1}{z^{n_p+1}}\right),\;\;
% \bar{e}_p > 0\,,
\label{eq:psi-sync-1}
\\
\psi^{(m_q)}(\bar{f}_q-z) &= \psi^{(m_q)}(1-z)
    - (-1)^n n! \left(\frac{1}{(-|\bar{f}_q|-z)^{m_q+1}} + \ldots + 
    \frac{1}{(-z)^{m_q+1}}\right). 
\label{eq:psi-sync-2}
\end{align}
Eq.~(\ref{eq:psi-sync-1}) holds for $\bar{e}_p > 0$ while eq.~(\ref{eq:psi-sync-2}) is true for $\bar{f}_q < 1$.
\\
Then, we use the relation
\begin{equation}
    \psi^{(m_q)}(1-z) = (-1)^{m_q}
    \left[\psi^{(m_q)}(z) + \pi \frac{\partial^{m_q}}{\partial z^{m_q}}
    \cot(\pi z)\right]
\label{eq:psi-reflect}
\end{equation}
to convert all $\psi^{(m_q)}(1-z)$ functions to $\psi^{(m_q)}(z)$ functions.

After synchronizing the gamma and $\psi$ functions in this way, the 
remaining one-dimensional integral in eq.~(\ref{eq:1dMB-gen-2}) takes the 
form
\begin{equation}
\begin{split}
&
    \int_{\bar{z}_0-i\infty}^{\bar{z}_0+i\infty} \frac{dz}{2\pi i} 
    R(\bar{a}_i,\bar{b}_j,\bar{c}_k,\bar{d}_l,z)
    \frac{\Gamma^m(z)\Gamma^n(1-z)}
    {\Gamma^p(z)\Gamma^q(1-z)} 
    \prod_{p=1}^{r}\bigg[\psi^{(n_p)}(z) + R_{n_p}(\bar{e}_p,z)\bigg]
\\ &\times
    \prod_{q=1}^{s}\bigg[(-1)^{m_q}\psi^{(m_q)}(z) 
    +(-1)^{m_q}\pi \frac{\partial^{m_q}}{\partial z^{m_q}}\cot(\pi z)
    + R_{m_q}(\bar{f}_q,z)\bigg] A^z\,,
\end{split}
\end{equation}
where $R(\bar{a}_i,\bar{b}_j,\bar{c}_k,\bar{d}_l,z)$, $R_{n_p}(\bar{e}_p,z)$ 
and $R_{m_q}(\bar{f}_q,z)$ are rational functions in $z$, such that their 
denominators only contains factors of $(n_i + z)$ with $n_i \in {\mathbb N}$. 
This last property is quite convenient, since it ensures that the poles of 
the integrand come from the gamma functions only and not from the rational 
functions $R$. However, it is easy to see that this property only holds if 
$\bar{a}_i\ge 0$, $\bar{d}_l \le 1$ as well as $\bar{e}_p\ge 0$ and 
$\bar{f}_1 \le 1$. This explains our initial transformation of the integration 
variable. Thus, after expanding the products, we find that our integral can 
be decomposed into a sum of integrals with the following general structure
\begin{equation}
    \begin{split}
        \sum 
        \int_{\bar{z}_0-i\infty}^{\bar{z}_0+i\infty} \frac{dz}{2\pi i}
        \Gamma^{N}(z)\Gamma^{M}(1-z) 
        R'(\bar{a}_i,\ldots,\bar{f}_q,z) A^z
        \prod_{p=1}^{P'} \psi^{n'_p}(z)
        \prod_{q}^{Q'}
        \pi \frac{\partial^{m'_q}}{\partial z^{m'_q}}\cot(\pi z)
    \end{split}
\label{eq:1dMB-gen-basic}
\end{equation}
Here $R'(\bar{a}_i,\ldots,\bar{f}_q,z)$ is again a rational function, 
generally a product of $R(\bar{a}_i,\bar{b}_j$, $\bar{c}_k,\bar{d}_l,z)$ and 
factors of $R_{n_p}(\bar{e}_p,z)$ and $R_{m_q}(\bar{f}_q,z)$. Then it is clear
that just as for $R$, $R_{n_p}$ and $R_{m_q}$, the denominator of $R'$ also 
factorizes into factors of $(n_i + z)$ with $n_i \in {\mathbb N}$. Thus, after partial fractioning in $z$ we have simply
\begin{equation}
    R'(\bar{a}_i,\ldots,\bar{f}_q,z) =
    \sum_{i} \frac{\alpha_i}{(n_i+z)^{p_i}} + \sum_j \beta_j z^{j}\,,
\label{eq:R-prime}
\end{equation}
where $n_i$ and $p_i$ are non-negative integers. The fact that all $n_i$ above 
are non-negative is significant, because then all poles of the integrand to 
the right of the imaginary axis only come from $\Gamma(1-z)$ or the derivatives 
of $\cot(\pi z)$, but not from the rational function $R'$. However, since we have 
chosen $\bar{z}_0 \in (0,1)$ (and we assumed that $0<Re A<1$, so that we must 
close the contour to the right), it is precisely these poles, i.e., the ones 
at $z=1,2,\ldots$, whose residues contribute to the integral.

After this preliminary work, we now want to show that the residue of the 
pole at some generic $z\to n$, $n=1,2,\ldots$ can be expressed in terms of 
objects such as factors of $(c+n)^k$, $Z(n-1,\ldots)$ and so forth, such that we will be able to perform the summation over $n$ form one to infinity in 
terms of $Z$-sums. To do this, we will assume that our integrals are 
\emph{balanced}, see section~\ref{sec:choosingcontour}.
%\begin{svgraybox}
%Recall that informally, a balanced integral is one where for each function 
%of the form $\Gamma(\ldots+z_j)$, a function of the form $\Gamma(\ldots-z_j)$ also %appears in the integrand, with gamma functions in the denominator 
%being counted a negative number of times. 
%\end{svgraybox}

Evidently, the integrals in eq.~(\ref{eq:1dMB-gen-basic}) are 
balanced if $N=M$, and we will also require that $N=M > 0$. Now, 
computing the residue of the integrand at some generic positive integer 
$n$ simply amounts to expanding the integrand in $z$ around $n$, with the 
coefficient of $(z-n)^{-1}$ essentially giving the residue (up to a factor of $2\pi i$). Thus, the basic question becomes: {\it Can we perform this expansion 
symbolically around a generic $n$ and what kind of objects appear in this 
expansion?} Remember that $N$, $R'$ as well as the $n'_p$ and $m'_q$ are all 
assumed to be explicitly known, however the integer $n$ must be kept symbolic. 
In order to perform the expansion and the computation of the residue, let us 
set $z = n + \delta$ so that the expansion is around $\delta = 0$, and recall 
the following facts.
\begin{enumerate}
    \item The gamma functions appearing in the integrand in 
    eq.~(\ref{eq:1dMB-gen-basic}), $\Gamma(n+\delta)$ and $\Gamma(1-n-\delta)$, 
    may be expanded in $\delta$ using eqs.~(\ref{eq:Gamma-exp-pos})~and~(\ref{eq:Gamma-exp-neg}). These expressions will generally yield $S$- and 
    $Z$-sums as well as their products, which can be converted to an expression 
    involving only the linear combination of, say, $Z$-sums as explained above. 
    However, for the balanced integrals we are considering here, this procedure
    can be circumvented altogether by noting that
    \begin{equation}
        \Gamma(z)\Gamma(1-z) = \frac{\pi}{\sin(\pi z)}\,.
    \label{eq:gg2csc}
    \end{equation}
    Computing the expansion of this expression around integers is made quite 
    simple by observing that for $z=n+\delta$, $n\in\mathbb{N}$ we have
    \begin{equation}
        \frac{\pi}{\sin[\pi (n+\delta)]} = \frac{(-1)^n \pi}{\sin(\pi \delta)}
        = (-1)^n\bigg[\frac{1}{\delta} + \frac{\pi^2}{6} \delta 
        + \frac{7\pi^4}{360} \delta^3 + \mathcal{O}(\delta^5)\bigg]\,.
    \label{eq:csc-exp}
    \end{equation}
    Thus the expansion of the product of gamma functions will simply yield 
    a power of $[(-1)^n]^N$ times a Laurent-series in $\delta$ starting at order 
    $1/\delta^N$ with explicitly known numeric coefficients independent of $n$.
    
    \item The expansion of the rational function $R'$ in eq.~(\ref{eq:R-prime}) 
    is trivial. Recalling that by construction $R'$ has no poles at positive 
    integers, we may simply compute the Taylor-series of $R'$ around $z=n$,
    \begin{equation}
    \begin{split}
        R'(\bar{a}_i,\ldots,\bar{f}_q,n+\delta) &=
        R'(\bar{a}_i,\ldots,\bar{f}_q,n) + 
        \frac{d R'(\bar{a}_i,\ldots,\bar{f}_q,z)}{dz}\bigg|_{z=n} \delta 
    \\ &        
        +
        \frac{d^2 R'(\bar{a}_i,\ldots,\bar{f}_q,z)}{dz^2}\bigg|_{z=n} 
        \frac{\delta^2}{2} + \mathcal{O}(\delta^3)\,.
    \end{split}
    \end{equation}
    Obviously the coefficients in this expansion can be computed 
    explicitly for general $n$ given any specific $R'$ and they are simply 
    rational functions of $n$. It is also clear that the denominator of 
    each expansion coefficient continues to factorize into linear factors 
    of $(n_i+n)$ and so can be partial fractioned to yield a form
    \begin{equation}
        \sum_i \frac{\rho_i}{(n_i+n)^{k_i}} + \sum_j \sigma_j n^j\,.
    \end{equation}
    
    \item Clearly the expansion of the factor of $A^z$ is trivial and we have 
    simply
    \begin{equation}
        A^{n+\delta} = A^n \bigg[1 + \ln A \delta 
        + \ln^2 A \frac{\delta^2}{2} + \mathcal{O}(\delta^3)
        \bigg]\,.
    \end{equation}
    Thus the result involves a factor of $A^n$ as well as powers of 
    logarithms of $A$.
    
    \item Turning to the product of $\psi$ functions, we note that again 
    these do not have poles for positive integers, so we may simply Taylor-expand 
    around $n$. Using the definition $\psi^{(m+1)}(z) = d\psi^{(m)}(z)/dz$, we 
    have simply
    \begin{equation}
        \psi^{(m)}(n+\delta) = \psi^{(m)}(n) + \psi^{(m+1)}(n) \delta
        + \psi^{(m+2)}(n) \frac{\delta^2}{2} + \mathcal{O}(\delta^3)\,,
    \end{equation}
    and thus the expansion of a product of $\psi$ functions will involve 
    products of $\psi$ functions at the integer value $n$. We now make the 
    important observation that polygamma functions at positive integer 
    values can be written in terms of $Z$-sums as 
    \begin{align}
        \psi^{(0)}(n) &= -\gamma + Z(n-1;1;1)\,,
        \\
        \psi^{(m)}(n) &= (-1)^{m+1} m! \left[\zeta_{m+1} - Z(n-1;m+1;1)\right]\,,
        \qquad m > 0\,.
    \end{align}
    Thus, the expansion coefficients will in general involve products of 
    $Z$-sums which can be reduced to linear combinations of such sums using 
    the algebra of $Z$-sums as explained above.
    
    \item Last, let us consider the product of derivatives of $\cot(\pi z)$ 
    that appears in eq.~(\ref{eq:1dMB-gen-basic}). First, it is trivial to 
    show that 
    \begin{equation}
        \frac{\partial \cot(\pi z)}{\partial z} = -
        \pi \left[1 + \cot^2(\pi z)\right]\,,
    \end{equation}
    which implies that the $m$-th derivative of $\cot(\pi z)$ can be expressed 
    as an $m+1$-st degree polynomial in $\cot(\pi z)$. Hence, the product 
    of derivatives will evaluate to simply a polynomial in $\cot(\pi z)$. 
    Computing the expansion of this polynomial around an integer $n$ is 
    simplified by noticing that
    \begin{equation}
        \cot[\pi(n+\delta)] = \cot(\pi\delta) 
        = \frac{1}{\pi}\bigg[\frac{1}{\delta} - \frac{\pi^2}{3} \delta 
        - \frac{\pi^4}{45} \delta^3 + \mathcal{O}(\delta^5)\bigg]\,.
    \end{equation}
    Thus, the expansion of the product of derivatives of $\cot(\pi z)$ 
    yields a Laurent-series in $\delta$ whose coefficients are explicitly 
    known numbers independent of $n$.
\end{enumerate}
From the above, it is clear that the residue of the integrand in 
eq.~(\ref{eq:1dMB-gen-basic}) can indeed be computed for a generic positive 
integer $n$ and will involve at most the product of factors of $(n_i+n)^{k_i}$, 
$A^n$, $Z$-sums $Z(n-1,\ldots)$ (these are only present if the integrand involved 
polygamma functions) as well as numeric factors independent 
of $n$ and logarithms of $A$ raised to $n$-independent powers. Moreover, the 
summation over residues runs from $n=1$ to $n=\infty$ by construction (recall 
our choice of $0<\bar{z}_0<1$), so at most we encounter sums of the form
\begin{equation}
    \sum_{n=1}^{\infty} \frac{A^n}{(n_i + n)^{k_i}} Z(n-1,\ldots)\,.
\end{equation}
These sums are then straightforward to evaluate using the summation techniques 
presented above. The result will involve $Z$-sums to infinity which we recall 
are noting but a multiple polylogarithms, see eq.~(\ref{eq:Z-sum-to-Li}) and chapter~\ref{chapter:complex}.
\begin{overview}{Solving 1-dimensional Mellin-Barnes Integrals with Nested Sums}

The above approach to solving 1d Mellin-Barnes integrals with nested sums involves 
the following steps.
\begin{enumerate}
\item Transform the variable of integration $z\to z+r$, where the constant $r$ is 
given in eq.~(\ref{eq:r-def}).
\item Shift the contour of integration to the standard contour that it runs parallel 
to the imaginary axis form $\bar{z}_0-i\infty$ to $\bar{z}_0+i\infty$ with 
$\bar{z}_0 \in (0,1)$. If any poles of the integrand are crossed, compute their 
contribution using Cauchy's residue theorem.
\item Transform all gamma functions in the integrand to the form 
$\Gamma(z)$ and $\Gamma(1-z)$ using eqs.~(\ref{eq:Gamma-sync-1})--(\ref{eq:Gamma-sync-4}).
\item Bring all polygamma functions to the form $\psi^{(n)}(z)$ using eqs.~(\ref{eq:psi-sync-1})--(\ref{eq:psi-reflect}).
\item The first two steps generally introduce a rational function $R(z)$ of the 
integration variable $z$ into the integrand. Apply partial fractioning in $z$ to 
this rational function.
\item Compute the residue of the integrand symbolically at $z=n$, $n\in \mathbb{N}^+$. 
We have shown how to do this in terms of at most $Z$-sums $Z(n-1,\ldots)$, as well as 
factors of the form $(c + n)^k$, with $k\in\mathbb{Z}$ and $A^n$.
\item Perform the summation using the algorithms for $Z$-sums discussed above.
{Note that many of these algorithms, along with several others, have been implemented 
in publicly available packages such as {\tt XSummer} and {\tt Sigma}, see appendix~\ref{app:mplsums}.}
\item The result will involve $Z$-sums to infinity which can be expressed in terms of 
multiple polylogarithms if desired.
\end{enumerate}
\end{overview}

\subsection{Application of Sums to Angular Integrations}
In order to illustrate the procedure outlined above, we consider the 
angular integral with two massless denominators, discussed in section~\ref{sec:4MBrepr}. 
There we showed that this integral has a one dimensional Mellin-Barnes 
representation,
\begin{equation}
    \begin{split}
        \Omega_{j,k} &= \frac{2^{1-2\epsilon} \pi ^{-\epsilon} 
        \Gamma (1-\epsilon)}{\Gamma (1-2 \epsilon)}
        \int_{-1}^{1} d(\cos\theta)(\sin\theta)^{-2\eps}
        \int_{-1}^{1} d(\cos\phi)(\sin\phi)^{-1-2\eps}
\\& \times
        (1-\cos\theta)^{-j}
        (1-\cos\chi\cos\theta-\sin\chi\sin\theta\cos\phi)^{-k}
\\& =
        2^{2-j-k-2\epsilon} \pi^{1-\epsilon}
        \frac{1}{\Gamma(j)\Gamma(k)\Gamma(2-j-k-2\epsilon)}
\\& \times
        \int_{z_0-i\infty}^{z_0+i\infty} \frac{dz}{2\pi i}
        \Gamma(-z)\Gamma(j+z)\Gamma(k+z)\Gamma(1-j-k-\epsilon-z)
        v^z\,,
    \end{split}
\end{equation}
where $v=(1-\cos\chi)/2$. Let us now consider the specific case of 
$j=k=1$ and extract the overall factor in the first line by defining 
\begin{equation}
    \begin{split}
        I_{1,1}(v,\epsilon) &\equiv \frac{2^{-1+2\epsilon} \pi ^{\epsilon} 
        \Gamma (1-2\epsilon)}{\Gamma (1-\epsilon)}\Omega_{j,k}
\\& =
        \frac{\pi}{\Gamma(-\epsilon)}
        \int_{z_0-i\infty}^{z_0+i\infty} \frac{dz}{2\pi i}
        \Gamma(-z)\Gamma^2(1+z)\Gamma(-1-\epsilon-z) v^z\,.
    \end{split}
\end{equation}
This choice of normalization keeps the expanded expressions simpler and avoids 
the appearance of constants such as $\gamma$ and $\ln 2$ in the final result. 
For $\epsilon=0$ it is clearly impossible to find a straight line contour 
parallel to the imaginary axis that separates the poles of $\Gamma(1+z)$ 
and $\Gamma(-1-z)$. Put differently, there is no $z_0$ such that the real 
parts of the arguments of all gamma functions are positive. However for a 
proper choice of $\epsilon\ne 0$, such $z_0$ does exist. Hence we choose 
e.g. $z_0 = -\frac{1}{2}$ and $\epsilon = -1$ and analytically continue the 
integral to $\epsilon \to 0$, picking up the residues of the poles which 
are crossed as explained in chapter~\ref{chapter-singul}. We find
\begin{equation}
\begin{split}
    I_{1,1}(v,\epsilon) &= \pi \Gamma(-\epsilon)\Gamma(1+\epsilon) 
    v^{-1-\epsilon}
\\ &
    + \frac{\pi}{\Gamma(-\epsilon)}
    \int_{-\frac{1}{2}-i\infty}^{-\frac{1}{2}+i\infty}\frac{dz}{2\pi i}
    \Gamma(-z)\Gamma^2(1+z)\Gamma(-1-\epsilon-z) v^z\,.
\end{split}
\end{equation}
The $\epsilon$ expansion can now be performed under the integral sign. Keeping 
terms up to order $\epsilon^2$, we obtain
\begin{equation}
\begin{split}
    I_{1,1}(v,\epsilon) &= 
    \frac{\pi}{v}\bigg[
        -\frac{1}{\epsilon} 
        + \ln v
        - \bigg(\frac{\ln^2 v}{2} + \frac{\pi^2}{6}\bigg) \epsilon
        + \bigg(\frac{\ln^3 v}{6} + \frac{\pi^2 \ln v}{6}\bigg) \epsilon^2
        \bigg]
\\ & +
    \pi(-\epsilon + \gamma \epsilon^2)
    \int_{-\frac{1}{2}-i\infty}^{-\frac{1}{2}+i\infty}\frac{dz}{2\pi i}
    \Gamma(-z)\Gamma^2(1+z)\Gamma(-1-z) v^z
\\ & +
    \pi \epsilon^2
    \int_{-\frac{1}{2}-i\infty}^{-\frac{1}{2}+i\infty}\frac{dz}{2\pi i}
    \Gamma(-z)\Gamma^2(1+z)\Gamma(-1-z) \psi^{(0)}(-1-z) v^z    
    + \mathcal{O}(\epsilon^3)\,.
\end{split}
\end{equation}
Let us now evaluate the two remaining Mellin-Barnes integrals using our 
summation tools. Starting with
\begin{equation}
    I_1 = \int_{-\frac{1}{2}-i\infty}^{-\frac{1}{2}+i\infty}\frac{dz}{2\pi i}
    \Gamma(-z)\Gamma^2(1+z)\Gamma(-1-z) v^z\,,
\end{equation}
we first recognise that in this particular case $r$, defined in 
eq.~(\ref{eq:r-def}), is simply zero, $r=0$. So we do not need to shift 
the integration variable. Next, we choose a straight line contour running 
between zero and one, e.g. $\bar{z}_0 = +\frac{1}{2}$. As our original 
contour was at $z_0 = -\frac{1}{2}$, we must pick up the residue of the 
integrand at $z=0$ as we move to the new contour. This residue is just 
$1-\ln v$, however we are crossing the pole from left to right which means 
we encircle the pole clockwise, hence
\begin{equation}
    I_1 = -1 + \ln v +
    \int_{+\frac{1}{2}-i\infty}^{+\frac{1}{2}+i\infty}\frac{dz}{2\pi i}
    \Gamma(-z)\Gamma^2(1+z)\Gamma(-1-z) v^z\,.
\label{eq:1dMB-I1}
\end{equation}
Concentrating on the left-over integral, our next step is the synchronization 
of gamma functions. Using eqs.~(\ref{eq:Gamma-sync-1})--(\ref{eq:Gamma-sync-4}), 
we find
\begin{equation}
\begin{split}
    I'_1 &= 
    \int_{+\frac{1}{2}-i\infty}^{+\frac{1}{2}+i\infty}\frac{dz}{2\pi i}
    \Gamma(-z)\Gamma^2(1+z)\Gamma(-1-z) v^z
\\& =
    \int_{+\frac{1}{2}-i\infty}^{+\frac{1}{2}+i\infty}\frac{dz}{2\pi i}
    \Gamma^2(1-z)\Gamma^2(z)\frac{1}{-1-z} v^z
    \,.
\end{split}
\end{equation}
As expected, the rational function in the integrand on the second line 
does not have poles for positive $z$. Now we are ready to compute the residue 
of the integrand at a generic positive integer $n$. Using eqs.~(\ref{eq:gg2csc})~and~(\ref{eq:csc-exp}), we find
\begin{equation}
    \mathrm{Res}_{z\to n} \Gamma(1-z)^2\Gamma^2(z)\frac{1}{-1-z} v^z 
    =
    \frac{v^n}{(1+n)^2} - \frac{a^n}{(1+n)} \ln v\,,
    \quad n \in \mathbb{N}^+\,,
\end{equation}
where we have used that $(-1)^{2n} = 1$ for all positive integers $n$. 
Recalling that $v = (1-\cos\chi)/2$, it is clear that {$0\le v \le 1$}, thus 
we must close the contour of integration to the right (clockwise) and so
\begin{equation}
    I'_1 = -\sum_{n=1}^{\infty}\bigg[
    \frac{v^n}{(1+n)^2} - \frac{a^n}{(1+n)} \ln v
    \bigg]\,.
\end{equation}
The sums that appear are of the form discussed in eq.~(\ref{eq:c-shift}) with 
$c=1$. It is then trivial to reduce the shift to zero and we are left with
$Z$-sums of depth one to infinity,
\begin{equation}
    I'_1 = -\frac{1}{v}\sum_{n=1}^{\infty}\frac{v^n}{(n)^2} + 1 
    + \frac{1}{v}\sum_{n=1}^{\infty}\frac{a^n}{(n)} \ln v - \ln v 
    = -\frac{Z(\infty,2,v)}{v} + 1 + \frac{Z(\infty,1,v)}{v} \ln v - \ln v\,.
\end{equation}
Recalling that depth one $Z$-sums to infinity simply correspond to 
classical polylogarithms, see eq.~(\ref{eq:Z1-to-LIm}), thus we have
\begin{equation}
    I'_1 = -\frac{\mathrm{Li}_2(v)}{v} - \frac{\ln(1-v)\ln(v)}{v} + 1 - \ln v\,,
\end{equation}
where we have used that $\mathrm{Li}_1(v) = -\ln(1-v)$. Thus the complete 
integral $I_1$ in eq.(\ref{eq:1dMB-I1}) evaluates to
\begin{equation}
    I_1 = -\frac{\mathrm{Li}_2(v)}{v} - \frac{\ln(1-v)\ln(v)}{v}\,.
\end{equation}

Turning to the last integral, 
\begin{equation}
    I_2 = \int_{-\frac{1}{2}-i\infty}^{-\frac{1}{2}+i\infty}\frac{dz}{2\pi i}
    \Gamma(-z)\Gamma^2(1+z)\Gamma(-1-z) \psi^{(0)}(-1-z) v^z\,,
\end{equation}
once more we find that $r$ (defined in eq.~(\ref{eq:r-def}) is zero, so we 
can move on straight away to shifting the contour to one running between zero 
and one and once more we choose $\bar{z}_0 = +\frac{1}{2}$. As before, we 
must account for the residue of the pole at $z=0$ which we cross from left 
to right, and we find
\begin{equation}
\begin{split}
    I_2 &= 
    - 1 + \gamma + \frac{\pi^2}{6} - \gamma \ln v + \frac{\ln^2 v}{2}
\\& +
    \int_{+\frac{1}{2}-i\infty}^{+\frac{1}{2}+i\infty}\frac{dz}{2\pi i}
    \Gamma(-z)\Gamma^2(1+z)\Gamma(-1-z) \psi^{(0)}(-1-z) v^z\,.
\end{split}
\label{eq:1dMB-I2}
\end{equation}
In order to evaluate the remaining Mellin-Barnes integral, we first synchronize 
the gamma and $\psi$ functions using
eqs.~(\ref{eq:Gamma-sync-1})--(\ref{eq:Gamma-sync-4}) as well as 
eqs.~(\ref{eq:psi-sync-2})~and~(\ref{eq:psi-reflect}), which leads to
\begin{equation}
    \begin{split}
        I'_2 &= \int_{+\frac{1}{2}-i\infty}^{+\frac{1}{2}+i\infty}\frac{dz}{2\pi i}
    \Gamma(-z)\Gamma^2(1+z)\Gamma(-1-z) \psi^{(0)}(-1-z) v^z
\\& =
    \int_{+\frac{1}{2}-i\infty}^{+\frac{1}{2}+i\infty}\frac{dz}{2\pi i}
    \Gamma^2(1-z)\Gamma^2(z)\frac{1}{-1-z}
    \bigg[-\frac{1}{-1-z} - \frac{1}{-z} + \psi^{(0)}(z) + \pi \cot(\pi z)\bigg]\,.
    \end{split}
\end{equation}
Notice that as before, the rational functions which appear do not have poles 
for positive $z$. Computing the residues of the integrand at a generic 
positive integer $n$ is straightforward using the methods outlined above and 
the result reads
\begin{equation}
    \begin{split}
\mathrm{Res}_{z \to n} &
        \Gamma^2(1-z)\Gamma^2(z)\frac{1}{-1-z}
    \bigg[-\frac{1}{-1-z} - \frac{1}{-z} + \psi^{(0)}(z) + \pi \cot(\pi z)\bigg]
\\ 
   \qquad =& 
    \frac{v^n}{(n+1)^2}Z(n-1,1,1)
    +\frac{v^n}{n+1} Z(n-1,2,1)
    -\ln v \frac{v^n}{n+1} Z(n-1,1,1)
\\ &
    +\frac{v^n}{(n+1)^3}
    -(1+\gamma )\frac{v^n}{(n+1)^2}
    -\left(\frac{\ln^2 v}{2} - \gamma  \ln v - \ln v + \frac{\pi^2}{6}\right)
    \frac{v^n}{n+1}
\\ &
    +\frac{v^n}{n^2}
    -\ln v\frac{v^n}{n}\,,\quad n \in \mathbb{N}^+\,.
    \end{split}
\label{eq:I2-res}
\end{equation}
Once more we close the contour to the right, picking up the residues at
$n=1,2,\ldots$. Thus, we must sum the right hand side of eq.~(\ref{eq:I2-res}) 
from $n=1$ to $n=\infty$. Clearly the sums can be evaluated in term of $Z$-sums, 
after reducing the offset to zero where needed. By way of illustration, 
let us consider the summation of the first term on the right hand side. 
Using eq.~(\ref{eq:c-shift-2}), we obtain
\begin{equation}
\begin{split}
&
    \sum_{n=1}^{\infty} \frac{v^n}{(n+1)^2}Z(n-1,1,1) = 
    \frac{1}{v}\sum_{n=1}^{\infty} \frac{v^n}{n^2}Z(n-1,1,1)
    -\sum_{n=1}^{\infty} \frac{v^n}{n(1+n)^2}Z(n-1)
\\& \qquad=
    \frac{1}{v}\sum_{n=1}^{\infty} \frac{v^n}{n^2}Z(n-1,1,1)
    + \sum_{n=1}^{\infty} \frac{v^n}{(1+n)^2}
    + \sum_{n=1}^{\infty} \frac{v^n}{1+n}
    - \sum_{n=1}^{\infty} \frac{v^n}{n}
\end{split}
\end{equation}
where we used $Z(n-1)=1$ for $n \in \mathbb{N}^+$ and perfomred partial fractioning of the expression with respect to $n$. The second and thrid 
sums on the second line still has a non-zero offset, so we apply
eq.~(\ref{eq:c-shift}) 
and obtain
\begin{equation}
\begin{split}
&
    \sum_{n=1}^{\infty} \frac{v^n}{(n+1)^2}Z(n-1,1,1) 
\\& \qquad= 
    \frac{1}{v}\sum_{n=1}^{\infty} \frac{v^n}{n^2}Z(n-1,1,1)
    + \frac{1}{v}\sum_{n=1}^{\infty} \frac{v^n}{n^2} - 1
    + \frac{1}{v}\sum_{n=1}^{\infty} \frac{v^n}{n} - 1
    - \sum_{n=1}^{\infty} \frac{v^n}{n}
\\& \qquad=
    \frac{1}{v}Z(\infty,2,1,v,1) + \frac{1}{v}Z(\infty,2,v)
    + \frac{1}{v}Z(\infty,1,v) - Z(\infty,1,v) - 2\,,
\end{split}
\end{equation}
where in the last line we used the definition of $Z$-sums to evaluate the 
summations. We can evaluate the rest of the sums coming from eq.~(\ref{eq:I2-res}) 
in a similar fashion to arrive at the final result (recall that the contour is 
closed in the clockwise direction)
\begin{equation}
    \begin{split}
        I'_2 &= 
        -\frac{Z(\infty,1,2,v,1)}{v}
        -\frac{Z(\infty,2,1,v,1)}{v}
        +\ln v \frac{Z(\infty ,1,1,v,1)}{v}
        -\frac{Z(\infty,3,v)}{v}
\\ &
        +\gamma \frac{Z(\infty ,2,v)}{v}
        +\bigg(
        \frac{\ln^2 v}{2}
        -\gamma \ln v
        +\frac{\pi^2}{6}
        \bigg)\frac{Z(\infty,1,v)}{v}
        + 1 - \gamma - \frac{\pi^2}{6} + \gamma  \ln v 
\\ &        
        - \frac{\ln^2 v }{2}\,.
    \end{split}
\end{equation}
Then, the complete $I_2$ integral of eq.~(\ref{eq:1dMB-I2}) becomes
\begin{equation}
    \begin{split}
        I_2 &= 
        -\frac{Z(\infty,1,2,v,1)}{v}
        -\frac{Z(\infty,2,1,v,1)}{v}
        +\ln v \frac{Z(\infty ,1,1,v,1)}{v}
        -\frac{Z(\infty,3,v)}{v}
\\ &
        +\gamma \frac{Z(\infty ,2,v)}{v}
        +\bigg(
        \frac{\ln^2 v}{2}
        -\gamma \ln v
        +\frac{\pi^2}{6}
        \bigg)\frac{Z(\infty,1,v)}{v}\,.
    \end{split}
\end{equation}
This result can be written in terms of more familiar functions, since
the depth one sums are simply classical polylogarithms, see
eq.~(\ref{eq:Z1-to-LIm}), while the depth two sums can be expressed 
in terms of harmonic polylogarithms using 
eqs.~(\ref{eq:Z-sum-to-Li})~and~(\ref{eq:H-to-Li}),
\begin{align}
    Z(\infty,1,2,v,1) &= \mathrm{Li}_{2,1}(1,v) = H_{1,2}(v) = H_{1,0,1}(v)\,,
    \\
    Z(\infty,2,1,v,1) &= \mathrm{Li}_{1,2}(1,v) = H_{2,1}(v) = H_{0,1,1}(v)\,,
    \\
    Z(\infty,1,1,v,1) &= \mathrm{Li}_{1,1}(1,v) = H_{1,1}(v),
\end{align}
where on the right hand sides of the first two equations above, we have given the HPLs both in `m'- and `a'-notation ($H_{1,1}(v)$ is clearly the same in both notations).
Finally, MPLs (and thus also HPLs) of weight three or less can always be written in terms of classical polylogarithms (with perhaps arguments that are rational functions the original argument)~\cite{Lewin:1981,Kellerhals1995,Goncharov:1996dce}. For the specific harmonic polylogarithms appearing above, we find
\begin{align}
    H_{1,2}(v) &= 2 \mathrm{Li}_3(1-v) 
    + \mathrm{Li}_2(v) \ln(1-v)
    + \ln v \ln^2(1-v)
    -\frac{\pi^2}{3} \ln(1-v)
    -2 \zeta_3
    \,,
\\
    H_{2,1}(v) &= -\mathrm{Li}_3(1-v)
    -\mathrm{Li}_2(v) \ln(1-v)
    -\frac{1}{2} \ln v \ln^2(1-v)
    +\frac{\pi^2}{6} \ln(1-v)
    +\zeta_3\,
\\ 
    H_{1,1}(v) &= \frac{1}{2} \ln^2(1-v)\,.
\end{align}
Thus, $I_2$ can indeed be expressed with just classical polylograithms,
\begin{equation}
\begin{split}
        I_2 &= 
        -\frac{\mathrm{Li}_3(v)}{v}
        -\frac{\mathrm{Li}_3(1-v)}{v}
        +\gamma \bigg(\frac{\mathrm{Li}_2(v)}{v} + \frac{\ln(1-v) \ln v}{v}\bigg)
        -\frac{\ln(1-v) \ln^2 v}{2 v}
        +\frac{\zeta_3}{v}\,.
\end{split}
\end{equation}
Combining our previous results, the complete solution for the angular 
integral $I_{1,1}(v,\epsilon)$ up to and including terms at order $\epsilon^2$ 
may be written as (see also \wwwaux{I11\_massless})
\begin{equation}
\begin{split}
    I_{1,1}(v,\epsilon) = \frac{\pi}{v}\Bigg\{
    &-\frac{1}{\epsilon}
    +\ln v 
    +\bigg[\mathrm{Li}_2(v) 
        - \frac{1}{2} \ln^2v 
        + \ln(1-v) \ln v 
        - \frac{\pi^2}{6}
    \bigg]\epsilon
\\ &
    -\bigg[\mathrm{Li}_3(v)
        + \mathrm{Li}_3(1-v)
        - \frac{1}{6}\ln^3(v) 
        + \frac{1}{2}\ln(1-v) \ln^2 v 
\\ &
        - \frac{\pi^2}{6} \ln v
        - \zeta_3
   \bigg]\epsilon ^2 
   +\mathcal{O}(\epsilon^3)
    \bigg\}\,.
\end{split}
\end{equation}

\subsection{Expansions of Special Functions}
\label{sec:expansion-special-fcns}

In section~\ref{sec:gamma_hyperg}, we have shown that the generalized Gauss hypergeometric function ${}_A F_B$ is represented as a one-dimensional sum. There exist several multi-variable generalizations of these hypergeometric functions that can be defined through sums of depths greater than one. Some of the better known generalizations to the two-variable case are the Appell functions $F_1$, $F_2$, $F_3$ and $F_4$, which also have simple representations as two-dimensional Mellin-Barnes integrals. In this section, we briefly comment on the fact that the summation technology presented above can be employed to expand such functions around integer values of their parameters. 

We start by considering the Gauss hypergeometric function ${}_2F_1$, given by the following sum
\begin{equation}
    {}_2F_1(\alpha_1,\alpha_2;\beta_1;z) = 
    \frac{\Gamma(\beta_1)}{\Gamma(\alpha_1)\Gamma(\alpha_2)}
    \sum_{n=0}^{\infty}
    \frac{\Gamma(n+\alpha_1)\Gamma(n+\alpha_2)}
    {\Gamma(n+\beta_1)\Gamma(n+1)}
    z^n.
\end{equation}
If the parameters $\alpha_1$, $\alpha_2$ and $\beta_1$ are of the form $n+m\epsilon$, $n\in\mathbb{Z}$, we can compute the $\epsilon$ expansion of ${}_2F_1(\alpha_1,\alpha_2;\beta_1;z)$ as follows. Let us write $\alpha_i = a_i+c_i \epsilon$ and $
\beta_1 = b_1 + d_1 \epsilon$, where we assume that $a_i$ and $b_1$ are integers. Let us concentrate on the following sum
\begin{equation}
    \sum_{n=1}^{\infty}
    \frac{\Gamma(n+a_1+c_1\eps)\Gamma(n+a_2+c_2\eps)}
    {\Gamma(n+b_1+d_1\eps)\Gamma(n+1)}
    z^n.
\end{equation}
First, we expand the gamma functions using eq.~(\ref{eq:Gamma-exp-pos}). This expansion will introduce products of $Z$-sums of the form $Z_{11\ldots 1}(n+a_1-1)$, $Z_{11\ldots 1}(n+a_2-1)$, $Z_{11\ldots 1}(n+b_1-1)$ and $Z_{11\ldots 1}(n)$. Note that the expansion of the gamma functions in the denominator can be brought to this form using eq.~(\ref{eq:Zsum-invert}), followed by a conversion of $S$-sums to $Z$-sums. Next, we can synchronize the sums and use the algebra of $Z$-sums to write the products in terms of just single sums. After performing these manipulations, our problem reduces to the computation of sums of the form
\begin{equation}
    \sum_{n=1}^{\infty} \frac{z^n}{(n+c)^m}Z(n-1,\ldots),
\end{equation}
where $c$ is a non-negative integer. However, these are just the sums we have already encountered in eq.~(\ref{eq:c-shifted-sum})! We have seen there that they can be systematically reduced to $Z$-sums by recursively shifting $c\to c-1$ until this offset is reduced to zero, at which point the sums trivially evaluate to $Z$-sums.

In fact, it is clear that the same algorithm (referred to as algorithm A in~\cite{Moch:2001zr}) can be used to solve any sum of the form
\begin{equation}
\begin{split}
&
    \sum_{n=1}^N \frac{z^n}{(n+c)^m}
    \frac{\Gamma(n+a_1+c_1\eps)}{\Gamma(n+b_1+d_1\eps)}
    \cdots
    \frac{\Gamma(n+a_k+c_k\eps)}{\Gamma(n+b_k+d_k\eps)}
\\ &\qquad\times
    Z(n+o-1;m_1,\ldots,m_l;z_1,\ldots,z_l)
\end{split}
\label{eq:A-type-sum}
\end{equation}
in terms of $Z$-sums. The upper limit of summation is allowed to be infinity. This type of sum includes in particular the generalized Gauss hypergeometric function ${}_A F_B$.

Similar algorithms exist for other types of sums~\cite{Moch:2001zr}, which allow to perform the expansions of more general special functions. As an example, we briefly discuss one specific two-variable generalization of the Gauss hypergeometric function, the Appell function of the first kind, $F_1$. This function was introduced in chapter~\ref{chapter:complex}, see eq.~(\ref{eq:Appel-F1-def}). There it was defined by a double sum,
\begin{equation}
F_1(a,b_1,b_2;c,z_1,z_2) = 
    \sum_{m=0}^{\infty} \sum_{n=0}^{\infty}
    \frac{(a)_{m+n} (b_1)_m (b_2)_n }{(c)_{m+n} m! n!}
    z_1^m z_2^n.
\end{equation}
This double series is absolutely convergent for $|z_1|<1$ and $|z_2|<1$. In order to make contact with the language of nested sums, we rewrite the double sum in the following way. First, we separate the terms where $m$ or $n$ are zero, then we make the change of summation variable $n \to k=m+n$. Clearly the summation in $k$ runs from 2 to infinity, while the positivity of $n = k-m$ implies that the summation in $m$ runs from 1 to $k-1$. Then we find
\begin{equation}
    \begin{split}
    &
    F_1(a,b_1,b_2;c,z_1,z_2) = 
    1
\\ &\qquad
    + \frac{\Gamma(c)}{\Gamma(a)\Gamma(b_1)}
    \sum_{m=1}^{\infty} z_1^m 
    \frac{\Gamma(m+a)\Gamma(m+b_1)}{\Gamma(m+c)\Gamma(m+1)}
\\ &\qquad
    + \frac{\Gamma(c)}{\Gamma(a)\Gamma(b_2)}
    \sum_{n=1}^{\infty} z_2^n 
    \frac{\Gamma(n+a)\Gamma(n+b_2)}{\Gamma(n+c)\Gamma(n+1)}
\\ &\qquad
    +\frac{\Gamma(c)}{\Gamma(a)\Gamma(b_1)\Gamma(b_2)}
    \sum_{k=1}^{\infty} 
    \frac{\Gamma(k+a)}{\Gamma(k+c)}
    \sum_{m=1}^{k-1} 
    z_1^m \frac{\Gamma(m+b_1)}{\Gamma(m+1)}
    z_2^{k-m} \frac{\Gamma(k-m+b_2)}{\Gamma(k-m+1)}.
    \end{split}
\label{eq:F1-nested-sum}
\end{equation}
Let us assume that all parameters $a$, $b_1$, $b_2$ and $c$ are of the form $i+j\eps$, where $i$ is an integer. Then the expansion of the sums on the second line of the above equation can be performed with the algorithm just discussed. However, the nested sum on the fourth line is new. In order to perform the expansion of this sum, let us first concentrate on the inner sum over $m$. We again use eqs.~(\ref{eq:Gamma-exp-pos})~and~(\ref{eq:Zsum-invert}), convert the $S$-sums to $Z$-sums, syncronize the subsums and use algebra of $Z$-sums to write products of sums with the same upper limit of summation in terms of just single sums. Then we obtain sums of the form
\begin{equation}
    \sum_{m=1}^{k-1} \frac{z_1^m}{(m+o)^p}Z(m-1;m_1,\ldots)
    \frac{z_2^{k-m}}{(k-m+o')^{p'}} Z(k-m-1,m'_1,\ldots).
\end{equation}
After performing the partial fraction decomposition of the denominators (recall $p$ and $p'$ are fixed integers), the resulting expressions will involve either the denominator $(m+o)$ or $(k-m+o')$, but not both at the same time. Thus by changing the index of summation $m\to k-m$ in those terms that involve the denominator $(k-m+o')$, we can further reduce the above sum to sums of the type
\begin{equation}
    \sum_{m=1}^{k-1} \frac{z^m}{(m+o)^p}Z(m-1;m_1,\ldots)
    Z(k-m-1;m'_1,\ldots).
\label{eq:B-type-inter}
\end{equation}
If the depth of $Z(k-m-1;m'_1,\ldots)$ is zero, we simply recover a sum of the form
\begin{equation}
    \sum_{m=1}^{k-1} \frac{z^m}{(m+o)^p}Z(m-1;m_1,\ldots),
\end{equation}
with upper summation index $k-1$. We have already discussed that this sum can be evaluated in terms of $Z$-sums of the form $Z(k-1,\ldots)$. Thus, the double sum on the last line of eq.~(\ref{eq:F1-nested-sum}) is reduced to a sum of the form of eq.~(\ref{eq:A-type-sum}) which we can solve with algorithm A presented above. If the depth of $Z(k-m-1;m'_1,\ldots)$ is greater than one, we can write eq.~(\ref{eq:B-type-inter}) as 
\begin{equation}
    \sum_{n=1}^{k-1}\left[\sum_{m=1}^{n-1} 
    \frac{z^m}{(m+o)^p}Z(m-1;m_1,\ldots)
    \frac{(x'_1)^{n-m}}{(n-m)^{m'_1}}
    Z(n-m-1;m'_2,\ldots)\right].
\end{equation}
But this last sum is just of the form that we had started with, but the depth of the second $Z$-sum is reduced. Thus, we can use recursion to lower this depth to zero, when the recursion terminates. More generally, the algorithm just described (called algorithm B in~\cite{Moch:2001zr}) can be used to evaluate sums of the form
\begin{equation}
    \begin{split}
    &
        \sum_{m=1}^{k-1} \frac{x^m}{(m+o)^p}
        \frac{\Gamma(m+a_1+c_1\eps)}{\Gamma(m+b_1+d_1\eps)}
        \cdots
        \frac{\Gamma(m+a_k+c_k\eps)}{\Gamma(m+b_k+d_k\eps)} 
\\ &\qquad\times
        Z(m+r-1;m_1,\ldots,m_l;z_1,\ldots,z_l)
\\ &\qquad\times
        \frac{y^{k-m}}{(k-m+o')^{p'}}
        \frac{\Gamma(k-m+a'_1+c'_1\eps)}{\Gamma(k-m+b'_1+d'_1\eps)}
        \cdots
        \frac{\Gamma(k-m+a'_k+c'_k\eps)}{\Gamma(k-m+b'_k+d'_k\eps)}
\\ &\qquad\times
        Z(k-m+r'-1;m'_1,\ldots,m'_{l'};z'_1,\ldots,z'_{l'}).
    \end{split}
\end{equation}
Thus, the expansion of the first Appell function can be computed with the help of algorithms A and B. To finish, we note that further algorithms exist to handle also certain sums in which binomial coefficients appear. We refer the interested reader to the original literature for further details~\cite{Moch:2001zr}. Moreover, these algorithms have been implemented in the {\tt XSummer} package~\cite{Moch:2005uc} for the computer algebra system FORM. An adaptation of algorithm A described above, useful for the expansion of Gauss generalized hypergeometric functions ${}_A F_B$, has also been implemented in the {\tt HypExp} package for \math{}~\cite{Huber:2005yg}, see appendix~\ref{app:mplsgpls}. 

To finish, we reiterate that throughout we have assumed that the expansions are performed around integer values of the parameters of the hypergeometric functions and their generalizations. However, in practical calculations, sometimes special functions with half-integer parameters also appear. In this case, the algorithms described above cannot be applied in a straightforward manner and new algorithms are needed. In particular, the expansion of generalized hypergeometric functions of the form ${}_PF_{P-1}$ around half-integer values can be performed with the methods described in~\cite{Huber:2007dx}, which have been implemented in the {\tt HypExp~2} package for \math{}.

\begin{tips}{The Angular Integral with Two Denominators and One Mass}
Consider the angular integral with two denominators and one mass, eq.~(\ref{eq:Om-2-1m-fin}), with $j=k=1$. This integral appears e.g. when integrating over real soft radiation in a process with massive external legs. We can apply the techniques and tools described above to obtain the expansion of this integral around $\eps=0$. Using the {\tt XSummer} package, up to ${\mathcal O}(\eps^2)$ accuracy we find (see also \wwwaux{O11\_onemass})
\begin{equation}
    \begin{split}
    \Omega_{1,1}(v_{11},v_{12},\eps) &=
    -\Omega_{1-2\eps}
    \frac{\pi}{2v_{12}}
    \bigg\{
    \frac{1}{\eps}
    - \mathrm{Li}_1(V_{-}) 
    - \mathrm{Li}_1(V_{+})
\\ &    
    - \bigg[
    2 \mathrm{Li}_2(V_{-}) 
    + 2 \mathrm{Li}_2(V_{+}) 
    + \mathrm{Li}_{1,1}(1,V_{-}) 
    + \mathrm{Li}_{1,1}(1,V_{+})     
\\ &
    - \mathrm{Li}_{1,1}\bigg(\frac{V_{-}}{V_{+}},V_{+}\bigg)      
    - \mathrm{Li}_{1,1}\bigg(\frac{V_{+}}{V_{-}},V_{-}\bigg)      
    \bigg]\eps + {\mathcal O}(\eps^2)
    \bigg\},
    \end{split}
\label{eq:O2d1m-11-exp}
\end{equation}
%\commjg{aux file????}
where we have set $V_{\pm} = \frac{2v_{12}-1\pm \sqrt{1-4v_{11}}}{2v_{12}}$.
Furthermore, we have extracted a factor of $\Omega_{1-2\eps} = 2^{1 - 2 \eps} \pi^{-\eps} \frac{\gamma(1 - \eps)}{\gamma(1 - 2 \eps)}$, whose expansion would introduce irrelevant constants like $\ln(4\pi)$ and $\gamma$ into the final expression. In order to apply {\tt XSummer}, we have used the nested sums representation of the Appell $F_1$ function as given in eq.~(\ref{eq:F1-nested-sum}).
We note in passing that the result in eq.~(\ref{eq:O2d1m-11-exp}) can be expressed with just the logarithm and dilogarithm functions as follows,
\begin{equation}
    \begin{split}
    &
    \Omega_{1,1}(v_{11},v_{12},\eps) =
    -\Omega_{1-2\eps}
    \frac{\pi}{2v_{12}}
    \bigg\{
    \frac{1}{\eps}
    + \ln\frac{v_{11}}{v_{12}^2}
    - \bigg[
    \frac{1}{2} \ln^2\bigg(\frac{1-\sqrt{1-4v_{11}}}{1+\sqrt{1-4v_{11}}}\bigg)
\\ &\qquad   
    + 2 \mathrm{Li}_2\bigg(\frac{2v_{12}-1 - \sqrt{1-4v_{11}}}{2v_{12}}\bigg)
    + 2 \mathrm{Li}_2\bigg(\frac{2v_{12}-1 + \sqrt{1-4v_{11}}}{2v_{12}}\bigg)
    \bigg]\eps 
%\\ &   
    + {\mathcal O}(\eps^2)
    \bigg\}.
    \end{split}
\label{eq:O2d1m-11-exp-simp}
\end{equation}

\end{tips}

\subsection{More General Sums, {Massive Propagators} 
\label{sec:moregensums}}

Up to this point, we have been dealing with \MB{} integrals where the integration variables enter the gamma functions only in the form $\Gamma(\ldots\pm a z), \; a=1$, but not with more general $a$ coefficients. Indeed, as alluded to at the beginning of this Chapter, most of the tools that we have discussed are most useful in this case. However, more general \MB{} integrals certainly arise in practical calculations ($a \neq 1$ appears in case of massive propagators and as we will see leads to series with binomial sums~\cite{Davydychev:2003mv}). In order to illustrate the complications that arise in this case, consider the self-energy diagram of Fig.~\ref{self-energy} in chapter~\ref{chapter-MBrepr}, where the internal propagators are massive, with mass $m$. The corresponding \MB{} representation is given in eq.~(\ref{MB-SE1l2m-lemmas}). Let us consider the first (constant) term in the $\epsilon$ expansion of this integral, %\added{Is $1/(2\pi i)$ missing???}
\bq
 G(1)_{\text{SE2l2m}} \equiv I_1 =  \int_{-1/4-i \infty}^{-1/4+i \infty} \frac{d z_1}{2 \pi i} \left(-\frac{m^2}{s}\right)^{z}  \frac{\Gamma(1 - z)^2 \Gamma(-z) \Gamma(z)}{ \Gamma(2 - 2 z)}. \label{eq:se1l2m-again}
\eq
(This is just the function $aux(\epsilon = 0, \Re(z)=-1/4)$ of eq.~(\ref{eq:auxbasic}).) Notice that the gamma function in the denominator has the argument $2-2z$, i.e., the coefficient of the integration variable in the argument is $-2$.

Let us now apply Cauchy's residue theorem. For the sake of being explicit, we assume that $|m^2/s| < 1$ and close the contour to the right. The first few residues of the integrand read
\begin{eqnarray}
 \mathrm{Res}_{z=0}\; \left(-\frac{m^2}{s}\right)^{z}  \frac{\Gamma(1 - z)^2 \Gamma(-z) \Gamma(z)}{ \Gamma(2 - 2 z)} &=& 
 -2 - \ln\left(-\frac{m^2}{s}\right),
 \\
 \mathrm{Res}_{z=1}\; \left(-\frac{m^2}{s}\right)^{z}  \frac{\Gamma(1 - z)^2 \Gamma(-z) \Gamma(z)}{ \Gamma(2 - 2 z)} &=& \frac{2m^2}{s}
 \left[-1 + \ln\left(-\frac{m^2}{s}\right)\right], 
 \\
 \mathrm{Res}_{z=2}\; \left(-\frac{m^2}{s}\right)^{z}  \frac{\Gamma(1 - z)^2 \Gamma(-z) \Gamma(z)}{ \Gamma(2 - 2 z)} &=& \frac{m^4}{s^2}
 \left[1 + 2\ln\left(-\frac{m^2}{s}\right)\right].
 \end{eqnarray}
The residue for a general non-negative integer can be written as follows,\footnote{For $n=0$, the correct residue, $-2 - \ln\left(-\frac{m^2}{s}\right)$, is obtained by computing the limit as $n\to 0$ of the expression in eq.~(\ref{eq:SE212m-res-n}).}
\begin{equation}
    \begin{split}
    &
        \mathrm{Res}_{z=n} \left(-\frac{m^2}{s}\right)^{z}  \frac{\Gamma(1 - z)^2 \Gamma(-z) \Gamma(z)}{ \Gamma(2 - 2 z)}
        = 
        2 \bigg(\frac{m^2}{s}\bigg)^n
        \frac{\Gamma(2n-1)}{\Gamma(n)\Gamma(n+1)}
    \\&\qquad\times
        \left[\ln\left(-\frac{m^2}{s}\right) - \psi^{(0)}(n) 
        - \psi^{(0)}(n+1) + 2 \psi^{(0)}(2n-1)\right],
        \qquad n\in \mathbb{N}.
    \end{split}
    \label{eq:SE212m-res-n}
\end{equation}
{\it However, upon attempting to perform the summation over residues, we immediately face a problem: the summand cannot be reduced to a rational expression in $n$, and so cannot be brought to the form of a $Z$-sum.} 

This is evident from appearance of the overall factor of $\frac{\Gamma(2n-1)}{\Gamma(n)\Gamma(n+1)}$. Using eq.~(\ref{AsymptoticnGamma}), it is easy to show that this factor grows exponentially as $n\to \infty$ and so it cannot correspond to a rational function of $n$. Thus, the solution cannot be expressed as a combination of polylogarithms of argument $(m^2/s)$.

In order to get an idea of the form of the solution, let us examine the sum proportional to $\ln(-m^2/s)$ in eq.~(\ref{eq:SE212m-res-n}). This sum can be evaluated in closed form:
\begin{equation}
    \sum_{n=0}^{\infty} 2 \bigg(\frac{m^2}{s}\bigg)^n \frac{\Gamma(2n-1)}{\Gamma(n)\Gamma(n+1)}\ln\left(-\frac{m^2}{s}\right)  =
    -\sqrt{1 - \frac{4m^2}{s}} \ln\left(-\frac{m^2}{s}\right).
\label{eq:gamma2n-sum}
\end{equation}
\begin{tips}{Fractional and Inverse Binomial Sums}
Let us evaluate the sum
\begin{equation}
    \sum_{n=1}^{\infty} \frac{\Gamma(2n-1)}{\Gamma(n)\Gamma(n+1)} x^n.
\label{eq:binom-sum}
\end{equation}
We begin by writing
\begin{equation}
\begin{split}
    \sum_{n=1}^{\infty} \frac{\Gamma(2n-1)}{\Gamma(n)\Gamma(n-1)}  x^n &=
    \sum_{n=1}^{\infty} \frac{(2n-2)!}{(n-1)! n!} x^n
\\&=
    \sum_{n=1}^{\infty} \frac{1\cdot 2\cdot 3\cdot 4\cdot\ldots\cdot(2n-3)\cdot(2n-2)}{(n-1)! n!} x^n
\\&=
    \sum_{n=1}^{\infty} 2^{2n-2}\frac{\frac{1}{2}\cdot 1\cdot \frac{3}{2}\cdot 2\cdot\ldots\cdot\frac{2n-3}{2}\cdot(n-1)}{(n-1)! n!} x^n.
\end{split}
\label{eq:binom-sum-1}
\end{equation}
Written in this way, it is clear that we can cancel a factor of $(n-1)!$ between the numerator and denominator. Then we have
\begin{equation}
\begin{split}
    &
    \sum_{n=1}^{\infty} 2^{2n-2}\frac{\frac{1}{2}\cdot 1\cdot \frac{3}{2}\cdot 2\cdot\ldots\cdot\frac{2n-3}{2}\cdot(n-1)}{(n-1)! n!} x^n 
\\ &\qquad=
    \frac{1}{4}\sum_{n=1}^{\infty} (-1)^{n-1}\frac{\left(-\frac{1}{2}\right)\cdot \left(-\frac{3}{2}\right)\cdot\ldots\cdot\left(-\frac{2n-3}{2}\right)}{n!} (4x)^n
\\ &\qquad=
    -\frac{1}{2}\sum_{n=1}^{\infty} \frac{\frac{1}{2}\cdot\left(-\frac{1}{2}\right)\cdot \left(-\frac{3}{2}\right)\cdot\ldots\cdot\left(-\frac{2n-3}{2}\right)}{n!} (-4x)^n    
\\ &\qquad=
    -\frac{1}{2}\sum_{n=1}^{\infty} \frac{\frac{1}{2}\cdot\left(\frac{1}{2}-1\right)\cdot \left(\frac{1}{2}-2\right)\cdot\ldots\cdot\left(\frac{1}{2}-n+1\right)}{n!} (-4x)^n.    
\end{split}
\label{eq:binom-sum-2}
\end{equation}
The sum on the last line can now be computed using the binomial theorem for fractional exponents, see eq.~(\ref{eq:powersum2}). Examining the sum in eq.~(\ref{eq:binom-sum-2}), we see that apart from missing the $n=0$ term (which is equal to one), it is just the generalized binomial sum of eq.~(\ref{eq:powersum2}) with $\alpha=\frac{1}{2}$. Thus
\begin{equation}
\begin{split}
    &
    -\frac{1}{2}\sum_{n=1}^{\infty} \frac{\frac{1}{2}\cdot\left(\frac{1}{2}-1\right)\cdot \left(\frac{1}{2}-2\right)\cdot\ldots\cdot\left(\frac{1}{2}-n+1\right)}{n!} (-4x)^n =
    -\frac{1}{2}\left(\sqrt{1-4x}-1\right).
\end{split}
\label{eq:binom-sum-3}
\end{equation}
Hence, we have established that
\begin{equation}
    \sum_{n=1}^{\infty} \frac{\Gamma(2n-1)}{\Gamma(n)\Gamma(n+1)} x^n 
    =
    \frac{1}{2} - \frac{1}{2}\sqrt{1-4x}.
\label{eq:binom-sum-res}
\end{equation}
A straightforward application of this result then leads immediately to 
eq.~(\ref{eq:gamma2n-sum}).

We note in passing that if we choose to close the contour of integration to the left in eq.~(\ref{eq:se1l2m-again}), the residue of the integrand at a generic negative integer $z=-n$, $n\in\mathbb{N}^+$ reads
\begin{equation}
    \begin{split}
        \mathrm{Res}_{z=-n} \left(-\frac{m^2}{s}\right)^{z}  \frac{\Gamma(1 - z)^2 \Gamma(-z) \Gamma(z)}{ \Gamma(2 - 2 z)}
        &= 
       	\bigg(\frac{m^2}{s}\bigg)^{-n}
        \frac{\Gamma(n)\Gamma(n+1)}{\Gamma(2n+2)}
    \\&=
       	\bigg(\frac{m^2}{s}\bigg)^{-n}
        \frac{1}{(2n+1)n\binom{2n}{n}}    	
	,
        \qquad -n\in \mathbb{N}^+.
    \end{split}
    \label{eq:SE212m-res-mn}
\end{equation}
After performing partial fractioning in $n$, we see that the solution involves \emph{inverse binomial sums} such as
\begin{equation}
\sum_{n=1}^{\infty} \frac{1}{n\binom{2n}{n}} x^n
\qquad\mbox{and}\qquad
\sum_{n=1}^{\infty} \frac{1}{(2n+1)\binom{2n}{n}} x^n.
\label{eq:inv-binom-sums}
\end{equation}
Again, these sums are not of a form that we can directly treat with the algorithms presented above for nested sums.

In order to illustrate the one typical way of dealing with such sums, consider e.g. the sum
\begin{equation}
f(x) = \sum_{n=0}^{\infty} \frac{1}{(2n+1)\binom{2n}{n}} x^n = \sum_{n=0}^{\infty} c_n x^n,
\label{eq:binom-sum-ex}
\end{equation}
which appeared already in eq.~(\ref{eq:anal3}). Let us denote the $n$-th expansion coefficient as $c_n = \frac{1}{\binom{2n}{n}}$ and examine the $n+1$-st coefficient, $c_{n+1}$. It is easy to show using the definition of the binomial coefficient that $c_{n+1} = \frac{n+1}{2(2n+3)}c_n$, which implies the recursion relation
\begin{equation}
4(n+1)c_{n+1} + 2c_{n+1} = n c_n + c_n.
\end{equation}
But then, multiplying this equation by $x^n$ and summing over $n$ form zero to infinity leads to the differential equation
\begin{equation}
4f'(x) + \frac{2}{x}[f(x)-1] = x f'(x) + f(x).
\label{eq:inv-binom-de}
\end{equation}
Here we have used that sums involving the terms $(n+1)c_{n+1}$ and $n c_n$ can be related to the derivative function $f'(x)$ 
and also that $c_0=1$. Eq.~(\ref{eq:inv-binom-de}) is an ordinary first order inhomogeneous differential equation, and hence it can be reduced to integrations using the method of variation of constants. The integrals that appear are elementary and we find
\begin{equation}
f(x) = \frac{2}{\sqrt{x(4-x)}}\left[\pi
	- 2 \arctan\left(\frac{\sqrt{4-x}}{\sqrt{x}}\right)\right].
\end{equation}
Interestingly, the constant of integration is fixed by requiring that the expansion of the solution around $x=0$ should only involve integer powers of $x$. The other sum in eq.~(\ref{eq:inv-binom-sums}) can be handled in a similar fashion. We see that again square roots appear in the solution, that can be simplified by using the conformal variable of eq.~(\ref{eq:conformal}).

Another approach to evaluating inverse binomial sums starts from the observation that the general inverse binomial coefficient admits the integral representation
\begin{equation}
{\binom{n}{k}}^{-1} = (n+1)\int_0^1 dt\,t^k(1-t)^{n-k}. 
\end{equation}
Then e.g. for the particular sum in eq.~(\ref{eq:binom-sum-ex}) we find
\begin{equation}
\begin{split}
\sum_{n=0}^{\infty} \frac{1}{(2n+1)\binom{2n}{n}} x^n &= 
\sum_{n=0}^{\infty} \frac{1}{(2n+1)} (2n+1)\int_0^1 dt\,t^n (1-t)^n x^n
\\ &=
 \int_0^1 \frac{dt}{1-t(1-t)x},
\end{split}
\end{equation}
where we have assumed that we can interchange the order of summation and integration. We can now try to perform the integration in $t$ to obtain the solution. We will not pursue this example any further here, but we will have more to say on how such parametric integrals can be evaluated in section~\ref{ssec:symbolic-int}.

We finish by noting that efficient methods computing for certain classes of sums involving inverse binomials have been implemented in the \texttt{IBIS} package~\cite{vanHoegaerden:2025nmh}. The method is based on recursion relations and yields results in the form of $S$-sums, where the dependence on the upper limit of the sum is analytic. The class of sums that can be handled by \texttt{IBIS} can also be computed with the general-purpose summation codes \texttt{Sigma} and \texttt{EvaluateMultiSums}, albeit \texttt{IBIS}, being specialized to these cases, is typically much faster.

\end{tips}
Thus, we generally expect the expression $\sqrt{1-\frac{4m^2}{s}}$ to appear in the solution. Let us then introduce (with some foresight) the so-called conformal variable
\begin{equation}
    y = \frac{\sqrt{1-4 m^2/s}-1}{\sqrt{1-4 m^2/s}+1}.
\label{eq:conformal}
\end{equation}
Then $\left(-\frac{m^2}{s}\right)$ becomes $\frac{y}{(1-y)^2}$. We can get an idea for the structure of the final result by computing the sum of the first few residues, expressed as a series in $y$ around zero. It is straightforward to perform this computation in \math{} (recall we are closing the contour in the clockwise direction),
\begin{minted}[frame=single,breaklines,fontsize=\small]{mathematica}
In[6]:= I1 = (y/(1-y)^2)^z Gamma[1-z]^2 Gamma[-z] Gamma[z]/Gamma[2-2z];
In[7]:= sum = Sum[-Residue[I1, {z, i}], {i, 0, 10}];  
In[8]:= series = Normal[Series[sum, {y, 0, 10}]] // Simplify
Out[8]:= 2 + (1 + 2 y + 2 y^2 + 2 y^3 + 2 y^4 + 2 y^5 + 2 y^6 + 2 y^7 + 2 y^8 + 2 y^9 + 2 y^10) Log[y]
\end{minted}
Evidently part of the result builds a geometric series, which for $|y|<1$ gives a final result
\begin{equation}
    G(1)_{\text{SE2l2m}} = I_1 = 2 + \frac{1+y}{1-y} \ln y.
\label{eq:G1SE2l2m-sol}
\end{equation}
The condition $|y|<1$ is satisfied for any values of $m,s$, see Fig.~\ref{fig:conf}. This kind of changes of variable is also very useful in calculation of Feynman integrals using the differential equations approach~\cite{Kotikov:1990kg,Kotikov:1991hm,Kotikov:1991pm,Remiddi:1997ny,Henn:2013pwa,Czakon:2004wm}.

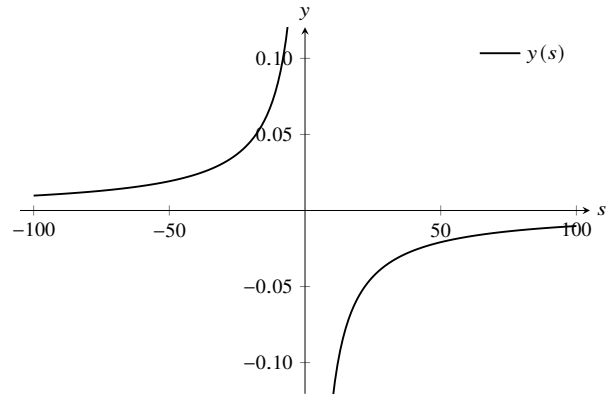
\begin{figure}[tbph]
%%\vspace*{7mm}
%%%%\hspace*{4mm}
\begin{center}
\begin{tikzpicture}[scale=0.9]
  \begin{axis}[
      axis lines=middle,
        width=10cm, height=7cm,
        xmin=-105,xmax=105,
        ymin=-0.1205,ymax=0.1205,
        label style={font=\footnotesize},
        ticklabel style={font=\footnotesize},
        xlabel=$s$,
        ylabel=$X$,
        x label style={at={(axis description cs:1,0.5)},anchor=west},
        y label style={at={(axis description cs:0.5,1)},anchor=south},
        legend style={draw=none},
        legend style={font=\footnotesize},
%        minor x tick num = 4,
%        minor y tick num = 4,
        ytick={-0.1,-0.05,0.05,0.1},
        yticklabels={$-0.10$, $-0.05$, $0.05$, $0.10$},
      ]

    \addplot[line width=0.8pt, domain=1:100,samples=100] {(sqrt((x-4)/x)-1)/(sqrt((x-4)/x)+1)};
    \addplot[line width=0.8pt, domain=-100:-1,samples=100] {(sqrt((x-4)/x)-1)/(sqrt((x-4)/x)+1)};
%    \addplot[
%      line width=0.8pt]
%       {sqrt(x)/100};
%    \legend{$Re\mathrm{Li}_2(x)$,$Im\mathrm{Li}_2(x)$}
    \legend{$y(s)$}
  \end{axis}
\end{tikzpicture}
\end{center}
\caption[]{
The conformal transformation of eq.~(\ref{eq:conformal}). Since $y$ only depens on the ratio of $m^2/s$, we have set $m=1$ in the plot. The conformal variable $y$ approaches $\pm 1$ as $s \to 0^{\pm}$ and $0$ as $s \to \infty$.}
\label{fig:conf}
\end{figure}

The ideas discussed above can be applied to other cases as well. For example, the the vertex function of eq.~(\ref{parts}),
\begin{eqnarray}
V^{\epsilon^{-1}}_{V3l2m}(s) &=& 
-
% typos corrected 2015-01-21 (Johann Usovitsch)
 ~ \frac{1}{2s} 
\int\limits_{-\frac{1}{2}-i \infty}^{-\frac{1}{2}+i \infty} 
\frac{dz}{2\pi i}
{(-s)^{-z}}
{\frac{\Gamma^3(-z)\Gamma(1+z)} {\Gamma(-2z)}}, \label{eq:V3l2m-again}
\end{eqnarray}
considered in section~\ref{sec:simpleinvitation}, can be computed in a very similar way. Setting $s = -\frac{(1-y)^2}{y}$, which corresponds to the conformal transformation of eq.~(\ref{eq:conformal}) with $m=1$, we again compute the first several residues and examine the expansion in $y$ around zero. Once more, we are able to identify part of this expansion as a geometric series and we find the result
\begin{equation}
V^{\epsilon^{-1}}_{V3l2m}(s) = -\frac{y \ln y}{1-y^2}.
\label{eq:V3l2m-res}
\end{equation}
 
\begin{tips}{Analytic Solutions With a Matching Procedure}
In our examples above, we have relied in the last step on identifying parts of expansions with geometic series. There is however a more general way of approaching this problem, which can also be used in more complicated cases. This relies on building a proper ansatz for the solution. Then, we can `sum up' such series in an easy way by constructing a general solution, expanding it and matching the coefficients of the specific sum with the coefficients of the this general solution.

Choosing a proper ansatz requires some understanding of the structure of the result. For example, in the case of the self-energy integral which we discussed above, we must realize that the basic variable entering the ansatz should be the conformal variable of eq.~(\ref{eq:conformal}). Furthermore, we expect that a logarithmic function will play a pivotal role, since as a general rule, we know that such functions appear in one-loop computations. However, if we do not know the precise form of the solution, we should include more general functions as well, trying to guess the possible complete set of functions and coefficients, which can be monomials, polynomials or rational functions in the conformal variables. Even if we choose a wide base of functions, the redundant ones should disappear when matching the series to the unknown sum, as we will see below.

So for the self-energy integral of eq.~(\ref{eq:se1l2m-again}), let us consider  the minimal set $S_1$ of functions 
\bq 
S_1 = \{ 1, \ln(y), \ln(1-y), \ln(1+y)\}, \label{eq:set1}
\eq 
with $y$ being the conformal variable of eq.~(\ref{eq:conformal}). We also define the minimal set of monomials
\bq 
S_2 = \{ 1, y, \frac{1}{1-y}, \frac{1}{1+y} \}. \label{eq:set2}
\eq 
Taking the Cartesian product of $S_1$ and $S_2$, we obtain the following {\it{basis}}
\begin{equation}
    \begin{split}
    S_B = S_1 \times S_2 = 
    \bigg\{&1, y, \frac{1}{1-y}, \frac{1}{1+y},
    \ln(y), y \ln(y), \frac{\ln(y)}{1-y}, \frac{\ln(y)}{1+y},
\\ &    
    \ln(1-y), y \ln(1-y), \frac{\ln(1-y)}{1-y}, \frac{\ln(1-y)}{1+y},
\\\ &    
    \ln(1+y), y \ln(1+y), \frac{\ln(1+y)}{1-y}, \frac{\ln(1+y)}{1+y}
    \bigg\},
    \end{split}
\label{eq:setB}
\end{equation}
which contains 16 terms, i=1,\ldots,16. Our ansatz is then simply a linear combination of these 16 terms with undetermined coefficients $c[i]$,
\begin{equation}
{\rm Ansatz}=\sum_{i=1}^{16} c[i] S_B[i].
\label{eq:ansatz1}
\end{equation}
In order to determine the $c[i]$, we first Taylor expand our ansatz in $y$ around zero. Since we must determine 16 unknowns, this expansion should include at least the first 16 terms. Second, we take the sum of residues\footnote{The sum of residues must include enough terms such that all terms in the Taylor expansion to the required order are generated in the following step.} of the unknown function $G(1)_{\text{SE2l2m}}$ in eq.~(\ref{eq:se1l2m-again}), and Taylor expand this too in $y$ around zero. Next, we compare terms of the same order $y$ in the two expansions. Since the function $\ln(y)$ does not have a Taylor expansion around zero, terms of the form $y^k \ln(y)$ will also appear in addition to just powers like $y^k$. The coefficients of terms of the same order in $y$ that are also proportional to $\ln(y)$ can be matched separately from those terms that are free of the logarithm. In this way, we solve for the unknown unknown coefficients $c[i]$. If the constructed basis $S_b$ covers properly the space of functions which are necessary to build the solution of $G(1)_{\text{SE2l2m}}$, then the $c[i]$ will have unique solutions. Thus finally, we obtain a solution for the ansatz in eq.~(\ref{eq:ansatz1}) and so for $G(1)_{\text{SE2l2m}}$. 

The above actions can be written in \math{} as follows

\begin{minted}[frame=single,breaklines,fontsize=\small]{mathematica}
In[9]:= eq = Array[c, Length[basis]].basis; 
In[10]:= serEq1 = Normal[Series[eq, {y, 0, 16}]] // Simplify; 
In[11]:= sumresid = Sum[-Residue[fun, {z, i}], {i, 0, 16}];
In[12]:= serEq2 = Normal[Series[sumresid, {y, 0, 16}]] // Simplify;  
In[13]:= nologs = serEq1 - serEq2; /. a_. Log[y]^_. -> 0;   
In[14]:= logs = serEq1 - serEq2 - nologs // Simplify;  
In[15]:= tab1 = Table[0 == Coefficient[nologs,y,i-1],{i,16}]; 
In[16]:= coeff1 = Solve[tab1, Array[c, 16]]  
Out[16]:= {{c[1] -> 2, c[2] -> 0, c[3] -> 0, c[4] -> 0, c[9] -> 0, c[10] -> 0, c[11] -> 0, c[12] -> 0, c[13] -> 0, c[14] -> 0, c[15] -> 0, c[16] -> 0}}
In[17]:= rule1 = coeff1[[1]]; 
In[18]:= tab2 = Table[0 == Coefficient[logs,y,i-1],{i,16}]; Solve[tab2, Array[c, 16]] 
Out[18]:= {{c[5] -> -1, c[6] -> 0, c[7] -> 2, c[8] -> 0}} 
In[19]:= rule2 = coeff2[[1]]; 
\end{minted}
 {In the code above the `basis' refers to eq.~(\ref{eq:setB}), `fun' is (the integrand of) our unknown function $G(1)_{\text{SE2l2m}}$ in eq.~(\ref{eq:se1l2m-again}). The complete \math{} file is available in \wwwaux{SE2l2m}.}

Finally, we can apply the rules derived above, {\tt rule1} and {\tt rule2}, to the ansatz in eq.~(\ref{eq:ansatz1}). Doing so, we obtain the solution
\begin{minted}[frame=single,breaklines,fontsize=\small]{mathematica}
In[20]:= eq /. Join[rule1, rule2] // Simplify
Out[20]:= 2 - ((1 + y) Log[y])/(-1 + y)
\end{minted}
Clearly the solution obtained agrees with the one we found in eq.~(\ref{eq:G1SE2l2m-sol}).

To finish, let us make two final comments. First, after obtaining the solution in the way described above, it is of course possible to Taylor expand both the ansatz and the unknown sum to some order in $y$ that is higher than what was used to fix the solution. Verifying that the two expansions continue to match also at higher orders serves as a useful check of the result. Second, it can also happen that we do not find a solution for the $c[i]$ in eq.~(\ref{eq:ansatz1}). This indicates that our ansatz was incorrect or incomplete. In this case, we may try to modify or enlarge our basis and search for a solution with the new ansatz. {Or we are unlucky and the solution goes beyond the known classes of functions.}

\end{tips}

\subsection{\mb{} sums with conic hulls, contour geometry and triangulations \label{sec:conic}} 

Recently a new insight for finding convergent sums generated with the \mb{} method was developed in \cite{Banik:2022bmk,Banik:2023rrz,Banik:2024jqj}.
The authors introduce two non-iterative geometrical methods for the analytic evaluation of multifold Mellin--Barnes (MB) integrals: the \emph{conic-hull method} and the \emph{triangulation method}. Both associate geometrical objects to the $\Gamma$-function structure of the MB integrand and combine them with multivariate residue calculus to construct hypergeometric-type series representations. 

For the conic-hull approach, one first selects nonsingular $N$-tuples of numerator Gamma functions in the \mb{} integral representations and computes the corresponding residue ``building blocks''. To each such block one associates a conic hull generated by the coefficient vectors appearing in the Gamma functions. The largest intersections of these conic hulls identify compatible sets of building blocks and therefore yield series representations of the original \mb{} integral. The common intersection, or \emph{master conic hull}, can additionally be used to infer convergence properties of the resulting master series. The method was subsequently extended to arbitrary straight contours. Using the generalized reflection formula $\Gamma(z-n)=\frac{\Gamma(z)\,\Gamma(1-z)\,(-1)^n}{\Gamma(n+1-z)}$, the integrand can be transformed so that the real parts of the relevant numerator Gamma functions are positive on the contours. The same conic-hull construction can then be applied. 

A more efficient formulation is obtained from the observation that the relevant intersections of conic hulls are dual to \emph{regular triangulations} of suitable point configurations. After transforming the \mb{} integral to a canonical form containing $\Gamma(-z_1)\cdots\Gamma(-z_N)$ structures, the coefficient vectors of the remaining Gamma functions define a finite configuration of points. Regular triangulations of this configuration are in one-to-one correspondence with the relevant series representations. Thus, triangulations replace the expensive search through intersections of many conic hulls by a combinatorial geometric problem. The main computational advantage is substantial for high-fold \mb{} integrals. For example, the massless off-shell conformal hexagon has a $9$-fold \mb{} representation with $194160$ different series representations. The triangulation approach can explore such large sets much more efficiently than the original conic-hull construction. 

The methods described above have been implemented in the Mathematica package \texttt{MBConicHulls.wl}, using \texttt{TOPCOM} \cite{topcom} for triangulations. In summary, the analytic strategy can be represented schematically as 
\\
\[ \boxed{ \begin{gathered} \text{MB integrand} \;\hookrightarrow\; \text{coefficient vectors} \;\hookrightarrow\; \begin{cases} \text{conic hulls}\\[-1mm] \text{point configurations} \end{cases} \\[1mm] \hookrightarrow\; \text{intersections / triangulations} \;\hookrightarrow\; \text{multivariate residues} \\[1mm] \;\hookrightarrow\; \text{series representations} \end{gathered} } \] 
\\
The two constructions are equivalent at the level of the resulting series representations, while the triangulation formulation is computationally much faster for complicated multifold MB integrals.

The above ideas are discussed in detail in a recent talk by S.~Banik \cite{banikWGFCCtalk}. The talk also addresses the problem of so-called ``white zones'', namely, regions of parameter space in which none of the series representations obtained from a direct evaluation of the \mb{} integrals converge. Several analytical and numerical approaches to overcoming this problem are discussed. The methods developed in this framework have already found practical applications, including the treatment of MB integrals with polygamma functions \cite{Banik:2025kwa} and sunset integrals with up to three mass scales in chiral perturbation theory \cite{Ananthanarayan:2025rti}.

\section{Decoupling Integrals Through a Change of Variables}
\label{sec:decoupl} 

In our case decoupling of \mb{} representations or integrals means their decomposition into a product of two or more integrals with lower dimensionality.

One of the simplest examples can be found for example in eq.~(10) of~\cite{Czakon:2005rk}. It corresponds to an $1 / \epsilon^2$ pole of a planar QED double-box. The integral is only two-dimensional so decoupling is visible by eyes. One also should point here that decoupled 1-dimensional integrals up to a coefficient have same structure as the integral in eq.~(\ref{parts}). In its turn eq.~(\ref{parts}) corresponds to $1 / \epsilon$ pole of the one-loop QED vertex.

A more complicated example is given in eq.~(\ref{eq:decoupl_1}) and corresponds to a diagram in Fig.~\ref{fig:decoupl_1}.
\begin{figure}[!h]
\begin{center}
    \includegraphics[scale=0.45]{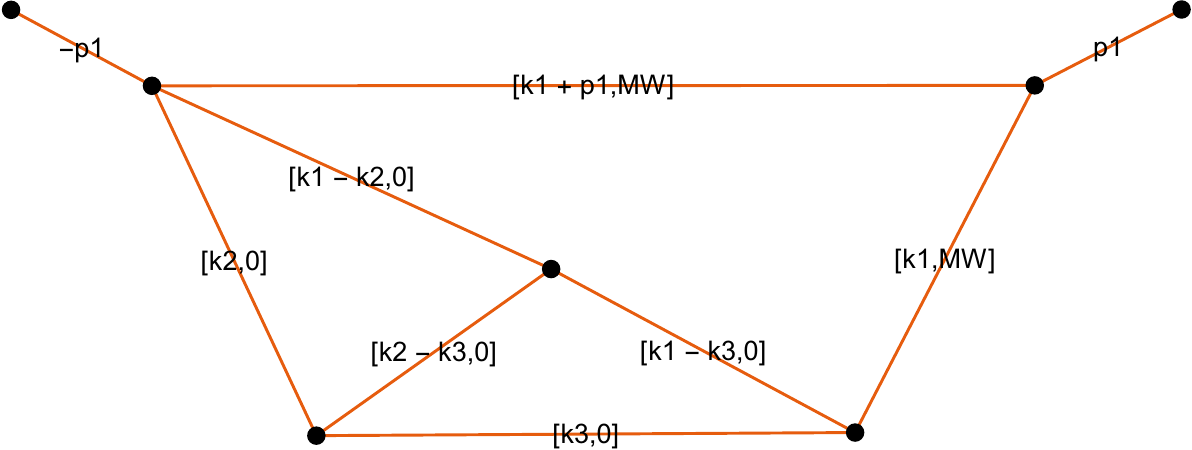}
\end{center}
    \caption{A diagram corresponding to the integral in eq.~(\ref{eq:decoupl_1}).}
    \label{fig:decoupl_1}
\end{figure}

\begin{multline}
    I = -\frac{1}{(-s)^{1+3\epsilon}} \int_{- i \infty}^{+ i \infty} \prod_{i=1}^4 \frac{d z_i}{2 \pi i} \,\, \left(- \frac{M_W^2}{s} \right)^{z_3}
    \frac{\Gamma(-\epsilon-z_1)\Gamma(-z_1)\Gamma(1+2\epsilon+z_1)}
    {\Gamma(1-2\epsilon)\Gamma(1-3\epsilon-z_1)} \\
    \times \frac{\Gamma(-2\epsilon-z_{12})\Gamma(1-\epsilon+z_2)\Gamma(1+z_{12})\Gamma(1+\epsilon+z_{12})\Gamma(1+3\epsilon+z_3)
    \Gamma(1-\epsilon-z_4)}
    {\Gamma(1-z_2)\Gamma(2+\epsilon+z_{12})} \\
    \times \frac{\Gamma(-\epsilon-z_2)\Gamma(-z_2)\Gamma(1+z_3-z_4)\Gamma(-z_4)\Gamma(-z_3+z_4)\Gamma(-3\epsilon-z_3+z_4)}
    {\Gamma(1-4\epsilon-z_3)\Gamma(2+2\epsilon+z_3-z_4)}.
\label{eq:decoupl_1}    
\end{multline}

It is not clear from the diagram, but decoupling happens due to a presence of a massless two-loop self-energy subdiagram corresponding to the loop momenta $k_1$ and $k_2$. In this particular case decoupling was obtained automatically by a proper choose of order of integration in \la{} approach. 

In general case, decoupled parts can be mixed up and one needs to establish a procedure to check, if the decoupling is possible in principle. To answer this question, let us again consider the matrix representation used in section~\ref{sec:BLeff} to study the possibility of \mb{} integrals decoupling. 
Decoupling means that the matrix representation should have a block form
\begin{equation}
    M_{\Gamma} \longrightarrow 
    \begin{pmatrix}
        A & 0 \\
        0 & B
\end{pmatrix}
\end{equation}
Transition to a block form is possible if it is possible to find a suitable transformation of the integration variables. For a faster search of the transformation, one can consider a smaller matrix $M_{\Gamma} \rightarrow \bar M_{\Gamma}$ where we deleted all duplicated lines including lines that differ by overall sing. Now, elements of transformation matrix $U = \{x_{ij}\}$, see eq. (\ref{eq:Utransform}), should fulfill the following conditions 
\begin{equation}
\bar M_{\Gamma} =  \begin{pmatrix} \bar M^{(1)}_{\Gamma} \\ \bar M^{(2)}_{\Gamma} \end{pmatrix}, \,\,\,\,\, 
\begin{matrix} \bar M^{(1)}_{\Gamma} X^{(1)} = 0 \\ M^{(2)}_{\Gamma} X^{(2)} = 0 \end{matrix},
\label{eq:DecEq}
\end{equation}
where we assume that the first system of linear equations corresponds to the upper 0 in the block form, and the second system - to the bottom 0. Solutions will form columns of the transformation
matrix $U = \begin{pmatrix} X^{(2)} & X^{(1)} \end{pmatrix}$. 

Since the matrix $U$ must be reversible, only linearly independent solutions should be taken. To have the total number of independent solutions equal to the dimensionality of matrix $U$,  one can formulate the following condition
\begin{equation}
    \mbox{rank}(M^{(1)}_{\Gamma}) + \mbox{rank}(M^{(2)}_{\Gamma}) = \mbox{dim}
    \label{DecoupCond}
\end{equation}
where {\tt dim} is the dimensionality of the corresponding \mb{} integral.

Now a test of decoupling possibility can be formulated in a simple algorithmic way: i) construct matrix $\bar M_{\Gamma}$; ii) divide it into two submatrices in all possible ways; iii) check the condition in eq.~(\ref{DecoupCond}). An instructive example for eq.~(\ref{eq:decoupl_1}) is given in \wwwaux{Decoupling}.

\section{Solving via Integration} 
\label{sec:3anal}  
One can obtain analytic solutions for certain types of Mellin-Barnes 
integrals by transforming them to Euler-type integrals over real 
variables and evaluating the latter. In the present subsection, we 
explore this method.

\subsection{From Mellin-Barnes to Euler Integrals \label{subsec:EulerInt}}

To begin, we will assume that the Mellin-Barnes integral we are evaluating is balanced. As we have discussed above, this is not a real limitation for practical applications. Moreover it is straightforward to derive a real Euler-type integral representation for balanced Mellin-Barnes integrals as we will see in a moment. However, this representation will only be convergent if the arguments of all gamma functions have positive real parts. So in the following, let us assume that the integration contours are all straight lines parallel to the imaginary axis such that the real parts of the arguments of all gamma functions are positive. (We note that in dimensional regularization such contours may not exist for $\ep=0$.) In this case, one can derive an Euler-type integral representation by noticing first of all that since the integral is balanced, its integrand can be expressed as a product of beta functions
\begin{equation}
    B(a,b) = \frac{\Gamma(a)\Gamma(b)}{\Gamma(a+b)}\,.
\label{eq:Beta-def}
\end{equation}
For example, consider the integral
\begin{equation}
    I = \int_{-i\infty-\frac12}^{+i\infty-\frac12} \frac{dz_1}{2\pi i}
    \int_{-i\infty-\frac14}^{+i\infty-\frac14} \frac{dz_2}{2\pi i}  
    \Gamma(-z_1)\Gamma(-z_2) \Gamma(1+z_1+z_2) A^{z_1} B^{z_2}\,.
\end{equation}
Clearly, the contours are such that the argument of each gamma 
function has a positive real part. Now, notice that, e.g.
\begin{equation}
    \Gamma(-z_2) \Gamma(1+z_1+z_2) = \Gamma(1+z_1) B(-z_2,1+z_+z_2)\,,
\end{equation}
so we have
\begin{equation}
    I = \int_{-i\infty-\frac12}^{+i\infty-\frac12} \frac{dz_1}{2\pi i}
    \int_{-i\infty-\frac14}^{+i\infty-\frac14} \frac{dz_2}{2\pi i}  
    \Gamma(-z_1)\Gamma(1+z_1) B(-z_2,1+z_1+z_2) A^{z_1} B^{z_2}\,.
\end{equation}
Similarly, using
\begin{equation}
    \Gamma(-z_1)\Gamma(1+z_1) = \Gamma(1) B(-z_1,1+z_1)
\end{equation}
we find
\begin{equation}
    I = \int_{-i\infty-\frac12}^{+i\infty-\frac12} \frac{dz_1}{2\pi i}
    \int_{-i\infty-\frac14}^{+i\infty-\frac14} \frac{dz_2}{2\pi i}  
    B(-z_1,1+z_1) B(-z_2,1+z_1+z_2) A^{z_1} B^{z_2}\,.
\end{equation}
The point is to ``pair'' gamma functions such that at least one 
integration variable cancels in the sum of arguments and replace 
such a pair by a corresponding beta function,
\begin{equation}
    \Gamma(a \pm z)\Gamma(b \mp z) = \Gamma(a+b)B(a \pm z,b \mp z)\,.
\end{equation}
Because the integral is balanced by assumption, this pairing can 
be continued until all gamma functions whose arguments contain 
integration variables have been exhausted. Notice that ratios 
of gamma functions can be replaced using
\begin{equation}
    \frac{\Gamma(a \pm z)}{\Gamma(b \pm z)} = 
    \frac{1}{\Gamma(b-a)}B(b-a,a \pm z)\,.
\end{equation}
In this case, we are looking to pair a gamma function in the 
numerator with a gamma function in the denominator such that the 
difference of their arguments is free of at least one integration 
variable.

Next, we use the real integral representation of $B(a,b)$,\footnote{We note that in practice one can very often also use the alternative integral representation of the beta function,
\begin{equation}
    B(a,b) = \int_0^1 dx\, x^{a - 1}(1-x)^{b - 1},
    \qquad 
    \Re(a)\,,\Re(b) > 0
    \label{eq:beta-int-one}
\end{equation}
with limits of integration zero and one. Most of the discussion that follows goes through without change in this case, but the appearance of boundary conditions can lead to certain technical subtleties.}
\begin{equation}
    B(a,b) = \int_0^{\infty} dx\, x^{a - 1}(1+x)^{-a-b}\,,
    \qquad 
    \Re(a)\,,\Re(b) > 0
    \label{eq:beta-int-inf}
\end{equation}
and write all beta functions in the integrand of the Mellin-Barnes 
integral using one of these representations. Both integrals 
converge when $\Re(a)\,,\Re(b)>0$ and it is easy to check that 
these conditions are satisfied whenever the real parts of all 
arguments of gamma functions are positive in the original 
Mellin-Barnes integral. Because all integrals are convergent, 
we can exchange the Mellin-Barnes integrations and the real 
integrations coming form the beta functions. Returning to our 
example above, we obtain
\begin{equation}
\begin{split}
    I &= \int_{-i\infty-\frac12}^{+i\infty-\frac12} \frac{dz_1}{2\pi i}
    \int_{-i\infty-\frac14}^{+i\infty-\frac14} \frac{dz_2}{2\pi i}  
    B(-z_1,1+z_1) B(-z_2,1+z_1+z_2) A^{z_1} B^{z_2}
\\ &=
    \int_{-i\infty-\frac12}^{+i\infty-\frac12} \frac{dz_1}{2\pi i}
    \int_{-i\infty-\frac14}^{+i\infty-\frac14} \frac{dz_2}{2\pi i} 
    \int_0^1 dx_1\, x_1^{-z_1-1}(1-x_1)^{z_1}
\\ &\qquad\times
    \int_0^1 dx_2\, x_2^{-z_2-1}(1-x_2)^{z_1+z_2} A^{z_1} B^{z_2}
\\ &=
    \int_0^1 dx_1\, \int_0^1 dx_2\, \frac{1}{x_1 x_2}
    \int_{-i\infty-\frac12}^{+i\infty-\frac12} \frac{dz_1}{2\pi i}
    \int_{-i\infty-\frac14}^{+i\infty-\frac14} \frac{dz_2}{2\pi i}
    \left[\frac{(1-x_1)(1-x_2)A}{x_1}\right]^{z_1}
\\ &\qquad\times
    \left[\frac{(1-x_2)B}{x_2}\right]^{z_2}\,.
\end{split}
\label{eq:MB-to-int-ex}
\end{equation}
In general, after this step, one finds an integral of the form
\begin{equation}
    \int_0^1 \left(\prod_{k=1}^{K} dx_k\right)
    R_0(\mathbf{x},\mathbf{v})
    \left[R_{\epsilon}(\mathbf{x},\mathbf{v})\right]^{\epsilon}
    \int_{-i\infty}^{+i\infty}
    \prod_{l=1}^{L} \frac{dz_l}{2\pi i}
    \left[R_l(\mathbf{x},\mathbf{v})\right]^{z_l}\,,
\label{eq:Euler-0}
\end{equation}
where $\mathbf{x} = (x_1,\ldots,x_K)$ are $K$ real integration 
variables, $\mathbf{v} = (v_1,\ldots,v_N)$ are the variables of 
the original Mellin-Barnes integral and the $R_l$ are ratios of 
products of $x_i$, $(1+x_i)$ and $v_j$. Although obvious, we note 
that $R_{\epsilon}$ is only present if the original Mellin-Barnes
integral depends on $\epsilon$. Hence, if we have already expanded 
the Mellin-Barnes integral in $\epsilon$ and are considering only 
the expansion coefficients, $R_{\epsilon}$ is absent.

Next, we need to perform the Mellin-Barnes integrations. This can 
be done using the formula
\begin{equation}
    \int_{-i\infty+z_0}^{+i\infty+z_0} 
    \frac{dz}{2\pi i} A^z = \delta(1-A)\,,
    \qquad A>0\,.
\label{eq:MBtoDelta}
\end{equation}
Importantly, in the Euclidean regime, where all variables $v_j$ 
are positive by definition, each $R_l$ is clearly positive for 
$0<x_i<\infty$, $i=1,\ldots,K$. Thus, eq.~(\ref{eq:Euler-0}) takes the form
\begin{equation}
    \int_0^\infty \left(\prod_{k=1}^{K} dx_k\right)
    R_0(\mathbf{x},\mathbf{v})
    \left[R_{\epsilon}(\mathbf{x},\mathbf{v})\right]^{\epsilon}
    \prod_{l=1}^{L}
    \delta\left[1-R_l(\mathbf{x},\mathbf{v})\right]\,.
\label{eq:Euler-1}
\end{equation}
By solving the constraints implied by the $\delta$ functions, 
we arrive at the desired parametric integral representation. 
Continuing with our example in eq.~(\ref{eq:MB-to-int-ex}), we find
\begin{equation}
    I = 
    \int_0^1 dx_1\, \int_0^1 dx_2\, \frac{1}{x_1 x_2}
    \delta\left[1-\frac{(1-x_1)(1-x_2)A}{x_1}\right]
    \delta\left[1-\frac{(1-x_2)B}{x_2}\right]\,.
\end{equation}
The second $\delta$ function implies $x_2 = B/(1+B)$ and so 
performing the integration over $x_2$ we obtain
\begin{equation}
    I = 
    \int_0^1 dx_1\, \frac{1}{(1+B)x_1}
    \delta\left[1-\frac{(1-x_1)A}{(1+B)x_1}\right]\,,
\end{equation}
where we have used
\begin{equation}
    \delta\left[1-\frac{(1-x_2)B}{x_2}\right] =
    \frac{B}{(1+B)^2}\delta\left(x_2 - \frac{B}{1+B}\right)\,.
\end{equation}
The remaining $\delta$ function now fixes $x_1 = A/(1+A+B)$ and 
making use of
\begin{equation}
    \delta\left[1-\frac{(1-x_1)A}{(1+B)x_1}\right] =
    \frac{A(1+B)}{(1+A+B)^2}\delta\left(x_1 - \frac{A}{1+A+B}\right)\,,
\end{equation}
we finally obtain a zero-dimensional integral representation
\begin{equation}
    I = 
    \int_0^1 dx_1\, \frac{A}{(1+A+B)^2 x_1}
    \delta\left(x_1 - \frac{A}{1+A+B}\right)
    =
    \frac{1}{1+A+B}\,.
\end{equation}
Of course, in general, one obtains a finite dimensional real
integral representation, which must be solved to obtain the desired result. One way of implementing the above manipulations in \math{} is presented in \wwwaux{Euler}.

Finally, we mention that if the \MB{} integral involves polygamma functions (this will typically happen if we expand in $\eps$ first, and apply the above construction to the \MB{} integrals that appear in the expansion coefficients), we can use eq.~(\ref{eq:polygamma-int-inf}),
\begin{equation}
    \psi(z) = \int_0^\infty [(1+x)^{-1}-(1+x)^{-z}]\frac{dx}{x} - \gamma,\qquad \Re(z)>0,
    \label{eq:polygamma-int-inf-again}
\end{equation}
to express also the polygamma functions as real integrals.\footnote{As with the beta function, the integral representation over the finite interval $[0,1]$ of eq.~(\ref{eq:polygamma-int-one}),
\begin{equation}
    \psi(z) = \int_0^1 \frac{t^{z-1}-1}{t-1}dt - \gamma,\qquad \Re(z)>0,
    \label{eq:polygamma-int-one}
\end{equation}
can also very often be used in practice, with the same caveat as for the beta function.} The integral representation converges if the real part of the argument of the polygamma function is positive.

\begin{overview}{Mellin-Barnes Integrals to Euler-Type Integrals: an Overview}

Given a balanced Mellin-Barnes integral such that the 
real parts of the arguments of all gamma functions are 
positive, one can construct an Euler-type integral representation 
for the Mellin-Barnes integral as follows.
\begin{enumerate}
    \item Express the integrand as a product of beta functions. 
    To do so, pair products (ratios) of two gamma functions 
    such that at least one Mellin-Barnes integration variable 
    cancels in the sum (difference) of their arguments. Continue 
    this process until no gamma function remains whose argument 
    contains an integration variable. Since the original Mellin-Barnes 
    integral is balanced by assumption, this can always be achieved.
    
    \item Replace all beta and polygamma functions by their integral 
    representations as given in 
    eqs.~(\ref{eq:beta-int-inf})~and~(\ref{eq:polygamma-int-inf-again}). 
    The real integrations introduced in this step are all convergent if 
    the real parts of the arguments of all gamma and polygamma functions 
    in the original Mellin-Barnes integral are positive.
    
    \item Perform all Mellin-Barnes integrations using the result 
    of eq.~(\ref{eq:MBtoDelta}). In the Euclidean region, the necessary
    condition is always fulfilled.
    
    \item Finally, solve the constraints implied by the $\delta$ 
    functions to obtain the desired real parametric representation.
\end{enumerate}
\end{overview}

\subsection{Symbolic Integration of Euler Integrals}
\label{ssec:symbolic-int}

Let us now turn to briefly describing how the parametric integrals that we obtain can be computed. More precisely, we consider an integral of the following form,
\begin{equation}
    I(\mathbf{v};\epsilon) = \int_0^\infty \left(\prod_{k=1}^{n} dx_k\right)
    \prod_{l=1}^{m} P_l(\mathbf{x},\mathbf{v})^{a_l+b_l \epsilon}    
\label{eq:Euler-gen}
\end{equation}
where again, $\mathbf{x} = (x_1,\ldots,x_n)$ are $n$ real integration variables and $\mathbf{v} = (v_1,\ldots,v_m)$ are $m$ parameters fixed during the integration. We assume that $a_l$ and $b_l$ are integers and the $P_l(\mathbf{x},\mathbf{v})$ are polynomials that cannot be factorized further into products of lower degree polynomials. We assume that the limits of integration are zero and infinity. It is not very difficult to see that eq.~(\ref{eq:Euler-1}) can generally be written as a sum of such integrals. We assume that the integral is convergent for $\epsilon=0$, so that we can expand the integrand in $\epsilon$ before integration. This is not always true, as in practical calculations the integrand sometimes develops non-integrable singularities when $\epsilon=0$, typically when some of the $x_i$ are zero (however, more elaborate structure of singularities can also show up). In such cases, we must first resolve the poles of the integral in $\epsilon$, e.g. by subtracting the divergences. The detailed discussion of how this can be achieved for general integrals is beyond the scope of our discussion and we point the reader to the literature~\cite{Heinrich:2008si,Panzer:2015ida}.

After performing the Taylor-expansion of the integrand in eq.~(\ref{eq:Euler-gen}) in $\epsilon$, the expansion coefficients will involve rational functions (recall that the $a_l$ can also be negative) and (powers of) logarithms of polynomials in the integration variables. Our goal is then to integrate out the $x_i$ one after the other. This involves iterated integrations of rational functions and logarithms and hence we expect that multiple polylogarithms, introduced in chapter~\ref{chapter:complex}, will show up in the solution. We emphasize that of course not every integral can be solved in this way in terms of multiple polylogarithms, however this is still a powerful method for obtaining solutions in many situations. In fact, one can state a sufficient condition, first derived in~\cite{Brown:2008um}, on the integrand, which if satisfied, allows one to express the integral in terms of multiple polylogarithms. To do so, let us start by defining $S$ as the set of all polynomials that are not monomials which appear in the integrand in eq.~(\ref{eq:Euler-gen}),
\begin{equation}
    S = \left\{P_l(\mathbf{x},\mathbf{v})\right\}.
\end{equation}
Recalling the definition of multiple polylogarithms as iterated integrals in eq.~(\ref{eq:G-def}),
\begin{equation}
G(a_1,\ldots,a_n;z) = \int_0^z \frac{dt}{t-a_1} G(a_2,\ldots,a_n;t)\,,
\label{eq:G-def-again}
\end{equation}
we see that in order to begin the integration, we have to assume that that there exists some integration variable, say $x_a$, such that all elements of $S$ are linear in this variable. In this case, let us write each of the $P_l$'s as
\begin{equation}
    P_l(\mathbf{x},\mathbf{v}) = Q_l^a(\mathbf{x}',\mathbf{v}) x_a 
    + R_l^a(\mathbf{x}',\mathbf{v}),
\end{equation}
where $\mathbf{x}' = \{x_1,\ldots,x_{a-1},x_{a},\ldots,x_{n}\}$ denotes the set of integration variables \emph{without} $x_a$. Of course, after the $\epsilon$ expansion, logarithms of $P_l$ will also appear. These can be written in terms of multiple polylogarithms as
\begin{equation}
\begin{split}
    \ln P_l &= 
    \ln \left(Q_l^a x_a + R_l^a\right)
    =
    \ln R_l^a + \ln  \left(1 + \frac{Q_l^a}{R_l^a} x_a\right)
\\ &=
    \ln R_l^a + G\left(-\frac{R_l^a}{Q_l^a};x_a\right).
\end{split}
\end{equation}
Moreover, products of multiple polylogarithms can be expressed in terms of functions of the form $G(a_1(\mathbf{x}',\mathbf{v}),\ldots,a_l(\mathbf{x}',\mathbf{v});x_a)$ using the shuffle algebra of multiple polylogarithms, e.g.,
\begin{equation}
    G\left(-\frac{R_k^a}{Q_k^a};x_a\right) G\left(-\frac{R_l^a}{Q_l^a};x_a\right)
    =
    G\left(-\frac{R_k^a}{Q_k^a},-\frac{R_l^a}{Q_l^a};x_a\right)
+
    G\left(-\frac{R_l^a}{Q_l^a},-\frac{R_k^a}{Q_k^a};x_a\right).
\end{equation}
Thus, the integration over $x_a$ involves (sums of) integrals of the form
\begin{equation}
    \int_0^\infty 
    \frac{dx_a}{(Q_1^a x_a + R_1^a)^{-a_1}\ldots(Q_m^a x_a + R_m^a)^{-a_m}}
    G(\vec{a};x_a).
\end{equation}
In order to compute the integral, we partial fraction in $x_a$, e.g.,
\begin{equation}
    \frac{1}{(Q_k^a x_a + R_k^a)(Q_l^a x_a + R_l^a)} =
    \frac{1}{Q_k^a R_l^a - Q_l^a R_k^a}
    \left(\frac{1}{x_a + R_k^a/Q_k^a} - \frac{1}{x_a + R_l^a/Q_kl^a}\right),
\end{equation}
then simply apply the recursive definition of multiple polylogarithems in eq.~(\ref{eq:G-def-again}) to obtain a primitive. Finally, we must compute the limits of this primitive as $x_a\to 0$ and $x_a\to\infty$.

We would like to iterate this process and integrate over the rest of the variables one by one. However, to do this we must once again find an integration variable, say $x_b$ in which all polynomials in the new integrand are linear. These polynomials are nothing but the $Q_l^a$ and $R_l^a$ introduced above, \emph{as well as} the combinations $Q_k^a R_l^a - Q_l^a R_k^a$ that are introduced by partial fractioning. This last polynomial is not necessarily linear, even if the $Q_k^a$ and $R_l^a$ are. Thus, in order to proceed, it is necessary that the $Q_k^a R_l^a - Q_l^a R_k^a$ factor into polynomials that are linear in some integration variable. This analysis must be repeated after each step of the integration. In order to formalize this procedure, we can define $S_{(x_a)}$ as the set of all irreducible factors that appear inside the polynomials $Q_l^a$, $R_l^a$ and $Q_k^a R_l^a - Q_l^a R_k^a$. Now, if we can find an integration variable, say $x_b$, such that all elements of $S_{(x_a)}$ are linear in $x_b$, we can restart the above procedure and integrate over $x_b$. If we can iterate this procedure and construct a sequence of sets of polynomials,
\begin{equation}
    S_{(x_a)},\,S_{(x_a,x_b)},\,S_{(x_a,x_b,x_c)},\,\ldots
    \label{eq:S-sets-seq}
\end{equation}
such that all polynomials in each set are linear in at least one integration variable, then we can perform the integrations one by one and express the result in terms of multiple polylogarithms. {\it We stress three important facts.} 
\begin{enumerate}
    \item 
First, the existence of such a sequence of sets in general depends on the ordering of the integration variables. Of course, it is enough to find one ordering where all polynomials in each set are linear in at least one integration variable. 
\item Second, it can be shown that this condition, called \emph{linear reducibility}, is sufficient but not necessary for the integral to be expressible in terms of multiple polylogarithms. In fact, even if we fail to find a suitable sequence of sets, the result may still be given in terms of multiple polylogarithms, e.g. after a suitable change of variables. In addition, it is not necessary that the factors are linear for the last integration step, since in this case we can factor the polynomials into linear factors whose roots involve algebraic expressions of the variables $\mathbf{v}$.
\item Third, notice that we can check if such a sequence of sets exists without actually performing any integrations. We remark though that in practical applications, it is rarely necessary to actually construct this sequence of sets: many times several orderings of variables can be chosen, some more convenient than others (e.g. leading to smaller intermediate expressions) and it is often quite simple to proceed by inspection.
\end{enumerate}
We now have a criterion for determining if a given integral can be evaluated in terms of multiple polylogarithms, however in order to actually proceed and compute the result, we must address the following points. Let us assume that we have found a proper ordering of integration variables and have performed the (indefinite) integration over $x_a$. Then in order to continue and perform the next integration, we must
\begin{enumerate}
    \item take the limits of the primitive with respect to $x_a$ as $x_a\to 0$ and $x_a \to \infty$,
    \item find a way to write all multiple polylogarithms in the form $G(\vec{a},x_b)$, where $\vec{a}$ is \emph{independent} of $x_b$. This is non-trivial, as the arguments of multiple polylogarithms of the form $G\left(\ldots,-R_l^a/Q_l^a,\ldots\right)$ may well depend on $x_b$. 
\end{enumerate}

Regarding the first point, the limit at $x_a\to 0$ can easily be taken by using the shuffle algebra to extract pure logarithms in $x_a$ (see the discussion in chapter~\ref{chapter:complex}), then using the sum representation of multiple polylogarithms, eq.~(\ref{eq:Li-def}). In fact, in many situations, it is enough to consider that
\begin{equation}
    \lim_{x_a\to 0} G(\vec{a};x_a) = 0,\qquad \vec{a} \ne \vec{0}.
\label{eq:MPL-lim-0}
\end{equation}
Turning to the limit at $x_a\to\infty$, we note that by introducing $\bar{x}_a = 1/x_a$, this problem can be reduced to taking the limit $\bar{x}_a \to  0$,
\begin{equation}
   \lim_{x_a\to\infty} G(\vec{a};x_a) 
   =
   \lim_{\bar{x}_a\to 0} G\left(\vec{a};\frac{1}{\bar{x}_a}\right). 
 \end{equation}
If we can find a way of writing the multiple polylogarithms on the right hand side as linear combinations of multiple polylogarithms of the form $G(\ldots;\bar{x}_a)$, then taking the $\bar{x}_a \to  0$ limit can be performed as discussed above. At this point, it is worth noticing that the problem of finding the appropriate inversion relations is formally equivalent to the second issue that we have not yet addressed, namely, how to bring the multiple polylogarithms entering the integrand to a form where $x_b$ only appears in the last position, as $G(\vec{a};x_b)$. In both cases we are looking for functional equations that bring the multiple polylogarithms into ``standard form'', where a certain variable only enters as the explicit argument but not in the parameters.

{\it The framework necessary to address this problem involves Hopf algebra structure of multiple polylogarithms.} The description of this topic is beyond the scope of this book, however, we mention one main result. If the sufficient condition described above regarding the existence of the sequence of sets in eq.~(\ref{eq:S-sets-seq}) is met and \emph{furthermore} the integration ragnes are $[0,\infty]$, it can be shown that the multiple polylogarithms that appear can always be brought to a form
\begin{equation}
    \sum_i c_i G(\vec{a}_i;x),
\end{equation}
for some variable $x$ such that $\vec{a}_i$ is independent of $x$ and the coefficidents $c_i$ involve only multiple polylogarithms that are independent of $x$. Furthermore, this rewriting can be performed in a constructive algorithmic way. The appropriate algorithms are implemented e.g. in the \math{} package {\tt PolyLogToos} and the {\tt Maple} program {\tt HyperInt}, see appendix~\ref{app:mplsgpls}.

We end our discussion with a brief comment about generic limits of integration. Although many statements described above remain true for generic integration boundaries, nevertheless in certain cases, some algorithms do break down due to the appearance of boundary terms. However, a generic region of integration $x\in[a,b]$ can always be mapped to $y\in[0,\infty]$ by the change of variables
\begin{equation}
    y = \frac{x-b}{x-a}.
\end{equation}
Thus from a purely formal point of view, when constricting Euler integral representations of \MB{} integrals, we should use the representation of the beta function as an integral between zero and infinity. However, from a practical standpoint, we can very often proceed with the representation based on the integration over $[0,1]$ as well.

\subsection{Merging \mb{} Integrals with Euler Integrals and
Symbolic Integration}
\label{sec:MBtoEulertoInt}

Finally, as an illustration of the above procedure, let us discuss a simple example. Consider the phase space integral
\begin{equation}
    I = \int d\Omega_{d-1}(k)\,d\Omega_{d-1}(l)\,
    \frac{1}{p_1\cdot (k + l)},
\label{eq:I-ex-mb-int}
\end{equation}
where $p_1$ is fixed massless four-vectors, while $k$ and $l$ are massless vectors whose directions we integrte over with the measures $d\Omega_{d-1}(k)$ and $d\Omega_{d-1}(k)$ that are given in eq.~(\ref{eq:dOmq-def}). Such integrands appear for example in factorization formulae describing the double soft limits of QCD real-emission matrix elements. We can easily relate this integral to the phase space integrals studied in chapter~\ref{chapter:complex} by applying the basic \MB{} formula to the denominator,
\begin{equation}
    I = \int d\Omega_{d-1}(k)\,d\Omega_{d-1}(l)\,
    \int_{-i\infty}^{+i\infty} \frac{dz}{2\pi i}
    \Gamma(-z)\Gamma(1+z) (p_1\cdot k)^{z} (p_1\cdot l)^{-1-z}.
\end{equation}
Now, let us exchange the order of integrations and then apply the eq.~(\ref{eq:Om-j-0}),
\begin{equation}
    \int d\Omega_{d-1}(q) \frac{1}{(p_1\cdot q)^j} = \Omega_j(0,\epsilon) = 2^{2-j-2\epsilon}\pi^{1-\epsilon}
    \frac{\Gamma(1-j-\epsilon)}{\Gamma(2-j-2\epsilon)},
\label{eq:Om-j-0-again}
\end{equation}
to perform the angular integrations.\footnote{Eq.~(\ref{eq:Om-j-0-again}) is valid provided that all vectors in the integrand are normalized such that their zeroth components are one. It is obviously trivial to rescale the integrand in eq.~(\ref{eq:I-ex-mb-int}) to achieve this and we assume that this has been done.} We obtain
\begin{align}
    I &= \int_{-i\infty}^{+i\infty} \frac{dz}{2\pi i}
    \Gamma(-z)\Gamma(1+z)
    \int d\Omega_{d-1}(k)\frac{1}{(p_1\cdot k)^{-z}}\,
    d\Omega_{d-1}(l)\frac{1}{(p_1\cdot l)^{1+z}}\,
\notag \\ &=
    2^{3-4\epsilon}\pi^{2-2\epsilon}
    \int_{-i\infty}^{+i\infty} \frac{dz}{2\pi i}
    \Gamma(-z)\Gamma(1+z)
    \frac{\Gamma(1+z-\epsilon)}{\Gamma(2+z-2\epsilon)}
    \frac{\Gamma(-z-\epsilon)}{\Gamma(1-z-2\epsilon)}.
\label{eq:I-ex-mb-int-2}
\end{align}
Now we must find a contour for the $z$ integration, such that the poles of the gamma functions with $\Gamma(\ldots+z)$ and those with $\Gamma(\ldots-z)$ are separated. In this case, it is quite straightforward to find a suitable straight-line contour parallel to the imaginary axis. By simple inspection, we see that for $\epsilon=0$, we may choose any $z_0 \in (-1,0)$ and the contour running from $z_0 - i\infty$ to $z_0 + i\infty$ will separate the poles as required. As we have found a suitable contour directly at $\epsilon=0$, there is no need to analytically continue to $\epsilon \to 0$.

Having derived the \MB{} representation, let us now turn obtaining an analytic solution by converting it into an Euler-type integral and integrating that. To do this, we can combine the factors of $\Gamma(-z)$ and $\Gamma(1+z)$ in the numerator using\footnote{The choice of contour is such that the real parts of the arguments of all gamma functions in eq.~(\ref{eq:I-ex-mb-int-2}) are positive, thus all integral representations we write will automatically converge.}
\begin{equation}
    \Gamma(-z)\Gamma(1+z) = 
    \Gamma(1) \int_0^\infty dx_1\, x_1^{-1 - z} (1 + x_1)^{-1},
\end{equation}
while the two ratios of gamma functions can be written as
\begin{align}
    \frac{\Gamma(1+z-\epsilon)}{\Gamma(2+z-2\epsilon)} &=
    \frac{1}{\Gamma(1 - \epsilon)}
    \int_0^\infty dx_2\, x_2^{z-\epsilon} (1 + x_2)^{-2 + 2 \epsilon - z},
\\
    \frac{\Gamma(-z-\epsilon)}{\Gamma(1-z-2\epsilon)} &=
    \frac{1}{\Gamma(1 - \epsilon)}
    \int_0^\infty dx_3\, x_3^{-1 - \epsilon - z} (1 + x_3)^{-1 + 2 \epsilon + z}.
\end{align}
Then we obtain
\begin{equation}
    \begin{split}
    I &= \frac{2^{3-4\epsilon}\pi^{2-2\epsilon}}
    {\Gamma^2(1-\epsilon)}
    \int_{-i\infty}^{+i\infty} \frac{dz}{2\pi i}
    \int_0^\infty dx_1\, \int_0^\infty dx_2\, \int_0^\infty dx_3 
\\ &\times
    \frac{x_2^{-\epsilon } (1+x_2)^{2
   \epsilon -2} x_3^{-\epsilon -1} (1+x_3)^{2 \epsilon -1}}
   {x_1 (1+x_1) \Gamma (1-\epsilon )^2}
   \left(\frac{x_2 (1+x_3)}{x_1 (1+x_2) x_3}\right)^{z}.
    \end{split}
\end{equation}
Next, we perform the \MB{} integration using eq.~(\ref{eq:MBtoDelta}). Notice that the quantity being raised to the power $z$, $\frac{x_2 (1+x_3)}{x_1 (1+x_2) x_3}$, is positive in the integration region, essentially by construction. Then we find
\begin{equation}
    \begin{split}
    I &= \frac{2^{3-4\epsilon}\pi^{2-2\epsilon}}
    {\Gamma^2(1-\epsilon)}
    \int_0^\infty dx_1\, \int_0^\infty dx_2\, \int_0^\infty dx_3\,
    \frac{x_2^{-\epsilon } (1+x_2)^{2
   \epsilon -2} x_3^{-\epsilon -1} (1+x_3)^{2 \epsilon -1}}
   {x_1 (1+x_1) \Gamma (1-\epsilon )^2}
\\ &\times
   \delta\left(1-\frac{x_2 (1+x_3)}{x_1 (1+x_2) x_3}\right).
    \end{split}
\end{equation}
Finally, we must resolve the Dirac $\delta$ function by performing one integration. Solving the constraint for $x_1$ yields
\begin{equation}
    x_1 = \frac{x_2 (1 + x_3)}{(1 + x_2) x_3}
\end{equation}
and we find the Euler-type integral representation
\begin{equation}
    I = \frac{2^{3-4 \epsilon} \pi^{2-2 \epsilon}}
    {\Gamma^2(1-\epsilon)}
    \int_0^\infty dx_2\, \int_0^\infty dx_3\,
    \frac{x_2^{-\epsilon}(1+x_2)^{-1+2 \epsilon} x_3^{-\epsilon } (1+x_3)^{-1+2 \epsilon}}{x_2 + x_3+ 2 x_2 x_3}.
\end{equation}
Obviously this representation is not unique: it depends on the particular way in which we combine the gamma functions into beta functions as well as on the choice of which variable we solve the Dirac $\delta$ constraint for.
\begin{svgraybox}
Of course, we may choose to solve the Dirac $\delta$ constraint for any of the variables. However, some choices are less ideal than others. To see why, consider solving for $x_2$ in our example instead of $x_1$. We find
\begin{equation}
    x_2 = \frac{x_1 x_3}{1 + x_3 - x_1 x_3}.
\end{equation}
But $x_2$ must be non-negative in the integration region, which leads to the constraint $1 + x_3 - x_1 x_3 \ge 0$. If $0\le x_1 \le 1$, this is clearly satisfied for any non-negative $x_3$, however for $x_1 > 1$ it implies $x_3 \le \frac{1}{x_1 - 1}$. We must then be careful to implement this constraint in the subsequent integrations, in this case e.g. by splitting the integration in $x_1$ at one. Notice that the choice of solving for $x_1$ instead completely circumvents such issues.
\end{svgraybox}

Now we are ready to solve the remaining real integrations. To do so, we first note that the result is finite as $\epsilon\to 0$, we we may simply expand in $\epsilon$ under the integral sign. Let us consider the leading term of this expansion, i.e., the finite part,
\begin{equation}
    I = 8\pi^2
    \int_0^\infty dx_2\, \int_0^\infty dx_3\,
    \frac{1}{(1+x_2)(1+x_3)(x_2 + x_3+ 2 x_2 x_3)}
    + {\mathcal O}(\epsilon).
\label{eq:I-ex-mb-int-3}
\end{equation}
It can be checked that this integral is linearly reducible, and hence expressible with multiple polylogarithms (at some specific arguments). The order of integrations clearly cannot matter in this example, since the integrand is symmetric under $x_2 \leftrightarrow x_3$ exchange. Thus, let us partial fraction the integrand in $x_2$,
\begin{equation}
    \frac{1}{(1+x_2)(1+x_3)(x_2 + x_3+ 2 x_2 x_3)} = 
    -\frac{1}{(1+x_3)^2 (1+x_2)} 
    +\frac{1}{(1+x_3)^2 [x_3/(1+2x_3) + x_2]}.
\end{equation}
It is straightforward to compute a primitive with respect $x_2$ now and we find
\begin{equation}
\begin{split}
    &
    \int dx_2\,\left[
    -\frac{1}{(1+x_3)^2 (1+x_2)} 
    +\frac{1}{(1+x_3)^2 [x_3/(1+2x_3) + x_2]}
    \right]
\\ &= 
    -\frac{1}{(1+x_3)^2} G(-1;x_2)
    +\frac{1}{(1+x_3)^2} G\left(-\frac{x_3}{1+2x_3};x_2\right).
\end{split}
\end{equation}
Although the integration is of course elementary, we have expressed the result with multiple polylogarithms in order to illustrate the general case. Next, we must evaluate this primitive in the limits $x_2 \to 0$ and $x_2 \to \infty$. The limit at $x_2=0$ is trivial: using eq.~(\ref{eq:MPL-lim-0}) we see that in the limit, the primitive is simply zero. In order to find the limit at $x_2 \to \infty$, let us set $x_2 = 1/\bar{x}_2$ and consider $\bar{x}_2 \to 0$. This requires writing the multiple polylogarithms $G\left(-1;\frac{1}{\bar{x}_2}\right)$ 
and $G\left(-\frac{x_3}{1+2x_3};\frac{1}{\bar{x}_2}\right)$ in terms of multiple polylogarithms of the form $G(\ldots;\bar{x}_2)$. As discussed above, the procedures for how this can be done in general go beyond the scope of this book, however for the specific case at hand, we may proceed simply by recalling that a weight one multiple polylogarithm is simply an ordinary logarithm in disguise, $G(a;z) = \ln\left(1-\frac{z}{a}\right)$, if $a\ne 0$. So we have
\begin{equation}
    G\left(-1;\frac{1}{\bar{x}_2}\right) = 
    \ln\left(1+\frac{1}{\bar{x}_2}\right) 
    =
    \ln(1+\bar{x}_2) - \ln(\bar{x}_2)
    =
    G(-1;\bar{x}_2) - G(0;\bar{x}_2)
\end{equation}
and similarly
\begin{equation}
\begin{split}
    &
    G\left(-\frac{x_3}{1+2x_3};\frac{1}{\bar{x}_2}\right) = 
    \ln\left(1+\frac{1+2x_3}{\bar{x}_2 t3}\right) 
\\ &=
    \ln(1 + 2 x_3) 
    + \ln\left(1 + \frac{x_3}{1 + 2 x_3}\bar{x}_2\right) 
    - \ln(x_3) 
    - \ln(\bar{x}_2)
\\ &=
    G\left(-\frac{1}{2};x_3\right)
    + G\left(-\frac{1+2x_3}{x_3};\bar{x}_2\right)
    -G(0;x_3)
    -G(0;\bar{x}_2).
\end{split}
\end{equation}
Using the above, it is straightforward to compute the limit of the primitive at $\bar{x}_2 \to 0$. Notice in particular that terms involving $G(0;\bar{x}_2) = \ln (\bar{x}_2)$ cancel and the limit is finite. Thus we obtain
\begin{equation}
\begin{split}
    &
    \int_0^\infty dx_2\,\left[
    -\frac{1}{(1+x_3)^2 (1+x_2)} 
    +\frac{1}{(1+x_3)^2 [x_3/(1+2x_3) + x_2]}
    \right]
\\ &= 
    \frac{G\left(-\frac{1}{2};x_3\right)-G(0;x_3)}{(1+x_3)^2}. 
\end{split}
\end{equation}
Finally, we must evaluate
\begin{equation}
    \int_0^\infty dx_3\, 
    \frac{G\left(-\frac{1}{2};x_3\right)-G(0;x_3)}{(1+x_3)^2}.
\end{equation}
Notice that the denominator involves $(1+x_3)$ to the second power, so the primitive can be computed using integration by parts and will not not involve mulitple polylogarithms of weight higher than one. Computing the limits then proceeds in the same way as above, using the functional identities for the logarithm and we find the the above integral is simply equal to $2 G(0;2)=2\ln 2$. Recalling the overall factor of $8\pi^2$ in eq.~(\ref{eq:I-ex-mb-int-3}), we have finally (see \wwwaux{symbolicint}
\begin{equation}
    I = 16 \pi^2 \ln 2 + {\mathrm O}(\epsilon).
\end{equation}
In a similar fashion it is possible to compute the result at higher orders in the $\epsilon$ expansion and the linear reducibility of the integrand guarantees that the expansion coefficients at higher orders are still expressible in terms of mulitple polylogarithms (at specific arguments). However, already at ${\mathrm O}(\epsilon)$, the full machinery of the Hopf-algebra of mulitple polylogarithms comes into play. The method discussed here has been used in the literature to compute real radiation integrals, e.g. in Higgs boson production at N${}^3$LO accuracy~\cite{Anastasiou:2013srw}.

\begin{tips}{Choosing Proper Methods of \mb{} Integration}

Sometimes it is beneficial to combine the methods presented above. 
For example, it may happen that after replacing only some of the 
gamma functions in the integrand of a Mellin-Barnes integral by 
beta functions, the Mellin-Barnes integrations simplify. One way
this may happen is that multi-dimensional Mellin-Barnes integrals 
decouple into products of one-dimensional integrals, but this is 
not the only option as we will see shortly. In such situations, it is 
not necessary to then continue to replace further gamma functions by 
beta functions until we arrive at a proper Euler integral representation, 
but one may go ahead and evaluate the simplified Mellin-Barnes 
integrals directly, e.g. using the methods of symbolic summation. 
This way of proceeding can sometimes be more efficient than if we 
had derived a full Euler integral representation. As an application 
of these ideas, here we derive the solution of the two-denominator 
angular integral with one mass, eq.~(\ref{eq:Om-2-1m-fin}). 

To start, recall from eq.~(\ref{eq:Om-2-1m-MB}) that the Mellin-Barnes representation of the angular integral $\Omega_{j,k}(v_{11},v_{12},\epsilon)$ involves the two-dimensional integral
\begin{equation}
    \begin{split}
    I &= 
        \int_{-i\infty}^{+i\infty} 
        \frac{dz_1}{2\pi i} \frac{dz_2}{2\pi i}
        \Gamma(-z_1)\Gamma(-z_2)
        \Gamma(a+2z_1+z_2)\Gamma(b+z_2)
\\& \times
        \Gamma(c-z_1-z_2) x^{z_1} y^{z_2}\,.
    \end{split}
\label{eq:Om-2-1m-MB-I-1}
\end{equation}
Notice that this integral involves a gamma function of the form $\Gamma(\dots + 2z_1)$ 
and hence it cannot easily be computed using the summation methods presented above.
Throughout we tacitly assume that all parameters lie in a strip 
of the convex plane such that each integral we write is convergent. 
Then let us begin by writing the product of the second and fourth 
gamma functions in the integrand as a beta function. Here we 
choose the representation of the beta function where the domain of 
integration is $[0,1]$ and write
\begin{equation}
    \Gamma(-z_2)\Gamma(b+z_2) = \Gamma(b) \int_0^1 dt\, 
    t^{-1-z_2}(1-t)^{b-1+z_2}\,.
\end{equation}
Then, eq.~(\ref{eq:Om-2-1m-MB-I-1}) becomes
\begin{equation}
    \begin{split}
    I &= \Gamma(b)
        \int_{-i\infty}^{+i\infty} 
        \frac{dz_1}{2\pi i} \frac{dz_2}{2\pi i}
        \int_0^1 dt\, t^{-1-z_2}(1-t)^{b-1+z_2}
\\& \times
        \Gamma(-z_1)\Gamma(a+2z_1+z_2)
        \Gamma(c-z_1-z_2) x^{z_1} y^{z_2}\,.
    \end{split}
\label{eq:Om-2-1m-MB-I-2}
\end{equation}
Next, we perform the change of variables $z_2 \to -a - 2 z_1 - z_2$, which makes the following manipulations more straightforward to follow. Then we find
\begin{equation}
    \begin{split}
    I &= \Gamma(b)
        \int_{-i\infty}^{+i\infty} 
        \frac{dz_1}{2\pi i} \frac{dz_2}{2\pi i}
        \int_0^1 dt\, t^{a-1+2z_1+z_2}(1-t)^{b-a-1-2z_1-z_2}
\\& \times
        \Gamma(-z_1)\Gamma(-z_2)
        \Gamma(a+c+z_1+z_2) x^{z_1} y^{-a-2z_1-z_2}\,.
    \end{split}
\label{eq:Om-2-1m-MB-I-3}
\end{equation}
After exchanging the order of integrations and rearranging some factors, $I$ can be written in the following form,
\begin{equation}
    \begin{split}
    I &= y^{-a} \Gamma(b)
        \int_0^1 dt\, t^{a-1}(1-t)^{a+b+2c-1}
        \int_{-i\infty}^{+i\infty} 
        \frac{dz_1}{2\pi i} \frac{dz_2}{2\pi i}
\\& \times
        \Gamma(-z_1)\Gamma(-z_2)
        \Gamma(a+c+z_1+z_2) \left(\frac{x s^2}{y^2}\right)^{z_1}
        \left[\frac{s(1-s)}{y}\right]^{z_2}
        [(1-s^2)]^{-a-c-z_1-z_2}\,.
    \end{split}
\label{eq:Om-2-1m-MB-I-4}
\end{equation}
At this point we could continue to express products of gamma functions as real Euler-type integrals. However, a more clever way to proceed presents itself. Indeed, the Mellin-Barnes integrations are now formally easy to perform, since the integral we are left with is just a special case of the general formula
\begin{equation}
    \frac{1}{(A+B+C)^\nu} = \frac{1}{\Gamma(\nu)}
    \int_{-i \infty}^{+i \infty} \frac{dz_1\,dz_2}{(2\pi i)^2}
    \Gamma(-z_1)\Gamma(-z_2)\Gamma(\nu+z_1+z_2) A^{z_1} B^{z_1} 
    C^{-\nu-z_1-z_2},
\end{equation}
where of course the appropriate choice of contour is understood. Thus, eq.~(\ref{eq:Om-2-1m-MB-I-4}) becomes
\begin{equation}
    I = y^{-a} \Gamma(b) \Gamma(a+c) \int_0^1 dt\,
    t^{a-1}(1-t)^{a+b+2c-1}
    \left[(1-s)^2+\frac{s(1-s)}{y} + \frac{x s^2}{y^2}\right]^{-a-c}.
\end{equation}
The last factor of the integrand can be factored in the integration variable $s$ as follows,
\begin{equation}
    (1-s)^2+\frac{s(1-s)}{y} + \frac{x s^2}{y^2} =
    \left(1-\frac{2y-1-\sqrt{1-4x}}{2y}s\right)
    \left(1-\frac{2y-1+\sqrt{1-4x}}{2y}s\right).
\end{equation}
Note that the appearance of square roots is not an issue, since only the parameter $x$ appears under the root, the variable of integration $s$ does not. (Recall also the discussion below eq.~(\ref{eq:S-sets-seq}) regarding the last, and in the present case only, integration step). Finally, we obtain the one-dimensional real integral representation
\begin{equation}
\begin{split}
    I &= y^{-a} \Gamma(b) \Gamma(a+c) \int_0^1 dt\,
    t^{a-1}(1-t)^{a+b+2c-1}
\\ &\times
     \left(1-\frac{2y-1-\sqrt{1-4x}}{2y}s\right)^{-a-c}
    \left(1-\frac{2y-1+\sqrt{1-4x}}{2y}s\right)^{-a-c}.
\end{split}
\end{equation}
The last integral can now be performed in terms of the Appell function of the first kind, see eq.~(\ref{eq:F1-1d-real-int}), 
and we find the final result
\begin{equation}
\begin{split}
    I &= y^{-a} \frac{\Gamma(b)\Gamma(b)\Gamma(a+c)\Gamma(a+b+2c)}
    {\Gamma(2a+b+2c}
\\ &\times
    F_1\left(a,a+c,a+c,2a+b+2c,\frac{2y-1-\sqrt{1-4x}}{2y},
    \frac{2y-1+\sqrt{1-4x}}{2y}\right).
\end{split}
\end{equation}
We note that if the parameters $a$, $b$ and $c$ are all of the form $n+m\epsilon$ with $n\in\mathbb{Z}$, the integral representation can serve as a starting point for the computation of the $\epsilon$ expansion of the result. Alternatively, the Appell function can be represented as nested sum of depth two and the techniques of symbolic summation discussed above may also be applied to evaluate the $\epsilon$ expansion.

\end{tips}

\subsection{More General Integrals}
\label{sec:moregenints}

We have seen in section~\ref{sec:moregensums} that for \MB{} integrals that contain the integration variable also in the form $\Gamma(\ldots \pm az)$ with $a\ne 1$ (as is typical for example for problems with massive propagators), obtaining analytic solutions by applying Cauchy's theorem and summing over residues becomes problematic even for one-dimensional integrals. Indeed, as demonstrated there, already in such simple cases we can encounter sums that we are no longer able to write as $S$- or $Z$-sums and thus the systematic way of obtaining the solution given in section~\ref{sect:1dmb-sol-with-sums} is no longer applicable.

Here we want to give a brief description of how these issues manifest themselves in the integration approach to obtaining analytic solutions. To do so, let us consider the vertex function of eq.~(\ref{parts}), which we have also seen in eq.~(\ref{eq:V3l2m-again}),
\begin{eqnarray}
V^{\epsilon^{-1}}_{V3l2m}(s) &=& 
- ~ \frac{1}{2s} 
\int\limits_{-\frac{1}{2}-i \infty}^{-\frac{1}{2}+i \infty} 
\frac{dz}{2\pi i}
{(-s)^{-z}}
{\frac{\Gamma^3(-z)\Gamma(1+z)} {\Gamma(-2z)}}. \label{eq:V3l2m-again2}
\end{eqnarray}
Noting that the contour runs parallel to the imaginary axis with real part $Re(z)=-\frac{1}{2}$, we see immediately that the real parts of the arguments of all gamma functions are positive. Thus, we can express the integral as a convergent Euler-type integral as described in section~\ref{subsec:EulerInt}. E.g. using
\begin{equation}
    \frac{\Gamma^2(-z)}{\Gamma(-2z)} = 
    \int_0^\infty dx_1\, x_2^{-1 - z} (1 + x_1)^{2z}
\end{equation}
and
\begin{equation}
    \Gamma(-z)\Gamma(1+z) = 
    \int_0^\infty dx_2\, x_2^{-1 - z} (1 + x_2)^{-1},
\end{equation}
we find the integral representation
\begin{equation}
V^{\epsilon^{-1}}_{V3l2m}(s) = -\frac{1}{2s} 
    \int_0^\infty dx_1\,\int_0^\infty dx_2\,
    \frac{1}{x_1 x_2 (1+x_2)}
    \delta\left(1 - \frac{(1 + x_1)^2}{(-s) x_1 x_2}\right).
\label{eq:V3l2m-DD}
\end{equation}
To proceed, we must solve the Dirac delta constraint for one of the variables. Solving for $x_2$, we find
\begin{equation}
    x_2 = \frac{(1 + x_1)^2}{(-s) x_1}
\end{equation}
and the following integral representation for the vertex function,
\begin{equation}
    V^{\epsilon^{-1}}_{V3l2m}(s) = \frac{1}{2} 
    \int_0^\infty dx_1\, \frac{1}{1 + (2 - s) x_1 + x_1^2}.
\label{eq:V3l2m-1}
\end{equation}
For the sake of simplicity, let us assume that we are working in the Euclidean region where $s<0$ and thus there are no singularities in the integration domain. Examining the above integral we see the following issue: the denominator is not a product of factors that are linear in the integration variable $x_1$, so we cannot perform the integration in a straightforward manner in terms of MPLs. However, in this particular case, it is easy to see how to proceed. Let us simply factor the denominator by solving for the roots of the quadratic,
\begin{equation}
    1 + (2 - s) x_1 + x_1^2 = 
    \left(x_1 - \frac{s - 2 + \sqrt{(s - 4) s}}{2}\right) 
    \left(x_2 - \frac{s - 2 - \sqrt{(s - 4) s}}{2}\right).
\end{equation}
Using this factorization, the integral in eq.~(\ref{eq:V3l2m-1}) can be computed easily after performing the partial fraction decomposition. {\it Evidently, logarithms of the roots $\frac{s - 2 \pm \sqrt{(s - 4) s}}{2}$
will appear in the result.} The final expression can be simplified by introducing once again the conformal variable $y$, given in eq.~(\ref{eq:conformal}) (with $m=1$). In fact, setting $s = -\frac{(1 - y)^2}{y}$, the quadratic expression in the denominator of eq.~(\ref{eq:V3l2m-1}) factors immediately
\begin{equation}
    1 + (2 - s) x_1 + x_1^2 = \frac{(x_1 + y) (1 + x_1 y)}{y},
\end{equation}
and the integration yields the same result that we found in eq.~(\ref{eq:V3l2m-res}),
\begin{equation}
V^{\epsilon^{-1}}_{V3l2m}(s) = -\frac{y \ln y}{1-y^2}.
\label{eq:V3l2m-res-again}
\end{equation}
Thus, once more we find that a suitable change of variable allows to compute the solution in terms of MPLs.

This is, however, not true generally. In order to get a feeling for why this is the case, let us return to eq.~(\ref{eq:V3l2m-DD}) and solve the Dirac delta constraint for $x_1$. In this case we find
\begin{equation}
    x_1 = \frac{-2 - s x_2 \pm \sqrt{s x_2 (4 + s x_2)}}{2}.
\label{eq:V3l2m-x1sol}
\end{equation}
Recalling that 
\begin{equation}
    \delta(f(x)) = \sum_{i} \frac{1}{|f'(x_{0,i})|} \delta(x-x_{0,i}),
\end{equation}
where $x_{0,i}$ are the zeros of $f(x)$, we find the following integral representation in this case
\begin{equation}
    V^{\epsilon^{-1}}_{V3l2m}(s) = \frac{1}{2} 
    \int_{-\frac{4}{s}}^\infty dx_2\, 
    \frac{1}{(1 + x_2) \sqrt{s x_2 (4 + s x_2)}}.
\label{eq:V3l2m-1-2}
\end{equation}
The non-trivial lower limit of integration appears because the solutions for $x_1$ in eq.~(\ref{eq:V3l2m-x1sol}) must be non-negative and hence real. But recall that we are working in the Euclidean region where a$s<0$. Thus, the expression under the square root in eq.~(\ref{eq:V3l2m-x1sol}) must be non-negative, which implies $-\frac{4}{s} \le x_2$ (recall that $x_2>0$). Now we are faced with the following problem: the integrand includes a polynomial of the integration variable under a square root. Such integrals are genuine generalizations of multiple polylogarithms and in general cannot be evaluated in terms MPLs. Notice too, that introducing the conformal variable $y$ does not solve this problem. Setting $x_2 \to x_2 - \frac{4}{s}$ to transform the limits of integration to zero and infinity and using $s= -\frac{(1-y)^2}{y}$ together with the fact that in the Euclidean region $0<y<1$, we find
\begin{equation}
    V^{\epsilon^{-1}}_{V3l2m}(s) = \frac{1}{2} 
    \int_{0}^\infty dx_2\, 
    \frac{y(1 - y)}
    {[(1 + y)^2 + (1 - y)^2 x_2]\sqrt{x_2(4 y + (1 - y)^2 x_2)}}.
\label{eq:V3l2m-2}
\end{equation}
Evidently, the integration variable $x_2$ still appears under the square root. In order to proceed, we must find some transformation of the integration variable which \emph{rationalizes the root}, i.e., a transformation which converts the integrand into a rational function of $x_2$. There are some systematic procedures known which can be used to search for such transformations, see e.g. the algorithm of~\cite{Besier:2019kco}, which has also been implemented in the \math{} package {\tt RationalizeRoots}. However, in general such transformations are not guaranteed to exist. In our particular case, the transformation
\begin{equation}
    x_2 \to -\frac{4 y}{(1 - y)^2 - 16 y^2 t_2^2}
\end{equation}
does the trick. Note that the lower limit of integration after this change of variables becomes $\frac{1-y}{4y}$, while the upper limit is still infinity. Finally, we shift the integration variable, $t_2 \to t_2 + \frac{1-y}{4y}$, in order set the lower limit of integration to zero and obtain
\begin{equation}
    V^{\epsilon^{-1}}_{V3l2m}(s) = \frac{1}{2} 
    \int_{0}^\infty dt_2\, 
    \frac{2 y (1 - y)}{[1 - y + 2 (1 + y) t_2][1 - y + 2 y (1 + y) t_2]}.
\label{eq:V3l2m-3}
\end{equation}
This final integral is easy to perform and once more gives the result of eq.~(\ref{eq:V3l2m-res-again}).

However, the fact that we were able to find a transformation which rationalized the integrand is not generic. In fact, we emphasize that in the general case, one can encounter square roots that cannot be rationalized (with a rational transformation). Such integrands lead to genuinely new functions, such as the complete elliptic integrals
\begin{equation}
    K(z) = \int_0^1 dt\, \frac{1}{\sqrt{(1-t^2)(1-z t^2)}}
    \qquad\mbox{and}\qquad
    E(z) = \int_0^1 dt\, \frac{\sqrt{1-z t^2}}{\sqrt{(1-t^2)}}.
\end{equation}
Moreover, in complicated examples, one typically encounters several different square roots in the integrand. In such cases, the variable transformations must rationalize all square roots that appear together simultaneously. The complete elliptic integrals already showcase this scenario. In fact, we can easily find transformations that rationalize either $\sqrt{1-t^2}$ or $\sqrt{1-z t^2}$, but not both roots together. Hence, in the general case we are left with some square roots and cannot proceed to compute the integral in terms of MPLs. Instead, the set of functions must be enlarged to include these more general integrals. The study of the appropriate generalizations, such as elliptic polylogarithms, is an active area of research and their discussion is beyond the scope of this book. For a summary of the current state-of-the-art, we refer the reader to~\cite{Bourjaily:2022bwx}.

\section{Approximations \label{section:approx}}
 
\subsection{Expansions in the Ratios of Kinematic Parameters} 
\label{sec:1exp}  
 
Expansions of the \mb{} integrals depends on the kinematic variables and physical regimes. As an example we start from the integrals where two variables, say $m$ and $s$ are present, as in eq.~(\ref{MB-SE1l2m-lemmas}), which can be cast in the form $ms=-m^2/s$
\bq
I_1 =  \int_{-1/4-i \infty}^{-1/4+i \infty} \frac{d z_1}{2 \pi i}\; (ms)^{z_1}  \frac{\Gamma(1 - z_1)^2 \Gamma(-z_1) \Gamma(z_1)}{ \Gamma(2 - 2 z_1)}. \label{eq:se1l2mchap6}
\eq  
We can approximate the integral and expand it in the $ms$ parameter by closing the contour on right-half in a complex plane, means $s \gg m^2$. This has been done in section~\ref{ssec:MBtoSums} and the \math{} file with derivations can be found in \wwwaux{SE2l2m}. There you can also find how the result can be obtained directly with the package \texttt{MBasymptotics.m}, see appendix~\ref{app:analsoftware}. 
If another kinematical variables exist, we can work our approximations systematically in regions where the ratios of them are small. For instance, in case of the scattering processes like Bhabha, we have $s$ and $t$ Mandelstam variables, we can expand in the limit $s,t \gg m^2$ by considering ratios 
\begin{equation}
\ln\left(-\frac{m^2}{s}\right) \longrightarrow
\ln\left(-\frac{m^2}{t}\right) = \ln\left(-\frac{m^2}{s}\right)
-\ln\left(\frac{t}{s}\right),
\end{equation}
and
\begin{equation}
\frac{t}{s} \longrightarrow \frac{s}{t} = 1/\left(\frac{t}{s}\right).
\end{equation}
The above transformations may imply some algebra in transforming the arguments of the polylogarithms to a unique form or basis, see \texttt{HPLs} section~\ref{sec:hplsmpls}.

In general, integrals of the form  

\begin{equation} I = (m^2)^{-2\epsilon}\int_{-i \infty}^{i \infty} \frac{d z}{2 \pi i} \;
\left(-\frac{m^2}{s}\right)^z f\left(\frac{t}{s},z\right), 
\label{eq:mbnpbbhabha1}
\end{equation} 
have been considered in kinematic expansions in~\cite{Czakon:2006pa}.
Here the $f$ function contains, amongst others, a product of gamma, or
possibly polygamma functions, which have poles in $z$.  The $f$ function is given in general by a
multidimensionl \mb{} integral. To expand such integrals we take a trick since it is difficult to directly take
residues in this form, and we change the order of integration and close
the $z$ contour to the right. This procedure is subsequently applied
recursively, until no further poles at the required order of expansion
occur. We can check resulting approximations numerically with existing \mb{} software (appendix~\ref{app:mbnum}) and other methods and packages (appendix~\ref{app:othernum}).  

For more complicated cases, we consider \texttt{4PF} integrals with three kinematic Mandelstam variables $s,t,u$, we follow~\cite{Czakon:2006pa}.
 
To summarize, we would like to evaluate \mb{} integrals in eq.~(\ref{eq:mbnpbbhabha1}) using the following series variables
\bq
x=\frac{t}{s},\;\;\;\; L = \ln\left(-\frac{m^2}{s}\right).
\label{eq:xvar}
\eq
The advantage of such choice is that the  results are explicitly real in the
Euclidean domain, $s,t < 0$. 
We should note that in the limit $s,t \gg m^2$ we have $s+t+u=u$, which simplifies possible analytic solutions of \mb{} integrals involved. In fact, we know already that without this approximation the analytic solutions goes beyond \texttt{HPLs}~\cite{Henn:2013woa}. 

\begin{svgraybox}
That approximations of \mb{} integrals leads to analytic solutions with simpler classes of functions is a typical observation. This fact can be certainly explored systematically in frontier studies of evaluation of multiloop and multileg \texttt{FI} integrals.  
\end{svgraybox}

\begin{figure}
  \begin{center}
      \includegraphics[scale=0.5]{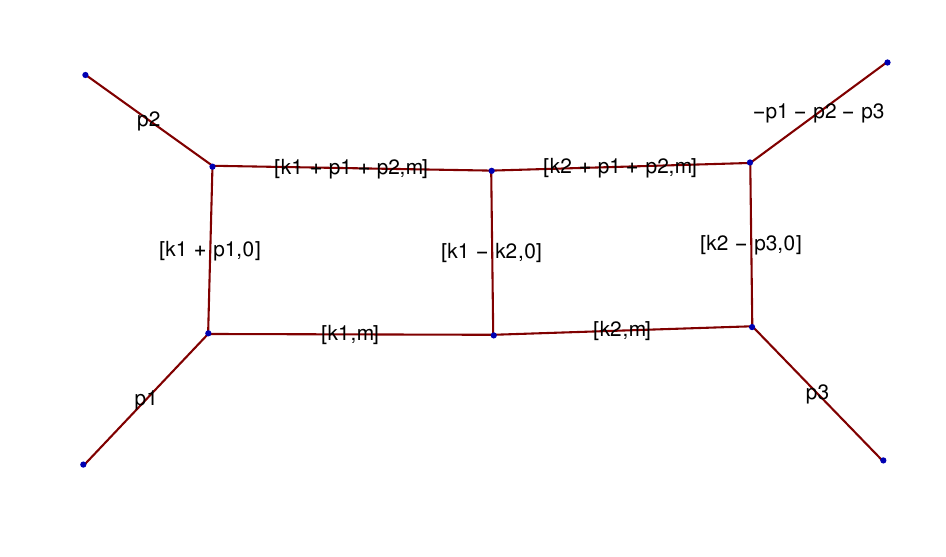}
  \end{center}
  \caption{\label{fig:7lin}
The 7-propagators planar two-loop topology \texttt{B7l4m1} with the momentum distribution defined in the file \texttt{MB\_B7l4m1\_Springer.nb}, see the text for details.
}
\end{figure}

Saying that, in \wwwaux{B74m1} we derive six dimensional \mb{} representation for the scalar integral given in Fig.~\ref{fig:7lin}, the integral considered first time in~\cite{Smirnov:2001cm}. Also, \texttt{MBasymptotic.m} package is used there to expand the integral in the $ms$ variable. 
After expansion and additional manipulations  we are left with a bunch of one dimensional \mb{} integrals which can be sum up by one of summation packages described in appendix~\ref{app:mplsums}.
We leave it as the Problem~\ref{prob:expandxsummer}. 

After summation over the $x$ variable in eq.~(\ref{eq:xvar}), we get~\cite{Czakon:2006pa}
\begin{eqnarray}
%checked
{\tt B7l4m1} =
       & + & \frac{1}{\epsilon^2}\; \frac{2 {\rm L}^2}{s^2t}
\nonumber \\
%tord:  &+& \frac{1}{\epsilon}\; \frac{m_s^3}{3 x}
        &-& \frac{1}{\epsilon}\; \frac{1}{3 s^2t}
 \left[-10\;{\rm L}^3 + 6\;{\rm L}\;\zeta_2 + 6\;\zeta_3 + 12\;{\rm L}^2\;\ln(x) \right] \nonumber \\
%tord:  &+& \frac{m_s^3}{6 x} \biggl\{
        &-& \frac{1}{6s^2t} \biggl\{
-12\;{\rm L}^4 + 28\;{\rm L}^3\;\ln(x) - 4\;{\rm L}^2\;(-30\;\zeta_2 +
3\;\ln^2(x)) \nonumber \\
&-& 4\;{\rm L}\;(9\;\zeta_3 + 30\;\zeta_2\;\ln(x) + \ln^3(x) -
18\;\zeta_2\;\ln(1 + x) \nonumber \\
&-&
3\;\ln^2(x)\;\ln(1 + x) - 6\;\ln(x)\;{\litwo}( -x) + 6\;{\litr}( -x))
\nonumber \\
&-&15 \zeta_4-24 \zeta_3 \ln(x) \biggr\}
\label{B7l4m1ms}
\end{eqnarray}

\subsection{Taylor and Region Expansions with \mb{} \label{sec:taylorregionMB}}

{Another interesting idea to evaluate difficult multi-dimensional \mb{} integrals, which is different from direct kinematic expansions just discussed in the previous section is to decrease the number of \mb{} variables by direct propagators expansion in the first step in some ratios of the propagator masses. 
A simple way to do that is to Taylor-expand, as discussed and applied in more general contexts in~\cite{Smirnov:1994tg, Misiak:2004ew}.
In~\cite{Dubovyk:2021lqe} an example has been shown.
The idea is simple. Imagine we have two propagators with different masses, $M_W$ and $MZ$, say. Then integrand can be expanded at e.g. $M_W = M_Z$ and the expanded propagator takes the following form
 
\begin{equation}    
    \frac{1}{k^{2} - M^{2}_{W}} = \frac{1}{k^{2} - M^{2}_{Z}} + 
     \frac{(M^{2}_{W} - M^{2}_{Z})}{(k^{2} - M^{2}_{Z})^{2}} +
     \frac{(M^{2}_{W} - M^{2}_{Z})^{2}}{(k^{2} - M^{2}_{Z})^{3}} +
     \frac{(M^{2}_{W} - M^{2}_{Z})^{3}}{(k^{2} - M^{2}_{Z})^{4}} + \dots ~.
     \label{eq:Taylor_exp}
\end{equation}

As the result, one gets to calculate more terms (depending on aiming accuracy) but simpler integrals (propagators with $M_W$ mass are eliminated in eq.~(\ref{eq:Taylor_exp})). For more masses involved, it can be applied recursively in order to reduce the number of scales in the integral.  

In~\cite{Mishima:2018olh} the method of expansions by regions~\cite{Beneke:1997zp} has been merged with the \mb{} method.  
The contribution from each region
is expressed in terms of \mb{} integrals, and the integrals have been solved systematically for the Higgs pair production cross-section at two-loop order,
in the high energy limit. However, the method can have wider applications. In work, the generalized Barnes lemmas
have been considered. In general, using \texttt{1BL}, \texttt{2BL} and generalized Barnes lemmas four-dimensional \mb{} integrals were reduced up to two-dimensional cases which were solved analytically in terms of generalized hypergeometric functions discussed in section~\ref{sec:gamma_hyperg} or summed (in the case of one dimensional \mb{} integrals). For essential points of the \mb{} techniques used, see~\cite{Davies:2018ood}.

 }
  
\subsection{\mb{}: Other Directions \label{sec:otherdir}}

Another strategy useful for getting analytical results for \texttt{FI} expanded in large or small limits has been developed in~\cite{Friot:2005cu}.
The starting point is the Feynman parametrization of eq.~(\ref{FeynSgen}) rewritten to the form
\begin{equation}
        \cF (\rho_j)=\int_0^1 d\alpha_1 \int_0^1 d\alpha_2 \cdots\int_0^1 d\alpha_F \frac{N(\alpha_i)}{\left[\sum_j D_{j}(\alpha_i)\rho_j\right]^{n+ (p/q)
\epsilon}}\,,
\end{equation}
where the $\rho_j$ denote scalar products of the external momenta and squared masses normalized to a fixed mass scale in the given Feynman diagram, so that $\rho_0 =1$, and $n+(p/q)\epsilon >0$ with $n$, $p$ and $q$ some positive integers, in general. For instance, in the simplest case there can be only one $\rho$--parameter
$
        \sum_j D_{j}(\alpha_i)\rho_j=D_{0}(\alpha_i)+D_{1}(\alpha_i)\rho\,,
$
and the behaviour of the integral $\cF (\rho)$ for $\rho\ll 1$ can be examined, for details see~\cite{Friot:2005cu} and related works by the authors.
Using the basic \mb{} relation, we get the relation between the $\rho$ real space and the complex Mellin s-plane
\begin{eqnarray}\label{feyparMB}
          \cF (\rho)&=&\frac{1}{2\pi i}\int\limits_{c-i\infty}^{c+i\infty} ds
          \  \rho^{-s}\  \cM[\cF](s),  \\
    \cM[\cF](s) &=& \int_0^1 d\alpha_1 \int_0^1 d\alpha_2 \cdots\int_0^1 d\alpha_F \ \frac{N(\alpha_i)}{\left[D_{0}(\alpha_i)\right]^{\nu}}\left(\frac{D_0(\alpha_i)}{D_1(\alpha_i)}\right)^{s}\frac{\Gamma(s)\Gamma(\nu-s)}{\Gamma(\nu).} \nn
\end{eqnarray}
Relation in eq.~(\ref{feyparMB}) constitutes the so-called Mellin transform.
{\it{the asymptotic behaviours of $\cF (\rho)$, both, for $\rho\ll 1$ and $\rho\gg 1$, are  encoded in the so-called {\it converse mapping theorem}~\cite{flajolet1995mellin}  which establishes a relation between the singularities in the Mellin $s$--plane and the asymptotic behaviour(s) one is looking for.} }
The method has been applied to the calculation of the vacuum polarization contributions to the muon anomaly in~\cite{Friot:2005cu}.

The problem of the convergence of \mb{} sums is difficult, either in asymptotic or full analytic form~\cite{Friot:2011ic,Blumlein:2014maa}. The public programs are gathered in appendix~\ref{app:mplsums}.  

In this chapter we discussed possible ways in which \mb{} representations are used in analytic studies of \texttt{FI}.
%\label{sec:3poles}  
%from Gudrun
Certainly the subject is not `out-of-print'.
For instance, as shown in~\cite{Kalmykov:2012rr}, Mellin-Barnes
representations can be used to derive linear systems of homogeneous differential equations for the original
Feynman integrals with arbitrary powers of propagators,
without the need for IBP relations.
This method in addition can be used to deduce extra relations between master integrals, beyond the
IBP reduction~\cite{Kalmykov:2011yy,Kalmykov:2016lxx,Bitoun:2017nre}. 

\section*{Problems}
\addcontentsline{toc}{section}{Problems}
\begin{problem} 
 \label{prob:inf1} Discuss convergence of the term $\Gamma(a+z)\Gamma(b-z) x^z$ in eq.~(\ref{eq:exinf1}) for  ${z\rightarrow \infty}$  
\\ \noindent \hint{} Notice that the large-$z$ behaviour is controlled by the exponential function $e^{z \ln x}$, since this grows or vanishes faster than any power-law function of $z$. Thus, the limit at ${z\rightarrow \infty}$ is zero or infinity, depending on the sign of the real part of the product $z \ln x$.
\end{problem}

\begin{problem} \label{prob:1BLalternat}
Derive Barnes' 1st lemma in eq.~(\ref{eq:1stBL}).  \\ \noindent
\hint{} Consider
    \begin{equation}
        I = \int_{-i\infty}^{+i\infty} \frac{dz}{2\pi i}
        \Gamma(a+z)\Gamma(b+z)\Gamma(c-z)\Gamma(d-z)
    \end{equation}
    and use the methods discussed above to derive a real integral 
    representation for this one-fold Mellin-Barnes integral. You 
    should end up with a one-dimensional real integral that can be 
    immediately performed in terms of a beta function. 
%    \\ \noindent \hint{} 
    See also discussion in~\cite{Jantzen:2012cb}.
\end{problem}
    
\begin{problem}
Show that
    \begin{equation}
        I = \int_{-i\infty}^{+i\infty} \frac{dz}{2\pi i} 
        \frac{\Gamma(a+z)\Gamma(b+z)\Gamma(-z) }{\Gamma(c+z)}
        x^z
    \end{equation}
    can be evaluated in terms of Gauss' hypergeometric function 
    ${}_2F_1$. Note that ${}_2F_1(a,b;c;x)$ has the following 
    standard integral representation
    \begin{equation}
        {}_2F_1(a,b;c;x) = \frac{\Gamma(c)}{\Gamma(b)\Gamma(b-c)}
        \int_0^1 dt\, t^{b-1} (1-t)^{c-b-1} (1-t x)^{-a}\,.
    \end{equation}
    \\ \noindent \hint{} See discussion in section~\ref{sec:gamma_hyperg}, eq.~(\ref{eq:ressum2F1}) and Problem \ref{prob:mb2F1}.
\end{problem}
\begin{problem}
\label{prob:expandms} Analyze construction and expansion of the \mb{} integral given in Fig.~\ref{fig:7lin}, which is given in \wwwaux{B7l4m1}. Note that the original \mb{} integral is six dimensional, after analytic expansion in $\eps$ up to the $\eps^0$ term it is four dimensional, and after $ms$ expansion (\texttt{2BL} must be applied afterwards) the final result is a set of 0- and 1-dimensional \mb{} integrals.
\end{problem}
\begin{problem}
\label{prob:expandxsummer} Taking the set of expanded in the $ms$ paprameter \mb{} integrals obtained in Problem~\ref{prob:expandms}, derive analytic result and check it against eq.~(\ref{B7l4m1ms}).
\\ \noindent
\hint{} Use e.g. \texttt{MBsums.m} and one of packages for summing series, e.g. \texttt{Xsummer}, see appendix~\ref{app:mplsums}.
\end{problem}

\putbib[%
bibs/refs,%
bibs/2loops_LL16,%
bibs/Phd_Dubovyk,%
bibs/LRrefa,%
bibs/2loopsreport]
\end{bibunit}
%\input{references5}

%% file: chapter6.tex
%%%%%%%%%%%%%%%%%%%%% chapter.tex %%%%%%%%%%%%%%%%%%%%%%%%%%%%%%%%%
%
% sample chapter
%
% Use this file as a template for your own input.
%
\begin{bibunit}[elsarticle-num-ID] % define the bib-style for the unit: elsarticle-num.bst
%  text-1; this is the corresponding section
%\putbib[2loops] % the *.bib
%\end{bibunit}
% go-on
%--- from: bibunits.sty, adapts the font size of ``References'' to section
\let\stdthebibliography\thebibliography
\renewcommand{\thebibliography}{%
\let\section\subsection
\stdthebibliography} 
 
\chapter{\mb{} Numerical Methods}
\label{chapter-MBnum}  
 
 \abstract{We discuss the main issues which lead to numerical instabilities of multiloop \MB{} integrals in the Minkowskian region. We then present practical procedures for overcoming the obstacles such as the transformation of variables to finite intervals, the shifting and deformation of integration contours, and the construction of MB representations that take kinematic thresholds into account. Steepest descent and Lefschetz thimbles are applied to find optimal integration contours. The numerical evaluation of phase space \MB{} integrals is also briefly discussed.}

\section{Introduction}

Since the very beginning, the computer language {\tt FORTRAN} (1957, John Backus, IBM) has been used for numeric scientific computing. For purpose of algebraic evaluations a special software  {\tt Reduce} (started in 1963 by Anthony Hearn),  Schoonschip (started in 1967 by Martinus J.~G.~Veltman), and its descendant, the {\tt Form} package (initially released in 1989 by Jos Vermaseren) has been also developed. Nowadays {\tt Form}, {\tt C++}, \math{} and python are the main environments used in particle physics. More on the history of numerical and algebra systems and software development can be read in~\cite{Weinzierl:2002cg}. 

In the 1980s, many multiloop methods to calculate higher order corrections in particle physics have been developed, based on generalized unitarity, tree-duality, simultaneous numerical integration of amplitudes over the phase space and the loop momentum, reductions of the integrals at the integrand level, improved  diagrammatic approach and recursion relations applied to higher-rank tensor integrals, contour deformations, expansions by regions,  sector-decomposition, dispersion relations, differential equations, summation of series, Mellin-Barnes representations. Many methods aim at
direct Feynman integral calculations. For 
general reviews see~\cite{Smirnov:2004ym,Smirnov:2006rya,Anastasiou:2005cb,Gluza:2014jxa,Freitas:2016sty}. Due to experimental requirements defined by precision physics at present and future colliders, notably HL-LHC and Tera-Z physics at FCC-ee, there is a large activity in the field and the methods are in permanent development~\cite{Blondel:2018mad,Heinrich:2020ybq}. See also a recent workshop at CERN on the subject (June 2022)~\cite{cern2022}.

To solve Feynman integrals,  analytical methods can be used, though they exhibit natural limitations when sophisticated integrals with many parameters appear. Some analytical methods connected with \mb{} integrals have been discussed in previous sections, and some useful software is summarized in the Appendix. Thus, with going to higher and higher loop levels and growing complexity (number of parameters, dimensionality of integrals), it is natural that numerical methods become more and more relevant.  
The trade-off between analytical and numerical or semi-numerical methods has been nicely summarized in~\cite{Heinrich:2020ybq}, see Tab.~\ref{tab:anaversusnum}. 

\begin{table}
  \begin{center}
%  \hspace*{-1cm}
  \begin{tabular}{|l|c|c|}
    \hline
  &  analytic & numerical\\
    \hline
    pole cancellation& exact & with numerical uncertainty\\
    control of integrable singularities & analytic continuation
              & less straightforward\\
    fast and stable evaluation & yes (mostly)& depends\\
    extension to more scales/loops & difficult & promising\\
    automation & difficult & less difficult\\
    \hline    
  \end{tabular} 
  \caption{Strong and weak points of analytic versus numerical evaluation
    of loop integrals as discussed in~\cite{Heinrich:2020ybq}.}
  \label{tab:anaversusnum}
  \end{center}
\end{table}

So far a lot has been done for precise analytical and numerical calculations at the one-loop order (which is also called `next-to-leading', \texttt{NLO}). Many \texttt{NLO} programs allow to consider processes automatically and in a numerical way, see for instance~\cite{Dubovyk:2017cqw,Heinrich:2020ybq}.
Going beyond the one-loop calculations in Minkowskian regions, the situation is quite different.   
There are technical obstructions in the calculation of \texttt{FI} in this setting due to threshold effects, singularities, on-shellness or several  mass parameters involved.  
 With increasing number of loops and increasing number of mass and momentum scales, it becomes more difficult to compute these corrections analytically, or even semi-analytically.  
In 2014  the only  advanced  automatic numerical  two-loop method in Minkowskian kinematics was \texttt{SD}~\cite{Hepp:1966eg,Binoth:2000ps}, where  a complex contour deformation of the Feynman parameter integrals is implemented in the publicly available numerical  packages {\tt FIESTA 3}~\cite{Smirnov:2013eza} (since 2013) and {\tt SecDec 2}~\cite{Borowka:2012yc}  (since 2012), followed by pySecDec~\cite{Borowka:2017idc}. 
 
 Evaluation of the multiloop, multiscale integrals is in general very challenging. For instance, for the 2-loop Z-decay (representative Feynman diagram is shown in Fig.~\ref{fig:npSDMB}), the typical evaluation problems which we meet~\cite{Dubovyk:2016aqv} are connected with (i) up to four dimensionless scales at $s=M_Z^2$  with a variance of masses
   $M_Z, M_W, m_t, M_H$ involved and  (ii) intricate threshold and on-shell effects.
 To tackle these problems, the semi-numerical approach to the calculation of Feynman integrals based on  Mellin-Barnes representations has been developed. 
In calculations the \MB{} method is used with many suitable packages gathered in what we call the \MB{}-suite. 
An important part of the \MB{}-suite is the {\tt AMBRE} project, see chapter~\ref{chapter-MBrepr}.  It includes among others the construction of \MB{} representations~\cite{Gluza:2007rt,Gluza:2010rn,Dubovyk:2016ocz,ambrewww}) and the recognition of planarity of Feynman diagrams~\cite{Bielas:2013v11,Bielas:2013rja} with the corresponding packages  \arm{} and \pltestm{}.  
For the extraction of $\epsilon$-singularities in dimensional regularization of multiloop integrals, the \mbm{}~\cite{Czakon:2005rk} and \mbresolve{}~\cite{Smirnov:2009up} packages are used, see chapter~\ref{chapter-singul}. As they offer the possibility of numerical integrations in Euclidean kinematics, they are used also as a numerical cross check of analytical results for multiloop integrals. Working on the completion of two-loop corrections to the $Z$-boson decay, it has been observed that serious convergence problems for some classes of integrals both with {\tt SecDec} and {\tt FIESTA} appear and the new numerical Mellin-Barnes (MB) approach to the calculation of multiloop and massive Feynman integrals is used, with the \MB{} package \mbn{}~\cite{Usovitsch:201606,Usovitsch:2018shx}, see section E6 in~\cite{Blondel:2018mad} and~\cite{Usovitsch:2018pea}. This package is able to evaluate \mb{} integrals directly in Minkowskian kinematics, solving so far untouchable regions of precision physics. Most of the mentioned  packages can be found at web pages~\cite{mbtools,ambrewww}, see also appendix~\ref{introA}.

 %warsaw 2021 talk
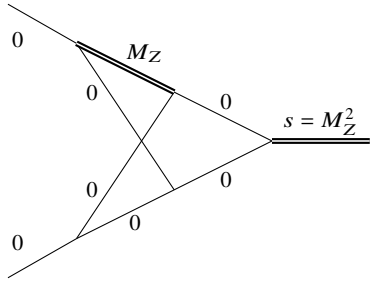
\begin{figure}[h!]
\centering
\begin{tikzpicture}[scale=1.3]
\begin{feynman}
 
\vertex at (0,-2) (i1);
\vertex at (0,2,1) (i2);
\vertex at (1,1) (i3);
\vertex at (1,0) (i4);
\vertex at (1,-1) (i5);
\vertex at (2,.5) (i6);
\vertex at (3,0) (f1);
\vertex at (4,0) (f2);
 
\draw    (0.3,1.4) -- (1,1) ; 
\draw    (1,1) -- (2,-.5) ;
%\draw    (1,0) -- (1.8,-.5) ;
\draw    (0.3,-1.4) -- (1,-1) ;
\draw    (1,-1) -- (2,.5) ;
\draw  [thick, double]  (2,.5) -- (1,1) ;
\draw   (2,.5) -- (3,0) ;
\draw    (1,-1) -- (3,0) ;
\draw [thick, double]  (3,0) -- (4,0) ;
 
\node[below] at (0.4,1.2) {{{\bf $0$}}};
\node[above] at (0.4,-1.2) {{{\bf $0$}}};
\node at (3.5,0.2) {{{\bf $s=M_Z^2$}}};
\node[left] at (1.3,.5) {{$0$}};
\node[left] at (1.3,-.5) {{$0$}};
%\node at (1.5,0.1) {{$6$}};  
\node at (1.7,0.9) {{$M_Z$}};
\node at (2.53,.4) {{$0$}};
\node at (2.53,-.4) {{$0$}};
\node at (1.6,-0.84) {{$0$}};
\end{feynman}
\end{tikzpicture}
\caption{The non-planar vertex diagram \texttt{V6l1m} which is presented for instance in a calculation of the Z-decay process to two massless fermions. It is a difficult case for the \texttt{SD} method due to a single massive propagator, which mass $M_Z$ coincides with external invariant energy. On the other hand, it can be nicely integrated with high accuracy with the \MB{} method, see, Tab.~\ref{tab:SDmink}. This \mb{} integral has been discussed initially in~\cite{Dubovyk:2016zok}.}
\label{fig:npSDMB}
\end{figure}
 
In Tab.~\ref{tab:SDmink} we show the problem that meets the \texttt{SD} method and the laborious progress made in recent years to overcome it. It appears that for such cases where most of the propagators are massless, the \mb{} method may serve as a complementary tool. In this way, using \texttt{SD} and \mb{} approaches the calculation of the two-loop \texttt{SM} \texttt{FI} has been completed~\cite{Dubovyk:2016aqv,Dubovyk:2018rlg,Dubovyk:2019szj}.

 \begin{table}[]
     \centering
     $
        \begin{array}{|l|l|}
        \hline  \hline
{\rm Analytical}:  & -{\color{black}{\bf 0.77859960897968}} - {\color{black}{\bf 4.12351259339631}} \cdot i \\
 \hline
{\rm MBnumerics}:  & -{\color{black}{\bf 0.778599608}}32476 - {\color{black}{\bf 4.123512}}60051601 \cdot i\\
{\rm MB+thresholds}:  & -{\color{black}{\bf 0.7785}}4282426410 - {\color{black}{\bf 4.123}}49826423109 \cdot i\\
 \hline 
{\rm SecDec}:      &      {\rm big\; error\; {\color{gray}{[2016]}}}\\ &{\color{black}{\bf -0.77}} - i \cdot {\color{black}{\bf 4.1}}\;
 {\color{gray}{[2017]}}\\& -{\color{black}{\bf 0.778}} - i \cdot {\color{black}{\bf 4.123}} \;{\color{gray}{[2019]}} \\%using splitting 
 {\rm pySecDec+rescaling  }:      & -{\color{black}{\bf 0.77859}}8 - i \cdot {\color{black}{\bf 4.123512}} \;{\color{gray}{[2020]}}\\
  \hline  \hline 
\end{array}
$
\caption{Minkowskian results for the diagram given in Fig.~\ref{fig:npSDMB}, for the constant part in $\epsilon$. As discussed in~\cite{Dubovyk:2016aqv}, see also section~\ref{sec:3complex}, the integral \texttt{V6l1m} which corresponds to the Feynman diagram in Fig.~\ref{fig:npSDMB} is evaluated at $s=M_Z^2\equiv 1+i \delta'$ where $\delta'$ is a small parameter. In numerical evaluations $\delta' = i \cdot 10^{-7}$.  To get the correct sign of the imaginary part of \texttt{FI} in calculations made in this chapter, the same $i \delta'$ prescription is assumed everywhere.
}
 \label{tab:SDmink}
 \end{table}

The analytical result can be found in~\cite{Fleischer:1998nb}.
The \mb{} running file \texttt{MB\_V6l1m\_Sprin\-ger.sh} for the scalar integral which corresponds to Fig.~\ref{fig:npSDMB}, evaluated with \mbn{} is available at~\cite{www_aux_springer}.  Similarly configuration files \texttt{SD\_V6l1m\_Springer\_generate.py} and  \texttt{SD\_V6l1m\_Springer\_integrate.py} for the \texttt{pySecDec} evaluation are given in~\cite{www_aux_springer}.  
The  discussion of the integral by the SecDec team (2017) can be found
in~\cite{Borowka:2017idc}.  In Tab.~\ref{tab:SDmink} the {\rm pySecDec} new result obtained with the so-called rescaling is also given.  

\begin{tips}{The Idea of Rescaling}

Let us briefly discuss the rescaling of Feynman parameters in the context of numerical integration with the \texttt{SD} approach, though the method can be used elsewhere.
The basic idea is that we can rescale each Feynman parameter by an arbitrary factor $x_i \rightarrow \kappa_i x_i, \kappa > 0$
resulting in a new integral $I = \prod \kappa_i^{n_i} I'$ where each term in $U$ and $F$ polynomials will get new coefficients.
The transformation is possible due to the Cheng-Wu theorem discussed in chapter~\ref{chapter-MBrepr}. One can sequentially extract Feynman parameters from the common delta function
and integrate them from $0$ to $\infty$. Rescaling, in this case, doesn't change integration boundaries and doesn't affect
the delta function. After this manipulation, one can reverse the theorem and put the corresponding parameter back into the delta function, making changes only
in coefficients of graph polynomials. Starting from Eq.~(\ref{FeynSgen}) and omitting the coefficient in front of the integral, it works as follows
\begin{multline}
 \int \limits_0^1 \prod \limits_{j=1}^N dx_j ~ x_j^{n_j-1}
 \delta(1-\sum \limits_{i=1}^N x_i)
 \frac{U(\vec x)^{N_{\nu}-d(L+1)/2}}{F(\vec x)^{N_{\nu}-dL/2}} \\
 = 
\kappa_{i_1}^{n_{i_1}} \int\limits_0^{\infty} dx_{i_1} \int \limits_0^1 
 \prod \limits_{j \neq i_1}^N d x_j
 \prod \limits_{j=1}^N ~ x_j^{n_j-1}
 \delta(1-\sum \limits_{i \neq i_1}^N x_i)
 \frac{U(\dots, \kappa_{i_1} x_{i_1}, \dots)^{N_{\nu}-d(L+1)/2}}{F(\dots, \kappa_{i_1} x_{i_1}, \dots)^{N_{\nu}-dL/2}} \\
 =
\prod \limits_{j=1}^N \kappa_i^{n_i} \int \limits_0^1 \prod \limits_{j=1}^N dx_j ~ x_j^{n_j-1}
 \delta(1-\sum \limits_{i=1}^N x_i)
 \frac{U(\vec x, \vec \kappa)^{N_{\nu}-d(L+1)/2}}{F(\vec x, \vec \kappa)^{N_{\nu}-dL/2}}. 
\end{multline}
After rescaling, one can proceed further in the usual way with the \texttt{SD} algorithm.

In the SD method, a deformation of integration contours is performed to avoid a problem with the so-called Landau singularity for physical kinematic points. Integrations start at 0 and end at one but go in a complex plane. In the case shown in Tab.~\ref{tab:SDmink} we do the calculation at a specific point for which coefficients in the $F$ polynomial
are equal to $1$ or $-1$. In such a situation, the Landau singularity is located at some corner of the integration domain and cannot be fixed by the contour's transformation. The main purpose of rescaling, in this case, is to move the singularity to the interior of the integration domain. In principle, rescaling may also improve the behavior of the integrand and increase the accuracy of results.
\end{tips}
 
At the end of this introductory section, we should acknowledge that very recently, numerical calculations have been pushed forward independently by \texttt{DEs}~\cite{Dubovyk:2022frj,Liu:2022chg,Cordero:2022gsh,Armadillo:2022ugh}. For other exploratory methods in \texttt{FI} computation, see a review~\cite{Heinrich:2020ybq}.  Now we will consider in more detail the main problems and features of the \mb{} numerical approach to evaluating the \texttt{FI}.

\section{\MB{} Numerical Evaluation Using Bromwich Contours}

We begin from a discussion of numerical integrations over the Bromwich contour introduced in section~\ref{ex4_sec2}.

\subsection{Straight Line Contours and Their Limitations
}
\label{sec:1num}

As shown in section~\ref{chapter-singul} the general form of the \mb{} representation in Eq.~(\ref{MBgenForm}) is well defined if real parts of all gamma functions are positive (equivalent to the separation of all right poles of gamma functions from all left ones). That is accomplished by the appropriate choice of points where integration contours cross the real axis. The final form of \mb{} integrals suited to numerical integration  is following:
\begin{equation}
I = \frac{1}{(2\pi i)^r} \int\limits_{-i \infty + z_{10}}^{+i \infty + z_{10}} \dots \int\limits_{-i \infty + z_{r0}}^{+i \infty + z_{r0}}  \underset{i}{\overset{r}{\Pi}} dz_i \;
 {\bf{F}}(Z,S)  
 \frac{\prod \limits_{j=1}^{N_n} \; \Gamma(\varLambda_j)}{\prod \limits_{k=1}^{N_d} \; \Gamma(\varLambda_k)} f_{\psi}(Z) .
 \label{MBIntgenForm}%MBgenForm
\end{equation}
Notations are the same as in Eq.~(\ref{MBgenForm}), but now it doesn't depend on $\epsilon$. The positions of contours are fixed by
$z_{i0}$. The part $f_{\psi}(Z)$ may depend on polygamma functions and constants like Euler's constant $\gamma$ or it is equal to 1 if the
corresponding Feynman integral has no $\epsilon$ poles. 

To understand the problems appearing during numerical integrations of \mb{} integrals, one has first to study the asymptotic behavior of integrands. The main building blocks of \mb{} integrals are gamma and polygamma functions which have the following asymptotics for large arguments $\lvert z \rvert \rightarrow \infty$
\begin{align}
 \Gamma(z)_{|z| \rightarrow \infty} & = \sqrt{2 \pi} e^{-z}z^{z-\frac12} \left[ 1 + \frac{1}{12z} + \frac{1}{288z^2} + \ldots \right], \label{gammaAsy} \\
 \psi(z)_{|z| \rightarrow \infty}                & =  \ln z - \frac{1}{2z} - \frac{1}{12z^2} + \ldots         \nonumber  \\ 
 \psi^\prime(z)_{|z| \rightarrow \infty}         & =  \frac{1}{z} + \frac{1}{2z^2} + \frac{1}{6z^3} + \ldots  \nonumber  \\ 
 \psi^{\prime\prime}(z)_{|z| \rightarrow \infty} & =  - \frac{1}{z^2} - \frac{1}{z^3} + \ldots                           \\
 \psi^{(3)}(z)_{|z| \rightarrow \infty}          & =  \frac{2}{z^3} + \ldots                                  \nonumber  \\
 \ldots                                          & \,\,\,\,\, .                                               \nonumber
\end{align}

We are interested only in the leading term of the expansion for large $|z|$. It is evident that the polygamma functions, as well as the
linear fractional part of the gamma function asymptotics, do not contribute, and we are looking only at the 
$e^{-z}z^{z-\frac12}$ part of the Stirling formula in Eq.~(\ref{gammaAsy}). It is easy to see that $e^{-z}$ parts from different 
functions cancel each other. That is a general property of \mb{} representations and extends to any multi-dimensional
integral. Based on that the asymptotic behavior of \mb{} integrals is determined by the $z^{z-\frac12}=e^{(z-\frac12)\ln z}$ part of Eq.~(\ref{gammaAsy}).

Let's now consider a simple one-dimensional example
\begin{equation}
 I_{5, \epsilon^{-2}}^{0h0w} = \frac{1}{2s} \frac{1}{2 \pi i} \int\limits_{-i \infty - \frac12}^{+i \infty - \frac12} dz \left( \frac{M^2_Z}{-s} \right)^z 
 \frac{\Gamma^3(-z) \Gamma(1+z)}{\Gamma^2(1-z)}. %\frac{1}{\epsilon^2}
 \label{1dimNumEx1}
\end{equation}
where the ratio of gamma functions, which we will call a {\it{core}} of the \mb{} representation, in the limit $\lvert z \rvert \rightarrow \infty$ is 
\begin{equation}
 \frac{\Gamma^3(-z) \Gamma(1+z)}{\Gamma^2(1-z)} \xrightarrow{|z| \rightarrow \infty} e^{z (\ln z - \ln(-z)) + \frac12 \ln z 
 - \frac52 \ln(-z)}. 
\end{equation}
Here one should notice that independently of the contour of integration, the following relation holds: $\ln z - \ln(-z) = i \pi \, \mbox{sign}(\Im (z))$, see discussion in section~\ref{sec:3complex} and  Eq.~(\ref{appx147}).  

In case of practical applications the integration contour is the Bromwich contour
(a straight line parallel to the imaginary axis),
so $z = z_0 + i \, t, \,\,\, t \in (-\infty, \infty)$ and the {\it{core}} of the \mb{} integral in Eq.~(\ref{1dimNumEx1}) in the limit 
$|z| \rightarrow \infty \Leftrightarrow t \rightarrow \pm \infty$ is
\begin{equation}
 \frac{\Gamma^3(-z) \Gamma(1+z)}{\Gamma^2(1-z)} \longrightarrow e^{- \pi |t|} \frac{1}{\lvert t \rvert^2}.
 \label{Ex1limit}
\end{equation}
It is a well-behaving,  non-oscillating function.

We can extend our reasoning to any multi-dimensional integral. The asymptotic of the representation core in generalized spherical coordinates has the following form
\begin{equation}
  \frac{\prod_j \; \Gamma(\varLambda_j)}{\prod_k \; \Gamma(\varLambda_k)}
  \xlongrightarrow[|z_i| \rightarrow \infty]{r \rightarrow \infty}
  \frac{e^{- \beta r}}{r^{\alpha}}, \,\,\, \beta = \beta(\vec \theta) \geq \pi,
  \,\,\, \alpha = \alpha(z_{i0})
 \label{eq:assymp-spher} 
\end{equation}
where the coefficient $\beta$ depends on direction and the expression is valid only for contours parallel to the imaginary axis.

Now lets look at the kinematic term $\left( \frac{M^2_Z}{-s} \right)^z$. In the Euclidean case $s < 0$,
this factor gives oscillations which are well damped by the factor $e^{- \pi |t|}$. In the Minkowskian case $s \rightarrow s + i
\delta \, (s > 0)$, where $i \delta$ comes from the propagator definition and is needed to choose the proper branch of the
logarithm,
\begin{equation}
 \left( \frac{M^2_Z}{-s} \right)^z = e^{z \ln (- \frac{M_Z^2}{s} + i \, \delta )}  \longrightarrow e^{i \, t \ln \frac{M_Z^2}{s}} e^{- \pi t}, s>0.
\label{eq:MZ2overs} 
\end{equation}
As one can see, the {\it{core}} factor $e^{- \pi |t|}$ cancels with $e^{- \pi t}$ only  when $t \rightarrow - \infty$ and oscillations are not 
damped in general. Moreover, it may happen that an exponent $\alpha$ in the fractional part
$\frac{1}{|t|^{\alpha}}$ 
of the limit in Eq.~(\ref{Ex1limit}) (in our example $\alpha=2$) is not large enough
to guarantee the convergence of the integral (see e.g. Sec.~(3.4) in~\cite{Czakon:2005rk}). 
That is the main problem of numerical integration of \mb{} integrals for physical kinematics.

One of the ways to avoid the described problem is to perform a deformation of the integration path in a way that
the overall exponential damping factor is restored. This will be discussed in section~\ref{sec:shift}.

\subsection{Transforming Variables to the Finite Integration Range \label{sec:finite}} 

In practice, the integration over infinite intervals requires their transformation into finite ones (for example, in \texttt{CUBA} numerical library~\cite{Hahn:2004fe} it is the interval $[0,1]$).
In the package \texttt{MB.m} this transformation is done in the following way
\begin{equation}
 t_i \rightarrow \ln \left( \frac{x_i}{1 - x_i} \right), \;\;\;
       dt_i \rightarrow \frac{d x_i}{x_i(1 - x_i)}.
\label{lntransform}       
\end{equation}
In case of the example in Eq.~(\ref{1dimNumEx1}), the limit $t \rightarrow - \infty$ is equivalent to $x \rightarrow 0$ and in this limit
the integrand behaves like
\begin{equation}
 \frac{1}{x \; \ln^2 x} \xrightarrow{x \rightarrow 0} \infty.
\end{equation}
This singularity is integrable but prevents reaching a high accuracy result.
As an alternative one can transform the integration interval $(- \infty, \infty)$ into $[0,1]$ differently:
\begin{equation}
t_i \rightarrow \tan \left( \pi(x_i - \frac12) \right), \;\;\; 
d t_i \rightarrow \frac{\pi d x_i}{\cos^2\left(\pi(x_i - \frac12)\right)}.
\label{tantransform}
\end{equation}
The corresponding limit now is 
\begin{equation}
 \frac{1}{\sin^2 \left(\pi(x_i - \frac12)\right)} \xrightarrow{x \rightarrow 0} 1,
\end{equation}
and the integration can be easily performed. One should stress that 
with the new type of transformation imaginary parts of arguments of gamma functions grow
much faster than with $\ln$-type transformation. At some moment gamma functions in the denominator 
become equal to 0, numerically. To avoid this problem we compute the {\it{core}} of \mb{} integral in the following way:
\begin{equation}
 \frac{\prod \limits_{j=1}^{N_n} \; \Gamma(\varLambda_j)}{\prod \limits_{k=1}^{N_d} \; \Gamma(\varLambda_k)} =
 \mbox{Exp}\left( \sum \limits_{j=1}^{N_n} \; \ln\Gamma(\varLambda_j) - \sum \limits_{k=1}^{N_d} \; \ln\Gamma(\varLambda_k) \right),
\end{equation}
where $\ln\Gamma$ denotes the log-gamma function (see, e.g. \texttt{CernLib} documentation and appendix~\ref{app:mbnum}).

%------>
Let's consider a 3-dimensional \MB{} integral 
\begin{align}
 I_{2, II}^{0h0w} = & \frac{1}{s^2} \frac{1}{(2 \pi i)^3}  
 \int\limits_{-i \infty - \frac{47}{37}}^{i \infty - \frac{47}{37}} dz_1 
 \int\limits_{-i \infty - \frac{139}{94}}^{i \infty - \frac{139}{94}} dz_2
 \int\limits_{-i \infty - \frac{176}{235}}^{i \infty - \frac{176}{235}} dz_3
 \left(- \frac{s}{m^2} \right)^{-z_1}
 \Gamma(-1 - z_1) \nonumber \\ 
 & \Gamma(2 + z_1) \Gamma(-1 - z_2) \Gamma(z_1 - z_2) \Gamma(1 + z_2 - z_3)^2 \Gamma(-z_3) \Gamma(1 + z_3) \nonumber \\
 & \Gamma(-z_1 + z_3)^2 \Gamma(-z_2 + z_3) /\Gamma(-z_1) \Gamma(1 + z_1 - z_2) \Gamma(1 - z_1 + z_3).
\label{3dimEx12} 
\end{align}
Its derivation, corresponding diagram and files needed for calculations are in \wwwaux{3dimNum}.

The integral in Eq.~(\ref{3dimEx12}) has the cancellation
of the overall damping factor along the 
$z_1(t_1)$-axis ($t_1=t, \, t_2=t_3=0$) or in spherical coordinates along the direction $\vec \theta = (\theta = \pi/2, \phi = 0)$. Numerical results for this integral obtained with different combinations of transformations in Eqs.~(\ref{lntransform})~and~(\ref{tantransform}) are compared with an analytical solution in Tab.~\ref{tab:NumTabI2}. This example is taken from~\cite{Dubovyk:2019krd}.

\begin{table}[!h]
\centering
\begin{tabular}{|l|ll|l|} \hline
\rule{0pt}{2.3ex}\texttt{AB}      & $-1.{\bf199526183135}$ & $+5.{\bf567365907880} i$ &   \\ \hline 
\rule{0pt}{2.3ex}\texttt{MB}$1$   & $-1.{\bf19952}5259137$ & $+5.{\bf56736}7419371 i$ & Cuhre,  $10^7$, $10^{-8}$  \\ \hline 
\rule{0pt}{2.3ex}\texttt{MB}$2$   & $-1.{\bf1995261831}68$ & $+5.{\bf567365907}904 i$ & Cuhre,  $10^7$, $10^{-8}$  \\ \hline 
\multicolumn{4}{|l|}{} \\ \hline
\rule{0pt}{2.3ex}\texttt{MB}$3$   & $-1.{\bf20}4597845834$ & $+5.{\bf567}518701898 i$ & Vegas,  $10^7$, $10^{-3}$  \\ \hline
\multicolumn{4}{|l|}{} \\ \hline
\rule{0pt}{2.3ex}\texttt{MB}$4$   & $-1.{\bf1995}16455248$ & $+5.{\bf5673}76681167 i$ & QMC, $10^6$, $10^{-5}$     \\ \hline
\rule{0pt}{2.3ex}\texttt{MB}$5$   & $-1.{\bf19952}7580305$ & $+5.{\bf56736}7345229 i$ & QMC, $10^7$, $10^{-6}$     \\ \hline
\end{tabular}
\caption{\label{tab:NumTabI2}
Numerical results for the integral Eq.~(\ref{3dimEx12}) for $s = m^2 = 1$. \texttt{AB} - analytical solution~\cite{Aglietti:2004tq}. \texttt{MB}$1$ to \texttt{MB}$5$ -- numerical integration of the MB integrals with different integration routines and transformations of the infinite integration region as described in the text.} 
\end{table}

In the table, the label
\texttt{MB}$1$ corresponds to the numerical integration of Eq.~(\ref{3dimEx12}), where the mapping into
the integration interval $[0,1]$ is done by the $\tan$-type of transformation of Eq.~(\ref{tantransform}) for all variables. 
\texttt{MB}$2$ denotes $\tan$-mapping for $t_1$ and $\ln$-mapping for the remaining variables. 
Integrations are done by the \texttt{CUHRE} routine 
of the \texttt{CUBA} library. The maximum number of integrand evaluations allowed was set to $10^7$. 
The absolute error reported by the routine is at the level of $10^{-8}$.
In \texttt{MB}$3$ the integration is done by the \texttt{VEGAS} routine~\cite{Lepage:1977sw,Lepage:1980dq} and the
$\ln$-type of transformation is used for all variables. Corresponding error estimation is of the order of $\sim 10^{-3}$. 
The last two rows \texttt{MB}$4$ and \texttt{MB}$5$ show results for the numerical integration of Eq.~(\ref{3dimEx12}) and
$\tan$-mapping for all variables with the quasi-Monte Carlo library \texttt{QMC}~\cite{Borowka:2018goh}. 
Numbers in the last column give the maximum number of integrand evaluations and the absolute error.

%--------------->
Integration of \MB{} integrals in the Minkowskian region with only logarithmic mapping and a
deterministic algorithm implemented in \texttt{CUHRE} leads to a \texttt{NaN} result. However, the
Monte-Carlo algorithm in \texttt{VEGAS} can handle integrable singularities and give a few
correct digits for relatively simple integrals.  
As one can see from Tab.~\ref{tab:NumTabI2} the highest real accuracy was obtained in the case \mb$2$. 
Here one should point out that for both results \mb$1$ and \mb$2$ an absolute error returned by the integration routine 
is at the same level $\sim 10^{-8}$.
{\it{This shows that the direct integration of \mb{} integrals is possible but limited by several factors.}}
First, the best accuracy is achieved when the cancellation of the exponential damping factor happens along with one
of the axes in the integration space. In the case of integrals with more than one scale, a cancellation takes place
in multiple directions or in some sector of the integration space, and getting a high accuracy result becomes
much more complicated. Second, the exponent $\alpha$ in $\frac{1}{|t|^\alpha}$ in the asymptotical expansion along the direction of the damping factor cancellation maybe not be big enough for a good convergence. In a multi-dimensional case in contrast to the one-dimensional case, this problem can be solved
by shifts of integration contours $z_i \rightarrow z_{i0} + s_i + i t_i, \, s_i \in Z$ which is the subject of the next section. An appropriate set of $\{s_i\}$ allows
to increase $\alpha$. For example, in case of the integral in Eq.~(\ref{3dimEx12}), taking $z_3 \rightarrow  -176/235 - 1 + i t_3$, one can get damping factor $\alpha = 881/235 \simeq 3.75$
instead of $646/235 \simeq 2.75$.   In addition to the shifted integral, contributions from the residues must be added.
The additional terms are \mb{} integrals with one integration less, see chapter~\ref{chapter-singul}, and hence simpler to evaluate.

\subsection{Shifting and Deforming Contours of Integration \label{sec:shift}}

Let's come back to the example from section~\ref{sec:simpleinvitation} and Eq.~(\ref{parts}), with a solution in Eq.~(\ref{eq:anal3}), see also section~\ref{sec:moregensums}

\begin{eqnarray}
V^{\epsilon^{-1}}_{V3l2m}(s) &=& 
-
% typos corrected 2015-01-21 (Johann Usovitsch)
 ~ \frac{1}{2s} 
\int\limits_{-\frac{1}{2}-i \infty}^{-\frac{1}{2}+i \infty} 
\frac{dz}{2\pi i}
{\left(\frac{-s}{M_Z^2}\right)^{-z}}
{\frac{\Gamma^3(-z)\Gamma(1+z)} {\Gamma(-2z)}}, \label{eq:V3l2m-again2-ch6}
\end{eqnarray}

Doing the same asymptotic analysis as in  secion~\ref{sec:1num} one can find that the integrand behaves like
$\tfrac{1}{\sqrt{t}}$ and this is not enough for a fast convergence. Moreover, shifts in the integration contour in
one-dimensional cases do not affect the asymptotic, see Problem~\ref{problem_num1d}. The analytical solution
for this integral is well defined, so numerical integration should also be possible.

One of the methods for overcoming difficulties connected with slow convergence is the deformation of the integration contours.
It can be written in general form like
\begin{equation}
z_i = z_{i0} + f_i(t_1,...,t_n) + i t_i    
\end{equation}
where we add some real-valued function to our standard contour parallel to the imaginary axis. This function $f_i$ must fulfill two basic conditions. First, it must restore the exponential damping factor in the integration region. Second, it must prevent the crossing of poles of gamma functions. Briefly speaking, when the imaginary part of arguments of gamma functions is equal to zero, the real part should not be equal to zero or negative integers. We automatically have the second condition with the standard contour, but in general, this is a highly non-trivial task. In addition, this function can be chosen to improve the general behavior of the integrand by removing oscillations and making it smoother, see section~\ref{sec:thimbles} and~\cite{Gluza:2016fwh}. 

The simplest ansatz for the function $f_i$ which fulfills the second condition is a linear function $f_i(t_1,...,t_n) = \theta t_i$ with the same coefficient $\theta$ for all integration variables and we have
\begin{equation}
z_i = z_{i0} + (i+\theta)t_i.    
\end{equation}
Formally we rotate all integration contours at the same angle. This approach is described in~\cite{Freitas:2010nx} and works well for certain integrals but it is not general. The main problem is that we have only one parameter to fix asymptotic behavior in the whole integration domain. Let's look at Eq.~(\ref{eq:MZ2overs}) and use a new contour
\begin{equation}
 \left( \frac{M^2_Z}{-s} \right)^z \longrightarrow  e^{t (i + \theta)(\ln \frac{M_Z^2}{s} + i\pi)} \longrightarrow e^{t (\theta \ln \frac{M_Z^2}{s} - \pi)}, s>0.
\label{eq:MZ2overs_2} 
\end{equation}
Keeping in mind the new asymptotic for gamma functions, one can choose $\theta$ in such a way as to control the overall exponential damping everywhere, see Problem~\ref{prob:thetalimit}. The situation becomes more complicated when we go to the integrals with more scales. From
Eq.~(\ref{eq:MZ2overs_2}) we see that even in the Euclidean case, rotated contours give some exponential factor, and if we have more than one kinematic coefficient, it can be in general not possible to find $\theta$ which
will provide exponential damping in all directions.

In Fig.~\ref{fig:contours_fig} a parabolic deformation is shown in addition to the parallel and rotated contours. It is a special case of transformations described in section~\ref{sec:thimbles}. The function $\theta t^2$ doesn't fulfill the condition of not crossing poles of gamma functions, but this type of deformation can be beneficial for one-dimensional integrals. Numerical examples for the contours in Eq.~(\ref{parts}) are given in \wwwaux{1dimNum}.

\begin{figure}[h!]
%\sidecaption
\centering
\begin{tikzpicture}
\begin{axis}[
    legend style={at={(0.03,0.35)},anchor=south west},
    width=10cm, height=7cm, 
    xlabel={$\Re(z)$},
    ylabel={$\Im(z)$},
    xmin=-5.2, xmax=2.2,
    ymin=-2.2, ymax=2.2,
    xtick={-5,-4,-3,-2,-1,0,1,2},
    ytick={-2,-1,0,1,2},
    minor tick num = 4
%    legend pos=north west,
%    ymajorgrids=true,
%    grid style=dashed,
]

\addplot[color=black] coordinates {(-0.5, -2) (-0.5, 2)};
\addlegendentry{$z=z_0 + i t$}

\addplot[color=black, dashed, domain=-2.5:1.5,samples=10]{-0.5-x};
\addlegendentry{$z=z_0 + (i + \theta) t$}

\addplot[color=black, dotted, domain=-4.5:1.5, line width=0.8pt, samples=100]{sqrt(-0.5-x)};
\addplot[color=black, dotted, domain=-4.5:1.5, line width=0.8pt, samples=100]{-sqrt(-0.5-x)};
\addlegendentry{$z=z_0 + \theta t^2 + i t$}

\addplot[color=gray] coordinates {(0, -2.2) (0, 2.2)};
\addplot[color=gray] coordinates {(-5.2, 0) (2.2, 0)};
    
\end{axis}
\end{tikzpicture}
\caption{Contours deformations discussed in~\cite{Dubovyk:2016ocz}.}
\label{fig:contours_fig}       % Give a unique label
\end{figure}
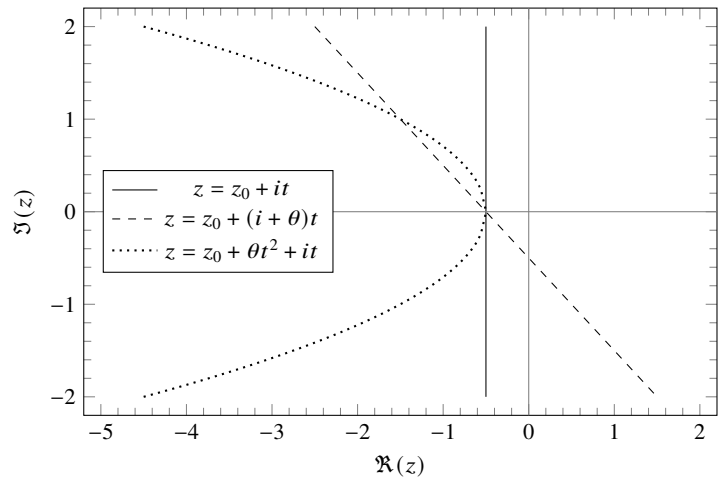 

Contours deformations can be easily implemented in \math{} as seen below.  

\begin{minted}[frame=single,breaklines,fontsize=\small]{mathematica}
In[1]:= V1 = MBint[-((Gamma[-z1]^3*Gamma[1+z1])/(-s)^z1)/(2*s*Gamma[-2*z1]), {{eps->0},{z1->-(1/2)}}];
In[2]:= {int,z0} = V1/.MBint[mb_,{{__},{Rule[z0_,val_],___}}]->{mb,val};
\end{minted}

\begin{minted}[frame=single,breaklines,fontsize=\small]{mathematica}
In[3]:= f1 = z0 + I t;  (* straight contour *)
In[4]:= J1 = D[f1, t];  (* Jacobian *)
In[5]:= intobj1 = J1/(2 Pi I) int /. z1 -> f1;
In[6]:= res1 = NIntegrate[intobj1 /. s -> 2, {t, -Infinity, Infinity}, Method -> DoubleExponential]
Out[6]:= 0.572124 - 0.0364984 I
\end{minted}

\newpage 

\begin{minted}[frame=single,breaklines,fontsize=\small]{mathematica}
In[7]:= f2 = z0 + theta t + I t; (* rotation *)
In[8]:= J2 = D[f2, t];           (* Jacobian *)
In[9]:= intobj2 = J2/(2 Pi I) int /. z1 -> f2;
In[10]:= res2 = NIntegrate[intobj2 /. s -> 2 /. theta -> -1, {t, -Infinity, Infinity}, Method -> DoubleExponential]
Out[10]:= 0.785398
\end{minted}
%\newpage
\begin{minted}[frame=single,breaklines,fontsize=\small]{mathematica}
In[11]:= f3 = z0 + theta t^2 + I t; (* parabolic deformation *)
In[12]:= J3 = D[f3, t];             (* Jacobian *)
In[13]:= intobj3 = J3/(2 Pi I) int /. z1 -> f3;
In[14]:= res3 = NIntegrate[intobj3 /. s -> 2 /. theta -> -1, {t, -Infinity, Infinity}, Method -> DoubleExponential]
Out[14]:= 0.785398
\end{minted}

The exact value of the integral ${\rm V1}$ at $s=2$ is $\pi/2 \simeq 0.78539816339744830962$.
Notice that the numerical evaluation based on a straight-line contour fails to reproduce the correct result.

Another method of \mb{} integrals evaluation is based on the shifts of integration contours.
It relies on  properties of \mb{} integrals described above, and its main feature is that by shifts of integration contours $\{s_i\}$ discussed at the end of section~\ref{sec:finite},
one can change the asymptotic behavior of the integrand and make the absolute value of the integral negligibly small.
The low accuracy of the result for the shifted integral plays no role in this case. This procedure is applied recursively to lower-dimensional integrals which appear after shifts. The algorithm is implemented in the package \texttt{MBnumerics.m}, for more
details see~\cite{Usovitsch:2018shx}. 
\subsection{Thresholds and no Need for Contour Deformations
\label{sec:2num}}

Here we would like to discuss an approach to the construction of \mb{} representations which allows for numerical integration without contour deformations and shifts.

For some Feynman integrals, the \mb{}-suite works without \texttt{MBnumerics} and any additional
tricks because even in the Minkowskian kinematic case, cancellation of the damping factor does not happen, and the integrand of the corresponding \mb{}
integral has always Euclidean-like asymptotic by default. As an example, Tab.~\ref{tab:NumTabI43} gives a result for the integral which corresponds to the Feynman diagram shown in Fig.~\ref{fig:I43}. The numerical evaluation has been obtained just with the \texttt{MB.m} package.

\begin{figure}[h!]
\centering
\begin{tikzpicture}[scale=1.3]
\begin{feynman}
 
\vertex at (0,-2) (i1);
\vertex at (0,2,1) (i2);
\vertex at (1,1) (i3);
\vertex at (1,0) (i4);
\vertex at (1,-1) (i5);
\vertex at (2,.5) (i6);
\vertex at (3,0) (f1);
\vertex at (4,0) (f2);
 
\draw    (0.3,1.4) -- (1,1) ; 
%\draw    (1,1) -- (2,-.5) ;
\draw  [thick, double]  (1,1) -- (1,-1) ;
%\draw    (1,0) -- (1.8,-.5) ;
\draw    (0.3,-1.4) -- (1,-1) ;
%\draw    (1,-1) -- (2,.5) ;
\draw    (2,-.5) -- (2,.5) ;
\draw    (2,.5) -- (1,1) ;
\draw   (2,.5) -- (3,0) ;
\draw    (1,-1) -- (3,0) ;
\draw [thick, double]  (3,0) -- (4,0) ;
 
\node[below] at (0.4,1.2) {{{\bf $0$}}};
\node[above] at (0.4,-1.2) {{{\bf $0$}}};
\node at (3.5,0.2) {{{\bf $s=M_Z^2$}}};
\node[left] at (1,0) {{$M_Z$}};
\node[left] at (2,0) {{0}};
%\node[left] at (1.3,-.5) {{$0$}};
%\node at (1.5,0.1) {{$6$}};  
\node at (1.6,0.84) {{$0$}};
\node at (2.53,.4) {{$0$}};
\node at (2.53,-.4) {{$0$}};
\node at (1.6,-0.84) {{$0$}};
\end{feynman}
\end{tikzpicture}
        \caption{\label{fig:I43}
The two-loop planar vertex diagram.}
\end{figure}
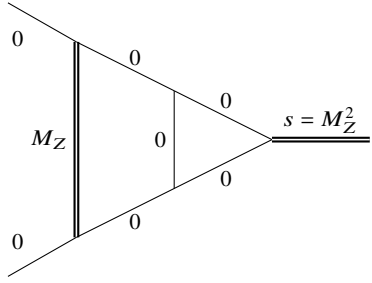

\begin{table}[!h]
\centering
\begin{tabular}{|l|ll|} \hline
               & $s = 1$                  &                              \\ \hline     
\texttt{AB}    & $-5.{\bf46318633201975}$ & $+3.{\bf83353434712221} i$   \\ \hline 
\mb            & $-5.{\bf463186332019}24$ & $+3.{\bf83353434712}128 i$   \\ \hline
               & $s = 2$                  &                              \\ \hline     
\texttt{AB}    & $-2.{\bf01402077672107}$ & $+2.{\bf06704179425209} i$   \\ \hline 
\mb            & $-2.{\bf014020776721}70$ & $+2.{\bf067041794252}42 i$   \\ \hline

\end{tabular}        
\caption{Numerical results for the diagram in Fig.~\ref{fig:I43} for two different values of $s$ and $M_Z^2=1$. 
\texttt{AB} - analytical solution~\cite{Aglietti:2003yc}. 
\mb{} - numerical integration of the corresponding representation is performed by the \texttt{MB.m} package. }
\label{tab:NumTabI43}
\end{table}

As shown in section~\ref{sec:1MBrepr}, Eq.~(\ref{eq:spanning2forest}), the $F$ polynomial for a given Feynman diagram can be written as a sum of two parts

\begin{equation}
\label{eq:spanning2forest_chap6}
 F  = F_0 + U \sum\limits_{i=1}^n x_i m_i^2.
\end{equation}

The first one, denoted as $F_0(x)$, corresponds to a diagram with all massless propagators. It depends on kinematic invariants. The second term $U(x)\sum_i m_i^2 x_i$ depends on masses of internal particles.
To minimize the dimensionality of the representation within the \ga{} approach, {\it{we do not expand the second term}} and we construct the representation as follows
\begin{multline}
\label{eq:thresholdtrick}
G(X) \sim \frac{U(x)^{N_{\nu}-d(L+1)/2}}
{\left( F_0(x) + U(x) {\color{gray}{\sum_i m_i^2 x_i}} \right)^{N_{\nu}-dL/2}} %\\
\sim
{\color{gray}{\prod_i  \left(m_i^2 x_i\right)^{z_i}}}
\frac{U(x)^{N_{\nu}-d(L+1)/2+\sum_i z_i}}
{F_0(x)^{N_{\nu}-dL/2+\sum_i z_i}}.
\end{multline}
This way, we effectively get a massless diagram plus as many additional integrations as many massive propagators we have. In practice, the number of additional integrations equals to the number of different masses in each chain of the diagram. Equal masses in one chain can be collected, giving a linear combination of Feynman parameters. Later such a combination can be simplified by \texttt{1BL} as discussed in section~\ref{sec:BLeff}.

This approach gives optimal dimensionality, but at a price: we lose information about physical and pseudo-thresholds. {\it{By a physical threshold, we mean a kinematic point where the $F$ polynomial starts to be negative. For the pseudo-threshold, one of the terms in $F$ becomes negative.}} For a general two-loop diagram with a single
kinematic invariant $s$, for example, a $Z b \bar b$ vertex,  we can rearrange terms in $F$ and write it schematically in the following  way  
\begin{eqnarray}
   F(x,s) &=& -s \sum\limits_{l,n,k\in \Omega_{F_0}} x_{l} x_{n} x_{k} +   
   \sum m_i^2 \sum\limits_{l,n,k\in \Omega_{m_i}} x_{l} x_{n} x_{k} \label{eq:thbefore}\\
   &=& -s \sum\limits_{l,n,k\in \Omega_{F_0} \backslash \Omega_{m_i}} x_{l} x_{n} x_{k} + 
   \sum \left(m^2_i -s\right) \sum\limits_{l,n,k\in \Omega_{F_0} \cap \Omega_{m_i}} x_{l} x_{n} x_{k} \nonumber \\
   &+& \sum m_i^2 \sum\limits_{l,n,k\in \Omega_{m_i} \backslash \Omega_{F_0}} x_{l} x_{n} x_{k}, \label{eq:thafter}
\end{eqnarray}
where in Eq.~(\ref{eq:thbefore}) the first term corresponds to $F_0$ in
Eq.~(\ref{eq:spanning2forest_chap6}) and the second is an expanded
$U \sum\limits_{i=1}^n x_i m_i^2$ term. In Eq.~(\ref{eq:thafter}) we separated and collected $x_l x_n x_k$ terms with a common  $\sum \left(m^2_i -s\right)$ dependency. There is a physical threshold at $s=0$ and pseudo-thresholds at points $s = \sum m_i^2$.

According to our observations, sometimes it is better to construct the \MB{} representation using the expanded version of the $F$ polynomial where all pseudo-thresholds are separated and collected explicitly. This leads to representations with higher dimensionality. However, in this case, we have a chance to get the exponential damping of the integrand in all kinematic regions.
More precisely, in this way of construction the coefficient $\beta(\vec \theta)$ in Eq.~(\ref{eq:assymp-spher}) fulfills a new condition
$\beta(\vec \theta) \geq 2 \pi$ and the overall damping factor
remains for Minkowskian kinematical points. In this case, we don't need to perform a contour deformation or other tricks. Integration can be done straightforwardly with the logarithmic mapping of Eq.~(\ref{lntransform}) and \texttt{MB.m} package. 

The integral in Fig.~(\ref{fig:I43}) has no intersection between $\Omega_{F_0}$ and the massive part $\Omega_{m_i}$, so explicit thresholds separation is unnecessary, and the integral can be easily integrated without additional tricks. For this example, see \wwwaux{TH1}. In the file, the representation was obtained with the help of the \la{} method. To collect all the peculiarities of the $F$ polynomial, in general, using the \ga{} method is mandatory, even for planar diagrams.

Now let's consider a more complicated example for the diagram shown in Fig.~\ref{fig:npSDMB} with results in Tab.\ref{tab:SDmink}. The diagram has one massive propagator. In addition to the two-dimensional representation in section~\ref{sec:2loopMBgeneral}, Eq.~(\ref{eq:V6l0m2dim}), we have one more integration, and the minimal dimensionality for this integral is three.  
\begin{minted}[frame=single,breaklines,fontsize=\small]{mathematica}
In[1]:= MBreprNP[{1}, {PR[k1, 0, n1] PR[k1 - k2, 0, n2] PR[k2, 0, n3] PR[k1 - k2 + p1, m, n4] PR[k2 + p2, 0, n5] PR[k1 + p1 + p2, 0, n6]}, {k1, k2}];
...
In[2]:= Fauto[0];
...
In[3]:= fupc = m^2 x[1] x[2] x[4] + m^2 x[1] x[3] x[4] + m^2 x[2] x[3] x[4] + m^2 x[1] x[4]^2 + m^2 x[3] x[4]^2 (* + m^2 x[1] x[4] x[5]-s x[1] x[4] x[5]*) + m^2 x[2] x[4] x[5] + m^2 x[4]^2 x[5] - s x[1] x[2] x[6] - s x[1] x[3] x[6] - s x[2] x[3] x[6] - s x[1] x[4] x[6] + m^2 x[2] x[4] x[6] + m^2 x[3] x[4] x[6] + m^2 x[4]^2 x[6] - s x[1] x[5] x[6] + m^2 x[4] x[5] x[6];
\end{minted}
 In the frame above, we show a partial output from  \texttt{AMBREv3} package for this integral. The program allows manual manipulation with the $F$ polynomial. The pseudo-threshold terms \verb|m^2 x[1] x[4] x[5] - s x[1] x[4] x[5]|
are commented.
We can keep the threshold terms in the form \verb|m^2 x[1] x[4] x[5] - s x[1] x[4] x[5]| $\rightarrow$ \verb|m2s x[1] x[4] x[5]|. In this case rescaling of Feynman parameters as in Eq.~(\ref{TransformRule}) reduces dimensionality to 7.
For the specific kinematic point  $s = m^2$ the term $m^2 - s = $ \verb|m2s| can be dropped out, reducing dimensionality by one. The final result after $\epsilon$-expansion and simplifications with {\tt BL}s is 4-dimensional, which is only one dimension higher than the optimal one. As in the case of the diagram in Fig.~\ref{fig:I43}, the \mb{} integrand now does not cause numerical problems, and the integration can be done with \texttt{MB.m} package. The result is given in Tab.~\ref{tab:SDmink}.  All the manipulations for the $F$ polynomial discussed here are given in \wwwaux{TH2}.

At this point, we should stress that at the pseudo-threshold point $s = m^2$, the integral is continuous, and we can drop out the related term in the $F$ polynomial. Calculations at physical thresholds, in our case $s=0$, should be done differently, by analytical expansion around the point $s \ll 1$. After that, we can take the limit $s = 0$ and compute left and right limits $s \rightarrow 0^{ \pm}$.  
For analytical expansion of \mb{} integrals, see section~\ref{sec:1exp}.

 \section{\mb{} Numerical Evaluation by Steepest Descent 
 \label{sec:thimbles}}
 
 The method of steepest descent has already been used in works by Riemann, who applied it to estimate the hypergeometric function, and by Cauchy, Debye, and Nekrasov. For historical records and references to their original works, see~\cite{Petrova:1997}. The method is often used for the asymptotic evaluation of integrals and is based on the idea that many functions in the complex
plane have a stationary point. There is one direction at that stationary point in which the function decreases rapidly, and there is an orthogonal direction in which the function increases rapidly.  In oher words, the stationary point is a saddle point; the function decreases in one direction while it increases in another direction. 
By deforming the integration path so that it goes through the stationary point in the direction in which the function decreases, one can evaluate the integral asymptotically. 

Here we will connect stationary points with \mb{} through the Lefschetz thimbles (\texttt{LT}),   
by searching for a stationary phase contours $\mcC$  as solutions of properly defined differential equations.

\texttt{LT} are applied in research of mathematics, crossing many issues
like behaviour of \texttt{LT} in presence of poles, singularities and branch cuts, behaviour at complex infinity, Stokes phenomenon, relation to relative homology of a punctured Riemann sphere. In physics, it can be applied to the 
analytical continuation of 3d Chern-Simons theory, QCD with chemical potential, resurgence theory, counting master integrals or the repulsive Hubbard model. Applying this method to the numerical evaluation of \MB{} integrals is at the exploratory stage and has been discussed in~\cite{Gluza:2016fwh,Sidorov:2017aea}. 
Here we will present the main idea for the lowest one-dimensional \MB{} integrals,  
in both Euclidean ($s<0$) and Minkowski $(s>0)$ regions. These cases have been explored in fine details in~\cite{Gluza:2016fwh}.
For higher dimensions, even solving two-dimensional \mb{} integrals is not fully understood and explored and can be a potential subject of a nice research work by the reader (see Problem~\ref{problem_mbde}). 

 \subsection{General Idea \label{subsLTgeneralidea}}

Let us write then a general \mb{} integrand $F(z)$, transformed into exponential form. For brevity, we suppress the dependence on $s$ 
and shall use $F(z)$ instead of $F(s,z)$  
 \bq \label{Is}
I(s)=\frac{1}{2\pi i}\int\limits_{\mcC_0}
%c_0-i\infty}^{c_0+i\infty}
dz\,F(z)=\frac{1}{2\pi i}\int\limits_{c_0-i\infty}^{c_0+i\infty}dz\,e^{-f(z)}.
\eq
$\mcC_0$ is a contour  defined by $\mathrm{\Re}(z)=c_0$ while $f(z)=-\ln F(z)$.

The core of the problem with integration over $\mcC_0$ is highly-oscillatory behaviour of the integrand $F(z)$. 
For such a class of integrands, standard methods of numerical integration are often not adequate.

One of possible ways to  get rid of 
 numerical problems with the \MB{} integrand $F(z)$ which is of  highly-oscillatory behaviour (see section~\ref{sec:simpleinvitation} and~\cite{Dubovyk:2016ocz}) is to integrate Eq.~(\ref{Is}) over a new contour 
$\mcC= 
\mcJ_1+\mcJ_2+\mcA$. 

A typical example is sketched in Fig.~\ref{fig:LTcontours} where $\mcC$ is a sum of three contours $\mcJ_1$, $\mcJ_2$ and $\mcA$  along which 
the behaviour of $f$ is under control. 

\begin{figure}[h!]
\sidecaption
\includegraphics[scale=.45]{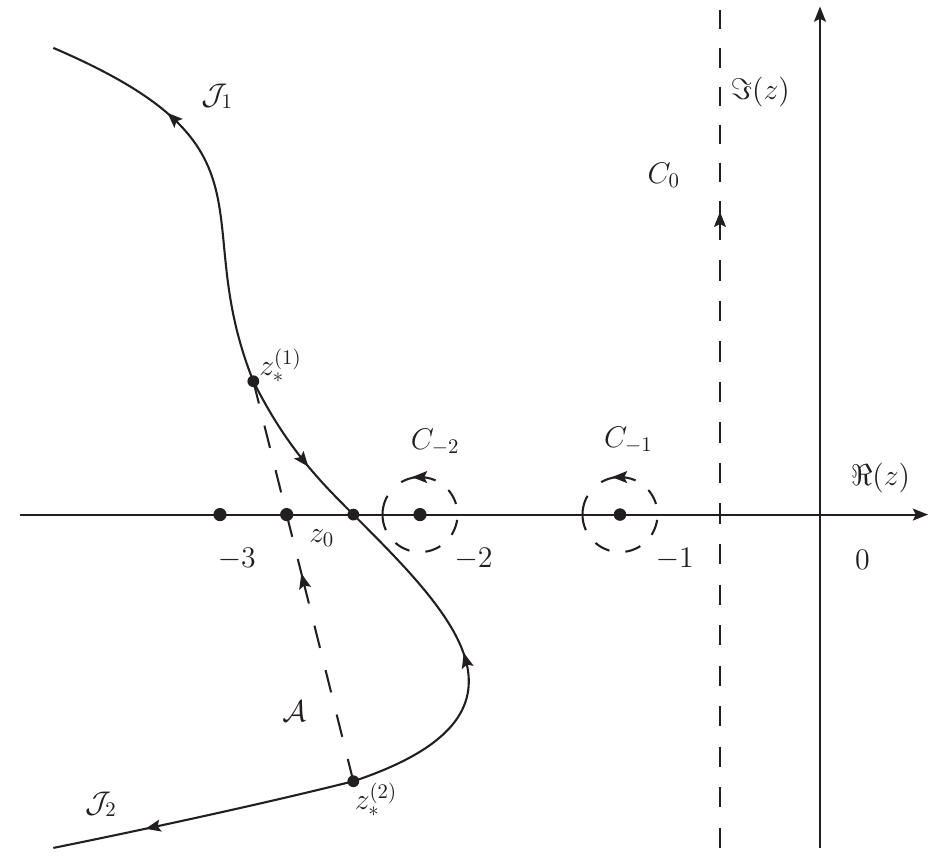}
\caption{A deformation of the integration contour 
$\mcC_0$ defined by $\mathrm{\Re}(z)=c_0$ to a contour $\mcC=\mcJ_1+\mcJ_2+\mcA$.
$\mcJ_{1,2}$ are two Lefschetz thimbles which start at saddle points $z_*^{(1,2)}$ and go towards infinity. 
The compact contour $\mcA$ (interval) connects the two saddle points $z_{*}^{(1)}$ and $z_*^{(2)}$. 
When there is an obstruction in deriving the parameterization of $\mcJ_{1,2}$ around some point, e.g. $z_0$, one can bypass that region 
using the contour $\mcA$. 
Note that here a deformation $\mcC_0\to \mcC$ requires taking into account integrals over two `small' contours, $\mcC_{-2}$ and $\mcC_{-1}$ around poles $z=-2$ and $z=-1$ which contribute to $\sum\res F$ in \eqref{IsC}. }
\label{fig:LTcontours}       % Give a unique label
\end{figure} 

Let us describe in details this decomposition. 

As $\mcJ_{1,2}$ we choose such stationary phase contours 
%$\mcJ_{1,2}$ are two contours 
which start at saddle points $z_*^{(1,2)}$ and go towards infinity without hitting other poles. Both contours are chosen such that  $\mathrm{\Im}(f)$ is constant along them and function $\mathrm{\Re}(f)$ is stricly increasing when one moves away from $z_*^{(1,2)}$. 
Varietes defined in such a way are called 
steepest descent contours~\cite{Bender:1999,Wong:2001} or Lefschetz thimbles~\cite{Pham:1983,Witten:2010cx,Witten:2010zr,Harlow:2011ny,Kanazawa:2014qma,Tanizaki:2014xba}. Usage of $\mcJ_{1,2}$ allows to control the behaviour of $f(z)$ when $z\to\infty$. Because $\mathrm{\Re}(f)$ is stricly increasing, the integrand $e^{-f}$ decreases rapidly at the ends of $\mcJ_{1,2}$. That transform the integral \eqref{Is} into a form which is more suitable for numerical treatment.  

In this method, $\mcA$ is an auxillary compact contour with an explicit analytical para\-meterization. It connects two saddle points on different Lefschetz thimbles $\mcJ_k$.\footnote{More generally, the contour $\mcA$ can connect any two points $z_1$ and $z_2$ which belongs to $\mcJ_1$ and $\mcJ_2$ respectively.  
They do not necessarily have to be saddle points of $f$.} $\mcA$ is choosen such that $F$ does not have singularities along $\mcA$. Then the integral $\int_{\mcA}dz\,F$ can be calculated with 
%arbitrary
high precision with the help of standard numerical methods. The simplest choice of $\mcA$ is an interval spanned between two saddle points $z_*^{(1,2)}$. In some cases of integrands $F$, it is convenient to introduce such an additional contour to improve convergence and accuracy of the numerical integration. 

After deformation $\mcC_0\to\mcC$ the original MB integral \eqref{Is} can be written as
\bq \label{IsC}
I(s)=\sum_{k=1,2}\frac{e^{-i\phi_k}}{2\pi i}\int\limits_{\mcJ_k}dz\,e^{-\mathrm{\Re}(f)}+\frac{1}{2\pi i}\int\limits_{\mcA}dz\,F+\sum\limits_{\mcC_0\to\mcC}\mathrm{Res}\,F,
\eq
where $\phi_k=\left.\mathrm{\Im}(f)\right|_{\mcJ_k}$. According to the Cauchy theorem, the sum over residues in Eq.~\eqref{IsC} is necessary when deformation from $\mcC_0$ to $\mcC$ encircles any poles of $F=e^{-f}$. 
The additional condition is  that during transition from $\mcC_0$ to $\mcC$ one surrounds only a finite set of poles of $F$. In other words, $\mcC$ is such that
the difference $\mcC-\mcC_0$ consists of only a finte sum of `small' closed contours encircling respective poles of $F$. 

The analytical formula describing $\mcJ_k$ can be found only in the simplest cases by explicitly solving the equation $\mathrm{\Im}(f)=\mathrm{const.}$ Instead, we use the fact that the  function $\mathrm{\Re}(f)$ defines a Morse 
flow~\cite{Nicolaescu:2011,Arnold:2012}.  
Such a flow is realized by a parametrization $t\mapsto z(t)$ of $\mcJ_k$ which  obeys the following differential equation~\cite{Witten:2010cx,Kanazawa:2014qma}:
%
% downward flow equation
\begin{equation}\label{LTmaster}
\frac{dz}{dt}=(\partial_zf)^*
\end{equation}
with an initial condition:
\begin{equation}
\lim\limits_{t\to-\infty}z(t)=z_*.
\end{equation}
With respect to various methods known in the literature~\cite{jgll2016,Anastasiou:2005cb, Czakon:2005rk, Freitas:2010nx} which shift/rotate contours or use approximate forms  thereof, this method relies on deriving numerically a parametrization $z(t)$ of $\mcJ_k$ 
%with the help 
as a solution of differential equation \eqref{LTmaster} and then, again numerically, integrating the function  $e^{-\mathrm{\Re}(f)}$ along contour  $\mcC$ composed of Lefschetz thimbles $\mcJ_k$ (and compact contour $\mcA$ if necessary).  
This purely numerical approach is complementary to the Pad\'e approximation~\cite{Gluza:2016fwh}. 

We choose $z_*$ as a starting point for differential equation \eqref{LTmaster}. 
It is worth mentioning that other initial points 
%than saddle points 
are also possible. However, from the practical point of view, it is easier to find a point $z_*$ which satisfies $\partial_zf(z_*)=0$ and starting from that point properly construct $\mcJ_k$,
%, along which $\mathrm{Re}\,f$ is incresing, 
than to find a point $z_{\mathrm{in}}$ lying on some stationary phase contour and check if $\mathrm{\Re}(f)$ is increasing along that contour.  

Let us note that at first sight the presented method seems to be not appropriate for finding a parametrization of a Lefschetz thimble around
%in the vicinity of 
a saddle point $z_*$ which is, at the same time, a zero of the integrand $F$. The reason is that in this region the differential equation \eqref{LTmaster} is not well-defined because $\partial_zf=-\partial_zF/F$ diverges at the saddle point $z_*$. However, one can shift $F$ by a holomorphic function e.g. $F\to\widetilde{F}=F+1$ without changing the value of the integral $I(s)=\int_C dz\,\widetilde{F}=\int_C dz\,F$. Now, $\widetilde{F}(z_*)=1$ and $\partial_z\widetilde{F}(z_*)=0$. Hence one can use $\widetilde{f}=-\ln\widetilde{F}$ in the differential equation \eqref{LTmaster} to derive parametrization of a Lefschetz thimble related to a saddle point $z_*$. 

\subsection{Implementation of the Method \label{subsLTimplmethod}}

%%%%%%%%%%%%%%%%%%%%%%%%%%%%
%%%%%%%%%%%%%%%%%%%%%%%%%%%%
Let us now describe how to numerically derive parametrizations $z_k(t)$ of Lefschetz thimbles $\mcJ_{k}$ using 
%solve in practice 
the differential equation \eqref{LTmaster}. Below we present all crucial steps of the method and 
enumerate the routines used from the  {\tt Mathematica} language. 
To solve \eqref{LTmaster} one has to specify three ingredients: 

(a) the saddle point $z_*$ from which Lefschetz thimble starts, 

(b) a line $l(t)=z_*+t e^{i\beta}$ tangent to the curve $z(t)$ at point $z_*$, and 

(c) a distortion 
\beq\label{dist}
z_*(\epsilon)=z_*+\epsilon e^{i\beta}
\eeq
from $z_*$ along $l(t)$ which is controlled by a small parameter $\epsilon\ll1$, the smaller $\epsilon$ is the more accurate solution of \eqref{LTmaster} one gets. 

In practice, the distorted point $z_*(\epsilon)$ plays a role of the initial point at $t=0$  for the {\tt NDSolve}.  
Such a small deviation   
is necessary to start off the abovementioned numerical routine.  At $z_*$ the derivative $dz/dt$ is zero so setting $z_*$ as the initial point, i.e. $z(0)=z_*$, would only generate static solution $z(t)\equiv z_*$.  

The phase $\beta$ in \eqref{dist} is related to the slope of $l(t)$. To find possible values of $\beta$, one uses the Laurent series of $f$ around $z_*$ and solve the condition for stationary phase $\mathrm{\Im}(f)=\phi$ keeping only leading terms.  
The easiest way to decide which value of $\beta$ corresponds to Lefschetz thimble(s) $\mcJ$ is to solve \eqref{LTmaster} for all possible distortions $z_*(\epsilon)$ in some small range of $t$ 
and then compare along which solution $\mathrm{\Re}(f)$ is increasing. In cases when $\partial^2_zf(z_*)\neq0$ one can use eigenvectors of the Hessian matrix of $\mathrm{\Re}(f)$ at $z_*$ to find directions along which $\mathrm{\Re}(f)$ is increasing.

To find saddle points we use {\tt FindRoot} which looks for a solution of an equation $\partial_zf=0$ in the vicinity of a chosen point. It is convenient to first start looking for $z_*$ within the orginal fundamental region $-1<\mathrm{\Re}(z)<0$. If this fails then one jumps to another region $-2<\mathrm{\Re}(z)<-1$, etc.
In the Minkowski region, the position of a saddle point is not restricted at all, while in the Euclidean region, a saddle point has to be real or it can be complex but then it must come in pair with its complex conjugate $(z_*)^*$. 

After finding a saddle point $z_*^{(1)}$, one construct all its possible distortions \eqref{dist} and use them in {\tt NDSolve}. If {\tt NDSolve} returns two Lefschetz thimbles $\mcJ_{1,2}$ starting from $z_*^{(1)}$ and such that the sum  $\mcJ_1+\mcJ_2$ is well-defined deformation of $\mcC_0$, see Fig.~\ref{fig:LTcontours}, then this step of the method is accomplished and one can go further. 

In cases in which {\tt NDSolve} is not able to find two Lefschetz thimbles starting from one saddle point, one can use the already mentioned trick which relies on using auxillary compact contour $\mcA$. It allows to bypass, in a fully-controlled way, regions in which {\tt NDSolve} does not work well or returns warnings/errors. 
In such a situation, one construct only one Lefschetz thimble $\mcJ_{1}$ which starts at $z_*^{(1)}$ and goes to $\infty$ and then one finds another saddle point $z_*^{(2)}$ and applies the same procedure as for $z_*^{(1)}$. If for $z_*^{(2)}$ it is possible to construct at least one Lefschetz thimble $\mcJ_{2}$ such that $z_*^{(1)}$ and $z_*^{(2)}$ can be connected via a compact contour $\mcA$ and the sum  $\mcJ_1+\mcJ_2+\mcA$ is well-defined deformation of $\mcC_0$, see Fig.~\ref{fig:LTcontours}, then this step of the method is accomplished and one can go to the next step.  
If one fails, e.g. because the obtained contours are closed or they hit poles, then one has to choose another saddle point of $f$ and repeat the whole procedure. 

The final output of {\tt NDSolve} is  
the {\tt InterpolatingFunction} which is further used to integrate numerically the integrand $F$ using {\tt NIntegrate}.

To monitor whether the solution of \eqref{LTmaster} parametrizes a contour which goes towards $\infty$  without hitting a pole, one can make use of the slope of the line which is tangent to that contour:
\beq
%\frac{dy}{dx}=
\frac{\mathrm{\Im}\frac{dz}{dt}}{\mathrm{\Re}\frac{dz}{dt}}=-\tan\left[\arg\left(\partial_zf\right)\right].
\eeq
If $\mathcal{J}$ asymptotically approaches a line $\mathrm{\Im}(f(z\to\infty))=\mathrm{const.}$, then along $\mcJ$, in the limit $t\to\infty$, one gets
\beq
\Delta\theta=\theta_{\pm\infty}+\arg(\partial_zf)|_{z=z(t)}\to n\pi,\quad n\in\mathbb{Z}
\eeq 
In other words, $\Delta\theta$ measures whether $\mcJ$ approaches a line of constant phase $\mathrm{\Im}(f(z\to\infty))=\mathrm{const.}$ when $t\to\infty$. 

 $\theta_{\pm\infty}$ for the stationary-phase contour are defined as follows,
\begin{equation}
\begin{aligned}
z &\simas^{t\rightarrow +\infty} \zinf + i e^{i \theta_{+\infty}} t\,,
\quad t>0\,,\\
z &\simas^{t\rightarrow -\infty} \zinf + i e^{-i\theta_{-\infty}} t\,,
\quad t<0\,,
\end{aligned}
\label{MinkowskiAsymptoticForm}
\end{equation}
their general solutions are discussed in~\cite{Gluza:2016fwh}. 

Let us now discuss numerical integration of $F$ over contour $\mcC$. Taking into account the change of the integration measure given by \eqref{LTmaster}, the integral \eqref{IsC} can be written as:
\beqa\label{Isum}
I(s)&=&\frac{e^{-i\phi_{1}}}{2\pi i}\int\limits_{-\infty}^{+\infty}dt\,\left. (\partial_z f)^*e^{-\mathrm{\Re}(f)}\right|_{z=z_1(t)}-\frac{e^{-i\phi_{2}}}{2\pi i}\int\limits_{-\infty}^{+\infty}dt\,\left.(\partial_z f)^*e^{-\mathrm{\Re}(f)}\right|_{z=z_2(t)}\nonumber\\
&&+\frac{1}{2\pi i}\int\limits_{0}^1 dt\,F\left(z_{\mcA}(t)\right)\frac{dz_{\mcA}(t)}{dt}+\sum\limits_{\mcC_0\to\mcC}\mathrm{Res}\,F,
\eeqa
where $\phi_{1,2}=\mathrm{\Im}(f)|_{\mcJ_{1,2}}$ take into account the fact that the value of $\mathrm{\Im}(f)$ can 
%jump across saddle point where two flows meet. 
be different on each Lefschetz thimble $\mcJ_k$. 
The minus sign in front of the second term in \eqref{Isum} corresponds to the opposite orientation of $\mcJ_2$  with respect to the orientation of $\mcC$, see Fig.~\ref{fig:LTcontours}.
$z_{k}=z_{k}(t)$ are parametrizations of $\mcJ_{k}$ derived from {\tt  NDSolve} as discussed above.
In both cases the flow starts from the saddle points $z_*^{(k)}=z_{k}(t=-\infty)$ but because the distorted point $z_*(\epsilon)=z(t=0)$ is close to $z_*=z(t=-\infty)$, due to the small parameter $\epsilon$, one can integrate over $t\in (0,+\infty)$ instead of $t\in (-\infty,+\infty)$.
In practice one integrates over $(0,t_{\mathrm{max}})$ where $t_{\mathrm{max}}$ is chosen such that all contributions from the integrand, up to declared precision, are taken into account. Varying $t_{\mathrm{max}}$ provides additional test whether the method works properly. When one makes $t_{\mathrm{max}}$ bigger and bigger then $I(s)$ should asymptotically approach a finite value. 
Finally, $z_{\mcA}(t)$ used in \eqref{Isum} is an explicit analytical parameterization of the contour $\mcA$. When $\mcA$ is an interval then its parameterization can be of the following form:
\beq
t\mapsto z_{\mcA}(t)=z_2+t(z_1-z_2).
\eeq

\begin{tips}{\texttt{MBDE} solutions in Euclidean and Minkowskian kinematics}
Let us consider the badly behaving integrand Eq.~(\ref{parts}), considered in sections~\ref{sec:simpleinvitation}~and~\ref{sec:moregensums}, see  Eqs.~(\ref{eq:V3l2m-again})~and~(\ref{eq:V3l2m-res})~\cite{Czakon:2005rk,Gluza:2016fwh}.

\begin{eqnarray}
F_1(z) &=& 
 ~  
{(-s)^{-z}}
{\frac{\Gamma^3(-z)\Gamma(1+z)} {\Gamma(-2z)}} \label{eq:LTex1}
\end{eqnarray}
at three different kinematic points  $s=-1/20$, $1+i \delta$, $5$. Functional behavior of the function is given in Figs.~\ref{sln20-f1-sm120}--\ref{sln20-f1-s5} with the aid
of a function designed to compress the vertical scale (left side figures),

\begin{equation}
\sln_{m} x \equiv \mathrm{sign}(x)\,\ln(1+ |x| e^m)\,.
\end{equation}

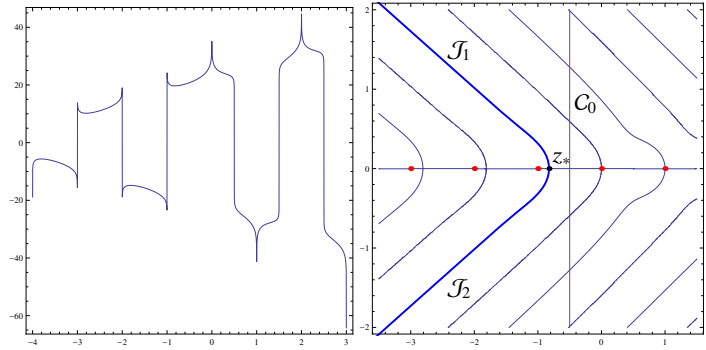
\begin{figure}[h!]
\begin{center}
\includegraphics[scale=0.5]{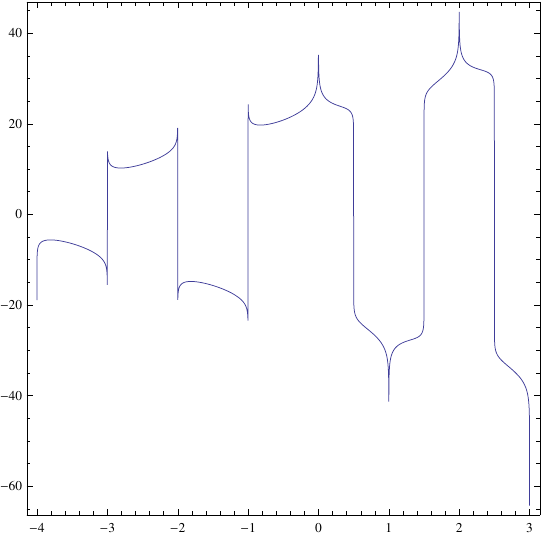}
\includegraphics[scale=0.5]{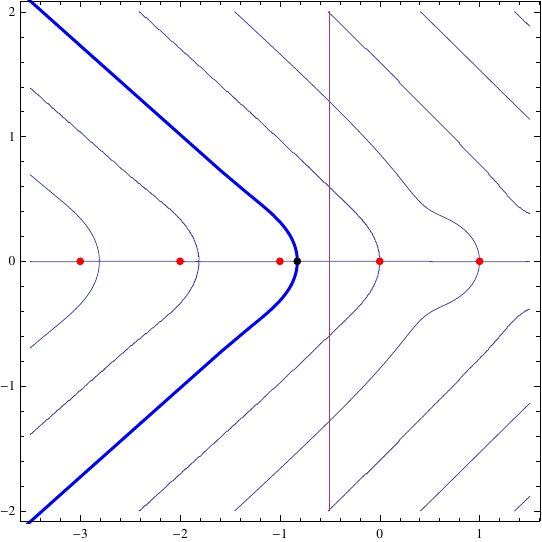}
\begin{picture}(0,0)
\put(-100,110){\small $\mcJ_1$}
\put(-100,20){\small $\mcJ_2$}
\put(-60,70){\small $z_*$}
\put(-52,90){\small $\mcC_0$}
\end{picture}
\caption{
A plot of $\mathrm{sln}_{20}\,F_1$ for $s=-1/20$ (left).
$\mathrm{\Im}(f_1)=0$ contours for $s=-1/20$. Thick blue line corresponds to the solution of \eqref{LTmaster} with $z_*=-0.8256$ (right).  
}\label{sln20-f1-sm120}
\end{center}
\end{figure}

\begin{figure}[h!]
\begin{center}
\includegraphics[scale=0.5]{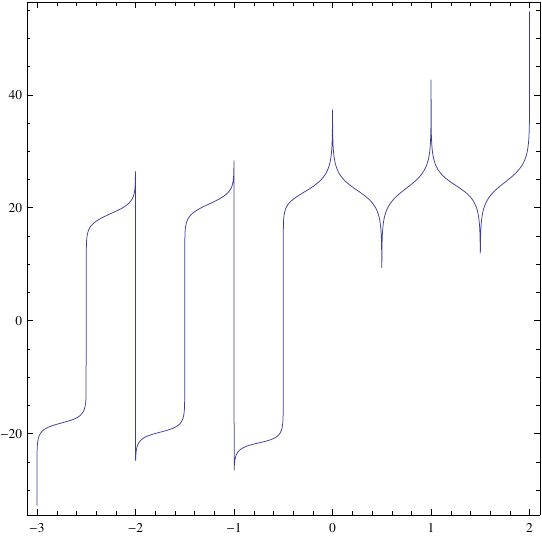}
\includegraphics[scale=0.5]{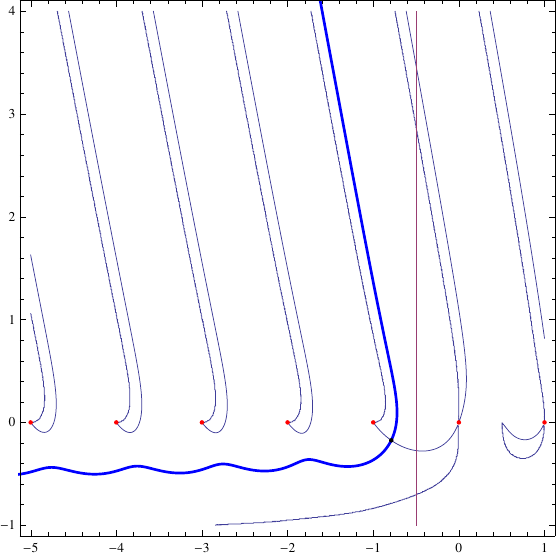}
\begin{picture}(0,0)
\put(-55,125){\small $\mcJ_1$}
\put(-120,12){\small $\mcJ_2$}
\put(-43,22){\small $z_*$}
\put(-47,85){\small $\mcC_0$}
\end{picture}
\caption{
 A plot of $\sln_{20}\,\mathrm{\Re}\,(F_1)$ for $s=1+i \delta$ (left).
$\mathrm{\Im}(f_1)=2.289$ contours for $s=1+i \delta$. Thick blue line corresponds to the solution of \eqref{LTmaster} with $z_*=-0.7893-0.1745i$ (right).  
}\label{sln20-f1-s1}
\end{center}
\end{figure}
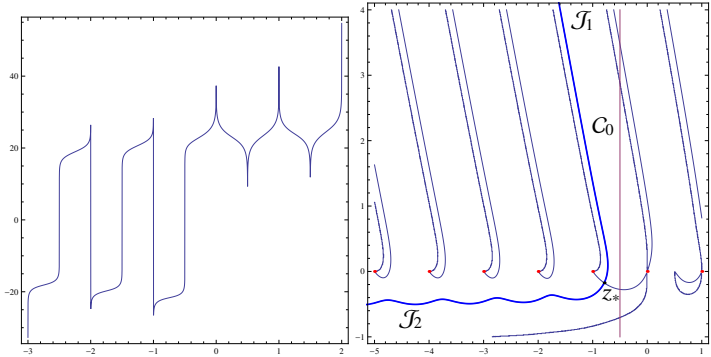

\begin{figure}[h!]
\begin{center}
\includegraphics[scale=0.5]{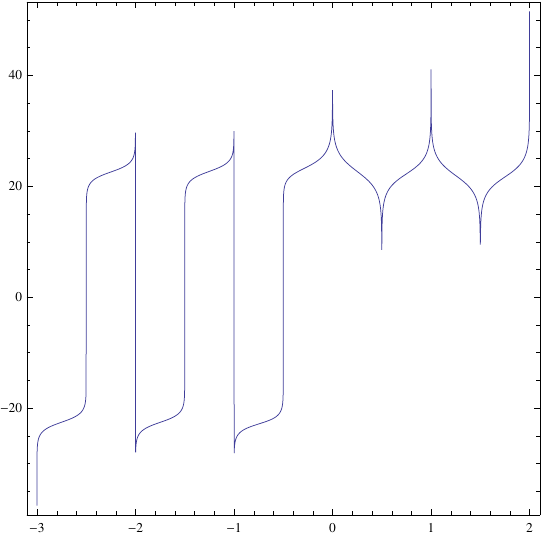}
\includegraphics[scale=0.5]{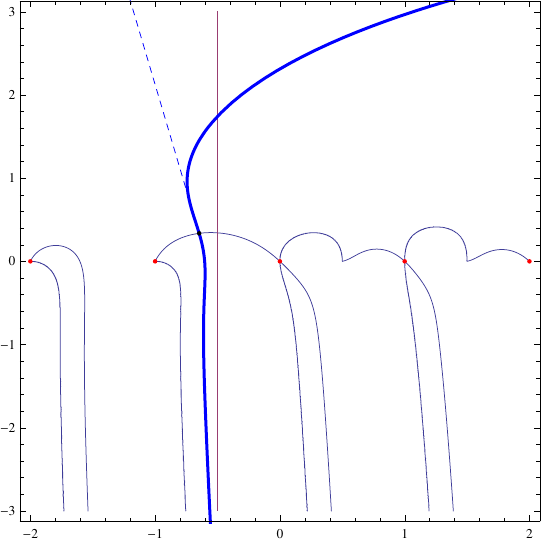}
\begin{picture}(0,0)
\put(-70,120){\small $\mcJ_1$}
\put(-94,20){\small $\mcJ_2$}
\put(-93,76){\small $z_*$}
\put(-90,120){\small $\mcC_0$}
\end{picture}
\caption{
A plot of $\mathrm{sln}_{20}\,\mathrm{\Re}(F_1)$ for $s=5$ (left).
$\mathrm{\Im}(f_1)=-1.895$ contours for $s=5$. Thick blue line corresponds to the solution of \eqref{LTmaster} with $z_*=-0.6470+0.3366i$ (right).  
}\label{sln20-f1-s5}
\end{center}
\end{figure}

The file 
\texttt{MBDE\_Springer.nb} with a complete solution can be found in the book auxiliary files repository~\cite{www_aux_springer}.
Based on  a discussion in section~\ref{subsLTgeneralidea}, the easy to understand  major steps for the integral evaluation of the function in Eq.~(\ref{eq:LTex1}) are easy to follow there.
 
\end{tips}

\subsection{Tricks and pitfalls for `Vanishing Derivatives'}

Let us consider the integrand
$F_4(z)=F_1(z)\psi^4(1-z)$,  $F_4^{(0-3)}(z_*)=0$ for $z_*\in(-1,0)$, $s=-2$. 
In some cases, for example $F_4$ in the Minkowski region, it is hard to find a saddle point $z_*$ through which a contour goes to infinity, see Fig.~\ref{sln20-f4-s2}. 
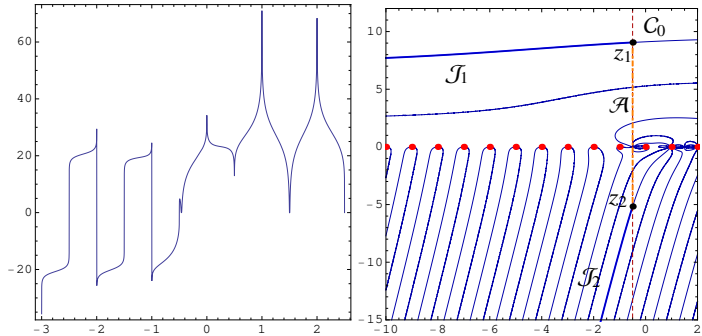
\begin{figure}[h!]
\begin{center}
\includegraphics[scale=0.5]{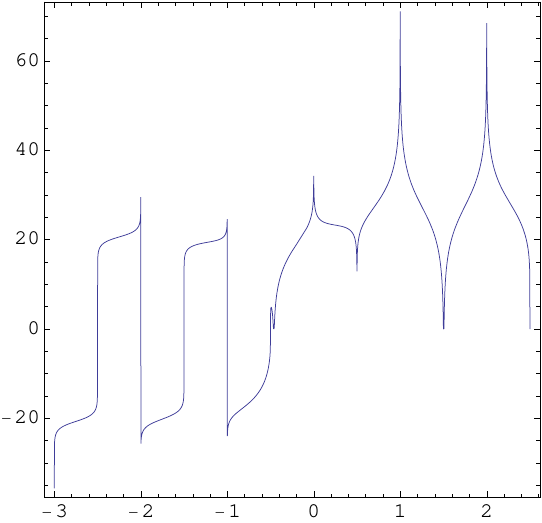}
\includegraphics[scale=0.5]{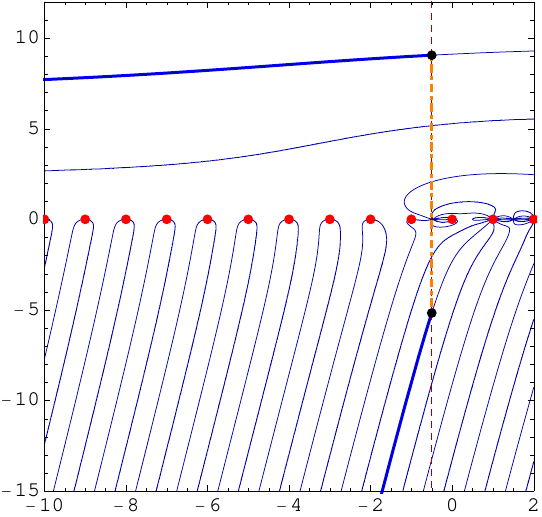}
\begin{picture}(0,0)
\put(-100,97){\small $\mcJ_1$}
\put(-50,20){\small $\mcJ_2$}
\put(-36,105){\small $z_1$}
\put(-38,50){\small $z_2$}
\put(-25,115){\small $\mcC_0$}
\put(-38,85){\small $\mcA$}
\end{picture}\caption{
A plot of $\mathrm{sln}_{20}\,F_4$ for $s=2$ (left).
$\mathrm{\Im}(f_4)=-2.793$ contours for $s=2$. Thick blue lines correspond to contours $\mcJ_{1,2}$ which go to $\infty$. They intersect $\mathrm{\Re}(z)=-1/2$ line at $z_1=-1/2+9.057i$ and $z_2=-1/2-5.169i$  respectively. The orginal integration contour $\mcC_0$ is deformed to $\mathcal{J}_1+\mcA+\mathcal{J}_2$, where $\mathcal{J}_1$ is an upper part of the contour $\mcC$, starting at $z_1$, $\mathcal{J}_2$ is a lower part of the contour $\mcC$, starting at $z_2$, and $\mcA=\{z\in\mathbb{C}:z=z_2+t(z_1-z_2), t\in\left<0,1\right>\}$ is an interval connecting $z_1$ and $z_2$, orange, dot-dashed line (right).}\label{sln20-f4-s2}
\end{center}
\end{figure}
Then one can use the following trick. One looks for a new contour $\mcC$ which can be decomposed into three pieces $\mcJ_1+\mcA+\mcJ_2$, where $\mcJ_{1,2}$ are 
two $\mathrm{\Im}\,(F_4)=0$ contours ending at $\infty$ while $\mcA=\{z\in\mathbb{C
}:z=z_2+t(z_1-z_2), t\in\left<0,1\right>\}$ is an interval connecting $z_1$ and 
$z_2$. These two points lie on $\mcJ_1$ and $\mcJ_2$, respectively, and are such
 that $F_4^{(k)}(z_{1,2})\neq0$. The additional requirement is that $\mcC_0$ can
 be deformed into $\mcC$ by encircling only finite number of poles of $F_4$. For
 example, $z_{1,2}$ can be choosen as intersection points of two $\mathrm{\Im}\,(F
_4)=0$ contours with the line $\mathrm{\Re}(z)=-1/2$, see Fig.~\ref{sln20-f4-s2} (right).
From the practical point of view, there is a way to control if numerically derived $\mcJ_{1,2}$ really go to $\infty$ and do not hit a zero or pole. Namely, one can use \eqref{LTmaster} to monitor the asymptotic behaviour of $\mathcal{J}_{1,2}$:
\begin{equation}
\frac{dy}{dx}=\frac{\mathrm{\Im}\frac{dz}{dt}}{\mathrm{\Re}\frac{dz}{dt}}=-\tan\left[\arg\left(\partial_zf_4\right)\right].
\end{equation}
So, for $t\to\infty$ the value of $dy/dx$ along $\mcJ_{\pm}$ should asymptotically approach the slope of the lines defined by $\mathrm{\Im}(f_4(z\to \infty))=0$.

%\end{tips}

\subsection{Beyond One-dimensional Cases}

There is so far not much developement within the LT method beyond the one-dimensional case. Eq.~(\ref{LTmaster}) can be generalized to
\begin{equation}\label{flow}
 \frac{d z_{i}}{d t}=-\left(\frac{\partial f}{\partial z_{i}}\right)^{*}, \quad i=1,...,n.
 \end{equation}
 This allows us to define Lefschetz thimbles associated with a given critical point as
%jg
 \begin{equation}
 \mathcal{J}_{\alpha}=\lbrace z\in \mathbb{C}^{n}: \lim_{t\rightarrow \infty} z(t)=z_{\alpha} \rbrace,
 \end{equation}
 where $z(t)$ is the solution of the flow equation \eqref{flow}.
The Lefschetz thimble is a collection of curves flowing in the direction of the steepest descent that asymptotically reach the corresponding critical point. Note that the real part of the integrand for critical points are at the same time saddle points. There exists an alternative direction in which the real part changes the most, i.e., the direction of the steepest ascent. With this direction are associated geometrical structures similar to the Lefschetz thimbles, called anti-thimbles, which satisfy the system of upward flow differential equations
%jg
 \begin{equation}
 \mathcal{K}_{\alpha}= \lbrace z \in \mathbb{C}^{n} : \lim_{t\rightarrow \infty} \bar{z}(t)=z_{\alpha} \rbrace,
 \end{equation}
 where $\bar{z}(t)$ are curves that satisfy the following system of upward flow equations
 \begin{equation}
 \frac{d z_{i}}{d t}=\left(\frac{\partial f}{\partial z_{i}}\right)^{*}, \quad i=1,...,n.
 \end{equation}
 The idea of thimbles as a surface of steepest descent appeared in~\cite{Pham:1983}. It has been shown there that within the homology theory, an initial contour of integration $\mathcal{C}$ can be changed into a finite sum of thimbles, each associated with one critical point
\begin{equation}
 \mathcal{C}=\sum_{\alpha}N_{\alpha}\mathcal{J}_{\alpha}.
 \label{split}
 \end{equation}
 The sum in the above equation goes over critical points and $N_{\alpha} \in \mathbb{Z}$. This way, we can change the integration over one surface into the sum of integrals over thimbles.
See W.~Flieger's contribution to~\cite{Blondel:2018mad} for more details.
 
\section{Numerical evaluation of phase space \MB{} integrals} 

The tools described above for the numerical evaluation of \MB{} integrals can also be very effective for obtaining numerical solutions for phase space integrals. The power of the \MB{} approach is nicely demonstrated in the following example. Let us consider an angular integral with three denominators,
\begin{equation}
I_{j,k,l} = \int d\Omega_{d-1}(q) \frac{1}{(p_1\cdot q)^j (p_2\cdot q)^k (p_3\cdot q)^l},
\end{equation}
where the fixed momenta $p_1$, $p_2$ and $p_3$ are massless and we are integrating over the angular variables of the (massless) momentum $q$. Such integrals appear e.g. when integrating over soft radiation beyond NLO accuracy~\cite{Somogyi:2008fc,Bolzoni:2009ye,Bolzoni:2010bt}. Let us attempt to evaluate the above integral numerically. One obvious choice is to use a parametrization in terms of angular variables as described in section~\ref{sec:4MBrepr}. So let us choose a Lorentz-frame such that
\begin{equation}
    \begin{split}
        p_1^\mu &= (1,\vec{p}_1) = (1,\vec{0}_{d-2},1),
        \\
        p_2^\mu &= (1,\vec{p}_2) = (1,\vec{0}_{d-3},
        \sin\chi_2^{(1)},
        \cos\chi_2^{(1)}),
        \\
        p_3^\mu &= (1,\vec{p}_3) = (1,\vec{0}_{d-4},
        \sin\chi_3^{(2)}\sin\chi_3^{(1)},
        \cos\chi_3^{(2)}\sin\chi_3^{(1)},
        \cos\chi_3^{(1)}).
    \end{split}
    \label{eq:pmu-def-again}  
\end{equation}
In this frame, the vector $q^{\mu}$ has the form
\begin{equation}
    \begin{split}
        q^\mu = (1,..\mbox{``angles''}..,
        \cos\vartheta_n\prod_{k=1}^{n-1}\sin\vartheta_k,
        \cos\vartheta_{n-1}\prod_{k=1}^{n-2}
        \sin\vartheta_k,\ldots,
        \cos\vartheta_2\sin\vartheta_1,
        \cos\vartheta_1),
    \end{split}
    \label{eq:qmu-def-again}
\end{equation}
where  $..\mbox{``angles''}..$ denotes those angular variables that can be integrated trivially (since the integrand does not depend on them). In this frame, our integral $I_{j,k,l}$ becomes
\begin{equation}
\begin{split}
I_{j,k,l} &= \int d\Omega_{d-4} 
	\int_{-1}^{1} d(\cos\vartheta_1)
%	\int_{-1}^{1} 
	d(\cos\vartheta_2)
%	\int_{-1}^{1} 
	d(\cos\vartheta_3)(\sin\vartheta_1)^{-2\eps}(\sin\vartheta_2)^{-1-2\eps}
\\&\times
	(\sin\vartheta_3)^{-2-2\eps}
	(1-\cos\vartheta_1)^{-j}
	(1-\cos\chi_2^{(1)}\cos\vartheta_1
		-\sin\chi_2^{(1)}\sin\vartheta_1\cos\vartheta_2)^{-k}
\\&\times
	(1-\cos\chi_3^{(1)}\cos\vartheta_1
		-\cos\chi_3^{(2)}\sin\chi_3^{(1)}\sin\vartheta_1\cos\vartheta_2
\\&\quad	
		-\sin\chi_3^{(2)}\sin\chi_3^{(1)}\sin\vartheta_1\sin\vartheta_2\cos\vartheta_3)^{-l}.
\end{split}
\label{eq:Ijkl}
\end{equation}
However, the straightforward evaluation of this integral runs into the following difficulty. If at least two of the the exponents $j$, $k$ and $l$ are positive integers (at $\eps=0$), already the integrand has a line singularity inside the integration region. Indeed, if say $j=1$ and $k=1$, the integrand is singular not only at $\cos\vartheta_1 = 1$, but also along the line $\cos\vartheta_1 = \cos\chi_2^{(1)}$ whenever $\cos\vartheta_2 = 1$. Obviously these singularities lead to poles in $\eps$ which must be made explicit before any actual numerical computation can be attempted. However, the presence of the line singularity means that the usual procedure of making these poles explicit, \emph{sector decomposition} (see e.g.~\cite{Heinrich:2008si} and references therein) is not directly applicable to Eq.~(\ref{eq:Ijkl}), as it relies on the assumption that the only singularities are on the integration boundaries. If all three of the exponents are positive integers, the last factor in Eq.~(\ref{eq:Ijkl}) develops singularities inside the integration region as well. These divergences must then be treated on a case-by-case basis before numerical integration can proceed.
\begin{tips}{Removing Line Singularities by Partial Fractioning}
In certain cases the partial fraction decomposition of the integrand can be applied to eliminate line singularities. To see this in its simplest form, consider the case $j=1$, $k=1$ and $l=0$. Clearly we are in the situation described above, where in addition to the singularity on the integration boundary at $\cos\vartheta_1 = 1$, the line singularity at $\cos\vartheta_1 = \cos\chi_2^{(1)}$, $\cos\vartheta_2 = 1$ is also present. However, we can remove this singularity as follows. Let us perform partial fractioning of the denominator,
\begin{equation}
\frac{1}{(p_1\cdot q) (p_2\cdot q)} =
\left[\frac{1}{(p_1\cdot q)} + \frac{1}{(p_2\cdot q)}\right]\frac{1}{(p_1+p_2)\cdot q}.
\end{equation}
Now, it is easy to check that the denominator $(p_1+p_2)\cdot q$ is no longer singular. Hence, all singularities are now associated with the two separate denominators in the bracket. The first term is only divergent at $\cos\vartheta_1 = 1$ in the chosen frame and we can proceed to apply sector decomposition to resolve the $\eps$ pole. The second term has the line singularity \emph{in this particular frame}. However, we are free to use rotational invariance to compute this second integral in another Lorentz-frame, where the roles of $p_1$ and $p_2$ are interchanged! In this rotated frame, the line singularity would appear in the denominator $(p_1\cdot q)$, while $(p_2\cdot q)^{-1}$ is singular only at $\cos\vartheta'_1 = 1$, where $\cos\vartheta'$ is the variable in the rotated Lorentz-frame. Thus, we do not have to deal with the line singularity explicitly.

Another approach to dealing with the line singularity in this particular case would be to split the region of integration in $\cos\vartheta_1$ as 
\begin{equation}
\int_{-1}^{1} d(\cos\vartheta_1) = 
	\int_{-1}^{\cos\chi_2^{(1)}} d(\cos\vartheta_1)
	+ \int_{\cos\chi_2^{(1)}}^{1} d(\cos\vartheta_1).
\end{equation}
Obviously the two integrals only have singularities on the boundaries of the integration region in this case. Note however, that this approach requires the precise knowledge of the geometry of the singularities inside the integration region, which in general can be quite elaborate.
\end{tips}
\begin{svgraybox}
However, all of these complications can be sidestepped immediately by using the \MB{} representation of the angular integral, discussed in section~\ref{sec:4MBrepr}. 
\end{svgraybox}
Indeed, for general exponents $j$, $k$ and $l$ we find~\cite{Somogyi:2011ir}
\begin{equation}
\begin{split}
I_{j,k,l} &= 2^{2-j-k-l-2\eps} \pi^{1-\eps} \frac{1}{\Gamma(j)\Gamma(k)\Gamma(l)
	\Gamma(2-j-k-l-2\eps)}
	\int_{-i\infty}^{+i\infty} \frac{dz_1\,dz_2\,dz_3}{(2\pi i)^3}
\\ &\times
	\Gamma(-z_1) \Gamma(-z_2) \Gamma(-z_3)
%\\ &\times
	\Gamma(j+z_1+z_2) \Gamma(k+z_1+z_3) \Gamma(l+z_2+z_3)
\\ &\times
	\Gamma(1-j-k-l-\eps-z_1-z_2-z_3)
	v_{12}^{z_1}
	v_{13}^{z_2}
	v_{23}^{z_3},
\end{split}
\end{equation}
where
\begin{equation}
v_{12} = \frac{p_1\cdot p_2}{2},
\qquad
v_{13} = \frac{p_1\cdot p_3}{2}
\qquad\mbox{and}\qquad
v_{23} = \frac{p_2\cdot p_3}{2}.
\end{equation}
Thus, for any specific exponents $j$, $k$ and $l$ we can directly proceed with resolving the poles, performing the $\eps$-expansion and evaluating the resulting \MB{} integrals numerically. As an example, consider e.g. $j=1$, $k=1+\eps$ and $l=\eps$. This particular set of exponents appears when integrating over soft-gluon emission from the interference of tree-level and one-loop matrix elements, as happens in QCD computations at NNLO accuracy~\cite{Somogyi:2008fc,Bolzoni:2010bt}. Then we must evaluate the integral
\begin{equation}
\begin{split}
I_{1,1+\eps,\eps} &= 2^{-4\eps} \pi^{1-\eps} \frac{1}{\Gamma(1+\eps)\Gamma(\eps)
	\Gamma(-4\eps)}
	\int_{-i\infty}^{+i\infty} \frac{dz_1\,dz_2\,dz_3}{(2\pi i)^3}
	\Gamma(-z_1) 
\\ &\times
	\Gamma(-z_2) \Gamma(-z_3)
	\Gamma(1+z_1+z_2) \Gamma(1+\eps+z_1+z_3) \Gamma(\eps+z_2+z_3)
\\ &\times
	\Gamma(-1-3\eps-z_1-z_2-z_3)
	v_{12}^{z_1}
	v_{13}^{z_2}
	v_{23}^{z_3}.
\end{split}
\label{eq:I3-example}
\end{equation}
In this particular case, we must address one further subtlety: it is relatively straightforward to see that we cannot find straight line contours running parallel to the imaginary axis such that the real parts of the arguments of all gamma functions involving integration variables are positive, for any real value of $\eps$. Indeed, clearly on the one hand we must have $\Re(z_i) < 0$ from the first three gamma functions in the integrand. But then, looking at the last gamma function on the second line, $\Gamma(\eps+z_2+z_3)$, we deduce that $\eps > 0$. However, the sum of the arguments of the last three gamma function in the integrand, $\Gamma(1+\eps+z_1+z_3)$, $\Gamma(\eps+z_2+z_3)$ and $\Gamma(-1-3\eps-z_1-z_2-z_3)$ is $-\eps+z_3$. If the real parts of all arguments where positive, this would imply $\Re(-\eps+z_3)>0$ and given that $\Re(z_3) < 0$, we derive $\eps < 0$, in contradiction to our previous observation! One way to proceed is to introduce an auxiliary regulator, say we set $l=\eps$ to $l=\eps + \delta$. Now it becomes possible to find straight line contours and we may then proceed to analytically continue the integral first to $\delta \to 0$ and then to $\eps \to 0$. This strategy has been also discussed in chapter~\ref{chapter-singul}. It is imperative to check that this procedure does not introduce any poles or logarithms of the regulator $\delta$, so that the $\delta\to 0$ limit exists and is finite (in $\delta$). The algorithmic procedure for performing the analytic continuations was detailed in chapter~\ref{chapter-singul} and we can implement the computation in \math{} using the {\tt MB.m} package as follows (see also \wwwaux{RealPSNum})
\begin{minted}[frame=single,breaklines,fontsize=\small]{mathematica}
In[1]:= Ijkl = 2^(2-j-k-l-2ep) Pi^(1-ep) * 1/(Gamma[j] Gamma[k] Gamma[l] Gamma[2-j-k-l-2ep]) Gamma[-z1] Gamma[-z2] Gamma[-z3] Gamma[j+z1+z2] Gamma[k+z1+z3] Gamma[l+z2+z3] Gamma[1-j-k-l-ep-z1-z2-z3] v12^z1 v13^z2 v23^z3 /. {j->1, k->1+ep, l->ep+delta};
In[2]:= contour = MBoptimizedRules[Ijkl, delta->0, {}, {delta}];
In[3]:= Ijkl = MBcontinue[Ijkl, delta->0, contour];
In[4]:= Ijkl = MBexpand[Ijkl, 1, {delta, 0, 0}];
In[5]:= FreeQ[Ijkl /. MBint[int_, rule_]->int, delta]
Out[5]:= True
In[6]:= Ijkl = Ijlk /. {MBint[int_,rule_] :> MBint[int, {Select[rule[[2]], !FreeQ[#,ep]&], Select[rule[[2]], FreeQ[#,ep]&]}];
In[7]:= Ijkl = Ijkl/. {MBint[int_, rule_]:>MBcontinue[int, ep->0, rule];
\end{minted}
The fifth input checks that indeed the $\delta\to 0$ limit is free of $\delta$ (and hence does not involve poles of logarithms of $\delta$), while the sixth input brings the rule in {\tt MBint} to a form where the analytic continuation can then performed in $\eps$ in the next step. (For further details, consult the documentation of the {\tt MB.m} package~\cite{Czakon:2005rk}.) 

\begin{figure}
\begin{center}
\includegraphics[width=0.5\textwidth]{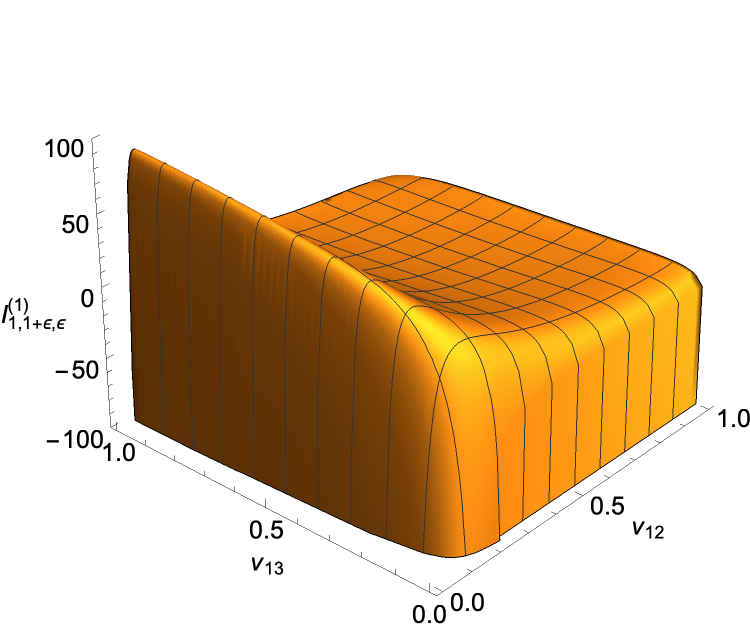}
\end{center}
\caption{\label{fig:I3-expample} Numerical result for the third expansion coefficient  $I^{(1)}_{1,1+\eps,\eps}$ defined in Eqs.~(\ref{eq:I3-example})~and~(\ref{eq:I3-example-exp}) for the co-planar case $\cos\chi_3^{(2)}=-1$. In this case, we have simply $v_{12} = (1-\cos\chi_{2}^{(1)})/2$ and $v_{13} = (1-\cos\chi_{3}^{(1)})/2$.}
\end{figure}

Now we are ready to expand in $\eps$ to the desired order and obtain a numerical solution for the expansion coefficients. In fact, in this example, the expansion starts at $1/\epsilon$:
\begin{equation}
I_{1,1+\eps,\eps} = 
\frac{1}{\eps} I^{(-1)}_{1,1+\eps,\eps}
+
I^{(0)}_{1,1+\eps,\eps}
+
\eps I^{(1)}_{1,1+\eps,\eps} 
+
{\mathcal O}(\eps^2)
\label{eq:I3-example-exp}
\end{equation}
and it turns out that the coefficient of $\frac{1}{\ep}$ and the finite part are given by zero dimensional \MB{} integrals, while the coefficient of $\ep$ involves only one-dimensional \MB{} integrals. However, three-dimensional \MB{} integrals appear for all higher-order expansion coefficients.

In terms of the numerical evaluation, we note that for real angles, all variables $v_{ij}$ are positive and so in this sense, we are dealing with `Euclidean' kinematics. Thus, we may simply use straight line integration contours. In Fig.~\ref{fig:I3-expample} we show the order $\epsilon$ expansion coefficient $I^{(1)}_{1,1+\eps,\eps}$, while the complete numerical solution for Eq.~(\ref{eq:I3-example-exp}) can be found in \wwwaux{RealPSNum}. In the numerical calculation, we have specialized to the case where the three-vectors $\vec{p}_1$, $\vec{p}_2$ and $\vec{p}_3$ are co-planar.
%, i.e., $\cos\chi_3^{(2)} = -1$. 
Then the result depends only on the angles $\chi_2^{(1)}$ and $\chi_3^{(1)}$ and we can represent the functions as surfaces in three dimensions.

\section*{Problems}
\addcontentsline{toc}{section}{Problems} 
\begin{problem}
\label{problem_num1d}
Calculate the asymptotics for 1-dim integrals Eq.~(\ref{1dimNumEx1}) and 
Eq.~(\ref{parts}), check that they do not depend on the position of the contour $z_0$.
\\
\hint{} See derivation above Eq.~(\ref{Ex1limit}).
\end{problem}

\begin{problem}
\label{problem_num1d-2}
For Eq.~(\ref{3dimEx12}) check the direction of the damping factor cancellation and calculate the value of the parameter $\alpha$ in the fractional part $\frac{1}{|t|^{\alpha}}$ of the asymptotic for a contour position before and after the shift $z_2 \rightarrow z_2 +2.$ \\
\hint{} For checking the direction keep in mind that the $i t$ term in the gamma function leads to $e^{- \pi |t|/2}$ in the limit. For calculating $\alpha$,  put all irrelevant $t_i$ equal to 0.
\end{problem}

\begin{problem}
\label{prob:thetalimit}
Find optimal $\theta$ parameter values for the rotated and parabolic contours illustrated in Fig.~\ref{fig:contours_fig} for the integral in Eq.~(\ref{parts}). \\
\hint{} The relation $\ln z - \ln(-z) = i \pi \, \mbox{sign}(\Im (z))$ doesn't depend on parametrization. Compare asymptotics of gamma functions ratio with the kinematic part, $\theta$ should restore exponential damping in both directions $t \rightarrow \pm \infty$.
The result should have the form of an inequality.
\end{problem}

\begin{problem}
\label{prob:mbnumerics}
Install and run examples for numerical integration of \mb{} integrals using the \texttt{MBnumerics} package given in~\cite{ambrewww}. \\
\hint{} For installation, follow the Appendix and \texttt{HOWTO} discussed there. 
\end{problem}
%\begin{question}{Problems}
%\begin{enumerate}
%\item 
\begin{problem}
\label{problem_mbde} Concerning the \texttt{MBDE} approach described in section~\ref{sec:thimbles}, it is a very efficient method for  1-dim \mb{} integrals. However, already 2-dim cases are not fully understood, see for instance W.~Flieger's contribution to~\cite{Blondel:2018mad} and the PhD thesis by Z.~Peng~\cite{Peng:2012zpa}. 
In the case of \mbr{} integrals we have to face the following issues:
  
\begin{itemize}
  
\item since the gamma function has infinitely many critical points we have to consider how many of them are necessary to match the assumed numerical precision;
  
\item the parametrization of thimbles by solving flow equations for each relevant critical point. This stage requires also calculation of the Hessian matrix and the associated eigensystem;
  
\item the knowledge of singularities of the MB integrand is necessary to avoid them by thimbles; 
  
\item the determination of the inverse function for the cuts approach;
  
\item the summation of thimbles with suitable coefficients in \eqref{split} to reconstruct the whole integral. 
\end{itemize}
  
You are invited to make a reasearch in this direction!
\end{problem} 

\putbib[%
bibs/refs,%
bibs/MBmethods,%
bibs/2loops_LL16,%
bibs/Phd_Dubovyk,%
bibs/LRrefa,%
bibs/2loopsreport]
\end{bibunit} 

%% file: appendixv2.tex
%%%%%%%%%%%%%%%%%%%%% appendix.tex %%%%%%%%%%%%%%%%%%%%%%%%%%%%%%%%%
%
% sample appendix
%
% Use this file as a template for your own input.
%
%%%%%%%%%%%%%%%%%%%%%%%% Springer-Verlag %%%%%%%%%%%%%%%%%%%%%%%%%%
\begin{bibunit}[elsarticle-num-ID] 
\let\stdthebibliography\thebibliography
\renewcommand{\thebibliography}{%
\let\section\subsection
\stdthebibliography} 

\appendix 
\chapter{Public Software and Codes For \mb{} Studies}
\label{introA}  

\section{Analytic software \label{app:analsoftware}}
 For a construction of \mb{} representations, analytic continuation in $\eps$, expansions, decreasing dimensionality of the \mb{} integrals, see chapters \ref{chapter-MBrepr}-\ref{chapter-MBanal}.   

\begin{enumerate}
   
%\end{enumerate}
%\section
\item[{\bf A.1.1}] 
{\ambrem \label{app:ambre}} %\cite{Dubovyk:2016ocz,Dubovyk:2016zok,Dubovyk:2016aqv}} 
The Mathematica toolkit AMBRE derives \mblong{} representations for Feynman integrals in $d = n-2 \epsilon$ dimensions.
The package is described in \cite{Dubovyk:2016ocz,Dubovyk:2016zok,Dubovyk:2016aqv}.
Connected useful \mb{} literature is \cite{Gluza:2007rt, Gluza:2007bd, Gluza:2009mj, Gluza:2010mz, Gluza:2010rn, Bielas:2013rja, Blumlein:2014maa, Dubovyk:2015yba, Ochman:2015fho,
      Dubovyk:2016ocz, Dubovyk:2016zok, Gluza:2016fwh, Dubovyk:2017cqw, Prausa:2017frh, Czakon:2005rk,
      Freitas:2010nx,
      Freitas:2012sy, Dubovyk:2016aqv,
      Tausk:1999vh,
      Gonzalez:2007ry, Gonzalez:2008xm, Gonzalez:2010uz,
      Cvitanovic:1974uf,
      Smirnov:2009up, Anastasiou:2005cb,
      Bogner:2010kv, nakanishi1971graph,
      mbtools-kosower, mbtools-smirnov, Bielas:2013v12, Gluza:2010v22, Dubovyk:201509v30x}.
      
\texttt{AMBRE} versions overview:
\begin{itemize}
 \item iteratively to each subloop -- loop-by-loop approach (LA): mostly for planar \\
       ({\texttt{AMBREv1.3.1}} \& {\texttt{AMBREv2.1.1}}) \\
%       \begin{center}
       \includegraphics[scale=0.45]{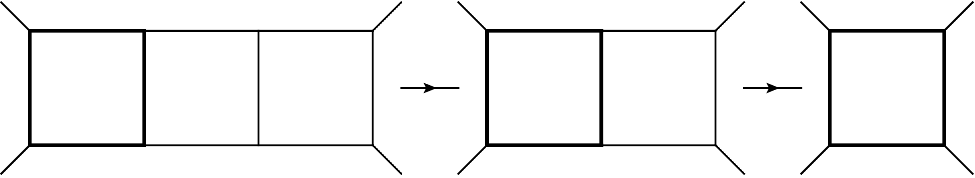}
%       \end{center}
 \item in one step to the complete U and F polynomials -- global approach (GA): general \\
       ({\texttt{AMBREv3.1.1}})

 \item combination of the above methods -- Hybrid approach (HA) \\
 (\texttt{AMBREv4}, not public)

%       \begin{center}
       \includegraphics[scale=0.35]{FIGS/mxpl.pdf} \\
       \includegraphics[scale=0.35]{FIGS/mxnp.pdf} 
%       \end{center}
\end{itemize}

Examples, description, updates, links to basic tools and literature can be found on the webpage \cite{ambrewww}.

The \math{} package \cite{Bielas:2013rja,ambrewww} determines planarity of a Feynman diagram upon its propagators (momentum flows). %See chapters \ref{chapter-MBrepr}-\ref{chapter-MBanal}.

\item[{\bf A.1.2}] {\texttt{MBcreate.m} \label{app:mbcreate}} \\
The \math{} code \cite{Belitsky:2022gba} for automatic construction of a concise Mellin-Barnes representation for a given parametric integral.

\item[{\bf A.1.3}] {\pltestm \label{app:pltest}}\\
The \math{} package \cite{Bielas:2013rja,ambrewww} determines planarity of a Feynman diagram upon its propagators (momentum flows). %See chapters \ref{chapter-MBrepr}-\ref{chapter-MBanal}.

%\section{\pltestm \label{app:pltest}}

% \section{\mbm{} \label{app:MB.m}}
\item[{\bf A.1.4}] {\mbm{} \label{app:MB.m}}\\
The {\math{} package \cite{Czakon:2005rk,mbtools} allows to analytically continue any \MB{} integral in a given parameter and to resolve the singularity structure in this parameter.}  %See chapters \ref{chapter-MBrepr}-\ref{chapter-MBnum}.
 
% \section{\texttt{MBresolve.m} \label{app:MB.resolve}}
\item[{\bf A.1.5}] {\texttt{MBresolve.m} \label{app:MB.resolve}} \\
The \math{} package \cite{Smirnov:2009up,mbtools-smirnov}   resolves singularities of multifold Mellin-Barnes integrals in the dimensional regularization parameter $\eps$.  An alternative to \mbm{}.%, see chapter \ref{chapter-singul}.
 
% \section{\texttt{barnesroutines.m}}
 \item[{\bf A.1.6}] {\texttt{barnesroutines.m}}\\
 A tool for an automatic application of the first and second Barnes lemmas on lists of multiple Mellin-Barnes integrals \cite{mbtools-kosower}. 
% See in particular section \ref{sec:5MBrepr,sec:BLeff}. 
 
 % \section{\texttt{MBasymptotics.m}}
 \item[{\bf A.1.7}] \label{app:mbasymptotics} {\texttt{MBasymptotics.m}}\\
 The \math{} routine by Michał Czakon which expands Mellin-Barnes integrals in a small parameter \cite{mbtools}.

%\section{\texttt{asy2.m} \label{asy2}} 
\item[{\bf A.1.8}] {\texttt{asy2.m} \label{asy2}}
The \math{} code \cite{Jantzen:2012mw} performs asymptotic expansions using the strategy of expansion by regions and recognizes potential regions in threshold expansions or Glauber regions.

%\section{ \texttt{ASPIRE}} 
\item[{\bf A.1.9}] { \texttt{ASPIRE}} \\
The \math{} code \cite{Ananthanarayan:2018tog} which is  an alternative to \texttt{asy2.m}.% \ref{asy2}.  
\\ 
\end{enumerate}
\section{Multiple sums \label{app:mplsums}}
 
As discussed in chapters \ref{chapter:complex} and \ref{chapter-MBanal}, taking residues of \mb{} integrals leads to multiple sums over  hyper\-geo\-metric expressions.  Several techniques have been developed to manipulate these sums see e.g. \cite{Moch:2001zr,Davydychev:2003mv,Weinzierl:2004bn,Kalmykov:2007dk,Blumlein:2018cms,McLeod:2020dxg,Kotikov:2020ccc}.\\ Below we list some of the codes which deal with multiple sums. 
%\end{svgraybox}
\begin{enumerate}
  
%\section{\texttt{MBsums.m}}
\item[\bf{A2.1}] {\texttt{MBsums}}\\
The \math{} package \cite{Ochman:2015fho,ambrewww} which  transforms \mblong{} integrals into multiple sums. 

%\section{\texttt{MBConicHulls}}
\item[\bf{A2.2}] {\texttt{MBConicHulls}}\\
The \math{} package 
\cite{Ananthanarayan:2020fhl} for generating multiple series representations of N-fold MB integrals using geometrical analysis based on conic hulls, $N>2$. Relies on the \texttt{MultivariateResidues} \math{} package.
In \cite{Banik:2025kwa}, the package was extended using the conic-hull and triangulation methods to compute \MB{} integrals whose integrands contain polygamma functions.

%{\texttt{MultivariateResidues}}.%\ref{multivariate}. 

%\section{\texttt{MultivariateResidues} \label{multivariate}}
\item[\bf{A2.3}] {\texttt{MultivariateResidues}}\\
The \math{} package \cite{Larsen:2017aqb} for efficient evaluation of multidimensional residues based on methods from computational algebraic geometry.

%\section{\texttt{SUMMER}}
\item[\bf{A2.4}] {\texttt{SUMMER}}\\
The  package \cite{Vermaseren:1998uu} written in \texttt{FORM} \cite{Vermaseren:2000nd} for evaluation of nested symbolic sums over combinations of harmonic series, binomial coefficients and denominators. In addition it treats Mellin transforms and the inverse Mellin transformation for functions that are encountered in Feynman diagram calculations. 

%\section{\texttt{nestedsums}}
\item[\bf{A2.5}]  {\texttt{SUMMER}}\\
The C++ program \cite{Weinzierl:2002hv} with \texttt{GiNaC} library for expansion of transcendental functions \cite{Moch:2001zr} with the algorithms  based on the Hopf algebra of nested sums. 

%\section{\texttt{XSummer}}
\item[\bf{A2.6}] {\texttt{XSummer}}\\
The  package \cite{Moch:2005uc} written in \texttt{FORM} \cite{Vermaseren:2000nd} for evaluation of nested sums, where the harmonic sums and their generalizations appear as building blocks, originating for example from the expansion of generalized hypergeometric functions around integer values of the parameters.

%\section{\texttt{HarmonicSums}}
\item[\bf{A2.7}] {\texttt{HarmonicSums}} \\
The \math{} package \cite{Ablinger:2013cf} for evaluation of S-sums and generalized harmonic polylogarithms.% (see chapter \ref{chapter-MBanal}).

%\section{\texttt{SumProduction}, \texttt{EvaluateMultiSums}, \texttt{Sigma} }
\item[\bf{A2.8}] {\texttt{SumProduction}, \texttt{EvaluateMultiSums}, \texttt{Sigma} }\\
The packages \texttt{EvaluateMultiSums}, \texttt{SumProduction} and \texttt{Sigma} 
\cite{Schneider:2013zna,Ablinger:2013eba}
deal with multi-sums over hypergeometric expressions. The results are given, if possible,   in terms of harmonic sums, generalized harmonic sums, cyclotomic harmonic sums or binomial sums.
\item[\bf{A2.9}] {\texttt{HypExp}}, {\texttt{HypExp 2}}  \\        
   The Mathematica package \texttt{HypExp} which allows to expand hypergeometric functions $_JF_{J-1}$ around integer and half-integer parameters to arbitrary order  \cite{Huber:2005yg,Huber:2007dx}.
\item[\bf{A2.10}] {\texttt{Olsson}}  \\        
   The Mathematica package \cite{Ananthanarayan:2021yar}
which aims to find linear transformations for some classes of multivariable hypergeometric functions. It is based on \cite{Olsson:1964}.
\item[\bf{A2.11}] {\texttt{IBIS}}  \\        
   The ``Inverse BInomial sum Solver'', a program~\cite{vanHoegaerden:2025nmh} written in FORM to perform  classes of sums involving inverse binomials. Results are expressed in terms of S-sums, where the dependence on the upper limit of the sum is analytic.
\item[\bf{A2.12}] {\texttt{MultiHypExp.wl}}  \\        
   The Mathematica package \cite{Bera:2023pyz} allows for the expansion of multivariate hypergeometric functions in the dimensional regularisation parameter using a differential operator.

\end{enumerate}

\section{Polylogarithms and generalizations \label{app:mplsgpls}}

%text from 0508212
In particle physics HPLs $H(\{a\},x)$ have been introduced in \cite{Remiddi:1999ew} in order to describe systematically QED self-energies and vertices, $\{a\} \in \{1,0,-1\}$. 
HPLs with some generalized arguments are introduced in \cite{Aglietti:2003yc}.
A generalization of HPLs for three-scale problems (e.g. QED boxes) are the two-dimensional harmonic polylogarithms $G(\{b\},x)$ (GPLs);
the index vector $\{b\}$ has elements $1,0,-1$, but now also those depending on a second kinematic variable $y$.
This was observed in \cite{Gehrmann:1999as}
and systematically worked out in \cite{Gehrmann:2000zt}
 for a planar massless problem.
There one may find analytical expressions for the
GPLs until weight 3 with indexes $0, 1, z, 1-z$.
GPLs with some generalized arguments are introduced in \cite{Aglietti:2003yc}.
A generalization of HPLs for three-scale problems (e.g. QED boxes) are the two-dimensional harmonic polylogarithms $G(\{b\},x)$ (GPLs);
the index vector $\{b\}$ has elements $1,0,-1$, but now also those depending on a second kinematic variable $y$.
This was observed in \cite{Gehrmann:1999as}
and systematically worked out in \cite{Gehrmann:2000zt}
 for a planar massless problem.
There one may find analytical expressions for the
GPLs until weight 3 with indexes $0, 1, z, 1-z$. 
In fact, the minimal basis has been systematically worked out in \cite{Blumlein:2009ta}.
See also for MPLs studies in \cite{Ablinger:2013cf,Bonciani:2010ms,Czakon:2005jd}. 

As far as numerical solutions are concerned, the problem is solved for any
weight of \texttt{HPLs} and \texttt{GPLs}. 
Classical and Nielsen generalized polylogarithms are available within both the \math{} and \texttt{GiNac} \cite{Bauer:2000cp,ginacwww} language as built-in functions
(the evaluation for arbitrary complex arguments and without any restriction on the weight).
In addition many special packages have been developed. 

\begin{enumerate}
%\section{{\texttt{hplog}}}
\item[\bf{A3.1}] {{\texttt{hplog}}}\\
The numerical \texttt{Fortran} \cite{Gehrmann:2001pz} evaluation of \texttt{HPLs} up to weight four with indexes
$\{0,1,-1\}$ and the analytic continuation.

%\section{\texttt{dhpl}}
\item[\bf{A3.2}] {\texttt{dhpl}}\\
The systematic numerical treatment
of GPLs with indexes $\{0,1,1-z,-z\}$, performed in
\texttt{Fortran} \cite{Gehrmann:2001jv}. 

%\section{\texttt{HPOLY.f}}
\item[\bf{A3.3}] {\texttt{HPOLY.f}} \\
The \texttt{Fortran} numerical code \cite{Ablinger:2018sat} for \texttt{HPLs} up to weight eight. 

%\section{\texttt{handyG}}
\item[\bf{A3.4}] {\texttt{handyG}} \\
The \texttt{Fortran} library \cite{Naterop:2019xaf} for the evaluation of \texttt{GPLs}, in principle of any weight, suitable for Monte Carlo integration.

%\section{\texttt{FastGPL}}
\item[\bf{A3.5}] {\texttt{FastGPL}} \\
The \texttt{C++} library \cite{Wang:2021imw} for  efficient and accurate  evaluation of \texttt{GPLs} based on work \cite{Vollinga:2004sn}, suitable for Monte Carlo integration and event generation. 
         
%\section{\texttt{HPL.m}}
\item[\bf{A3.6}] {\texttt{HPL.m}}\\
The \math{} \texttt{HPLs} package \cite{Maitre:2005uu} with implementation of the product algebra, the derivative properties, series expansion and numerical evaluation of \texttt{HPLs}.

%\section{PolyLogTools} 
\item[\bf{A3.7}] {\texttt{PolyLogTools}}\\
The \math{} package \cite{Duhr:2019tlz} for multiple polylogarithms, including the Hopf algebra of the multiple polylogarithms and the symbol map, as well as the construction of single valued multiple polylogarithms with an algorithm for given symbol combination of MPLs, including the so-called fibration bases. 

%\section{\texttt{HyperInt}}
\item[\bf{A3.8}] {\texttt{HyperInt}}\\
The package \cite{Panzer:2014caa} based on the \texttt{Maple} CAS \cite{} for symbolic integration of hyperlogarithms multiplied by rational functions, which also include \texttt{MPLs} when their arguments are rational functions. 
 
%\section{\texttt{CHAPLIN}}
\item[\bf{A3.9}] {\texttt{CHAPLIN}}\\
The \texttt{Fortran} library \cite{Buehler:2011ev} to evaluate all \texttt{HPLs} up to weight four numerically for any complex argument. 

\item[\bf{A3.9}] {\texttt{HYPERDIRE}}\\
The \math{} package \cite{Bytev:2013gva} devoted to the creation of a set of based programs for the differential reduction of hypergeometric functions.

\item[\bf{A3.10}] {\texttt{Diogenes}}\\
The \math{} package \cite{Bezuglov:2023owj} for generalized hypergeometric functions of one variable and the Appell and Lauricella functions and their expansions in terms of the Goncharov polylogarithms.
 
\item[\bf{A3.11}] {\texttt{FeynGKZ}}\\
The \math{} package \cite{Ananthanarayan:2022ntm}   for solving Feynman integrals using GKZ hypergeometric systems.

\item[\bf{A3.12}] {\texttt{HyperPrecision}}\\
The \math{} package \cite{Banik:2026cfo}  for high-precision numerical evaluation of general Horn-type multivariate hypergeometric functions and their Laurent expansions in a small parameter $\epsilon$.

 \item[\bf{A3.12}] {\texttt{Subtropica}}\\
The \math{} package \cite{Giroux:2026tgd} for symbolic integration of multi-polylogarithmic integrals using tropical geometry.

\end{enumerate}

\section{\MB{} numerical software \label{app:mbnum}}

There are at the moment two publicly available packages for \mb{} numerical evaluations.

\begin{enumerate}
%\section{\mbm{}}
\item[\bf{A4.1}] {\mbm{}}\\
\mbm{} listed in {\bf A1.3} %\ref{app:MB.m} 
also allows to solve \mb{} representations numerically. It works fine in the Euclidean kinematic region. 

\mbm{} requires the \texttt{CUBA} Monte Carlo and {\tt quasi} Monte Carlo library \cite{Hahn:2004fe} for numerical integration and \texttt{CERNLIB} for evaluation of gamma and polygamma functions. 

In the original Mathematica package \MB{} \cite{Czakon:2005rk}, the gamma and polygamma functions are obtained from links to the libraries \texttt{libmathlib.a} and \texttt{libkernlib.a} of \texttt{CERNLIB}.
%In fact, only the libraries libmathlib.a and libkernlib.a are actually required.
 Unfortunately, 
the development and support of the original \texttt{CERNLIB} are no longer available \cite{cernlib}
%, see \url{http://cernlib.web.cern.ch/cernlib/} 
(last update 10 Oct 2014, 12 Feb 2018).
It may, and in fact causes problems when running numeric within \mbm{} and programs based on \mbm{}. 

In \wwwaux{HowTo}, we write a short manual how to solve the problems if numerical \mbm{} output gives NaNs.
\mbm{} available at \cite{mbtools}.

%\section{\texttt{MBnumerics.m}}
\item[\bf{A4.2}] {\texttt{MBnumerics.m}} \\
The \math{} package for evaluation of \mb{} integrals in any kinematic point (Minkowskian or Euclidean).
\texttt{MBnumerics.m} available at \cite{ambrewww}.
\end{enumerate}

\section{Other methods for \texttt{FI} numerical integrations \label{app:othernum}}

The well established method of calculations is the sector decomposition method. The programs \texttt{pySecDec} and \texttt{FIESTA} allow now for direct numerical integration 
of \texttt{FI} in physical (Minkowskian) space. In recent years numerical evaluation of \texttt{FI} based on differential equations has been also developed very much, there are several unpublic programs or set of connected packages \cite{Mandal:2018cdj,Czakon:2020vql,Bronnum-Hansen:2020mzk} with exception of{\texttt{AMFlow}} and \texttt{SeaSyde} \cite{Liu:2022chg,Armadillo:2022ugh}.  Some of them relies on \texttt{DiffExp} \cite{Hidding:2020ytt},  a \math{} package for integrating families of \texttt{FI} order-by-order in the dimensional regulator from their systems of differential equations, in terms of one-dimensional series expansions along lines in phase-space, the idea based on the work \cite{Moriello:2019yhu}.

\begin{enumerate}
%\section{\mbm{}}
\item[\bf{A5.1}] {\texttt{pySecDec}} \\
The project \texttt{SecDec} performs the factorization of dimensionally regulated poles in parametric integrals, and the subsequent numerical evaluation of the finite coefficients. 
The latest version with the algebraic part of the program is written with python modules \cite{Borowka:2017idc,Borowka:2018goh}.
\item[\bf{A5.2}] {\texttt{FIESTA}} \\
 The program FIESTA is an alternative package to \texttt{pySecDec}, in the newest version it is supported by graphical processor units (GPUs) for the numerical integration aiming at optimal performance at large scales when one is increasing the number of sampling points in order to reduce the uncertainty estimates \cite{Smirnov:2015mct}.
 
\item[\bf{A5.3}] {\texttt{AMFlow}} \\
 The \math{}  package \cite{Liu:2022chg} for numerical computation of dimensionally regularized \texttt{FI} via so-called auxiliary mass flow method. In this framework, integrals are treated as functions of an auxiliary mass parameter and their results can be obtained by constructing and solving differential systems with respect to this parameter, in an automatic way. 
\item[\bf{A5.4}] {\texttt{SeaSyde}} \\
The \math{} package \cite{Armadillo:2022ugh} for numerical computation of dimensionally regularized \texttt{FI}, with internal masses being in general complex-valued. The implementation solves by series expansions the system of differential equations satisfied by the Master Integrals.  
\item[\bf{A5.5}] {\texttt{AmpRed}} \\
The semi-automatic package \cite{Chen:2024xwt} for the  calculations of multi-loop Feynman amplitudes. It implements the methods of integration by parts and differential equations.
\item[\bf{A5.6}] {\texttt{Line}} \\
 An open source computer program \cite{Prisco:2025wqs} for  numerical evaluation of the master integrals by solving first-order ordinary differential equations through series expansion.   
\end{enumerate} 

\chapter{Additional Working Files \label{appB1}}
%\section{Additional Working Files}

At the web pages \cite{www_aux_springer} we gathered files  with materials and solutions discussed in the book. For readers convenience, \math{} files are accompanied with corresponding pdf files:

\begin{enumerate}
    \item MB\_miscellaneous\_Springer.nb
    \item MB\_SE2l2m\_Springer.nb
    \item MB\_Simpl\_Springer.nb
    \item MB\_V63MZ1MH\_Springer.nb  
    \item MB\_V6l3m1M\_Springer.nb
    \item MB\_V6l1MZ\_Springer.nb
    \item MB\_V6l0m\_Springer.nb
    \item MB\_SE6l0m\_Springer.nb
    \item MB\_V8l0m\_Springer.nb
    \item MB\_AMBREnew\_Springer.nb
    \item MB\_Zmatrix\_Springer.nb
    \item MB\_3L\_Springer.nb
    \item MB\_Gamma0\_Springer.nb
    \item MB\_PSint\_Springer.nb
    \item MB\_B5l2m2\_Springer.nb
    \item MB\_B5nf\_Springer.nb
    \item MB\_basicsums\_Springer.nb
    \item MB\_I11\_massless\_Springer.nb
    \item MB\_O11\_onemass\_Springer.nb
    \item MB\_Decoupling\_Springer.nb
    \item MB\_Euler\_Springer.nb
    \item MB\_symbolicint\_Springer.nb
    \item MB\_B7l4m1\_Springer.nb
    \item MB\_3dimNum\_Springer.nb
    \item MB\_1dimNum\_Springer.nb
    \item MB\_TH1\_Springer.nb
    \item MB\_TH2\_Springer.nb
    \item MB\_RealPSNum\_Springer.nb
    \item MB\_V6l1m\_Springer.sh  
    \item SD\_V6l1m\_generate\_Springer.py 
    \item SD\_V6l1m\_integrate\_Springer.py  
\end{enumerate}

\chapter{Some \mb{} Talks \label{appB}}

\begin{enumerate}
    \item V.A.~Smirnov, {\it{The method of Mellin–Barnes Representation}}. \newline In: School of Analytic Computing in Theoretical High-Energy Physics,  \newline Atrani, Italy 2013,
     \href{https://indico.cern.ch/event/248025/}{https://indico.cern.ch/event/248025/}.
      \item T.~Riemann, {\it{Integrals, Mellin-Barnes representations and Sums}}. \newline In: Computer Algebra and Particle Physics 2009,  \newline DESY, Zeuthen, Germany, 2009,
     \href{https://indico.desy.de/event/1573}{https://indico.desy.de/event/1573}.
\item G.~Passarino, {\it{"Feynman Fox integrals"}}, \newline J.~Gluza, {\it "Status and perspectives for Mellin-Barnes integrals"}, \newline
In: Loop Summit 2 - new perturbative results and methods in precision physics: Quantum Field Theory and Collider Physics, 
\href{https://indico.desy.de/event/44685/}{https://indico.desy.de/event/44685/}. 
\end{enumerate} 

 More generally, numerous talks on the calculation of Feynman integrals can be found on the Indico page of the \href{https://indico.cern.ch/category/21380/}{FCC Precision Calculations Working Group}, including a talk specifically devoted to Mellin--Barnes techniques \cite{banikWGFCCtalk}. 
 
\putbib[%
bibs/refs,%
bibs/2loops_LL16,%
bibs/Phd_Dubovyk,%
bibs/LRrefa,%
bibs/2loopsreport]
\end{bibunit}